\documentclass[runningheads]{llncs}
\usepackage{amsmath}%
\usepackage{pifont}%
\usepackage{amssymb}%
\usepackage{stmaryrd}%
\usepackage{array}%
\usepackage{epic}%
\usepackage[fixamsmath]{mathtools}
\usepackage{cancel}
\usepackage{url}
\usepackage{bigstrut}
\usepackage{wrapfig}
\usepackage{mathabx}

\DeclareMathAlphabet\mathbb{U}{msb}{m}{n}
\DeclareMathAlphabet\mathbfcal{OMS}{cmsy}{b}{n}
\DeclareMathAlphabet\mathbfit{OML}{cmm}{b}{it}
\newtheorem{fact}{Fact}
\DeclareMathAlphabet{\mathrm}    {OT1}{cmr}{m}{n}
\DeclareMathAlphabet{\mathrmbf}  {OT1}{cmr}{bx}{n}
\DeclareMathAlphabet{\mathrmit}  {OT1}{cmr}{m}{it}
\DeclareMathAlphabet{\mathrmbfit}{OT1}{cmr}{bx}{it}
\DeclareMathAlphabet{\mathsf}    {OT1}{cmss}{m}{n}
\DeclareMathAlphabet{\mathsfbf}  {OT1}{cmss}{bx}{n}
\DeclareMathAlphabet{\mathsfit}  {OT1}{cmss}{m}{sl}
\DeclareMathAlphabet{\mathtt}    {OT1}{cmtt}{m}{n}
\DeclareMathAlphabet{\mathttbf}  {OT1}{cmtt}{bx}{n}
\DeclareMathAlphabet{\mathttit}  {OT1}{cmtt}{m}{it}
\DeclareMathAlphabet{\mathpzc}   {OT1}{pzc}{m}{it}
\newcommand{\keywords}[1]{\par\addvspace\baselineskip\noindent\enspace\ignorespaces{\bfseries Keywords:\,}#1}
\newcommand{\comment}[1]{}
\newenvironment{aside}
{\begin{flushleft}\begin{minipage}{340pt}\textbf{Aside:}}
{\end{minipage}\end{flushleft}}
\makeatletter
\def\dual#1{\expandafter\dual@aux#1\@nil}
\def\dual@aux#1/#2\@nil{\begin{tabular}{@{}c@{}}#1\\#2\end{tabular}}
\makeatother
\begin{document}

\pagestyle{headings}
\title{Relational Operations in \texttt{FOLE}} 
\titlerunning{Relational Operations in \texttt{FOLE}}  
\author{Robert E. Kent}
\institute{Ontologos}
\maketitle

\begin{abstract}
This paper discusses relational operations 
in the first-order logical environment {\ttfamily FOLE}
(Kent~\cite{kent:iccs2013}).
%This paper studies 
Here we demonstrate
how \texttt{FOLE} 
expresses the relational operations of database theory 
in a clear and implementable representation.
An analysis of 
the representation of database tables/relations in \texttt{FOLE}
(Kent~\cite{kent:fole:era:tbl})
reveals a principled way to express the relational operations.
This representation is expressed in terms of a distinction between 
basic components versus composite relational operations.
The 9 basic components fall into three categories: 
reflection (2), Booleans or basic operations (3), 
and adjoint flow (4). 
Adjoint flow is given for signatures (2) and for type domains (2), 
which are then combined into full adjoint flow.
%After an explanation of 
The basic components  
are used to express various composite operations,
where we illustrate each of these with a flowchart.
Implementation of the composite operations
is then expressed in an input/output table
containing four parts:
constraint,
construction,
input, and output.
We explain how limits and colimits are constructed from diagrams of tables,
and then classify composite relational operations into three categories:
limit-like, colimit-like and unorthodox.
%Many others can be similarly defined.
%Next,
%We also discuss adjoint flow at large.
%Finally,
%we discuss a change of basis for both type domains and signatures,
%and study how these affect the Booleans and local flow,
%and thus flowcharts.
%
\keywords{tables, adjoint flow, relational operations.}
\end{abstract}

\tableofcontents

%%%%%%%%%%%%%%%%%%%%%%%%%%%%%%%%%%%%%%%%%%%%%%%%%%%%%%%%%%%%%%%%%%%%%%%%%%%%%%%
%%%%%%%%%%%%%%%%%%%%%%%%%%%%%%%%%%%%%%%%%%%%%%%%%%%%%%%%%%%%%%%%%%%%%%%%%%%%%%%%

%%%%%%%%%%%%%%%%%%%%%%%%%%%%%%%%%%%%%%%%%%%%%%%%%%%%%%%%%%%%%%%%%%
%%%%%%%%%%%%%%%%%%%%%%%%%%%%%%%%%%%%%%%%%%%%%%%%%%%%%%%%%%%%%%%%%%
%%%%%%%%%%%%%%%%%%%%%%%%%%%%%%%%%%%%%%%%%%%%%%%%%%%%%%%%%%%%%%%%%%
%
\newpage
\section{Introduction}
\label{sec:intro}
%%%%%%%%%%%%%%%%%%%%%%%%%%%%%%%%%%%%%%%%%%%%%%%%%%%%%%%%%%%%%%%%%
%%%%%%%%%%%%%%%%%%%%%%%%%%%%%%%%%%%%%%%%%%%%%%%%%%%%%%%%%%%%%%%%%
%%%%%%%%%%%%%%%%%%%%%%%%%%%%%%%%%%%%%%%%%%%%%%%%%%%%%%%%%%%%%%%%%

%%%%%%%%%%%%%%%%%%%%%%%%%%%%%%%%%%%%%%%%%%%%%%%%%%%%%%%%%%%%%%%%%%
%%%%%%%%%%%%%%%%%%%%%%%%%%%%%%%%%%%%%%%%%%%%%%%%%%%%%%%%%%%%%%%%%%
%\newpage
%\subsection{The Relational Model}
%\label{subsec:rel:mod}
\paragraph{The Relational Model.}
%%%%%%%%%%%%%%%%%%%%%%%%%%%%%%%%%%%%%%%%%%%%%%%%%%%%%%%%%%%%%%%%%
%%%%%%%%%%%%%%%%%%%%%%%%%%%%%%%%%%%%%%%%%%%%%%%%%%%%%%%%%%%%%%%%%

%\mbox{}\newline
%$abcdefghijklmnopqrstuvwxyz
%\mbox{}\newline
%\alpha\beta\gamma\delta\epsilon\zeta\eta\theta\iota\kappa\lambda\mu\nu\xi
%%\omicron
%\pi\rho\sigma\tau\upsilon\phi\chi\psi\omega$
%$\rho_{.}^{.}$
%$\acute{\alpha}\grave{\alpha}\acute{\beta}\grave{\beta}$
%\acute{\imath}
%\grave{\imath}
%\acute{\jmath}
%\grave{\jmath}

%$\acute{p}\acute{\imath}\acute{e}\acute{r}$
%$\wr\int\oint\flat\natural\sharp\uparrow\downarrow\varkappa
%\nabla\triangle\curvearrowleft$

Many-sorted (multi-sorted) first-order predicate logic 
represents a communities ``universe of discourse'' as 
a heterogeneous collection of objects
by conceptually scaling
%partitioning 
the universe according to types.
The relational model for database management 
%\cite{codd:90}
uses a structure and language 
consistent with this logic.
The relational model was initially discussed in two papers:
``A Relational Model of Data for Large Shared Data Banks''
by Codd \cite{codd:70} 
and
``The Entity-Relationship Model -- Toward a Unified View of Data'' 
by Chen \cite{chen:76}. 

The relational model follows many-sorted logic
by representing data in terms of many-sorted relations, 
subsets of the Cartesian product of multiple data-types. 
All data is represented horizontally in terms of tuples, 
which are grouped vertically into relations. 
A database organized in terms of the relational model 
is a called relational database.
The relational model provides a method 
for modeling the data stored in a relational database 
and for defining queries upon it. 
In the relational model there are two approaches for database management:
%The relational algebra and the relational calculus;
the relational algebra, which defines an imperative language,
and
the relational calculus, which defines a declarative language.
%

%%%%%%%%%%%%%%%%%%%%%%%%%%%%%%%%%%%%%%%%%%%%%%%%%%%%%%%%%%%%%%%%%%
%%%%%%%%%%%%%%%%%%%%%%%%%%%%%%%%%%%%%%%%%%%%%%%%%%%%%%%%%%%%%%%%%%
%\newpage
%\subsection{FOLE}
%\label{subsec:FOLE}
\paragraph{\texttt{FOLE}.}
%%%%%%%%%%%%%%%%%%%%%%%%%%%%%%%%%%%%%%%%%%%%%%%%%%%%%%%%%%%%%%%%%
%%%%%%%%%%%%%%%%%%%%%%%%%%%%%%%%%%%%%%%%%%%%%%%%%%%%%%%%%%%%%%%%%

The first order logical environment {\texttt{FOLE}} is a category-theoretic approach to many-sorted first order predicate logic.
The {\texttt{FOLE}} approach to logic, 
and hence to databases, 
%uses lists and classifications.
%The {\texttt{FOLE}} representation for database tables 
relies upon two mathematical concepts:
(1) lists and (2) classifications.
Lists represent database signatures and tuples;
classifications represent data-types and logical predicates.
{\texttt{FOLE}} 
represents the header of a database table as a list of sorts, and
represents the body of a database table as a set of tuples 
classified by the header.
The notion of a list is common in category theory
%``Categories for the Working Mathematician''
(Mac Lane
\cite{maclane:71}).
The notion of a classification is described in two books:
``Information Flow: The Logic of Distributed Systems''
by Barwise and Seligman \cite{barwise:seligman:97} and
``Formal Concept Analysis: Mathematical Foundations''
by Ganter and Wille \cite{ganter:wille:99}. 

In this paper
we explain how 
%the first order logical environment 
{\texttt{FOLE}} 
provides a categorically well-founded semantics for relational algebra.
We present that semantics in a detail
%that will be 
suitable for implementation.
The {\texttt{FOLE}} semantics for the relational calculus 
will appear elsewhere.
%
%The first order logical environment 
{\texttt{FOLE}} is described in multiple papers by the author:
``Database Semantics''
\cite{kent:db:sem},
``The First-order Logical Environment''
\cite{kent:iccs2013},
``The {\ttfamily ERA} of {\ttfamily FOLE}: Foundation''
\cite{kent:fole:era:found},
``The {\ttfamily ERA} of {\ttfamily FOLE}: Superstructure''
\cite{kent:fole:era:supstruc},
and
``The {\ttfamily FOLE} Table''
\cite{kent:fole:era:tbl}.
\newline
%

%\item[Overview of the Paper:] 
%We first review various aspects of the \texttt{FOLE} notion of a table.
%

%%%%%%%%%%%%%%%%%%%%%%%%%%%%%%%%%%%%%%%%%%%%%%%%%%%%%%%%%%%%%%%%%%
%%%%%%%%%%%%%%%%%%%%%%%%%%%%%%%%%%%%%%%%%%%%%%%%%%%%%%%%%%%%%%%%%%
%\newpage
%\subsection{The Relational Model}
%\label{subsec:rel:mod}
\paragraph{Comparisons.}
%%%%%%%%%%%%%%%%%%%%%%%%%%%%%%%%%%%%%%%%%%%%%%%%%%%%%%%%%%%%%%%%%
%%%%%%%%%%%%%%%%%%%%%%%%%%%%%%%%%%%%%%%%%%%%%%%%%%%%%%%%%%%%%%%%%

To a large extent this paper follows 
the relational operators that Codd introduced 
for relational algebra.
In his original paper \cite{codd:70} 
Codd introduced
the eight relational operators:
union, intersection, difference, Cartesian product,
selection (restriction), projection, join, and division.
In his book \;\cite{codd:90}, 
the basic operators
(chap.4) 
are:
selection, 
projection, 
natural join, 
select-join,
union, 
difference,
intersection, and
division.
Various advanced operators (chap.5)
include:
semi-theta-join,
outer-equi-join, and
outer-natural-join.
%{\fbox{\textbf{Review these!}}}
%

This contrasts somewhat to the approach in \texttt{FOLE} (this paper),
which consists of the concise collection of basic components of \S\,\ref{sub:sec:base:ops}
and the large collection of composite operators of 
\S\,\ref{sub:sec:comp:ops:type:dom},\,\ref{sub:sec:comp:ops:sign},\,\ref{sub:sec:non-trad:ops}.
%which are defined in terms of the basic components.
The \emph{basic components} are of three kinds:
reflection (\S\,\ref{sub:sub:sec:reflect})
%
%%%%%%%%%%%%%%%%%%%%%%%%%%%%%%%%%%%%%%%%%%%%%%%%%%%%%%%%%%%%%%%%%%%%%%%%%%%%%%%%
%%%%%%%%%%%%%%%%%%%%%%%%%%%%%%%%%%%%%%%%%%%%%%%%%%%%%%%%%%%%%%%%%%%%%%%%%%%%%%%%
\footnote{Tables and relations are informationally equivalent.
The inclusion of relations into tabular operations can be implicit.
The image relation of a tabular result should be made explicit.} 
%%%%%%%%%%%%%%%%%%%%%%%%%%%%%%%%%%%%%%%%%%%%%%%%%%%%%%%%%%%%%%%%%%%%%%%%%%%%%%%%
%%%%%%%%%%%%%%%%%%%%%%%%%%%%%%%%%%%%%%%%%%%%%%%%%%%%%%%%%%%%%%%%%%%%%%%%%%%%%%%%
%
(image, include),
Booleans or basic operators (\S\,\ref{sub:sub:sec:boole}) 
%(see Tbl.\,\ref{fig:fole:base:ops}) 
(union, intersection, difference), and 
adjoint flow (\S\,\ref{sub:sub:sec:adj:flow})
%\fbox{(see Tbl.\,\ref{tbl:fole:adj:flow})} 
(projection/inflation and expansion/restriction).
The \emph{composite operators} 
%(\S\,\ref{sub:sec:comp:ops:type:dom})
(Tbl.\,\ref{fig:fole:comp:ops}) 
are defined in terms of the basic components.
These consist of three kinds:
%\textbf{limit:} 
limit operators
(\S\,\ref{sub:sec:comp:ops:type:dom}),
%(quotient, core, natural join, semi-join, anti-join, and generic meet);
%\textbf{colimit:} 
colimit operators
(\S\,\ref{sub:sec:comp:ops:sign})
%(co-quotient, co-core, data-type join, data-type semi-join, data-type anti-join, and generic join);
and
%\textbf{unorthodox:} 
unorthodox operators
(\S\,\ref{sub:sec:non-trad:ops}).
%(selection, select-join, filtered join, data-type meet, subtraction, division, and outer-join.)
%
Many more could be defined.
The operations in the \texttt{FOLE} approach
are based upon principles implicit 
in the structure of the mathematical context of tables  $\mathrmbf{Tbl}$,
as presented in detail in the paper \cite{kent:fole:era:tbl}
and summarized in the appendix \S\,\ref{sec:append} below.

%%%%%%%%%%%%%%%%%%%%%%%%%%%%%%%%%%%%%%%%%%%%%%%%%%%%%%%%%%%%%%%%%%
%%%%%%%%%%%%%%%%%%%%%%%%%%%%%%%%%%%%%%%%%%%%%%%%%%%%%%%%%%%%%%%%%%
%\newpage
\section{Overview}
\label{sec:overview}
%\paragraph{Overview.}
%%%%%%%%%%%%%%%%%%%%%%%%%%%%%%%%%%%%%%%%%%%%%%%%%%%%%%%%%%%%%%%%%
%%%%%%%%%%%%%%%%%%%%%%%%%%%%%%%%%%%%%%%%%%%%%%%%%%%%%%%%%%%%%%%%%
%
%This overview 
%reviews Codd's relational operators and 
%their \texttt{FOLE} representation.
%This paper discusses relational operations 
%in the (FOLE) First order Logical Environment.
Previous papers about \texttt{FOLE} have been largely theoretical.
This paper is more applied.
Here we discuss one representation of 
the relational model called relational algebra.
%
%%%%%%%%%%%%%%%%%%%%%%%%%%%%%%%%%%%%%%%%%%%%%%%%%%%%%%%%%%%%%%%%%%%%%%
%%%%%%%%%%%%%%%%%%%%%%%%%%%%%%%%%%%%%%%%%%%%%%%%%%%%%%%%%%%%%%%%%%%%%%
\footnote{Following the original discussion 
of {\ttfamily FOLE} (Kent~\cite{kent:iccs2013})
in the knowledge-representation community, 
we use 
``mathematical context'' for the term ``category'',
``passage'' for the term ``functor'', and
``bridge'' for the term ``natural transformation''.}
%%%%%%%%%%%%%%%%%%%%%%%%%%%%%%%%%%%%%%%%%%%%%%%%%%%%%%%%%%%%%%%%%%%%%%
%%%%%%%%%%%%%%%%%%%%%%%%%%%%%%%%%%%%%%%%%%%%%%%%%%%%%%%%%%%%%%%%%%%%%%
%
Relational algebra involves operations
%, such as projection, selection and join,
that are used to formulate database queries.
To a large extent this paper follows 
the eight relational operators that Codd introduced. 
%in his original paper \cite{codd:70} on relational algebra:
%union, intersection, difference, cartesian product,
%selection (restriction), projection, join, and division.
%\begin{itemize}
%\item 
Here we define relational operations 
%of relational databases
in terms of 
%by using the components of 
a \texttt{FOLE} relational database.
Each \texttt{FOLE} relational database
is built from the basic concept of a \texttt{FOLE} table,
which is discussed in detail in the paper 
\cite{kent:fole:era:tbl}.
%Kent, R.E.:
%``The {\ttfamily FOLE} Table''.
The mathematical context of \texttt{FOLE} tables
is divided into fiber contexts based at each of the tabular components:
%The building blocks of \texttt{FOLE} tables are
signatures,
type domains, and
signed domains. 
Tabular flow is used to move tables between fiber contexts.
In the appendix \S\,\ref{sec:append}
we review some basic concepts of \texttt{FOLE}:
the tabular components in
\S\,\ref{sub:sec:tbl:comp}
of
%\begin{itemize}
%\item 
type domains,
% in
%\S\,\ref{sub:sec:cls}; 
%\item 
signatures,
% in
%\S\,\ref{sub:sec:sign}; 
%\item 
and
signed domains;
% in
%\S\,\ref{sub:sec:sign:dom};  
and
%\item 
tables, tabular flow and relations in
\S\,\ref{sub:sec:tbl:rel}.
%,
%and
%%\item 
%tabular flow in
%\S\,\ref{sub:sub:sec:tbl:flow:adj}.
%\end{itemize}
%
The paper 
\cite{kent:fole:era:tbl}
%Kent, R.E.:
%``The {\ttfamily FOLE} Table''.
has a more complete presentation of these concepts.

%\item 
This paper
%\S\,\ref{sec:rel:alg}
offers a framework in which to define the operators of relational algebra.
To a certain extent
we follow category-theoretic guidelines and spirit.
The operators of relational algebra are split into two groups: 
basic 
%(Fig.\;\ref{fig:fole:base:ops})
and composite.
%(Fig.\;\ref{fig:fole:comp:ops}).
Basic operations (\S\,\ref{sub:sub:sec:boole}) 
are defined at the small scope of a fixed signed domain.
Most 
composite operators 
%(Tbl.\,\ref{fig:fole:comp:ops})
are defined at an intermediate scope, 
either for a fixed type domain or for a fixed signature.
%
%%%%%%%%%%%%%%%%%%%%%%%%%%%%%%%%%%%%%%%%%%%%%%%%%%%%%%%%%%%%%%%%%%%%%%%%%%%%%%%%
%%%%%%%%%%%%%%%%%%%%%%%%%%%%%%%%%%%%%%%%%%%%%%%%%%%%%%%%%%%%%%%%%%%%%%%%%%%%%%%%
\footnote{The generic composite operators are defined at the large scope of all tables.}
%%%%%%%%%%%%%%%%%%%%%%%%%%%%%%%%%%%%%%%%%%%%%%%%%%%%%%%%%%%%%%%%%%%%%%%%%%%%%%%%
%%%%%%%%%%%%%%%%%%%%%%%%%%%%%%%%%%%%%%%%%%%%%%%%%%%%%%%%%%%%%%%%%%%%%%%%%%%%%%%%
%
At a fixed type domain,
composite operators are defined in terms of 
the basic operators (Booleans) and 
adjoint flow 
along signature morphisms
(\S\,\ref{sub:sub:sec:adj:flow:A}).
At a fixed signature,
composite operators are defined in terms of 
the basic operators (Booleans) and 
adjoint flow 
along type domain morphisms
(\S\,\ref{sub:sub:sec:adj:flow:S}).
%
%As the name implies,
%each composite operation 
%is the composition of several basic operations,
%and is exemplified by a flow chart.
Composite operators are of two kinds,
either orthodox or unorthodox.
Orthodox composite operators 
follow category-theoretic guidelines,
corresponding to 
either the computation of 
limits (\S\,\ref{sub:sec:comp:ops:type:dom}) 
or the computation of 
colimits (\S\,\ref{sub:sec:comp:ops:sign}).
Unorthodox operations (\S\,\ref{sub:sec:non-trad:ops}), 
for various reasons, do not.
A flowchart is used to visualize the definition of each composite operator.
The links between flowchart symbols are typed by 
either 
a signature (when operating in the context of a fixed type domain),
a type domain (when operating in the context of a fixed signature),
or
a signed domain (when operating in the full table context).
%\begin{description}
%
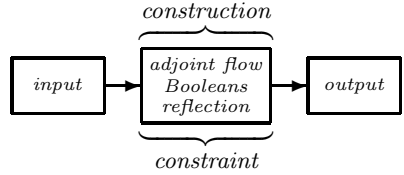
\begin{figure}
\begin{center}
{{\begin{tabular}{c}
%%%%%%%%%%%%%%%%%%%%%%%%%%%%%%%%%%%%%%%%%%%%%%%%%%%%%%%%%%%%%%%%%%%%%%%%%%%%%%%%
\begin{picture}(160,70)(-25,5)
\setlength{\unitlength}{0.35pt}
\put(-70,80){\begin{picture}(0,0)(0,0)
\thicklines
\put(70,40){\makebox(0,0){\scriptsize{{\textit{{input}}}}}}
\put(20,70){\line(1,0){100}}
\put(20,10){\line(1,0){100}}
\put(20,70){\line(0,-1){60}}
\put(120,70){\line(0,-1){60}}
\put(120,40){\vector(1,0){40}}
\end{picture}}
\put(90,80){\begin{picture}(0,0)(0,0)
\thicklines
\put(70,62){\makebox(0,0){\scriptsize{{\textit{{adjoint flow}}}}}}
\put(70,40){\makebox(0,0){\scriptsize{{\textit{{Booleans}}}}}}
\put(70,17){\makebox(0,0){\scriptsize{{\textit{{reflection}}}}}}
\put(0,80){\line(1,0){140}}
\put(0,0){\line(1,0){140}}
\put(0,80){\line(0,-1){80}}
\put(140,80){\line(0,-1){80}}
\put(70,110){\makebox(0,0){\normalsize{${
\overset{\textstyle{{\footnotesize{{\textit{{construction}}}}}}}
{\overbrace{\hspace{52pt}}}}$}}}
\put(70,-30){\makebox(0,0){\normalsize{${
\underset{\textstyle{{\footnotesize{{\textit{{constraint}}}}}}}
{\underbrace{\hspace{52pt}}}}$}}}
\end{picture}}
\put(250,80){\begin{picture}(0,0)(0,0)
\thicklines
\put(70,40){\makebox(0,0){\scriptsize{{\textit{{output}}}}}}
\put(20,70){\line(1,0){100}}
\put(20,10){\line(1,0){100}}
\put(20,70){\line(0,-1){60}}
\put(120,70){\line(0,-1){60}}
\put(-20,40){\vector(1,0){40}}
\end{picture}}
\end{picture}
%%%%%%%%%%%%%%%%%%%%%%%%%%%%%%%%%%%%%%%%%%%%%%%%%%%%%%%%%%%%%%%%%%%%%%%%%%%%%%%%
\\
%%%%%%%%%%%%%%%%%%%%%%%%%%%%%%%%%%%%%%%%%%%%%%%%%%%%%%%%%%%%%%%%%%%%%%%%%%%%%%%%
{\scriptsize{\begin{tabular}{p{160pt}}
The input obeys the constraint.
The constraint specifies the construction using 
limits or colimits.
%colimits of signatures and limits of type domains.
The construction provides a framework for building the output
%The output is constructed along the constraint 
using adjoint flow, Boolean operators and reflection.
\end{tabular}}}
%%%%%%%%%%%%%%%%%%%%%%%%%%%%%%%%%%%%%%%%%%%%%%%%%%%%%%%%%%%%%%%%%%%%%%%%%%%%%%%%
\end{tabular}}}
\end{center}
\caption{\texttt{FOLE} Composite Relational Operations (Specification/Evaluation)}
\label{tbl:fole:rel:ops:spec:eval}
\end{figure}

Fig.\;\ref{tbl:fole:rel:ops:spec:eval}
%gives a compact overview that 
illustrates the 
%structure
specification and evaluation
of the composite operations defined in this paper.
Many of these
%the composite operations
need to use only a sufficient or adequate collection of tables
(Def.\;\ref{def:suff:adequ:lim}) and
(Def.\;\ref{def:suff:adequ:colim})
for their input;
in particular,
quotient and co-quotient,
natural join and data-type join,
and
generic meet and generic join.
%
%\begin{itemize}
%\comment{
%%\item 
%\fbox{Fig.\;\ref{fig:fole:base:ops}}
%is a list of the three \texttt{FOLE} basic operations (Booleans),
%consisting of meet, join and difference.
%These correspond to the logical operations of 
%conjunction (and), disjunction (or) and negation (not).
%%\item 
%\fbox{Fig.\;\ref{tbl:fole:adj:flow}}
%is a depiction of the \texttt{FOLE} adjoint flow operators:
%for a fixed type domain,
%projection is left adjoint to inflation; and
%for a fixed signature,
%expansion is left adjoint to restriction.
%}
%\item 
The notation and description of the composite operators is 
given
%illustrated 
in Tbl.\,\ref{fig:fole:comp:ops}.
%Tbl.\,\ref{fig:fole:comp:ops}
%is a list of the composite operations,
These consist of limit, colimit and unorthodox varieties.
The limit operations 
%(symbols with squarish shapes) 
are applied versions of 
the theoretical limit operation in various scopes.
The dual holds for the colimit operations.
%
%%%%%%%%%%%%%%%%%%%%%%%%%%%%%%%%%%%%%%%%%%%%%%%%%%%%%%%%%%%%%%%%%%%%%%%%%%%%%%%%%%%%%%%%%%%%%%%%%%%%%%%%%%%%%%%%%%%%%%%%%%%%%%%%%%%%%%%%%%%%%%%%%%%%%%%%%%%%%%
\footnote{The limit operation symbols have squarish shapes,
the colimit operation symbols have roundish shapes.}
%%%%%%%%%%%%%%%%%%%%%%%%%%%%%%%%%%%%%%%%%%%%%%%%%%%%%%%%%%%%%%%%%%%%%%%%%%%%%%%%%%%%%%%%%%%%%%%%%%%%%%%%%%%%%%%%%%%%%%%%%%%%%%%%%%%%%%%%%%%%%%%%%%%%%%%%%%%%%%%
%
%%%%%%%%%%%%%%%%%%%%%%%%%%%%%%%%%%%%%%%%%%%%%%%%%%%%%%%%%%%%%%%%%%%%
%%%%%%%%%%%%%%%%%%%%%%%%%%%%%%%%%%%%%%%%%%%%%%%%%%%%%%%%%%%%%%%%%%%%
%\footnote{
To a certain extent 
the limit operations follow the spirit of classical relational database operations
by controlling the database through the header.
The colimit operations allow you to get control of your data at the query level
by manipulating the data-types.
%}
%%%%%%%%%%%%%%%%%%%%%%%%%%%%%%%%%%%%%%%%%%%%%%%%%%%%%%%%%%%%%%%%%%%%
%%%%%%%%%%%%%%%%%%%%%%%%%%%%%%%%%%%%%%%%%%%%%%%%%%%%%%%%%%%%%%%%%%%%
%
Unorthodox operations do not follow these prescriptions for various reasons;
either the definition is deviant 
%
%%%%%%%%%%%%%%%%%%%%%%%%%%%%%%%%%%%%%%%%%%%%%%%%%%%%%%%%%%%%%%%%%%%%%%%%%%%%%%%%%%%%%%%%%%%%%%%%%%%%%%%%%%%%%%%%%%%%%%%%%%%%%%%%%%%%%%%%%%%%%%%%%%%%%%%%%%%%%%
\footnote{See comments in 
\S\,\ref{sub:sub:sec:filtered:join} and \S\,\ref{sub:sub:sec:boolean:meet}
on the deviancy of filtered join and data-type meet.}
%%%%%%%%%%%%%%%%%%%%%%%%%%%%%%%%%%%%%%%%%%%%%%%%%%%%%%%%%%%%%%%%%%%%%%%%%%%%%%%%%%%%%%%%%%%%%%%%%%%%%%%%%%%%%%%%%%%%%%%%%%%%%%%%%%%%%%%%%%%%%%%%%%%%%%%%%%%%%%%
or more complex.
%
%Both the limit and the colimit categories consist of basic and auxiliary parts.
Each limit operation 
has a dual,
which is a colimit operation. 
For example,
natural join is a limit operation,
whose dual is the data-type join,
which is a colimit operation.
In addition,
in the unorthodox category,
the filtered join is dual to the data-type meet.
%\end{itemize}
%\newpage

%\noindent{\fbox{\textbf{Review To Here: 02-03-2021}}}\newline
%\begin{description}

%
%%%%%%%%%%%%%%%%%%%%%%%%%%%%%%%%%%%%%%%%%%%%%%%%%%%%%%%%%%%%%%%%%%%%%%%%%%%%%%%%
\newpage
\paragraph{Limit Operations.}
%%%%%%%%%%%%%%%%%%%%%%%%%%%%%%%%%%%%%%%%%%%%%%%%%%%%%%%%%%%%%%%%%%%%%%%%%%%%%%%%
%%%%%%%%%%%%%%%%%%%%%%%%%%%%%%%%%%%%%%%%%%%%%%%%%%%%%%%%%%%%%%%%%
%\item[Limits.]
%in $\mathrmbf{Tbl}(\mathcal{A})$.] 
%\paragraph{Composite operators in $\mathrmbf{Tbl}(\mathcal{A})$.}
%%%%%%%%%%%%%%%%%%%%%%%%%%%%%%%%%%%%%%%%%%%%%%%%%%%%%%%%%%%%%%%%%
In
\S\,\ref{sub:sec:comp:ops:type:dom},
we define (Tbl.\,\ref{tbl:fole:comp:rel:ops:lim}) the composite operations for limits.
%In Tbl.\,\ref{tbl:fole:comp:rel:ops:lim}we define the composite operations.
%and in Tbl.\,\ref{tbl:fole:set:analogs}
%we give set-theoretic analogies to these operations.
%Tbl.\,\ref{tbl:op:flo:chrt:basic:comp:lim} 
%lists the basic components needed for each composite.
%\begin{itemize}
%\item 
%\S\,\ref{sub:sub:sec:quotient} covers quotient.
Quotient (\S\,\ref{sub:sub:sec:quotient}) represents equalizer.
As the flowchart in Fig.\;\ref{fig:fole:quotient:flo:chrt} indicates,
quotient is a unary operator taking one argument,
one table whose signature is the target of a parallel pair of 
$X$-signature morphisms; 
it is defined by the single operation of inflation.
%\item 
\S\,\ref{sub:sub:sec:core} covers core.
As the flowchart in Fig.\;\ref{fig:fole:core:flo:chrt} indicates,
core is a binary operator taking two arguments,
two tables whose type domains are linked to a third type domain
through an opspan of $X$-type domain morphisms;
it is the composition of restriction (twice) followed by meet.
%\item 
%\S\,\ref{sub:sub:sec:nat:join}
%covers natural join.
%, selection and select-join.
Natural join (\S\,\ref{sub:sub:sec:nat:join}) represents pullback.
As the flowchart in Fig.\;\ref{fig:fole:nat:join:flo:chrt} indicates,
natural join is a binary operator taking two arguments,
two tables whose signatures are linked to a third signature
through a span of $X$-signature morphisms;
it is the composition of inflation (twice) followed by meet.
Cartesian product is a special case.
%\item 
%\newpage
\S\,\ref{sub:sub:sec:semi:join} covers semi-join.
There are two semi-joins,
left or right.
As the flowchart in Fig.\;\ref{fole:semi:join:flo:chrt} indicates,
semi-join takes two arguments,
two tables constrained as in the natural join;
it is the composition of natural join followed by projection.
%\item 
\S\,\ref{sub:sub:sec:anti:join} covers anti-join.
There are two anti-joins,
left or right.
As the flowchart in Fig.\;\ref{fole:anti:join:flo:chrt} indicates,
anti-join takes two arguments,
two tables constrained as in the natural join;
it is the composition of semi-join followed by the difference
with one of the two arguments.
%\item 
%\S\,\ref{sub:sub:sec:generic:meet}
%covers generic meet.
Generic meet (\S\,\ref{sub:sub:sec:generic:meet}) represents limit.
As the flowchart in Fig.\;\ref{fig:fole:boolean:meet:flo:chrt} indicates,
generic meet takes 
a sufficient indexed collection of 
%$n$ tables,
$n$ tables 
%linked through a sufficient sub-diagram 
(Def.\;\ref{def:suff:adequ:lim});
it is the composition of 
%the composite 
restriction$\,\circ\,$inflation ($n$ times) followed by meet.
%\end{itemize}
%\end{itemize}
%\end{itemize}
%

%%%%%%%%%%%%%%%%%%%%%%%%%%%%%%%%%%%%%%%%%%%%%%%%%%%%%%%%%%%%%%%%%%%%%%%%%%%%%%%%
%\mbox{}\newpage
\paragraph{Colimit Operations.}
%%%%%%%%%%%%%%%%%%%%%%%%%%%%%%%%%%%%%%%%%%%%%%%%%%%%%%%%%%%%%%%%%%%%%%%%%%%%%%%%

%%%%%%%%%%%%%%%%%%%%%%%%%%%%%%%%%%%%%%%%%%%%%%%%%%%%%%%%%%%%%%%%%
%\mbox{}\newline
%\item[Colimits.]
%in $\mathrmbf{Tbl}(\mathcal{S})$.]
%\paragraph{Composite operators in $\mathrmbf{Tbl}(\mathcal{S})$.}
%%%%%%%%%%%%%%%%%%%%%%%%%%%%%%%%%%%%%%%%%%%%%%%%%%%%%%%%%%%%%%%%%
In
\S\,\ref{sub:sec:comp:ops:sign},
we define (Tbl.\,\ref{tbl:fole:comp:rel:ops:colim}) 
the composite operations for colimits.
%In Tbl.\,\ref{tbl:fole:comp:rel:ops:colim}we define the composite operations.
%and in {\fbox{Tbl.\;{???}}}
%we give set-theoretic analogies to these operations.
%Tbl.\,\ref{tbl:op:flo:chrt:basic:comp:colim} 
%lists the basic components needed for each composite.
%\begin{itemize}
%\item 
%\S\,\ref{sub:sub:sec:co-quotient} covers co-quotient.
Co-quotient (\S\,\ref{sub:sub:sec:co-quotient}) represents co-equalizer.
As the flowchart in Fig.\;\ref{fig:fole:co-quotient:flo:chrt} indicates,
co-quotient is a unary operator taking one argument,
one table whose type domain is the source of a parallel pair of 
$X$-type domain morphisms; 
it is defined by the single operation of expansion.
%\item 
\S\,\ref{sub:sub:sec:co-core} covers co-core.
As the flowchart in Fig.\;\ref{fig:fole:co-core:flo:chrt} indicates,
co-core is a binary operator taking two arguments,
two tables 
whose signatures are linked to a third signature
through a span of $X$-signature morphisms;
it is the composition of projection (twice) followed by join.
%\item 
%\S\,\ref{sub:sub:sec:boole:join} 
%covers data-type join.
Data-type join (\S\,\ref{sub:sub:sec:boole:join}) represents pushout.
As the flowchart in Fig.\;\ref{fig:fole:boole:join:flo:chrt} indicates,
data-type join is a binary operator taking two arguments,
two tables whose type domains are linked to a third type domain
through an op-span of $X$-type domain morphisms:
it is the composition of expansion (twice) followed by join.
Disjoint sum is a special case.
%\item 
\S\,\ref{sub:sub:sec:boole:semi:join}
covers data-type semi-join.
There are two semi-joins,
left or right.
As the flowchart in Fig.\;\ref{fole:boole:semi:join:flo:chrt} indicates,
data-type semi-join takes two arguments,
two tables constrained as in the data-type join;
it is the composition of data-type join followed by restriction.
%\item 
\S\,\ref{sub:sub:sec:boole:anti:join}
covers data-type anti-join.
There are two anti-joins,
left or right.
As the flowchart in Fig.\;\ref{fole:boole:anti:join:flo:chrt} indicates,
anti-join takes two arguments,
two tables constrained as in the data-type join;
it is the composition of data-type semi-join followed by the difference
with one of the two arguments.
%\item 
%\S\,\ref{sub:sub:sec:generic:join}
%covers generic join.
Generic join (\S\,\ref{sub:sub:sec:generic:join}) represents colimit.
As the flowchart in Fig.\;\ref{fig:fole:boolean:join:flo:chrt} indicates,
generic join takes 
a sufficient indexed collection of 
%$n$ tables,
$n$ tables 
%linked through a sufficient sub-diagram 
(Def.\;\ref{def:suff:adequ:colim});
it is the composition of 
%the composite 
projection$\,\circ\,$expansion ($n$ times) followed by join.
%\end{itemize}
%\mbox{}
%\newline

%%%%%%%%%%%%%%%%%%%%%%%%%%%%%%%%%%%%%%%%%%%%%%%%%%%%%%%%%%%%%%%%%%%%%%%%%%%%%%%%
%\mbox{}\newpage
\paragraph{Unorthodox Operations.}
%%%%%%%%%%%%%%%%%%%%%%%%%%%%%%%%%%%%%%%%%%%%%%%%%%%%%%%%%%%%%%%%%%%%%%%%%%%%%%%%
%%%%%%%%%%%%%%%%%%%%%%%%%%%%%%%%%%%%%%%%%%%%%%%%%%%%%%%%%%%%%%%%%
%\newpage
%\item[Unorthodox.]
%%%%%%%%%%%%%%%%%%%%%%%%%%%%%%%%%%%%%%%%%%%%%%%%%%%%%%%%%%%%%%%%%
In
\S\,\ref{sub:sec:non-trad:ops},
we define (Tbl.\;\ref{tbl:fole:comp:rel:ops:unortho})
the unorthodox composite operations.
%\begin{itemize}
%\item 
\S\,\ref{sub:sub:sec:sel}
covers selection.
As the flowchart in Fig.\;\ref{fig:fole:select:aux:rel:flo:chrt} indicates,
selection takes two arguments,
a principal table and an auxiliary relation 
selected against;
%
%%%%%%%%%%%%%%%%%%%%%%%%%%%%%%%%%%%%%%%%%%%%%%%%%%%%%%%%%%%%%%%%%%%%%%%%%%%%%%%%
%%%%%%%%%%%%%%%%%%%%%%%%%%%%%%%%%%%%%%%%%%%%%%%%%%%%%%%%%%%%%%%%%%%%%%%%%%%%%%%%
\footnote{In \S\,\ref{sub:sec:tbl:rel}
we give several examples of common auxiliary relations.}
%%%%%%%%%%%%%%%%%%%%%%%%%%%%%%%%%%%%%%%%%%%%%%%%%%%%%%%%%%%%%%%%%%%%%%%%%%%%%%%%
%%%%%%%%%%%%%%%%%%%%%%%%%%%%%%%%%%%%%%%%%%%%%%%%%%%%%%%%%%%%%%%%%%%%%%%%%%%%%%%%
%
it is the composition of inflation (twice) followed by meet.
Selection is a special case of natural join.
%\item 
\S\,\ref{sub:sub:sec:sel:join} 
covers select-join.
As the flowchart in Fig.\;\ref{fig:fole:theta:join:flo:chrt} indicates,
select-join takes three arguments,
two principal tables that are joined,
and an auxiliary relation selected against;
it is the composition of natural join followed by selection.
Select-join is a natural multi-join.
%\item 
\S\,\ref{sub:sub:sec:filtered:join}
covers filtered join.
As the flowchart in Fig.\;\ref{fig:fole:filter:join:flo:chrt} indicates,
filtered join takes two arguments,
two tables whose type domains are linked to a third type domain
through an op-span of $X$-type domain morphisms;
it is the composition of restriction (twice) followed by join.
%\item 
\S\,\ref{sub:sub:sec:boolean:meet}
covers data-type meet.
As the flowchart in Fig.\;\ref{fig:fole:boole:meet:flo:chrt} indicates,
data-type meet takes two arguments,
two tables constrained as in the data-type join;
it is the composition of expansion (twice) followed by meet.
%\item 
\S\,\ref{sub:sub:sec:subtrac}
covers subtraction.
As the flowchart in Fig.\;\ref{fole:subtraction:flo:chrt} indicates,
subtraction takes two arguments,
two tables with one subtracted from the other.
There are two type domains,
with one table having the product type domain,
and the other table having one of the component type domains.
It is the composition of expansion (once) followed by difference.
%\item 
\S\,\ref{sub:sub:sec:div} covers division.
As the flowchart in Fig.\;\ref{fole:division:flo:chrt} indicates,
division takes two arguments,
two tables with one divided by the other.
There are two signatures,
with one table having the coproduct signature,
and the other table having one of the component signatures.
Division uses projection (twice), Cartesian product and difference (twice).
%\mbox{}\newline{\fbox{\textbf{To Here!}}}
%\item 
\S\,\ref{sub:sub:sec:out:join} covers outer-join.
There are two outer-joins,
left or right.
As the flowchart 
in Fig.\;\ref{fole:outer:join:in:square} 
%and Fig.\;\ref{fole:outer:join:in:square}
indicates,
left outer-join takes two arguments,
two tables with one outer-joined with the other.
For outer-join,
we expand the type domain by adding a null value.
%In addition,
%to define the outer join, 
%we need to expand the type domain 
%by adding a null value.
We then 
add a third table 
%with the expanded type domain
that consists of a single tuple of null values.
Outer-join uses natural join, anti-join, expansion, Cartesian product and union.
%As the equivalent flowchart in Fig.\;\ref{fole:outer:join:in:square} indicates,
Since outer-join uses two signatures and two type domains with a common sort set,
outer-join is an interesting and non-trivial case of adjoint flow 
in the square (\S\,\ref{sub:sub:sec:adj:flo:square}).
%\end{itemize}
%\end{description}
%

%
\begin{table}
\begin{center}
{{\begin{tabular}
{|
l@{\hspace{5pt}}
r@{\hspace{3pt}{$\circ$}\hspace{3pt}}
l@{\hspace{8pt}{$\doteq$}\hspace{6pt}}
c@{\hspace{6pt}}
c@{\hspace{6pt}}
c@{\hspace{6pt}}
l|}
%%%%%%%%%%%%%%%%%%%%%%%%%%%%%%%%%%%%%%%%%%%%%%%%%%%%%%%%%%%%%%%%%%%%%%
\multicolumn{1}{l}{}
& \multicolumn{2}{c}{\textsf{definition}}
& \multicolumn{1}{c}{\textsf{name}} 
& \multicolumn{1}{c}{\textsf{symbol}}
& \multicolumn{1}{c}{\textsf{arity}} 
& \multicolumn{1}{c}{\textsf{scope}}
\\
%%%%%%%%%%%%%%%%%%%%%%%%%%%%%%%%%%%%%%%%%%%%%%%%%%%%%%%%%%%%%%%%%%%%%%
\multicolumn{7}{l}{\S\,\ref{sub:sec:comp:ops:type:dom} \textbf{Limit}} 
%%%%%%%%%%%%%%%%%%%%%%%%%%%%%%%%%%%%%%%%%%%%%%%%%%%%%%%%%%%%%%%%%%%%%%
\\\hline
\S\,\ref{sub:sub:sec:quotient}
& \multicolumn{1}{r}{\textit{inflate}} & \multicolumn{1}{r}{$\doteq\hspace{7pt}$} 
&	\textit{quotient}
& $\Yright$
& unary
& $\mathrmbf{Tbl}(\mathcal{A})$ 
\\
\S\,\ref{sub:sub:sec:core}
& \textit{restrict}$\times${2} & \textit{meet} 
&	\textit{core} 
& $\sqcap$
& binary
& $\mathrmbf{Tbl}(\mathcal{S})$ 
\\
\S\,\ref{sub:sub:sec:nat:join}
& \textit{inflate}$\times${2} & \textit{meet} 
&	\textit{natural join} 
& $\boxtimes$
& binary
& $\mathrmbf{Tbl}(\mathcal{A})$ 
\\\hline\hline
\S\,\ref{sub:sub:sec:semi:join}
& \textit{natural join} & \textit{project} 
& \textit{semi-join}
& $\boxleft$ $\boxright$
& binary
& $\mathrmbf{Tbl}(\mathcal{A})$ 
\\
\S\,\ref{sub:sub:sec:anti:join}
& (\textit{id},\textit{semi-join}) & \textit{diff} 
& \textit{anti-join}
& $\boxslash$ $\boxbackslash$
& binary
& $\mathrmbf{Tbl}(\mathcal{A})$ 
\\\hline\hline
\S\,\ref{sub:sub:sec:generic:meet}
& (\textit{restrict}$\,\circ\,$\textit{inflate})$\times${$n$} & \textit{meet} 
&	\textit{generic meet}
& \scriptsize{$\prod$}
& $n$-ary
& $\mathrmbf{Tbl}$ 
\\\hline
\multicolumn{7}{l}{} 
\\
%%%%%%%%%%%%%%%%%%%%%%%%%%%%%%%%%%%%%%%%%%%%%%%%%%%%%%%%%%%%%%%%%%%%%%
\multicolumn{7}{l}{\textbf{Colimit}} 
%%%%%%%%%%%%%%%%%%%%%%%%%%%%%%%%%%%%%%%%%%%%%%%%%%%%%%%%%%%%%%%%%%%%%%
\\\hline
\S\,\ref{sub:sub:sec:co-quotient}
& \multicolumn{1}{r}{\textit{expand}} & \multicolumn{1}{r}{$\doteq\hspace{7pt}$} 
&	\textit{co-quotient}
& $\Yleft$
& unary
& $\mathrmbf{Tbl}(\mathcal{S})$ 
\\
\S\,\ref{sub:sub:sec:co-core}
& \textit{project}$\times${2} & \textit{join}		
&	\textit{co-core}
& $\cup$
& binary
& $\mathrmbf{Tbl}(\mathcal{A})$ 
\\ 
\S\,\ref{sub:sub:sec:boole:join}
& \textit{expand}$\times$2 & \textit{join}	
&	\textit{data-type} \textit{join}
& $\oplus$
& binary
& $\mathrmbf{Tbl}(\mathcal{S})$ 
\\\hline\hline
\S\,\ref{sub:sub:sec:boole:semi:join}
& \textit{data-type} \textit{join} & \textit{restrict} 
& \textit{data-type} \textit{semi-join}
& $\oleft$ $\oright$
& binary
& $\mathrmbf{Tbl}(\mathcal{S})$ 
\\
\S\,\ref{sub:sub:sec:boole:anti:join}
& (\textit{id},\textit{semi-join}) & \textit{diff} 
& \textit{data-type} \textit{anti-join}
& $\oslash$ $\obackslash$
& binary
& $\mathrmbf{Tbl}(\mathcal{S})$ 
\\\hline\hline
\S\,\ref{sub:sub:sec:generic:join}
& 
%project-expand 
(\textit{project}$\,\circ\,$\textit{expand})$\times${$n$}
& \textit{join} 
&	\textit{generic join}
& \scriptsize{$\coprod$}
& $n$-ary
& $\mathrmbf{Tbl}$ 
\\\hline
\multicolumn{7}{l}{}
\\
%%%%%%%%%%%%%%%%%%%%%%%%%%%%%%%%%%%%%%%%%%%%%%%%%%%%%%%%%%%%%%%%%%%%%%
\multicolumn{7}{l}{\textbf{Unorthodox}} 
%%%%%%%%%%%%%%%%%%%%%%%%%%%%%%%%%%%%%%%%%%%%%%%%%%%%%%%%%%%%%%%%%%%%%%
\\\hline
\S\,\ref{sub:sub:sec:sel}
& 
(\textit{inflate},\textit{id}) & \textit{meet} 
& \textit{selection}
& $\sigma$
& binary
& $\mathrmbf{Tbl}(\mathcal{A})$ 
\\
\S\,\ref{sub:sub:sec:sel:join}
& 
(\textit{id},\;\textit{natural join}) & \textit{select} 
& \textit{select-join}
& $\hat{\boxtimes}$
& ternary
& $\mathrmbf{Tbl}(\mathcal{A})$ 
\\\hline\hline
\S\,\ref{sub:sub:sec:filtered:join}
& \textit{restrict}$\times$2 & \textit{join}	
& \textit{filtered join}
& 
%{\tiny{\textbf{M}}} 
$\varominus$
& binary
& $\mathrmbf{Tbl}(\mathcal{S})$ 
%\\
%\S???
%& \textit{project}$\times$2 & \textit{meet} 
%&	\textit{axis}
%& ???
%& binary
%& $\mathrmbf{Tbl}(\mathcal{A})$ 
\\
\S\,\ref{sub:sub:sec:boolean:meet}
& \textit{expand}$\times$2 & \textit{meet} 
&	\textit{data-type meet}
& $\boxbar$
& binary
& $\mathrmbf{Tbl}(\mathcal{S})$ 
\\\hline\hline
%%%%%%%%%%%%%%%%%%%%%%%%%%%%%%%%%%%%%%%%%%%%%%%%%%%%%%%%%%%%%%%%%%%%%%%
\S\,\ref{sub:sub:sec:subtrac}
& (\textit{id},\textit{expand}) & \textit{diff} 
& \textit{subtraction}
& $\thicksim$ 
& binary
& $\mathrmbf{Tbl}(\mathcal{S})$ 
\\\hline\hline
\S\,\ref{sub:sub:sec:div}
& \multicolumn{2}{c}{\textit{multiple}}
& \textit{division}
& $\div$
& binary
& $\mathrmbf{Tbl}(\mathcal{A})$ 
\\
\S\,\ref{sub:sub:sec:out:join}
& \multicolumn{2}{c}{\textit{multiple}}
& \textit{outer-join}
& $\rgroup\!\boxtimes$ $\boxtimes\!\lgroup$ 
& binary
& $\mathrmbf{Tbl}$
\\\hline
\end{tabular}}}
\end{center}
\caption{\texttt{FOLE} Composite Operations}
\label{fig:fole:comp:ops}
\end{table}
%

%%%%%%%%%%%%%%%%%%%%%%%%%%%%%%%%%%%%%%%%%%%%%%%%%%%%%%%%%%%%%%%%%%%%%%%%%%%%%%%%
%%%%%%%%%%%%%%%%%%%%%%%%%%%%%%%%%%%%%%%%%%%%%%%%%%%%%%%%%%%%%%%%%%%%%%%%%%%%%%%%
\comment{ % unecessary notation
\begin{figure}
\begin{center}
{\scriptsize{\begin{tabular}
{|@{\hspace{7pt}}c@{\hspace{15pt}}c@{\hspace{15pt}}c@{\hspace{2pt}}|}
%%%%%%%%%%%%%%%%%%%%%%%%%%%%%%%%%%%%%%%%%%%%%%%%%%%%%%%%%%%%%%%%%%%%%%
\hline
\multicolumn{1}{|c}{\textsf{name}} 
& \multicolumn{1}{c}{\textsf{definition}}
& \multicolumn{1}{r|}{\textsf{scope}}
\\\hline
\textit{core} 
&
$\mathcal{T}_{1}\sqcap\mathcal{T}_{2}
=
(\grave{\wr}_{g}(\mathcal{T}_{1})\wedge\grave{\wr}_{g}(\mathcal{T}_{2}))$
&
$\mathrmbf{Tbl}(\mathcal{S})$
\\
\textit{natural join}
&
$\mathcal{T}_{1})\boxtimes\mathcal{T}_{2}
=
(\grave{\wr}_{h}(\mathcal{T}_{1})\wedge\grave{\wr}_{h}(\mathcal{T}_{2}))$
&
$\mathrmbf{Tbl}(\mathcal{A})$
\\
\textit{semi-join}
&
$\mathcal{T}_{1}\boxright\mathcal{T}_{2}
=
\grave{\wr}_{h}(\mathcal{T}_{1}\boxtimes\mathcal{T}_{2})$
& 
$\mathrmbf{Tbl}(\mathcal{A})$ 
\\
\textit{anti-join}
&
$\mathcal{T}_{1}\boxslash\mathcal{T}_{2}
=
\mathcal{T}_{1}{-}(\mathcal{T}_{1}\boxright\mathcal{T}_{2})$
& 
$\mathrmbf{Tbl}(\mathcal{A})$ 
\\
\textit{generic meet}
&
$\prod\mathrmbfit{T}
=
\bigwedge_{i \in I}
\grave{\wr}_{({\hat{h}_{i},\hat{f}_{i},\hat{g}_{i}})}(\mathcal{T}_{i})
$
& 
$\mathrmbf{Tbl}$ 
\\\hline
\end{tabular}}}
$
\prod\mathrmbfit{T} =
$
\end{center}
\caption{\texttt{FOLE} Operations Functional Notation}
\label{fig:fole:ops:fn:notation}
\end{figure}} % unecessary notation
\begin{aside}
For logical interpretation in \texttt{FOLE}, 
the \emph{domain of discourse}
is the context of tables,
with tuples representing individuals
and tables representing propositions.
Interpretation is defined in terms of propositional and predicate logic.
Propositional logic uses 
conjunction, 
%(\emph{and}), 
disjunction 
%(\emph{or}) 
and negation.
% (\emph{not}). 
%\begin{description}
%\item[conjunction:] 
Conjunction (\emph{and})
is represented by the meet at various scopes:
the small scope of a signed domain 
(\S\,\ref{sub:sub:sec:boole} intersection $\mathrmbfit{T}_{1}{\wedge}\mathrmbfit{T}_{2}$),
the intermediate scope of a type domain 
(\S\,\ref{sub:sub:sec:nat:join} natural join $\mathrmbfit{T}_{1}{\,\boxtimes\,}\mathrmbfit{T}_{2}$), and
the large scope of all tables 
(\S\,\ref{sub:sub:sec:generic:meet} generic meet $\prod\mathrmbfit{T}$).
%\newline
%\item[disjunction:] 
Disjunction (\emph{or})
is represented by the join at various scopes:
the small scope of a signed domain 
(\S\,\ref{sub:sub:sec:boole} union $\mathrmbfit{T}_{1}{\vee}\mathrmbfit{T}_{2}$),
the intermediate scope of a signature 
(\S\,\ref{sub:sub:sec:boole:join} data-type join $\mathrmbfit{T}_{1}{\,\oplus\,}\mathrmbfit{T}_{2}$), and
the large scope of all tables 
(\S\,\ref{sub:sub:sec:generic:join} generic join $\coprod\mathrmbfit{T}$).
%\newline
%\item[difference:] 
Negation (\emph{not})
is represented by the difference at various scopes:
the small scope of a signed domain 
(\S\,\ref{sub:sub:sec:boole} difference $\mathrmbfit{T}_{1}{\mathbf{-}}\mathrmbfit{T}_{2}$), and
the intermediate scope of a signature 
(\S\,\ref{sub:sub:sec:subtrac} subtraction $\mathrmbfit{T}_{1}{\,\thicksim\,}\mathrmbfit{T}_{2}$).
%\end{description}
Predicate logic adds the flow of tables to logical interpretation:
%\begin{itemize}
%\item 
(\S\,\ref{sub:sub:sec:adj:flow:A}
projection/inflation
{\footnotesize{${{\bigl\langle{\acute{\mathrmbfit{tbl}}_{\mathcal{A}}(h)
{\;\dashv\;}
\grave{\mathrmbfit{tbl}}_{\mathcal{A}}(h)}\bigr\rangle}}$}})
%at the intermediate scope of a type domain
between the intermediate scope of signatures,
%along an 
%$X$-sorted signature morphism
%$\mathcal{S}'\xrightarrow{\;h\;}\mathcal{S}$, and
%\item 
(\S\,\ref{sub:sub:sec:adj:flow:S}
expansion/restriction
{\footnotesize{${{\bigl\langle{\acute{\mathrmbfit{tbl}}_{\mathcal{S}}(g)
{\;\dashv\;}
\grave{\mathrmbfit{tbl}}_{\mathcal{S}}(g)}\bigr\rangle}}$}})
between the intermediate scope of type domains,
%along an 
%$X$-sorted type domain morphism
%$\mathcal{A}'\xrightarrow{\;g\;}\mathcal{A}$.
and 
%\item 
(\S\;\ref{sub:sub:sec:flow:sign:dom:mor}
projection$\,\circ\,$expansion/restriction$\,\circ\,$inflation
{\footnotesize{$
{{\bigl\langle{\acute{\mathrmbfit{tbl}}(h,f,g)
{\;\dashv\;}
\grave{\mathrmbfit{tbl}}(h,f,g)}\bigr\rangle}}
$}})
%along a signed domain morphism
%{\footnotesize{$
%\mathcal{D}'
%%={\langle{\mathcal{S}',\mathcal{A}'}\rangle}
%\xrightarrow{{\langle{h,f,g}\rangle}}
%%{\langle{\mathcal{S},\mathcal{A}}\rangle}=
%\mathcal{D}
%$.}\normalsize}
between the small scope of signed domains.
%\]
%\end{itemize}
%
%%%%%%%%%%%%%%%%%%%%%%%%%%%%%%%%%%%%%%%%%%%%%%%%%%%%%%%%%%%%%%%%%%%%%%%%%%%%%%%%
%%%%%%%%%%%%%%%%%%%%%%%%%%%%%%%%%%%%%%%%%%%%%%%%%%%%%%%%%%%%%%%%%%%%%%%%%%%%%%%%
\footnote{For more on this,
see \textit{Formula Interpretation} 
\S\;2.2.1  of the paper
``The {\ttfamily ERA} of {\ttfamily FOLE}: Superstructure''
\cite{kent:fole:era:supstruc}.}
%
%%%%%%%%%%%%%%%%%%%%%%%%%%%%%%%%%%%%%%%%%%%%%%%%%%%%%%%%%%%%%%%%%%%%%%%%%%%%%%%%
%%%%%%%%%%%%%%%%%%%%%%%%%%%%%%%%%%%%%%%%%%%%%%%%%%%%%%%%%%%%%%%%%%%%%%%%%%%%%%%%
%
\footnote{To allow algebraic computations on the data domains, 
%such as defining algebraic operations 
%(multiply, divide, add, subtract, etc.) 
%on the values from two columns, 
%we must extend to 
%the full version of \texttt{FOLE}
see the paper``The First-order Logical Environment'' 
\cite{kent:iccs2013},
which defines syntactic flow along term vectors.}
%%%%%%%%%%%%%%%%%%%%%%%%%%%%%%%%%%%%%%%%%%%%%%%%%%%%%%%%%%%%%%%%%%%%%%%%%%%%%%%%
%%%%%%%%%%%%%%%%%%%%%%%%%%%%%%%%%%%%%%%%%%%%%%%%%%%%%%%%%%%%%%%%%%%%%%%%%%%%%%%%
\end{aside}
%\end{minipage}}}}}
%

%%%%%%%%%%%%%%%%%%%%%%%%%%%%%%%%%%%%%%%%%%%%%%%%%%%%%%%%%%%%%%
%%%%%%%%%%%%%%%%%%%%%%%%%%%%%%%%%%%%%%%%%%%%%%%%%%%%%%%%%%%%%%
\newpage
\section{Basic Components}
\label{sub:sec:base:ops}
%%%%%%%%%%%%%%%%%%%%%%%%%%%%%%%%%%%%%%%%%%%%%%%%%%%%%%%%%%%%
%%%%%%%%%%%%%%%%%%%%%%%%%%%%%%%%%%%%%%%%%%%%%%%%%%%%%%%%%%%%

%\textsf{Motto:} ``Project forward, select back, join here.''
%\newline\newline

Basic components are elements to be used in flowcharts.
%either reflection, Booleans (basic operations), or adjoint flow.
A case in point is the quotient composite operation 
of \S\,\ref{sub:sub:sec:quotient},
whose flowchart has only one component --- inflation.
There are three kinds of basic components:
two reflectors, three Booleans (basic operations), and 
four components of adjoint flow
(two each for type domain and signature).

%%%%%%%%%%%%%%%%%%%%%%%%%%%%%%%%%%%%%%%%%%%%%%%%%%%%%%%%%%%%%%%%%%
%\newpage
\subsection{Reflection.}\label{sub:sub:sec:reflect}
%%%%%%%%%%%%%%%%%%%%%%%%%%%%%%%%%%%%%%%%%%%%%%%%%%%%%%%%%%%%%%%%%%

%
\begin{figure}
\begin{center}
{{\begin{tabular}{c}
\begin{picture}(180,14)(-3,0)
%%%%%%%%%%%%%%%%%%%%%%%%%FIGURES%%%%%%%%%%%%%%%%%%%%%%%%%
%%%%%%%%%%%%%%%%%%%%%%%%%FIGURES%%%%%%%%%%%%%%%%%%%%%%%%%
\comment{\put(120,100){\begin{picture}(0,0)(0,0)
\setlength{\unitlength}{0.45pt}
\put(95,90){\makebox(0,0){\Large{{\textit{{flow chart}}}}}}
\put(95,60){\makebox(0,0){\Large{{\textbf{{Set}}}}}}
\put(95,25){\makebox(0,0){\Large{{\textbf{{Table}}}}}}
\put(0,0){\line(1,0){200}}
\put(0,120){\line(1,0){200}}
\put(0,0){\line(0,1){120}}
\put(200,0){\line(0,1){120}}
\end{picture}}}
%%%%%%%%%%%%%%%%%%%%%%%%%FIGURES%%%%%%%%%%%%%%%%%%%%%%%%%
%%%%%%%%%%%%%%%%%%%%%%%%%FIGURES%%%%%%%%%%%%%%%%%%%%%%%%%
%
%%%%%%%%%%%%%%%%%%%%%%%%%%%%%%%%%%%%%%%%%%%%%%%%%%
%%%%%%%%%%%%%%%%%%%%%%%%%%%%%%%%%%%%%%%%%%%%%%%%%%
\put(20,-9){\begin{picture}(0,0)(0,0)
\setlength{\unitlength}{0.45pt}
%\thicklines
%\put(106,40){\makebox(0,0){\normalsize{$\boldsymbol{\circ}$}}}
%\put(4.7,40){\makebox(0,0){\normalsize{$\boldsymbol{\circ}$}}}
\put(40,10){\line(1,0){60}}
\put(40,70){\line(1,0){60}}
\put(100,70){\line(0,-1){60}}
\put(40,40){\oval(60,60)[bl]}
\put(40,40){\oval(60,60)[tl]}
\put(57,52){\makebox(0,0){\footnotesize{{\textit{{include}}}}}}
\put(55,30){\makebox(0,0){\LARGE{${\Leftarrow}$}}}
\put(10,40){\vector(-1,0){30}}
\put(130,40){\vector(-1,0){30}}
\end{picture}}
%%%%%%%%%%%%%%%%%%%%%%%%%%%%%%%%%%%%%%%%%%%%%%%%%%
\put(105,-9){\begin{picture}(0,0)(0,0)
\setlength{\unitlength}{0.45pt}
%\thicklines
%\put(106,40){\makebox(0,0){\normalsize{$\boldsymbol{\circ}$}}}
%\put(5.8,40){\makebox(0,0){\normalsize{$\boldsymbol{\circ}$}}}
\put(10,10){\line(1,0){60}}
\put(10,70){\line(1,0){60}}
\put(11,70){\line(0,-1){60}}
\put(70,40){\oval(60,60)[br]}
\put(70,40){\oval(60,60)[tr]}
\put(55,52){\makebox(0,0){\footnotesize{{\textit{{image}}}}}}
\put(56,30){\makebox(0,0){\LARGE{${\Rightarrow}$}}}
\put(-20,40){\vector(1,0){30}}
\put(100,40){\vector(1,0){30}}
\end{picture}}
%%%%%%%%%%%%%%%%%%%%%%%%%%%%%%%%%%%%%%%%%%%%%%%%%%
%\put(29,7.5){\vector(-1,0){12}}
%\put(73,7.5){\vector(-1,0){12}}
\end{picture}
\end{tabular}}}
\end{center}
\caption{\texttt{FOLE} Reflection Operators}
\label{fig:fole:reflec:proc}
\end{figure}
Let $\mathcal{D}={\langle{\mathcal{S},\mathcal{A}}\rangle}$ be a fixed signed domain.
Here,
we define reflection between the smallest fiber contexts
$\mathrmbf{Tbl}(\mathcal{D}){\;\rightleftarrows\;}\mathrmbf{Rel}(\mathcal{D})$.
The context of relations forms a sub-context of tables:
there is an inclusion passage 
$\mathrmbf{Rel}(\mathcal{D})
\xhookrightarrow{\;\mathrmbfit{inc}_{\mathcal{D}}}
\mathrmbf{Tbl}(\mathcal{D})$.
Conversely,
there is an image passage 
$\mathrmbf{Tbl}(\mathcal{D})
\xrightarrow{\;\mathrmbfit{im}_{\mathcal{D}}}
\mathrmbf{Rel}(\mathcal{D})$
defined as follows.
A table
${\langle{K,t}\rangle}{\;\in\;}\mathrmbf{Tbl}(\mathcal{D})$
with tuple function
$K\xrightarrow{\,t\;}\mathrmbfit{tup}_{\mathcal{A}}(\mathcal{S})$
is mapped to the relation 
${\langle{{\wp{t}}(K),i}\rangle}{\;\in\;}\mathrmbf{Rel}(\mathcal{D})$
with inclusion tuple function
${\wp{t}}(K)\xhookrightarrow{\,i\;}\mathrmbfit{tup}_{\mathcal{A}}(\mathcal{S})$,
which is essentially its tuple subset
${\wp{t}}(K) \subseteq \mathrmbfit{tup}_{\mathcal{A}}(\mathcal{S})$.
A table morphism
{\footnotesize{$
\mathcal{T}' = {\langle{K',t'}\rangle}
\xleftarrow{\;k\;} 
{\langle{K,t}\rangle} = \mathcal{T}
$}\normalsize}
with table morphism condition 
$k{\;\cdot\;}t' = t$
is mapped to the relation morphism 
{\footnotesize{$
\mathcal{R}' = {\langle{{\wp{t'}}(K'),i'}\rangle}
\xhookleftarrow{r} 
{\langle{{\wp{t}}(K),i}\rangle} = \mathcal{R}
$}\normalsize}
with relation morphism condition
${\wp{t'}}(K'){\;\supseteq\;}{\wp{t}}(K)$.
%in $\mathrmbf{Rel}$
%guaranteed by 
%
\begin{center}
{{\begin{tabular}{c}
{\setlength{\unitlength}{0.5pt}\begin{picture}(160,150)(0,-15)
\put(0,120){\makebox(0,0){\footnotesize{$K'$}}}
\put(0,60){\makebox(0,0){\footnotesize{${\wp{t'}}(K)'$}}}
%\put(20,65){\makebox(0,0){\footnotesize{$\supseteq$}}}
%\put(40,50){\makebox(0,0){\footnotesize{$\exists_{h}(R)$}}}
%\put(280,110){\makebox(0,0){\footnotesize{$R$}}}
%\put(260,95){\makebox(0,0){\footnotesize{$\supseteq$}}}
\put(160,120){\makebox(0,0){\footnotesize{$K$}}}
\put(160,60){\makebox(0,0){\footnotesize{${\wp{t}}(K)$}}}
%\put(220,65){\makebox(0,0){\footnotesize{$\supseteq$}}}
%\put(200,50){\makebox(0,0){\footnotesize{${h}^{{\scriptscriptstyle-}1}(R')$}}}
\put(0,0){\makebox(0,0){\footnotesize{$\mathrmbfit{tup}_{\mathcal{A}}(\mathcal{S})$}}}
\put(160,0){\makebox(0,0){\footnotesize{$\mathrmbfit{tup}_{\mathcal{A}}(\mathcal{S})$}}}
\put(80,130){\makebox(0,0){\scriptsize{$k$}}}
\put(80,70){\makebox(0,0){\scriptsize{$r$}}}
\put(80,0){\makebox(0,0){\scriptsize{$=$}}}
%\mathrmbfit{tup}_{\mathcal{A}}(h)
%
\put(145,120){\vector(-1,0){130}}
\put(130,60){\vector(-1,0){100}}\put(130,65){\oval(6,6)[r]}
%\put(110,0){\vector(-1,0){60}}
\put(0,105){\vector(0,-1){30}}
\put(160,105){\vector(0,-1){30}}
\put(0,42){\vector(0,-1){27}}\put(4,43){\makebox(0,0){\scriptsize{$\cap$}}}
%\put(0,33){\line(0,-1){10}}\put(10,30){\makebox(0,0){\scriptsize{$\bigcap$}}}
\put(160,42){\vector(0,-1){27}}\put(164,43){\makebox(0,0){\scriptsize{$\cap$}}}
\put(-40,30){\makebox(0,0)[r]{\footnotesize{$
\mathrmbfit{inc}_{\mathcal{D}}(\mathrmbfit{im}_{\mathcal{D}}(\mathcal{T}'))
\left\{\rule{0pt}{24pt}\right.$}}}
\put(200,30){\makebox(0,0)[l]{\footnotesize{$
\left.\rule{0pt}{24pt}\right\}
\mathrmbfit{inc}_{\mathcal{D}}(\mathrmbfit{im}_{\mathcal{D}}(\mathcal{T}))$}}}
%\mathrmbfit{im}_{\mathcal{S}'}^{\mathcal{A}}(K',t')
%\mathrmbfit{im}_{\mathcal{S}}^{\mathcal{A}}(K,t)
%\put(278,96){\vector(-1,-4){20}}
%\put(215,40){\vector(3,-4){20}}
%\put(-38,96){\vector(1,-4){20}}
%\put(35,40){\vector(-3,-4){20}}
%
%\put(10,100){\oval(20,20)[tl]}
%\put(10,110){\line(1,0){220}}
%\put(230,100){\oval(20,20)[tr]}
%\put(0,94){\vector(0,-1){0}}
%
\end{picture}}
\end{tabular}}}
\end{center}
The diagram above
factors the condition 
$k{\;\cdot\;}t' = t$
%$k{\;\cdot\;}t'=t{\;\cdot\;}\mathrmbfit{tup}_{\mathcal{A}}(h)$ 
by diagonal fill-in.
%A $\mathcal{D}$-relation morphism 
%$\mathcal{R}' = {\langle{R',i'}\rangle}\xhookleftarrow{r}{\langle{R,i}\rangle} = \mathcal{R}$
%is an inclusion tuple function $R'\xhookleftarrow{\,r\,}R$
%satisfying the condition
%$\mathrmbfit{tup}_{\mathcal{A}}(\mathcal{S}) \supseteq R' \supseteq R$.
This gives the $\mathcal{D}$-table morphism 
$\mathrmbfit{inc}_{\mathcal{D}}(\mathrmbfit{im}_{\mathcal{D}}(\mathcal{T}'))
\xhookleftarrow{\;r\;}
\mathrmbfit{inc}_{\mathcal{D}}(\mathrmbfit{im}_{\mathcal{D}}(\mathcal{T}))$,
which is the image-inclusion composite passage 
applied to 
the $\mathcal{D}$-table morphism 
$\mathcal{T}'
%={\langle{\mathcal{S}',K',t'}\rangle}
\xleftarrow{\;k\;}
%{\langle{\mathcal{S},K,t}\rangle}=
\mathcal{T}$.
%\end{proof}
%

%\comment{
\begin{proposition}\label{tbl:rel:refl}
Image and inclusion form reflections on full and fiber contexts:
\begin{center}
{{{\begin{tabular}{c}
\setlength{\extrarowheight}{2pt}
{\footnotesize{$\begin{array}[c]{r@{\hspace{2pt}{\;:\;}\hspace{2pt}}
l@{\hspace{4pt}{\;\rightleftarrows\;}\hspace{4pt}}l}
{\langle{\mathrmbfit{im}{\;\dashv\;}\mathrmbfit{inc}}\rangle}
&
\mathrmbf{Tbl}
&
\mathrmbf{Rel}
\\
{\langle{\mathrmbfit{im}_{\mathcal{A}}{\;\dashv\;}\mathrmbfit{inc}_{\mathcal{A}}}\rangle}
&
\mathrmbf{Tbl}(\mathcal{A})
&
\mathrmbf{Rel}(\mathcal{A})
\\
{\langle{\mathrmbfit{im}_{\mathcal{D}}{\;\dashv\;}\mathrmbfit{inc}_{\mathcal{D}}}\rangle}
&
\mathrmbf{Tbl}(\mathcal{D})
&
\mathrmbf{Rel}(\mathcal{D})
={\langle{{\wp}\mathrmbfit{tup}_{\mathcal{A}}(\mathcal{S}),\subseteq}\rangle}
\end{array}$}}
\end{tabular}}}}
\end{center}
\end{proposition}
\begin{proof}
The reflection at type domain $\mathcal{A}$
appears in appendix \S\;A.1 of 
%the paper 
%``The {\ttfamily FOLE} Table'' 
\cite{kent:fole:era:tbl}.
\mbox{}\hfill\rule{5pt}{5pt}
\end{proof}
Each reflection embodies the notion of informational equivalence.
The inclusion operator 
$\mathrmbf{Tbl}(\mathcal{D})
\xhookleftarrow{\mathrmbfit{inc}_{\mathcal{D}}}
\mathrmbf{Rel}(\mathcal{D})$
can be used at the input of any composite operator on tables.
Dually,
the image operator 
$\mathrmbf{Tbl}(\mathcal{D})
\xrightarrow{\mathrmbfit{im}_{\mathcal{D}}}
\mathrmbf{Rel}(\mathcal{D})$
can be used at the output of any composite operator on tables.
%
%%%%%%%%%%%%%%%%%%%%%%%%%%%%%%%%%%%%%%%%%%%%%%%%%%%%%%%%%%%%
%%%%%%%%%%%%%%%%%%%%%%%%%%%%%%%%%%%%%%%%%%%%%%%%%%%%%%%%%%%%
\footnote{
%Amongst other things,
For example,
the inclusion operator is used before the inflation operator
in \S\,\ref{sub:sub:sec:nat:join}
to define the selection operation,
and
the image operator can be used after the projection operator 
%in \S\,\ref{sub:sub:sec:proj:infl:typ:dom:fbr}
in \S\,\ref{sub:sub:sec:semi:join}
to define the semi-join operation.}

{% not needed by projection
%%%%%%%%%%%%%%%%%%%%%%%%%%%%%%%%%%%%%%%%%%%%%%%%%%%%%%%%%%%%
%%%%%%%%%%%%%%%%%%%%%%%%%%%%%%%%%%%%%%%%%%%%%%%%%%%%%%%%%%%%

%%%%%%%%%%%%%%%%%%%%%%%%%%%%%%%%%%%%%%%%%%%%%%%%%%%%%%%%%%%%%%%%%%
%
\newpage
\subsection{Booleans.}\label{sub:sub:sec:boole}
%\subsubsection{Booleans in $\mathrmbf{Tbl}(\mathcal{D})$.}
%%%%%%%%%%%%%%%%%%%%%%%%%%%%%%%%%%%%%%%%%%%%%%%%%%%%%%%%%%%%%%%%%

%
\begin{figure}
\begin{center}
{{\begin{tabular}[h]{c}
\begin{picture}(180,25)(-5,0)
%%%%%%%%%%%%%%%%%%%%%%%%%%%%%%%%%%%%%%%%%%%%%%%%%%
\put(-5,-5){\begin{picture}(0,0)(0,3)
\setlength{\unitlength}{0.40pt}
\put(60,30){\makebox(0,0){\Large{$\wedge$}}}
%\thicklines
%\put(33,84){\line(0,1){10}}
%\put(33,76){\oval(14,14)[t]}
\put(33,95){\vector(0,-1){25}}
\put(87,95){\vector(0,-1){25}}
\put(60,10){\vector(0,-1){25}}
%\put(33,76){\makebox(0,0){\large{$\boldsymbol{\circ}$}}}
%\put(87,76){\makebox(0,0){\normalsize{$\boldsymbol{\circ}$}}}
%\put(60,3){\makebox(0,0){\normalsize{$\boldsymbol{\circ}$}}}
\put(10,70){\line(1,0){100}}
\put(40,10){\line(1,0){40}}
\put(10,70){\line(0,-1){30}}
\put(110,70){\line(0,-1){30}}
\put(40,40){\oval(60,60)[bl]}
\put(80,40){\oval(60,60)[br]}
\put(60,55){\makebox(0,0){\footnotesize{{\textit{{meet}}}}}}
\end{picture}}
%%%%%%%%%%%%%%%%%%%%%%%%%%%%%%%%%%%%%%%%%%%%%%%%%%
%%%%%%%%%%%%%%%%%%%%%%%%%%%%%%%%%%%%%%%%%%%%%%%%%%
\put(60,-5){\begin{picture}(0,0)(0,3)
\setlength{\unitlength}{0.40pt}
\put(60,30){\makebox(0,0){\Large{$\vee$}}}
%\thicklines
%\put(33,84){\line(0,1){10}}
%\put(33,76){\oval(14,14)[t]}
\put(33,95){\vector(0,-1){25}}
\put(87,95){\vector(0,-1){25}}
\put(60,10){\vector(0,-1){25}}
%\put(33,76){\makebox(0,0){\large{$\boldsymbol{\circ}$}}}
%\put(87,76){\makebox(0,0){\normalsize{$\boldsymbol{\circ}$}}}
%\put(60,3){\makebox(0,0){\normalsize{$\boldsymbol{\circ}$}}}
\put(10,70){\line(1,0){100}}
\put(40,10){\line(1,0){40}}
\put(10,70){\line(0,-1){30}}
\put(110,70){\line(0,-1){30}}
\put(40,40){\oval(60,60)[bl]}
\put(80,40){\oval(60,60)[br]}
\put(60,55){\makebox(0,0){\footnotesize{{\textit{{join}}}}}}
\end{picture}}
%%%%%%%%%%%%%%%%%%%%%%%%%%%%%%%%%%%%%%%%%%%%%%%%%%
%%%%%%%%%%%%%%%%%%%%%%%%%%%%%%%%%%%%%%%%%%%%%%%%%%
\put(125,-5){\begin{picture}(0,0)(0,3)
\setlength{\unitlength}{0.40pt}
%\put(60,40){\makebox(0,0){\large{$-$}}}
\put(33,95){\vector(0,-1){25}}
\put(87,95){\vector(0,-1){25}}
\put(60,10){\vector(0,-1){25}}
%\put(33,76){\makebox(0,0){\normalsize{$\boldsymbol{\circ}$}}}
%\put(87,76){\makebox(0,0){\normalsize{$\boldsymbol{\circ}$}}}
%\put(60,3){\makebox(0,0){\normalsize{$\boldsymbol{\circ}$}}}
\put(10,70){\line(1,0){100}}
\put(40,10){\line(1,0){40}}
\put(10,70){\line(0,-1){30}}
\put(110,70){\line(0,-1){30}}
\put(40,40){\oval(60,60)[bl]}
\put(80,40){\oval(60,60)[br]}
\thicklines
\put(50,30){\line(1,0){20}}
\put(60,50){\makebox(0,0){\footnotesize{{\textit{{diff}}}}}}
\end{picture}}
%%%%%%%%%%%%%%%%%%%%%%%%%%%%%%%%%%%%%%%%%%%%%%%%%%
%%%%%%%%%%%%%%%%%%%%%%%%%%%%%%%%%%%%%%%%%%%%%%%%%%
\end{picture}
\end{tabular}}}
\end{center}
\caption{\texttt{FOLE} Boolean Operators}
\label{fig:fole:boole:proc}
\end{figure}

\comment{
&
%%%%%%%%%%%%%%%%%%%%%%%%%%%%%%%%%%%%%%%%%%%%%%%%%%
{{\begin{tabular}{c}
\setlength{\unitlength}{0.4pt}
\begin{picture}(120,80)(0,0)
\thicklines
\put(106,40){\makebox(0,0){\normalsize{$\boldsymbol{\circ}$}}}
\put(4.4,40){\makebox(0,0){\normalsize{$\boldsymbol{\circ}$}}}
\put(40,10){\line(1,0){61}}
\put(40,70){\line(1,0){61}}
\put(100,70){\line(0,-1){60}}
\put(40,40){\oval(60,60)[bl]}
\put(40,40){\oval(60,60)[tl]}
\put(57,35){\makebox(0,0){\huge{${\Leftarrow}$}}}
\put(61,-10){\makebox(0,0){\scriptsize{{\textit{{project}}}}}}
\end{picture}
\end{tabular}}}
%%%%%%%%%%%%%%%%%%%%%%%%%%%%%%%%%%%%%%%%%%%%%%%%%%
&
%%%%%%%%%%%%%%%%%%%%%%%%%%%%%%%%%%%%%%%%%%%%%%%%%%
{{\begin{tabular}{c}
\setlength{\unitlength}{0.4pt}
\begin{picture}(120,80)(0,0)
\thicklines
%\put(101,46){\vector(-1,0){0}}
\put(106,40){\makebox(0,0){\normalsize{$\boldsymbol{\circ}$}}}
%\put(-6,50){\vector(-1,0){0}}
%\put(5,40){\makebox(0,0){\normalsize{${\bullet}$}}}
%\put(5,50){\makebox(0,0){\large{$\mathbf{\circ}$}}}
\put(5,40){\makebox(0,0){\normalsize{$\boldsymbol{\circ}$}}}
\put(10,10){\line(1,0){61}}
\put(10,70){\line(1,0){61}}
\put(11,70){\line(0,-1){60}}
%\put(100,70){\line(0,-1){60}}
%
%\put(40,40){\oval(60,60)[bl]}
%\put(40,40){\oval(60,60)[tl]}

\put(70,40){\oval(60,60)[br]}
\put(70,40){\oval(60,60)[tr]}
\put(57,35){\makebox(0,0){\huge{${\Rightarrow}$}}}
\put(61,-10){\makebox(0,0){\scriptsize{{\textit{{inflate}}}}}}
\end{picture}
\end{tabular}}}
}

The Boolean operators are binary operators on tables
based upon 
the traditional mathematical set operations.
These operators require that 
both tables have the same set of attributes; i.e. 
type domain $\mathcal{A} = {\langle{X,Y,\models_{\mathcal{A}}}\rangle}$.
In addition,
they require compatibility.
Two tables are said to be compatible 
when 
both tables have same number of attributes and 
corresponding attributes have the same data type 
(int, char, float, date, \dots).
Hence,
two tables are compatible when they have 
the same type domain $\mathcal{A} = {\langle{X,Y,\models_{\mathcal{A}}}\rangle}$ and 
the same signature $\mathcal{S} = {\langle{I,x,X}\rangle}$;
in short,
when they have the same signed domain
$\mathcal{D} = {\langle{\mathcal{S},\mathcal{A}}\rangle}$. 
Here,
we define the Boolean operators 
of meet, join and difference 
in the smallest table fiber context
$\mathrmbf{Tbl}(\mathcal{D})
=\mathrmbf{Tbl}_{\mathcal{A}}(\mathcal{S})$.
%(see footnote\;\ref{lim:colim:tbl})
%$^\ref{lim:colim:tbl}$
%
The corresponding Boolean operators in 
%the smallest relation fiber context
$\mathrmbf{Rel}(\mathcal{D})$
are called set intersection, set union and set difference.
%
%%%%%%%%%%%%%%%%%%%%%%%%%%%%%%%%%%%%%%%%%%%%%%%%%%%%%%%%%%%%
%%%%%%%%%%%%%%%%%%%%%%%%%%%%%%%%%%%%%%%%%%%%%%%%%%%%%%%%%%%%
\footnote{Codd also listed the Cartesian product as a Boolean operator.
However,
in this paper the Cartesian product is defined 
in the larger fiber $\mathrmbf{Tbl}(\mathcal{A})$
and is closely connected to the natural join operator
(see \S\,\ref{sub:sub:sec:nat:join}).}
%%%%%%%%%%%%%%%%%%%%%%%%%%%%%%%%%%%%%%%%%%%%%%%%%%%%%%%%%%%%
%%%%%%%%%%%%%%%%%%%%%%%%%%%%%%%%%%%%%%%%%%%%%%%%%%%%%%%%%%%%
%
Let
$\mathcal{T} = {\langle{K,t}\rangle}$
and
$\mathcal{T}' = {\langle{K',t'}\rangle}$
be two \texttt{FOLE} tables 
in $\mathrmbf{Tbl}(\mathcal{D})$.
These tables have tuple functions 
$K\xrightarrow{\,t\,}\mathrmbfit{tup}(\mathcal{D})
=\mathrmbfit{tup}_{\mathcal{A}}(\mathcal{S})$
and $K'\xrightarrow{\,t'\,}\mathrmbfit{tup}(\mathcal{D})=\mathrmbfit{tup}_{\mathcal{A}}(\mathcal{S})$
with two image \texttt{FOLE} relations
${\wp}t(K),{\wp}t'(K')\subseteq
\mathrmbfit{tup}(\mathcal{D})=\mathrmbfit{tup}_{\mathcal{A}}(\mathcal{S})$
in $\mathrmbf{Rel}(\mathcal{D})
=\mathrmbf{Rel}_{\mathcal{A}}(\mathcal{S})
={\langle{{\wp}\mathrmbfit{tup}_{\mathcal{A}}(\mathcal{S}),\subseteq}\rangle}$.
%%%%%%%%%%%%%%%%%%%%%%%%%%%%%%%%%%%%%%%%%%%%%%%%%%%%%%%%%%%%%%%%%%%%%%%%%%%%%%%%
\newpage
%%%%%%%%%%%%%%%%%%%%%%%%%%%%%%%%%%%%%%%%%%%%%%%%%%%%%%%%%%%%%%%%%%%%%%%%%%%%%%%%
%
\begin{itemize}
%
%%%%%%%%%%%%%%%%%%%%%%%%%%%%%%%%%%%%%%%%%%%%%%%%%%%%%%%%%%%%%%%%%%%%%%%%%%%%%%%%
\item[$\bigwedge:$] 
The intersection (meet, conjunction) operator produces the set of tuples that two tables share in common. %Intersection is implemented in SQL in the form of the INTERSECT operator.
%\item
The intersection operation defines the \texttt{FOLE} table
$\mathcal{T}\wedge\mathcal{T}' = {\langle{\widehat{K},{(t,t')}}\rangle}$
whose key set 
$\widehat{K} \subseteq K{\times}K'$
is the pullback and whose tuple map is the mediating function 
$\widehat{K} \xrightarrow{{(t,t')}}\mathrmbfit{tup}_{\mathcal{A}}(\mathcal{S})$
of the opspan
$K\xrightarrow{t}
\mathrmbfit{tup}_{\mathcal{A}}(\mathcal{S})
\xleftarrow{t'}K'$,
%, where the tuple map ${(t,t')}$
which maps a pair of keys 
$(k,k') \in \widehat{K}$
to the common tuple
$t(k)=t'(k') \in \mathrmbfit{tup}_{\mathcal{A}}(\mathcal{S})$.
The image relation
is the set-theoretic intersection
%${\wp}t(\widehat{K})
%={\wp}t(K){\,\cap\,}{\wp}t'(K')
%\subseteq\mathrmbfit{tup}_{\mathcal{A}}(\mathcal{S})$.
$\mathrmbfit{im}_{\mathcal{D}}(\mathcal{T}\wedge\mathcal{T}')=
\mathrmbfit{im}_{\mathcal{D}}(\mathcal{T})\cap\mathrmbfit{im}_{\mathcal{D}}(\mathcal{T}')$.
Intersection is the product in $\mathrmbf{Tbl}(\mathcal{D})$
with projection morphisms
\newline\mbox{}\hfill
$\mathcal{T}\xleftarrow{\hat{\pi}}\mathcal{T}
\wedge\mathcal{T}'
\xrightarrow{\hat{\pi}'}\mathcal{T}'$.
\hfill\mbox{}
%\newline
%satisfying the commuting diagrams
%$\pi \cdot t = (t,t') = \pi' \cdot t'$.
%%%%%%%%%%%%%%%%%%%%%%%%%%%%%%%%%%%%%%%%%%%%%%%%%%%%%%%%%%%%%%%%%%%%%%%%%%%%%%%%
\newline
%%%%%%%%%%%%%%%%%%%%%%%%%%%%%%%%%%%%%%%%%%%%%%%%%%%%%%%%%%%%%%%%%%%%%%%%%%%%%%%%
\item[$\bigvee:$] The union (join, disjunction) operator combines the tuples of two tables and removes all duplicate tuples from the result. 
%\item
The union operation defines the \texttt{FOLE} table
$\mathcal{T}\vee\mathcal{T}' = {\langle{K{+}K',{[t,t']}}\rangle}$
whose key set is the disjoint union $K{+}K'$
and whose tuple map is the comediating function
$K{+}K'\xrightarrow{[t,t']}\mathrmbfit{tup}_{\mathcal{A}}(\mathcal{S})$
of the opspan
$K\xrightarrow{t}
\mathrmbfit{tup}_{\mathcal{A}}(\mathcal{S})
\xleftarrow{t'}K'$,
which maps a key $k \in K$ 
to $t(k) \in \mathrmbfit{tup}_{\mathcal{A}}(\mathcal{S})$
and
maps a key $k' \in K'$ 
to $t'(k') \in \mathrmbfit{tup}_{\mathcal{A}}(\mathcal{S})$.
The image relation
is the set-theoretic union
$\mathrmbfit{im}_{\mathcal{D}}(\mathcal{T}\vee\mathcal{T}')=
\mathrmbfit{im}_{\mathcal{D}}(\mathcal{T})\cup\mathrmbfit{im}_{\mathcal{D}}(\mathcal{T}')$.
Union is the coproduct in $\mathrmbf{Tbl}(\mathcal{D})$
with injection morphisms
\newline\mbox{}\hfill
$\mathcal{T}\xrightarrow{\check{\iota}}\mathcal{T}
\vee\mathcal{T}'
\xleftarrow{\check{\iota}'}\mathcal{T}'$.
%
%%%%%%%%%%%%%%%%%%%%%%%%%%%%%%%%%%%%%%%%%%%%%%%%%%%%%%%%%%%%%%%%%%%%%%%%%%%%%%%%
%%%%%%%%%%%%%%%%%%%%%%%%%%%%%%%%%%%%%%%%%%%%%%%%%%%%%%%%%%%%%%%%%%%%%%%%%%%%%%%%
\footnote{Intersection $\bigwedge_{i \in I} \mathcal{T}_{i}$
and union $\bigvee_{i \in I} \mathcal{T}_{i}$
can be generalized to any number of $\mathcal{D}$-tables.}
%%%%%%%%%%%%%%%%%%%%%%%%%%%%%%%%%%%%%%%%%%%%%%%%%%%%%%%%%%%%%%%%%%%%%%%%%%%%%%%%
%%%%%%%%%%%%%%%%%%%%%%%%%%%%%%%%%%%%%%%%%%%%%%%%%%%%%%%%%%%%%%%%%%%%%%%%%%%%%%%%
%
\hfill\mbox{}
%with key set
%${\wp}t(K{+}K')
%={\wp}t(K){\,\cup\,}{\wp}t'(K')$.
%\subseteq\mathrmbfit{tup}_{\mathcal{A}}(\mathcal{S})$.
%When the tuple maps are injections,
%this is essentially the union set of tuples.
%
%%%%%%%%%%%%%%%%%%%%%%%%%%%%%%%%%%%%%%%%%%%%%%%%%%%%%%%%%%%%%%%%%%%%%%%%%%%%%%%%
\newline
%%%%%%%%%%%%%%%%%%%%%%%%%%%%%%%%%%%%%%%%%%%%%%%%%%%%%%%%%%%%%%%%%%%%%%%%%%%%%%%%
%
\item[$\mathbf{-}:$] The difference operator acts on two tables and produces the set of tuples from the first table that do not exist in the second table. 
%Difference is implemented in SQL in the form of the EXCEPT or MINUS operator.
%\item
The difference operation defines the \texttt{FOLE} table
$\mathcal{T}{-}\mathcal{T}' = {\langle{\bar{K},\bar{t}}\rangle}$
whose key set
$\bar{K}$ is the tuple inverse image of the difference tuple set
$\bar{K} = t^{\text{-}1}({\wp}t(K){\,{-}\,}{\wp}t'(K')){\,\subseteq\,}K$
and whose tuple map
$\bar{t} : \bar{K} \xhookrightarrow{\bar{k}}
K\xrightarrow{t}\mathrmbfit{tup}_{\mathcal{A}}(\mathcal{S})$
restricts to this subset.
The image relation
is the set-theoretic difference
$\mathrmbfit{im}_{\mathcal{D}}(\mathcal{T}{-}\mathcal{T}')=
\mathrmbfit{im}_{\mathcal{D}}(\mathcal{T}){-}\mathrmbfit{im}_{\mathcal{D}}(\mathcal{T}')$.
There is an inclusion morphism
\newline\mbox{}\hfill
$\mathcal{T}\xhookleftarrow{\bar{\omega}}(\mathcal{T}{-}\mathcal{T}')$.
\hfill\mbox{}
%${\wp}\bar{t}(\bar{K})
%={\wp}t(K){\,{-}\,}{\wp}t'(K')
%\subseteq\mathrmbfit{tup}_{\mathcal{A}}(\mathcal{S})$.
%
\end{itemize}
%
%%%%%%%%%%%%%%%%%%%%%%%%%%%%%%%%%%%%%%%%%%%%%%%%%%%%%%%%%%%%%%%%%%%%%%%%%%%%%%%%
%%%%%%%%%%%%%%%%%%%%%%%%%%%%%%%%%%%%%%%%%%%%%%%%%%%%%%%%%%%%%%%%%%%%%%%%%%%%%%%%

%
\begin{proposition}\label{boolean:laws}
There are many algebraic laws for the Boolean operations:
associativity, commutativity, idempotency for $\wedge$ and $\vee$;
Distributive laws for $\wedge$ over $\vee$ and $\vee$ over $\wedge$; 
Distributive laws for $\times$ w.r.t. $\wedge$ and $\vee$;
Complement and double negation laws for $-$;
DeMorgans laws for $-$ w.r.t. $\wedge$ and $\vee$.
\end{proposition}
\begin{proof}
Well-known.
\hfill\rule{5pt}{5pt}
\end{proof}
%

%%%%%%%%%%%%%%%%%%%%%%%%%%%%%%%%%%%%%%%%%%%%%%%%%%%%%%%%%%%%%
%%%%%%%%%%%%%%%%%%%%%%%%%%%%%%%%%%%%%%%%%%%%%%%%%%%%%%%%%%%%%
%
\newpage
\subsection{Adjoint Flow.}
\label{sub:sub:sec:adj:flow}
%%%%%%%%%%%%%%%%%%%%%%%%%%%%%%%%%%%%%%%%%%%%%%%%%%%%%%%%%%%%%
%%%%%%%%%%%%%%%%%%%%%%%%%%%%%%%%%%%%%%%%%%%%%%%%%%%%%%%%%%%%%
%

%%%%%%%%%%%%%%%%%%%%%%%%%%%%%%%%%%%%%%%%%%%%%%%%%%%%%%%%%%%%%
%\newpage
\subsubsection{Fixed Type Domain.}
\label{sub:sub:sec:adj:flow:A}
%\subsubsection{Projection/Inflation in $\mathrmbf{Tbl}(\mathcal{A})$.}
%%%%%%%%%%%%%%%%%%%%%%%%%%%%%%%%%%%%%%%%%%%%%%%%%%%%%%%%%%%%%
%
%%%%%%%%%%%%%%%%%%%%%%%%%%%%%%%%%%%%%%%%%%%%%%%%%%%%%%%%%%%%
%%%%%%%%%%%%%%%%%%%%%%%%%%%%%%%%%%%%%%%%%%%%%%%%%%%%%%%%%%%%
\comment{There is a 
%preliminary 
discussion 
of signed domain indexing
in
\S\;3.2 of 
%the paper 
``The {\ttfamily FOLE} Table'' 
\cite{kent:fole:era:tbl}.}
%%%%%%%%%%%%%%%%%%%%%%%%%%%%%%%%%%%%%%%%%%%%%%%%%%%%%%%%%%%%
%%%%%%%%%%%%%%%%%%%%%%%%%%%%%%%%%%%%%%%%%%%%%%%%%%%%%%%%%%%%
%

%
\begin{figure}
\begin{center}
{{\begin{tabular}{c}
\begin{picture}(180,14)(-3,0)
%%%%%%%%%%%%%%%%%%%%%%%%%%%%%%%%%%%%%%%%%%%%%%%%%%
%%%%%%%%%%%%%%%%%%%%%%%%%%%%%%%%%%%%%%%%%%%%%%%%%%
\put(20,-9){\begin{picture}(0,0)(0,0)
\setlength{\unitlength}{0.45pt}
%\thicklines
%\put(106,40){\makebox(0,0){\normalsize{$\boldsymbol{\circ}$}}}
%\put(4.7,40){\makebox(0,0){\normalsize{$\boldsymbol{\circ}$}}}
\put(40,10){\line(1,0){60}}
\put(40,70){\line(1,0){60}}
\put(100,70){\line(0,-1){60}}
\put(40,40){\oval(60,60)[bl]}
\put(40,40){\oval(60,60)[tl]}
\put(58,50){\makebox(0,0){\footnotesize{{\textit{{project}}}}}}
\put(56,30){\makebox(0,0){\LARGE{${\Leftarrow}$}}}
\put(10,40){\vector(-1,0){30}}
\put(130,40){\vector(-1,0){30}}
\end{picture}}
%%%%%%%%%%%%%%%%%%%%%%%%%%%%%%%%%%%%%%%%%%%%%%%%%%
\put(105,-9){\begin{picture}(0,0)(0,0)
\setlength{\unitlength}{0.45pt}
%\thicklines
%\put(106,40){\makebox(0,0){\normalsize{$\boldsymbol{\circ}$}}}
%\put(5.8,40){\makebox(0,0){\normalsize{$\boldsymbol{\circ}$}}}
\put(10,10){\line(1,0){60}}
\put(10,70){\line(1,0){60}}
\put(10,70){\line(0,-1){60}}
\put(70,40){\oval(60,60)[br]}
\put(70,40){\oval(60,60)[tr]}
\put(55,50){\makebox(0,0){\footnotesize{{\textit{{inflate}}}}}}
\put(56,30){\makebox(0,0){\LARGE{${\Rightarrow}$}}}
\put(-20,40){\vector(1,0){30}}
\put(100,40){\vector(1,0){30}}
\end{picture}}
%%%%%%%%%%%%%%%%%%%%%%%%%%%%%%%%%%%%%%%%%%%%%%%%%%
%\put(29,7.5){\vector(-1,0){12}}
%\put(73,7.5){\vector(-1,0){12}}
\end{picture}
\end{tabular}}}
\end{center}
\caption{\texttt{FOLE} Adjoint Flow Operators: Base $\mathcal{A}$}
\label{fig:fole:adj:flow:proc}
\end{figure}
\comment{%
Prop.\;\ref{prop:lim:tbl} 
of \S\,\ref{sub:sec:lim:colim:tbl}
defines the limit of a diagram of tables
in the full context 
$\mathrmbf{Tbl}$.
There we used adjoint flow along
signed domain morphisms
(Eqn.\;\ref{fbr:adj:sign:dom:mor})
with 
left adjoint flow defining projection
and
right adjoint flow defining inflation.
However,
for simplicity of explanation of the various relational operations,}
In this section we use
%adjoint flow along 
%$X$-sorted signature morphisms
%in the mid-sized table fiber context
%of 
a fixed type domain
%
%Let
$\mathcal{A} = {\langle{X,Y,\models_{\mathcal{A}}}\rangle}$.
%be a fixed type domain.
Here, 
we define adjoint flow 
in the mid-sized table fiber context
$\mathrmbf{Tbl}(\mathcal{A})$.
%of the signed domain morphism 
%in \S\,\ref{sub:sec:sign:dom}.
%(Disp.\;\ref{fbr:adj:sign:dom:mor}).
%
%%%%%%%%%%%%%%%%%%%%%%%%%%%%%%%%%%%%%%%%%%%%%%%%%%%%%%%%%%%%
%%%%%%%%%%%%%%%%%%%%%%%%%%%%%%%%%%%%%%%%%%%%%%%%%%%%%%%%%%%%
\footnote{
There is a 
%preliminary 
discussion 
of type domain indexing
in
\S\;3.4.1 of 
%the paper 
``The {\ttfamily FOLE} Table'' 
\cite{kent:fole:era:tbl}.}
%%%%%%%%%%%%%%%%%%%%%%%%%%%%%%%%%%%%%%%%%%%%%%%%%%%%%%%%%%%%
%%%%%%%%%%%%%%%%%%%%%%%%%%%%%%%%%%%%%%%%%%%%%%%%%%%%%%%%%%%%
%
Let
$\mathcal{S}'\xrightarrow{\;h\;}\mathcal{S}$
be an $X$-sorted signature morphism
%the sort function $X'\xrightarrow{f}X$ and 
with the arity (index) function $I'\xrightarrow{h}I$,
which satisfies the condition 
$h{\,\cdot\,}s = s'$.
Its tuple function
%\newline\mbox{}\hfill
{\footnotesize{$
\mathrmbfit{tup}_{\mathcal{A}}(\mathcal{S}')
\xleftarrow[h{\,\cdot\,}{(\mbox{-})}]
{\mathrmbfit{tup}_{\mathcal{A}}(h)}
\mathrmbfit{tup}_{\mathcal{A}}(\mathcal{S})
$}\normalsize}
%%%%%%%%%%%%%%%%%%%%%%%%%%%%%%%%%%%%%%%%%%%%%%%%%%%%%%%%%%%%
%%%%%%%%%%%%%%%%%%%%%%%%%%%%%%%%%%%%%%%%%%%%%%%%%%%%%%%%%%%%
\footnote{Visually,
$({\cdots\,}t_{h(i')}{\,\cdots}{\,\mid\,}i'{\,\in\,}I')
\mapsfrom
({\cdots\,}t_{i}{\,\cdots}{\,\mid\,}i{\,\in\,}I)$.}
%%%%%%%%%%%%%%%%%%%%%%%%%%%%%%%%%%%%%%%%%%%%%%%%%%%%%%%%%%%%
%%%%%%%%%%%%%%%%%%%%%%%%%%%%%%%%%%%%%%%%%%%%%%%%%%%%%%%%%%%%
%\hfill\mbox{}\newline
%
defines by composition/pullback
%(see \cite{kent:fole:era:tbl})
%Kent, R.E.:
%``The {\ttfamily FOLE} Table''.
a fiber adjunction of tables
\begin{equation}\label{def:fbr:adj:sign:mor}
{{\begin{picture}(120,10)(0,-4)
\put(60,0){\makebox(0,0){\footnotesize{$
\mathrmbf{Tbl}_{\mathcal{A}}(\mathcal{S}')
{\;\xleftarrow
[{\bigl\langle{
{\scriptscriptstyle\sum}_{h}{\;\dashv\;}{h}^{\ast}
\bigr\rangle}}]
{{\bigl\langle{\acute{\mathrmbfit{tbl}}_{\mathcal{A}}(h)
{\;\dashv\;}
\grave{\mathrmbfit{tbl}}_{\mathcal{A}}(h)}\bigr\rangle}}
\;}
\mathrmbf{Tbl}_{\mathcal{A}}(\mathcal{S})
$}}}
\end{picture}}}
\end{equation}
\begin{itemize}
\item[\textbf{project:}]  
The left adjoint (existential quantifier)
%table passage
$\mathrmbf{Tbl}_{\mathcal{A}}(\mathcal{S}')
\xleftarrow[{\scriptscriptstyle\sum}_{h}]
{\;\acute{\mathrmbfit{tbl}}_{\mathcal{A}}(h)\;}
\mathrmbf{Tbl}_{\mathcal{A}}(\mathcal{S})$
defines projection.
An $\mathcal{A}$-table
$\mathcal{T}={\langle{K,t}\rangle} \in \mathrmbf{Tbl}_{\mathcal{A}}(\mathcal{S})$
%the projection along $I'\xrightarrow{h}I$ is 
is mapped to the $\mathcal{A}$-table
${\scriptstyle\sum}_{h}(\mathcal{T})
= \mathcal{T}' = {\langle{K,t'}\rangle} 
\in \mathrmbf{Tbl}_{\mathcal{A}}(\mathcal{S}')$
with its tuple function
$K \xrightarrow{t'} \mathrmbfit{tup}_{\mathcal{A}}(\mathcal{S}')$
defined by composition,
$t' = t{\,\cdot\,}\mathrmbfit{tup}_{\mathcal{A}}(h)$.
%with the tuple function
%$\mathrmbfit{tup}_{\mathcal{A}}(\mathcal{S}')
%\xleftarrow{\;\mathrmbfit{tup}_{\mathcal{A}}(h)\;}
%\mathrmbfit{tup}_{\mathcal{A}}(\mathcal{S})$.
Here we have 
``horizontally abridged'' 
(projected out sub-tuples from)
%data value 
tuples in $\mathrmbf{List}(Y)$ 
%for data value set $Y$ 
by tuple composition with the index function 
$I'\xrightarrow{\;h\;}I$.
There is an $\mathcal{A}$-table morphism 
(LHS Fig.\;\ref{fig:proj:arity})
%\[\mbox
{\footnotesize{
{{$\mathcal{T}'={\langle{\mathcal{S}',K,t'}\rangle}\xleftarrow{\langle{h,1_{K}}\rangle}{\langle{\mathcal{S},K,t}\rangle}=\mathcal{T}$.}}
}\normalsize}
%\]
%
We say that
$\mathcal{A}$-table $\mathcal{T}'={\scriptstyle\sum}_{h}(\mathcal{T})$ 
is the 
%$h^{\text{th}}$ 
\underline{projection} of $\mathcal{A}$-table $\mathcal{T}$
along 
$X$-sorted signature morphism
$\mathcal{S}'\xrightarrow{\;h\;}\mathcal{S}$.
Being left adjoint in flow 
%(Disp.\;\ref{fbr:adj:sign:dom:mor}),
projection preserves colimits.
\item[\textbf{inflate:}] 
The right adjoint (inverse image)
$\mathrmbf{Tbl}_{\mathcal{A}}(\mathcal{S}')
{\;\xrightarrow[{h}^{\ast}]
{\;\grave{\mathrmbfit{tbl}}_{\mathcal{A}}(h)\;}\;}
\mathrmbf{Tbl}_{\mathcal{A}}(\mathcal{S})$
defines inflation.
An $\mathcal{A}$-table
$\mathcal{T}'={\langle{K',t'}\rangle} \in 
\mathrmbf{Tbl}_{\mathcal{A}}(\mathcal{S}')$
%\mathrmbf{Tbl}_{\mathcal{S}'}({f}^{-1}(\mathcal{A}))$
%the projection along $I'\xrightarrow{h}I$ is 
is mapped to the $\mathcal{A}$-table
${h}^{\ast}(\mathcal{T}')
= \mathcal{T} = {\langle{K,t}\rangle} 
\in \mathrmbf{Tbl}_{\mathcal{A}}(\mathcal{S})$,
with its tuple function
$K \xrightarrow{t} \mathrmbfit{tup}_{\mathcal{A}}(\mathcal{S})$
defined by pullback,
$k{\,\cdot\,}t' = t{\,\cdot\,}\mathrmbfit{tup}_{\mathcal{A}}(h)$. 
Here we have ``horizontally inflated'' 
%data value 
tuples in $\mathrmbf{List}(Y)$ 
%for data value set $Y$ 
by tuple pullback back along the index function 
$I'\xrightarrow{\;h\;}I$.
There is an $\mathcal{A}$-table morphism 
(RHS Fig.\;\ref{fig:proj:arity})
%\[\mbox
{\footnotesize{
{{$\mathcal{T}'={\langle{\mathcal{S}',K',t'}\rangle}\xleftarrow{\langle{h,
k}\rangle}{\langle{\mathcal{S},K,t}\rangle}=\mathcal{T}$.}}
}\normalsize}
%\]
%
We say that
$\mathcal{A}$-table 
$\mathcal{T}={h}^{\ast}(\mathcal{T}')$
is the 
%$h^{\text{th}}$ 
\underline{inflation} of $\mathcal{A}$-table $\mathcal{T}'$
%back 
along 
$X$-sorted signature morphism
$\mathcal{S}'\xrightarrow{\;h\;}\mathcal{S}$.
Being right/left adjoint in flow,
%
%%%%%%%%%%%%%%%%%%%%%%%%%%%%%%%%%%%%%%%%%%%%%%%%%%%%%%%%%%%%
%%%%%%%%%%%%%%%%%%%%%%%%%%%%%%%%%%%%%%%%%%%%%%%%%%%%%%%%%%%%
\footnote{
Inverse image is left adjoint
{\footnotesize{$
\mathrmbf{Tbl}_{\mathcal{A}}(\mathcal{S}')
{\;\xrightarrow
{{\bigl\langle{
{h}^{\ast}{\;\dashv\;}{\scriptscriptstyle\prod}_{h}
\bigr\rangle}}}
%{{\bigl\langle{\acute{\mathrmbfit{tbl}}_{\mathcal{A}}(h)
%{\;\dashv\;}
%\grave{\mathrmbfit{tbl}}_{\mathcal{A}}(h)}\bigr\rangle}}
\;}
\mathrmbf{Tbl}_{\mathcal{A}}(\mathcal{S})
$}}
to universal quantification.}
%%%%%%%%%%%%%%%%%%%%%%%%%%%%%%%%%%%%%%%%%%%%%%%%%%%%%%%%%%%%
%%%%%%%%%%%%%%%%%%%%%%%%%%%%%%%%%%%%%%%%%%%%%%%%%%%%%%%%%%%%
%(Disp.\;\ref{fbr:adj:sign:dom:mor}),
inflation preserves limits/colimits.
\end{itemize}
\begin{figure}
\begin{center}
{{\begin{tabular}{c@{\hspace{75pt}}c}
%%%%%%%%%%%%%%%%%%%%%%%%%%%%%%%%%%%%%%%%%%%%%%%%%%
{{\begin{tabular}{c}
\setlength{\unitlength}{0.65pt}
%\begin{picture}(120,120)(0,-20)
\begin{picture}(120,90)(0,-5)
\put(0,80){\makebox(0,0){\footnotesize{$K$}}}
%\qbezier(-10,80)(-40,75)(-48,50)\put(-48,50){\vector(-1,-3){0}}
%\qbezier(-50,28)(-50,15)(-45,8)\put(-45,8){\vector(1,-2){0}}
%\put(-46,28){\oval(8,8)[t]}
%\put(-49,40){\makebox(0,0){\footnotesize{$R$}}}
\put(120,80){\makebox(0,0){\footnotesize{$K$}}}
%\put(-10,0){\makebox(0,0){\footnotesize{$\mathrmbfit{tup}_{\mathcal{A}}(\mathcal{S}')$}}}
\put(-10,-10){\makebox(0,0){\footnotesize{$
%\underset{\textstyle{= \mathrmbfit{ext}_{\mathrmbf{List}(\mathcal{A})}(1',s')}}
\underset{\textstyle{\subseteq \mathrmbf{List}(Y)}}
{\mathrmbfit{tup}_{\mathcal{A}}(\mathcal{S}')}
$}}}
\put(130,-10){\makebox(0,0){\footnotesize{$
%\underset{\textstyle{= \mathrmbfit{ext}_{\mathrmbf{List}(\mathcal{A})}(1,s)}}
\underset{\textstyle{\subseteq \mathrmbf{List}(Y)}}
{\mathrmbfit{tup}_{\mathcal{A}}(\mathcal{S})}$}}}
%\put(60,-40){\makebox(0,0){\footnotesize{$\mathrmbfit{tup}_{\mathcal{A}}(\mathcal{S}')$}}}
\put(63,95){\makebox(0,0){\scriptsize{$1_{K}$}}}
\put(65,-12){\makebox(0,0){\scriptsize{$\mathrmbfit{tup}_{\mathcal{A}}(h)$}}}
\put(6,40){\makebox(0,0)[l]{\scriptsize{$t'$}}}
%\put(-42,73){\makebox(0,0)[r]{\scriptsize{$\hat{e}$}}}
%\put(-53,20){\makebox(0,0)[r]{\scriptsize{$\hat{m}$}}}
\put(128,40){\makebox(0,0)[l]{\scriptsize{$t$}}}
%\put(105,-22){\makebox(0,0)[l]{\scriptsize{$\mathrmbfit{tup}_{\mathcal{A}}(h)$}}} 
%\put(15,-25){\makebox(0,0)[r]{\scriptsize{$\underset{\textstyle{=1}}{\grave{\tau}_{{\langle{\mathrmit{1}_{X},\mathrmit{1}_{Y}}\rangle}}(\mathcal{S}')}$}}}
\put(0,65){\vector(0,-1){50}}
%\put(-10,65){\vector(-3,-2){22}}
\put(120,65){\vector(0,-1){50}}
\put(100,80){\vector(-1,0){80}}
\put(90,0){\vector(-1,0){56}}
%\put(45,-30){\vector(-3,2){30}}
%\put(105,-10){\vector(-3,-2){30}}
%%%%%%%%%%
\put(-10,-44){\makebox(0,0){\normalsize{$\underset{\textstyle{
\mathrmbf{Tbl}_{\mathcal{A}}(\mathcal{S}')
}}{\underbrace{\rule{50pt}{0pt}}}$}}}
%\put(-27,45){\makebox(0,0){\huge{
%$\overset{\;\textit{\scriptsize{im}}}{\Leftarrow}$}}}
%%%%%%%%%%
\put(60,45){\makebox(0,0){\huge{
$\overset{\textit{\scriptsize{project}}}{\Leftarrow}$}}}
%%%%%%%%%%
\put(180,40){\makebox(0,0){\Large{$\rightleftarrows$}}}
%%%%%%%%%%
\end{picture}
\end{tabular}}}
%%%%%%%%%%%%%%%%%%%%%%%%%%%%%%%%%%%%%%%%%%%%%%%%%%
&
%%%%%%%%%%%%%%%%%%%%%%%%%%%%%%%%%%%%%%%%%%%%%%%%%%
{{\begin{tabular}{c}
\setlength{\unitlength}{0.65pt}
\begin{picture}(120,90)(0,-5)
\put(0,80){\makebox(0,0){\footnotesize{$K'$}}}
\put(120,80){\makebox(0,0){\footnotesize{$K$}}}
%\put(-27,0){\makebox(0,0){\footnotesize{$\mathrmbfit{tup}_{f^{-1}(\mathcal{A})}(\mathcal{S}')$}}}
%\put(-10,-10){\makebox(0,0){\footnotesize{$
%\underset{\textstyle{\subseteq \mathrmbf{List}(Y)}}
%{
%%\mathrmbfit{tup}_{\mathcal{A}'}(\mathcal{S}')
%\mathrmbfit{tup}_{\mathcal{S}'}({f}^{-1}(\mathcal{A}))
%}$}}}
\put(-10,-10){\makebox(0,0){\footnotesize{$
%\underset{\textstyle{= \mathrmbfit{ext}_{\mathrmbf{List}(\mathcal{A})}(1',s')}}
\underset{\textstyle{\subseteq \mathrmbf{List}(Y)}}
{\mathrmbfit{tup}_{\mathcal{A}}(\mathcal{S}')}
$}}}
\put(130,-10){\makebox(0,0){\footnotesize{$
\underset{\textstyle{\subseteq \mathrmbf{List}(Y)}}
{\mathrmbfit{tup}_{\mathcal{A}}(\mathcal{S})}$}}}
%\put(60,-40){\makebox(0,0){\footnotesize{$\mathrmbfit{tup}_{\mathcal{A}}({\scriptstyle\sum}_{f}(\mathcal{S}'))$}}}
\put(63,90){\makebox(0,0){\scriptsize{$k$}}}
%\put(65,12){\makebox(0,0){\scriptsize{$\tau_{{\langle{h,f,g}\rangle}}$}}}
%\put(70,-12){\makebox(0,0){\scriptsize{$
%\tau_{{\langle{h,f}\rangle}}(\mathcal{A})$}}}
\put(65,-12){\makebox(0,0){\scriptsize{$
\mathrmbfit{tup}_{\mathcal{A}}(h)$}}}
\put(-6,40){\makebox(0,0)[r]{\scriptsize{$t'$}}}
\put(128,40){\makebox(0,0)[l]{\scriptsize{$t$}}}
%\put(105,-22){\makebox(0,0)[l]{\scriptsize{$\mathrmbfit{tup}_{\mathcal{A}}(h)$}}} 
%\put(15,-25){\makebox(0,0)[r]{\scriptsize{$\underset{\textstyle{=1}}{\grave{\tau}_{{\langle{f,\mathrmit{1}_{Y}}\rangle}}(\mathcal{S}')}$}}}
\put(0,65){\vector(0,-1){50}}
\put(120,65){\vector(0,-1){50}}
\put(100,80){\vector(-1,0){80}}
\put(90,0){\vector(-1,0){56}}
%\put(45,-30){\vector(-3,2){30}}
%\put(105,-10){\vector(-3,-2){30}}
%
\qbezier(40,30)(30,30)(20,30)
\qbezier(40,30)(40,20)(40,10)
\put(120,-44){\makebox(0,0){\normalsize{$\underset{\textstyle{
\mathrmbf{Tbl}_{\mathcal{A}}(\mathcal{S})}}{\underbrace{\rule{50pt}{0pt}}}$}}}
%%%%%%%%%%
\put(60,45){\makebox(0,0){\huge{
$\overset{\textit{\scriptsize{inflate}}}{\Rightarrow}$}}}
%%%%%%%%%%
\end{picture}
\end{tabular}}}
%
%%%%%%%%%%%%%%%%%%%%%%%%%%%%%%%%%%%%%%%%%%%%%%%%%%
\\&\\
%\multicolumn{2}{c}{$\mathrmbf{Tbl}(\mathcal{A})$}
%&
\\
\multicolumn{2}{c}{
{\footnotesize{$\mathcal{S}'\xrightarrow{\;h\;}\mathcal{S}$}}}
\\
${\scriptstyle\sum}_{h}(\mathcal{T})
\xleftarrow{\langle{h,1_{K}}\rangle}
\mathcal{T}$
&
{\footnotesize{
$\mathcal{T}'\xleftarrow{\langle{h,k}\rangle}
{h}^{\ast}(\mathcal{T}')$
}\normalsize}
\end{tabular}}}
\end{center}
\caption{
\texttt{FOLE} 
Adjoint Flow in $\mathrmbf{Tbl}(\mathcal{A})$
%Projection/Inflation in $\mathrmbf{Tbl}(\mathcal{A})$
}
\label{fig:proj:arity}
\end{figure}
\begin{figure}
\begin{center}
{{\begin{tabular}{c}
\setlength{\unitlength}{0.5pt}
\begin{picture}(280,140)(-8,-80)
\put(10,0){\makebox(0,0){\footnotesize{${\mathrmbf{Tbl}_{\mathcal{A}}(\mathcal{S}')}$}}}
\put(230,0){\makebox(0,0){\footnotesize{${\mathrmbf{Tbl}_{\mathcal{A}}(\mathcal{S})}$}}}
\put(120,50){\makebox(0,0){\scriptsize{$\textit{project}$}}}
\put(120,30){\makebox(0,0){\scriptsize{$\acute{\mathrmbfit{tbl}}_{\mathcal{A}}(h)$}}}
\put(120,-30){\makebox(0,0){\scriptsize{$\grave{\mathrmbfit{tbl}}_{\mathcal{A}}(h)$}}}
\put(120,-50){\makebox(0,0){\scriptsize{$\textit{inflate}$}}}
\put(120,0){\makebox(0,0){\footnotesize{${\;\dashv\;}$}}}
\put(180,14){\vector(-1,0){120}}
\put(60,-14){\vector(1,0){120}}
%
%\put(120,90){\makebox(0,0){\normalsize{${
%\overset{\textstyle{\mathrmbf{Tbl}(\mathcal{A})}}{\overbrace{\hspace{160pt}}}}$}}}
\put(120,-85){\makebox(0,0){\normalsize{${
\underset{\textstyle{\mathrmbf{Tbl}(\mathcal{A})}}{\underbrace{\hspace{160pt}}}}$}}}
%\put(-40,0){\makebox(0,0)[r]{\footnotesize{$
%\mathrmbf{Tbl}(\mathcal{S}')\left\{\rule{0pt}{25pt}\right.$}}}
%\put(275,0){\makebox(0,0)[l]{\footnotesize{$
%\left.\rule{0pt}{25pt}\right\}\mathrmbf{Tbl}(\mathcal{S})$}}}
%
\end{picture}
\end{tabular}}}
\end{center}
\caption{Adjoint Flow Factor}
% in $\mathrmbf{Tbl}(\mathcal{A})$}
\label{fig:adj:flo:A}
\end{figure}
\comment
{
\begin{center}
\begin{tabular}{c}
%{\footnotesize{
{{${\scriptstyle\sum}_{h}(\mathcal{T})
%\mathcal{T}'
%={\langle{\mathcal{S}',K,t'}\rangle}
\xleftarrow{\langle{h,1_{K}}\rangle}
%{\langle{\mathcal{S},K,t}\rangle}=
\mathcal{T}\;\;\;$}}
%}\normalsize}
\\
in $\mathrmbf{Tbl}(\mathcal{A})$
\end{tabular}
\end{center}

\begin{equation}\label{def:tbl:cxt}
{{\begin{picture}(120,10)(0,-4)
\put(60,0){\makebox(0,0){\footnotesize{$
\underset{\textstyle{\text{in}\;
\mathrmbf{Tbl}_{\mathcal{A}}({\scriptstyle\sum}_{f}(\mathcal{S}'))
}}
{\mathcal{T}'\xleftarrow{\;1\;}{\scriptstyle\sum}_{h}(\mathcal{T})}
{\;\;\;\;\;\;\;\;\rightleftarrows\;\;\;\;\;\;\;\;}
\underset{\textstyle{\text{in}\;\mathrmbf{Tbl}_{\mathcal{A}}(\mathcal{S})}}
{{h}^{\ast}(\mathcal{T}')\xleftarrow{\;k\,}\mathcal{T}}
$}}}
\end{picture}}}
\end{equation}}
%

%\newpage

%
\begin{flushleft}
{\fbox{\fbox{\footnotesize{\begin{minipage}{345pt}
{\underline{\textsf{How does projection work?}}}
When the index function is an inclusion 
$I'\xhookrightarrow{\,h\,}I$,
the projection ${\scriptstyle\sum}_{h}(\mathcal{T})$
of an $\mathcal{A}$-table
$\mathcal{T}={\langle{K,t}\rangle} 
\in \mathrmbf{Tbl}_{\mathcal{A}}(\mathcal{S})$
consists of the sub-tuples of $\mathcal{T}$ indexed by $I'$.
Hence,
%as mentioned before,
projection abridges the horizontal aspect of tables, 
ending with a subset of columns.
%a sub-header, or subset of columns.
%arity (index) function $I'\xrightarrow{h}I$.
%%``$\mathcal{A}$-table $\mathcal{T}'$ is the $h^{\text{th}}$ projection of $\mathcal{A}$-table $\mathcal{T}$''
%
In particular,
for any $\mathcal{A}$-table
$\mathcal{T}={\langle{\mathcal{S},K,t}\rangle}$,
an index $i \in I$
defines an arity (index) function $1\xrightarrow{i}I$,
thus forming an indexing $X$-sorted signature morphism 
${\langle{1,x}\rangle}\xrightarrow{\;i\;}\mathcal{S}$
from signature (sort) $1\xrightarrow{\,x\,}X$
satisfying the naturality condition 
$i{\,\cdot\,}s = s_{i} = x$.
%The $i^{\text{th}}$ projection
Projection along $1\xrightarrow{i}I$ 
defines an $\mathcal{A}$-table morphism 
%\[\mbox
{\footnotesize{
{{$\mathcal{T}_{i}={\langle{1,x,K,t_{i}}\rangle}\xleftarrow{\langle{i,1_{K}}\rangle}{\langle{\mathcal{S},K,t}\rangle}=\mathcal{T}$}}
}\normalsize}
%\]
%consists of 
%an indexing $X$-sorted signature morphism 
%${\langle{1,s_{i}}\rangle}\xrightarrow{\;i\;}\mathcal{S}$,
satisfying the naturality condition $t_{i}=t{\;\cdot\;}\mathrmbfit{tup}_{\mathcal{A}}(i)$,
which states that ``$t_{i}$ is the $i^{\text{th}}$ projection of $t$''.
%
%%%%%%%%%%%%%%%%%%%%%%%%%%%%%%%%%%%%%%%%%%%%%%%%%%%%%%%%%%%%%%%%%%%%%%
%%%%%%%%%%%%%%%%%%%%%%%%%%%%%%%%%%%%%%%%%%%%%%%%%%%%%%%%%%%%%%%%%%%%%%
\footnote{The $\mathcal{A}$-table
$\mathcal{T}_{i}={\langle{1,x,K,t_{i}}\rangle}$,
essentially the $i^{\text{th}}$-column of $\mathcal{T}$,
consists of 
signature ${\langle{1,x,X}\rangle}$
(sort $x = s_{i} \in X$),
the same set $K$ of keys, and
the tuple (data value) function 
$K\xrightarrow{t_{i}}\mathrmbfit{tup}_{\mathcal{A}}(1,s)=
\mathrmbfit{ext}_{\mathrmbf{List}(\mathcal{A})}(1,s)\cong
\mathrmbfit{ext}_{\mathcal{A}}(x)=\mathcal{A}_{x}\subseteq{Y}$.}
%%%%%%%%%%%%%%%%%%%%%%%%%%%%%%%%%%%%%%%%%%%%%%%%%%%%%%%%%%%%%%%%%%%%%%
%%%%%%%%%%%%%%%%%%%%%%%%%%%%%%%%%%%%%%%%%%%%%%%%%%%%%%%%%%%%%%%%%%%%%%
%
%This composes the $i^{\text{th}}$ tuple projection onto the table tuple map.
%Hence,
%Table $\mathcal{T}_{i}$ is the $i^{\text{th}}$ column of table $\mathcal{T}$.
%
%\end{minipage}}}}}
%\end{flushleft}
%
\newline\newline
%\begin{flushleft}
%{\fbox{\fbox{\footnotesize{\begin{minipage}{345pt}
{\underline{\textsf{How does inflation work?}}}
When the index function is an inclusion 
$I'\xhookrightarrow{\,h\,}I$,
an $\mathcal{A}$-relation
$\mathcal{R}'={\langle{R',i'}\rangle} 
\in 
%\mathrmbf{Tbl}_{\mathcal{S}'}({f}^{-1}(\mathcal{A}))
\mathrmbf{Rel}_{\mathcal{A}}(\mathcal{S}')
{\;\subseteq\;}
\mathrmbf{Tbl}_{\mathcal{A}}(\mathcal{S}')$
with
tuple subset
$R'{\,\subseteq\,}
%\mathrmbfit{tup}_{\mathcal{S}'}({f}^{-1}(\mathcal{A}))=
%\mathrmbfit{tup}_{\mathcal{A}}(\mathcal{S}')=
\mathrmbfit{tup}_{\mathcal{A}}(\mathcal{S}')$
is mapped (isomorphically) to the inflation relation
${h}^{\ast}(\mathcal{R}')
\in 
\mathrmbf{Rel}_{\mathcal{A}}(\mathcal{S})
{\;\subseteq\;}
\mathrmbf{Tbl}_{\mathcal{A}}(\mathcal{S})$
with tuple subset
$R'{\,\times\,}
\mathrmbfit{tup}_{\mathcal{A}}(I'',s'')
{\,\subseteq\,}
\mathrmbfit{tup}_{\mathcal{A}}(\mathcal{S})$,
where
$I''=I{{-}}I'$
is the index set complement
and
$s'' : I'' \rightarrow X$
is the restriction of signature function 
$s = [s',s''] : I=I'{+}I'' \rightarrow X$
to this complement.
The target tuple set factors as
$\mathrmbfit{tup}_{\mathcal{A}}(\mathcal{S})
=\mathrmbfit{tup}_{\mathcal{A}}(\mathcal{S}')
{\,\times\,}\mathrmbfit{tup}_{\mathcal{A}}(\mathcal{S}'')$.
The tuple set
$\mathrmbfit{tup}_{\mathcal{A}}(\mathcal{S}'')$
is the inflation referred to in the name.
Inflation enlarges the horizontal aspect of tables.
\end{minipage}}}}}
\end{flushleft}
\begin{aside}
Although we usually think of projection and inflation
along an injective index function,
here is an example along a surjective index function.
The copower $X$-signature $\mathcal{S}{+}\mathcal{S}$ 
in the opspan
{\footnotesize{$\mathcal{S}\xrightarrow{i_{1}\,}
{\mathcal{S}{+}\mathcal{S}}
\xleftarrow{\;i_{2}}\mathcal{S}$}\normalsize}
has inclusion index functions
{\footnotesize{$I\xhookrightarrow{i_{1}\,} 
I{+}I
\xhookleftarrow{\;i_{2}}I$.}\normalsize}
There is an $X$-signature morphism
$\mathcal{S}{+}\mathcal{S}
\xrightarrow
%[{[1,1]}]
{\;\triangledown\;}\mathcal{S}$
with (surjective) index function
{\footnotesize{$I{+}I\xrightarrow{\;\triangledown\;}I$}\normalsize}
that erases the origin:
$i_{1}{\;\cdot\;}\triangledown = 1_{I} = i_{2}{\;\cdot\;}\triangledown$.
Hence,
the projection
{\footnotesize{$\mathrmbf{Tbl}_{\mathcal{A}}(\mathcal{S}{\times}\mathcal{S})
{\;\xleftarrow
%[{\scriptscriptstyle\sum}_{\natural}]
{\acute{\mathrmbfit{tbl}}_{\mathcal{A}}(\triangledown)}
\;}
\mathrmbf{Tbl}_{\mathcal{A}}(\mathcal{S})$}}
is actually a ``creation'' or a ``duplication'',
and
the inflation
{\footnotesize{$\mathrmbf{Tbl}_{\mathcal{A}}(\mathcal{S}{\times}\mathcal{S})
{\;\xrightarrow
%[{\scriptscriptstyle\sum}_{\natural}]
{\grave{\mathrmbfit{tbl}}_{\mathcal{A}}(\triangledown)}
\;}
\mathrmbf{Tbl}_{\mathcal{A}}(\mathcal{S})$}}
is actually an ``erasure''.
\end{aside}
\begin{proposition}\label{project:inflate:preserve}
Projection preserves union $\vee$.
Projection is decreasing on intersection $\wedge$.
Inflation preserves union $\vee$ and intersection $\wedge$.
\end{proposition}
\begin{proof}
Projection is co-continuous and preserves order.
Inflation is continuous, co-continuous and preserves order.
\hfill\rule{5pt}{5pt}
\end{proof}
%

%%%%%%%%%%%%%%%%%%%%%%%%%%%%%%%%%%%%%%%%%%%%%%%%%%%%%%%%%%%%%%%%%%
%%%%%%%%%%%%%%%%%%%%%%%%%%%%%%%%%%%%%%%%%%%%%%%%%%%%%%%%%%%%%%%%%%
%\newpage
%\subsection{Ov}
%\label{subsec:overview}
\paragraph{\textbf{Application.}}
%%%%%%%%%%%%%%%%%%%%%%%%%%%%%%%%%%%%%%%%%%%%%%%%%%%%%%%%%%%%%%%%%
%%%%%%%%%%%%%%%%%%%%%%%%%%%%%%%%%%%%%%%%%%%%%%%%%%%%%%%%%%%%%%%%%
%
The fiber adjunction of tables
{\footnotesize{$\mathrmbf{Tbl}_{\mathcal{A}}(\mathcal{S}')
{\;\xleftarrow
{{\bigl\langle{\acute{\mathrmbfit{tbl}}_{\mathcal{A}}(h)
{\;\dashv\;}
\grave{\mathrmbfit{tbl}}_{\mathcal{A}}(h)}\bigr\rangle}}\;}
\mathrmbf{Tbl}_{\mathcal{A}}(\mathcal{S})$}}
for an $X$-sorted signature morphism
$\mathcal{S}'\xrightarrow{\;h\;}\mathcal{S}$
is used as follows.
\begin{itemize}
\item 
To define \underline{inflations} 
\newline\mbox{}\hfill
{\footnotesize{$
\mathrmbf{Tbl}_{\mathcal{A}}(\mathcal{S}_{1})
{\;\xrightarrow{\grave{\mathrmbfit{tbl}}_{\mathcal{A}}(h_{1})}\;}
\mathrmbf{Tbl}_{\mathcal{A}}(\mathcal{S})
{\;\xleftarrow{\grave{\mathrmbfit{tbl}}_{\mathcal{A}}(h_{2})}\;}
\mathrmbf{Tbl}_{\mathcal{A}}(\mathcal{S}_{2})$}}
\hfill\mbox{}\newline
from two peripheral signatures $\mathcal{S}_{1}$ and $\mathcal{S}_{2}$
to a central signature $\mathcal{S}$,
you need an opspan of
$\mathcal{S}_{1}
\rightarrow
\mathcal{S}
\leftarrow
\mathcal{S}_{2}$
of
$X$-sorted signature morphisms.
One way to get this is to assume
an $X$-sorted signature span 
$\mathcal{S}_{1}\xleftarrow{h_{1}}\mathcal{S}\xrightarrow{h_{2}}\mathcal{S}_{2}$
to define a
coproduct $X$-signature opspan
{\footnotesize{$\mathcal{S}_{1}\xrightarrow{\iota_{1}\,} 
{\mathcal{S}_{1}{+_{\mathcal{S}}}\mathcal{S}_{2}}
\xleftarrow{\;\iota_{2}}\mathcal{S}_{2}.$}\normalsize}
This is used by natural join 
in \S\,\ref{sub:sub:sec:nat:join};
hence, 
it is also used 
by Cartesian product, selection and select join 
there.
%in \S\,\ref{sub:sub:sec:nat:join}.
Furthermore,
it is used 
by semi-join
in \S\,\ref{sub:sub:sec:semi:join}, 
by anti-join in \S\,\ref{sub:sub:sec:anti:join}, 
by outer-join 
in \S\,\ref{sub:sub:sec:out:join},
and by division
in \S\,\ref{sub:sub:sec:div}.
\newline
\item 
To define \underline{projections} 
\newline\mbox{}\hfill
{\footnotesize{$
\mathrmbf{Tbl}_{\mathcal{A}}(\mathcal{S}_{1})
{\;\xrightarrow{\acute{\mathrmbfit{tbl}}_{\mathcal{A}}(h_{1})}\;}
\mathrmbf{Tbl}_{\mathcal{A}}(\mathcal{S})
{\;\xleftarrow{\acute{\mathrmbfit{tbl}}_{\mathcal{A}}(h_{2})}\;}
\mathrmbf{Tbl}_{\mathcal{A}}(\mathcal{S}_{2})$}}
\hfill\mbox{}\newline
from two peripheral signatures $\mathcal{S}_{1}$ and $\mathcal{S}_{2}$
to a central signature $\mathcal{S}$,
you need a span of
$\mathcal{S}_{1}
\xleftarrow{\;h_{1}\;}
\mathcal{S}
\xrightarrow{\;h_{2}\;}
\mathcal{S}_{2}$
of
$X$-sorted signature morphisms.
This is used by project-join
in 
\S\,\ref{sub:sub:sec:co-core}.
\end{itemize}
%
%\mbox{}\newline
%\item 

%%%%%%%%%%%%%%%%%%%%%%%%%%%%%%%%%%%%%%%%%%%%%%%%%%%%%%%%%%%%%
%
\newpage
\subsubsection{Fixed Signature.}
\label{sub:sub:sec:adj:flow:S}
%%%%%%%%%%%%%%%%%%%%%%%%%%%%%%%%%%%%%%%%%%%%%%%%%%%%%%%%%%%%%
%
Part of this paper, 
the adjoint flow in \S\,\ref{sub:sub:sec:adj:flow:A}
and the classic relational operations of \S\,\ref{sub:sec:comp:ops:type:dom},
deals with traditional relation algebra.
This assumes that we can manipulate the header part of table,
but that the data part is fixed.
In the \texttt{FOLE} representation 
the header part is represented by a signature, and
the data part is represented by a type domain.
Hence,
using \texttt{FOLE} to represent traditional relation algebra,
we fix the type domain, 
%(here called the base),
and allow the signature to vary.
However,
in the \texttt{FOLE} representation of relation databases,
we can manipulate both the header part \underline{and} the data part 
(the dual approach to relational algebra).
This section and \S\,\ref{sub:sec:comp:ops:sign}
explain this dual approach.
%There are two ways to do this:
%(1) add basic operations by varying the data part, and thus extending flowcharts; and/or
%(2) define base transformations.
%
%This section discusses the first way,
%whereas \S\,\ref{sub:sub:sec:transform}
%discusses the second.}
%
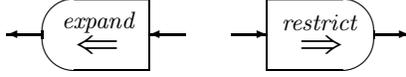
\begin{figure}
\begin{center}
{{\begin{tabular}{c}
\begin{picture}(180,14)(-3,0)
%%%%%%%%%%%%%%%%%%%%%%%%%%%%%%%%%%%%%%%%%%%%%%%%%%
%%%%%%%%%%%%%%%%%%%%%%%%%%%%%%%%%%%%%%%%%%%%%%%%%%
\put(20,-9){\begin{picture}(0,0)(0,0)
\setlength{\unitlength}{0.45pt}
%\thicklines
%\put(106,40){\makebox(0,0){\normalsize{$\boldsymbol{\circ}$}}}
%\put(4.7,40){\makebox(0,0){\normalsize{$\boldsymbol{\circ}$}}}
\put(40,10){\line(1,0){60}}
\put(40,70){\line(1,0){60}}
\put(100,70){\line(0,-1){60}}
\put(40,40){\oval(60,60)[bl]}
\put(40,40){\oval(60,60)[tl]}
\put(58,50){\makebox(0,0){\footnotesize{{\textit{{expand}}}}}}
\put(56,30){\makebox(0,0){\LARGE{${\Leftarrow}$}}}
\put(10,40){\vector(-1,0){30}}
\put(130,40){\vector(-1,0){30}}
\end{picture}}
%%%%%%%%%%%%%%%%%%%%%%%%%%%%%%%%%%%%%%%%%%%%%%%%%%
\put(105,-9){\begin{picture}(0,0)(0,0)
\setlength{\unitlength}{0.45pt}
%\thicklines
%\put(106,40){\makebox(0,0){\normalsize{$\boldsymbol{\circ}$}}}
%\put(5.8,40){\makebox(0,0){\normalsize{$\boldsymbol{\circ}$}}}
\put(10,10){\line(1,0){60}}
\put(10,70){\line(1,0){60}}
\put(10,70){\line(0,-1){60}}
\put(70,40){\oval(60,60)[br]}
\put(70,40){\oval(60,60)[tr]}
\put(55,50){\makebox(0,0){\footnotesize{{\textit{{restrict}}}}}}
\put(56,30){\makebox(0,0){\LARGE{${\Rightarrow}$}}}
\put(-20,40){\vector(1,0){30}}
\put(100,40){\vector(1,0){30}}
\end{picture}}
%%%%%%%%%%%%%%%%%%%%%%%%%%%%%%%%%%%%%%%%%%%%%%%%%%
%\put(29,7.5){\vector(-1,0){12}}
%\put(73,7.5){\vector(-1,0){12}}
\end{picture}
\end{tabular}}}
\end{center}
\caption{\texttt{FOLE} Adjoint Flow Operators: Base $\mathcal{S}$}
\label{fig:fole:adj:flow:proc:sign}
\end{figure}
%
%%%%%%%%%%%%%%%%%%%%%%%%%%%%%%%%%%%%%%%%%%%%%%%%%%%%%%%%%%%%
\comment{% temporary
Prop.\;\ref{prop:lim:tbl} 
of \S\,\ref{sub:sec:lim:colim:tbl}
defines the limit of a diagram of tables
in the full context 
$\mathrmbf{Tbl}$.
There we used adjoint flow along
signed domain morphisms
(Disp.\;\ref{fbr:adj:sign:dom:mor})
with 
left adjoint flow defining projection
and
right adjoint flow defining inflation.
However,
for simplicity of explanation of the various relational operations,
}% temporary
%%%%%%%%%%%%%%%%%%%%%%%%%%%%%%%%%%%%%%%%%%%%%%%%%%%%%%%%%%%%
%

In this section we use
a fixed signature
$\mathcal{S} = {\langle{I,x,X}\rangle}$.
Here, 
we define adjoint flow 
in the mid-sized table fiber context
$\mathrmbf{Tbl}(\mathcal{S})$.
%
%%%%%%%%%%%%%%%%%%%%%%%%%%%%%%%%%%%%%%%%%%%%%%%%%%%%%%%%%%%%
%%%%%%%%%%%%%%%%%%%%%%%%%%%%%%%%%%%%%%%%%%%%%%%%%%%%%%%%%%%%
\footnote{There is a 
discussion of signature indexing in \S\;3.3.1 of 
%the paper 
``The {\ttfamily FOLE} Table'' 
\cite{kent:fole:era:tbl}.}
%%%%%%%%%%%%%%%%%%%%%%%%%%%%%%%%%%%%%%%%%%%%%%%%%%%%%%%%%%%%
%%%%%%%%%%%%%%%%%%%%%%%%%%%%%%%%%%%%%%%%%%%%%%%%%%%%%%%%%%%%
%
Let
$\mathcal{A}'
%={\langle{Y',\models_{\mathcal{A}'}}\rangle} 
\xrightarrow{\;g\;} 
%{\langle{Y,\models_{g^{-1}(\mathcal{A}')}}\rangle}=
\mathcal{A}$
be an $X$-sorted type domain morphism
%
%%%%%%%%%%%%%%%%%%%%%%%%%%%%%%%%%%%%%%%%%%%%%%%%%%%%%%%%%%%%%%%%%%%%%%%%%%%%%%%%
%%%%%%%%%%%%%%%%%%%%%%%%%%%%%%%%%%%%%%%%%%%%%%%%%%%%%%%%%%%%%%%%%%%%%%%%%%%%%%%%
\footnote{This is described as
an $X$-sorted type domain morphism
$\mathcal{A}'
\xrightleftharpoons{{\langle{\mathrmit{1}_{X},g}\rangle}} 
\mathcal{A}$
in \cite{kent:fole:era:tbl}.}
%%%%%%%%%%%%%%%%%%%%%%%%%%%%%%%%%%%%%%%%%%%%%%%%%%%%%%%%%%%%%%%%%%%%%%%%%%%%%%%%
%%%%%%%%%%%%%%%%%%%%%%%%%%%%%%%%%%%%%%%%%%%%%%%%%%%%%%%%%%%%%%%%%%%%%%%%%%%%%%%%
%
with data value function $Y'\xleftarrow{\,g\,}Y$
satisfying the condition 
$\mathrmbfit{ext}_{\mathcal{A}'}{\;\cdot\;}g^{-1}
= \mathrmbfit{ext}_{\mathcal{A}}$;
or that
{{$g^{-1}(A'_{x}) = A_{x}$ for all $x \in X$}}.
This implies that
$A'_{x} \supseteq {\wp}g(A_{x})$ for all $x \in X$.
%
%%%%%%%%%%%%%%%%%%%%%%%%%%%%%%%%%%%%%%%%%%%%%%%%%%%%%%%%%%%%%%%%%%%%%%
%%%%%%%%%%%%%%%%%%%%%%%%%%%%%%%%%%%%%%%%%%%%%%%%%%%%%%%%%%%%%%%%%%%%%%
\footnote{Hence,
for an injective data value function $Y'\xhookleftarrow{\,g\,}Y$,
we have the inclusion
$A'_{x} \supseteq A_{x}$ for all $x \in X$.}
%%%%%%%%%%%%%%%%%%%%%%%%%%%%%%%%%%%%%%%%%%%%%%%%%%%%%%%%%%%%%%%%%%%%%%
%%%%%%%%%%%%%%%%%%%%%%%%%%%%%%%%%%%%%%%%%%%%%%%%%%%%%%%%%%%%%%%%%%%%%%
%
\comment{
\begin{itemize}
\item 
%$\mathcal{A}$
%is the sort-indexed collection of subsets of data values
$\mathcal{A} = \{ A_{x} \subseteq Y \mid x \in X \}$
\item 
%$\mathcal{A}'$
%is the sort-indexed collection of subsets of data values
$\mathcal{A}' = \{ A'_{x} \subseteq Y' \mid x \in X \}$
\item 
$\mathcal{A}'\xleftarrow{\;g\;}\mathcal{A}$
satisfies
$
g^{-1}(A'_{x}) = A_{x}
$
for all
$x \in X$.
%$X'\xrightarrow{\;\mathcal{A}'\;}{\wp}Y'\xrightarrow{\;g^{-1}\;}{\wp}Y 
%=X\xrightarrow{\;\mathcal{A}\;}{\wp}Y$ 
% 
\end{itemize}
}
%
%%%%%%%%%%%%%%%%%%%%%%%%%%%%%%%%%%%%%%%%%%%%%%%%%%%%%%%%%%%%
%%%%%%%%%%%%%%%%%%%%%%%%%%%%%%%%%%%%%%%%%%%%%%%%%%%%%%%%%%%%
\footnote{Hence,
\comment{
$\mathcal{A}' = {\langle{X,Y',\models_{\mathcal{A}'}}\rangle} 
\xrightleftharpoons{{\langle{1_{X},g}\rangle}} 
{\langle{X,Y,\models_{\mathcal{A}}}\rangle} = \mathcal{A}$
% = g^{-1}(\mathcal{A}')
is an infomorphism
satisfying the condition
$g(y){\;\models_{\mathcal{A}'}\;}x$
\underline{iff}
$y{\;\models_{\mathcal{A}}\;}x$
for any sort $x{\,\in\,}X$ and 
%source 
data value $y{\,\in\,}Y$.
This means that}
$\mathcal{A} = g^{-1}(\mathcal{A}')$,
as discussed by the \textbf{Yin} definition in 
\S\,\ref{sub:sec:tbl:comp}.}
%%%%%%%%%%%%%%%%%%%%%%%%%%%%%%%%%%%%%%%%%%%%%%%%%%%%%%%%%%%%
%%%%%%%%%%%%%%%%%%%%%%%%%%%%%%%%%%%%%%%%%%%%%%%%%%%%%%%%%%%%
%
Its tuple function 
\newline\mbox{}\hfill
{\footnotesize{$
\mathrmbfit{tup}_{\mathcal{S}}(\mathcal{A}')
\xleftarrow[{(\mbox{-})}{\,\cdot\,}g]
{\;\mathrmbfit{tup}_{\mathcal{S}}(g)\;}
\mathrmbfit{tup}_{\mathcal{S}}(\mathcal{A})$}\normalsize}
%%%%%%%%%%%%%%%%%%%%%%%%%%%%%%%%%%%%%%%%%%%%%%%%%%%%%%%%%%%%
%%%%%%%%%%%%%%%%%%%%%%%%%%%%%%%%%%%%%%%%%%%%%%%%%%%%%%%%%%%%
\comment{Visually,
$({\cdots\,}g(t_{i}){\,\cdots}{\,\mid\,}i{\,\in\,}I)
\mapsfrom
({\cdots\,}t_{i}{\,\cdots}{\,\mid\,}i{\,\in\,}I)$.}
%%%%%%%%%%%%%%%%%%%%%%%%%%%%%%%%%%%%%%%%%%%%%%%%%%%%%%%%%%%%
%%%%%%%%%%%%%%%%%%%%%%%%%%%%%%%%%%%%%%%%%%%%%%%%%%%%%%%%%%%%
\hfill\mbox{}\newline
defines by composition/pullback 
%(see \cite{kent:fole:era:tbl})
\comment{a table passage
$\mathrmbf{Tbl}_{\mathcal{S}}(\mathcal{A}')
\xrightarrow[{g}^{\ast}]
{\;\grave{\mathrmbfit{tbl}}_{\mathcal{S}}(g)\;}
\mathrmbf{Tbl}_{\mathcal{S}}(\mathcal{A})$,
which is the right adjoint of}
a fiber adjunction of tables
\begin{equation}\label{def:fbr:adj:type:dom:mor}
{{\begin{picture}(120,10)(0,-4)
\put(60,0){\makebox(0,0){\footnotesize{$
\mathrmbf{Tbl}_{\mathcal{S}}(\mathcal{A}')
{\;\xleftarrow
[{\bigl\langle{
{\scriptscriptstyle\sum}_{g}{\;\dashv\;}{g}^{\ast}
\bigr\rangle}}]
{{\bigl\langle{\acute{\mathrmbfit{tbl}}_{\mathcal{S}}(g)
{\;\dashv\;}
\grave{\mathrmbfit{tbl}}_{\mathcal{S}}(g)}\bigr\rangle}}
\;}
\mathrmbf{Tbl}_{\mathcal{S}}(\mathcal{A})
$.}}}
\end{picture}}}
\end{equation}
\begin{itemize}
\item[\textbf{restrict:}]  
The right adjoint 
$\mathrmbf{Tbl}_{\mathcal{S}}(\mathcal{A}')
\xrightarrow[{g}^{\ast}]
{\;\grave{\mathrmbfit{tbl}}_{\mathcal{S}}(g)\;}
\mathrmbf{Tbl}_{\mathcal{S}}(\mathcal{A})$
defines restriction.
%which restricts the vertical aspect of tables.
An $\mathcal{S}$-table
$\mathcal{T}'={\langle{K',t'}\rangle} \in \mathrmbf{Tbl}_{\mathcal{S}}(\mathcal{A}')$
%the projection along $I'\xrightarrow{h}I$ is 
is mapped to the $\mathcal{S}$-table
${g}^{\ast}(\mathcal{T}')
=\mathcal{T}
={\langle{K,t}\rangle} 
\in \mathrmbf{Tbl}_{\mathcal{S}}(\mathcal{A})$,
with its tuple function
$K\xrightarrow{t}\mathrmbfit{tup}_{\mathcal{S}}(\mathcal{A})$
defined by pullback,
%and
%satisfying the condition
$k{\;\cdot\;}t' = t{\;\cdot\;}\mathrmbfit{tup}_{\mathcal{S}}(g)$. 
Here we have ``vertically abridged'' 
%data value 
tuples in $\mathrmbf{List}(Y')$ 
%for data value set $Y$ 
by tuple pullback with the data value function 
$Y'\xleftarrow{\,g\,}Y$.
There is an $\mathcal{S}$-table morphism 
(RHS Fig.\;\ref{fig:restr:expan:sign})
%\[
{\mbox{\footnotesize{
$\mathcal{T}'={\langle{\mathcal{A}',K',t'}\rangle}
\xleftarrow{{\langle{g,k}\rangle}}
{\langle{\mathcal{A},K,t}\rangle}=\mathcal{T}$.
}\normalsize}}
%\]
%
%from a table $\mathcal{T}'$
%satisfying the naturality condition
%$t'=t{\;\cdot\;}\mathrmbfit{tup}_{\mathcal{A}}(h)$.
%This composes the $h^{\text{th}}$ tuple projection onto the table tuple map.
We say that
$\mathcal{S}$-table $\mathcal{T}={g}^{\ast}(\mathcal{T}')$ 
is the 
%$h^{\text{th}}$ 
\underline{restriction} of $\mathcal{S}$-table $\mathcal{T}'$
along $X$-sorted type domain morphism
$\mathcal{A}'
%={\langle{Y',\models_{\mathcal{A}'}}\rangle} 
\xrightarrow{\;g\;} 
%{\langle{Y,\models_{\mathcal{A}}}\rangle}=
\mathcal{A}$.
%
%{\fbox{\textbf{To Here 2/5/2020:}}}
%
\item[\textbf{expand:}]  
The left adjoint 
$\mathrmbf{Tbl}_{\mathcal{S}}(\mathcal{A}')
{\;\xleftarrow{
{\;\acute{\mathrmbfit{tbl}}_{\mathcal{S}}(g)\;}
}\;}
\mathrmbf{Tbl}_{\mathcal{S}}(\mathcal{A})$
defines expansion.
An $\mathcal{S}$-table
$\mathcal{T}={\langle{K,t}\rangle} \in 
\mathrmbf{Tbl}_{\mathcal{S}}(\mathcal{A})$
is mapped to the $\mathcal{S}$-table
${\scriptstyle\sum}_{g}(\mathcal{T})
= \mathcal{T}' = {\langle{K,t'}\rangle} 
\in \mathrmbf{Tbl}_{\mathcal{S}}(\mathcal{A}')$,
with its tuple function 
$K\xrightarrow{t'}\mathrmbfit{tup}_{\mathcal{S}}(\mathcal{A}')$
defined by composition,
$t' = t{\,\cdot\,}\mathrmbfit{tup}_{\mathcal{S}}(g)$.
Here we have ``vertically merged'' 
tuples in $\mathrmbf{List}(Y)$ 
with
tuples in $\mathrmbf{List}(Y')$ 
by tuple composition with the data value function 
$Y'\xleftarrow{\;g\;}Y$.
There is an $\mathcal{S}$-table morphism 
(LHS Fig.\;\ref{fig:restr:expan:sign})
%\[\mbox
{\footnotesize{$
\mathcal{T}'={\langle{\mathcal{A}',K,t'}\rangle}
\xleftarrow{\langle{1_{K},g}\rangle}
{\langle{\mathcal{A},K,t}\rangle}=\mathcal{T}
$.}\normalsize}
%\]
%
We say that
$\mathcal{S}$-table 
$\mathcal{T}'={\scriptstyle\sum}_{g}(\mathcal{T})$ 
is the 
%$h^{\text{th}}$ 
\underline{expansion} of $\mathcal{S}$-table $\mathcal{T}$
along $X'$-sorted type domain morphism
$\mathcal{A}'\xrightarrow{\;g\;}\mathcal{A}$.
\end{itemize}
\begin{figure}
\begin{center}
{{\begin{tabular}{c@{\hspace{75pt}}c}
%%%%%%%%%%%%%%%%%%%%%%%%%%%%%%%%%%%%%%%%%%%%%%%%%%
{{\begin{tabular}{c}
\setlength{\unitlength}{0.65pt}
%\begin{picture}(120,120)(0,-20)
\begin{picture}(120,90)(0,-5)
\put(0,80){\makebox(0,0){\footnotesize{$K$}}}
\put(120,80){\makebox(0,0){\footnotesize{$K$}}}
\put(0,-10){\makebox(0,0){\footnotesize{$
\underset{\textstyle{\subseteq \mathrmbf{List}(Y')}}
{\mathrmbfit{tup}_{\mathcal{S}}(\mathcal{A}')}$}}}
\put(140,-10){\makebox(0,0){\footnotesize{$
\underset{\textstyle{\subseteq \mathrmbf{List}(Y)}}
{\mathrmbfit{tup}_{\mathcal{S}}(\mathcal{A})}$}}}
\put(63,95){\makebox(0,0){\scriptsize{$1_{K}$}}}
\put(70,-12){\makebox(0,0){\scriptsize{$\mathrmbfit{tup}_{\mathcal{S}}(g)$}}}
\put(-6,40){\makebox(0,0)[r]{\scriptsize{$t'$}}}
\put(128,40){\makebox(0,0)[l]{\scriptsize{$t$}}}
\put(0,65){\vector(0,-1){50}}
\put(120,65){\vector(0,-1){50}}
\put(100,80){\vector(-1,0){80}}
\put(90,0){\vector(-1,0){50}}
%\put(45,-30){\vector(-3,2){30}}
%\put(105,-10){\vector(-3,-2){30}}
%%%%%%%%%%
\put(0,-44){\makebox(0,0){\normalsize{$\underset{\textstyle{
\mathrmbf{Tbl}_{\mathcal{S}}(\mathcal{A}')
}}{\underbrace{\rule{50pt}{0pt}}}$}}}
%%%%%%%%%%
\put(60,45){\makebox(0,0){\huge{
$\overset{\textit{\scriptsize{expand}}}{\Leftarrow}$}}}
%%%%%%%%%%
\put(180,40){\makebox(0,0){\Large{$\rightleftarrows$}}}
%%%%%%%%%%
\end{picture}
\end{tabular}}}
%%%%%%%%%%%%%%%%%%%%%%%%%%%%%%%%%%%%%%%%%%%%%%%%%%
&
%%%%%%%%%%%%%%%%%%%%%%%%%%%%%%%%%%%%%%%%%%%%%%%%%%
{{\begin{tabular}{c}
\setlength{\unitlength}{0.65pt}
\begin{picture}(120,90)(0,-5)
\put(0,80){\makebox(0,0){\footnotesize{$K'$}}}
\put(120,80){\makebox(0,0){\footnotesize{$K$}}}
\put(0,-10){\makebox(0,0){\footnotesize{$
\underset{\textstyle{\subseteq \mathrmbf{List}(Y')}}
{\mathrmbfit{tup}_{\mathcal{S}}(\mathcal{A}')}$}}}
\put(140,-10){\makebox(0,0){\footnotesize{$
\underset{\textstyle{\subseteq \mathrmbf{List}(Y)}}
{\mathrmbfit{tup}_{\mathcal{S}}(\mathcal{A})}$}}}
\put(63,90){\makebox(0,0){\scriptsize{$k$}}}
\put(70,-12){\makebox(0,0){\scriptsize{$
\mathrmbfit{tup}_{\mathcal{S}}(g)$}}}
\put(-6,40){\makebox(0,0)[r]{\scriptsize{$t'$}}}
\put(128,40){\makebox(0,0)[l]{\scriptsize{$t$}}}
\put(0,65){\vector(0,-1){50}}
\put(120,65){\vector(0,-1){50}}
\put(100,80){\vector(-1,0){80}}
\put(90,0){\vector(-1,0){50}}
\qbezier(40,30)(30,30)(20,30)
\qbezier(40,30)(40,20)(40,10)
\put(120,-44){\makebox(0,0){\normalsize{$\underset{\textstyle{
\mathrmbf{Tbl}_{\mathcal{S}}(\mathcal{A})}}{\underbrace{\rule{50pt}{0pt}}}$}}}
%%%%%%%%%%
\put(60,45){\makebox(0,0){\huge{
$\overset{\textit{\scriptsize{restrict}}}{\Rightarrow}$}}}
%%%%%%%%%%
\end{picture}
\end{tabular}}}
%
%%%%%%%%%%%%%%%%%%%%%%%%%%%%%%%%%%%%%%%%%%%%%%%%%%
\\&\\
%&\\
%\multicolumn{2}{c}{$\mathrmbf{Tbl}(\mathcal{S})$}
\\
\multicolumn{2}{c}{{\footnotesize{$
\mathcal{A}'\xrightarrow{\;g\,}\mathcal{A}$}}}
\\
${\scriptstyle\sum}_{g}(\mathcal{T})
\xleftarrow{\langle{g,1_{K}}\rangle}
\mathcal{T}$
&
{\footnotesize{
$\mathcal{T}'\xleftarrow{\langle{g,k}\rangle}
{g}^{\ast}(\mathcal{T}')$
}\normalsize}
\end{tabular}}}
\end{center}
\caption{\texttt{FOLE} 
Adjoint Flow in $\mathrmbf{Tbl}(\mathcal{S})$}
%Restriction/Expansion in $\mathrmbf{Tbl}(\mathcal{S})$}
\label{fig:restr:expan:sign}
\end{figure}
\begin{figure}
\begin{center}
{{\begin{tabular}{c}
\setlength{\unitlength}{0.5pt}
\begin{picture}(280,140)(-8,-80)
\put(10,0){\makebox(0,0){\footnotesize{${\mathrmbf{Tbl}_{\mathcal{S}}(\mathcal{A}')}$}}}
\put(230,0){\makebox(0,0){\footnotesize{${\mathrmbf{Tbl}_{\mathcal{S}}(\mathcal{A})}$}}}
\put(120,50){\makebox(0,0){\scriptsize{$\textit{expand}$}}}
\put(120,30){\makebox(0,0){\scriptsize{$\acute{\mathrmbfit{tbl}}_{\mathcal{S}}(g)$}}}
\put(120,-30){\makebox(0,0){\scriptsize{$\grave{\mathrmbfit{tbl}}_{\mathcal{S}}(g)$}}}
\put(120,-50){\makebox(0,0){\scriptsize{$\textit{restrict}$}}}
\put(120,0){\makebox(0,0){\footnotesize{${\;\dashv\;}$}}}
\put(180,14){\vector(-1,0){120}}
\put(60,-14){\vector(1,0){120}}
%
%\put(120,90){\makebox(0,0){\normalsize{${
%\overset{\textstyle{\mathrmbf{Tbl}(\mathcal{A})}}{\overbrace{\hspace{160pt}}}}$}}}
\put(120,-85){\makebox(0,0){\normalsize{${
\underset{\textstyle{\mathrmbf{Tbl}(\mathcal{S})}}{\underbrace{\hspace{160pt}}}}$}}}
%\put(-40,0){\makebox(0,0)[r]{\footnotesize{$
%\mathrmbf{Tbl}(\mathcal{A}')\left\{\rule{0pt}{25pt}\right.$}}}
%\put(275,0){\makebox(0,0)[l]{\footnotesize{$
%\left.\rule{0pt}{25pt}\right\}\mathrmbf{Tbl}(\mathcal{A})$}}}
%
\end{picture}
\end{tabular}}}
\end{center}
\caption{Adjoint Flow Factor}
%in $\mathrmbf{Tbl}(\mathcal{S})$}
\label{fig:adj:flo:S}
\end{figure}
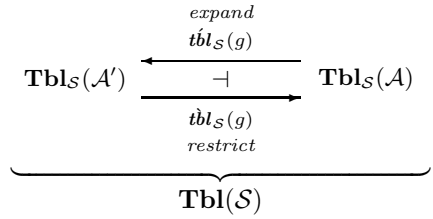
\begin{flushleft}
{\fbox{\fbox{\footnotesize{\begin{minipage}{345pt}
{\underline{\textsf{How does restriction work?}}}
We can think of restriction as a \emph{validation} process.
When the data value function is an inclusion
$Y'\xhookleftarrow{\,g\,}Y$,
we can regard the values in 
$Y'{\,\setminus\,}Y$ as being inauthentic and non-usable.
For each sort $x \in X$,
the data value function 
%$Y'\xleftarrow{\;g\;}Y$
maps the source data type 
$\mathcal{A}'_{x}$ to the
restricted target data type
$\mathcal{A}_{x} = \mathcal{A}'_{x} \cap Y$,
and the restriction ${g}^{\ast}(\mathcal{T}')$
consists of only those tuples of $\mathcal{T}'$ 
with data values in $Y$;
all other tuples are omitted.
Hence,
restriction abridges the vertical aspect of tables, 
ending with a subset of rows.
%
%\end{minipage}}}}}
%\end{flushleft}
%
\newline\newline
%
%\begin{flushleft}
%{\fbox{\fbox{\footnotesize{\begin{minipage}{345pt}
{\underline{\textsf{How does expansion work?}}}
When the data value function is an inclusion 
$Y'\xhookleftarrow{\,g\,}Y$,
the expansion ${\scriptstyle\sum}_{g}(\mathcal{T})$ 
retains the same tuples in the target table,
but places the table in the larger context $\mathrmbf{List}(Y')$
of data tuples.
In this sense,
expansion ``extends'' the vertical aspect of tables.
With an inclusion data function,
expansion followed by restriction gives the original table.
\end{minipage}}}}}
\end{flushleft}
\begin{aside}
Although we usually think of expansion and restriction
along an injective data value function,
here is an example along a surjective data value function.
The power $X$-type domain $\mathcal{A}{\times}\mathcal{A}$ 
in the span
{\footnotesize{$\mathcal{A}\xleftarrow{i_{1}\,}
{\mathcal{A}{\times}\mathcal{A}}
\xrightarrow{\;i_{2}}\mathcal{A}$}\normalsize}
has inclusion data value functions
{\footnotesize{$Y\xhookrightarrow{i_{1}\,} 
Y{+}Y
\xhookleftarrow{\;i_{2}}Y$.}\normalsize}
There is an $X$-type domain morphism
$\mathcal{A}
\xrightarrow
%[{[1,1]}]
{\;\triangledown\;}
\mathcal{A}{\times}\mathcal{A}$
with (surjective) data value function
{\footnotesize{$Y\xleftarrow{\;\triangledown\;}Y{+}Y$}\normalsize}
that erases the origin:
$i_{1}{\;\cdot\;}\triangledown = 1_{Y} = i_{2}{\;\cdot\;}\triangledown$.
Hence,
the expansion
{\footnotesize{$\mathrmbf{Tbl}_{\mathcal{S}}(\mathcal{A})
{\;\xleftarrow
%[{\scriptscriptstyle\sum}_{\natural}]
{\acute{\mathrmbfit{tbl}}_{\mathcal{S}}(\triangledown)}
\;}
\mathrmbf{Tbl}_{\mathcal{S}}(\mathcal{A}{\times}\mathcal{A})$}}
is actually an ``erasure'',
and
the restriction
{\footnotesize{$\mathrmbf{Tbl}_{\mathcal{S}}(\mathcal{A})
{\;\xrightarrow
%[{\scriptscriptstyle\sum}_{\natural}]
{\grave{\mathrmbfit{tbl}}_{\mathcal{S}}(\triangledown)}
\;}
\mathrmbf{Tbl}_{\mathcal{S}}(\mathcal{A}{\times}\mathcal{A})$}}
is actually a ``creation'' or a ``duplication''.
\end{aside}
%

%%%%%%%%%%%%%%%%%%%%%%%%%%%%%%%%%%%%%%%%%%%%%%%%%%%%%%%%%%%%%%%%%%
%%%%%%%%%%%%%%%%%%%%%%%%%%%%%%%%%%%%%%%%%%%%%%%%%%%%%%%%%%%%%%%%%%
%\newpage
%\subsection{Ov}
%\label{subsec:overview}
\paragraph{\textbf{Application.}}
%%%%%%%%%%%%%%%%%%%%%%%%%%%%%%%%%%%%%%%%%%%%%%%%%%%%%%%%%%%%%%%%%
%%%%%%%%%%%%%%%%%%%%%%%%%%%%%%%%%%%%%%%%%%%%%%%%%%%%%%%%%%%%%%%%%
%
The fiber adjunction of tables
%
%\newline\mbox{}\hfill
%\rule[8pt]{0pt}{10pt}
{\footnotesize{$
\mathrmbf{Tbl}_{\mathcal{S}}(\mathcal{A}')
{\;\xleftarrow
%[{\bigl\langle{{\scriptscriptstyle\sum}_{g}{\;\dashv\;}{g}^{\ast}\bigr\rangle}}]
{{\bigl\langle{\acute{\mathrmbfit{tbl}}_{\mathcal{S}}(g)
{\;\dashv\;}
\grave{\mathrmbfit{tbl}}_{\mathcal{S}}(g)}\bigr\rangle}}\;}
\mathrmbf{Tbl}_{\mathcal{S}}(\mathcal{A})$}}
%\hfill\mbox{}\newline
for an $X$-sorted type domain morphism
$\mathcal{A}'\xrightarrow{\;g\;}\mathcal{A}$
is used as follows.
\begin{itemize}
\item 
To define \underline{expansions} 
\newline\mbox{}\hfill
{\footnotesize{$
\mathrmbf{Tbl}_{\mathcal{S}}(\mathcal{A}_{1})
{\;\xrightarrow{\acute{\mathrmbfit{tbl}}_{\mathcal{S}}(g_{1})}\;}
\mathrmbf{Tbl}_{\mathcal{S}}(\mathcal{A})
{\;\xleftarrow{\acute{\mathrmbfit{tbl}}_{\mathcal{S}}(g_{2})}\;}
\mathrmbf{Tbl}_{\mathcal{S}}(\mathcal{A}_{2})$}}
\hfill\mbox{}\newline
from two peripheral type domains $\mathcal{A}_{1}$ and $\mathcal{A}_{2}$
to a central type domain $\mathcal{A}$,
you need a span of
$\mathcal{A}_{1}
\leftarrow
\mathcal{A}
\rightarrow
\mathcal{A}_{2}$
of
$X$-sorted type domain morphisms.
One way to get this is to assume
an $X$-sorted type domain opspan 
$\mathcal{A}_{1}\xrightarrow{g_{1}}\mathcal{A}\xleftarrow{g_{2}}\mathcal{A}_{2}$
to define a
product $X$-type domain span
{\footnotesize{$\mathcal{A}_{1}\xleftarrow{\pi_{1}\,} 
{\mathcal{A}_{1}{\times_{\mathcal{A}}}\mathcal{A}_{2}}
\xrightarrow{\;\pi_{2}}\mathcal{A}_{2}.$}\normalsize}
This is used 
by data-type join 
in \S\,\ref{sub:sub:sec:boole:join}; 
hence, it is also used 
by disjoint sum and data-type meet
there. 
%in \S\,\ref{sub:sub:sec:boole:join}. 
This is used 
by data-type semi-join 
in \S\,\ref{sub:sub:sec:boole:semi:join}; 
hence, it is also used by data-type semi-meet
there. 
%in \S\,\ref{sub:sub:sec:boole:semi:join}. 
%Furthermore,
Finally,
\comment{
it is used 
by data-type anti-meet 
in \S\,\ref{sub:sub:sec:boole:anti:meet},
and 
}
by subtraction
in \S\,\ref{sub:sub:sec:subtrac}.
\newline
\item 
To define \underline{restrictions} 
\newline\mbox{}\hfill
{\footnotesize{$
\mathrmbf{Tbl}_{\mathcal{S}}(\mathcal{A}_{1})
{\;\xrightarrow{\grave{\mathrmbfit{tbl}}_{\mathcal{S}}(g_{1})}\;}
\mathrmbf{Tbl}_{\mathcal{S}}(\mathcal{A})
{\;\xleftarrow{\grave{\mathrmbfit{tbl}}_{\mathcal{S}}(g_{2})}\;}
\mathrmbf{Tbl}_{\mathcal{S}}(\mathcal{A}_{2})$}}
\hfill\mbox{}\newline
from two peripheral type domains $\mathcal{A}_{1}$ and $\mathcal{A}_{2}$
to a central type domain $\mathcal{A}$,
you need an opspan of
$\mathcal{A}_{1}
\xrightarrow{\;g_{1}\;}
\mathcal{A}
\xleftarrow{\;g_{2}\;}
\mathcal{A}_{2}$
of
$X$-sorted type domain morphisms.
This is used 
by filter-join
in \S\,\ref{sub:sub:sec:filtered:join}.
%\newline
%\newline\rule{10pt}{1pt}\newline
%\newline\rule{10pt}{1pt}
%\newline
%
\end{itemize}
%
%\end{itemize}
%

%
%%%%%%%%%%%%%%%%%%%%%%%%%%%%%%%%%%%%%%%%%%%%%%%%%%%%%%%%%%%%%
\newpage
\subsubsection{All Tables.}
\label{sub:sub:sec:flow:sign:dom:mor}
%sub:sub:sec:flow:factor
%%%%%%%%%%%%%%%%%%%%%%%%%%%%%%%%%%%%%%%%%%%%%%%%%%%%%%%%%%%%%

%\mbox{}{{\textbf{Tabular Flow Adjunction:}}}
%\S\,\ref{sub:sub:sec:tbl:flow:adj}
%\begin{itemize}
%\item define tuple function
%\item define (Tables and Relations)\;\ref{sub:sec:tbl:rel}
%\item state tabular flow adjunction
%\item factor signed domain morphism
%\item factor tuple function
%\item factor adjoint flow (All Tables)\;\ref{sub:sub:sec:flow:sign:dom:mor}
%\end{itemize}
%

%
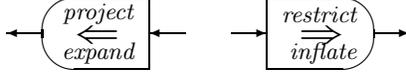
\begin{figure}
\begin{center}
{{\begin{tabular}{c}
\begin{picture}(180,14)(-3,0)
%%%%%%%%%%%%%%%%%%%%%%%%%%%%%%%%%%%%%%%%%%%%%%%%%%
%%%%%%%%%%%%%%%%%%%%%%%%%%%%%%%%%%%%%%%%%%%%%%%%%%
\put(20,-9){\begin{picture}(0,0)(0,0)
\setlength{\unitlength}{0.45pt}
%\thicklines
%\put(106,40){\makebox(0,0){\normalsize{$\boldsymbol{\circ}$}}}
%\put(4.7,40){\makebox(0,0){\normalsize{$\boldsymbol{\circ}$}}}
\put(40,10){\line(1,0){60}}
\put(40,70){\line(1,0){60}}
\put(100,70){\line(0,-1){60}}
\put(40,40){\oval(60,60)[bl]}
\put(40,40){\oval(60,60)[tl]}
\put(58,56){\makebox(0,0){\footnotesize{{\textit{{project}}}}}}
\put(56,37){\makebox(0,0){\LARGE{${\Leftarrow}$}}}
\put(58,23){\makebox(0,0){\footnotesize{{\textit{{expand}}}}}}
\put(10,40){\vector(-1,0){30}}
\put(130,40){\vector(-1,0){30}}
\end{picture}}
%%%%%%%%%%%%%%%%%%%%%%%%%%%%%%%%%%%%%%%%%%%%%%%%%%
\put(105,-9){\begin{picture}(0,0)(0,0)
\setlength{\unitlength}{0.45pt}
%\thicklines
%\put(106,40){\makebox(0,0){\normalsize{$\boldsymbol{\circ}$}}}
%\put(5.8,40){\makebox(0,0){\normalsize{$\boldsymbol{\circ}$}}}
\put(10,10){\line(1,0){60}}
\put(10,70){\line(1,0){60}}
\put(10,70){\line(0,-1){60}}
\put(70,40){\oval(60,60)[br]}
\put(70,40){\oval(60,60)[tr]}
\put(55,56){\makebox(0,0){\footnotesize{{\textit{{restrict}}}}}}
\put(56,37){\makebox(0,0){\LARGE{${\Rightarrow}$}}}
\put(58,23){\makebox(0,0){\footnotesize{{\textit{{inflate}}}}}}
\put(-20,40){\vector(1,0){30}}
\put(100,40){\vector(1,0){30}}
\end{picture}}
%%%%%%%%%%%%%%%%%%%%%%%%%%%%%%%%%%%%%%%%%%%%%%%%%%
%\put(29,7.5){\vector(-1,0){12}}
%\put(73,7.5){\vector(-1,0){12}}
\end{picture}
\end{tabular}}}
\end{center}
\caption{\texttt{FOLE} Adjoint Flow Operators}
\label{fig:fole:adj:flow:proc:sign:dom}
\end{figure}

In this section we define adjoint flow 
in the full table context
$\mathrmbf{Tbl}$.
%of the signed domain morphism 
%in \S\,\ref{sub:sec:sign:dom}.
%
Let
%\[\mbox
{\footnotesize{$
%\mathcal{D}'=
{\langle{\mathcal{S}',\mathcal{A}'}\rangle}
\xrightarrow{{\langle{h,f,g}\rangle}}
{\langle{\mathcal{S},\mathcal{A}}\rangle}
%=\mathcal{D}
$}\normalsize}
%\]
be a signed domain morphism.
%the sort function $X'\xrightarrow{f}X$ and 
%with the arity (index) function $I'\xrightarrow{h}I$,
%which satisfies the condition 
%$h{\,\cdot\,}s = s'$.
%
By Prop.\;\ref{prop:tup:func:ident} 
in \S\;\ref{sub:sec:tbl:rel},
%of \S\,\ref{sub:sec:sign:dom},
its tuple function
\comment{$\mathrmbfit{tup}(I_{2},s_{2},\mathcal{A}_{2})
\xleftarrow[(h{\cdot}{(\mbox{-})})\cdot({(\mbox{-})}{\cdot}g)]{\mathrmbfit{tup}(h,f,g)}
\mathrmbfit{tup}(I_{1},s_{1},\mathcal{A}_{1})$}
factors into two parts
(Fig.\;\ref{tup:func:idents}).
\begin{center}
{{\begin{tabular}{c}
{{\setlength{\extrarowheight}{4pt}{\footnotesize{$
\begin{array}[c]{@{\hspace{5pt}}r@{\hspace{5pt}}c@{\hspace{5pt}}
c@{\hspace{5pt}}c@{\hspace{5pt}}l@{\hspace{5pt}}}
\overset{\textstyle{
\mathrmbfit{tup}_{\mathcal{A}'}(\mathcal{S}')}}
{\underset{\textstyle{\mathrmbfit{tup}_{\mathcal{S}'}(\mathcal{A}')}}
{=\rule[6pt]{0pt}{1pt}}}
&
\xleftarrow[\;\mathrmbfit{tup}_{\mathcal{S}'}(g)\;]
{\grave{\tau}_{{\langle{f,g}\rangle}}(\mathcal{S}')}
&
\overset{\textstyle{
\mathrmbfit{tup}_{\mathcal{A}}({\scriptstyle\sum}_{f}(\mathcal{S}')}}
{\underset{\textstyle{
\mathrmbfit{tup}_{\mathcal{S}'}({g}^{-1}(\mathcal{A}'))}}
{=\rule[6pt]{0pt}{1pt}}}
&
\xleftarrow[\tau_{{\langle{h,f}\rangle}}(\mathcal{A})]
{\mathrmbfit{tup}_{\mathcal{A}}(h)}
&
\overset{\textstyle{
\mathrmbfit{tup}_{\mathcal{A}}(\mathcal{S})}}
{\underset{\textstyle{
\mathrmbfit{tup}_{\mathcal{S}}(\mathcal{A})
}}
{=\rule[6pt]{0pt}{1pt}}}
\\
%\hline
\end{array}$}}}}
%%%%%%%%%%%%%%%%%%%%%%%%%%%%%%%%%%%%%%%%%%%%%%%%%%%%%%%%%%%%
\end{tabular}}}
\end{center}
Hence, the fiber adjunction of tables
%%%%%%%%%%
\comment{
\[{\mbox{
{\footnotesize{$
\mathrmbf{Tbl}_{\mathcal{S}'}(\mathcal{A}')
{\;\xleftarrow
[{\bigl\langle{{\scriptscriptstyle\sum}_{{\langle{h,f,g}\rangle}}
{\;\dashv\;}
{{\langle{h,f,g}\rangle}}^{\ast}\bigr\rangle}}]
{{\bigl\langle{\acute{\mathrmbfit{tbl}}(h,f,g)
{\;\dashv\;}
\grave{\mathrmbfit{tbl}}(h,f,g)}\bigr\rangle}}\;}
\mathrmbf{Tbl}_{\mathcal{A}}(\mathcal{S})$
}\normalsize}}}\]
defined by 
the tuple function,}
%%%%%%%%%%
also factors into two parts:
%in two directions of motion%:
either forward (right-to-left) by composition
or backward (left-to-right) by pullback.
%
%%%%%%%%%%%%%%%%%%%%%%%%%%%%%%%%%%%%%%%%%%%%%%%%%%%%%%%%%%%%
%%%%%%%%%%%%%%%%%%%%%%%%%%%%%%%%%%%%%%%%%%%%%%%%%%%%%%%%%%%%
\footnote{There is an $h$-closure for $h$-projection
and an $h$-interior for inflation.
\newline
There is a $g$-closure for $g$-expansion
and a $g$-interior for restriction.}
%%%%%%%%%%%%%%%%%%%%%%%%%%%%%%%%%%%%%%%%%%%%%%%%%%%%%%%%%%%%
%%%%%%%%%%%%%%%%%%%%%%%%%%%%%%%%%%%%%%%%%%%%%%%%%%%%%%%%%%%%
%
%The above tuple function
%defines by composition/pullback
%(see \cite{kent:fole:era:tbl})
%Hence,
%the tuple function of a signed domain morphism
%defines by composition/pullback
%(see \cite{kent:fole:era:tbl})
%Kent, R.E.:
%``The {\ttfamily FOLE} Table''.
%the fiber adjunction of tables 
%in \S\ref{sub:sec:flow:factor}.
%{\fbox{\textbf{better reference?}}}
%
\begin{equation}\label{fbr:adj:sign:dom:mor:factor}
{{\begin{picture}(300,20)(0,-5)
\put(0,0){\scriptsize{$
\overset{\textstyle{
\mathrmbf{Tbl}_{\mathcal{A}'}(\mathcal{S}')}}
{\underset{\textstyle{
\mathrmbf{Tbl}_{\mathcal{S}'}(\mathcal{A}')
}}
{=\rule[4pt]{0pt}{1pt}}}
{\;\xleftarrow
[{\bigl\langle{\acute{\mathrmbfit{tbl}}_{\mathcal{S}'}(g)
{\;\dashv\;}
\grave{\mathrmbfit{tbl}}_{\mathcal{S}'}(g)}\bigr\rangle}]
{{\makebox(0,0){\Large{$
\overset{\textit{\scriptsize{expand}}}
{\Leftarrow}
\!\!\text{/}\!\!
\overset{\textit{\scriptsize{restrict}}}
{\Rightarrow}
$}}}\rule[-8pt]{0pt}{10pt}}
\;}
\overset{\textstyle{
\mathrmbf{Tbl}_{\mathcal{A}}({\scriptstyle\sum}_{f}(\mathcal{S}')}}
{\underset{\textstyle{
\mathrmbf{Tbl}_{\mathcal{S}'}({g}^{-1}(\mathcal{A}'))}}
{=\rule[4pt]{0pt}{1pt}}}
{\;\xleftarrow
[{\bigl\langle{\acute{\mathrmbfit{tbl}}_{\mathcal{A}}(h)
{\;\dashv\;}
\grave{\mathrmbfit{tbl}}_{\mathcal{A}}(h)}\bigr\rangle}]
{{\makebox(0,0){\Large{$
\overset{\textit{\scriptsize{project}}}
{\Leftarrow}
\!\!\text{/}\!\!
\overset{\textit{\scriptsize{inflate}}}
{\Rightarrow}
$}}}\rule[-8pt]{0pt}{10pt}}
\;}
\overset{\textstyle{
\mathrmbf{Tbl}_{\mathcal{A}}(\mathcal{S})}}
{\underset{\textstyle{
\mathrmbf{Tbl}_{\mathcal{S}}(\mathcal{A})
}}
{=\rule[4pt]{0pt}{1pt}}}
$.}}
\end{picture}}}
\end{equation}
\begin{itemize}
\item[\textbf{project/expand:}]  
The left adjoint
$\mathrmbf{Tbl}_{\mathcal{S}'}(\mathcal{A}')
{\;\xleftarrow[{\scriptscriptstyle\sum}_{h,f,g}]
{\;\acute{\mathrmbfit{tbl}}(h,f,g)\;}\;}
\mathrmbf{Tbl}_{\mathcal{A}}(\mathcal{S})$
defines projection-expansion.
An $\mathcal{A}$-table
$\mathcal{T}={\langle{K,t}\rangle} \in 
\mathrmbf{Tbl}_{\mathcal{A}}(\mathcal{S})$
is mapped to the $\mathcal{S}'$-table
${\scriptstyle\sum}_{h,f,g}(\mathcal{T})
= \mathcal{T}' = {\langle{K,t'}\rangle} 
\in \mathrmbf{Tbl}_{\mathcal{S}'}(\mathcal{A}')$,
with its tuple function 
$K\xrightarrow{t'}\mathrmbfit{tup}_{\mathcal{S}'}(\mathcal{A}')$
defined by composition,
$t' = t{\,\cdot\,}\mathrmbfit{tup}(h,f,g)$.
Here we have ``horizontally abridged'' and then ``vertically extended'' tuples in 
$\mathrmbf{List}(Y)$ 
by tuple composition with the signed domain morphism
${\langle{\mathcal{S}',\mathcal{A}'}\rangle}
\xrightarrow{{\langle{h,f,g}\rangle}}
{\langle{\mathcal{S},\mathcal{A}}\rangle}$.
There is a table morphism 
(LHS Fig.\;\ref{fig:flow:sign:dom})
%\[\mbox
{\footnotesize{$
\mathcal{T}'={\langle{\mathcal{S}',\mathcal{A}',K,t'}\rangle}
\xleftarrow{\langle{{\langle{h,f,g}\rangle},1_{K}}\rangle}
{\langle{\mathcal{S},\mathcal{A},K,t}\rangle}=\mathcal{T}$
.}\normalsize}
%\]
We say that
table 
$\mathcal{T}'={\scriptstyle\sum}_{h,f,g}(\mathcal{T})$ 
is the 
\underline{projection-expansion} of 
table $\mathcal{T}$
along signed domain morphism
${\langle{\mathcal{S}',\mathcal{A}'}\rangle}
\xrightarrow{{\langle{h,f,g}\rangle}}
{\langle{\mathcal{S},\mathcal{A}}\rangle}$.
\item[\textbf{restrict/inflate:}]  
The right adjoint 
$\mathrmbf{Tbl}_{\mathcal{S}'}(\mathcal{A}')
\xrightarrow[{\langle{h,f,g}\rangle}^{\ast}]
{\;\grave{\mathrmbfit{tbl}}(h,f,g)\;}
\mathrmbf{Tbl}_{\mathcal{A}}(\mathcal{S})$
defines restriction-inflation.
A table
$\mathcal{T}'={\langle{K',t'}\rangle} \in \mathrmbf{Tbl}_{\mathcal{S}'}(\mathcal{A}')$
is mapped to the table
${\langle{h,f,g}\rangle}^{\ast}(\mathcal{T}')
=\mathcal{T}
={\langle{K,t}\rangle} 
\in \mathrmbf{Tbl}_{\mathcal{A}}(\mathcal{S})$,
with its tuple function
$K\xrightarrow{t}\mathrmbfit{tup}_{\mathcal{A}}(\mathcal{S})$
defined by pullback,
$k{\;\cdot\;}t' = t{\;\cdot\;}\mathrmbfit{tup}(h,f,g)$. 
Here we have ``vertically abridged'' and then ``horizontally inflated'' tuples in 
$\mathrmbf{List}(Y')$ 
by tuple pullback along the signed domain morphism
${\langle{\mathcal{S}',\mathcal{A}'}\rangle}
\xrightarrow{{\langle{h,f,g}\rangle}}
{\langle{\mathcal{S},\mathcal{A}}\rangle}$.
There is a table morphism 
(RHS Fig.\;\ref{fig:flow:sign:dom})
%\[\mbox
{\footnotesize{$
\mathcal{T}'={\langle{\mathcal{S}',\mathcal{A}',K,t'}\rangle}
\xleftarrow{\langle{{\langle{h,f,g}\rangle},k}\rangle}
{\langle{\mathcal{S},\mathcal{A},K,t}\rangle}=\mathcal{T}$
.}\normalsize}
%\]
We say that table 
$\mathcal{T}={\langle{h,f,g}\rangle}^{\ast}(\mathcal{T}')$ 
is the 
\underline{restriction-inflation} of 
table $\mathcal{T}'$
along signed domain morphism
${\langle{\mathcal{S}',\mathcal{A}'}\rangle}
\xrightarrow{{\langle{h,f,g}\rangle}}
{\langle{\mathcal{S},\mathcal{A}}\rangle}$.
\end{itemize}
\begin{figure}
\begin{center}
{{\begin{tabular}{c@{\hspace{65pt}}c}
%%%%%%%%%%%%%%%%%%%%%%%%%%%%%%%%%%%%%%%%%%%%%%%%%%%%%%%%%%%%
{{\begin{tabular}{c}
\setlength{\unitlength}{0.7pt}
%\begin{picture}(120,120)(0,-20)
\begin{picture}(120,90)(0,-5)
\put(0,80){\makebox(0,0){\footnotesize{$K$}}}
\put(120,80){\makebox(0,0){\footnotesize{$K$}}}
\put(0,-10){\makebox(0,0){\footnotesize{$
\underset{\textstyle{\subseteq \mathrmbf{List}(Y')}}
{\mathrmbfit{tup}_{\mathcal{S}'}(\mathcal{A}')}$}}}
\put(125,-10){\makebox(0,0){\footnotesize{$
\underset{\textstyle{\subseteq \mathrmbf{List}(Y)}}
{\mathrmbfit{tup}_{\mathcal{A}}(\mathcal{S})}$}}}
\put(63,95){\makebox(0,0){\scriptsize{$1_{K}$}}}
\put(65,-12){\makebox(0,0){\scriptsize{$
\mathrmbfit{tup}(h,f,g)$}}}
\put(-6,40){\makebox(0,0)[r]{\scriptsize{$t'$}}}
\put(128,40){\makebox(0,0)[l]{\scriptsize{$t$}}}
\put(0,65){\vector(0,-1){50}}
\put(120,65){\vector(0,-1){50}}
\put(100,80){\vector(-1,0){80}}
\put(90,0){\vector(-1,0){50}}
%\put(45,-30){\vector(-3,2){30}}
%\put(105,-10){\vector(-3,-2){30}}
%%%%%%%%%%
\put(0,-44){\makebox(0,0){\normalsize{$\underset{\textstyle{
\mathrmbf{Tbl}_{\mathcal{S}'}(\mathcal{A}')
}}{\underbrace{\rule{50pt}{0pt}}}$}}}
%%%%%%%%%%
\put(60,45){\makebox(0,0){\huge{
$\overset{\textit{\scriptsize{project-expand}}}{\Leftarrow}$}}}
%%%%%%%%%%
\put(195,40){\makebox(0,0){\Large{$\rightleftarrows$}}}
%%%%%%%%%%
\end{picture}
\end{tabular}}}
%%%%%%%%%%%%%%%%%%%%%%%%%%%%%%%%%%%%%%%%%%%%%%%%%%%%%%%%%%%%
&
%%%%%%%%%%%%%%%%%%%%%%%%%%%%%%%%%%%%%%%%%%%%%%%%%%%%%%%%%%%%
{{\begin{tabular}{c}
\setlength{\unitlength}{0.7pt}
\begin{picture}(120,90)(0,-5)
\put(0,80){\makebox(0,0){\footnotesize{$K'$}}}
\put(120,80){\makebox(0,0){\footnotesize{$K$}}}
\put(0,-10){\makebox(0,0){\footnotesize{$
\underset{\textstyle{\subseteq \mathrmbf{List}(Y')}}
{\mathrmbfit{tup}_{\mathcal{S}'}(\mathcal{A}')}$}}}
\put(130,-10){\makebox(0,0){\footnotesize{$
\underset{\textstyle{\subseteq \mathrmbf{List}(Y)}}
{\mathrmbfit{tup}_{\mathcal{A}}(\mathcal{S})}$}}}
\put(63,90){\makebox(0,0){\scriptsize{$k$}}}
\put(65,-12){\makebox(0,0){\scriptsize{$
\mathrmbfit{tup}(h,f,g)$}}}
\put(-6,40){\makebox(0,0)[r]{\scriptsize{$t'$}}}
\put(128,40){\makebox(0,0)[l]{\scriptsize{$t$}}}
\put(0,65){\vector(0,-1){50}}
\put(120,65){\vector(0,-1){50}}
\put(100,80){\vector(-1,0){80}}
\put(90,0){\vector(-1,0){50}}
\qbezier(40,30)(30,30)(20,30)
\qbezier(40,30)(40,20)(40,10)
\put(120,-44){\makebox(0,0){\normalsize{$\underset{\textstyle{
\mathrmbf{Tbl}_{\mathcal{A}}(\mathcal{S})}}{\underbrace{\rule{50pt}{0pt}}}$}}}
%%%%%%%%%%
\put(60,45){\makebox(0,0){\huge{
$\overset{\textit{\scriptsize{restrict-inflate}}}{\Rightarrow}$}}}
%%%%%%%%%%
\end{picture}
\end{tabular}}}
%
%%%%%%%%%%%%%%%%%%%%%%%%%%%%%%%%%%%%%%%%%%%%%%%%%%%%%%%%%%%%
\\&\\&\\
\multicolumn{2}{c}{$\mathrmbf{Tbl}$}
\\
\multicolumn{2}{c}{{\footnotesize{$
{\langle{\mathcal{S}',\mathcal{A}'}\rangle}
\xrightarrow{{\langle{h,f,g}\rangle}}
{\langle{\mathcal{S},\mathcal{A}}\rangle}$}}}
\\
${\scriptstyle\sum}_{{\langle{h,f,g}\rangle}}(\mathcal{T})
\xleftarrow{\langle{{\langle{h,f,g}\rangle},1_{K}}\rangle}
\mathcal{T}$
&
{\footnotesize{
$\mathcal{T}'
\xleftarrow{\langle{{\langle{h,f,g}\rangle},k}\rangle}
{{\langle{h,f,g}\rangle}}^{\ast}(\mathcal{T}')$
}\normalsize}
\end{tabular}}}
\end{center}
\caption{\texttt{FOLE} 
Adjoint Flow in $\mathrmbf{Tbl}$}
\label{fig:flow:sign:dom}
\end{figure}
\begin{flushleft}
{\fbox{\fbox{\footnotesize{\begin{minipage}{345pt}
{\underline{\textsf{How does the left-adjoint of flow work?}}}
The left-adjoint of flow is projection followed by -expansion.
When the index function is 
an inclusion $I'\xhookrightarrow{\,h\,}I$,
the flow 
${\scriptstyle\sum}_{h}(\mathcal{T})$
consists of the sub-tuples of $\mathcal{T}$ indexed by $I'$.
Hence,
the first part of flow (projection) 
restricts the horizontal aspect of tables, 
ending with a subset of columns.
When the data value function is an inclusion 
$Y'\xhookleftarrow{\,g\,}Y$,
the flow 
${\scriptstyle\sum}_{g}({\scriptstyle\sum}_{h}(\mathcal{T}))$
does not alter tuples,
but places them in the larger context $\mathrmbf{List}(Y')$
of data tuples.
Hence,
the second part of flow (expansion) 
does not alter the table ${\scriptstyle\sum}_{h}(\mathcal{T})$,
but places it in a larger context of data tuples.
For pure projection,
use an identity data value function $Y\xleftarrow{\,1_{Y}\,}Y$.
For pure expansion,
use an identity index function $I\xrightarrow{\,1_{I}\,}I$.
\newline\newline
{\underline{\textsf{How does right-adjoint of flow work?}}}
The right-adjoint of flow is restriction followed by inflation.
When the data value function is an inclusion 
$Y'\xhookleftarrow{\,g\,}Y$,
the flow ${g}^{\ast}(\mathcal{T}')$
consists of only those tuples of $\mathcal{T}$ 
with data values in $Y$;
all other tuples are omitted.
Hence,
the first part of flow (restriction) restricts the vertical aspect of tables, 
ending with a subset of rows.
When the index function is an inclusion 
$I'\xhookrightarrow{\,h\,}I$,
the flow ${h}^{\ast}({g}^{\ast}(\mathcal{T}'))$
factors on the complement subset
$\mathrmbfit{tup}_{\mathcal{A}}(I'',s'')$
for
$I''=I{{-}}I'$.
Hence,
the second part of flow (inflation)
enlarges the horizontal aspect of tables.
For pure restriction,
use an identity index function $I\xrightarrow{\,1_{I}\,}I$.
For pure inflation,
use an identity data value function $Y\xleftarrow{\,1_{Y}\,}Y$.
\end{minipage}}}}}
\end{flushleft}
%

%
%%%%%%%%%%%%%%%%%%%%%%%%%%%%%%%%%%%%%%%%%%%%%%%%%%%%%%%%%%%%%%%%%%%%%%
%%%%%%%%%%%%%%%%%%%%%%%%%%%%%%%%%%%%%%%%%%%%%%%%%%%%%%%%%%%%%%%%%%%%%%
%%%%%%%%%%%%%%%%%%%%%%%%%%%%%%%%%%%%%%%%%%%%%%%%%%%%%%%%%%%%%%%%%%%%%%
%%%%%%%%%%%%%%%%%%%%%%%%%%%%%%%%%%%%%%%%%%%%%%%%%%%%%%%%%%%%%%%%%%%%%%
%%%%%%%%%%%%%%%%%%%%%%%%%%%%%%%%%%%%%%%%%%%%%%%%%%%%%%%%%%%%%%%%%%%%%%
%\comment{% alas and alack for economy I must omit this amusing diversion
%
\newpage
\subsubsection{The Square.}\label{sub:sub:sec:adj:flo:square}
%%%%%%%%%%%%%%%%%%%%%%%%%%%%%%%%%%%%%%%%%%%%%%%%%%%%%%%%%%%%%%%%%%%%%%
%$\bigstar$ 

%%%%%%%%%%%%%%%%%%%%%%%%%%%%%%%%%%%%%%%%%%%%%%%%%%%%%%%%%%%%%%%%%%%%%%%%%%%%%%%%
%%%%%%%%%%%%%%%%%%%%%%%%%%%%%%%%%%%%%%%%%%%%%%%%%%%%%%%%%%%%%%%%%%%%%%%%%%%%%%%%
\comment{\begin{aside}
In the \texttt{FOLE} relational model
there is a certain asymmetry between signatures and type domains.
Signatures are lists,
but
type domains are classifications.
For example,
the natural join 
with a fixed type domain 
$\mathcal{A} = {\langle{X,Y,\models}\rangle}$
ends up using a subset of sorts in its tuples,
with the remaining sorts still there in the background;
whereas (its ``dual'') a data-type join
with a fixed signature 
$\mathcal{S} = {\langle{I,X,s}\rangle}$
continues to use all sorts,
just with a different set of data values for each sort. 
\end{aside}}
%%%%%%%%%%%%%%%%%%%%%%%%%%%%%%%%%%%%%%%%%%%%%%%%%%%%%%%%%%%%%%%%%%%%%%%%%%%%%%%%
%%%%%%%%%%%%%%%%%%%%%%%%%%%%%%%%%%%%%%%%%%%%%%%%%%%%%%%%%%%%%%%%%%%%%%%%%%%%%%%%
%\rule{344pt}{1pt}\newline
Adjoint flow in the square
works in the full context $\mathrmbf{Tbl}$. 
By fixing a set of sorts $X$,
we allow the header indexing to vary and
we allow the data-types to vary.
%
%%%%%%%%%%%%%%%%%%%%%%%%%%%%%%%%%%%%%%%%%%%%%%%%%%%%%%%%%%%%%%%%%%%%%%%%%%%%%%%%
%%%%%%%%%%%%%%%%%%%%%%%%%%%%%%%%%%%%%%%%%%%%%%%%%%%%%%%%%%%%%%%%%%%%%%%%%%%%%%%%
\footnote{In concert with 
%the previous two sections
\S\,\ref{sub:sub:sec:adj:flow:A}
and
\S\,\ref{sub:sub:sec:adj:flow:S},
this might be entitled
``\textbf{Fixed Sort Set}''.}
%%%%%%%%%%%%%%%%%%%%%%%%%%%%%%%%%%%%%%%%%%%%%%%%%%%%%%%%%%%%%%%%%%%%%%%%%%%%%%%%
%%%%%%%%%%%%%%%%%%%%%%%%%%%%%%%%%%%%%%%%%%%%%%%%%%%%%%%%%%%%%%%%%%%%%%%%%%%%%%%%
%
This allows both
processing in 
$\mathrmbf{Tbl}(\mathcal{A})$
and
processing in 
$\mathrmbf{Tbl}(\mathcal{S})$.
%Flow in the square (Fig.\;\ref{fig:square}). 
%Let $X$ be a fixed sort set.
Let
%\[\mbox
{\footnotesize{$
\mathcal{D}'={\langle{\mathcal{S}',\mathcal{A}'}\rangle}
\xrightarrow{{\langle{h,1,g}\rangle}}
{\langle{\mathcal{S},\mathcal{A}}\rangle}=\mathcal{D}
$}\normalsize}
%\]
be a signed domain morphism
with an identity sort function $X\xrightarrow{1}X$.
This consists of 
an $X$-sorted signature morphism 
$\mathcal{S}'\xrightarrow{\;h\;}\mathcal{S}$ 
which satisfies the condition 
$h{\,\cdot\,}s = s'$,
and 
an $X$-sorted type domain morphism 
$\mathcal{A}'
\xrightarrow{\;g\;}\mathcal{A}$
which satisfies the condition 
$\mathrmbfit{ext}_{\mathcal{A}'}{\;\cdot\;}g^{-1}
= \mathrmbfit{ext}_{\mathcal{A}}$;
or that
{{$g^{-1}(A'_{x}) = A_{x}$ for all $x \in X$,}}
so that
${g}^{-1}(\mathcal{A}') = \mathcal{A}$.
%
%\begin{itemize}
%\item 
\comment{
\begin{center}
{{\begin{tabular}{c}
\setlength{\unitlength}{0.33pt}
\begin{picture}(200,190)(0,0)
\put(0,180){\makebox(0,0){\footnotesize{$I'$}}}
%\put(0,90){\makebox(0,0){\footnotesize{$X'$}}}
\put(0,0){\makebox(0,0){\footnotesize{${\wp}(Y')$}}}
\put(180,180){\makebox(0,0){\footnotesize{$I$}}}
\put(90,90){\makebox(0,0){\footnotesize{$X$}}}
\put(180,0){\makebox(0,0){\footnotesize{${\wp}(Y)$}}}
\put(12,140){\makebox(0,0)[l]{\scriptsize{$s'$}}}
\put(154,140){\makebox(0,0)[l]{\scriptsize{$s$}}}
\put(20,50){\makebox(0,0)[l]{\scriptsize{$\mathcal{A}'$}}}
\put(145,50){\makebox(0,0)[l]{\scriptsize{$\mathcal{A}$}}}
\put(90,194){\makebox(0,0){\scriptsize{$h$}}}
\put(98,14){\makebox(0,0){\scriptsize{$g^{-1}$}}}
\put(20,180){\vector(1,0){140}}
\put(20,0){\vector(1,0){140}}
\put(10,170){\vector(1,-1){70}}
\put(170,170){\vector(-1,-1){70}}
\put(80,75){\vector(-1,-1){60}}
\put(100,75){\vector(1,-1){60}}
\end{picture}
\end{tabular}}}
\end{center}}
Here
$\mathcal{S}'\xrightarrow{\;h\;}\mathcal{S}$
could be a morphism in
$\mathrmbf{Tbl}(\mathcal{A}')$
\underline{or}
$\mathrmbf{Tbl}(\mathcal{A})$
(top/bottom of 
%the square 
Fig.\;\ref{fig:square}),
%\item 
and
$\mathcal{A}'
%={\langle{Y',\models_{\mathcal{A}'}}\rangle} 
\xrightarrow{\;g\;} 
%{\langle{Y,\models_{g^{-1}(\mathcal{A}')}}\rangle}=
\mathcal{A}$
could be a morphism in
$\mathrmbf{Tbl}(\mathcal{S}')$
\underline{or}
$\mathrmbf{Tbl}(\mathcal{S})$
(left/right of 
%the square 
Fig.\;\ref{fig:square}).
%\end{itemize}
%
\begin{figure}
\begin{center}
{{\begin{tabular}{c}
\setlength{\unitlength}{0.4pt}
\begin{picture}(260,270)(-8,-44)
\put(120,240){\makebox(0,0){\normalsize{${
\overset{\textstyle{\mathrmbf{Tbl}(\mathcal{A}')}}{\overbrace{\hspace{120pt}}}}$}}}
\put(120,-60){\makebox(0,0){\normalsize{${
\underset{\textstyle{\mathrmbf{Tbl}(\mathcal{A})}}{\underbrace{\hspace{120pt}}}}$}}}
\put(-50,90){\makebox(0,0)[r]{\footnotesize{$
\mathrmbf{Tbl}(\mathcal{S}')\left\{\rule{0pt}{50pt}\right.$}}}
\put(290,90){\makebox(0,0)[l]{\footnotesize{$
\left.\rule{0pt}{50pt}\right\}\mathrmbf{Tbl}(\mathcal{S})$}}}
\put(0,180){\makebox(0,0){\footnotesize{${\langle{\mathcal{S}',\mathcal{A}'}\rangle}$}}}
\put(240,180){\makebox(0,0){\footnotesize{${\langle{\mathcal{S},\mathcal{A}'}\rangle}$}}}
\put(0,0){\makebox(0,0){\footnotesize{${\langle{\mathcal{S}',\mathcal{A}}\rangle}$}}}
\put(240,0){\makebox(0,0){\footnotesize{${\langle{\mathcal{S},\mathcal{A}}\rangle}$}}}
\put(120,200){\makebox(0,0){\scriptsize{$h$}}}
\put(120,20){\makebox(0,0){\scriptsize{$h$}}}
\put(20,95){\makebox(0,0)[l]{\scriptsize{$g$}}}
\put(260,95){\makebox(0,0)[l]{\scriptsize{$g$}}}
\put(0,150){\vector(0,-1){120}}
\put(240,150){\vector(0,-1){120}}
\put(60,180){\vector(1,0){120}}
\put(60,0){\vector(1,0){120}}
\put(120,90){\makebox(0,0){\scriptsize{$\blacksquare$}}}
\put(120,90){\makebox(0,0){\large{$\square$}}}
\end{picture}
\end{tabular}}}
\end{center}
\caption{The Square}
\label{fig:square}
\end{figure}
%
%\newpage
Hence,
the
signed domain morphism 
factors in two ways,
as visualized in the square 
(Fig.\;\ref{fig:adj:flo:square}). 
Because of this flexibility,
by interspersing Booleans at the signed domains (corners),
there are a variety of flowcharts definable here.
The outer-join in \S\,\ref{sub:sub:sec:out:join} demonstrates
one possible use for the square.
% not needed
%%%%%%%%%%%%%%%%%%%%%%%%%%%%%%%%%%%%%%%%%%%%%%%%%%%%%%%%%%%%%%%%%%%%%%%%%%%%%%%%
%%%%%%%%%%%%%%%%%%%%%%%%%%%%%%%%%%%%%%%%%%%%%%%%%%%%%%%%%%%%%%%%%%%%%%%%%%%%%%%%
%

%\newpage

%
\begin{figure}
\begin{center}
{{\begin{tabular}{c}
\setlength{\unitlength}{0.56pt}
\begin{picture}(260,280)(-8,-44)
%
%\put(-80,180){\makebox(0,0){\footnotesize{$\mathrmbf{3}$}}}
%\put(-80,0){\makebox(0,0){\footnotesize{$\mathrmbf{4}$}}}
%\put(0,-60){\makebox(0,0){\footnotesize{$\mathrmbf{1}$}}}
%\put(240,-60){\makebox(0,0){\footnotesize{$\mathrmbf{2}$}}}
%\thicklines
%\put(-70,180){\vector(1,0){30}}
%\put(-70,0){\vector(1,0){30}}
%\put(0,-45){\vector(0,1){30}}
%\put(240,-45){\vector(0,1){30}}
%\thinlines
%
\put(0,180){\makebox(0,0){\footnotesize{${\mathrmbf{Tbl}_{\mathcal{A}'}(\mathcal{S}')}$}}}
\put(240,180){\makebox(0,0){\footnotesize{${\mathrmbf{Tbl}_{\mathcal{A}'}(\mathcal{S})}$}}}
\put(0,0){\makebox(0,0){\footnotesize{${\mathrmbf{Tbl}_{\mathcal{A}}(\mathcal{S}')}$}}}
\put(240,0){\makebox(0,0){\scriptsize{${\mathrmbf{Tbl}_{\mathcal{A}}(\mathcal{S})}$}}}
\put(120,230){\makebox(0,0){\scriptsize{$\textit{project}$}}}
\put(120,210){\makebox(0,0){\scriptsize{$\acute{\mathrmbfit{tbl}}_{\mathcal{A}'}(h)$}}}
\put(120,150){\makebox(0,0){\scriptsize{$\grave{\mathrmbfit{tbl}}_{\mathcal{A}'}(h)$}}}
\put(120,130){\makebox(0,0){\scriptsize{$\textit{inflate}$}}}
\put(120,50){\makebox(0,0){\scriptsize{$\textit{project}$}}}
\put(120,30){\makebox(0,0){\scriptsize{$\acute{\mathrmbfit{tbl}}_{\mathcal{A}}(h)$}}}
\put(120,-30){\makebox(0,0){\scriptsize{$\grave{\mathrmbfit{tbl}}_{\mathcal{A}}(h)$}}}
\put(120,-50){\makebox(0,0){\scriptsize{$\textit{inflate}$}}}
\put(30,95){\makebox(0,0)[l]{\scriptsize{$\acute{\mathrmbfit{tbl}}_{\mathcal{S}'}(g)$}}}
\put(30,75){\makebox(0,0)[l]{\scriptsize{$\textit{expand}$}}}
\put(-30,95){\makebox(0,0)[r]{\scriptsize{$\grave{\mathrmbfit{tbl}}_{\mathcal{S}'}(g)$}}}
\put(-30,75){\makebox(0,0)[r]{\scriptsize{$\textit{restrict}$}}}
\put(270,95){\makebox(0,0)[l]{\scriptsize{$\acute{\mathrmbfit{tbl}}_{\mathcal{S}}(g)$}}}
\put(270,75){\makebox(0,0)[l]{\scriptsize{$\textit{expand}$}}}
\put(210,95){\makebox(0,0)[r]{\scriptsize{$\grave{\mathrmbfit{tbl}}_{\mathcal{S}}(g)$}}}
\put(210,75){\makebox(0,0)[r]{\scriptsize{$\textit{restrict}$}}}
\put(0,90){\makebox(0,0){\footnotesize{${\;\dashv\;}$}}}
\put(240,90){\makebox(0,0){\footnotesize{${\;\dashv\;}$}}}
\put(120,180){\makebox(0,0){\footnotesize{${\;\dashv\;}$}}}
\put(120,0){\makebox(0,0){\footnotesize{${\;\dashv\;}$}}}
\put(-16,150){\vector(0,-1){120}}
\put(16,30){\vector(0,1){120}}
\put(224,150){\vector(0,-1){120}}
\put(256,30){\vector(0,1){120}}
\put(180,192){\vector(-1,0){120}}
\put(60,168){\vector(1,0){120}}
\put(180,12){\vector(-1,0){120}}
\put(60,-12){\vector(1,0){120}}
%
%\put(-40,210){\oval(50,50)[tr]}
%\put(-50,210){\oval(50,50)[l]}
%\put(-50,235){\line(1,0){10}}
%\put(-45,185){\vector(1,0){0}}
%\put(-40,-20){\oval(50,50)[br]}
%\put(-50,-20){\oval(50,50)[l]}
%\put(-50,-45){\line(1,0){10}}
%\put(-15,-20){\vector(0,1){0}}
%\put(280,210){\oval(50,50)[tl]}
%\put(290,210){\oval(50,50)[r]}
%\put(280,235){\line(1,0){10}}
%\put(255,210){\vector(0,-1){0}}
%\put(280,-20){\oval(50,50)[bl]}
%\put(290,-20){\oval(50,50)[r]}
%\put(280,-45){\line(1,0){10}}
%\put(285,5){\vector(-1,0){0}}
\put(0,210){\makebox(0,0){\footnotesize{$\vee\wedge-$}}}
\put(0,-20){\makebox(0,0){\footnotesize{$\vee\wedge-$}}}
\put(240,210){\makebox(0,0){\footnotesize{$\vee\wedge-$}}}
\put(240,-20){\makebox(0,0){\footnotesize{$\vee\wedge-$}}}
%\put(0,204){\oval(34,34)[t]}
%\put(0,204){\oval(34,34)[br]}
%\put(-17,196){\vector(0,-1){0}}
%\put(0,230){\makebox(0,0){\scriptsize{$\vee$}}}
%\put(0,242){\makebox(0,0){\scriptsize{$\textit{join}$}}}
%\put(240,204){\oval(34,34)[t]}
%\put(223,196){\vector(0,-1){0}}
%\put(240,230){\makebox(0,0){\scriptsize{$\vee$}}}
%\put(240,242){\makebox(0,0){\scriptsize{$\textit{join}$}}}
%\put(290,0){\oval(30,30)[r]}
%\put(284,15){\vector(-1,0){0}}
%\put(315,0){\makebox(0,0){\scriptsize{$\wedge$}}}
%\put(325,0){\makebox(0,0)[l]{\scriptsize{$\textit{meet}$}}}
%\put(290,180){\oval(30,30)[r]}
%\put(284,195){\vector(-1,0){0}}
%\put(315,180){\makebox(0,0){\scriptsize{$\wedge$}}}
%\put(325,180){\makebox(0,0)[l]{\scriptsize{$\textit{meet}$}}}
%
\put(120,90){\makebox(0,0){\scriptsize{$\blacksquare$}}}
%\put(120,90){\makebox(0,0){\large{$\square$}}}
%
\end{picture}
\end{tabular}}}
\end{center}
\caption{Adjoint Flow in the Square}
\label{fig:adj:flo:square}
\end{figure}
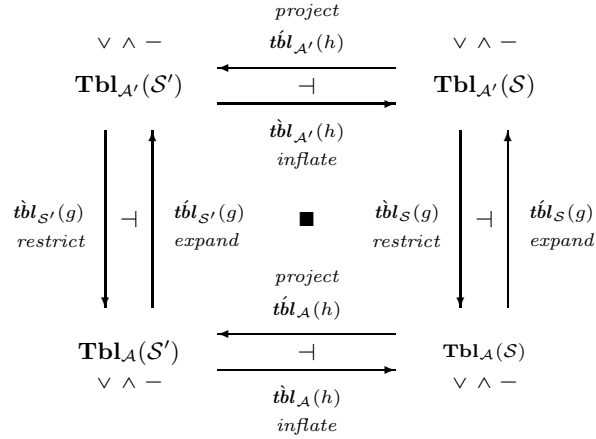
%

%}% alas and alack for economy I must omit this amusing diversion
%%%%%%%%%%%%%%%%%%%%%%%%%%%%%%%%%%%%%%%%%%%%%%%%%%%%%%%%%%%%%%
%%%%%%%%%%%%%%%%%%%%%%%%%%%%%%%%%%%%%%%%%%%%%%%%%%%%%%%%%%%%%%
%%%%%%%%%%%%%%%%%%%%%%%%%%%%%%%%%%%%%%%%%%%%%%%%%%%%%%%%%%%%%%
%%%%%%%%%%%%%%%%%%%%%%%%%%%%%%%%%%%%%%%%%%%%%%%%%%%%%%%%%%%%%%
%%%%%%%%%%%%%%%%%%%%%%%%%%%%%%%%%%%%%%%%%%%%%%%%%%%%%%%%%%%%%%
%%%%%%%%%%%%%%%%%%%%%%%%%%%%%%%%%%%%%%%%%%%%%%%%%%%%%%%%%%%%%%

%%%%%%%%%%%%%%%%%%%%%%%%%%%%%%%%%%%%%%%%%%%%%%%%%%%%%%%%%%%%%%
%%%%%%%%%%%%%%%%%%%%%%%%%%%%%%%%%%%%%%%%%%%%%%%%%%%%%%%%%%%%%%
%%%%%%%%%%%%%%%%%%%%%%%%%%%%%%%%%%%%%%%%%%%%%%%%%%%%%%%%%%%%%%
\newpage
\section{Composite Operations for Limits.}
\label{sub:sec:comp:ops:type:dom}
%%%%%%%%%%%%%%%%%%%%%%%%%%%%%%%%%%%%%%%%%%%%%%%%%%%%%%%%%%%%%%
%%%%%%%%%%%%%%%%%%%%%%%%%%%%%%%%%%%%%%%%%%%%%%%%%%%%%%%%%%%%%%
%%%%%%%%%%%%%%%%%%%%%%%%%%%%%%%%%%%%%%%%%%%%%%%%%%%%%%%%%%%%%%

The basic components of
\S\,\ref{sub:sec:base:ops}
are the components to be used in flowcharts.
Composite operations are operations whose flowcharts 
are composed of one or more basic components.
In addition to its basic components,
a composite operation also has a constraint,
which is used to construct its output.
In this section
we define the the composite relational operations
(Tbl.\,\ref{tbl:fole:comp:rel:ops:lim})
related to limits.
Each composite operation defined here
has a dual relational operation 
(Tbl.\,\ref{tbl:fole:comp:rel:ops:colim})
related to colimits.
%
%%%%%%%%%%%%%%%%%%%%%%%%%%%%%%%%%%%%%%%%%%%%%%%%%%%%%%%%%%%%%%%%%%%%%%%%%%%%%%%%
%%%%%%%%%%%%%%%%%%%%%%%%%%%%%%%%%%%%%%%%%%%%%%%%%%%%%%%%%%%%%%%%%%%%%%%%%%%%%%%%
\footnote{
The quotient operation is dual to the co-quotient operation.
The core operation is dual to the co-core operation.
The natural join operation is dual to the data-type join operation.
The generic meet operation is dual to the generic join operation.}
%%%%%%%%%%%%%%%%%%%%%%%%%%%%%%%%%%%%%%%%%%%%%%%%%%%%%%%%%%%%%%%%%%%%%%%%%%%%%%%%
%%%%%%%%%%%%%%%%%%%%%%%%%%%%%%%%%%%%%%%%%%%%%%%%%%%%%%%%%%%%%%%%%%%%%%%%%%%%%%%%
%
For limit operations
we need only a sufficient collection of tables linked by the given collection of signatures
(see Def.\;\ref{def:suff:adequ:lim}).
Fig.\;\ref{fig:routes:flow:lim}
gives the possible routes of flow for limits.

\begin{table}
\begin{center}
{{{\begin{tabular}{c}
\setlength{\extrarowheight}{2pt}
{\scriptsize{$\begin{array}[c]
{|@{\hspace{5pt}}r@{\hspace{10pt}}l@{\hspace{5pt}\in\hspace{4pt}}l@{\hspace{5pt}}|}
\hline
\textbf{quotient:}
&
{\Yright\!}_{\mathcal{A}}(\mathcal{T})
=
\grave{\mathrmbfit{tbl}}_{\mathcal{A}}(\hat{h})(\mathcal{T})
&
\mathrmbf{Tbl}_{\mathcal{A}}(\widehat{\mathcal{S}})
\\
\textbf{core:}
&
\mathcal{T}_{1}{\,{\sqcap}_{\mathcal{S}}}\mathcal{T}_{2} 
= 
\grave{\mathrmbfit{tbl}}_{\mathcal{S}}(g_{1})(\mathcal{T}_{1})
{\;\wedge\;}
\grave{\mathrmbfit{tbl}}_{\mathcal{S}}(g_{2})(\mathcal{T}_{2})
&
\mathrmbf{Tbl}_{\mathcal{S}}(\mathcal{A})
\\
\textbf{natural join:}
&
{\mathcal{T}_{1}}{\,\boxtimes_{\mathcal{A}}}{\mathcal{T}_{2}}
= 
\grave{\mathrmbfit{tbl}}_{\mathcal{A}}(\iota_{1})(\mathcal{T}_{1})
{\;\wedge\;}
\grave{\mathrmbfit{tbl}}_{\mathcal{A}}(\iota_{2})(\mathcal{T}_{2})
&
\mathrmbf{Tbl}_{\mathcal{A}}(\mathcal{S}_{1}{{+}_{\mathcal{S}}}\mathcal{S}_{2})
\\\hline\hline
\textbf{semi-join:}
&
{\mathcal{T}_{1}}{\,\boxleft_{\mathcal{A}}}{\mathcal{T}_{2}} =
%\mathrmbfit{im}(
\acute{\mathrmbfit{tbl}}_{\mathcal{A}}(\iota_{1})
\bigl(\mathcal{T}_{1}{\,\boxtimes_{\mathcal{A}}}\mathcal{T}_{2}\bigr)
%)
&
\mathrmbf{Tbl}_{\mathcal{A}}(\mathcal{S}_{1})
\\
\textbf{anti-join:}
&
{\mathcal{T}_{1}}{\,\boxslash_{\mathcal{A}}}{\mathcal{T}_{2}} =
\mathcal{T}_{1}{-}
({\mathcal{T}_{1}}{\,\boxleft_{\mathcal{A}}}{\mathcal{T}_{2}})
&
\mathrmbf{Tbl}_{\mathcal{A}}(\mathcal{S}_{1})
\\\hline\hline
\textbf{generic meet:}
&
\prod\mathrmbfit{T} =
\bigwedge \Bigl\{ 
\grave{\mathrmbfit{tbl}}({\hat{h}_{i},\hat{f}_{i},\hat{g}_{i}})(\mathcal{T}_{i}) 
\mid i \in I \Bigr\}
&
\mathrmbf{Tbl}(\widehat{\mathcal{D}})
\\\hline
\end{array}$}}
\end{tabular}}}}
\end{center}
\caption{\texttt{FOLE} Composite Relational Operations for Limits}
\label{tbl:fole:comp:rel:ops:lim}
\end{table}
\begin{figure}
\begin{center}
{{\begin{tabular}{c}
\setlength{\unitlength}{0.65pt}
\begin{picture}(120,90)(100,-20)
\put(104,55){\makebox(0,0){\huge{$\overset{\textit{\scriptsize{restrict}}}{\Rightarrow}$}}}
\put(105,34.8){\makebox(0,0){\huge{${\textit{\scriptsize{inflate}}}$}}}
\put(206,55){\makebox(0,0){\huge{$\overset{\textit{\scriptsize{restrict}}}{\Leftarrow}$}}}
\put(205,34.8){\makebox(0,0){\huge{${\textit{\scriptsize{inflate}}}$}}}
%%%%%
\put(155,30){\makebox(0,0){\huge{$\overset{\textit{\scriptsize{meet}}}{\Downarrow}$}}}
%%%%%
\put(105,4.1){\makebox(0,0){\huge{$\overset{\textit{\scriptsize{project}}}{\Leftarrow}$}}}
\put(105,-17){\makebox(0,0){\huge{${\textit{\scriptsize{expand}}}$}}}
\put(205,4.1){\makebox(0,0){\huge{$\overset{\textit{\scriptsize{project}}}{\Rightarrow}$}}}
\put(205,-17){\makebox(0,0){\huge{${\textit{\scriptsize{expand}}}$}}}
%
%\put(40,48){\line(1,0){50}}
\put(90,48){\line(1,0){54}}
\put(144,38){\oval(20,20)[tr]}
\put(166,38){\oval(20,20)[tl]}
\put(220,48){\line(-1,0){54}}
\put(145,-3){\line(-1,0){54}}
\put(165,-3){\line(1,0){54}}
\put(154,9){\line(0,1){29}}
\put(156,9){\line(0,1){29}}
\put(145,7){\oval(20,20)[br]}
\put(165,7){\oval(20,20)[bl]}
\end{picture}
\end{tabular}}}
\end{center}
\caption{Routes of Flow: Limits}
\label{fig:routes:flow:lim}
\end{figure}
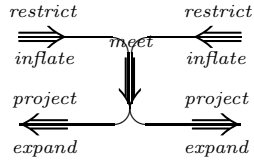
%

%%%%%%%%%%%%%%%%%%%%%%%%%%%%%%%%%%%%%%%%%%%%%%%%%%%%%%%%%%%%%%
%
\newpage
\subsection{Quotient.}\label{sub:sub:sec:quotient}
%\ast\ast\ast$} 
%(Pullback)}

%%%%%%%%%%%%%%%%%%%%%%%%%%%%%%%%%%%%%%%%%%%%%%%%%%%%%%%%%%%%%%
%
%%%%%%%%%%%%%%%%%%%%%%%%%%%%%%%%%%%%%%%%%%%%%%%%%%%%%%%%%%%%%%
%%%%%%%%%%%%%%%%%%%%%%%%%%%%%%%%%%%%%%%%%%%%%%%%%%%%%%%%%%%%%%
\footnote{Basic components 
(\S\,\ref{sub:sec:base:ops})
are components to be used in flowcharts.
In particular, 
the quotient \underline{composite} operation 
of this section
has a flowchart with only one component --- inflation.
In addition to its one component,
it also has a constraint,
which is used to construct its output.}
%%%%%%%%%%%%%%%%%%%%%%%%%%%%%%%%%%%%%%%%%%%%%%%%%%%%%%%%%%%%%%
%%%%%%%%%%%%%%%%%%%%%%%%%%%%%%%%%%%%%%%%%%%%%%%%%%%%%%%%%%%%%%
%
\begin{figure}
\begin{center}
{{{\begin{tabular}{c}
\begin{picture}(120,40)(55,45)
\setlength{\unitlength}{0.97pt}
%%%%%%%%%%%%%%%%%%%%%%%%%%%%%%%%%%%%%%%%%%%%%%%%%%
\put(96.5,35){\begin{picture}(0,0)(0,0)
\setlength{\unitlength}{0.35pt}
%\thicklines
%\put(106,40){\makebox(0,0){\normalsize{$\boldsymbol{\circ}$}}}
%\put(4.7,40){\makebox(0,0){\normalsize{$\boldsymbol{\circ}$}}}
\put(40,10){\line(1,0){60}}
\put(40,70){\line(1,0){60}}
\put(100,70){\line(0,-1){60}}
\put(40,40){\oval(60,60)[bl]}
\put(40,40){\oval(60,60)[tl]}
\put(58,50){\makebox(0,0){\scriptsize{{\textit{{inflate}}}}}}
\put(56,30){\makebox(0,0){\Large{${\Leftarrow}$}}}
\end{picture}}
%%%%%%%%%%%%%%%%%%%%%%%%%%%%%%%%%%%%%%%%%%%%%%%%%%
\put(120,85){\makebox(0,0){\footnotesize{{\textit{{quotient}}}}}}
\put(120,75){\makebox(0,0){\footnotesize{$\Yright$}}}
%[]\sim
%%%%%%%%%%%%%%%%%%%%%%%%%%%%%%%%%%%%%%%%%%%%%%%%%%
\put(100,50){\vector(-1,0){20}}
\put(153,50){\vector(-1,0){20}}
%%%%%%%%%%%%%%%%%%%%%%%%%%%%%%%%%%%%%%%%%%%%%%%%%%
\end{picture}
\end{tabular}}}}
\end{center}
\caption{\texttt{FOLE} Quotient Flow Chart}
\label{fig:fole:quotient:flo:chrt}
\end{figure}
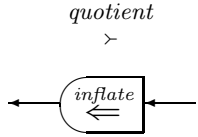
%
%Let 
%$\mathcal{A} = {\langle{X,Y,\models_{\mathcal{A}}}\rangle}$
%be a fixed type domain.
In this section, we focus on tables in the context $\mathrmbf{Tbl}(\mathcal{A})$ 
for fixed type domain $\mathcal{A}$.
In this context,
generic meets 
--- for the special case of the equalizer of a parallel pair --- 
are called quotients.
Here, 
we define an equivalence on attributes, specifically on indexes.
In \S\,\ref{sub:sub:sec:co-quotient}
we discuss a dual notion;
there we define an equivalence on data values.
The quotient operation defined here 
is dual to 
the co-quotient operation defined in \S\,\ref{sub:sub:sec:co-quotient}.

%\newline\rule{50pt}{1pt}\newline
%\newpage
\begin{description}
\item[Constraint:] 
Consider a parallel pair of 
$X$-signature morphisms 
$h_{1},h_{2}:\mathcal{S}{\;\leftleftarrows}\;\mathcal{S}'$
in $\mathrmbf{List}(X)$
consisting of 
a parallel pair of index functions 
{\footnotesize{$h_{1},h_{2}:I\;\leftleftarrows\;I'$}}
satisfying
$h_{1} \cdot s = s'$ and
$h_{2} \cdot s = s'$.
This is the constraint for quotient (Tbl.\,\ref{tbl:fole:quotient:input:output}).
\newline
\item[Construction:] 
We can construct 
the coequalizer of this constraint 
in $\mathrmbf{List}(X)$
with quotient $X$-signature
$\widehat{\mathcal{S}} = {\langle{\hat{I},\hat{s}}\rangle}$
and projection $X$-signature morphism
$\widehat{\mathcal{S}}\xleftarrow{\,\hat{h}\,}\mathcal{S}$.
%
%%%%%%%%%%%%%%%%%%%%%%%%%%%%%%%%%%%%%%%%%%%%%%%%%%%%%%%%%%%%%%%%%%%%%%
%%%%%%%%%%%%%%%%%%%%%%%%%%%%%%%%%%%%%%%%%%%%%%%%%%%%%%%%%%%%%%%%%%%%%%
\footnote{The coequalizer of 
the parallel pair of index functions 
$h_{1},h_{2}:I\;\leftleftarrows\;I'$
is 
(Mac Lane \cite{maclane:71})
the projection function
$\hat{h} : \widehat{I}=I/E \leftarrow I$
on the quotient set of $I$ by 
the least equivalence relation $E \subset I \times I$
that contains all pairs $(h_{1}(i'),h_{2}(i'))$ for $i' \in I'$.}
%%%%%%%%%%%%%%%%%%%%%%%%%%%%%%%%%%%%%%%%%%%%%%%%%%%%%%%%%%%%%%%%%%%%%%
%%%%%%%%%%%%%%%%%%%%%%%%%%%%%%%%%%%%%%%%%%%%%%%%%%%%%%%%%%%%%%%%%%%%%%
%
Each index $\hat{i} \in \hat{I}$ is an equivalence class of pairs 
%$\{ (i_{1},i_{2}) \mid i_{1},i_{2} \in I \}$
generated by the relation
$\{ (h_{1}(i'),h_{2}(i')) \mid i' \in I' \}$.
The index function
$\hat{I}\;\xleftarrow{\hat{h}}\;I$
maps an index $i \in I$
to the
equivalence class generated by these pairs.
%$(h_{1}(i'),h_{2}(i')) \mid i' \in I'$.
%
The sort map $\hat{I} \xrightarrow{\hat{s}} X$
%of the quotient $X$-signature $\widehat{\mathcal{S}}$
is well-defined 
$\hat{s}(\hat{i}) = s(h_{1}(i')) = s(h_{2}(i'))$ 
for any equivalent pair
$(h_{1}(i'),h_{2}(i'))$ for $i' \in I'$.
This is the construction for quotient (Tbl.\,\ref{tbl:fole:quotient:input:output}).
\newline
\item[Input/Output:] 
Consider a table
$\mathcal{T} = {\langle{K,t}\rangle} \in 
\mathrmbf{Tbl}_{\mathcal{A}}(\mathcal{S})$.
This table forms an adequate collection 
(Def.\;\ref{def:suff:adequ:lim})
to compute the equalizer.
This is the input for quotient (Tbl.\,\ref{tbl:fole:quotient:input:output}).
The output is computed with one inflation.
\begin{itemize}
\item 
Inflation 
{\footnotesize{$
\mathrmbf{Tbl}_{\mathcal{A}}(\widehat{\mathcal{S}})
{\;\xleftarrow
%[\hat{h}^{\ast}]
{\;\grave{\mathrmbfit{tbl}}_{\mathcal{A}}(\hat{h})\;}\;}
\mathrmbf{Tbl}_{\mathcal{A}}(\mathcal{S})
$}\normalsize}
(\S\,\ref{sub:sub:sec:adj:flow:A})
along the tuple function 
%{\footnotesize{$
%\mathrmbfit{tup}_{\mathcal{A}}(\widehat{\mathcal{S}})
%\xrightarrow
%{\mathrmbfit{tup}_{\mathcal{A}}(\hat{h})}
%{\mathrmbfit{tup}_{\mathcal{A}}(\mathcal{S})}
%$}\normalsize}
of the $X$-signature morphism
$\widehat{\mathcal{S}}\xleftarrow{\,\hat{h}\,}\mathcal{S}$
maps the table $\mathcal{T}$
to the $\mathcal{A}$-table
$\widehat{\mathcal{T}}
= \grave{\mathrmbfit{tbl}}_{\mathcal{A}}(\hat{h})(\mathcal{T})
%= \hat{h}^{\ast}(\mathcal{T})
= {\langle{\widehat{K},\hat{t}}\rangle} 
\in \mathrmbf{Tbl}_{\mathcal{A}}(\widehat{\mathcal{S}})$,
with its tuple function
$\widehat{K} \xrightarrow{\hat{t}} 
\mathrmbfit{tup}_{\mathcal{A}}(\widehat{\mathcal{S}})$
defined by pullback,
$\grave{k}{\,\cdot\,}t 
= \hat{t}{\,\cdot\,}\mathrmbfit{tup}_{\mathcal{A}}(\hat{h})$. 
%
%%%%%%%%%%%%%%%%%%%%%%%%%%%%%%%%%%%%%%%%%%%%%%%%%%%%%%%%%%%%%%%%%%%%%%
%%%%%%%%%%%%%%%%%%%%%%%%%%%%%%%%%%%%%%%%%%%%%%%%%%%%%%%%%%%%%%%%%%%%%%
\footnote{The quotient table $\widehat{\mathcal{T}}$
contains the set of all mergers of tuples in $\mathcal{T}$ 
of equivalence classes of sorts.}
%%%%%%%%%%%%%%%%%%%%%%%%%%%%%%%%%%%%%%%%%%%%%%%%%%%%%%%%%%%%%%%%%%%%%%
%%%%%%%%%%%%%%%%%%%%%%%%%%%%%%%%%%%%%%%%%%%%%%%%%%%%%%%%%%%%%%%%%%%%%%
%
This defines the $\mathcal{A}$-table morphism
$\widehat{\mathcal{T}}\xrightarrow{{\langle{\hat{h},\grave{k}}\rangle}}\mathcal{T}$,
%,
%illustrated in 
%Fig.\;\ref{fig:quotient}.
%This 
which is the output for quotient (Tbl.\,\ref{tbl:fole:quotient:input:output}).
\end{itemize}
\end{description}
%
%The quotient flowchart input/output is displayed in
%Tbl.\,\ref{tbl:fole:quotient:input:output}.
%where quotient is distinguish from mere inflation by having a constraint.
Quotient is the one-step process
%\mbox{}\newline\rule{50pt}{1pt}\newline
%
\newline\mbox{}\hfill
%\rule[-10pt]{0pt}{26pt}
${\Yright\!}_{\mathcal{A}}(\mathcal{T})
\doteq
\grave{\mathrmbfit{tbl}}_{\mathcal{A}}(\hat{h})(\mathcal{T})$.
\hfill\mbox{}\newline
\comment{
This results in the $\mathcal{A}$-table morphism
%\[\mbox
{\footnotesize{{$
{\Yright\!}_{\mathcal{A}}(\mathcal{T})
\xleftarrow{\;{\langle{\hat{h},\grave{k}}\rangle}\;} 
\mathcal{T}
$}}\normalsize}
%\]
illustrated in 
Fig.\;\ref{fig:quotient}.
}
\begin{aside}
Theoretically
this would represent the equalizer,
the limit 
(see the application discussion for completeness in 
\S\,\ref{sub:sec:lim:colim:tbl})
%(Chap.\;4 of \cite{kent:fole:era:tbl})
of a parallel pair 
{\footnotesize{${\langle{h_{1},k_{1}}\rangle},{\langle{h_{2},k_{2}}\rangle} :
\mathcal{T}\;\rightrightarrows\;\mathcal{T}'$}}
of $\mathcal{A}$-table morphisms.
But practically,
we are only given the constraint (parallel pair)
{\footnotesize{$h_{1},h_{2}:\mathcal{S}{\;\leftleftarrows}\;\mathcal{S}'$}}
of $X$-signature morphisms 
and the input 
$\mathcal{T}$
in Tbl.\,\ref{tbl:fole:quotient:input:output}.
Similar comments,
which distinguish the practical from the theoretical, 
hold for the natural join operation
in \S\,\ref{sub:sub:sec:nat:join}.
\end{aside}
\begin{table}
\begin{center}
{{\fbox{\begin{tabular}{c}
\setlength{\extrarowheight}{2pt}
{\scriptsize{$\begin{array}[c]{c@{\hspace{12pt}}l}
h_{1},h_{2}:\mathcal{S}{\;\leftleftarrows}\;\mathcal{S}'
&
\textit{constraint}
\\
\widehat{\mathcal{S}}\xleftarrow{\,\hat{h}\,}\mathcal{S}
&
\textit{construction}
\\
\hline
\mathcal{T}\in\mathrmbf{Tbl}_{\mathcal{A}}(\mathcal{S})
&
\textit{input}
\\
{\Yright\!}_{\mathcal{A}}(\mathcal{T})
%=\widehat{\mathcal{T}}
\xrightarrow{{\langle{\hat{h},\grave{k}}\rangle}} 
\mathcal{T}
&
\textit{output}
\end{array}$}}
\end{tabular}}}}
\end{center}
\caption{\texttt{FOLE} Quotient I/O}
\label{tbl:fole:quotient:input:output}
\end{table}
\comment{
\begin{figure}
\begin{center}
{{\begin{tabular}{c}
%%%%%%%%%%%%%%%%%%%%%%%%%%%%%%%%%%%%%%%%%%%%%%%%%%
{{\begin{tabular}{c}
\setlength{\unitlength}{0.65pt}
\begin{picture}(120,110)(0,-30)
\put(0,80){\makebox(0,0){\footnotesize{$\widehat{K}$}}}
\put(120,80){\makebox(0,0){\footnotesize{$K$}}}
\put(-5,0){\makebox(0,0){\footnotesize{$
{\mathrmbfit{tup}_{\mathcal{A}}(\widehat{\mathcal{S}})}$}}}
\put(125,0){\makebox(0,0){\footnotesize{$
{\mathrmbfit{tup}_{\mathcal{A}}(\mathcal{S})}$}}}
\put(63,90){\makebox(0,0){\scriptsize{$\grave{k}$}}}
\put(65,-12){\makebox(0,0){\scriptsize{$\mathrmbfit{tup}_{\mathcal{A}}(\hat{h})$}}}
\put(-6,40){\makebox(0,0)[r]{\scriptsize{$\hat{t}$}}}
\put(128,40){\makebox(0,0)[l]{\scriptsize{$t$}}}
\put(0,65){\vector(0,-1){50}}
\put(120,65){\vector(0,-1){50}}
\put(20,80){\vector(1,0){80}}
\put(34,0){\vector(1,0){56}}
%
%\qbezier(80,30)(90,30)(100,30)
%\qbezier(80,30)(80,20)(80,10)
\qbezier(90,30)(100,30)(110,30)
\qbezier(90,30)(90,20)(90,10)
\put(0,-34){\makebox(0,0){\normalsize{$\underset{\textstyle{
\mathrmbf{Tbl}_{\mathcal{A}}(\widehat{\mathcal{S}})}}{\underbrace{\rule{40pt}{0pt}}}$}}}
\put(120,-34){\makebox(0,0){\normalsize{$\underset{\textstyle{
\mathrmbf{Tbl}_{\mathcal{A}}(\mathcal{S})}}{\underbrace{\rule{40pt}{0pt}}}$}}}
%%%%%%%%%%
\put(55,45){\makebox(0,0){\huge{
$\overset{\textit{\scriptsize{inflate}}}{\Leftarrow}$}}}
%%%%%%%%%%
\end{picture}
\end{tabular}}}
\end{tabular}}}
\end{center}
\caption{\texttt{FOLE} Quotient}
\label{fig:quotient}
\end{figure}
}
%

%%%%%%%%%%%%%%%%%%%%%%%%%%%%%%%%%%%%%%%%%%%%%%%%%%%%%%%%%%%%%%
%
\newpage
\subsection{Core.}
\label{sub:sub:sec:core}
%$\ast{-}{-}$}
%%%%%%%%%%%%%%%%%%%%%%%%%%%%%%%%%%%%%%%%%%%%%%%%%%%%%%%%%%%%%%
%
%\subsubsection{Natural Meet.}
%\label{sub:sub:sec:nat:meet}

%
\begin{figure}
\begin{center}
{{{\begin{tabular}{c}
\begin{picture}(160,75)(37,27)
\setlength{\unitlength}{0.97pt}
%%%%%%%%%%%%%%%%%%%%%%%%%%%%%%%%%%%%%%%%%%%%%%%%%%
%\put(44,62){\begin{picture}(0,0)(0,0)
%\setlength{\unitlength}{0.46pt}
\put(54,65){\begin{picture}(0,0)(0,0)
\setlength{\unitlength}{0.35pt}
%\thicklines
%\put(106,40){\makebox(0,0){\normalsize{$\boldsymbol{\circ}$}}}
%\put(4.7,40){\makebox(0,0){\normalsize{$\boldsymbol{\circ}$}}}
\put(10,10){\line(1,0){60}}
\put(10,70){\line(1,0){60}}
\put(10,70){\line(0,-1){60}}
\put(70,40){\oval(60,60)[br]}
\put(70,40){\oval(60,60)[tr]}
\put(55,50){\makebox(0,0){\scriptsize{{\textit{{restrict}}}}}}
\put(56,30){\makebox(0,0){\Large{${\Rightarrow}$}}}
\end{picture}}
%%%%%%%%%%%%%%%%%%%%%%%%%%%%%%%%%%%%%%%%%%%%%%%%%%
\put(146.5,65){\begin{picture}(0,0)(0,0)
\setlength{\unitlength}{0.35pt}
%\thicklines
%\put(106,40){\makebox(0,0){\normalsize{$\boldsymbol{\circ}$}}}
%\put(4.7,40){\makebox(0,0){\normalsize{$\boldsymbol{\circ}$}}}
\put(40,10){\line(1,0){60}}
\put(40,70){\line(1,0){60}}
\put(100,70){\line(0,-1){60}}
\put(40,40){\oval(60,60)[bl]}
\put(40,40){\oval(60,60)[tl]}
\put(58,50){\makebox(0,0){\scriptsize{{\textit{{restrict}}}}}}
\put(56,30){\makebox(0,0){\Large{${\Leftarrow}$}}}
\end{picture}}
%%%%%%%%%%%%%%%%%%%%%%%%%%%%%%%%%%%%%%%%%%%%%%%%%%
\put(98,37){\begin{picture}(0,0)(0,3)
\setlength{\unitlength}{0.35pt}
\put(60,30){\makebox(0,0){\normalsize{$\wedge$}}}
%\thicklines
\put(40,10){\line(1,0){40}}
\put(10,70){\line(1,0){100}}
\put(10,70){\line(0,-1){30}}
\put(110,70){\line(0,-1){30}}
\put(40,40){\oval(60,60)[bl]}
\put(80,40){\oval(60,60)[br]}
\put(60,55){\makebox(0,0){\scriptsize{{\textit{{meet}}}}}}
\end{picture}}
%%%%%%%%%%%%%%%%%%%%%%%%%%%%%%%%%%%%%%%%%%%%%%%%%%
\put(120,100){\makebox(0,0){\footnotesize{{\textit{{core}}}}}}
\put(120,88){\makebox(0,0){\large{$\sqcap$}}}
%%%%%%%%%%%%%%%%%%%%%%%%%%%%%%%%%%%%%%%%%%%%%%%%%%
\put(38,80){\line(0,1){20}}
\put(38,80){\vector(1,0){20}}
\put(110,80){\line(-1,0){20}}
\put(110,80){\vector(0,-1){21}}
\put(120,38){\vector(0,-1){15}}
\put(130,80){\vector(0,-1){21}}
\put(130,80){\line(1,0){20}}
\put(203,80){\line(0,1){20}}
\put(203,80){\vector(-1,0){20}}
%\thicklines
%\put(15,110){\line(1,0){210}}
%\put(55,10){\line(1,0){130}}
%\put(15,50){\line(0,1){60}}
%\put(225,50){\line(0,1){60}}
%\qbezier(15,50)(15,10)(55,10)
%\qbezier(185,10)(225,10)(225,50)
%%%%%%%%%%%%%%%%%%%%%%%%%%%%%%%%%%%%%%%%%%%%%%%%%%
\end{picture}
\end{tabular}}}}
\end{center}
\caption{\texttt{FOLE} Core Flow Chart}
\label{fig:fole:core:flo:chrt}
\end{figure}
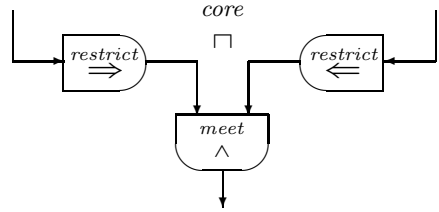
Let $\mathcal{S}$
% = {\langle{I,x,X}\rangle}
be a fixed signature.
%Let 
%$\mathcal{A} = {\langle{X,Y,\models_{\mathcal{A}}}\rangle}$
%be a fixed type domain.
This section discusses the core operation.
The core operation is somewhat unorthodox, 
since it does not have a construction process. 
However, although it is defined for tables with a fixed signature, 
it computes a limit-like result: 
it uses right adjoint flow (restriction) followed by intersection. 
In fact, it corresponds to the first half of flow 
along a signed domain morphism followed by intersection.
%${}^{\ref{colimit:in:A}}$ 
%
%%%%%%%%%%%%%%%%%%%%%%%%%%%%%%%%%%%%%%%%%%%%%%%%%%%%%%%%%%%%%%%%%%%%%%%%%%%%%%%%
%%%%%%%%%%%%%%%%%%%%%%%%%%%%%%%%%%%%%%%%%%%%%%%%%%%%%%%%%%%%%%%%%%%%%%%%%%%%%%%%
\comment{\S\,\ref{sub:sec:comp:ops:type:dom} discusses operations 
that compute a limit. 
%the scope 
%$\mathrmbf{Tbl}(\mathcal{A})$ 
%of $\mathcal{S}$-tables.
Although the core operation 
is defined in the scope 
$\mathrmbf{Tbl}(\mathcal{S})$,
it appears here since it 
computes a limit.}
%%%%%%%%%%%%%%%%%%%%%%%%%%%%%%%%%%%%%%%%%%%%%%%%%%%%%%%%%%%%%%%%%%%%%%%%%%%%%%%%
%%%%%%%%%%%%%%%%%%%%%%%%%%%%%%%%%%%%%%%%%%%%%%%%%%%%%%%%%%%%%%%%%%%%%%%%%%%%%%%%
%
This binary operation gives the core tuples
%
%%%%%%%%%%%%%%%%%%%%%%%%%%%%%%%%%%%%%%%%%%%%%%%%%%%%%%%%%%%%%%%%%%%%%%%%%%%%%%%%
%%%%%%%%%%%%%%%%%%%%%%%%%%%%%%%%%%%%%%%%%%%%%%%%%%%%%%%%%%%%%%%%%%%%%%%%%%%%%%%%
\footnote{A tuple is \emph{core} when 
it appears in both tables and its data values are taken from both data sets.}
%%%%%%%%%%%%%%%%%%%%%%%%%%%%%%%%%%%%%%%%%%%%%%%%%%%%%%%%%%%%%%%%%%%%%%%%%%%%%%%%
%%%%%%%%%%%%%%%%%%%%%%%%%%%%%%%%%%%%%%%%%%%%%%%%%%%%%%%%%%%%%%%%%%%%%%%%%%%%%%%%
%
in a pair of $\mathcal{S}$-tables.
We use the following routes of flow from 
%Tbl.\;\ref{tbl:routes:flow}.
Fig.\;\ref{fig:routes:flow:lim}.

%followed by image-factorization.
%using the header of $\mathcal{T}$. 
%If $a,..,a_{n}$ are the attribute names of $R$, then 
%$R{\;\rhd\,}S = \pi_{a,..,a_{n}}(R \bowtie S)$.
%\mbox{}\newline
%{\fbox{\textbf{Work zone: key factorization of projection.}}}
%\newline
%
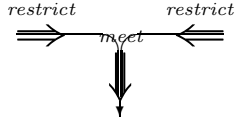
\begin{center}
{{\begin{tabular}{c}
\setlength{\unitlength}{0.6pt}
\begin{picture}(320,55)(0,10)
%\put(157,63){\makebox(0,0){\normalsize{$\textbf{3.}$}}}
%
%\put(325,25){\makebox(0,0){\normalsize{
%$\left.\rule{0pt}{24pt}\right\}
%\underset{\textstyle{\textsf{semi-join}}}{\textsf{right}}$}}}
%\put(-25,25){\makebox(0,0){\normalsize{
%$\underset{\textstyle{\textsf{semi-join}}}{\textsf{left}}
%\left\{\rule{0pt}{24pt}\right.$}}}
%
\put(100,55){\makebox(0,0){\huge{
$\overset{\textit{\scriptsize{restrict}}}{\Rightarrow}$}}}
\put(200,55){\makebox(0,0){\huge{
$\overset{\textit{\scriptsize{restrict}}}{\Leftarrow}$}}}
\put(150,30){\makebox(0,0){\huge{
$\overset{\textit{\scriptsize{meet}}}{\Downarrow}$}}}
%\put(100,4.5){\makebox(0,0){\huge{
%$\overset{\textit{\scriptsize{project}}}{\Leftarrow}$}}}
%\put(105,-14.5){\makebox(0,0){$\textit{\scriptsize{image}}$}}
%\put(200,4.5){\makebox(0,0){\huge{
%$\overset{\textit{\scriptsize{project}}}{\Rightarrow}$}}}
%\put(205,-14.5){\makebox(0,0){$\textit{\scriptsize{image}}$}}
%
\put(90,48){\line(1,0){54}}
\put(144,38){\oval(20,20)[tr]}
\put(166,38){\oval(20,20)[tl]}
\put(220,48){\line(-1,0){54}}
%\put(145,-2){\line(-1,0){54}}
%\put(165,-2){\line(1,0){54}}
\put(154,9){\line(0,1){29}}
\put(156,9){\line(0,1){29}}
%\put(145,8){\oval(20,20)[br]}
%\put(165,8){\oval(20,20)[bl]}
\put(155,6){\vector(0,-1){10}}
\end{picture}
\end{tabular}}}
\end{center}
\comment{% not necessary
Let $\mathcal{S}$ be a signature.
We are given 
two $\mathcal{S}$-tables 
$\mathcal{T}_{1} = {\langle{\mathcal{A}_{1},K_{1},t_{1}}\rangle}$ and
$\mathcal{T}_{2} = {\langle{\mathcal{A}_{2},K_{2},t_{2}}\rangle}$
connected 
through a third type domain $\mathcal{A}$
by an 
%connecting 
$X$-sorted type domain opspan
$\mathcal{A}_{1}
% = {\langle{X,Y_{1},\models_{\mathcal{A}_{1}}}\rangle} 
\xrightarrow{\;g_{1}\,} 
%{\langle{X,Y,\models_{\mathcal{A}}}\rangle} = 
\mathcal{A}
\xleftarrow{\;g_{2}\,} 
\mathcal{A}_{2}$,
%$\mathcal{S}$-table 
%$\mathcal{T} = {\langle{\mathcal{A},K,t}\rangle}$
which consists of an span of data value functions
$Y_{1}\xleftarrow{\;g_{1}\;}Y\xrightarrow{\;g_{2}\;}Y_{2}$. 
%
%Co-core is the two-step process
%illustrated in 
%Fig.\;\ref{fig:fole:co-core}.
%\newline
}% not necessary
We first restrict the data types of the tables to the common data values,
and then we intersect.
%
%%%%%%%%%%%%%%%%%%%%%%%%%%%%%%%%%%%%%%%%%%%%%%%%%%%%%%%%%%%%%%%%%%%%%%%%%%%%%%%%
%%%%%%%%%%%%%%%%%%%%%%%%%%%%%%%%%%%%%%%%%%%%%%%%%%%%%%%%%%%%%%%%%%%%%%%%%%%%%%%%
\footnote{The core operation is the  dual of the co-core operation.}
%\footnote{The restrict join operation is the  dual of the project join operation.}
%%%%%%%%%%%%%%%%%%%%%%%%%%%%%%%%%%%%%%%%%%%%%%%%%%%%%%%%%%%%%%%%%%%%%%%%%%%%%%%%
%%%%%%%%%%%%%%%%%%%%%%%%%%%%%%%%%%%%%%%%%%%%%%%%%%%%%%%%%%%%%%%%%%%%%%%%%%%%%%%%
%

%
\begin{description}
\item[Constraint/Construction:] 
Consider an $X$-sorted type domain opspan
$\mathcal{A}_{1}
% = {\langle{X,Y_{1},\models_{\mathcal{A}_{1}}}\rangle} 
\xrightarrow{\;g_{1}\,} 
%{\langle{X,Y,\models_{\mathcal{A}}}\rangle} = 
\mathcal{A}
\xleftarrow{\;g_{2}\,} 
\mathcal{A}_{2}$
%$\mathcal{S}$-table 
%$\mathcal{T} = {\langle{\mathcal{A},K,t}\rangle}$
consisting of an span of data value functions
$Y_{1}\xleftarrow{\;g_{1}\;}Y\xrightarrow{\;g_{2}\;}Y_{2}$. 
This is both the constraint and the construction for core 
(Tbl.\,\ref{tbl:fole:core:input:output}).
%\newline
%\item[Construction:] 
%The construction for core is the same as the constraint. 
%(Tbl.\,\ref{tbl:fole:core:input:output}).
We deviate from orthodoxy (limit construction) at this step.
\newline
\item[Input:] 
Consider a pair of tables
$\mathcal{T}_{1} = {\langle{K_{1},t_{1}}\rangle} \in 
\mathrmbf{Tbl}_{\mathcal{S}}(\mathcal{A}_{1})$
and
$\mathcal{T}_{2} = {\langle{K_{2},t_{2}}\rangle} \in 
\mathrmbf{Tbl}_{\mathcal{S}}(\mathcal{A}_{2})$.
This is the input for core 
(Tbl.\,\ref{tbl:fole:core:input:output}).
\newline
\item[Output:] 
The output is restriction (twice) followed by meet.
%\mbox{}
\begin{itemize}
%\begin{description}
%\item[restrict:] 
\item 
Restriction 
{\footnotesize{$\mathrmbf{Tbl}_{\mathcal{S}}(\mathcal{A}_{1})
\xrightarrow[\tilde{0}_{2}^{\ast}]
{\;\grave{\mathrmbfit{tbl}}_{\mathcal{S}}(g_{1})\;}
\mathrmbf{Tbl}_{\mathcal{S}}(\mathcal{A})$}\normalsize}
(\S\,\ref{sub:sub:sec:adj:flow:S})
along the tuple function 
of the 
$X$-sorted type domain morphism
$\mathcal{A}_{1}\xrightarrow{\;g_{1}\,}\mathcal{A}$
maps the $\mathcal{S}$-table
$\mathcal{T}_{1}=
{\langle{K_{1},t_{1}}\rangle} \in \mathrmbf{Tbl}_{\mathcal{S}}(\mathcal{A}_{1})$
to the $\mathcal{S}$-table
$\grave{\mathrmbfit{tbl}}_{\mathcal{S}}(g)(\mathcal{T}_{1})
%= {g}^{\ast}(\mathcal{T}_{1})
= \widehat{\mathcal{T}}_{1}
={\langle{\widehat{K}_{1},\hat{t}_{1}}\rangle} 
\in \mathrmbf{Tbl}_{\mathcal{S}}(\mathcal{A})$,
with its tuple function
$\widehat{K}_{1}
\xrightarrow{\hat{t}_{1}}
\mathrmbfit{tup}_{\mathcal{S}}(\mathcal{A})$
defined by pullback,
$\grave{k}_{1}{\;\cdot\;}t_{1} 
= \hat{t}_{1}{\;\cdot\;}\mathrmbfit{tup}_{\mathcal{S}}(g_{1})$. 
This is linked to the table $\mathcal{T}_{1}$ 
by the $\mathcal{S}$-table morphism 
%(RHS Fig.\;\ref{fig:restr:expan:sign})
%\[
{\mbox{\footnotesize{$
\mathcal{T}'={\langle{\mathcal{A}',K',t'}\rangle}
\xleftarrow{{\langle{g_{1},\grave{k}_{1}}\rangle}}
%{\langle{\mathcal{A},K,t}\rangle}=\mathcal{T}$.
{\langle{\widehat{K}_{1},\hat{t}_{1}}\rangle}=\widehat{\mathcal{T}}_{1}
$}\normalsize}}
%\]
%
Similarly for $\mathcal{S}$-table
$\mathcal{T}_{2} = {\langle{K_{2},t_{2}}\rangle} \in 
\mathrmbf{Tbl}_{\mathcal{S}}(\mathcal{A}_{2})$.
%
%\mbox{}\newline\newline
%{\fbox{$\blacktriangledown$\hspace{40pt}
%\textbf{Work zone: Resrict Meet $=$ dual Co-core.}
%\hspace{40pt}$\blacktriangledown$}}
%\newline
%\item[meet:] 
\item
Intersection (\S\,\ref{sub:sub:sec:boole})
of the two restriction tables $\widehat{\mathcal{T}}_{1}$ 
and $\widehat{\mathcal{T}}_{2}$ 
in the context $\mathrmbf{Tbl}_{\mathcal{S}}(\mathcal{A})$
defines the core table
$\mathcal{T}_{1}{\;{\sqcap}_{\mathcal{S}}\;}\mathcal{T}_{2}
= \widehat{\mathcal{T}}_{2}{\,\wedge\,}\widehat{\mathcal{T}}_{2}
= {\langle{\widehat{K}_{1}{\times}\widehat{K}_{2},{(\hat{t}_{1},\hat{t}_{2})}}\rangle}$,
%\widetilda{K}_{12}
%\item 
%The intersection operation defines the \texttt{FOLE} table
%$\mathcal{T}\wedge\mathcal{T}' = {\langle{\widehat{K},{(t,t')}}\rangle}$
whose key set $\widehat{K}_{1}{\times}\widehat{K}_{2}$ is the product and 
whose tuple map 
$\widehat{K}_{1}{\times}\widehat{K}_{2}
\xrightarrow{[\hat{t}_{1},\hat{t}_{2}]}
\mathrmbfit{tup}_{\mathcal{S}}(\mathcal{A})$
maps a pair of keys 
$(\hat{k}_{1},\hat{k}_{2}) \in \widehat{K}_{1}{\times}\widehat{K}_{2}$
to the common tuple
$\hat{t}_{1}(\hat{k}_{1})=\hat{t}_{2}(\hat{k}_{2}) 
\in \mathrmbfit{tup}_{\mathcal{A}}(\mathcal{S})$.
Intersection is the product in $\mathrmbf{Tbl}_{\mathcal{S}}(\mathcal{A})$
with span
%\newline\mbox{}\hfill
$\widehat{\mathcal{T}}_{1}\xleftarrow{\hat{\pi}_{1}}
\widehat{\mathcal{T}}_{1}{\,\wedge\,}\widehat{\mathcal{T}}_{2}
\xrightarrow{\hat{\pi}_{2}}\widehat{\mathcal{T}}_{2}$.
%\hfill\mbox{}
\newline
\end{itemize}
Restriction composed with meet 
defines the span of $\mathcal{A}$-table morphisms
\[\mbox{\footnotesize{
{$\mathcal{T}_{1}
\xleftarrow[\;{\langle{g_{1},\grave{k}_{1}}\rangle}{\circ\,}\check{\pi}_{1}\;]
{\;{\langle{g_{1},\hat{k}_{1}}\rangle}\;} 
\mathcal{T}_{1}{\;{\sqcap}_{\mathcal{S}}\;}\mathcal{T}_{2}
\xrightarrow[\;{\langle{g_{2},\grave{k}_{2}}\rangle}{\circ\,}\hat{\pi}_{2}\;]
{\;{\langle{g_{2},\hat{k}_{1}}\rangle}\;} 
\mathcal{T}_{2}
$,}}\normalsize}\]
which is the output for core (Tbl.\,\ref{tbl:fole:core:input:output}).
%\end{description}
\end{description}
%
%The project join flowchart input/output is displayed in 
%Tbl.\,\ref{tbl:fole:co-core:input:output}.
Core 
%(Fig.\;\ref{fig:fole:nat:join})
%within the context $\mathrmbf{Tbl}(\mathcal{A})$
is restriction followed by meet (conjunction or intersection).
This is the two-step process 
\newline\mbox{}\hfill
\rule[-10pt]{0pt}{26pt}
$\mathcal{T}_{1}{\;{\sqcap}_{\mathcal{S}}\;}\mathcal{T}_{2}
\doteq  
\grave{\mathrmbfit{tbl}}_{\mathcal{S}}(g_{1})(\mathcal{T}_{1})
{\;\wedge\;}
\grave{\mathrmbfit{tbl}}_{\mathcal{S}}(g_{2})(\mathcal{T}_{2})$.
%
%%%%%%%%%%%%%%%%%%%%%%%%%%%%%%%%%%%%%%%%%%%%%%%%%%%%%%%%%%%%%%%%%%%%%%
%%%%%%%%%%%%%%%%%%%%%%%%%%%%%%%%%%%%%%%%%%%%%%%%%%%%%%%%%%%%%%%%%%%%%%
\footnote{The core in \S\,\ref{sub:sub:sec:core},
$\mathcal{T}_{1}
\xleftarrow
%[\;{\langle{g_{1},\grave{k}_{1}}\rangle}{\circ\,}\check{\pi}_{1}\;]
{\;{\langle{g_{1},\hat{k}_{1}}\rangle}\;} 
\mathcal{T}_{1}{\;{\sqcap}_{\mathcal{S}}\;}\mathcal{T}_{2}
\xrightarrow
%[\;{\langle{g_{2},\grave{k}_{2}}\rangle}{\circ\,}\hat{\pi}_{2}\;]
{\;{\langle{g_{2},\hat{k}_{1}}\rangle}\;} 
\mathcal{T}_{2}$,
is homogeneous with and has a direct connection
to both tables $\mathcal{T}_{1}$ and $\mathcal{T}_{2}$. 
This is comparable with
the co-core (\S\,\ref{sub:sub:sec:co-core})
$\mathcal{T}_{1}
\xrightarrow
%[\;{\langle{h_{1},1}\rangle}{\circ\,}\check{\iota}_{1}\;]
{\;{\langle{h_{1},\check{\iota}_{1}}\rangle}\;} 
\mathcal{T}_{1}{\;{\cup}_{\mathcal{S}}\;}\mathcal{T}_{2}
\xleftarrow
%[\;{\langle{h_{2},1}\rangle}{\circ\,}\check{\iota}_{2}\;]
{\;{\langle{h_{2},\check{\iota}_{1}}\rangle}\;} 
\mathcal{T}_{2}$,
which is homogeneous with and has a direct connection
to both tables $\mathcal{T}_{1}$ and $\mathcal{T}_{2}$.}
%%%%%%%%%%%%%%%%%%%%%%%%%%%%%%%%%%%%%%%%%%%%%%%%%%%%%%%%%%%%%%%%%%%%%%
%%%%%%%%%%%%%%%%%%%%%%%%%%%%%%%%%%%%%%%%%%%%%%%%%%%%%%%%%%%%%%%%%%%%%%
%
\hfill\mbox{}\newline

\begin{table}
\begin{center}
{{\fbox{\begin{tabular}{c}
\setlength{\extrarowheight}{2pt}
{\scriptsize{$\begin{array}[c]{c@{\hspace{12pt}}l}
\mathcal{A}_{1}\xrightarrow{\;g_{1}\,}
\mathcal{A}\xleftarrow{\;g_{2}\,}\mathcal{A}_{2}
&
\textit{constraint}
\\
%\mathcal{A}_{1}\xrightarrow{\;g_{1}\,} 
%\mathcal{A}\xleftarrow{\;g_{2}\,}\mathcal{A}_{2}
&
\textit{/construction}
\\
\hline
\mathcal{T}_{1}\in\mathrmbf{Tbl}_{\mathcal{S}}(\mathcal{A}_{1})
\text{ and }
\mathcal{T}_{1}\in\mathrmbf{Tbl}_{\mathcal{S}}(\mathcal{A}_{1})
&
\textit{input}
\\
\mathcal{T}_{1}
\xleftarrow[\;{\langle{g_{1},\grave{k}_{1}}\rangle}{\circ\,}\check{\pi}_{1}\;]
{\;{\langle{g_{1},\hat{k}_{1}}\rangle}\;} 
\mathcal{T}_{1}{\;{\sqcap}_{\mathcal{S}}\;}\mathcal{T}_{2}
\xrightarrow[\;{\langle{g_{2},\grave{k}_{2}}\rangle}{\circ\,}\hat{\pi}_{2}\;]
{\;{\langle{g_{2},\hat{k}_{1}}\rangle}\;} 
\mathcal{T}_{2}
&
\textit{output}
\end{array}$}}
\end{tabular}}}}
\end{center}
\caption{\texttt{FOLE} Core I/O}
\label{tbl:fole:core:input:output}
\end{table}
\newpage
\subsection{Natural Join.}\label{sub:sub:sec:nat:join}
%\ast\ast\ast$} 
%(Pullback)}
%%%%%%%%%%%%%%%%%%%%%%%%%%%%%%%%%%%%%%%%%%%%%%%%%%%%%%%%%%%%%%
%%%%%%%%%%%%%%%%%%%%%%%%%%%%%%%%%%%%%%%%%%%%%%%%%%%%%%%%%%%%
\paragraph{Natural Join.}
%%%%%%%%%%%%%%%%%%%%%%%%%%%%%%%%%%%%%%%%%%%%%%%%%%%%%%%%%%%%
%
%%%%%%%%%%%%%%%%%%%%%%%%%%%%%%%%%%%%%%%%%%%%%%%%%%%%%%%%%%%%
%%%%%%%%%%%%%%%%%%%%%%%%%%%%%%%%%%%%%%%%%%%%%%%%%%%%%%%%%%%%
\footnote{For a brief discussion 
of natural join, see \S\;4.4 of 
%the paper 
``The {\ttfamily FOLE} Table'' 
\cite{kent:fole:era:tbl}.}
%%%%%%%%%%%%%%%%%%%%%%%%%%%%%%%%%%%%%%%%%%%%%%%%%%%%%%%%%%%%
%%%%%%%%%%%%%%%%%%%%%%%%%%%%%%%%%%%%%%%%%%%%%%%%%%%%%%%%%%%%

%
\begin{figure}
\begin{center}
{{{\begin{tabular}{c}
\begin{picture}(160,75)(37,27)
\setlength{\unitlength}{0.97pt}
%%%%%%%%%%%%%%%%%%%%%%%%%%%%%%%%%%%%%%%%%%%%%%%%%%
%\put(44,62){\begin{picture}(0,0)(0,0)
%\setlength{\unitlength}{0.46pt}
\put(54,65){\begin{picture}(0,0)(0,0)
\setlength{\unitlength}{0.35pt}
%\thicklines
%\put(106,40){\makebox(0,0){\normalsize{$\boldsymbol{\circ}$}}}
%\put(4.7,40){\makebox(0,0){\normalsize{$\boldsymbol{\circ}$}}}
\put(10,10){\line(1,0){60}}
\put(10,70){\line(1,0){60}}
\put(10,70){\line(0,-1){60}}
\put(70,40){\oval(60,60)[br]}
\put(70,40){\oval(60,60)[tr]}
\put(55,50){\makebox(0,0){\scriptsize{{\textit{{inflate}}}}}}
\put(56,30){\makebox(0,0){\Large{${\Rightarrow}$}}}
\end{picture}}
%%%%%%%%%%%%%%%%%%%%%%%%%%%%%%%%%%%%%%%%%%%%%%%%%%
\put(146.5,65){\begin{picture}(0,0)(0,0)
\setlength{\unitlength}{0.35pt}
%\thicklines
%\put(106,40){\makebox(0,0){\normalsize{$\boldsymbol{\circ}$}}}
%\put(4.7,40){\makebox(0,0){\normalsize{$\boldsymbol{\circ}$}}}
\put(40,10){\line(1,0){60}}
\put(40,70){\line(1,0){60}}
\put(100,70){\line(0,-1){60}}
\put(40,40){\oval(60,60)[bl]}
\put(40,40){\oval(60,60)[tl]}
\put(58,50){\makebox(0,0){\scriptsize{{\textit{{inflate}}}}}}
\put(56,30){\makebox(0,0){\Large{${\Leftarrow}$}}}
\end{picture}}
%%%%%%%%%%%%%%%%%%%%%%%%%%%%%%%%%%%%%%%%%%%%%%%%%%
\put(98,37){\begin{picture}(0,0)(0,3)
\setlength{\unitlength}{0.35pt}
\put(60,30){\makebox(0,0){\normalsize{$\wedge$}}}
%\thicklines
\put(40,10){\line(1,0){40}}
\put(10,70){\line(1,0){100}}
\put(10,70){\line(0,-1){30}}
\put(110,70){\line(0,-1){30}}
\put(40,40){\oval(60,60)[bl]}
\put(80,40){\oval(60,60)[br]}
\put(60,55){\makebox(0,0){\scriptsize{{\textit{{meet}}}}}}
\end{picture}}
%%%%%%%%%%%%%%%%%%%%%%%%%%%%%%%%%%%%%%%%%%%%%%%%%%
\put(120,100){\makebox(0,0){\footnotesize{{\textit{{natural join}}}}}}
\put(120,90){\makebox(0,0){\large{$\boxtimes$}}}
%%%%%%%%%%%%%%%%%%%%%%%%%%%%%%%%%%%%%%%%%%%%%%%%%%
\put(38,80){\line(0,1){20}}
\put(38,80){\vector(1,0){20}}
\put(110,80){\line(-1,0){20}}
\put(110,80){\vector(0,-1){21}}
\put(120,38){\vector(0,-1){15}}
\put(130,80){\vector(0,-1){21}}
\put(130,80){\line(1,0){20}}
\put(203,80){\line(0,1){20}}
\put(203,80){\vector(-1,0){20}}
%\thicklines
%\put(15,110){\line(1,0){210}}
%\put(55,10){\line(1,0){130}}
%\put(15,50){\line(0,1){60}}
%\put(225,50){\line(0,1){60}}
%\qbezier(15,50)(15,10)(55,10)
%\qbezier(185,10)(225,10)(225,50)
%%%%%%%%%%%%%%%%%%%%%%%%%%%%%%%%%%%%%%%%%%%%%%%%%%
\end{picture}
\end{tabular}}}}
\end{center}
\caption{\texttt{FOLE} Natural Join Flow Chart}
\label{fig:fole:nat:join:flo:chrt}
\end{figure}
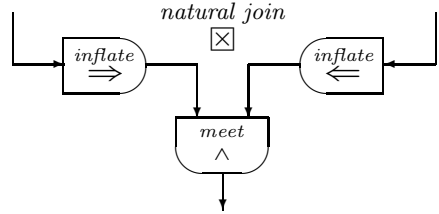
The natural join for tables 
is the relational counterpart of
the logical conjunction for predicates.
Where 
the \emph{meet} operation (\S\,\ref{sub:sub:sec:boole}) is the analogue for logical conjunction 
at the small scope $\mathrmbf{Tbl}(\mathcal{D})$ of a signed domain table fiber,
the \emph{natural join} is defined at the intermediate scope 
$\mathrmbf{Tbl}(\mathcal{A})$ of a type domain table fiber,
and the \emph{generic meet} (\S\,\ref{sub:sub:sec:generic:meet})
is defined at the large scope $\mathrmbf{Tbl}$ of all tables.
We identify these three concepts as limits at different scopes.
%We identify \texttt{FOLE} generic meets with all limits.
%

In this section, we focus on tables in the context $\mathrmbf{Tbl}(\mathcal{A})$ 
for fixed type domain $\mathcal{A}$.
In this context,
generic meets 
--- for the special case of pullback --- 
are called natural joins, 
the join of two $\mathcal{A}$-tables.
As we observed in \cite{kent:fole:era:tbl},
these limits are resolvable into inflations
(called substitutions there)
followed by meet.
We use the following routes of flow from 
%Tbl.\;\ref{tbl:routes:flow}.
Fig.\;\ref{fig:routes:flow:lim}.
%followed by image-factorization.
%using the header of $\mathcal{T}$. 
%If $a,..,a_{n}$ are the attribute names of $R$, then 
%$R{\;\rhd\,}S = \pi_{a,..,a_{n}}(R \bowtie S)$.
%\mbox{}\newline
%{\fbox{\textbf{Work zone: key factorization of projection.}}}
%\newline
%
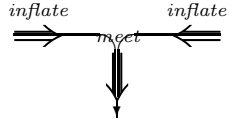
\begin{center}
{{\begin{tabular}{c}
\setlength{\unitlength}{0.6pt}
\begin{picture}(320,60)(0,5)
%\put(157,63){\makebox(0,0){\normalsize{$\textbf{4.}$}}}
%
%\put(325,25){\makebox(0,0){\normalsize{
%$\left.\rule{0pt}{24pt}\right\}
%\underset{\textstyle{\textsf{semi-join}}}{\textsf{right}}$}}}
%\put(-25,25){\makebox(0,0){\normalsize{
%$\underset{\textstyle{\textsf{semi-join}}}{\textsf{left}}
%\left\{\rule{0pt}{24pt}\right.$}}}
%
\put(100,55){\makebox(0,0){\huge{
$\overset{\textit{\scriptsize{inflate}}}{\Rightarrow}$}}}
\put(200,55){\makebox(0,0){\huge{
$\overset{\textit{\scriptsize{inflate}}}{\Leftarrow}$}}}
\put(150,30){\makebox(0,0){\huge{
$\overset{\textit{\scriptsize{meet}}}{\Downarrow}$}}}
%\put(100,4.5){\makebox(0,0){\huge{
%$\overset{\textit{\scriptsize{project}}}{\Leftarrow}$}}}
%\put(105,-14.5){\makebox(0,0){$\textit{\scriptsize{image}}$}}
%\put(200,4.5){\makebox(0,0){\huge{
%$\overset{\textit{\scriptsize{project}}}{\Rightarrow}$}}}
%\put(205,-14.5){\makebox(0,0){$\textit{\scriptsize{image}}$}}
%
\put(90,48){\line(1,0){54}}
\put(144,38){\oval(20,20)[tr]}
\put(166,38){\oval(20,20)[tl]}
\put(220,48){\line(-1,0){54}}
%\put(145,-2){\line(-1,0){54}}
%\put(165,-2){\line(1,0){54}}
\put(154,9){\line(0,1){29}}
\put(156,9){\line(0,1){29}}
%\put(145,8){\oval(20,20)[br]}
%\put(165,8){\oval(20,20)[bl]}
%\put(155,6){\line(0,-1){8}}
\put(155,6){\vector(0,-1){10}}
\end{picture}
\end{tabular}}}
\end{center}
The natural join operation is dual to the data-type join operation of 
\S\,\ref{sub:sub:sec:boole:join}.
Similar to data-type join,
we can define natural join for any number of tables 
$\{ \mathcal{T}_{1}, \mathcal{T}_{2}, \mathcal{T}_{3}, \cdots , \mathcal{T}_{n} \}$
with a comparable constraint.
%
%%%%%%%%%%%%%%%%%%%%%%%%%%%%%%%%%%%%%%%%%%%%%%%%%%%%%%%%%%%%%%%%%%%%%%%%%%%%%%%%
%%%%%%%%%%%%%%%%%%%%%%%%%%%%%%%%%%%%%%%%%%%%%%%%%%%%%%%%%%%%%%%%%%%%%%%%%%%%%%%%
\footnote{Natural join is a limit (meet-like) operation.
To fit better with the limit-colimit duality in this paper
(see Tbl.\,\ref{fig:fole:comp:ops}),
we modify the traditional symbol `$\bowtie$' for natural join,
using the symbol `$\boxtimes$' instead.}
%%%%%%%%%%%%%%%%%%%%%%%%%%%%%%%%%%%%%%%%%%%%%%%%%%%%%%%%%%%%%%%%%%%%%%%%%%%%%%%%
%%%%%%%%%%%%%%%%%%%%%%%%%%%%%%%%%%%%%%%%%%%%%%%%%%%%%%%%%%%%%%%%%%%%%%%%%%%%%%%%
%
%\newpage
%We focus on tables in the context $\mathrmbf{Tbl}(\mathcal{A})$ 
%for fixed type domain $\mathcal{A}$.
%
\begin{description}
\item[Constraint:] 
Consider an $X$-sorted signature span 
$\mathcal{S}_{1}\xleftarrow{h_{1}}\mathcal{S}\xrightarrow{h_{2}}\mathcal{S}_{2}$
in $\mathrmbf{List}(X)$
consisting of 
a span of index functions
$I_{1}\xleftarrow{h_{1}}I\xrightarrow{h_{2}}I_{2}$.
This is the constraint for natural join (Tbl.\,\ref{tbl:fole:natural:join:input:output}).
\newline
\item[Construction:] 
%
%We can form 
%(Chap.\;4 of \cite{kent:fole:era:tbl})
The pushout of this constraint
%colimiting cocone of this signature span
in $\mathrmbf{List}(X)$
is the opspan
{\footnotesize{$\mathcal{S}_{1}\xrightarrow{\iota_{1}\,} 
{\mathcal{S}_{1}{\!+_{\mathcal{S}}}\mathcal{S}_{2}}
\xleftarrow{\;\iota_{2}}\mathcal{S}_{2}$}\normalsize}
of injection $X$-signature morphisms
with pushout signature
$\mathcal{S}_{1}{+_{\mathcal{S}}}\mathcal{S}_{2}$
and
index function opspan
{\footnotesize{$I_{1} 
\xrightarrow{\iota_{1}\,} 
{\langle{I_{1}{+}_{I}I_{2},[s_{1},s_{2}]}\rangle}
\xleftarrow{\;\iota_{2}}I_{2}.$}\normalsize}
This is the construction for natural join (Tbl.\,\ref{tbl:fole:natural:join:input:output}).
\newline
\item[Input:] 
Consider a pair of tables
$\mathcal{T}_{1} = {\langle{K_{1},t_{1}}\rangle} \in 
\mathrmbf{Tbl}_{\mathcal{A}}(\mathcal{S}_{1})$
and
$\mathcal{T}_{2} = {\langle{K_{2},t_{2}}\rangle} \in 
\mathrmbf{Tbl}_{\mathcal{A}}(\mathcal{S}_{2})$.
These two tables form an adequate collection 
(Def.\;\ref{def:suff:adequ:lim})
to compute the pullback.
This is the input for natural join (Tbl.\,\ref{tbl:fole:natural:join:input:output}).
\newline
\item[Output:] 
The output is inflation (twice) followed by meet.
%\mbox{}
%\begin{description}
%%\item[restrict:] 
\begin{itemize}
\item 
Inflation 
{\footnotesize{$
\mathrmbf{Tbl}_{\mathcal{A}}(\mathcal{S}_{1})
{\;\xrightarrow
{\;\grave{\mathrmbfit{tbl}}_{\mathcal{A}}(\iota_{1})\;}\;}
\mathrmbf{Tbl}_{\mathcal{A}}({\mathcal{S}{+_{\mathcal{S}}}\mathcal{S}_{2}})
$}\normalsize}
(\S\,\ref{sub:sub:sec:adj:flow:A})
along the tuple function 
of the $X$-signature morphism
{\footnotesize{$\mathcal{S}_{1}\xrightarrow{\iota_{1}\,} 
{\mathcal{S}_{1}{\!+_{\mathcal{S}}}\mathcal{S}_{2}}$}\normalsize}
maps the table $\mathcal{T}_{1}$
to the table
$\widehat{\mathcal{T}}_{1}
= \grave{\mathrmbfit{tbl}}_{\mathcal{A}}(\iota_{1})(\mathcal{T}_{1})
%\iota_{1}^{\ast}(\mathcal{T}_{1})
= {\langle{\widehat{K}_{1},\hat{t}_{1}}\rangle} 
\in \mathrmbf{Tbl}_{\mathcal{A}}({\mathcal{S}_{1}{\!+_{\mathcal{S}}}\mathcal{S}_{2}})$,
with its tuple function
$\widehat{K} \xrightarrow{\hat{t}} 
\mathrmbfit{tup}_{\mathcal{A}}({\mathcal{S}_{1}{\!+_{\mathcal{S}}}\mathcal{S}_{2}})$
defined by pullback,
$\grave{k}_{1}{\,\cdot\,}t_{1} 
= \hat{t}_{1}{\,\cdot\,}\mathrmbfit{tup}_{\mathcal{A}}(\iota_{1})$. 
This is linked to the table $\mathcal{T}_{1}$ 
by the $\mathcal{A}$-table morphism 
%\[\mbox
{\footnotesize{{$
\mathcal{T}_{1} = {\langle{\mathcal{S}_{1},K_{1},t_{1}}\rangle}
\xleftarrow{{\langle{\iota_{1},\grave{k}_{1}}\rangle}} 
{\langle{{\mathcal{S}_{1}{\!+_{\mathcal{S}}}\mathcal{S}_{2}},\widehat{K}_{1},\hat{t}_{1}}\rangle} 
%= {\iota}_{1}^{\ast}(\mathcal{T}_{1})
= \widehat{\mathcal{T}}_{1}
$.}}\normalsize}
%\]
Similarly for $\mathcal{A}$-table
$\mathcal{T}_{2} = {\langle{K_{2},t_{2}}\rangle} \in 
\mathrmbf{Tbl}_{\mathcal{A}}(\mathcal{S}_{2})$.
\newline
\item 
Intersection (\S\,\ref{sub:sub:sec:boole})
of the two inflation tables $\widehat{\mathcal{T}}_{1}$
%{i}_{1}^{\ast}(\mathcal{T}_{1})$
and $\widehat{\mathcal{T}}_{2}$
%{i}_{2}^{\ast}(\mathcal{T}_{2})$ 
in the context 
$\mathrmbf{Tbl}_{\mathcal{A}}({\mathcal{S}_{1}{\!+_{\mathcal{S}}}\mathcal{S}_{2}})$
defines the natural join table
%$\mathcal{T}_{1}{\times_{\mathcal{T}}}\mathcal{T}_{2}$
${\mathcal{T}_{1}}{\,\boxtimes_{\mathcal{A}}}{\mathcal{T}_{2}}
=
\widehat{\mathcal{T}}_{1}\wedge\widehat{\mathcal{T}}_{2} 
= {\langle{\widehat{K}_{12},{(\hat{t}_{1},\hat{t}_{2})}}\rangle}$,
%\item 
%The intersection operation defines the \texttt{FOLE} table
%$\mathcal{T}\wedge\mathcal{T}' = {\langle{\widehat{K},{(t,t')}}\rangle}$
whose key set $\widehat{K}_{12}$ is the pullback and 
whose tuple map is the mediating function 
$\widehat{K}_{12} 
%\subseteq \widehat{K}_{1}{\times}\widehat{K}_{2}
\xrightarrow{{(\hat{t}_{1},\hat{t}_{2})}}
\mathrmbfit{tup}_{\mathcal{A}}({\mathcal{S}_{1}{\!+_{\mathcal{S}}}\mathcal{S}_{2}})$
of the opspan
$\widehat{K}_{1}\xrightarrow{\hat{t}_{1}}
\mathrmbfit{tup}_{\mathcal{A}}({\mathcal{S}_{1}{\!+_{\mathcal{S}}}\mathcal{S}_{2}})
\xleftarrow{\hat{t}_{2}}\widehat{K}_{2}$,
resulting in the span 
%of 
%${\langle{I_{1}{+}_{I}I_{2},[s_{1},s_{2}],\mathcal{A}}\rangle}$-
%table morphisms
%\newline\mbox{}\hfill
%\rule[-10pt]{0pt}{26pt}\bowtie
%\[\mbox
{\footnotesize{
{$
\widehat{\mathcal{T}}_{1}
%{i}_{1}^{\ast}(\mathcal{T}_{1})
%= {\langle{\widehat{K}_{1},\hat{t}_{1}}\rangle} 
\xleftarrow{\;\hat{\pi}_{1}\;} 
\mathcal{T}_{1}{\,\boxtimes_{\mathcal{A}}}\mathcal{T}_{2}
\xrightarrow{\;\hat{\pi}_{2}\;} 
%{\langle{\widehat{K}_{2},\hat{t}_{2}}\rangle} = 
\widehat{\mathcal{T}}_{2}
%{i}_{2}^{\ast}(\mathcal{T}_{2})
$.}}\normalsize}
%\]
%
%This results in the span of $\mathcal{A}$-table morphisms
%\item 
%
\end{itemize}
\mbox{}\newline
Inflation composed with meet 
defines the span of $\mathcal{A}$-table morphisms
\[\mbox{\footnotesize{
{$\mathcal{T}_{1}
\xleftarrow[\;\hat{\pi}_{1}{\circ\,}{\langle{\iota_{1},\grave{k}_{1}}\rangle}\;]
{\;{\langle{\iota_{1},\hat{k}_{1}}\rangle}\;} 
\mathcal{T}_{1}{\,\boxtimes_{\mathcal{A}}}\mathcal{T}_{2}
\xrightarrow[\;\hat{\pi}_{2}{\circ\,}{\langle{\iota_{2},\grave{k}_{2}}\rangle}\;]
{\;{\langle{\iota_{2},\hat{k}_{2}}\rangle}\;} 
\mathcal{T}_{2}
$,}}\normalsize}\]
%
%illustrated in 
%Fig.\;\ref{fig:fole:nat:join}.
which is the output for natural join (Tbl.\,\ref{tbl:fole:natural:join:input:output}).
\end{description}
%
%The natural join flowchart input/output is displayed in 
%Tbl.\,\ref{tbl:fole:natural:join:input:output}.
%\comment{
Natural join 
%(Fig.\;\ref{fig:fole:nat:join})
%within the context $\mathrmbf{Tbl}(\mathcal{A})$
is inflation followed by meet.
This is the two-step process 
\newline\mbox{}\hfill
\rule[-10pt]{0pt}{26pt}
$\mathcal{T}_{1}{\,\boxtimes}_{\mathcal{A}}\mathcal{T}_{2}
\doteq
\grave{\mathrmbfit{tbl}}_{\mathcal{A}}(\iota_{1})(\mathcal{T}_{1})
{\;\wedge\;}
\grave{\mathrmbfit{tbl}}_{\mathcal{A}}(\iota_{2})(\mathcal{T}_{2})$.
%
%%%%%%%%%%%%%%%%%%%%%%%%%%%%%%%%%%%%%%%%%%%%%%%%%%%%%%%%%%%%%%%%%%%%%%
{\footnote{
The natural join is empty, if one of the arguments is empty.}} 
%%%%%%%%%%%%%%%%%%%%%%%%%%%%%%%%%%%%%%%%%%%%%%%%%%%%%%%%%%%%%%%%%%%%%%
%
\hfill\mbox{}\newline
%which is the limit 
%(Chap.\;4 of \cite{kent:fole:era:tbl})
%of the table opspan Disp.\;\ref{tbl:opspan}.

%
\begin{aside}
Theoretically
this would represent pullback,
the limit 
(see the application discussion for completeness in 
\S\,\ref{sub:sec:lim:colim:tbl})
%(Chap.\;4 of \cite{kent:fole:era:tbl})
of an opspan 
{\footnotesize{$\mathcal{T}_{1}\xrightarrow{\langle{h_{1},k_{1}}\rangle} 
\mathcal{T}
\xleftarrow{\langle{h_{2},k_{2}}\rangle}\mathcal{T}_{2}$}}
of $\mathcal{A}$-tables.
But practically,
we are only given 
%
%%%%%%%%%%%%%%%%%%%%%%%%%%%%%%%%%%%%%%%%%%%%%%%%%%%%%%%%%%%%%%%%%%%%%%
%%%%%%%%%%%%%%%%%%%%%%%%%%%%%%%%%%%%%%%%%%%%%%%%%%%%%%%%%%%%%%%%%%%%%%
\footnote{In practice, 
the natural join is commonly understood to be 
the set of all combinations of tuples in $\mathcal{T}_{1}$ and $\mathcal{T}_{2}$ 
that are equal on their common attribute names.}
%%%%%%%%%%%%%%%%%%%%%%%%%%%%%%%%%%%%%%%%%%%%%%%%%%%%%%%%%%%%%%%%%%%%%%
%%%%%%%%%%%%%%%%%%%%%%%%%%%%%%%%%%%%%%%%%%%%%%%%%%%%%%%%%%%%%%%%%%%%%%
%
the constraint (span)
%an $X$-sorted signature span 
{\footnotesize{
$\mathcal{S}_{1}\xleftarrow{h_{1}}\mathcal{S}\xrightarrow{h_{2}}\mathcal{S}_{2}$}}
of $X$-sorted signatures
%in $\mathrmbf{List}(X)$
and the input tables 
$\mathcal{T}_{1}\in\mathrmbf{Tbl}_{\mathcal{A}}(\mathcal{S}_{1})$
and 
$\mathcal{T}_{1}\in\mathrmbf{Tbl}_{\mathcal{A}}(\mathcal{S}_{1})$
in Tbl.\,\ref{tbl:fole:natural:join:input:output}.
Similar comments,
which distinguish the practical from the theoretical, 
hold for the quotient operation
in \S\,\ref{sub:sub:sec:quotient}.
\end{aside}
\begin{table}
\begin{center}
{{\fbox{\begin{tabular}{c}
\setlength{\extrarowheight}{2pt}
{\scriptsize{$\begin{array}[c]{c@{\hspace{12pt}}l}
\mathcal{S}_{1}\xleftarrow{h_{1}}\mathcal{S}\xrightarrow{h_{2}}\mathcal{S}_{2}
&
\textit{constraint}
\\
\mathcal{S}_{1} \xrightarrow{\iota_{1}\,} 
{\mathcal{S}_{1}{\!+_{\mathcal{S}}}\mathcal{S}_{2}}
\xleftarrow{\;\iota_{2}}\mathcal{S}_{2}
&
\textit{construction}
\\
\hline
\mathcal{T}_{1}\in\mathrmbf{Tbl}_{\mathcal{A}}(\mathcal{S}_{1})
\text{ and }
\mathcal{T}_{2}\in\mathrmbf{Tbl}_{\mathcal{A}}(\mathcal{S}_{2})
&
\textit{input}
\\
\mathcal{T}_{1}
\xleftarrow{{\langle{\iota_{1},\hat{k}_{1}}\rangle}} 
\mathcal{T}_{1}{\,\boxtimes_{\mathcal{A}}}\mathcal{T}_{2}
\xrightarrow{{\langle{\iota_{2},\hat{k}_{2}}\rangle}} 
\mathcal{T}_{2}
&
\textit{output nat-join}
\\
\cline{2-2}
\mathcal{T}_{1}
\xleftarrow{\hat{k}_{1}} 
\mathcal{T}_{1}{\,\boxleft_{\mathcal{A}}}\mathcal{T}_{2}
\text{ \underline{or} }
\mathcal{T}_{1}{\,\boxright_{\mathcal{A}}}\mathcal{T}_{2}
\xrightarrow{\hat{k}_{2}} 
\mathcal{T}_{2}
%}
&
\textit{output semi-join}
\\
\cline{2-2}
\mathcal{T}_{1}
\xleftarrow{\bar{\omega}_{1}} 
\mathcal{T}_{1}{\,\boxslash_{\mathcal{A}}}\mathcal{T}_{2}
\text{ \underline{or} }
\mathcal{T}_{1}{\,\boxbackslash_{\mathcal{A}}}\mathcal{T}_{2}
\xrightarrow{\bar{\omega}_{2}} 
\mathcal{T}_{2}
%}
&
\textit{output anti-join}
\end{array}$}}
\end{tabular}}}}
\end{center}
\caption{\texttt{FOLE} Natural Join I/O}
\label{tbl:fole:natural:join:input:output}
\end{table}
\newpage
\paragraph{Cartesian Product.}
%%%%%%%%%%%%%%%%%%%%%%%%%%%%%%%%%%%%%%%%%%%%%%%%%%%%%%%%%%%%

%\begin{definition}
The Cartesian product 
%$\mathcal{T}_{1}{\times}\mathcal{T}_{2}$ 
%of two $\mathcal{A}$- tables 
%$\mathcal{T}_{1} = {\langle{\mathcal{S}_{1}.K_{1},t_{1}}\rangle}$
%and
%$\mathcal{T}_{2} = {\langle{\mathcal{S}_{2}.K_{2},t_{2}}\rangle}$
is a special case of the natural join.
%\begin{itemize}
%\item 
%In \texttt{FOLE},
%The Cartesian product $\mathcal{T}_{1}{\times}\mathcal{T}_{2}$ 
%of two $\mathcal{A}$- tables 
%$\mathcal{T}_{1} = {\langle{\mathcal{S}_{1}.K_{1},t_{1}}\rangle}$
%and
%$\mathcal{T}_{2} = {\langle{\mathcal{S}_{2}.K_{2},t_{2}}\rangle}$
%is a special case of a natural join
%--- 
%just link the tables through 
Let
$\mathcal{S}_{1}$
%\in
%\mathrmbf{Tbl}_{\mathcal{S}}(\mathcal{A}_{1})$
and
$\mathcal{S}_{2}$
% = {\langle{\mathcal{A}_{2}.K_{2},t_{2}}\rangle}
%\in\mathrmbf{Tbl}_{\mathcal{S}}(\mathcal{A}_{2})
be two $X$-signatures. 
These are linked by the span of $X$-signatures
%\begin{center}
{\footnotesize{$\mathcal{S}_{1}
%\rightarrow
\xhookleftarrow{0_{I_{1}}} 
\mathcal{S}_{\bot}
% = {\langle{\emptyset,0_{X},X}\rangle}
\xhookrightarrow{0_{I_{2}}} 
%\leftarrow
\mathcal{S}_{2}$}\normalsize}
with initial $X$-signature
$\mathcal{S}_{\bot} = {\langle{\emptyset,0_{X},X}\rangle}$
and
injection index functions
{\footnotesize{$
I_{1}
\xhookleftarrow{0_{I_{1}}}
\emptyset
\xhookrightarrow{0_{I_{2}}}
I_{2}$.}}
This is the \underline{constraint} for Cartesian product 
(Tbl.\;\ref{tbl:fole:cart:prod:input:output}).
It is a special case of the \underline{constraint} for Cartesian product. 
%(Tbl.\;\ref{tbl:fole:cart:prod:input:output}).
%\item 
The pushout (colimiting cocone) of this $X$-signature span
is the coproduct $X$-signature
$\mathcal{S}_{1}{+}\mathcal{S}_{2}
={\langle{I{+}_{I}I_{2},[s,s_{2}]}\rangle}$
with 
disjoint union index set $I{+}I_{2}$
and
injection $X$-signature morphisms
(opspan)
{\footnotesize{$
\mathcal{S}_{1} 
%= {\langle{I_{1},s_{1}}\rangle} 
\xhookrightarrow{\iota_{1}\,} 
%\overset{\textstyle
{\mathcal{S}_{1}{+\;}\mathcal{S}_{2}}
%}{\overbrace{{\langle{I_{1}{+}_{I}I_{2},[s_{1},s_{2}]}\rangle}}} 
\xhookleftarrow{\;\iota_{2}} 
%{\langle{I_{2},s_{2}}\rangle} = 
\mathcal{S}_{2}
$}\normalsize}
with inclusion index functions
{\footnotesize{$I_{1}\xhookrightarrow{\iota_{1}\,}
{I_{1}{+\;}I_{2}}
\xhookleftarrow{\;\iota_{2}}I_{2}.$}\normalsize}
The tuple set factors as
$\mathrmbfit{tup}_{\mathcal{A}}(\mathcal{S}_{1}{+\;}\mathcal{S}_{2})
\cong
\mathrmbfit{tup}_{\mathcal{A}}(\mathcal{S}_{1})
{\,\times\,}\mathrmbfit{tup}_{\mathcal{A}}(\mathcal{S}_{2})$.
This is the \underline{construction} for Cartesian product 
(Tbl.\;\ref{tbl:fole:cart:prod:input:output}).
%\item 
Let
$\mathcal{T}_{1} = {\langle{K_{1},t_{1}}\rangle}
\in
\mathrmbf{Tbl}_{\mathcal{A}}(\mathcal{S}_{1})$
and
$\mathcal{T}_{2} = {\langle{K_{2},t_{2}}\rangle}
\in
\mathrmbf{Tbl}_{\mathcal{S}}(\mathcal{S}_{2})$
be two $\mathcal{A}$-tables 
with
key sets 
%$K_{1}$ and $K_{2}$
and a tuple functions
$K_{1}\xrightarrow{\,t_{1}\;}\mathrmbfit{tup}_{\mathcal{A}}(\mathcal{S}_{1})$
and
$K_{2}\xrightarrow{\,t_{2}\;}\mathrmbfit{tup}_{\mathcal{A}}(\mathcal{S}_{2})$.
This is the \underline{input} for Cartesian product
(Tbl.\,\ref{tbl:fole:cart:prod:input:output}).
%\item 
The Cartesian product $\mathcal{T}_{1}{\,\times\,}\mathcal{T}_{2}$ 
of the two $\mathcal{A}$-tables 
$\mathcal{T}_{1}$
and
$\mathcal{T}_{2}$
is a special case of natural join
--- 
just link the tables through
the span of tuple functions
$\mathrmbf{tup}_{\mathcal{A}}(\mathcal{S}_{1})
{\;\xleftarrow{\mathrmbfit{tbl}_{\mathcal{A}}(\iota_{1})}\;}
\mathrmbf{tup}_{\mathcal{A}}(\mathcal{S}_{1}{+}\mathcal{S}_{2})
{\;\xrightarrow{\mathrmbfit{tbl}_{\mathcal{A}}(\iota_{2})}\;}
\mathrmbf{tup}_{\mathcal{A}}(\mathcal{S}_{2})$,
and then use inflation (twice) and intersection.
%\item 
The Cartesian product table $\mathcal{T}_{1}{\,\times\,}\mathcal{T}_{2}$, 
which has the binary product key set $K_{1}{\times}K_{2}$
with product tuple function
%
%%%%%%%%%%%%%%%%%%%%%%%%%%%%%%%%%%%%%%%%%%%%%%%%%%%%%%%%%%%%%%%%%%%%%%%%%%%%%%%%
%%%%%%%%%%%%%%%%%%%%%%%%%%%%%%%%%%%%%%%%%%%%%%%%%%%%%%%%%%%%%%%%%%%%%%%%%%%%%%%%
\footnote{The tuple subset
of the Cartesian product table $\mathcal{T}_{1}{\,\times\,}\mathcal{T}_{2}$ 
is the Cartesian product of the tuple sets
${\wp{t_{1}}}(K_{1}) 
\subseteq \mathrmbfit{tup}_{\mathcal{A}}(\mathcal{S}_{1})
\subseteq \mathrmbf{List}(Y_{1})$
and
${\wp{t_{2}}}(K_{2}) \subseteq \mathrmbfit{tup}_{\mathcal{A}}(\mathcal{S}_{2})
\subseteq \mathrmbf{List}(Y_{2})$.}
%%%%%%%%%%%%%%%%%%%%%%%%%%%%%%%%%%%%%%%%%%%%%%%%%%%%%%%%%%%%%%%%%%%%%%%%%%%%%%%%
%%%%%%%%%%%%%%%%%%%%%%%%%%%%%%%%%%%%%%%%%%%%%%%%%%%%%%%%%%%%%%%%%%%%%%%%%%%%%%%%
%
$K_{1}{\times}K_{2} 
\xrightarrow{t_{1}{\times}t_{1}} 
\mathrmbfit{tup}_{\mathcal{A}}(\mathcal{S}_{1})
{\,\times\,}\mathrmbfit{tup}_{\mathcal{A}}(\mathcal{S}_{2})$,
is linked to the component tables with the span of projection table morphisms
\[\mbox
{\footnotesize{
{$\mathcal{T}
\xleftarrow{{\langle{\iota_{1},\pi_{1}}\rangle}} 
\mathcal{T}_{1}{\,\times_{\mathcal{S}}}\mathcal{T}_{2}
\xrightarrow{{\langle{\iota_{2},\pi_{2}}\rangle}} 
\mathcal{T}_{2}
$.}}\normalsize}
\]
This is the \underline{output} for Cartesian product 
(Tbl.\;\ref{tbl:fole:cart:prod:input:output}).
%\item 
%\end{itemize}
%
%\end{definition}
%

%
\begin{table}
\begin{center}
{{\fbox{\begin{tabular}{c}
\setlength{\extrarowheight}{2pt}
{\scriptsize{$\begin{array}[c]{c@{\hspace{12pt}}l}
\mathcal{S}_{1}  
\text{ and }
\mathcal{S}_{2} 
&
\textit{constraint}
\\
\mathcal{S}_{1} \xhookrightarrow{\iota_{1}\,} 
{\mathcal{S}_{1}{+}\mathcal{S}_{2}}
\xhookleftarrow{\;\iota_{2}}\mathcal{S}_{2}
&
\textit{construction}
\\
\hline
\mathcal{T}_{1}\in\mathrmbf{Tbl}_{\mathcal{A}}(\mathcal{S}_{1})
\text{ and }
\mathcal{T}_{2}\in\mathrmbf{Tbl}_{\mathcal{A}}(\mathcal{S}_{2})
&
\textit{input}
\\
\mathcal{T}
\xleftarrow{{\langle{\iota_{1},\pi_{1}}\rangle}} 
\mathcal{T}_{1}{\,\times\,}\mathcal{T}_{2}
\xrightarrow{{\langle{\iota_{2},\pi_{2}}\rangle}} 
\mathcal{T}_{2}
&
\textit{output}
\end{array}$}}
\end{tabular}}}}
\end{center}
\caption{\texttt{FOLE} Cartesian Product I/O}
\label{tbl:fole:cart:prod:input:output}
\end{table}

\begin{proposition}\label{join:preserve:join:meet}
Natural join $\boxtimes$ distributes over union $\vee$ and intersection $\wedge$
\[\mbox{\footnotesize{{$
\mathcal{T}_{1}{\,\boxtimes_{\mathcal{A}}}
\bigl(\mathcal{T}_{2}{\;\wedge\;}\mathcal{T}'_{2}\bigr)
\cong
\bigl(\mathcal{T}_{1}{\,\boxtimes_{\mathcal{A}}}\mathcal{T}_{2}\bigr)
{\;\wedge\;}
\bigl(\mathcal{T}_{1}{\,\boxtimes_{\mathcal{A}}}
\mathcal{T}'_{2}\bigr)
$.}}\normalsize}\]
\end{proposition}
\begin{proof}
Inflation is continuous and co-continuous.
Intersection $\wedge$ is distributive over itself and union $\vee$.
\hfill\rule{5pt}{5pt}
\end{proof}
\begin{proposition}\label{join:assoc}
Natural join is associative
\[\mbox{\footnotesize{{$
\bigl(\mathcal{T}_{1}{\,\boxtimes_{\mathcal{A}}}\mathcal{T}_{2}\bigr)
{\,\boxtimes_{\mathcal{A}}}\mathcal{T}_{3} 
\cong
\mathcal{T}_{1}{\,\boxtimes_{\mathcal{A}}}
\bigl(\mathcal{T}_{2}{\,\boxtimes_{\mathcal{A}}}\mathcal{T}_{3}\bigr)
$.}}\normalsize}\]
\end{proposition}
\begin{proof}
Basic category theory;
see
Saunders Mac\;Lane
\cite{maclane:71}
\hfill\rule{5pt}{5pt}
\end{proof}
%

%%%%%%%%%%%%%%%%%%%%%%%%%%%%%%%%%%%%%%%%%%%%%%%%%%%%%%%%%%%%%%%%%%%%%%
%%%%%%%%%%%%%%%%%%%%%%%%%%%%%%%%%%%%%%%%%%%%%%%%%%%%%%%%%%%%%%%%%%%%%%
%\newpage
%\mbox{}\newline
%{\fbox{$\blacktriangledown$\hspace{50pt}
%\textbf{Work zone: limits by equalizer and pullback.}
%\hspace{50pt}$\blacktriangledown$}}
%\newline
%\comment{

%
\begin{aside}
At the intermediate scope, 
in the context $\mathrmbf{Tbl}(\mathcal{A})$ of a type domain table fiber,
generic meets 
--- for the special case of equalizer --- 
are called quotients, and
generic meets 
--- for the special case of pullback --- 
are called natural joins. 
%the join of two $\mathcal{A}$-tables.
%}
%
\begin{fact}
(Mac Lane \cite{maclane:71})
Limits can be constructed from equalizers and multi-pullbacks.
\end{fact}
\begin{proof}
%{\itshape Categories for the Working Mathematician}
Equalizers can be constructed from products and pullbacks.
%In particular,
%an equalizer of
%$f,g : B \rightrightarrows A$
%can be constructed as the pullback of the opspan
%$B \xrightarrow{(1_{B},f)} B{\,\times\,}A \xleftarrow{(1_{B},g)} B$.
%\end{fact}
%
%\begin{fact}
%(Mac Lane \cite{maclane:71})
Limits can be constructed from products and equalizers.
%In particular,
%an equalizer of
%$f,g : B \rightrightarrows A$
%can be constructed as the pullback of the opspan
%$B \xrightarrow{(1_{B},f)} B{\,\times\,}A \xleftarrow{(1_{B},g)} B$.
%\end{fact}
%
%\begin{fact}
Binary products are pullbacks from the terminal object.
%and
%.
%In particular,
%a binary product
%$A_{1} \xleftarrow{p_{1}} A_{1}{\times}A_{2} \xrightarrow{p_{2}} A_{2}$
%can be constructed as the pullback of the opspan
%$A_{1} \xrightarrow{!} 1 \xleftarrow{!} A_{2}$
%from the terminal object $1$.
%
Arbitrary products
%$\{ (A_{1} {\times} ... {\times} A_{n})
%%\prod_{i} A_{i} 
%\xrightarrow{p_{i}} A_{i} \mid i \in \{1 ... n\} \}$
are iterated binary products.
%$(A_{1} {\times} ... {\times} A_{n})
%= (A_{1} {\times} (A_{2} {\times} ... {\times} A_{n}))$.
\hfill\rule{5pt}{5pt}
\end{proof}
\end{aside}

% adjoint flow in Tbl(A)
%%%%%%%%%%%%%%%%%%%%%%%%%%%%%%%%%%%%%%%%%%%%%%%%%%%%%%%%%%%%%%%%%%%%%%%%%%%%%%%%
%%%%%%%%%%%%%%%%%%%%%%%%%%%%%%%%%%%%%%%%%%%%%%%%%%%%%%%%%%%%%%%%%%%%%%%%%%%%%%%%
%

%%%%%%%%%%%%%%%%%%%%%%%%%%%%%%%%%%%%%%%%%%%%%%%%%%%%%%%%%%%%%%
%
\newpage
\subsubsection{Semi-join.}\label{sub:sub:sec:semi:join}
%$\ast\ast\ast$}
%%%%%%%%%%%%%%%%%%%%%%%%%%%%%%%%%%%%%%%%%%%%%%%%%%%%%%%%%%%%%%

%
\begin{figure}
\begin{center}
{{{\begin{tabular}{c}
\begin{picture}(150,40)(40,35)
\setlength{\unitlength}{0.97pt}
%%%%%%%%%%%%%%%%%%%%%%%%%%%%%%%%%%%%%%%%%%%%%%%%%%
\put(98.4,58){\begin{picture}(0,0)(0,3)
\setlength{\unitlength}{0.35pt}
\put(60,23){\makebox(0,0){\normalsize{$\boxtimes$}}}
%\thicklines
%\put(33,76){\makebox(0,0){\normalsize{$\boldsymbol{\circ}$}}}
%\put(87,76){\makebox(0,0){\normalsize{$\boldsymbol{\circ}$}}}
%\put(60,3){\makebox(0,0){\normalsize{$\boldsymbol{\circ}$}}}
\put(40,10){\line(1,0){40}}
\put(10,70){\line(1,0){100}}
\put(10,70){\line(0,-1){30}}
\put(110,70){\line(0,-1){30}}
\put(40,40){\oval(60,60)[bl]}
\put(80,40){\oval(60,60)[br]}
\put(61,57){\makebox(0,0){\scriptsize{{\textit{{natural}}}}}}
\put(61,40){\makebox(0,0){\scriptsize{{\textit{{join}}}}}}
\end{picture}}
%%%%%%%%%%%%%%%%%%%%%%%%%%%%%%%%%%%%%%%%%%%%%%%%%%
\comment{\put(1,25.5){\begin{picture}(0,0)(0,0)
\setlength{\unitlength}{0.35pt}
%\thicklines
%\put(106.5,40){\makebox(0,0){\normalsize{$\boldsymbol{\circ}$}}}
%\put(4.7,40){\makebox(0,0){\normalsize{$\boldsymbol{\circ}$}}}
\put(40,10){\line(1,0){61}}
\put(40,70){\line(1,0){61}}
\put(100,70){\line(0,-1){60}}
\put(40,40){\oval(60,60)[bl]}
\put(40,40){\oval(60,60)[tl]}
\put(57,30){\makebox(0,0){\large{${\Leftarrow}$}}}
\put(61,50){\makebox(0,0){\scriptsize{{\textit{{image}}}}}}
\end{picture}}}
%%%%%%%%%%%%%%%%%%%%%%%%%%%%%%%%%%%%%%%%%%%%%%%%%%
\put(53.5,25.5){\begin{picture}(0,0)(0,0)
\setlength{\unitlength}{0.35pt}
%\thicklines
%\put(106,40){\makebox(0,0){\normalsize{$\boldsymbol{\circ}$}}}
%\put(4.7,40){\makebox(0,0){\normalsize{$\boldsymbol{\circ}$}}}
\put(40,10){\line(1,0){60}}
\put(40,70){\line(1,0){60}}
\put(100,70){\line(0,-1){60}}
\put(40,40){\oval(60,60)[bl]}
\put(40,40){\oval(60,60)[tl]}
\put(58,50){\makebox(0,0){\scriptsize{{\textit{{project}}}}}}
\put(56,30){\makebox(0,0){\Large{${\Leftarrow}$}}}
%%\put(61,22){\makebox(0,0){\scriptsize{{\textit{{image}}}}}}
\end{picture}}
%%%%%%%%%%%%%%%%%%%%%%%%%%%%%%%%%%%%%%%%%%%%%%%%%%
\put(146,25.5){\begin{picture}(0,0)(0,0)
\setlength{\unitlength}{0.35pt}
%\put(0,40){\line(1,0){130}}
%\thicklines
%\put(106,40){\makebox(0,0){\normalsize{$\boldsymbol{\circ}$}}}
%\put(5.8,40){\makebox(0,0){\normalsize{$\boldsymbol{\circ}$}}}
\put(10,10){\line(1,0){61}}
\put(10,70){\line(1,0){61}}
\put(11,70){\line(0,-1){60}}
\put(70,40){\oval(60,60)[br]}
\put(70,40){\oval(60,60)[tr]}
\put(55,50){\makebox(0,0){\scriptsize{{\textit{{project}}}}}}
\put(56,30){\makebox(0,0){\Large{${\Rightarrow}$}}}
%\put(56,22){\makebox(0,0){\scriptsize{{\textit{{image}}}}}}
\end{picture}}
%%%%%%%%%%%%%%%%%%%%%%%%%%%%%%%%%%%%%%%%%%%%%%%%%%
\comment{\put(198,25.5){\begin{picture}(0,0)(0,0)
\setlength{\unitlength}{0.35pt}
%\put(0,40){\line(1,0){130}}
%\thicklines
%\put(106,40){\makebox(0,0){\normalsize{$\boldsymbol{\circ}$}}}
%\put(5.8,40){\makebox(0,0){\normalsize{$\boldsymbol{\circ}$}}}
\put(10,10){\line(1,0){61}}
\put(10,70){\line(1,0){61}}
\put(11,70){\line(0,-1){60}}
\put(70,40){\oval(60,60)[br]}
\put(70,40){\oval(60,60)[tr]}
\put(57,30){\makebox(0,0){\Large{${\Rightarrow}$}}}
\put(56,50){\makebox(0,0){\scriptsize{{\textit{{image}}}}}}
%\put(56,22){\makebox(0,0){\scriptsize{{\textit{{image}}}}}}
\end{picture}}}
%%%%%%%%%%%%%%%%%%%%%%%%%%%%%%%%%%%%%%%%%%%%%%%%%%
\put(120,30){\makebox(0,0){\footnotesize{{\textit{{semi-join}}}}}}
\put(48,32){\makebox(0,0){\footnotesize{{\textit{{left}}}}}}
\put(48,46){\makebox(0,0){{$\boxleft$}}}
\put(192,46){\makebox(0,0){{$\boxright$}}}
\put(193,32){\makebox(0,0){\footnotesize{{\textit{{right}}}}}}
%\put(70,18){\makebox(0,0){\footnotesize{{\textit{{left}}}}}}
%\put(86,18){\makebox(0,0){\large{$\rhd$}}}
%\put(154,18){\makebox(0,0){\large{$\lhd$}}}
%\put(172,18){\makebox(0,0){\footnotesize{{\textit{{right}}}}}}
%%%%%%%%%%%%%%%%%%%%%%%%%%%%%%%%%%%%%%%%%%%%%%%%%%
\put(57,40){\vector(-1,0){20}}
\put(120,58.5){\line(0,-1){18.5}}
\put(120,40){\vector(-1,0){30}}
\put(120,40){\vector(1,0){30}}
\put(182,40){\vector(1,0){20}}
%\put(234,40){\vector(1,0){20}}
%%%%%%%%%%
\put(110,92){\vector(0,-1){12}}
\put(130,92){\vector(0,-1){12}}
%\thicklines
%\put(15,110){\line(1,0){210}}
%\put(55,10){\line(1,0){130}}
%\put(15,50){\line(0,1){60}}
%\put(225,50){\line(0,1){60}}
%\qbezier(15,50)(15,10)(55,10)
%\qbezier(185,10)(225,10)(225,50)
%%%%%%%%%%%%%%%%%%%%%%%%%%%%%%%%%%%%%%%%%%%%%%%%%%
\end{picture}
\end{tabular}}}}
\end{center}
\caption{\texttt{FOLE} Semi-Join Flow Chart}
\label{fole:semi:join:flo:chrt}
\end{figure}
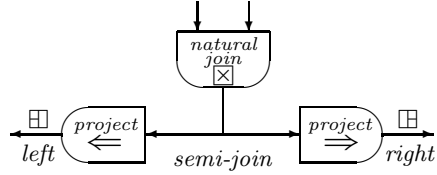
%
%The left and right semi-join are two other join-like operations.
%
Let $\mathcal{A}$
% = {\langle{X,Y,\models_{\mathcal{A}}}\rangle}$
be a fixed type domain.
For any two $\mathcal{A}$-tables 
$\mathcal{T}_{1}\in\mathrmbf{Tbl}_{\mathcal{A}}(\mathcal{S}_{1})$
 and 
$\mathcal{T}_{2}\in\mathrmbf{Tbl}_{\mathcal{A}}(\mathcal{S}_{2})$
that are linked through an $X$-sorted signature span 
$\mathcal{S}_{1}\xleftarrow{h_{1}}\mathcal{S}\xrightarrow{h_{2}}\mathcal{S}_{2}$
in $\mathrmbf{List}(X)$,
the left semi-join $\mathcal{T}_{1}{\;\boxleft_{\mathcal{A}}}\mathcal{T}_{2}$
is the set of all tuples in $\mathcal{T}_{1}$ 
for which there is a tuple in $\mathcal{T}_{2}$ 
that is equal on their common attribute names;
the other columns of $\mathcal{T}_{2}$ do not appear. 
Hence,
the left semi-join is defined to be the projection from the natural join.
The right semi-join is similar.
We use the following routes of flow from 
%Tbl.\;\ref{tbl:routes:flow}.
Fig.\;\ref{fig:routes:flow:lim}.
%followed by image-factorization.
%using the header of $\mathcal{T}$. 
%If $a,..,a_{n}$ are the attribute names of $R$, then 
%$R{\;\rhd\,}S = \pi_{a,..,a_{n}}(R \bowtie S)$.
%\mbox{}\newline
%{\fbox{\textbf{Work zone: key factorization of projection.}}}
%\newline
%
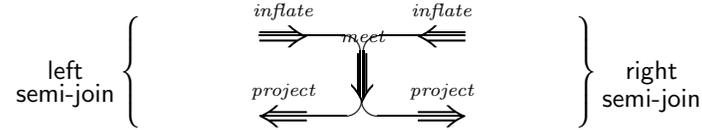
\begin{center}
{{\begin{tabular}{c}
\setlength{\unitlength}{0.6pt}
\begin{picture}(320,80)(0,-10)
%\put(157,63){\makebox(0,0){\normalsize{$\textbf{4.}$}}}
%
\put(325,25){\makebox(0,0){\normalsize{
$\left.\rule{0pt}{24pt}\right\}
\underset{\textstyle{\textsf{semi-join}}}{\textsf{right}}$}}}
\put(-25,25){\makebox(0,0){\normalsize{
$\underset{\textstyle{\textsf{semi-join}}}{\textsf{left}}
\left\{\rule{0pt}{24pt}\right.$}}}
\put(100,55){\makebox(0,0){\huge{
$\overset{\textit{\scriptsize{inflate}}}{\Rightarrow}$}}}
\put(200,55){\makebox(0,0){\huge{
$\overset{\textit{\scriptsize{inflate}}}{\Leftarrow}$}}}
\put(150,30){\makebox(0,0){\huge{
$\overset{\textit{\scriptsize{meet}}}{\Downarrow}$}}}
\put(100,4.5){\makebox(0,0){\huge{
$\overset{\textit{\scriptsize{project}}}{\Leftarrow}$}}}
%\put(105,-14.5){\makebox(0,0){$\textit{\scriptsize{image}}$}}
\put(200,4.5){\makebox(0,0){\huge{
$\overset{\textit{\scriptsize{project}}}{\Rightarrow}$}}}
%\put(205,-14.5){\makebox(0,0){$\textit{\scriptsize{image}}$}}
%
\put(90,48){\line(1,0){54}}
\put(144,38){\oval(20,20)[tr]}
\put(166,38){\oval(20,20)[tl]}
\put(220,48){\line(-1,0){54}}
\put(145,-3){\line(-1,0){54}}
\put(165,-3){\line(1,0){54}}
\put(154,9){\line(0,1){29}}
\put(156,9){\line(0,1){29}}
\put(145,7){\oval(20,20)[br]}
\put(165,7){\oval(20,20)[bl]}
\end{picture}
\end{tabular}}}
\end{center}
%
%\newpage
%Left semi-join is the two-step process
%illustrated in Fig.????
The constraint, construction and input for semi-join 
are identical to that for natural join.
Only the output is different.
\begin{description}
\item[Constraint:] 
The constraint for semi-join is the same as the constraint for natural join
(Tbl.\,\ref{tbl:fole:natural:join:input:output}):
an $X$-sorted signature span 
$\mathcal{S}_{1}\xleftarrow{h_{1}}\mathcal{S}\xrightarrow{h_{2}}\mathcal{S}_{2}$
in $\mathrmbf{List}(X)$.
\newline
\item[Construction:] 
The construction for semi-join is the same as the construction for natural join 
(Tbl.\,\ref{tbl:fole:natural:join:input:output}):
the opspan
{\footnotesize{$\mathcal{S}_{1}\xrightarrow{\iota_{1}\,} 
{\mathcal{S}_{1}{\!+_{\mathcal{S}}}\mathcal{S}_{2}}
\xleftarrow{\;\iota_{2}}\mathcal{S}_{2}$}\normalsize}
of injection $X$-signature morphisms
with pushout signature
$\mathcal{S}_{1}{+_{\mathcal{S}}}\mathcal{S}_{2}$.
\newline
\item[Input:] 
The input for semi-join is the same as the input for natural join 
(Tbl.\,\ref{tbl:fole:natural:join:input:output}):
a pair of tables
$\mathcal{T}_{1} = {\langle{K_{1},t_{1}}\rangle} \in 
\mathrmbf{Tbl}_{\mathcal{A}}(\mathcal{S}_{1})$
and
$\mathcal{T}_{2} = {\langle{K_{2},t_{2}}\rangle} \in 
\mathrmbf{Tbl}_{\mathcal{A}}(\mathcal{S}_{2})$.
\newline
\item[Output:] 
The output 
%for semi-join is defined in two steps:
is natural join followed by projection.
\newline
\begin{itemize}
\item 
%Natural join is defined by the two-step process
%of inflation followed by meet.  
%\newline\mbox{}\hfill
%\rule[-10pt]{0pt}{26pt}.
%\hfill\mbox{}\newline
%
Natural join results in the table
$\mathcal{T}_{1}{\,\boxtimes}_{\mathcal{A}}\mathcal{T}_{2}
\doteq
\grave{\mathrmbfit{tbl}}_{\mathcal{A}}(\iota_{1})(\mathcal{T}_{1})
{\;\wedge\;}
\grave{\mathrmbfit{tbl}}_{\mathcal{A}}(\iota_{2})(\mathcal{T}_{2})$
%{\scriptstyle\sum}_{\iota}(\mathcal{T}{\times}_{\mathcal{T}_{2}_{2}}\mathcal{T}_{2})=
%$\mathcal{T}_{1}{\,\bowtie_{\mathcal{S}}}\mathcal{T}_{2}$
%= {\langle{\hat{K},\hat{t}_{1}}\rangle}$ 
with key set
$\widehat{K}_{12}$ and
tuple function
%$K \xrightarrow{t'} \mathrmbfit{tup}_{\mathcal{A}}(\mathcal{S})$.
%defined by composition,
%$t' = t{\times}_{t''}t'{\,\cdot\,}\mathrmbfit{tup}_{\mathcal{A}}(\iota)$.
$
%\mathrmbfit{tup}_{\mathcal{A}}(\mathcal{S}_{1})
%\xleftarrow
%[\iota_{1}{\,\cdot\,}{(\mbox{-})}]
%{\mathrmbfit{tup}_{\mathcal{A}}(\iota_{1})}
\widehat{K}_{12}
\xrightarrow{(\hat{t}_{1},\hat{t}_{2})}
\mathrmbfit{tup}_{\mathcal{A}}({\mathcal{S}_{1}{+_{\mathcal{S}}}\mathcal{S}_{2}})$.
\newline
\item 
Projection 
{\footnotesize{$\mathrmbf{Tbl}_{\mathcal{A}}(\mathcal{S}_{1})
{\;\xleftarrow{\;\acute{\mathrmbfit{tbl}}_{\mathcal{A}}(\iota_{1})\;}\;}
\mathrmbf{Tbl}_{\mathcal{A}}({\mathcal{S}{+_{\mathcal{S}}}\mathcal{S}_{2}})$}\normalsize}
(\S\,\ref{sub:sub:sec:adj:flow:A})
along the tuple function 
of the $X$-signature morphism
{\footnotesize{$\mathcal{S}_{1}\xrightarrow{\iota_{1}\,} 
{\mathcal{S}_{1}{\!+_{\mathcal{S}}}\mathcal{S}_{2}}$}\normalsize}
maps
the natural join table
$\mathcal{T}_{1}{\,\boxtimes_{\mathcal{A}}}\mathcal{T}_{2}$
to the left semi-join table
${\mathcal{T}_{1}}{\,\boxleft_{\mathcal{A}}}{\mathcal{T}_{2}} =
\acute{\mathrmbfit{tbl}}_{\mathcal{A}}(\iota_{1})(
\mathcal{T}_{1}{\,\boxtimes_{\mathcal{A}}}\mathcal{T}_{2})
= {\langle{\hat{K}_{12},\acute{t}_{1}}\rangle}$ 
with key set $\widehat{K}_{12}$ and
tuple function
$\mathrmbfit{tup}_{\mathcal{A}}(\mathcal{S}_{1})
\xleftarrow{\acute{t}_{1}}\widehat{K}_{12}$
defined by composition 
$\mathrmbfit{tup}_{\mathcal{A}}(\mathcal{S}_{1})
\xleftarrow{\mathrmbfit{tup}_{\mathcal{A}}(\iota_{1})}
\mathrmbfit{tup}_{\mathcal{A}}({\mathcal{S}_{1}{+_{\mathcal{S}}}\mathcal{S}_{2}})
\xleftarrow{(\hat{t}_{1},\hat{t}_{2})}\widehat{K}_{12}$.
\newline
\end{itemize}
%
%\end{description}
%
%The semi-join flowchart input/output is displayed in 
%Tbl.\,\ref{tbl:fole:natural:join:input:output}.
Semi-join is natural join followed by projection.
For left semi-join
this is the two-step process 
\newline\mbox{}\hfill
\rule[-10pt]{0pt}{26pt}
${\mathcal{T}_{1}}{\,\boxleft_{\mathcal{A}}}{\mathcal{T}_{2}} 
\doteq
\acute{\mathrmbfit{tbl}}_{\mathcal{A}}(\iota_{1})
\bigl(\mathcal{T}_{1}{\,\boxtimes_{\mathcal{A}}}\mathcal{T}_{2}\bigr)$. 
\hfill\mbox{}\newline
%\begin{itemize}
%\item 
This defines the table morphism 
%\[\mbox\newline
{\footnotesize{{$
\mathcal{T}_{1}{\,\boxleft_{\mathcal{A}}}\mathcal{T}_{2} 
%= \acute{\mathcal{T}}_{1} = {\langle{\widehat{K}_{12},\acute{t}_{1}}\rangle}
\xleftarrow{\;{\langle{\iota_{1},1}\rangle}\,} 
%{\langle{\widehat{K}_{12},(\hat{t}_{1},\hat{t}_{2})}\rangle} = 
{\mathcal{T}_{1}}{\,\boxtimes_{\mathcal{A}}}{\mathcal{T}_{2}}
$}}\normalsize}
%\]
in $\mathrmbf{Tbl}(\mathcal{A})$.
%\item 
There is a sub-table relationship 
%(LHS Fig.\;\ref{fole:semi:join})
%\[\mbox\newline
{\footnotesize{{$
\mathcal{T}_{1} 
%= {\langle{K_{1},t_{1}}\rangle}
\xleftarrow{\;\hat{k}_{1}\,} 
%{\langle{\widehat{K}_{12},\acute{t}_{1}}\rangle} = \acute{\mathcal{T}}_{1} = 
\mathcal{T}_{1}{\,\boxleft_{\mathcal{A}}}\mathcal{T}_{2}
$}}\normalsize}
%\]
in the small fiber table context $\mathrmbf{Tbl}_{\mathcal{A}}(\mathcal{S}_{1})$, 
which is the output of left semi-join.
%}
%}
%\end{itemize}
%
The right semi-join has a similar definition
(Tbl.\,\ref{tbl:fole:natural:join:input:output}).
\end{description}
These factor 
%\[\mbox
{\footnotesize{
{$
\mathcal{T}_{1}
\xrightarrow{\hat{k}_{1}} 
\mathcal{T}_{1}{\,\boxleft_{\mathcal{A}}}\mathcal{T}_{2}
\xrightarrow{{\langle{\iota_{1},1}\rangle}} 
\mathcal{T}_{1}{\,\boxtimes_{\mathcal{A}}}\mathcal{T}_{2}
\xleftarrow{{\langle{\iota_{2},1}\rangle}} 
\mathcal{T}_{1}{\,\boxright_{\mathcal{A}}}\mathcal{T}_{2}
\xleftarrow{\hat{k}_{2}}
\mathcal{T}_{2}
$,}}\normalsize}
%\]
%
the span of $\mathcal{A}$-table morphisms
%\[\mbox
{\footnotesize{
{$\mathcal{T}_{1}
\xleftarrow[\;\hat{k}_{1}{\circ\,}{\langle{\iota_{1},1}\rangle}\;]
{\;{\langle{\iota_{1},\hat{k}_{1}}\rangle}\;} 
\mathcal{T}_{1}{\,\boxtimes_{\mathcal{A}}}\mathcal{T}_{2}
\xrightarrow[\;\hat{k}_{2}{\circ\,}{\langle{\iota_{2},1}\rangle}\;]
{\;{\langle{\iota_{2},\hat{k}_{2}}\rangle}\;} 
\mathcal{T}_{2}
$,}}\normalsize}
%\]
%
which is the output for natural join. 
%
%%%%%%%%%%%%%%%%%%%%%%%%%%%%%%%%%%%%%%%%%%%%%%%%%%%%%%%%%%%%%%%%%%%%%%
%%%%%%%%%%%%%%%%%%%%%%%%%%%%%%%%%%%%%%%%%%%%%%%%%%%%%%%%%%%%%%%%%%%%%%
{\footnote{The semi-join of a Cartesian product of non-empty tables 
gives either of the tables: the left semi-join gives the left table,
and the right gives the right.}}
\newpage
\subsubsection{Anti-join.}\label{sub:sub:sec:anti:join}
%$\ast{-}{-}$}
%%%%%%%%%%%%%%%%%%%%%%%%%%%%%%%%%%%%%%%%%%%%%%%%%%%%%%%%%%%%%%
%
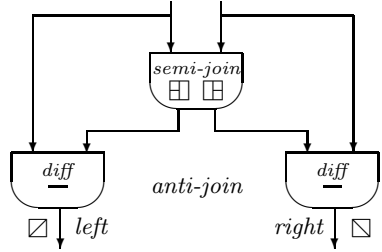
\begin{figure}
\begin{center}
{{{\begin{tabular}{c}
\begin{picture}(130,80)(50,15)
\setlength{\unitlength}{0.97pt}
%%%%%%%%%%%%%%%%%%%%%%%%%%%%%%%%%%%%%%%%%%%%%%%%%%
\put(98,60){\begin{picture}(0,0)(0,3)
\setlength{\unitlength}{0.35pt}
\put(42,30){\makebox(0,0){\normalsize{$\boxleft$}}}
\put(78,30){\makebox(0,0){\normalsize{$\boxright$}}}
%\thicklines
%\put(33,76){\makebox(0,0){\normalsize{$\boldsymbol{\circ}$}}}
%\put(87,76){\makebox(0,0){\normalsize{$\boldsymbol{\circ}$}}}
%\put(40,3){\makebox(0,0){\normalsize{$\boldsymbol{\circ}$}}}
%\put(80,3){\makebox(0,0){\normalsize{$\boldsymbol{\circ}$}}}
\put(40,10){\line(1,0){40}}
\put(10,70){\line(1,0){100}}
\put(10,70){\line(0,-1){30}}
\put(110,70){\line(0,-1){30}}
\put(40,40){\oval(60,60)[bl]}
\put(80,40){\oval(60,60)[br]}
%\put(61,52){\makebox(0,0){\scriptsize{{\textit{{left}}}}}}
\put(61,52){\makebox(0,0){\scriptsize{{\textit{{semi-join}}}}}}
\end{picture}}
%%%%%%%%%%%%%%%%%%%%%%%%%%%%%%%%%%%%%%%%%%%%%%%%%%
\put(44,21.3){\begin{picture}(0,0)(0,3)
\setlength{\unitlength}{0.35pt}
%\put(60,40){\makebox(0,0){\large{$-$}}}
%\thicklines
\put(50,33){\line(1,0){20}}
%\put(33,76){\makebox(0,0){\normalsize{$\boldsymbol{\circ}$}}}
%\put(87,76){\makebox(0,0){\normalsize{$\boldsymbol{\circ}$}}}
%\put(60,3){\makebox(0,0){\normalsize{$\boldsymbol{\circ}$}}}
\put(40,10){\line(1,0){40}}
\put(10,70){\line(1,0){100}}
\put(10,70){\line(0,-1){30}}
\put(110,70){\line(0,-1){30}}
\put(40,40){\oval(60,60)[bl]}
\put(80,40){\oval(60,60)[br]}
\put(61,50){\makebox(0,0){\scriptsize{{\textit{{diff}}}}}}
\end{picture}}
%%%%%%%%%%%%%%%%%%%%%%%%%%%%%%%%%%%%%%%%%%%%%%%%%%
\put(151,21.3){\begin{picture}(0,0)(0,3)
\setlength{\unitlength}{0.35pt}
%\put(60,40){\makebox(0,0){\large{$-$}}}
%\thicklines
\put(50,33){\line(1,0){20}}
%\put(33,76){\makebox(0,0){\normalsize{$\boldsymbol{\circ}$}}}
%\put(87,76){\makebox(0,0){\normalsize{$\boldsymbol{\circ}$}}}
%\put(60,3){\makebox(0,0){\normalsize{$\boldsymbol{\circ}$}}}
\put(40,10){\line(1,0){40}}
\put(10,70){\line(1,0){100}}
\put(10,70){\line(0,-1){30}}
\put(110,70){\line(0,-1){30}}
\put(40,40){\oval(60,60)[bl]}
\put(80,40){\oval(60,60)[br]}
\put(61,50){\makebox(0,0){\scriptsize{{\textit{{diff}}}}}}
\end{picture}}
%%%%%%%%%%%%%%%%%%%%%%%%%%%%%%%%%%%%%%%%%%%%%%%%%%
\put(120,30){\makebox(0,0){\footnotesize{{\textit{{anti-join}}}}}}
\put(72,14){\makebox(0,0)[l]{\footnotesize{{\textit{{left}}}}}}
\put(58,14){\makebox(0,0){{$\boxslash$}}}
\put(184,14){\makebox(0,0){{$\boxbackslash$}}}
\put(170,14){\makebox(0,0)[r]{\footnotesize{{\textit{{right}}}}}}
%\put(70,18){\makebox(0,0){\footnotesize{{\textit{{left}}}}}}
%\put(86,18){\makebox(0,0){\large{$\rhd$}}}
%\put(154,18){\makebox(0,0){\large{$\lhd$}}}
%\put(172,18){\makebox(0,0){\footnotesize{{\textit{{right}}}}}}
%%%%%%%%%%%%%%%%%%%%%%%%%%%%%%%%%%%%%%%%%%%%%%%%%%
%\put(49,40){\vector(-1,0){20}}
%\put(191,40){\vector(1,0){20}}
\put(112.7,60.3){\line(0,-1){8.5}}
\put(112.7,52){\line(-1,0){36}}
\put(77,46){\line(0,1){6}}
\put(77,43){\vector(0,-1){0}}
\put(127,60.3){\line(0,-1){8.5}}
\put(127,52){\line(1,0){36}}
\put(163,46){\line(0,1){6}}
\put(163,43){\vector(0,-1){0}}
\put(109.7,102.2){\vector(0,-1){20}}
\put(129.3,102.2){\vector(0,-1){20}}
\put(110,97){\line(-1,0){54}}
\put(56,97){\vector(0,-1){53.5}}
\put(129,97){\line(1,0){52}}
\put(181,97){\vector(0,-1){53.5}}
\put(66.6,22){\vector(0,-1){16}}
\put(173.6,22){\vector(0,-1){16}}
%\thicklines
%\put(15,110){\line(1,0){210}}
%\put(55,0){\line(1,0){130}}
%\put(15,40){\line(0,1){70}}
%\put(225,40){\line(0,1){70}}
%\qbezier(15,40)(15,0)(55,0)
%\qbezier(185,0)(225,0)(225,40)
%%%%%%%%%%%%%%%%%%%%%%%%%%%%%%%%%%%%%%%%%%%%%%%%%%
\end{picture}
\end{tabular}}}}
\end{center}
\caption{\texttt{FOLE} Anti-Join Flow Chart}
\label{fole:anti:join:flo:chrt}
\end{figure}
The anti-join operations are related to the semi-join operations.
The left (right) anti-join of two tables is the complement of the left (right) semi-join. 
%It returns rows from the first table, 
%where no matches occur
%on the common attributes of the two tables.
%
Let $\mathcal{A}$
% = {\langle{X,Y,\models_{\mathcal{A}}}\rangle}$
be a fixed type domain.
For any two $\mathcal{A}$-tables $\mathcal{T}_{1}$ and $\mathcal{T}_{2}$
that are linked through 
an $X$-sorted signature span 
{\footnotesize{$\mathcal{S}_{1}\xleftarrow{h_{1}} 
\mathcal{S}
\xrightarrow{h_{2}} \mathcal{S}_{2}$,}\normalsize}
the left anti-join $\mathcal{T}_{1}{\;\boxslash_{\mathcal{A}}\,}\mathcal{T}_{2}$
is the set of all tuples in $\mathcal{T}_{1}$ 
for which there is no tuple in $\mathcal{T}_{2}$ 
that is equal on their common attribute names. 
We use the following routes of flow.
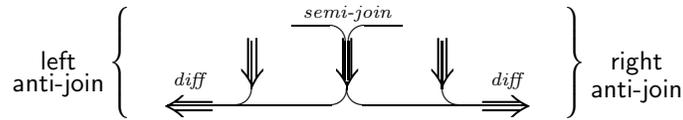
\begin{center}
{{\begin{tabular}{c}
\setlength{\unitlength}{0.6pt}
\begin{picture}(320,80)(0,-10)
\put(325,25){\makebox(0,0){\normalsize{
$\left.\rule{0pt}{24pt}\right\}
\underset{\textstyle{\textsf{anti-join}}}{\textsf{right}}$}}}
\put(-20,25){\makebox(0,0){\normalsize{
$\underset{\textstyle{\textsf{anti-join}}}{\textsf{left}}
\left\{\rule{0pt}{24pt}\right.$}}}
\put(90,30){\makebox(0,0){\huge{
$\overset{\textit{\scriptsize{}}}{\Downarrow}$}}}
\put(150,55){\makebox(0,0){\huge{
$\overset{\textit{\scriptsize{semi-join}}}{}$}}}
\put(150,30){\makebox(0,0){\huge{
$\overset{\textit{\scriptsize{}}}{\Downarrow}$}}}
\put(210,30){\makebox(0,0){\huge{
$\overset{\textit{\scriptsize{}}}{\Downarrow}$}}}
\put(50,5){\makebox(0,0){\huge{
$\overset{\textit{\scriptsize{diff}}}{\Leftarrow}$}}}
\put(250,5){\makebox(0,0){\huge{
$\overset{\textit{\scriptsize{diff}}}{\Rightarrow}$}}}
\put(95,8){\line(0,1){30}}
\put(85,8){\oval(20,20)[br]}
\put(215,8){\line(0,1){30}}
\put(225,8){\oval(20,20)[bl]}
%\put(82,48){\vector(1,0){0}}
\put(120,48){\line(1,0){24}}
\put(144,38){\oval(20,20)[tr]}
\put(190,48){\line(-1,0){24}}
\put(166,38){\oval(20,20)[tl]}
\put(154,14){\line(0,1){24}}
%\put(155,12){\line(0,-1){8}}
\qbezier(155,12)(155,10)(155,8)
\put(156,14){\line(0,1){24}}
\put(145,8){\oval(20,20)[br]}
\put(144,-2){\line(-1,0){102}}
\put(165,8){\oval(20,20)[bl]}
\put(166,-2){\line(1,0){102}}
\end{picture}
\end{tabular}}}
\end{center}
Left anti-join
within the context $\mathrmbf{Tbl}(\mathcal{A})$
is left semi-join, followed by difference.
This is the two-step process
illustrated in 
Fig.\;\ref{fole:anti:join:flo:chrt}.
The constraint, construction and input for anti-join 
are identical to that for natural join.
Only the output is different.
\begin{description}
\item[Constraint:] 
The constraint for anti-join is the same as the constraint for natural join
(Tbl.\,\ref{tbl:fole:natural:join:input:output}):
an $X$-sorted signature span 
$\mathcal{S}_{1}\xleftarrow{h_{1}}\mathcal{S}\xrightarrow{h_{2}}\mathcal{S}_{2}$
in $\mathrmbf{List}(X)$
\newline
\item[Construction:] 
The construction for anti-join is the same as the construction for natural join 
(Tbl.\,\ref{tbl:fole:natural:join:input:output}):
the opspan
{\footnotesize{$\mathcal{S}_{1}\xrightarrow{\iota_{1}\,} 
{\mathcal{S}_{1}{\!+_{\mathcal{S}}}\mathcal{S}_{2}}
\xleftarrow{\;\iota_{2}}\mathcal{S}_{2}$}\normalsize}
of injection $X$-signature morphisms
with pushout signature
$\mathcal{S}_{1}{+_{\mathcal{S}}}\mathcal{S}_{2}$.
\newline
\item[Input:] 
The input for anti-join is the same as the input for natural join 
(Tbl.\,\ref{tbl:fole:natural:join:input:output}):
a pair of tables
$\mathcal{T}_{1} = {\langle{K_{1},t_{1}}\rangle} \in 
\mathrmbf{Tbl}_{\mathcal{A}}(\mathcal{S}_{1})$
and
$\mathcal{T}_{2} = {\langle{K_{2},t_{2}}\rangle} \in 
\mathrmbf{Tbl}_{\mathcal{A}}(\mathcal{S}_{2})$.
\newline
\item[Output:] 
The output is semi-join followed by difference.
\newline
\begin{itemize}
\item 
Left semi-join 
%%is defined by the two-step process of 
%natural join followed by projection.
%This 
results in the table
${\mathcal{T}_{1}}{\,\boxleft_{\mathcal{A}}}{\mathcal{T}_{2}} =
\acute{\mathrmbfit{tbl}}_{\mathcal{A}}(\iota_{1})(
\mathcal{T}_{1}{\,\boxtimes_{\mathcal{A}}}\mathcal{T}_{2})
= {\langle{\hat{K}_{12},\acute{t}_{1}}\rangle}$ 
with key set $\widehat{K}_{12}$ and
tuple function
$\mathrmbfit{tup}_{\mathcal{A}}(\mathcal{S}_{1})
\xleftarrow{\acute{t}_{1}}\widehat{K}_{12}$.
%equal to the composite 
%$\mathrmbfit{tup}_{\mathcal{A}}(\mathcal{S}_{1})
%\xleftarrow{\mathrmbfit{tup}_{\mathcal{A}}(\iota_{1})}
%\mathrmbfit{tup}_{\mathcal{A}}({\mathcal{S}_{1}{+_{\mathcal{S}}}\mathcal{S}_{2}})
%\xleftarrow{(\hat{t}_{1},\hat{t}_{2})}\widehat{K}_{12}$.
\newline
\item 
Difference in the small table fiber context
$\mathrmbf{Tbl}_{\mathcal{A}}(\mathcal{S}_{1})$
gives the left anti-join table
${\mathcal{T}_{1}}{\,\boxslash_{\mathcal{A}}}{\mathcal{T}_{2}} 
= \mathcal{T}_{1}{\,-\,}({\mathcal{T}_{1}}{\,\boxleft_{\mathcal{A}}}{\mathcal{T}_{2}})
%\in \mathrmbf{Tbl}_{\mathcal{A}}(\mathcal{S}_{1})
$.
\end{itemize}
%
%The anti-join flowchart input/output is displayed in 
%Tbl.\,\ref{tbl:fole:natural:join:input:output}.
%\begin{description}
%\item[semi-join:] 
%
Anti-join is semi-join, followed by difference.
%For left anti-join,
For left anti-join,
this is the two-step process
\newline\mbox{}\hfill\rule[-10pt]{0pt}{26pt}
${\mathcal{T}_{1}}{\,\boxslash_{\mathcal{A}}}{\mathcal{T}_{2}} 
\doteq
\mathcal{T}_{1}{\,-\,}({\mathcal{T}_{1}}{\,\boxleft_{\mathcal{A}}}{\mathcal{T}_{2}})
%\in \mathrmbf{Tbl}_{\mathcal{A}}(\mathcal{S}_{1})
$.
\hfill\mbox{}\newline
There is an inclusion morphism
%\newline\mbox{}\hfill
$\mathcal{T}_{1}
\xhookleftarrow{\bar{\omega}_{1}}
{\mathcal{T}_{1}}{\,\boxslash_{\mathcal{A}}}{\mathcal{T}_{2}} $
in the small fiber table context $\mathrmbf{Tbl}_{\mathcal{A}}(\mathcal{S}_{1})$, 
which is the output for left anti-join.
The right anti-join has a similar definition
(Tbl.\,\ref{tbl:fole:natural:join:input:output}). 
%
%%%%%%%%%%%%%%%%%%%%%%%%%%%%%%%%%%%%%%%%%%%%%%%%%%%%%%%%%%%%%%%%%%%%%%
%%%%%%%%%%%%%%%%%%%%%%%%%%%%%%%%%%%%%%%%%%%%%%%%%%%%%%%%%%%%%%%%%%%%%%
{\footnote{The anti-join of a Cartesian product of non-empty tables is empty.}}
%%%%%%%%%%%%%%%%%%%%%%%%%%%%%%%%%%%%%%%%%%%%%%%%%%%%%%%%%%%%%%%%%%%%%%
%%%%%%%%%%%%%%%%%%%%%%%%%%%%%%%%%%%%%%%%%%%%%%%%%%%%%%%%%%%%%%%%%%%%%%
%
\end{description}
%
%The left anti-join $\mathcal{T}_{1}{\,\unrhd_{\mathcal{S}}}\mathcal{T}_{2}$
%is the difference between the table $\mathcal{T}_{1}$ and the semi-join
%$\mathcal{T}_{1}{\,\rhd_{\mathcal{S}}}\mathcal{T}_{2}$
%\end{description}
%\end{description}
%Hence,
%left anti-join has the following expression
%
%\newline\mbox{}\hfill\rule[-10pt]{0pt}{26pt}
%${\mathcal{T}_{1}}{\,\unrhd_{\mathcal{S}}}{\mathcal{T}_{2}} 
%\doteq
%\mathcal{T}_{1}{\,-\,}({\mathcal{T}_{1}}{\,\rhd_{\mathcal{S}}}{\mathcal{T}_{2}})
%\in \mathrmbf{Tbl}_{\mathcal{A}}(\mathcal{S}_{1})$.
%
%\hfill\mbox{}\newline
%\rule{330pt}{1pt}
%\newline

%
\begin{aside}\label{aside:anti:join}
%
%%%%%%%%%%%%%%%%%%%%%%%%%%%%%%%%%%%%%%%%%%%%%%%%%%%%%%%%%%%%%%%%%%%%%%
%%%%%%%%%%%%%%%%%%%%%%%%%%%%%%%%%%%%%%%%%%%%%%%%%%%%%%%%%%%%%%%%%%%%%%
%\footnote{Compare with the data-type anti-meet 
%(see the Aside\;\ref{aside:boole:anti:meet}).}
%%%%%%%%%%%%%%%%%%%%%%%%%%%%%%%%%%%%%%%%%%%%%%%%%%%%%%%%%%%%%%%%%%%%%%
%%%%%%%%%%%%%%%%%%%%%%%%%%%%%%%%%%%%%%%%%%%%%%%%%%%%%%%%%%%%%%%%%%%%%%
%
There is an alternate method for computing the anti-join.
For the left anti-join,
project $\mathcal{T}_{1}$ and $\mathcal{T}_{2}$
to the 
$X$-sorted signature $\mathcal{S}$
using 
the $X$-sorted 
signature span 
%$\mathcal{S}_{1}\xleftarrow{h_{1}}\mathcal{S}\xrightarrow{h_{2}}\mathcal{S}_{2}$
%an $X$-sorted signature span 
$\mathcal{S}_{1}\xhookleftarrow{h_{1}}\mathcal{S}\xhookrightarrow{h_{2}}\mathcal{S}_{2}$
%in $\mathrmbf{List}(X)$
%consisting of 
%a span of injective index functions
%$I_{1}\xhookleftarrow{h_{1}}I\xhookrightarrow{h_{2}}I_{2}$.
%
%%%%%%%%%%%%%%%%%%%%%%%%%%%%%%%%%%%%%%%%%%%%%%%%%%%%%%%%%%%%%%%%%%%%%%
%%%%%%%%%%%%%%%%%%%%%%%%%%%%%%%%%%%%%%%%%%%%%%%%%%%%%%%%%%%%%%%%%%%%%%
\footnote{We assume the index functions are injective
$I_{1}\xhookleftarrow{h_{1}}I\xhookrightarrow{h_{2}}I_{2}$.}
%Otherwise,
%we will need to add an inflation step before forming the Cartesian product.}
%%%%%%%%%%%%%%%%%%%%%%%%%%%%%%%%%%%%%%%%%%%%%%%%%%%%%%%%%%%%%%%%%%%%%%
%%%%%%%%%%%%%%%%%%%%%%%%%%%%%%%%%%%%%%%%%%%%%%%%%%%%%%%%%%%%%%%%%%%%%%
%
%in $\mathrmbf{List}(X)$
%\end{itemize}
%
getting $\mathcal{A}$-tables
$\acute{\mathrmbfit{tbl}}_{\mathcal{A}}(h_{1})(\mathcal{T}_{1}) 
= \widehat{\mathcal{T}}_{1}$ 
and 
$\acute{\mathrmbfit{tbl}}_{\mathcal{A}}(h_{2})(\mathcal{T}_{2}) 
= \widehat{\mathcal{T}}_{2}$.
Form the difference
$\bar{\mathcal{T}}_{12} =
\widehat{\mathcal{T}}_{1}{\;-\;}\widehat{\mathcal{T}}_{2}$.
Then,
form the natural join
$\mathcal{T}_{1}{\;\boxtimes}_{\mathcal{A}}\bar{\mathcal{T}}_{12}$.
\end{aside}
\begin{proposition}\label{anti:join}
The left anti-join is
$\mathcal{T}_{1}{\,\boxslash_{\mathcal{A}}}\mathcal{T}_{2}
%{\;\times\;}\mathcal{T}_{\bullet}
\cong 
\mathcal{T}_{1}{\;\boxtimes}_{\mathcal{A}}\bar{\mathcal{T}}_{12}$.
\end{proposition}
\begin{proof}
%
%%%%%%%%%%%%%%%%%%%%%%%%%%%%%%%%%%%%%%%%%%%%%%%%%%%%%%%%%%%%%%%%%%%%%%
%%%%%%%%%%%%%%%%%%%%%%%%%%%%%%%%%%%%%%%%%%%%%%%%%%%%%%%%%%%%%%%%%%%%%%
\footnote{\label{footnote:refl}
This argument is in terms of the underlying relations in the reflection of
Prop.\;\ref{tbl:rel:refl}.}
%%%%%%%%%%%%%%%%%%%%%%%%%%%%%%%%%%%%%%%%%%%%%%%%%%%%%%%%%%%%%%%%%%%%%%
%%%%%%%%%%%%%%%%%%%%%%%%%%%%%%%%%%%%%%%%%%%%%%%%%%%%%%%%%%%%%%%%%%%%%%
%
Consider a tuple
$t = (\acute{t}_{1},\widehat{t}_{12}) \in \mathrmbfit{tup}_{\mathcal{A}}(\mathcal{S}_{1})$,
\newline
with projection sub-tuple 
$\widehat{t}_{12} = \mathrmbfit{tup}_{\mathcal{A}}(h_{1})(t) 
\in \mathrmbfit{tup}_{\mathcal{A}}(\mathcal{S})$.
Then,
\newline
%A tuple 
%On the one hand,
$t \in \mathcal{T}_{1}{\,\boxslash_{\mathcal{A}}}\mathcal{T}_{2}
\text{\;iff\;}
%\newline
t \in \mathcal{T}_{1},\;
t \not\in {\mathcal{T}_{1}}{\,\boxleft_{\mathcal{A}}}{\mathcal{T}_{2}}
%$.
%%%%%%%%%%%%%%%%%%%%%%%%%%%%%%%%%%%%%%%%%%%%%%%%%%%
%\newline\rule{50pt}{1pt}\newline
%%%%%%%%%%%%%%%%%%%%%%%%%%%%%%%%%%%%%%%%%%%%%%%%%%%
%On the other hand,$
\text{\;iff\;}
t \in \mathcal{T}_{1},\;
\widehat{t}_{12} \in \bar{\mathcal{T}}_{12}
\text{\;iff\;}
t \in \mathcal{T}_{1}{\,\boxtimes}_{\mathcal{A}}\bar{\mathcal{T}}_{12}
%\newline
$.
\mbox{}\hfill\rule{5pt}{5pt}
\end{proof}
%

%%%%%%%%%%%%%%%%%%%%%%%%%%%%%%%%%%%%%%%%%%%%%%%%%%%%%%%%%%%%%%%%%%%%%%
%%%%%%%%%%%%%%%%%%%%%%%%%%%%%%%%%%%%%%%%%%%%%%%%%%%%%%%%%%%%%%%%%%%%%%
%\newpage
%\subsection{$\bigstar$ Extending Flowcharts.}\label{sub:sec:extend}
%Varying the Base \& 
%%%%%%%%%%%%%%%%%%%%%%%%%%%%%%%%%%%%%%%%%%%%%%%%%%%%%%%%%%%%%%%%%%%%%%
%%%%%%%%%%%%%%%%%%%%%%%%%%%%%%%%%%%%%%%%%%%%%%%%%%%%%%%%%%%%%%%%%%%%%%

%%%%%%%%%%%%%%%%%%%%%%%%%%%%%%%%%%%%%%%%%%%%%%%%%%%%%%%%%%%%%%
%
\newpage
\subsection{Generic Meet.}\label{sub:sub:sec:generic:meet}
%%%%%%%%%%%%%%%%%%%%%%%%%%%%%%%%%%%%%%%%%%%%%%%%%%%%%%%%%%%%%%
%
%
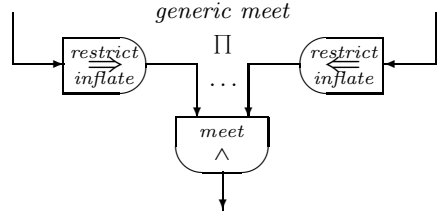
\begin{figure}
\begin{center}
{{{\begin{tabular}{c}
\begin{picture}(160,75)(37,27)
\setlength{\unitlength}{0.97pt}
%%%%%%%%%%%%%%%%%%%%%%%%%%%%%%%%%%%%%%%%%%%%%%%%%%
%\put(44,62){\begin{picture}(0,0)(0,0)
%\setlength{\unitlength}{0.46pt}
\put(54,65){\begin{picture}(0,0)(0,0)
\setlength{\unitlength}{0.35pt}
%\thicklines
%\put(106,40){\makebox(0,0){\normalsize{$\boldsymbol{\circ}$}}}
%\put(4.7,40){\makebox(0,0){\normalsize{$\boldsymbol{\circ}$}}}
\put(10,10){\line(1,0){60}}
\put(10,70){\line(1,0){60}}
\put(10,70){\line(0,-1){60}}
\put(70,40){\oval(60,60)[br]}
\put(70,40){\oval(60,60)[tr]}
\put(55,55){\makebox(0,0){\scriptsize{{\textit{{restrict}}}}}}
\put(56,40){\makebox(0,0){\Large{${\Rightarrow}$}}}
\put(55,25){\makebox(0,0){\scriptsize{{\textit{{inflate}}}}}}
\end{picture}}
%%%%%%%%%%%%%%%%%%%%%%%%%%%%%%%%%%%%%%%%%%%%%%%%%%
\put(146.5,65){\begin{picture}(0,0)(0,0)
\setlength{\unitlength}{0.35pt}
%\thicklines
%\put(106,40){\makebox(0,0){\normalsize{$\boldsymbol{\circ}$}}}
%\put(4.7,40){\makebox(0,0){\normalsize{$\boldsymbol{\circ}$}}}
\put(40,10){\line(1,0){60}}
\put(40,70){\line(1,0){60}}
\put(100,70){\line(0,-1){60}}
\put(40,40){\oval(60,60)[bl]}
\put(40,40){\oval(60,60)[tl]}
\put(58,55){\makebox(0,0){\scriptsize{{\textit{{restrict}}}}}}
\put(56,40){\makebox(0,0){\Large{${\Leftarrow}$}}}
\put(58,25){\makebox(0,0){\scriptsize{{\textit{{inflate}}}}}}
\end{picture}}
%%%%%%%%%%%%%%%%%%%%%%%%%%%%%%%%%%%%%%%%%%%%%%%%%%
\put(98,37){\begin{picture}(0,0)(0,3)
\setlength{\unitlength}{0.35pt}
\put(60,30){\makebox(0,0){\normalsize{$\wedge$}}}
%\thicklines
\put(40,10){\line(1,0){40}}
\put(10,70){\line(1,0){100}}
\put(10,70){\line(0,-1){30}}
\put(110,70){\line(0,-1){30}}
\put(40,40){\oval(60,60)[bl]}
\put(80,40){\oval(60,60)[br]}
\put(60,55){\makebox(0,0){\scriptsize{{\textit{{meet}}}}}}
\end{picture}}
%%%%%%%%%%%%%%%%%%%%%%%%%%%%%%%%%%%%%%%%%%%%%%%%%%
\put(120,100){\makebox(0,0){\footnotesize{{\textit{{generic meet}}}}}}
\put(120,88){\makebox(0,0){\scriptsize{$\prod$}}}
\put(121,72){\makebox(0,0){\ldots}}
%%%%%%%%%%%%%%%%%%%%%%%%%%%%%%%%%%%%%%%%%%%%%%%%%%
\put(38,80){\line(0,1){20}}
\put(38,80){\vector(1,0){20}}
\put(110,80){\line(-1,0){20}}
\put(110,80){\vector(0,-1){21}}
\put(120,38){\vector(0,-1){15}}
\put(130,80){\vector(0,-1){21}}
\put(130,80){\line(1,0){20}}
\put(203,80){\line(0,1){20}}
\put(203,80){\vector(-1,0){20}}
%\thicklines
%\put(15,110){\line(1,0){210}}
%\put(55,10){\line(1,0){130}}
%\put(15,50){\line(0,1){60}}
%\put(225,50){\line(0,1){60}}
%\qbezier(15,50)(15,10)(55,10)
%\qbezier(185,10)(225,10)(225,50)
%%%%%%%%%%%%%%%%%%%%%%%%%%%%%%%%%%%%%%%%%%%%%%%%%%
\end{picture}
\end{tabular}}}}
\end{center}
\caption{\texttt{FOLE} Generic Meet Flow Chart}
\label{fig:fole:boolean:meet:flo:chrt}
\end{figure}
The generic meet for tables 
is the relational counterpart of
the logical conjunction for predicates.
Where 
the \emph{meet} operation (\S\,\ref{sub:sub:sec:boole}) 
is the analogue for logical conjunction 
at the small scope $\mathrmbf{Tbl}(\mathcal{D})$ of a signed domain table fiber,
and the \emph{quotient} operation (\S\,\ref{sub:sub:sec:quotient})
and
the \emph{natural join} operation (\S\,\ref{sub:sub:sec:nat:join})
are special cases of the analogue 
%for logical conjunction 
at the intermediate scope 
$\mathrmbf{Tbl}(\mathcal{A})$ of a type domain table fiber,
the \emph{generic meet} operation
%(\S\,\ref{sub:sub:sec:generic:meet})
is defined at the large scope $\mathrmbf{Tbl}$ of all tables.
We identify all of these concepts as limits at different scopes.
%We identify \texttt{FOLE} generic meets with all limits.
%

In this section, we focus on tables in the full context $\mathrmbf{Tbl}$ 
of all tables.
These limits are resolvable into restriction-inflations
followed by meet.
The generic meet operation is dual to the generic join operation
(\S\,\ref{sub:sub:sec:generic:join}).
%
%%%%%%%%%%%%%%%%%%%%%%%%%%%%%%%%%%%%%%%%%%%%%%%%%%%%%%%%%%%%%%%%%%%%%%
%%%%%%%%%%%%%%%%%%%%%%%%%%%%%%%%%%%%%%%%%%%%%%%%%%%%%%%%%%%%%%%%%%%%%%
\footnote{Generic joins and colimits in the context of $\mathrmbf{Tbl}$ 
(\S\,\ref{sub:sub:sec:generic:join})
can be constructed out of 
limits in the context of signed domains $\mathrmbf{Dom}$,
the table projection-expansion operation along signed domain morphisms, and 
joins (colimits) in small table fibers.}
%%%%%%%%%%%%%%%%%%%%%%%%%%%%%%%%%%%%%%%%%%%%%%%%%%%%%%%%%%%%%%%%%%%%%%
%%%%%%%%%%%%%%%%%%%%%%%%%%%%%%%%%%%%%%%%%%%%%%%%%%%%%%%%%%%%%%%%%%%%%%
%
%\newpage
%We focus on tables in the context $\mathrmbf{Tbl}(\mathcal{A})$ 
%for fixed type domain $\mathcal{A}$.
%
%\newpage
%For the general case,
%we will use the terminology \emph{generic meet} when we use a sufficient
%
The \emph{generic meet} operation only needs a sufficient
collection of tables
(Def.\;\ref{def:suff:adequ:lim}).
To reiterate,
we identify \texttt{FOLE} generic meets with all limits
in the context $\mathrmbf{Tbl}$.
\begin{description}
\item[Constraint:] 
Consider a diagram 
$\mathrmbfit{D} 
%= \mathrmbfit{T}^{\text{op}}{\,\circ\,}\mathrmbfit{dom} 
: \mathrmbf{I}^{\text{op}} \rightarrow \mathrmbf{Dom}$
consisting of a linked collection of signed domains 
$\bigl\{ 
\mathcal{D}_{i} 
\xrightarrow{{\langle{{h},{f},{g}}\rangle}} 
\mathcal{D}_{j} 
%\mid i \in \mathrmbf{I} 
\bigr\}$.
%$\bigl\{ 
%\mathcal{D}_{i} 
%= \mathrmbfit{D}(i)
%= \mathrmbfit{dom}(\mathrmbfit{T}(i))
%\mid i \in \mathrmbf{I} \bigr\}$.
This is the constraint for generic meet 
(Tbl.\,\ref{tbl:fole:generic:meet:input:output}).
\newline
\item[Construction:] 
Let $\widehat{\mathcal{D}} = \coprod \mathrmbfit{D}$ be the colimit
% of 
%$\mathrmbfit{D} = \mathrmbfit{T}^{\text{op}}{\,\circ\,}\mathrmbfit{dom} :
%\mathrmbf{I}^{\text{op}} \rightarrow 
in $\mathrmbf{Dom}$
with injection signed domain morphisms 
$\bigl\{ 
\mathcal{D}_{i} 
\xrightarrow{{\langle{\hat{h}_{i},\hat{f}_{i},\hat{g}_{i}}\rangle}} 
\widehat{\mathcal{D}} 
\mid i \in \mathrmbf{I} \bigr\}$
that commute with the links in the constraint:
${\langle{\hat{h}_{i},\hat{f}_{i},\hat{g}_{i}}\rangle}
= {\langle{{h},{f},{g}}\rangle}{\;\circ\;}
{\langle{\hat{h}_{j},\hat{f}_{j},\hat{g}_{j}}\rangle}$.  
This is the construction for generic meet 
(Tbl.\,\ref{tbl:fole:generic:meet:input:output}).
\newline
\item[Input:] 
Let $I \xrightarrow{\mathrmbfit{T}} \mathrmbf{Tbl}$ 
be a \underline{sufficient} indexed collection of tables 
(Def.\;\ref{def:suff:adequ:lim})
{\footnotesize{$
\bigl\{ 
\mathcal{T}_{i} = \mathrmbfit{T}(i) \in \mathrmbf{Tbl}(\mathcal{D}_{i}) 
\mid i \in I \bigr\}$}}
for some indexing set
$I \subseteq obj(\mathrmbf{I})$.
%
%Let
%$I \xrightarrow{\mathrmbfit{D}} \mathrmbf{Dom}$
%be its signed domain projection
%consisting of an (unlinked) indexed collection of signed domains 
%$\bigl\{ 
%\mathcal{D}_{i} = \mathrmbfit{D}(i)
%\mid i \in I \bigr\}$.
%
This is the input for generic meet 
(Tbl.\,\ref{tbl:fole:generic:meet:input:output}).
\newpage
\item[Output:]
%\mbox{}
Generic meet is restriction/inflation ($i \in I$ times) followed by meet.
\newline
\begin{itemize}
\item 
For each index $i \in I$,
restriction/inflation
{\footnotesize{$\mathrmbf{Tbl}(\mathcal{D}_{i}) 
\xrightarrow{\grave{\mathrmbfit{tbl}}({\hat{h}_{i},\hat{f}_{i},\hat{g}_{i}})} 
\mathrmbf{Tbl}(\widehat{\mathcal{D}})$}}
(\S\,\ref{sub:sub:sec:flow:sign:dom:mor})
along the tuple function
of the signed domain morphism 
$\mathcal{D}_{i} 
\xrightarrow{{\langle{\hat{h}_{i},\hat{f}_{i},\hat{g}_{i}}\rangle}} 
\widehat{\mathcal{D}} = 
\coprod \mathrmbfit{D}$
maps the table 
$\mathcal{T}_{i} 
%= {\langle{K_{i},t_{i}}\rangle} 
\in \mathrmbf{Tbl}(\mathcal{D}_{i})$
%with its tuple function 
%$K_{i} \xrightarrow{t_{i}} \mathrmbfit{tup}(\mathcal{D}_{i})$
to the table
$\widehat{\mathcal{T}}_{i}
%\widehat{\mathrmbfit{T}}(i)
%= \grave{\mathrmbfit{tbl}}({h,f,g})(\mathrmbfit{T}(i))
= {\langle{\widehat{K}_{i},\hat{t}_{i}}\rangle} 
\in \mathrmbf{Tbl}({\widehat{\mathcal{D}}})$
with its tuple function
$\widehat{K}_{i} \xrightarrow{\hat{t}_{i}} 
\mathrmbfit{tup}(\widehat{\mathcal{D}})$
defined by pullback, 
$\grave{k}_{i}{\,\cdot\,}t_{i} 
= \hat{t}_{i}{\,\cdot\,}{\mathrmbfit{tup}(\hat{h}_{i},\hat{f}_{i},\hat{g}_{i})}$. 
Here we have ``vertically restricted'' and then ``horizontally inflated'' 
tuples in $\mathrmbfit{tup}(\mathcal{D}_{i}) \subseteq \mathrmbf{List}(Y_{i})$ 
by pullback 
along the tuple function
%\[\mbox
{\footnotesize{$
\mathrmbfit{tup}(\mathcal{D}_{i})
\xleftarrow
%[(h{\cdot}{(\mbox{-})})\cdot({(\mbox{-})}{\cdot}g)]
{\mathrmbfit{tup}(\hat{h}_{i},\hat{f}_{i},\hat{g}_{i})}
\mathrmbfit{tup}(\widehat{\mathcal{D}})
$}\normalsize}
(see RHS Fig.\;\ref{fig:flow:sign:dom}).
This is linked 
%(Fig.\;\ref{fig:fole:boolean:meet})
to the table $\mathcal{T}_{i}$ 
by the table morphism 
%\newline
%\[\mbox
{\footnotesize{{$
\mathcal{T}_{i} = {\langle{K_{i},t_{i}}\rangle}
\xleftarrow{{\langle{{\langle{\hat{h}_{i},\hat{f}_{i},\hat{g}_{i}}\rangle},\grave{k}_{i}}\rangle}}
{\langle{\widehat{K}_{i},\hat{t}_{i}}\rangle} 
= \widehat{\mathcal{T}}_{i}$.}}\normalsize}
%\]
%
\newline
\item
Intersection (\S\,\ref{sub:sub:sec:boole})
of the tables $\{ \widehat{\mathcal{T}}_{i} \mid i \in I \}$
in the fiber context $\mathrmbf{Tbl}(\widehat{\mathcal{D}})$
defines the 
%natural join 
%limit 
generic meet
$\prod\mathrmbfit{T}
= \widehat{\mathrmbfit{T}} 
= \bigwedge \bigl\{\widehat{\mathcal{T}}_{i} \mid i \in I \bigr\}
= {\langle{\widehat{K},\hat{t}}\rangle}$,
whose key set $\widehat{K}$ is the pullback and 
whose tuple map is the mediating function
$\widehat{K}\xrightarrow{{(\hat{t}_{i})}}
\mathrmbfit{tup}(\widehat{\mathcal{D}})$
of the multi-opspan
$\bigl\{ \widehat{K}_{i}\xrightarrow{\hat{t}_{i}}
\mathrmbfit{tup}(\widehat{\mathcal{D}}) \mid i \in I \bigr\}$,
resulting in the discrete multi-span (cone)
{\footnotesize{{$\Bigl\{ 
\widehat{\mathcal{T}}_{i} 
%= {\langle{\widehat{K}_{i},\hat{t}_{i}}\rangle} 
\xleftarrow{\;\hat{\pi}_{i}\;} 
%{\langle{\widehat{K},\hat{t}}\rangle} = 
\widehat{\mathrmbfit{T}} 
\mid i \in I \Bigr\}$.}}\normalsize}
\end{itemize}
\mbox{}\newline
Restriction-inflation composed with meet 
defines the multi-span of table morphisms
\[\mbox
{\footnotesize{{$\Bigl\{ 
\mathcal{T}_{i}
\xleftarrow
[\;\hat{\pi}_{i}{\circ\,}{\langle{{\langle{\hat{h}_{i},\hat{f}_{i},\hat{g}_{i}}\rangle},\grave{k}_{i}}\rangle}\;]
{\;{\langle{{\langle{\hat{h}_{i},\hat{f}_{i},\hat{g}_{i}}\rangle},\hat{k}_{i}}\rangle}\;} 
\prod\mathrmbfit{T} = \widehat{\mathrmbfit{T}}
\mid i \in I \Bigr\}$,}}\normalsize}\]
illustrated in 
Fig.\;\ref{fig:fole:generic:meet},
which is the output for generic meet 
(Tbl.\,\ref{tbl:fole:generic:meet:input:output}).
\end{description}
%
%The generic meet flowchart input/output is displayed in
%Tbl.\,\ref{tbl:fole:generic:meet:input:output}.
%where quotient is distinguish from mere inflation by having a constraint.
Generic meet is restriction/inflation ($i \in I$ times) followed by meet.
This is the two-step process
%\newline\mbox{}\hfill\rule[-10pt]{0pt}{26pt}
%
\[\mbox{\footnotesize{{$
\prod\mathrmbfit{T} 
\doteq
\bigwedge\;
\Bigl\{  
\grave{\mathrmbfit{tbl}}({\hat{h}_{i},\hat{f}_{i},\hat{g}_{i}})(\mathcal{T}_{i}) 
\mid i \in I \Bigr\}$.}}\normalsize}\]
%\hfill\mbox{}\newline
%
\begin{aside}
Theoretically
this would represent the limit
%(see the application discussion for completeness in \S\,\ref{sub:sub:sec:lim:colim:tbl})
of a diagram
$\mathrmbf{I}\xrightarrow{\mathrmbfit{T}}\mathrmbf{Tbl}$ 
consisting of a linked collection of tables.
But practically,
we are only given the constraint (a diagram)
$\mathrmbf{I}^{\text{op}}\xrightarrow{\mathrmbfit{D}}\mathrmbf{Dom}$ 
consisting of a linked collection of signed domains 
$\bigl\{ 
\mathcal{D}_{i} = \mathrmbfit{D}(i)
\mid i \in \mathrmbf{I} \bigr\}$
and the input 
$I\xrightarrow{\mathrmbfit{T}}\mathrmbf{Tbl}$ 
consisting of a \underline{sufficient} indexed collection of tables 
(Def.\;\ref{def:suff:adequ:lim})
{\footnotesize{$
\bigl\{ 
\mathcal{T}_{i} =
\mathrmbfit{T}(i) \in \mathrmbf{Tbl}(\mathcal{D}_{i}) 
\mid i \in I \subseteq obj(\mathrmbf{I}) \bigr\}$}}.
\end{aside}
%\newpage

%
\begin{table}
\begin{center}
{{\fbox{\begin{tabular}{c}
\setlength{\extrarowheight}{2pt}
{\scriptsize{$\begin{array}[c]{c@{\hspace{12pt}}l}
\mathcal{D}_{i}
\xrightarrow{{\langle{h,f,g}\rangle}}
\mathcal{D}_{j}
&
\textit{constraint}
\\
\mathcal{D}_{i}
\xrightarrow{{\langle{\hat{h}_{i},\hat{f}_{i},\hat{g}_{i}}\rangle}}
\coprod \mathrmbfit{D}
%\widehat{\mathcal{D}} 
&
\textit{construction}
\\\hline
%\mathcal{T}_{i}
%\xleftarrow{{\langle{{\langle{h,f,g}\rangle},k}\rangle}}
%\mathcal{T}_{j}
\bigl\{ 
\mathcal{T}_{i} 
%=\mathrmbfit{T}(i) 
\in \mathrmbf{Tbl}(\mathcal{D}_{i}) 
\mid i \in I \bigr\}
&
\textit{input}
\\
\mathcal{T}_{i}
%\mathrmbfit{T}(i)
\xleftarrow
{{\langle{{\langle{\hat{h}_{i},\hat{f}_{i},\hat{g}_{i}}\rangle},\hat{k}_{i}}\rangle}}
\prod \mathrmbfit{T}
&
\textit{output}
\end{array}$}}
\end{tabular}}}}
\end{center}
\caption{\texttt{FOLE} Generic Meet I/O}
\label{tbl:fole:generic:meet:input:output}
\end{table}
%
%\hfill\mbox{}
%\end{itemize}
%\end{description}
%\mbox{}\hfill\rule{5pt}{5pt}
%\end{proof}
%
%
\begin{figure}
\begin{center}
{{\begin{tabular}{c}
%@{\hspace{75pt}}c}
%%%%%%%%%%%%%%%%%%%%%%%%%%%%%%%%%%%%%%%%%%%%%%%%%%
{{\begin{tabular}{c}
\setlength{\unitlength}{0.56pt}
\begin{picture}(320,160)(0,-5)
\put(0,80){\makebox(0,0){\footnotesize{$K_{i}$}}}
\put(100,80){\makebox(0,0){\footnotesize{$\widehat{K}_{i}$}}}
\put(220,80){\makebox(0,0){\footnotesize{$\widehat{K}_{j}$}}}
\put(324,80){\makebox(0,0){\footnotesize{$K_{j}$}}}
\put(160,148){\makebox(0,0){\footnotesize{$\widehat{K}$}}}
\put(0,0){\makebox(0,0){\footnotesize{$
{\mathrmbfit{tup}(\mathcal{D}_{i})}$}}}
\put(320,0){\makebox(0,0){\footnotesize{$
{\mathrmbfit{tup}(\mathcal{D}_{j})}$}}}
\put(160,0){\makebox(0,0){\footnotesize{$
{\mathrmbfit{tup}(\widehat{\mathcal{D}})}$}}}
\put(55,90){\makebox(0,0){\scriptsize{$\grave{k}_{i}$}}}
\put(265,90){\makebox(0,0){\scriptsize{$\grave{k}_{j}$}}}
\put(80,-12){\makebox(0,0){\scriptsize{$
\mathrmbfit{tup}(\hat{h}_{i},\hat{f}_{i},\hat{g}_{i})$}}}
\put(240,-12){\makebox(0,0){\scriptsize{$
\mathrmbfit{tup}(\hat{h}_{j},\hat{f}_{j},\hat{g}_{j})$}}}
\put(-6,40){\makebox(0,0)[r]{\scriptsize{$t_{i}$}}}
\put(125,40){\makebox(0,0)[r]{\scriptsize{$\hat{t}_{i}$}}}
\put(125,116){\makebox(0,0)[r]{\scriptsize{$\hat{\pi}_{i}$}}}
\put(327,40){\makebox(0,0)[l]{\scriptsize{$t_{j}$}}}
\put(200,40){\makebox(0,0)[l]{\scriptsize{$\hat{t}_{j}$}}}
\put(195,115){\makebox(0,0)[l]{\scriptsize{$\hat{\pi}_{j}$}}}
\put(64,128){\makebox(0,0)[r]{\scriptsize{$\hat{k}_{i}$}}}
\put(260,128){\makebox(0,0)[l]{\scriptsize{$\hat{k}_{j}$}}}
\put(0,65){\vector(0,-1){50}}
\put(320,65){\vector(0,-1){50}}
\put(105,65){\vector(1,-1){50}}
\put(215,65){\vector(-1,-1){50}}
\put(150,135){\vector(-1,-1){45}}
\put(165,135){\vector(1,-1){45}}
\put(80,80){\vector(-1,0){60}}
\put(240,80){\vector(1,0){60}}
\put(120,0){\vector(-1,0){85}}
\put(200,0){\vector(1,0){85}}
\qbezier(20,96)(80,120)(140,144)\put(20,96){\vector(-2,-1){0}}
\qbezier(180,144)(240,120)(300,96)\put(300,96){\vector(2,-1){0}}
%\put(138,142){\vector(-3,-1){110}}
%\put(124,145){\vector(-2,-1){110}}
%\put(200,145){\vector(2,-1){110}}
%
\qbezier(40,30)(30,30)(20,30)
\qbezier(40,30)(40,20)(40,10)
\qbezier(146,36)(153,43)(160,50)
\qbezier(160,50)(167,43)(174,36)
\qbezier(280,30)(290,30)(300,30)
\qbezier(280,30)(280,20)(280,10)
%%%%%%%%%%
\put(60,45){\makebox(0,0){\huge{
$\overset{\textit{\scriptsize{inflate}}}{\Rightarrow}$}}}
\put(160,110){\makebox(0,0){\ldots}}
\put(160,80){\makebox(0,0){${\scriptsize{meet}}$}}
\put(250,45){\makebox(0,0){\huge{
$\overset{\textit{\scriptsize{inflate}}}{\Leftarrow}$}}}
%%%%%%%%%%
\end{picture}
\end{tabular}}}
%
%%%%%%%%%%%%%%%%%%%%%%%%%%%%%%%%%%%%%%%%%%%%%%%%%%
\\\\
%%%%%%%%%%%%%%%%%%%%%%%%%%%%%%%%%%%%%%%%%%%%%%%%%%
{{\begin{tabular}{c}
\setlength{\unitlength}{0.7pt}
\begin{picture}(200,70)(-100,60)
\put(-140,57){\makebox(0,0){\footnotesize{${\mathcal{T}_{i}}$}}}
\put(0,120){\makebox(0,0){\footnotesize{$\widehat{\mathrmbfit{T}}$}}}
\put(140,57){\makebox(0,0){\footnotesize{${\mathcal{T}_{j}}$}}}
\put(-63,57){\makebox(0,0){\footnotesize{$\widehat{\mathcal{T}}_{i}$}}}
\put(63,57){\makebox(0,0){\footnotesize{$\widehat{\mathrmbfit{T}}(j))$}}}
\put(-80,100){\makebox(0,0)[r]{\tiny{$
{\langle{{\langle{\hat{h}_{i},\hat{f}_{i},\hat{g}_{i}}\rangle},\hat{k}_{i}}\rangle}$}}}
\put(-105,45){\makebox(0,0){\tiny{$
{\langle{{\langle{\hat{h}_{i},\hat{f}_{i},\hat{g}_{i}}\rangle},\grave{k}_{i}}\rangle}$}}}
\put(82,100){\makebox(0,0)[l]{\tiny{$
{\langle{{\langle{\hat{h}_{j},\hat{f}_{j},\hat{g}_{j}}\rangle},\hat{k}_{j}}\rangle}$}}}
\put(105,45){\makebox(0,0){\tiny{$
{\langle{{\langle{\hat{h}_{i},\hat{f}_{i},\hat{g}_{i}}\rangle},\hat{k}_{i}}\rangle}$}}}
\put(-33,93){\makebox(0,0)[r]{\scriptsize{$\hat{\pi}_{i}$}}}
\put(35,93){\makebox(0,0)[l]{\scriptsize{$\hat{\pi}_{j}$}}}
%
%\put(-30,125){\vector(-2,-1){90}}
%\put(30,125){\vector(2,-1){94}}
\qbezier(-130,65)(-85,96)(-16,120)\put(-130,65){\vector(-2,-1){0}}
\qbezier(130,65)(85,96)(16,120)\put(130,65){\vector(2,-1){0}}
\put(-10,110){\vector(-1,-1){40}}
\put(10,110){\vector(1,-1){40}}
\put(-80,57){\vector(-1,0){50}}
\put(80,57){\vector(1,0){50}}
\put(0,80){\makebox(0,0){\ldots}}
\end{picture}
\end{tabular}}}
%%%%%%%%%%%%%%%%%%%%%%%%%%%%%%%%%%%%%%%%%%%%%%%%%%
\\
%%%%%%%%%%%%%%%%%%%%%%%%%%%%%%%%%%%%%%%%%%%%%%%%%%
\end{tabular}}}
\end{center}
\caption{\texttt{FOLE} Generic Meet}
\label{fig:fole:generic:meet}
\end{figure}
%

%%%%%%%%%%%%%%%%%%%%%%%%%%%%%%%%%%%%%%%%%%%%%%%%%%%%%%%%%%%%%%
%%%%%%%%%%%%%%%%%%%%%%%%%%%%%%%%%%%%%%%%%%%%%%%%%%%%%%%%%%%%%%
\newpage
\section{Composite Operations for Colimits.}
\label{sub:sec:comp:ops:sign}
%%%%%%%%%%%%%%%%%%%%%%%%%%%%%%%%%%%%%%%%%%%%%%%%%%%%%%%%%%%%%%
%%%%%%%%%%%%%%%%%%%%%%%%%%%%%%%%%%%%%%%%%%%%%%%%%%%%%%%%%%%%%%

To repeat,
% \S\;\ref{sub:sec:comp:ops:type:dom},
the basic components of
\S\,\ref{sub:sec:base:ops}
are the components to be used in flowcharts.
Composite operations are operations whose flowcharts 
are composed of one or more basic components.
In addition to its basic components,
a composite operation also has a constraint,
which is used to construct its output.
In this section
we define the the composite relational operations
(Tbl.\,\ref{tbl:fole:comp:rel:ops:colim})
related to colimits.
Each composite operation defined here
has a dual relational operation (Tbl.\,\ref{tbl:fole:comp:rel:ops:lim})
related to limits.
%
%%%%%%%%%%%%%%%%%%%%%%%%%%%%%%%%%%%%%%%%%%%%%%%%%%%%%%%%%%%%%%%%%%%%%%%%%%%%%%%%
%%%%%%%%%%%%%%%%%%%%%%%%%%%%%%%%%%%%%%%%%%%%%%%%%%%%%%%%%%%%%%%%%%%%%%%%%%%%%%%%
\footnote{
The co-quotient operation is dual to the quotient operation.
The co-core operation is dual to the core operation.
The data-type join operation is dual to the natural join operation.
The generic join operation is dual to the generic meet operation.}
%%%%%%%%%%%%%%%%%%%%%%%%%%%%%%%%%%%%%%%%%%%%%%%%%%%%%%%%%%%%%%%%%%%%%%%%%%%%%%%%
%%%%%%%%%%%%%%%%%%%%%%%%%%%%%%%%%%%%%%%%%%%%%%%%%%%%%%%%%%%%%%%%%%%%%%%%%%%%%%%%
%
For colimit operations
we need only a sufficient collection of tables linked by a complete collection of signatures
(see Def.\;\ref{def:suff:adequ:colim}).
Fig.\;\ref{fig:routes:flow:colim}
gives the possible routes of flow for colimits.

\begin{table}
\begin{center}
{{{\begin{tabular}{c}
\setlength{\extrarowheight}{2pt}
{\scriptsize{$\begin{array}[c]
{|@{\hspace{5pt}}r@{\hspace{10pt}}l@{\hspace{5pt}\in\hspace{4pt}}l@{\hspace{5pt}}|}
\hline
\textbf{co-quotient:}
&
{\Yleft}_{\mathcal{S}}(\mathcal{T})
= \acute{\mathrmbfit{tbl}}_{\mathcal{S}}(\tilde{g})(\mathcal{T})
&
\mathrmbf{Tbl}_{\mathcal{S}}(\widetilde{\mathcal{A}})
\\
\textbf{co-core:}
&
\mathcal{T}_{1}{\,{\cup}_{\mathcal{A}}}\mathcal{T}_{2} = 
\acute{\mathrmbfit{tbl}}_{\mathcal{A}}(h_{1})(\mathcal{T}_{1})
{\;\vee\;}
\acute{\mathrmbfit{tbl}}_{\mathcal{A}}(h_{2})(\mathcal{T}_{2})
&
\mathrmbf{Tbl}_{\mathcal{A}}(\mathcal{S})
\\
\textbf{data-type join:}
&
\mathcal{T}_{1}{\,\oplus_{\mathcal{S}}}\mathcal{T}_{2} 
= 
\acute{\mathrmbfit{tbl}}_{\mathcal{S}}(\tilde{g}_{1})(\mathcal{T}_{1})
{\;\vee\;}
\acute{\mathrmbfit{tbl}}_{\mathcal{S}}(\tilde{g}_{2})(\mathcal{T}_{2})
&
\mathrmbf{Tbl}_{\mathcal{S}}(\mathcal{A}_{1}{{\times}_{\mathcal{A}}}\mathcal{A}_{2})
\\\hline\hline
\textbf{data-type semi-join:}
&
{\mathcal{T}_{1}}{\,\oleft_{\mathcal{S}}}{\mathcal{T}_{2}} =
\grave{\mathrmbfit{tbl}}_{\mathcal{S}}(\tilde{g}_{1})
\bigl(\mathcal{T}_{1}{\,\oplus_{\mathcal{S}}}
\mathcal{T}_{2}\bigr)
&
\mathrmbf{Tbl}_{\mathcal{S}}(\mathcal{A}_{1})
\\
\textbf{data-type anti-join:}
&
{\mathcal{T}_{1}}{\,\oslash_{\mathcal{S}}}{\mathcal{T}_{2}} 
= \mathcal{T}_{1}{\,-\,}({\mathcal{T}_{1}}{\,\oleft_{\mathcal{S}}}{\mathcal{T}_{2}})
&
\mathrmbf{Tbl}_{\mathcal{S}}(\mathcal{A}_{1})
\\\hline\hline
\textbf{generic join:}
&
\coprod\mathrmbfit{T} = 
\bigvee \Bigl\{ 
\acute{\mathrmbfit{tbl}}({\tilde{h}_{i},\tilde{f}_{i},\tilde{g}_{i}})(\mathcal{T}_{i})
 \mid i \in I \Bigr\}
&
\mathrmbf{Tbl}(\widehat{\mathcal{D}})
\\\hline
\end{array}$}}
\end{tabular}}}}
\end{center}
\caption{\texttt{FOLE} Composite Relational Operations for Colimits}
\label{tbl:fole:comp:rel:ops:colim}
\end{table}
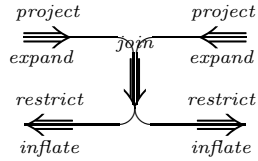
\begin{figure}
\begin{center}
{{\begin{tabular}{c}
\setlength{\unitlength}{0.65pt}
\begin{picture}(120,90)(100,-20)
\put(104,55){\makebox(0,0){\huge{$\overset{\textit{\scriptsize{project}}}{\Rightarrow}$}}}
\put(100,35.8){\makebox(0,0){\huge{${\textit{\scriptsize{expand}}}$}}}
\put(206,55){\makebox(0,0){\huge{$\overset{\textit{\scriptsize{project}}}{\Leftarrow}$}}}
\put(205,35.8){\makebox(0,0){\huge{${\textit{\scriptsize{expand}}}$}}}
%%%%%
\put(155,30){\makebox(0,0){\huge{$\overset{\textit{\scriptsize{join}}}{\Downarrow}$}}}
%%%%%
\put(105,4.1){\makebox(0,0){\huge{$\overset{\textit{\scriptsize{restrict}}}{\Leftarrow}$}}}
\put(105,-17){\makebox(0,0){\huge{${\textit{\scriptsize{inflate}}}$}}}
\put(205,4.1){\makebox(0,0){\huge{$\overset{\textit{\scriptsize{restrict}}}{\Rightarrow}$}}}
\put(205,-17){\makebox(0,0){\huge{${\textit{\scriptsize{inflate}}}$}}}
\put(90,48){\line(1,0){54}}
\put(144,38){\oval(20,20)[tr]}
\put(166,38){\oval(20,20)[tl]}
\put(220,48){\line(-1,0){54}}
\put(145,-3){\line(-1,0){54}}
\put(165,-3){\line(1,0){54}}
\put(154,9){\line(0,1){29}}
\put(156,9){\line(0,1){29}}
\put(145,7){\oval(20,20)[br]}
\put(165,7){\oval(20,20)[bl]}
\end{picture}
\end{tabular}}}
\end{center}
\caption{Routes of Flow: Colimits}
\label{fig:routes:flow:colim}
\end{figure}
%

%%%%%%%%%%%%%%%%%%%%%%%%%%%%%%%%%%%%%%%%%%%%%%%%%%%%%%%%%%%%%
%
\newpage
\subsection{Co-quotient.}
\label{sub:sub:sec:co-quotient}
%%%%%%%%%%%%%%%%%%%%%%%%%%%%%%%%%%%%%%%%%%%%%%%%%%%%%%%%%%%%%
%
%%%%%%%%%%%%%%%%%%%%%%%%%%%%%%%%%%%%%%%%%%%%%%%%%%%%%%%%%%%%%%
%%%%%%%%%%%%%%%%%%%%%%%%%%%%%%%%%%%%%%%%%%%%%%%%%%%%%%%%%%%%%%
\footnote{Basic components 
(\S\,\ref{sub:sec:base:ops})
are components to be used in flowcharts.
In particular, 
the co-quotient \underline{composite} operation 
of this section
has a flowchart with only one component --- expansion.
In addition to its one component,
it also has a constraint,
which is used to construct its output.}
%%%%%%%%%%%%%%%%%%%%%%%%%%%%%%%%%%%%%%%%%%%%%%%%%%%%%%%%%%%%%%
%%%%%%%%%%%%%%%%%%%%%%%%%%%%%%%%%%%%%%%%%%%%%%%%%%%%%%%%%%%%%%
%
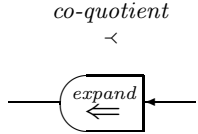
\begin{figure}
\begin{center}
{{{\begin{tabular}{c}
\begin{picture}(120,40)(55,45)
\setlength{\unitlength}{0.97pt}
%%%%%%%%%%%%%%%%%%%%%%%%%%%%%%%%%%%%%%%%%%%%%%%%%%
\put(96.5,35){\begin{picture}(0,0)(0,0)
\setlength{\unitlength}{0.35pt}
%\thicklines
%\put(106,40){\makebox(0,0){\normalsize{$\boldsymbol{\circ}$}}}
%\put(4.7,40){\makebox(0,0){\normalsize{$\boldsymbol{\circ}$}}}
\put(40,10){\line(1,0){60}}
\put(40,70){\line(1,0){60}}
\put(100,70){\line(0,-1){60}}
\put(40,40){\oval(60,60)[bl]}
\put(40,40){\oval(60,60)[tl]}
\put(58,50){\makebox(0,0){\scriptsize{{\textit{{expand}}}}}}
\put(56,30){\makebox(0,0){\Large{${\Leftarrow}$}}}
\end{picture}}
%%%%%%%%%%%%%%%%%%%%%%%%%%%%%%%%%%%%%%%%%%%%%%%%%%
\put(120,85){\makebox(0,0){\footnotesize{{\textit{{co-quotient}}}}}}
\put(120,75){\makebox(0,0){\footnotesize{$\Yleft$}}}
%[]\sim
%%%%%%%%%%%%%%%%%%%%%%%%%%%%%%%%%%%%%%%%%%%%%%%%%%
\put(80,50){\line(1,0){20}}
\put(153,50){\vector(-1,0){20}}
%%%%%%%%%%%%%%%%%%%%%%%%%%%%%%%%%%%%%%%%%%%%%%%%%%
\end{picture}
\end{tabular}}}}
\end{center}
\caption{\texttt{FOLE} Co-quotient Flow Chart}
\label{fig:fole:co-quotient:flo:chrt}
\end{figure}
%
%Let $\mathcal{S}$
% = {\langle{I,x,X}\rangle}
%be a fixed signature.
%Let 
%$\mathcal{A} = {\langle{X,Y,\models_{\mathcal{A}}}\rangle}$
%be a fixed type domain.
In this section, we focus on tables in the context $\mathrmbf{Tbl}(\mathcal{S})$ 
for fixed signature $\mathcal{S}$.
In this context,
generic joins 
--- for the special case of the co-equalizer of a parallel pair --- 
are called co-quotients.
Here,
we define an equivalence on data values.
In \S\,\ref{sub:sub:sec:quotient}
we discuss a dual notion;
there we define an equivalence on attributes, specifically on indexes.
The co-quotient operation defined here 
is dual to 
the quotient operation defined in \S\,\ref{sub:sub:sec:quotient}.
%
%\newline\rule{50pt}{1pt}\newline
%\newpage
%
\begin{description}
\item[Constraint:] 
Consider a parallel pair 
$g_{1},g_{2}:\mathcal{A}{\;\rightrightarrows}\;\mathcal{A}'$
of $X$-type domain morphisms 
in $\mathrmbf{Cls}(X)$.
This is the constraint for co-quotient 
(Tbl.\,\ref{tbl:fole:co-quotient:input:output}).
\newline
\item[Construction:] 
We can construct 
the equalizer of this constraint 
in $\mathrmbf{Cls}(X)$
with co-quotient $X$-type domain
$\mathcal{A}/\mathcal{I} 
= {\langle{X,\equiv_{R},\models_{\mathcal{A}/\mathcal{I}}}\rangle}$
and 
canonical quotient type domain morphism
$\mathcal{A}/\mathcal{I}\xrightarrow{\;[]_{R}\;}\mathcal{A}$
with 
%identity sort function $\overline{X}\xhookrightarrow{\;1_{X}\;}X$ and a 
surjective data value function $\equiv_{R}\xleftarrow{\;[]_{R}\;}Y$.
%
%%%%%%%%%%%%%%%%%%%%%%%%%%%%%%%%%%%%%%%%%%%%%%%%%%%%%%%%%%%%%%%%%%%%%%%%%%%%%%%%
%%%%%%%%%%%%%%%%%%%%%%%%%%%%%%%%%%%%%%%%%%%%%%%%%%%%%%%%%%%%%%%%%%%%%%%%%%%%%%%%
\footnote{The data value functions 
%{\footnotesize{$g_{1},g_{2}:Y\;\leftleftarrows\;Y'$}}
in the constraint
satisfy the invariant 
$\mathcal{I}={\langle{X,R}\rangle}$
%See the discussion of invariants and quotients in \S\,\ref{sub:sec:cls}.
consisting of 
the full subset of sorts $X$ and
the equivalence relation of data values 
$R = \{ (g_{1}(y'),g_{2}(y')) \mid  y' \in Y'\} \subseteq Y{\,\times\,}Y$
defined by the expression 
{\footnotesize{{$
g_{1}(y'){\;\models_{\mathcal{A}}\;}x
\text{ \underline{iff} }
y'{\;\models_{\mathcal{A}'}\;}x
\;\text{ \underline{iff} }\;
g_{2}(y'){\;\models_{\mathcal{A}}\;}x
$}}\normalsize}
for the parallel pair in the constraint.
The co-quotient $X$-type domain
$\mathcal{A}/\mathcal{I} 
= {\langle{X,\equiv_{R},\models_{\mathcal{A}/\mathcal{I}}}\rangle}$
has as data values the $R$-equivalence classes of data values in $Y'$. 
(invariants/quotients are discussed in \cite{barwise:seligman:97})
%the book
%{\itshape Information Flow: The Logic of Distributed Systems} 
%by Barwise and Seligman 
}
%%%%%%%%%%%%%%%%%%%%%%%%%%%%%%%%%%%%%%%%%%%%%%%%%%%%%%%%%%%%%%%%%%%%%%%%%%%%%%%%
%%%%%%%%%%%%%%%%%%%%%%%%%%%%%%%%%%%%%%%%%%%%%%%%%%%%%%%%%%%%%%%%%%%%%%%%%%%%%%%%
%\footnote{See Disp.\;\ref{typ:dom:mor:equiv:conds} in \S\,\ref{sub:sec:cls}.}
%
%\]
%
This is the construction for co-quotient 
(Tbl.\,\ref{tbl:fole:co-quotient:input:output}).
\newline
\item[Input/Output:] 
Consider a table
$\mathcal{T} = {\langle{K,t}\rangle} \in 
\mathrmbf{Tbl}_{\mathcal{S}}(\mathcal{A})$.
This table forms an adequate collection 
(Def.\;\ref{def:suff:adequ:colim})
to compute the coequalizer.
This is the input for co-quotient (Tbl.\,\ref{tbl:fole:co-quotient:input:output}).
The output is computed with one expansion.
%\newline
%\item[Output:] 
%
\begin{itemize}
\item 
Expansion 
$\mathrmbf{Tbl}_{\mathcal{S}}(\mathcal{A}/\mathcal{I})
{\;\xleftarrow{\acute{\mathrmbfit{tbl}}_{\mathcal{S}}([]_{R})}\;}
\mathrmbf{Tbl}_{\mathcal{S}}(\mathcal{A})$
along the tuple function of the 
%projection 
$X$-type domain morphism
{\footnotesize{$\mathcal{A}/\mathcal{I}\xrightarrow{\,[]_{R}\,}\mathcal{A}$}}
maps the table
$\mathcal{T}$
% \in \mathrmbf{Tbl}_{\mathcal{S}}(\mathcal{A})$
to the table 
$\widetilde{\mathcal{T}} 
= \acute{\mathrmbfit{tbl}}_{\mathcal{S}}([]_{R})(\mathcal{T})
= {\langle{K,\tilde{t}}\rangle} 
\in \mathrmbf{Tbl}_{\mathcal{S}}(\mathcal{A}/\mathcal{I})$,
with its tuple function
$K \xrightarrow{\tilde{t}} 
\mathrmbfit{tup}_{\mathcal{S}}(\mathcal{A}/\mathcal{I})$
defined by composition,
$\tilde{t} 
= t{\,\cdot\,}\mathrmbfit{tup}_{\mathcal{S}}([]_{R})$. 
%
%%%%%%%%%%%%%%%%%%%%%%%%%%%%%%%%%%%%%%%%%%%%%%%%%%%%%%%%%%%%%%%%%%%%%%
%%%%%%%%%%%%%%%%%%%%%%%%%%%%%%%%%%%%%%%%%%%%%%%%%%%%%%%%%%%%%%%%%%%%%%
\comment{The quotient table $\widetilde{\mathcal{T}}$
contains the set of all mergers of tuples in $\mathcal{T}$ 
of equivalence classes of sorts.}
%%%%%%%%%%%%%%%%%%%%%%%%%%%%%%%%%%%%%%%%%%%%%%%%%%%%%%%%%%%%%%%%%%%%%%
%%%%%%%%%%%%%%%%%%%%%%%%%%%%%%%%%%%%%%%%%%%%%%%%%%%%%%%%%%%%%%%%%%%%%%
%
This defines
%is linked to the table $\mathcal{T}$ by 
the $\mathcal{S}$-table morphism 
%\[\mbox
{\footnotesize{{$
\widetilde{\mathcal{T}}
%= {\langle{\mathcal{A}/\mathcal{I},K,\tilde{t}}\rangle} 
\xleftarrow{{\langle{[]_{R},1}\rangle}} 
%{\langle{\mathcal{A}.K,t}\rangle} =
\mathcal{T}$,}}\normalsize}
%\]
%
which is the output for co-quotient (Tbl.\,\ref{tbl:fole:co-quotient:input:output}).
\end{itemize}
\end{description}
%
%The co-quotient flowchart input/output is displayed in
%Tbl.\,\ref{tbl:fole:co-quotient:input:output}.
%where quotient is distinguish from mere inflation by having a constraint.
Co-quotient is the one-step process
%\mbox{}\newline\rule{50pt}{1pt}\newline
%
%\newline\mbox{}\hfill
%\rule[-10pt]{0pt}{26pt}
\[\mbox{\footnotesize{{$
{\Yleft}_{\mathcal{S}}(\mathcal{T})
\doteq
\acute{\mathrmbfit{tbl}}_{\mathcal{S}}([]_{R})(\mathcal{T})
$.}}\normalsize}\]
%\hfill\mbox{}\newline
%\mbox{}\newline\rule{300pt}{1pt}\newline
%
\begin{aside}
Theoretically
this would represent the co-equalizer,
the colimit 
(see the application discussion for co-completeness in 
\S\,\ref{sub:sec:lim:colim:tbl})
%(Chap.\;4 of \cite{kent:fole:era:tbl})
of a parallel pair 
{\footnotesize{${\langle{g_{1},k_{1}}\rangle},{\langle{g_{2},k_{2}}\rangle} :
\mathcal{T}\;\leftleftarrows\;\mathcal{T}'$}}
of $\mathcal{S}$-table morphisms.
But practically,
we are only given the constraint (parallel pair)
{\footnotesize{$g_{1},g_{2}:\mathcal{A}{\;\rightrightarrows}\;\mathcal{A}'$}}
of $X$-type domain morphisms 
and the input 
$\mathcal{T}$
in Tbl.\,\ref{tbl:fole:co-quotient:input:output}.
Similar comments,
which distinguish the practical from the theoretical, 
hold for the data-type join operation
in \S\,\ref{sub:sub:sec:boole:join}.
\end{aside}
\begin{table}
\begin{center}
{{\fbox{\begin{tabular}{c}
\setlength{\extrarowheight}{2pt}
{\scriptsize{$\begin{array}[c]{c@{\hspace{12pt}}l}
g_{1},g_{2}:\mathcal{A}{\;\rightrightarrows}\;\mathcal{A}'
&
\textit{constraint}
\\
\mathcal{A}/\mathcal{I}\xrightarrow{\,[]_{R}\,}\mathcal{A}
&
\textit{construction}
\\
\hline
\mathcal{T}\in\mathrmbf{Tbl}_{\mathcal{S}}(\mathcal{A})
&
\textit{input}
\\
{\Yleft}_{\mathcal{S}}(\mathcal{T})
%=\widetilde{\mathcal{T}}
\xleftarrow{{\langle{[]_{R},1}\rangle}} 
\mathcal{T}
&
\textit{output}
\end{array}$}}
\end{tabular}}}}
\end{center}
\caption{\texttt{FOLE} Co-quotient I/O}
\label{tbl:fole:co-quotient:input:output}
\end{table}
%

%%%%%%%%%%%%%%%%%%%%%%%%%%%%%%%%%%%%%%%%%%%%%%%%%%%%%%%%%%%%%%%%%%%%%%%
\newpage
\subsection{Co-core.}\label{sub:sub:sec:co-core}
%%%%%%%%%%%%%%%%%%%%%%%%%%%%%%%%%%%%%%%%%%%%%%%%%%%%%%%%%%%%%%%%%%%%%%%

%
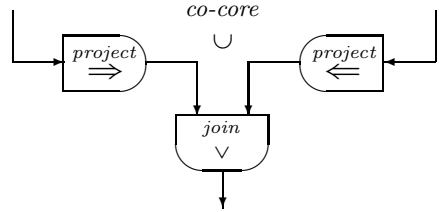
\begin{figure}
\begin{center}
{{{\begin{tabular}{c}
\begin{picture}(160,75)(37,27)
\setlength{\unitlength}{0.97pt}
%%%%%%%%%%%%%%%%%%%%%%%%%%%%%%%%%%%%%%%%%%%%%%%%%%
%\put(44,62){\begin{picture}(0,0)(0,0)
%\setlength{\unitlength}{0.46pt}
\put(54,65){\begin{picture}(0,0)(0,0)
\setlength{\unitlength}{0.35pt}
%\thicklines
%\put(106,40){\makebox(0,0){\normalsize{$\boldsymbol{\circ}$}}}
%\put(4.7,40){\makebox(0,0){\normalsize{$\boldsymbol{\circ}$}}}
\put(10,10){\line(1,0){60}}
\put(10,70){\line(1,0){60}}
\put(10,70){\line(0,-1){60}}
\put(70,40){\oval(60,60)[br]}
\put(70,40){\oval(60,60)[tr]}
\put(55,50){\makebox(0,0){\scriptsize{{\textit{{project}}}}}}
\put(56,30){\makebox(0,0){\Large{${\Rightarrow}$}}}
\end{picture}}
%%%%%%%%%%%%%%%%%%%%%%%%%%%%%%%%%%%%%%%%%%%%%%%%%%
\put(146.5,65){\begin{picture}(0,0)(0,0)
\setlength{\unitlength}{0.35pt}
%\thicklines
%\put(106,40){\makebox(0,0){\normalsize{$\boldsymbol{\circ}$}}}
%\put(4.7,40){\makebox(0,0){\normalsize{$\boldsymbol{\circ}$}}}
\put(40,10){\line(1,0){60}}
\put(40,70){\line(1,0){60}}
\put(100,70){\line(0,-1){60}}
\put(40,40){\oval(60,60)[bl]}
\put(40,40){\oval(60,60)[tl]}
\put(58,50){\makebox(0,0){\scriptsize{{\textit{{project}}}}}}
\put(56,30){\makebox(0,0){\Large{${\Leftarrow}$}}}
\end{picture}}
%%%%%%%%%%%%%%%%%%%%%%%%%%%%%%%%%%%%%%%%%%%%%%%%%%
\put(98,37){\begin{picture}(0,0)(0,3)
\setlength{\unitlength}{0.35pt}
\put(60,30){\makebox(0,0){\normalsize{$\vee$}}}
%\thicklines
\put(40,10){\line(1,0){40}}
\put(10,70){\line(1,0){100}}
\put(10,70){\line(0,-1){30}}
\put(110,70){\line(0,-1){30}}
\put(40,40){\oval(60,60)[bl]}
\put(80,40){\oval(60,60)[br]}
\put(60,55){\makebox(0,0){\scriptsize{{\textit{{join}}}}}}
\end{picture}}
%%%%%%%%%%%%%%%%%%%%%%%%%%%%%%%%%%%%%%%%%%%%%%%%%%
\put(120,100){\makebox(0,0){\footnotesize{{\textit{{co-core}}}}}}
\put(120,88){\makebox(0,0){\large{${\cup}$}}}
%%%%%%%%%%%%%%%%%%%%%%%%%%%%%%%%%%%%%%%%%%%%%%%%%%
\put(38,80){\line(0,1){20}}
\put(38,80){\vector(1,0){20}}
\put(110,80){\line(-1,0){20}}
\put(110,80){\vector(0,-1){21}}
\put(120,38){\vector(0,-1){15}}
\put(130,80){\vector(0,-1){21}}
\put(130,80){\line(1,0){20}}
\put(203,80){\line(0,1){20}}
\put(203,80){\vector(-1,0){20}}
%\thicklines
%\put(15,110){\line(1,0){210}}
%\put(55,10){\line(1,0){130}}
%\put(15,50){\line(0,1){60}}
%\put(225,50){\line(0,1){60}}
%\qbezier(15,50)(15,10)(55,10)
%\qbezier(185,10)(225,10)(225,50)
%%%%%%%%%%%%%%%%%%%%%%%%%%%%%%%%%%%%%%%%%%%%%%%%%%
\end{picture}
\end{tabular}}}}
\end{center}
\caption{\texttt{FOLE} Co-core Flow Chart}
\label{fig:fole:co-core:flo:chrt}
\end{figure}
%
%Let $\mathcal{S}$
% = {\langle{I,x,X}\rangle}
%be a fixed signature.
Let 
$\mathcal{A}$
% = {\langle{X,Y,\models_{\mathcal{A}}}\rangle}$
be a fixed type domain.
This section discusses the co-core operation.
%${}^{\ref{colimit:in:A}}$ 
%
The co-core operation is somewhat unorthodox, 
since it does not have a construction process. 
However, although it is defined for tables with a fixed type domain, 
it computes a colimit-like result: 
it uses left adjoint flow (projection) followed by union. 
In fact, it corresponds to the first half of flow 
along a signed domain morphism followed by union.
%${}^{\ref{colimit:in:A}}$ 
%
%%%%%%%%%%%%%%%%%%%%%%%%%%%%%%%%%%%%%%%%%%%%%%%%%%%%%%%%%%%%%%%%%%%%%%%%%%%%%%%%
%%%%%%%%%%%%%%%%%%%%%%%%%%%%%%%%%%%%%%%%%%%%%%%%%%%%%%%%%%%%%%%%%%%%%%%%%%%%%%%%
\comment{\S\,\ref{sub:sec:comp:ops:sign} discusses operations 
that compute a colimit.
%defined in the scope 
%$\mathrmbf{Tbl}(\mathcal{A})$ 
%of $\mathcal{A}$-tables.
Although the co-core operation 
is defined in the scope $\mathrmbf{Tbl}(\mathcal{A})$,
it appears here since it computes a colimit.}
%%%%%%%%%%%%%%%%%%%%%%%%%%%%%%%%%%%%%%%%%%%%%%%%%%%%%%%%%%%%%%%%%%%%%%%%%%%%%%%%
%%%%%%%%%%%%%%%%%%%%%%%%%%%%%%%%%%%%%%%%%%%%%%%%%%%%%%%%%%%%%%%%%%%%%%%%%%%%%%%%
%
We use the following routes of flow from 
%Tbl.\;\ref{tbl:routes:flow}.
Fig.\;\ref{fig:routes:flow:colim}.
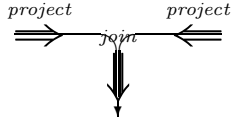
\begin{center}
{{\begin{tabular}{c}
\setlength{\unitlength}{0.6pt}
\begin{picture}(320,60)(0,5)
%\put(157,63){\makebox(0,0){\normalsize{$\textbf{2.}$}}}
%
%\put(325,25){\makebox(0,0){\normalsize{
%$\left.\rule{0pt}{24pt}\right\}
%\underset{\textstyle{\textsf{semi-join}}}{\textsf{right}}$}}}
%\put(-25,25){\makebox(0,0){\normalsize{
%$\underset{\textstyle{\textsf{semi-join}}}{\textsf{left}}
%\left\{\rule{0pt}{24pt}\right.$}}}
%
\put(100,55){\makebox(0,0){\huge{
$\overset{\textit{\scriptsize{project}}}{\Rightarrow}$}}}
\put(200,55){\makebox(0,0){\huge{
$\overset{\textit{\scriptsize{project}}}{\Leftarrow}$}}}
\put(150,30){\makebox(0,0){\huge{
$\overset{\textit{\scriptsize{join}}}{\Downarrow}$}}}
%\put(100,4.5){\makebox(0,0){\huge{
%$\overset{\textit{\scriptsize{project}}}{\Leftarrow}$}}}
%\put(105,-14.5){\makebox(0,0){$\textit{\scriptsize{image}}$}}
%\put(200,4.5){\makebox(0,0){\huge{
%$\overset{\textit{\scriptsize{project}}}{\Rightarrow}$}}}
%\put(205,-14.5){\makebox(0,0){$\textit{\scriptsize{image}}$}}
%
\put(90,48){\line(1,0){54}}
\put(144,38){\oval(20,20)[tr]}
\put(166,38){\oval(20,20)[tl]}
\put(220,48){\line(-1,0){54}}
%\put(145,-2){\line(-1,0){54}}
%\put(165,-2){\line(1,0){54}}
\put(154,9){\line(0,1){29}}
\put(156,9){\line(0,1){29}}
%\put(145,8){\oval(20,20)[br]}
%\put(165,8){\oval(20,20)[bl]}
%\put(155,6){\line(0,-1){8}}
\put(155,6){\vector(0,-1){10}}
\end{picture}
\end{tabular}}}
\end{center}
The idea here is to get rid of some of the unused indexes,
thereby getting rid of some of the unused data types.
%
%%%%%%%%%%%%%%%%%%%%%%%%%%%%%%%%%%%%%%%%%%%%%%%%%%%%%%%%%%%%%%%%%%%%%%%%%%%%%%%%
%%%%%%%%%%%%%%%%%%%%%%%%%%%%%%%%%%%%%%%%%%%%%%%%%%%%%%%%%%%%%%%%%%%%%%%%%%%%%%%%
\footnote{A header is \emph{co-core} when 
it appears in both tables and its index values are taken from either index set.}
%%%%%%%%%%%%%%%%%%%%%%%%%%%%%%%%%%%%%%%%%%%%%%%%%%%%%%%%%%%%%%%%%%%%%%%%%%%%%%%%
%%%%%%%%%%%%%%%%%%%%%%%%%%%%%%%%%%%%%%%%%%%%%%%%%%%%%%%%%%%%%%%%%%%%%%%%%%%%%%%%
%
This is comparable to the filtered join for signatures in 
\S\,\ref{sub:sub:sec:filtered:join},
where we want to get rid of some of the unreliable or inauthentic data values.
However,
co-core is an orthodox operation
since it does follow the construction of a limit, whereas
filtered join is an unorthodox operation 
since it does not follow the construction of a colimit.
%
%%%%%%%%%%%%%%%%%%%%%%%%%%%%%%%%%%%%%%%%%%%%%%%%%%%%%%%%%%%%%%%%%%%%%%%%%%%%%%%%
%%%%%%%%%%%%%%%%%%%%%%%%%%%%%%%%%%%%%%%%%%%%%%%%%%%%%%%%%%%%%%%%%%%%%%%%%%%%%%%%
\footnote{The co-core operation is the dual of the core operation.}
%\footnote{The restrict join operation is the  dual of the project join operation.}
%%%%%%%%%%%%%%%%%%%%%%%%%%%%%%%%%%%%%%%%%%%%%%%%%%%%%%%%%%%%%%%%%%%%%%%%%%%%%%%%
%%%%%%%%%%%%%%%%%%%%%%%%%%%%%%%%%%%%%%%%%%%%%%%%%%%%%%%%%%%%%%%%%%%%%%%%%%%%%%%%
%
%\newline
%{\fbox{\textbf{Actually,
%only filtered join and data-type meet are unorthodox operations.}}}
%\newline
%{\fbox{\textbf{Project join and restrict meet are OK.}}}
%\newline

\comment{
Let $\mathcal{A}$ be a type domain.
We are given 
two $\mathcal{A}$-tables 
$\mathcal{T}_{1} = {\langle{\mathcal{S}_{1},K_{1},t_{1}}\rangle}$ and
$\mathcal{T}_{2} = {\langle{\mathcal{S}_{2},K_{2},t_{2}}\rangle}$
connected 
through a third signature $\mathcal{S}$
by an 
%connecting 
$X$-sorted signature span
$\mathcal{S}_{1}
% = {\langle{X,Y_{1},\models_{\mathcal{A}_{1}}}\rangle} 
\xleftarrow{\;h_{1}\,} 
%{\langle{X,Y,\models_{\mathcal{A}}}\rangle} = 
\mathcal{S}
\xrightarrow{\;h_{2}\,} 
\mathcal{S}_{2}$
%$\mathcal{S}$-table 
%$\mathcal{T} = {\langle{\mathcal{A},K,t}\rangle}$
consisting of a span of index functions
$I_{1}\xleftarrow{\;h_{1}\;}I\xrightarrow{\;h_{2}\;}I_{2}$. 
%
%Co-core is the two-step process
%illustrated in 
%Fig.\;\ref{fig:fole:co-core}.
%\newline
}

\begin{description}
\item[Constraint/Construction:] 
Consider an $X$-sorted signature span 
$\mathcal{S}_{1}\xleftarrow{h_{1}}\mathcal{S}\xrightarrow{h_{2}}\mathcal{S}_{2}$
%in $\mathrmbf{List}(X)$
consisting of 
a span of index functions
$I_{1}\xleftarrow{h_{1}}I\xrightarrow{h_{2}}I_{2}$.
This is both the constraint and the construction for co-core 
(Tbl.\,\ref{tbl:fole:co-core:input:output}).
%\newline
%\item[Construction:] 
%
%The construction for co-core is the same as the constraint 
%(Tbl.\,\ref{tbl:fole:co-core:input:output}).
We deviate from orthodoxy (colimit construction) at this step.
\newline
\item[Input:] 
Consider a pair of tables
$\mathcal{T}_{1} = {\langle{K_{1},t_{1}}\rangle} \in 
\mathrmbf{Tbl}_{\mathcal{A}}(\mathcal{S}_{1})$
and
$\mathcal{T}_{2} = {\langle{K_{2},t_{2}}\rangle} \in 
\mathrmbf{Tbl}_{\mathcal{A}}(\mathcal{S}_{2})$.
This is the input for co-core 
(Tbl.\,\ref{tbl:fole:co-core:input:output}).
\newline
\item[Output:] 
The output is projection (twice) followed by join.
\newline
\begin{itemize}
\item 
Projection 
{\footnotesize{$
\mathrmbf{Tbl}_{\mathcal{A}}(\mathcal{S}_{1})
{\;\xrightarrow
%[{\scriptscriptstyle\sum}_{h_{2}}]
{\;\acute{\mathrmbfit{tbl}}_{\mathcal{A}}(h_{1})\;}\;}
\mathrmbf{Tbl}_{\mathcal{A}}(\mathcal{S})
$}\normalsize}
(\S\,\ref{sub:sub:sec:adj:flow:A})
along the tuple function 
of the $X$-signature morphism
$\mathcal{S}_{1}\xleftarrow{h_{1}}\mathcal{S}$
maps the $\mathcal{A}$-table
$\mathcal{T}_{1}$
% = {\langle{K_{1},t_{1}}\rangle} \in 
%\mathrmbf{Tbl}_{\mathcal{A}}(\mathcal{S}_{1})$
to the $\mathcal{A}$-table
$\widetilde{\mathcal{T}}_{1}
= \acute{\mathrmbfit{tbl}}_{\mathcal{A}}(h_{1})(\mathcal{T}_{1})
%= {\scriptstyle\sum}_{h_{1}}(\mathcal{T}_{1})
= {\langle{K_{1},\tilde{t}_{1}}\rangle} 
\in \mathrmbf{Tbl}_{\mathcal{A}}(\mathcal{S})$,
with its tuple function
$K_{1} \xrightarrow{\tilde{t}_{1}} 
\mathrmbfit{tup}_{\mathcal{A}}(\mathcal{S})$
defined by composition,
$\tilde{t}_{1} = t_{1}{\,\cdot\,}\mathrmbfit{tup}_{\mathcal{A}}(h_{1})$. 
This is linked to the table $\mathcal{T}_{1}$ 
by the $\mathcal{A}$-table morphism 
%\[\mbox
{\footnotesize{{$
\mathcal{T}_{1} = {\langle{\mathcal{S}_{1}.K_{1},t_{1}}\rangle}
\xrightarrow{{\langle{h_{1},1}\rangle}} 
{\langle{\mathcal{S},K_{1},\tilde{t}_{1}}\rangle} 
= \widetilde{\mathcal{T}}_{1}
%{\scriptstyle\sum}_{h_{1}}(\mathcal{T}_{1})
$.}}\normalsize}
%\]
Similarly for $\mathcal{A}$-table
$\mathcal{T}_{2} = {\langle{K_{2},t_{2}}\rangle} \in 
\mathrmbf{Tbl}_{\mathcal{A}}(\mathcal{S}_{2})$.
\newline
%\item[join:] 
%
\item 
Union (\S\,\ref{sub:sub:sec:boole})
of the two projection tables $\widetilde{\mathcal{T}}_{1}$ 
and $\widetilde{\mathcal{T}}_{2}$ 
in the context $\mathrmbf{Tbl}_{\mathcal{A}}(\mathcal{S})$
defines the co-core table
$\mathcal{T}_{1}{\;{\cup}_{\mathcal{A}}\;}\mathcal{T}_{2}
= \widetilde{\mathcal{T}}_{2}{\,\vee\,}\widetilde{\mathcal{T}}_{2}
= {\langle{K_{1}{+}K_{2},{[\tilde{t}_{1},\tilde{t}_{2}]}}\rangle}$,
%\widetilda{K}_{12}
%\item 
%The intersection operation defines the \texttt{FOLE} table
%$\mathcal{T}\wedge\mathcal{T}' = {\langle{\widehat{K},{(t,t')}}\rangle}$
whose key set $K_{1}{+}K_{2}$ is the disjoint union and 
whose tuple map 
$K_{1}{+}K_{2}
\xrightarrow{[\tilde{t}_{1},\tilde{t}_{2}]}
\mathrmbfit{tup}_{\mathcal{A}}(\mathcal{S})$
maps $k_{1} \in K_{1}$ 
to $\tilde{t}_{1}(k_{1}) \in \mathrmbfit{tup}_{\mathcal{A}}(\mathcal{S})$
and
maps $k_{2} \in K_{2}$ 
to $\tilde{t}_{2}(k_{2}) \in \mathrmbfit{tup}_{\mathcal{A}}(\mathcal{S})$.
Union is the coproduct in $\mathrmbf{Tbl}_{\mathcal{A}}(\mathcal{S})$
with opspan
%\newline\mbox{}\hfill
$\widetilde{\mathcal{T}}_{1}\xrightarrow{\check{\iota}_{1}}
\widetilde{\mathcal{T}}_{1}{\,\vee\,}\widetilde{\mathcal{T}}_{2}
\xleftarrow{\check{\iota}_{2}}\widetilde{\mathcal{T}}_{2}$.
%\hfill\mbox{}\newline
\end{itemize}
\mbox{}\newline
Projection composed with join 
defines the opspan of $\mathcal{A}$-table morphisms
\[\mbox{\footnotesize{
{$\mathcal{T}_{1}
\xrightarrow[\;{\langle{h_{1},1}\rangle}{\circ\,}\check{\iota}_{1}\;]
{\;{\langle{h_{1},\check{\iota}_{1}}\rangle}\;} 
\mathcal{T}_{1}{\;{\cup}_{\mathcal{A}}\;}\mathcal{T}_{2}
\xleftarrow[\;{\langle{h_{2},1}\rangle}{\circ\,}\check{\iota}_{2}\;]
{\;{\langle{h_{2},\check{\iota}_{1}}\rangle}\;} 
\mathcal{T}_{2}
$,}}\normalsize}\]
which is the output for co-core (Tbl.\,\ref{tbl:fole:co-core:input:output}).
%\end{description}
%
\end{description}
%
%The co-core flowchart input/output is displayed in 
%Tbl.\,\ref{tbl:fole:co-core:input:output}.
Co-core 
%(Fig.\;\ref{fig:fole:nat:join})
%within the context $\mathrmbf{Tbl}(\mathcal{A})$
is projection followed by join (disjunction or union).
This is the two-step process 
\newline\mbox{}\hfill
\rule[-10pt]{0pt}{26pt}
$\mathcal{T}_{1}{\;{\cup}_{\mathcal{A}}\;}\mathcal{T}_{2} 
\doteq  
\acute{\mathrmbfit{tbl}}_{\mathcal{A}}(h_{1})(\mathcal{T}_{1})
{\;\vee\;}
\acute{\mathrmbfit{tbl}}_{\mathcal{A}}(h_{2})(\mathcal{T}_{2})$.
%
%%%%%%%%%%%%%%%%%%%%%%%%%%%%%%%%%%%%%%%%%%%%%%%%%%%%%%%%%%%%%%%%%%%%%%%%%%%%%%%%
%%%%%%%%%%%%%%%%%%%%%%%%%%%%%%%%%%%%%%%%%%%%%%%%%%%%%%%%%%%%%%%%%%%%%%%%%%%%%%%%
\footnote{The co-core 
%in \S\,\ref{sub:sub:sec:co-core},
$\mathcal{T}_{1}
\xrightarrow
%[\;{\langle{h_{1},1}\rangle}{\circ\,}\check{\iota}_{1}\;]
{\;{\langle{h_{1},\check{\iota}_{1}}\rangle}\;} 
\mathcal{T}_{1}{\;{\cup}_{\mathcal{A}}\;}\mathcal{T}_{2}
\xleftarrow
%[\;{\langle{h_{2},1}\rangle}{\circ\,}\check{\iota}_{2}\;]
{\;{\langle{h_{2},\check{\iota}_{1}}\rangle}\;} 
\mathcal{T}_{2}$,
is homogeneous with and has a direct connection
to both tables $\mathcal{T}_{1}$ and $\mathcal{T}_{2}$. 
This is comparable with
the core (\S\,\ref{sub:sub:sec:core})
$\mathcal{T}_{1}
\xleftarrow
%[\;{\langle{g_{1},\grave{k}_{1}}\rangle}{\circ\,}\check{\pi}_{1}\;]
{\;{\langle{g_{1},\hat{k}_{1}}\rangle}\;} 
\mathcal{T}_{1}{\;{\cap}_{\mathcal{S}}\;}\mathcal{T}_{2}
\xrightarrow[\;{\langle{g_{2},\grave{k}_{2}}\rangle}{\circ\,}\hat{\pi}_{2}\;]
{\;{\langle{g_{2},\hat{k}_{1}}\rangle}\;} 
\mathcal{T}_{2}$,
which is homogeneous with and has a direct connection
to both tables $\mathcal{T}_{1}$ and $\mathcal{T}_{2}$.}
%%%%%%%%%%%%%%%%%%%%%%%%%%%%%%%%%%%%%%%%%%%%%%%%%%%%%%%%%%%%%%%%%%%%%%%%%%%%%%%%
%%%%%%%%%%%%%%%%%%%%%%%%%%%%%%%%%%%%%%%%%%%%%%%%%%%%%%%%%%%%%%%%%%%%%%%%%%%%%%%%
%
\hfill\mbox{}\newline

\begin{table}
\begin{center}
{{\fbox{\begin{tabular}{c}
\setlength{\extrarowheight}{2pt}
{\scriptsize{$\begin{array}[c]{c@{\hspace{12pt}}l}
\mathcal{S}_{1}\xleftarrow{h_{1}}\mathcal{S}\xrightarrow{h_{2}}\mathcal{S}_{2}
&
\textit{constraint}
\\
%\mathcal{S}_{1}\xleftarrow{h_{1}}\mathcal{S}\xrightarrow{h_{2}}\mathcal{S}_{2}
&
\textit{/construction}
\\
\hline
\mathcal{T}_{1}\in\mathrmbf{Tbl}_{\mathcal{A}}(\mathcal{S}_{1})
\text{ and }
\mathcal{T}_{1}\in\mathrmbf{Tbl}_{\mathcal{A}}(\mathcal{S}_{1})
&
\textit{input}
\\
\mathcal{T}_{1}
\xrightarrow[\;{\langle{h_{1},1}\rangle}{\circ\,}\check{\iota}_{1}\;]
{\;{\langle{h_{1},\check{\iota}_{1}}\rangle}\;} 
\mathcal{T}_{1}{\;{\cup}_{\mathcal{A}}\;}\mathcal{T}_{2}
\xleftarrow[\;{\langle{h_{2},1}\rangle}{\circ\,}\check{\iota}_{2}\;]
{\;{\langle{h_{2},\check{\iota}_{1}}\rangle}\;} 
\mathcal{T}_{2}
&
\textit{output}
\end{array}$}}
\end{tabular}}}}
\end{center}
\caption{\texttt{FOLE} Co-core I/O}
\label{tbl:fole:co-core:input:output}
\end{table}
\newpage
\subsection{Data-type Join.}
\label{sub:sub:sec:boole:join}
%%%%%%%%%%%%%%%%%%%%%%%%%%%%%%%%%%%%%%%%%%%%%%%%%%%%%%%%%%%%%

%
%%%%%%%%%%%%%%%%%%%%%%%%%%%%%%%%%%%%%%%%%%%%%%%%%%%%%%%%%%%%
%\newpage
%\paragraph{Data-type Join.}
%%%%%%%%%%%%%%%%%%%%%%%%%%%%%%%%%%%%%%%%%%%%%%%%%%%%%%%%%%%%
%

%
\begin{figure}
\begin{center}
{{{\begin{tabular}{c}
\begin{picture}(160,75)(37,27)
\setlength{\unitlength}{0.97pt}
%%%%%%%%%%%%%%%%%%%%%%%%%%%%%%%%%%%%%%%%%%%%%%%%%%
%\put(44,62){\begin{picture}(0,0)(0,0)
%\setlength{\unitlength}{0.46pt}
\put(54,65){\begin{picture}(0,0)(0,0)
\setlength{\unitlength}{0.35pt}
%\thicklines
%\put(106,40){\makebox(0,0){\normalsize{$\boldsymbol{\circ}$}}}
%\put(4.7,40){\makebox(0,0){\normalsize{$\boldsymbol{\circ}$}}}
\put(10,10){\line(1,0){60}}
\put(10,70){\line(1,0){60}}
\put(10,70){\line(0,-1){60}}
\put(70,40){\oval(60,60)[br]}
\put(70,40){\oval(60,60)[tr]}
\put(55,50){\makebox(0,0){\scriptsize{{\textit{{expand}}}}}}
\put(56,30){\makebox(0,0){\Large{${\Rightarrow}$}}}
\end{picture}}
%%%%%%%%%%%%%%%%%%%%%%%%%%%%%%%%%%%%%%%%%%%%%%%%%%
\put(146.5,65){\begin{picture}(0,0)(0,0)
\setlength{\unitlength}{0.35pt}
%\thicklines
%\put(106,40){\makebox(0,0){\normalsize{$\boldsymbol{\circ}$}}}
%\put(4.7,40){\makebox(0,0){\normalsize{$\boldsymbol{\circ}$}}}
\put(40,10){\line(1,0){60}}
\put(40,70){\line(1,0){60}}
\put(100,70){\line(0,-1){60}}
\put(40,40){\oval(60,60)[bl]}
\put(40,40){\oval(60,60)[tl]}
\put(58,50){\makebox(0,0){\scriptsize{{\textit{{expand}}}}}}
\put(56,30){\makebox(0,0){\Large{${\Leftarrow}$}}}
\end{picture}}
%%%%%%%%%%%%%%%%%%%%%%%%%%%%%%%%%%%%%%%%%%%%%%%%%%
\put(98,37){\begin{picture}(0,0)(0,3)
\setlength{\unitlength}{0.35pt}
\put(60,30){\makebox(0,0){\normalsize{$\vee$}}}
%\thicklines
\put(40,10){\line(1,0){40}}
\put(10,70){\line(1,0){100}}
\put(10,70){\line(0,-1){30}}
\put(110,70){\line(0,-1){30}}
\put(40,40){\oval(60,60)[bl]}
\put(80,40){\oval(60,60)[br]}
\put(60,55){\makebox(0,0){\scriptsize{{\textit{{join}}}}}}
\end{picture}}
%%%%%%%%%%%%%%%%%%%%%%%%%%%%%%%%%%%%%%%%%%%%%%%%%%
\put(120,100){\makebox(0,0){\footnotesize{{\textit{{data-type join}}}}}}
\put(120,88){\makebox(0,0){\large{$\oplus$}}}
%%%%%%%%%%%%%%%%%%%%%%%%%%%%%%%%%%%%%%%%%%%%%%%%%%
\put(38,80){\line(0,1){20}}
\put(38,80){\vector(1,0){20}}
\put(110,80){\line(-1,0){20}}
\put(110,80){\vector(0,-1){21}}
\put(120,38){\vector(0,-1){15}}
\put(130,80){\vector(0,-1){21}}
\put(130,80){\line(1,0){20}}
\put(203,80){\line(0,1){20}}
\put(203,80){\vector(-1,0){20}}
%\thicklines
%\put(15,110){\line(1,0){210}}
%\put(55,10){\line(1,0){130}}
%\put(15,50){\line(0,1){60}}
%\put(225,50){\line(0,1){60}}
%\qbezier(15,50)(15,10)(55,10)
%\qbezier(185,10)(225,10)(225,50)
%%%%%%%%%%%%%%%%%%%%%%%%%%%%%%%%%%%%%%%%%%%%%%%%%%
\end{picture}
\end{tabular}}}}
\end{center}
\caption{\texttt{FOLE} Data-type Join Flow Chart}
\label{fig:fole:boole:join:flo:chrt}
\end{figure}
The data-type join for tables 
is the relational counterpart of
the logical disjunction for predicates.
Where 
the \emph{join} operation is the analogue for logical disjunction 
at the small scope $\mathrmbf{Tbl}(\mathcal{D})$ of a signed domain table fiber,
the \emph{data-type join} is defined at the intermediate scope 
$\mathrmbf{Tbl}(\mathcal{S})$ of a signature table fiber, and 
the \emph{generic join} (\S\,\ref{sub:sub:sec:generic:join})
is defined at the large scope $\mathrmbf{Tbl}$ of all tables.
We identify these three concepts as colimits at different scopes.

In this section,
we focus on tables in the context $\mathrmbf{Tbl}(\mathcal{S})$ 
for fixed signature (header) $\mathcal{S}$.
% = {\langle{I,x,X}\rangle}$.
%In particular,
%we discuss the special case of pushout --- 
%the data-type join of two $\mathcal{S}$-tables.
%
In this context,
generic joins 
--- for the special case of pushout --- 
are called data-type joins, 
the join of two $\mathcal{S}$-tables.
%
%As we observed in \cite{kent:fole:era:tbl},
These colimits are resolvable into expansions
followed by join.
We use the following routes of flow from 
%Tbl.\;\ref{tbl:routes:flow}.
Fig.\;\ref{fig:routes:flow:colim}.
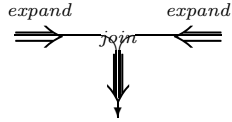
\begin{center}
{{\begin{tabular}{c}
\setlength{\unitlength}{0.6pt}
\begin{picture}(320,60)(0,5)
%\put(157,63){\makebox(0,0){\normalsize{$\textbf{1.}$}}}
%
%\put(325,25){\makebox(0,0){\normalsize{
%$\left.\rule{0pt}{24pt}\right\}
%\underset{\textstyle{\textsf{semi-join}}}{\textsf{right}}$}}}
%\put(-25,25){\makebox(0,0){\normalsize{
%$\underset{\textstyle{\textsf{semi-join}}}{\textsf{left}}
%\left\{\rule{0pt}{24pt}\right.$}}}
%
\put(100,55){\makebox(0,0){\huge{
$\overset{\textit{\scriptsize{expand}}}{\Rightarrow}$}}}
\put(200,55){\makebox(0,0){\huge{
$\overset{\textit{\scriptsize{expand}}}{\Leftarrow}$}}}
\put(150,30){\makebox(0,0){\huge{
$\overset{\textit{\scriptsize{join}}}{\Downarrow}$}}}
%\put(100,4.5){\makebox(0,0){\huge{
%$\overset{\textit{\scriptsize{project}}}{\Leftarrow}$}}}
%\put(105,-14.5){\makebox(0,0){$\textit{\scriptsize{image}}$}}
%\put(200,4.5){\makebox(0,0){\huge{
%$\overset{\textit{\scriptsize{project}}}{\Rightarrow}$}}}
%\put(205,-14.5){\makebox(0,0){$\textit{\scriptsize{image}}$}}
%
\put(90,48){\line(1,0){54}}
\put(144,38){\oval(20,20)[tr]}
\put(166,38){\oval(20,20)[tl]}
\put(220,48){\line(-1,0){54}}
%\put(145,-2){\line(-1,0){54}}
%\put(165,-2){\line(1,0){54}}
\put(154,9){\line(0,1){29}}
\put(156,9){\line(0,1){29}}
%\put(145,8){\oval(20,20)[br]}
%\put(165,8){\oval(20,20)[bl]}
%\put(155,6){\line(0,-1){8}}
\put(155,6){\vector(0,-1){10}}
\end{picture}
\end{tabular}}}
\end{center}
The data-type join operation is dual to the natural join operation of 
\S\,\ref{sub:sub:sec:nat:join}.
Similar to  natural join,
we can define data-type join for any number of tables 
$\{ \mathcal{T}_{1}, \mathcal{T}_{2}, \mathcal{T}_{3}, \cdots , \mathcal{T}_{n} \}$
with a comparable constraint.
%
%We are 
%Given 
%two $\mathcal{A}$-tables
%$\mathcal{T}_{1}={\langle{\mathcal{S}_{1},K_{1},t_{1}}\rangle}$
%and
%$\mathcal{T}_{2}={\langle{\mathcal{S}_{1},K_{2},t_{2}}\rangle}$,
%which are linked by 

%
\begin{description}
\item[Constraint:] 
Consider an $X$-sorted type domain opspan 
$\mathcal{A}_{1}
%={\langle{I_{1},s_{1}}\rangle} 
\xrightarrow{g_{1}} 
%\overset{
\mathcal{A}
%}{\mathcal{S}} 
\xleftarrow{g_{2}} 
%{\langle{I_{2},s_{2}}\rangle} = 
\mathcal{A}_{2}$
in $\mathrmbf{Cls}(X)$
consisting of 
a span of data value functions
$Y_{1}\xleftarrow{g_{1}}Y\xrightarrow{g_{2}}Y_{2}$.
This is the constraint for data-type join 
(Tbl.\,\ref{tbl:fole:boolean:join:input:output}).
\newline
\item[Construction:] 
%
%We can form (Chap.\;4 of \cite{kent:fole:era:tbl})
%(\S\,\ref{sub:sub:sec:sign:constr}) 
The pullback of this constraint
%colimiting cocone of this signature span
in $\mathrmbf{Cls}(X)$
is the span
{\footnotesize{$
\mathcal{A}_{1} 
%= {\langle{I_{1},s_{1}}\rangle} 
\xleftarrow{\;\tilde{g}_{1}\,} 
%\overset{\textstyle
{\mathcal{A}_{1}{\times_{\mathcal{A}}}\mathcal{A}_{2}}
%}{\overbrace{{\langle{I_{1}{+}_{I}I_{2},[s_{1},s_{2}]}\rangle}}} 
\xrightarrow{\;\tilde{g}_{2}} 
%{\langle{I_{2},s_{2}}\rangle} = 
\mathcal{A}_{2}
$}\normalsize}
of projection $X$-type domain morphisms
with pullback type domain
${\mathcal{A}_{1}{\times_{\mathcal{A}}}\mathcal{A}_{2}}$
and
data value function opspan
{\footnotesize{$Y_{1} 
\xrightarrow{\tilde{g}_{1}\,} 
{\langle{Y_{1}{+}_{Y}Y_{2},[g_{1},g_{2}]}\rangle}
\xleftarrow{\;\tilde{g}_{2}}
Y_{2}.$}\normalsize}
This is the construction for data-type join 
(Tbl.\,\ref{tbl:fole:boolean:join:input:output}).
\newline
\item[Input:] 
Consider a pair of tables
$\mathcal{T}_{1} = {\langle{K_{1},t_{1}}\rangle} \in 
\mathrmbf{Tbl}_{\mathcal{S}}(\mathcal{A}_{1})$
and
$\mathcal{T}_{2} = {\langle{K_{2},t_{2}}\rangle} \in 
\mathrmbf{Tbl}_{\mathcal{S}}(\mathcal{A}_{2})$.
These two tables form an adequate collection 
(Def.\;\ref{def:suff:adequ:colim})
to compute the pushout.
This is the input for data-type join 
(Tbl.\,\ref{tbl:fole:boolean:join:input:output}).
\newpage
\item[Output:] 
The output is expansion (twice) followed by join.
%blah blah blah
%\begin{description}
%\item[expansion:] 
%
\begin{itemize}
\item 
Expansion 
$\mathrmbf{Tbl}_{\mathcal{S}}(\mathcal{A}_{1})
{\;\xrightarrow{\acute{\mathrmbfit{tbl}}_{\mathcal{S}}(\tilde{g}_{1})}\;}
\mathrmbf{Tbl}_{\mathcal{S}}(\mathcal{A}_{1}{\times_{\mathcal{A}}}\mathcal{A}_{2})$
along the tuple function of the $X$-type domain morphism 
{\footnotesize{$\mathcal{A}_{1} 
\xleftarrow{\;\tilde{g}_{1}\,} 
{\mathcal{A}_{1}{\times_{\mathcal{A}}}\mathcal{A}_{2}}$}\normalsize}
maps
the $\mathcal{S}$-table
$\mathcal{T}_{1}$
% = {\langle{K_{1},t_{1}}\rangle} \in 
%\mathrmbf{Tbl}_{\mathcal{S}}(\mathcal{A}_{1})$
to the $\mathcal{S}$-table
$\widetilde{\mathcal{T}}_{1}
= \acute{\mathrmbfit{tbl}}_{\mathcal{S}}(\tilde{g}_{1})(\mathcal{T}_{1})
%= {\scriptstyle\sum}_{\tilde{g}_{1}}(\mathcal{T}_{1})
= {\langle{K_{1},\tilde{t}_{1}}\rangle} 
\in \mathrmbf{Tbl}_{\mathcal{S}}(\mathcal{A}_{1}{\times_{\mathcal{A}}}\mathcal{A}_{2})$,
with its tuple function
$K_{1} \xrightarrow{\tilde{t}_{1}} 
\mathrmbfit{tup}_{\mathcal{S}}
(\mathcal{A}_{1}{\times_{\mathcal{A}}}\mathcal{A}_{2})$
defined by composition,
$\tilde{t}_{1} 
= t_{1}{\,\cdot\,}\mathrmbfit{tup}_{\mathcal{S}}(\tilde{g}_{1})$. 
This is linked to the table $\mathcal{T}_{1}$ 
by the $\mathcal{S}$-table morphism 
%\[\mbox
{\footnotesize{{$
\mathcal{T}_{1} = {\langle{\mathcal{A}_{1}.K_{1},t_{1}}\rangle}
\xrightarrow{{\langle{1,\tilde{g}_{1}}\rangle}} 
{\langle{
\mathcal{A}_{1}{\times_{\mathcal{A}}}\mathcal{A}_{2},K_{1},\tilde{t}_{1}}\rangle} 
= \widetilde{\mathcal{T}}_{1}
%{\scriptstyle\sum}_{\tilde{g}_{1}}(\mathcal{T}_{1})
$.}}\normalsize}
%\]
Similarly for $\mathcal{S}$-table
$\mathcal{T}_{2} = {\langle{K_{2},t_{2}}\rangle} \in 
\mathrmbf{Tbl}_{\mathcal{S}}(\mathcal{A}_{2})$.
\newline
%\item[join:] 
%
\item 
Union (\S\,\ref{sub:sub:sec:boole})
of the two expansion tables 
$\widetilde{\mathcal{T}}_{1}$
%{\scriptstyle\sum}_{\tilde{g}_{1}}(\mathcal{T}_{1})$
and 
$\widetilde{\mathcal{T}}_{2}$
%${\scriptstyle\sum}_{\tilde{g}_{1}}(\mathcal{T}_{2})$ 
in the context 
$\mathrmbf{Tbl}_{\mathcal{S}}(\mathcal{A}_{1}{\times}_{\mathcal{A}}\mathcal{A}_{2})$
defines the data-type join 
%$\mathcal{T}_{1}{\times_{\mathcal{T}}}\mathcal{T}_{2}$
${\mathcal{T}_{1}}{\,\oplus_{\mathcal{S}}}{\mathcal{T}_{2}}
= \widetilde{\mathcal{T}}_{1} \vee \widetilde{\mathcal{T}}_{2}
= {\langle{K_{1}{+}K_{2},[\tilde{t}_{1},\tilde{t}_{2}]}\rangle}$,
whose key set is the disjoint union $K_{1}{+}K_{2}$
and whose tuple map
$K_{1}{+}K_{2}\xrightarrow{[\tilde{t}_{1},\tilde{t}_{2}]}\mathrmbfit{tup}_{\mathcal{A}}(\mathcal{S})$
is the comediator of the opspan
$K_{1}\xrightarrow{\tilde{t}_{1}} 
\mathrmbfit{tup}_{\mathcal{S}}
(\mathcal{A}_{1}{\times_{\mathcal{A}}}\mathcal{A}_{2})
\xleftarrow{\tilde{t}_{2}}K_{2}$,
resulting in the opspan 
%of 
%${\langle{I_{1}{+}_{I}I_{2},[s_{1},s_{2}],\mathcal{A}}\rangle}$-
%table morphisms
%\newline\mbox{}\hfill
%\rule[-10pt]{0pt}{26pt}\bowtie
%\[\mbox
{\footnotesize{{$
%{\scriptstyle\sum}_{\tilde{g}_{1}}(\mathcal{T}_{1})
\widetilde{\mathcal{T}}_{1}
%= {\langle{K_{1},\tilde{t}_{1}}\rangle} 
\xrightarrow{\;
\check{\iota}_{1}\;} 
\mathcal{T}_{1}{\,\oplus_{\mathcal{S}}}\mathcal{T}_{2}
\xleftarrow{\;\check{\iota}_{2}\;} 
%{\langle{K_{2},\tilde{t}_{2}}\rangle} = 
\widetilde{\mathcal{T}}_{2}
%{\scriptstyle\sum}_{\tilde{g}_{2}}(\mathcal{T}_{2})
$.}}\normalsize}
%\]
%
\end{itemize}
\mbox{}\newline
Expansion composed with join defines the opspan of $\mathcal{S}$-table morphisms
\[\mbox{\footnotesize{
{$
\mathcal{T}_{1}
\xrightarrow{\;{\langle{\tilde{g}_{1},\check{\iota}_{1}}\rangle}\;} 
\mathcal{T}_{1}{\,\oplus_{\mathcal{S}}}\mathcal{T}_{2}
\xleftarrow{\;{\langle{\tilde{g}_{2},\check{\iota}_{2}}\rangle}\;} 
\mathcal{T}_{2}
$,}}\normalsize}\]
which is the output for data-type join 
(Tbl.\,\ref{tbl:fole:boolean:join:input:output}).
%\end{description}
%
\end{description}
%
%The data-type join flowchart input/output is displayed in 
%Tbl.\,\ref{tbl:fole:boolean:join:input:output}. 
Data-type join is expansion followed by join. 
%Data-type join within the context $\mathrmbf{Tbl}(\mathcal{S})$
%is expansion followed by join.
This is the two-step process 
\newline\mbox{}\hfill
\rule[-10pt]{0pt}{26pt}
$\mathcal{T}_{1}{\,\oplus_{\mathcal{S}}}\mathcal{T}_{2} 
\doteq  
\acute{\mathrmbfit{tbl}}_{\mathcal{S}}(\tilde{g}_{1})(\mathcal{T}_{1})
{\;\vee\;}
\acute{\mathrmbfit{tbl}}_{\mathcal{S}}(\tilde{g}_{2})(\mathcal{T}_{2})$.
%
%%%%%%%%%%%%%%%%%%%%%%%%%%%%%%%%%%%%%%%%%%%%%%%%%%%%%%%%%%%%%%%%%%%%%%
%%%%%%%%%%%%%%%%%%%%%%%%%%%%%%%%%%%%%%%%%%%%%%%%%%%%%%%%%%%%%%%%%%%%%%
\footnote{The data-type join contains the set of all tuples in 
$\mathcal{T}_{1}$ and $\mathcal{T}_{2}$,
considered as taken from the data value set $Y_{1}{+}_{Y}Y_{2}$.}
%%%%%%%%%%%%%%%%%%%%%%%%%%%%%%%%%%%%%%%%%%%%%%%%%%%%%%%%%%%%%%%%%%%%%%
%%%%%%%%%%%%%%%%%%%%%%%%%%%%%%%%%%%%%%%%%%%%%%%%%%%%%%%%%%%%%%%%%%%%%%
%
\hfill\mbox{}\newline

\begin{aside}
Theoretically
this would represent pushout,
the colimit 
(see the application discussion for co-completeness in 
\S\,\ref{sub:sec:lim:colim:tbl})
%(Chap.\;4 of \cite{kent:fole:era:tbl})
of an span 
{\footnotesize{$\mathcal{T}_{1}\xleftarrow{\langle{g_{1},k_{1}}\rangle} 
\mathcal{T}
\xrightarrow{\langle{g_{2},k_{2}}\rangle}\mathcal{T}_{2}$}}
of $\mathcal{S}$-tables.
But practically,
we are only given 
%
%%%%%%%%%%%%%%%%%%%%%%%%%%%%%%%%%%%%%%%%%%%%%%%%%%%%%%%%%%%%%%%%%%%%%%
%%%%%%%%%%%%%%%%%%%%%%%%%%%%%%%%%%%%%%%%%%%%%%%%%%%%%%%%%%%%%%%%%%%%%%
%\comment{In practice, 
%the data-type join is understood to be 
%the set of all combinations of tuples in $\mathcal{T}_{1}$ and $\mathcal{T}_{2}$ 
%that are equal on their common data indexes.}
%%%%%%%%%%%%%%%%%%%%%%%%%%%%%%%%%%%%%%%%%%%%%%%%%%%%%%%%%%%%%%%%%%%%%%
%%%%%%%%%%%%%%%%%%%%%%%%%%%%%%%%%%%%%%%%%%%%%%%%%%%%%%%%%%%%%%%%%%%%%%
%
the constraint (opspan)
%an $X$-sorted signature span 
{\footnotesize{
$\mathcal{A}_{1}\xrightarrow{g_{1}}\mathcal{A}\xleftarrow{g_{2}}\mathcal{A}_{2}$}}
of $X$-indexed type domains
%in $\mathrmbf{List}(X)$
and the input tables 
$\mathcal{T}_{1}\in\mathrmbf{Tbl}_{\mathcal{S}}(\mathcal{A}_{1})$
and 
$\mathcal{T}_{1}\in\mathrmbf{Tbl}_{\mathcal{S}}(\mathcal{A}_{1})$
in Tbl.\,\ref{tbl:fole:boolean:join:input:output}.
Similar comments,
which distinguish the practical from the theoretical, 
hold for the co-quotient operation
in \S\,\ref{sub:sub:sec:co-quotient}.
\end{aside}
%

%\mbox{}\newline\rule{300pt}{1pt}\newline

%
\begin{table}
\begin{center}
{{\fbox{\begin{tabular}{c}
\setlength{\extrarowheight}{2pt}
{\scriptsize{$\begin{array}[c]{c@{\hspace{12pt}}l}
\mathcal{A}_{1}\xrightarrow{g_{1}}\mathcal{A}\xleftarrow{g_{2}}\mathcal{A}_{2}
&
\textit{constraint}
\\
\mathcal{A}_{1} 
\xleftarrow{\;\tilde{g}_{1}\,} 
{\mathcal{A}_{1}{\times_{\mathcal{A}}}\mathcal{A}_{2}}
\xrightarrow{\;\tilde{g}_{2}} 
\mathcal{A}_{2}
&
\textit{construction}
\\
\hline
\mathcal{T}_{1}\in\mathrmbf{Tbl}_{\mathcal{S}}(\mathcal{A}_{1})
\text{ and }
\mathcal{T}_{2}\in\mathrmbf{Tbl}_{\mathcal{S}}(\mathcal{A}_{2})
&
\textit{input}
\\
\mathcal{T}_{1}
\xrightarrow{\;{\langle{\tilde{g}_{1},\check{\iota}_{1}}\rangle}\;} 
\mathcal{T}_{1}{\,\oplus_{\mathcal{S}}}\mathcal{T}_{2}
\xleftarrow{\;{\langle{\tilde{g}_{2},\check{\iota}_{2}}\rangle}\;} 
\mathcal{T}_{2}
&
\textit{output boole-join}
\\
\cline{2-2}
\mathcal{T}_{1}
\xrightarrow{\hat{m}_{1}} 
\mathcal{T}_{1}{\oleft_{\mathcal{S}}}\mathcal{T}_{2}
\text{ \underline{or} }
\mathcal{T}_{1}{\oright_{\mathcal{S}}}\mathcal{T}_{2}
\xleftarrow{\hat{m}_{2}} 
\mathcal{T}_{2}
%}
&
\textit{output semi-join}
\\
%\cline{2-2}
%\mathcal{T}_{1}
%\xrightarrow{\tilde{m}_{1}} 
%\mathcal{T}_{1}{\,\succeq_{\mathcal{A}}}\mathcal{T}_{2}
%\text{ \underline{or} }
%\mathcal{T}_{1}{\,\preceq_{\mathcal{A}}}\mathcal{T}_{2}
%\xleftarrow{\tilde{m}_{2}} 
%\mathcal{T}_{2}
%}
%&
%\textit{output anti-join}
\cline{2-2}
\mathcal{T}_{1}
\xleftarrow{\bar{\omega}_{1}} 
\mathcal{T}_{1}{\,\oslash_{\mathcal{S}}}\mathcal{T}_{2}
\text{ \underline{or} }
\mathcal{T}_{1}{\,\obackslash_{\mathcal{S}}}\mathcal{T}_{2}
\xrightarrow{\bar{\omega}_{2}} 
\mathcal{T}_{2}
%}
&
\textit{output anti-join}
\end{array}$}}
\end{tabular}}}}
\end{center}
\caption{\texttt{FOLE} Data-type Join I/O}
\label{tbl:fole:boolean:join:input:output}
\end{table}
\begin{flushleft}
{\fbox{\fbox{\footnotesize{\begin{minipage}{345pt}
{\underline{\textsf{How does this work?}}}
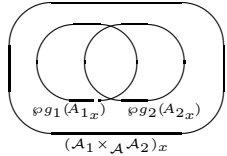
\begin{center}
{{\begin{tabular}{c}
{{\begin{tabular}{c}
\begin{picture}(100,50)(-35,-10)
\setlength{\unitlength}{0.45pt}
%\thicklines
%\put(106,40){\makebox(0,0){\normalsize{$\boldsymbol{\circ}$}}}
%\put(4.7,40){\makebox(0,0){\normalsize{$\boldsymbol{\circ}$}}}
%\put(70,40){\oval(40,40)[bl]}
%\put(70,40){\oval(40,40)[tl]}
%\put(70,40){\oval(40,40)[br]}
%\put(70,40){\oval(40,40)[tr]}
%\put(25,30){\oval(40,40)[bl]}
%\put(25,30){\oval(40,40)[tl]}
%\put(25,30){\oval(40,40)[br]}
%\put(25,30){\oval(40,40)[tr]}
\put(-20,5){\begin{picture}(20,20)(0,0)
\put(15,45){\oval(64,64)[bl]}
\put(15,45){\oval(64,64)[tl]}
\put(35,45){\oval(64,64)[br]}
\put(35,45){\oval(64,64)[tr]}
\put(15,77){\line(1,0){20}}
\put(10,13){\line(1,0){20}}
\put(10,4){\makebox(0,0){\tiny{${\wp}g_{1}(\!A_{1_{x}}\!)$}}}
\end{picture}}
%%%%%%%%%%%%%%%%%%%%%%%%%%%%%%%%%%%%%%%%%%%%%%%%%%%%%%%%%%%%
\put(20,5){\begin{picture}(20,20)(0,0)
\put(15,45){\oval(64,64)[bl]}
\put(15,45){\oval(64,64)[tl]}
\put(35,45){\oval(64,64)[br]}
\put(35,45){\oval(64,64)[tr]}
\put(15,77){\line(1,0){20}}
\put(15,13){\line(1,0){20}}
\put(50,4){\makebox(0,0){\tiny{${\wp}g_{2}(\!A_{2_{x}}\!)$}}}
\end{picture}}
%%%%%%%%%%%%%%%%%%%%%%%%%%%%%%%%%%%%%%%%%%%%%%%%%%%%%%%%%%%%
\put(-20,60){\oval(80,80)[tl]}
\put(-20,30){\oval(80,80)[bl]}
\put(80,60){\oval(80,80)[tr]}
\put(80,30){\oval(80,80)[br]}
\put(-60,30){\line(0,1){30}}
\put(120,30){\line(0,1){30}}
\put(-20,100){\line(1,0){100}}
\put(-20,-10){\line(1,0){100}}
\put(30,-20){\makebox(0,0){\tiny{$({\mathcal{A}_{1}{\times_{\mathcal{A}}}\mathcal{A}_{2}})_{x}$}}}
\end{picture}
\end{tabular}}}
\\
{{\begin{tabular}{c}
{\scriptsize{$
{\wp}g_{1}({A_{1}}_{x})
\subseteq
({\mathcal{A}_{1}{\times_{\mathcal{A}}}\mathcal{A}_{2}})_{x}
\supseteq 
{\wp}g_{2}({A_{2}}_{x})
$}}
\\
{{\textsf{inclusion of data-types}}}
\end{tabular}}}
\end{tabular}}}
\end{center}
An opspan of $\mathcal{S}$-type domains
\begin{center}
{\footnotesize{$
\mathcal{A}_{1}={\langle{X,Y_{1},\models_{\mathcal{A}_{1}}}\rangle}
\xrightarrow{{\langle{\mathrmit{1}_{X},g_{1}}\rangle}}
{\langle{X,Y,\models_{\mathcal{A}}}\rangle}=\mathcal{A}
\xleftarrow{{\langle{\mathrmit{1}_{X},g_{2}}\rangle}}
{\langle{X,Y_{2},\models_{\mathcal{A}_{2}}}\rangle}=\mathcal{A}_{2}$\;.}}
\end{center}
has
as its pullback the
span of $\mathcal{S}$-type domain morphisms.
%%\comment{
\newline\mbox{}\hfill
$\mathcal{A}_{1} = {\langle{X,Y_{1},\models_{\mathcal{A}_{1}}}\rangle}
\xleftarrow{{\langle{1_{X},\tilde{g}_{1}}\rangle}}
%{\mathcal{A}_{1}{\times_{\mathcal{A}}}\mathcal{A}_{2}}
\overset{\textstyle
{\mathcal{A}_{1}{\times_{\mathcal{A}}}\mathcal{A}_{2}}
}{\overbrace{{\langle{
X,Y_{1}{+}_{Y}Y_{2},\models_{{\mathcal{A}_{1}{\times_{\mathcal{A}}}\mathcal{A}_{2}}}
}\rangle}}} 
\xrightarrow{{\langle{1,\tilde{g}_{2}}\rangle}}
{\langle{X,Y_{2},\models_{\mathcal{A}_{2}}}\rangle}=\mathcal{A}_{2}$
\hfill\mbox{}\newline
%}
%
This consists of an opspan of data value functions
$Y_{1}\xrightarrow{\;\tilde{g}_{1}\;}Y_{1}{+}_{Y}Y_{2}\xleftarrow{\;\tilde{g}_{2}\;}Y_{2}$
satisfying the condition 
$\mathrmbfit{ext}_{{\mathcal{A}_{1}{\times_{\mathcal{A}}}\mathcal{A}_{2}}}
{\;\cdot\;}g_{1}^{-1}
= \mathrmbfit{ext}_{\mathcal{A}_{1}}$;
or that
%\newline\mbox{}\hfill
$g_{1}^{-1}({\mathcal{A}_{1}{\times_{\mathcal{A}}}\mathcal{A}_{2}})_{x} = ({A_{1}}_{x})$ for all $x \in X$.
%\hfill\mbox{}\newline
This implies that
$({\mathcal{A}_{1}{\times_{\mathcal{A}}}\mathcal{A}_{2}})_{x}
\supseteq {\wp}g_{1}({A_{1}}_{x})$ for all $x \in X$.
Hence,
for an injective data value function $Y_{1}{+}_{Y}Y_{2}\xhookleftarrow{\,g_{1}\,}Y_{1}$,
we have the inclusion
$({\mathcal{A}_{1}{\times_{\mathcal{A}}}\mathcal{A}_{2}})_{x} \supseteq {A_{1}}_{x}$ for all $x \in X$.
Same for an injective data value function 
$Y_{1}{+}_{Y}Y_{2}\xhookleftarrow{\,g_{2}\,}Y_{2}$.
Thus,
for every sort $x \in X$ the direct images
of data-types
${A_{1}}_{x}$
and
${A_{2}}_{x}$
are contained in data-type
$({\mathcal{A}_{1}{\times_{\mathcal{A}}}\mathcal{A}_{2}})_{x}$.
\end{minipage}}}}}
\end{flushleft}
%

%\begin{itemize}
%\item 
%Consider arbitrary diagrams of tables.
%\end{itemize}
%

%%%%%%%%%%%%%%%%%%%%%%%%%%%%%%%%%%%%%%%%%%%%%%%%%%%%%%%%%%%%%%
%
\newpage
\paragraph{Disjoint Sum.}
%%%%%%%%%%%%%%%%%%%%%%%%%%%%%%%%%%%%%%%%%%%%%%%%%%%%%%%%%%%%

%\begin{definition}
The disjoint sum 
is a special case of the data-type join.
%\begin{itemize}
%\item 
Let
$\mathcal{A}_{1}$
%\in
%\mathrmbf{Tbl}_{\mathcal{S}}(\mathcal{A}_{1})$
and
$\mathcal{A}_{2}$
% = {\langle{\mathcal{A}_{2}.K_{2},t_{2}}\rangle}
%\in\mathrmbf{Tbl}_{\mathcal{S}}(\mathcal{A}_{2})
be two $X$-type domains. 
These are linked by an opspan of $X$-type domains
{\footnotesize{$
\mathcal{A}_{1}
%={\langle{X,Y_{1},\models_{\mathcal{A}_{1}}}\rangle}
%\xrightarrow{{\langle{\mathrmit{1}_{X},0_{Y_{1}}}\rangle}}
\xrightarrow{0_{1}}
%{\langle{X,\emptyset,\models_{\mathcal{A}}}\rangle}=
\mathcal{A}_{\top}
\xleftarrow{0_{2}}
%\xleftarrow{{\langle{\mathrmit{1}_{X},0_{Y_{2}}}\rangle}}
%{\langle{X,Y_{2},\models_{\mathcal{A}_{2}}}\rangle}=
\mathcal{A}_{2}$}}
%\end{center}
%
with 
%the unique $X$-type domain morphisms
the terminal 
%(final) 
$X$-type domain 
$\mathcal{A}_{\top} = {\langle{X,\emptyset,\models_{\mathcal{A}}}\rangle}$
and injection data value functions
{\footnotesize{$Y_{1}\xhookleftarrow{0_{1}}\emptyset\xhookrightarrow{0_{2}}Y_{2}$\,.}}
This is the \underline{constraint} for disjoint sum
(Tbl.\,\ref{tbl:fole:disjoint:sum:input:output}).
It is a special case of the constraint for data-type join.
%\item 
The pullback (limiting cone) of this $X$-type domain opspan
is the product $X$-type domain 
$\mathcal{A}_{1}{\times}\mathcal{A}_{2}
={\langle{X,Y_{1}{+}Y_{2},\models_{\mathcal{A}_{1}{\times}\mathcal{A}_{1}}}\rangle}$ 
with 
the disjoint sum data value set $Y_{1}{+}Y_{2}$
and
projection
$X$-type domain morphisms (span)
{\footnotesize{$
\mathcal{A}_{1}
%={\langle{X,Y_{1},\models_{\mathcal{A}_{1}}}\rangle}
%\xleftarrow{{\langle{\mathrmit{1}_{X},i_{1}}\rangle}}
\xleftarrow{\,\tilde{0}_{1}}
\mathcal{A}_{1}{\times}\mathcal{A}_{2}
\xrightarrow{\tilde{0}_{2}\,}
%\xrightarrow{{\langle{\mathrmit{1}_{X},i_{2}}\rangle}}
%{\langle{X,Y_{2},\models_{\mathcal{A}_{2}}}\rangle}=
\mathcal{A}_{2}$}}
with inclusion data value functions 
{\footnotesize{$Y_{1}
\xhookrightarrow{\tilde{0}_{1}}
Y_{1}{+}Y_{2}
\xhookleftarrow{\tilde{0}_{2}}Y_{2}$.}}
This is the \underline{construction} for disjoint sum
(Tbl.\,\ref{tbl:fole:disjoint:sum:input:output}).
%\item 
Let
$\mathcal{T}_{1} = {\langle{K_{1},t_{1}}\rangle}
\in
\mathrmbf{Tbl}_{\mathcal{S}}(\mathcal{A}_{1})$
and
$\mathcal{T}_{2} = {\langle{K_{2},t_{2}}\rangle}
\in
\mathrmbf{Tbl}_{\mathcal{S}}(\mathcal{A}_{2})$
be two $\mathcal{S}$-tables 
with
key sets 
%$K_{1}$ and $K_{2}$
and a tuple functions
$K_{1}\xrightarrow{\,t_{1}\;}\mathrmbfit{tup}_{\mathcal{S}}(\mathcal{A}_{1})$
and
$K_{2}\xrightarrow{\,t_{2}\;}\mathrmbfit{tup}_{\mathcal{S}}(\mathcal{A}_{2})$.
This is the \underline{input} for disjoint sum
(Tbl.\,\ref{tbl:fole:disjoint:sum:input:output}).
%\item 
%In \texttt{FOLE},
The disjoint sum $\mathcal{T}_{1}{\,+\,}\mathcal{T}_{2}$ 
of the two $\mathcal{S}$-tables 
$\mathcal{T}_{1}$
and
$\mathcal{T}_{2}$
is a special case of data-type join
--- 
just link the tables through
the opspan of tuple functions
$\mathrmbf{tup}_{\mathcal{S}}(\mathcal{A}_{1})
{\;\xrightarrow{\mathrmbfit{tbl}_{\mathcal{S}}(\tilde{0}_{1})}\;}
\mathrmbf{tup}_{\mathcal{S}}(\mathcal{A}_{1}{\times}\mathcal{A}_{2})
{\;\xleftarrow{\mathrmbfit{tbl}_{\mathcal{S}}(\tilde{0}_{2})}\;}
\mathrmbf{tup}_{\mathcal{S}}(\mathcal{A}_{2})$,
and then use expansion (twice) and union.
%\item 
%
The disjoint sum table $\mathcal{T}_{1}{\,+\,}\mathcal{T}_{2}$, 
which has the binary sum (disjoint union) key set $K_{1}{\,+\,}K_{2}$
with co-mediating tuple function
%
%%%%%%%%%%%%%%%%%%%%%%%%%%%%%%%%%%%%%%%%%%%%%%%%%%%%%%%%%%%%%%%%%%%%%%%%%%%%%%%%
%%%%%%%%%%%%%%%%%%%%%%%%%%%%%%%%%%%%%%%%%%%%%%%%%%%%%%%%%%%%%%%%%%%%%%%%%%%%%%%%
\footnote{The tuple subset
of the disjoint sum table $\mathcal{T}_{1}{\,+\,}\mathcal{T}_{2}$ 
is the disjoint union of the tuple sets
${\wp{t_{1}}}(K_{1}) \subseteq 
\mathrmbfit{tup}_{\mathcal{S}}(\mathcal{A}_{1})
\subseteq \mathrmbf{List}(Y_{1})$
and
${\wp{t_{2}}}(K_{2}) 
\subseteq \mathrmbfit{tup}_{\mathcal{S}}(\mathcal{A}_{2})
 \subseteq \mathrmbf{List}(Y_{2})$.}
%%%%%%%%%%%%%%%%%%%%%%%%%%%%%%%%%%%%%%%%%%%%%%%%%%%%%%%%%%%%%%%%%%%%%%%%%%%%%%%%
%%%%%%%%%%%%%%%%%%%%%%%%%%%%%%%%%%%%%%%%%%%%%%%%%%%%%%%%%%%%%%%%%%%%%%%%%%%%%%%%
%
%\[\mbox
{\footnotesize{
{$
K_{1}{\,+\,}K_{2} 
\xrightarrow{
[t_{1}{\cdot}\mathrmbfit{tbl}_{\mathcal{S}}(\tilde{0}_{1}),
t_{1}{\cdot}\mathrmbfit{tbl}_{\mathcal{S}}(\tilde{0}_{2})]} 
\mathrmbfit{tup}_{\mathcal{A}}(\mathcal{S}_{1})
{\,\times\,}\mathrmbfit{tup}_{\mathcal{A}}(\mathcal{S}_{2})$,}}\normalsize}
%\]
%
%The disjoint union 
is linked to the component tables with the 
opspan of 
%projection 
table morphisms
\[\mbox
{\footnotesize{
{$\mathcal{T}_{1}
\xrightarrow{\;{\langle{\tilde{0}_{1},\check{\iota}_{1}}\rangle}\;} 
\mathcal{T}_{1}{\,+\,}\mathcal{T}_{2}
\xleftarrow{\;{\langle{\tilde{0}_{2},\check{\iota}_{2}}\rangle}\;} 
\mathcal{T}_{2}$.}}\normalsize}
\]
This is the \underline{output} for disjoint sum
(Tbl.\,\ref{tbl:fole:disjoint:sum:input:output}).
%\end{itemize}
%

%
%\end{definition}

%
\begin{table}
\begin{center}
{{\fbox{\begin{tabular}{c}
\setlength{\extrarowheight}{2pt}
{\scriptsize{$\begin{array}[c]{c@{\hspace{12pt}}l}
\mathcal{A}_{1}  
\text{ and }
\mathcal{A}_{2} 
&
\textit{constraint}
\\
\mathcal{A}_{1}\xleftarrow{\,\tilde{0}_{1}}
\mathcal{A}_{1}{\times}\mathcal{A}_{2}
\xrightarrow{\tilde{0}_{2}\,}\mathcal{A}_{2}
&
\textit{construction}
\\
\hline
\mathcal{T}_{1}\in\mathrmbf{Tbl}_{\mathcal{S}}(\mathcal{A}_{1})
\text{ and }
\mathcal{T}_{2}\in\mathrmbf{Tbl}_{\mathcal{S}}(\mathcal{A}_{2})
&
\textit{input}
\\
\mathcal{T}_{1}
\xrightarrow{\;{\langle{\tilde{0}_{1},\check{\iota}_{1}}\rangle}\;} 
\mathcal{T}_{1}{\,+\,}\mathcal{T}_{2}
\xleftarrow{\;{\langle{\tilde{0}_{2},\check{\iota}_{2}}\rangle}\;} 
\mathcal{T}_{2}
&
\textit{output}
\end{array}$}}
\end{tabular}}}}
\end{center}
\caption{\texttt{FOLE} Disjoint Sum I/O}
\label{tbl:fole:disjoint:sum:input:output}
\end{table}

%%%%%%%%%%%%%%%%%%%%%%%%%%%%%%%%%%%%%%%%%%%%%%%%%%%%%%%%%%%%%
\newpage
\subsubsection{Data-type Semi-join.}
\label{sub:sub:sec:boole:semi:join}
%%%%%%%%%%%%%%%%%%%%%%%%%%%%%%%%%%%%%%%%%%%%%%%%%%%%%%%%%%%%%

%
%%%%%%%%%%%%%%%%%%%%%%%%%%%%%%%%%%%%%%%%%%%%%%%%%%%%%%%%%%%%
%\newpage
%\paragraph{Data-type Semi-join.}
%%%%%%%%%%%%%%%%%%%%%%%%%%%%%%%%%%%%%%%%%%%%%%%%%%%%%%%%%%%%
%

%
\begin{figure}
\begin{center}
{{{\begin{tabular}{c}
\begin{picture}(150,40)(40,35)
\setlength{\unitlength}{0.97pt}
%%%%%%%%%%%%%%%%%%%%%%%%%%%%%%%%%%%%%%%%%%%%%%%%%%
\put(98.4,58){\begin{picture}(0,0)(0,3)
\setlength{\unitlength}{0.35pt}
\put(60,23){\makebox(0,0){\normalsize{$\oplus$}}}
%\thicklines
%\put(33,76){\makebox(0,0){\normalsize{$\boldsymbol{\circ}$}}}
%\put(87,76){\makebox(0,0){\normalsize{$\boldsymbol{\circ}$}}}
%\put(60,3){\makebox(0,0){\normalsize{$\boldsymbol{\circ}$}}}
\put(40,10){\line(1,0){40}}
\put(10,70){\line(1,0){100}}
\put(10,70){\line(0,-1){30}}
\put(110,70){\line(0,-1){30}}
\put(40,40){\oval(60,60)[bl]}
\put(80,40){\oval(60,60)[br]}
\put(61,58){\makebox(0,0){\scriptsize{{\textit{{data-type}}}}}}
\put(61,41){\makebox(0,0){\scriptsize{{\textit{{join}}}}}}
\end{picture}}
%%%%%%%%%%%%%%%%%%%%%%%%%%%%%%%%%%%%%%%%%%%%%%%%%%
\put(53.5,25.5){\begin{picture}(0,0)(0,0)
\setlength{\unitlength}{0.35pt}
%\thicklines
%\put(106,40){\makebox(0,0){\normalsize{$\boldsymbol{\circ}$}}}
%\put(4.7,40){\makebox(0,0){\normalsize{$\boldsymbol{\circ}$}}}
\put(40,10){\line(1,0){60}}
\put(40,70){\line(1,0){60}}
\put(100,70){\line(0,-1){60}}
\put(40,40){\oval(60,60)[bl]}
\put(40,40){\oval(60,60)[tl]}
\put(58,51){\makebox(0,0){\scriptsize{{\textit{{restrict}}}}}}
\put(56,30){\makebox(0,0){\Large{${\Leftarrow}$}}}
%%\put(61,22){\makebox(0,0){\scriptsize{{\textit{{image}}}}}}
\end{picture}}
%%%%%%%%%%%%%%%%%%%%%%%%%%%%%%%%%%%%%%%%%%%%%%%%%%
\put(146,25.5){\begin{picture}(0,0)(0,0)
\setlength{\unitlength}{0.35pt}
%\put(0,40){\line(1,0){130}}
%\thicklines
%\put(106,40){\makebox(0,0){\normalsize{$\boldsymbol{\circ}$}}}
%\put(5.8,40){\makebox(0,0){\normalsize{$\boldsymbol{\circ}$}}}
\put(10,10){\line(1,0){61}}
\put(10,70){\line(1,0){61}}
\put(11,70){\line(0,-1){60}}
\put(70,40){\oval(60,60)[br]}
\put(70,40){\oval(60,60)[tr]}
\put(55,51){\makebox(0,0){\scriptsize{{\textit{{restrict}}}}}}
\put(56,30){\makebox(0,0){\Large{${\Rightarrow}$}}}
%\put(56,22){\makebox(0,0){\scriptsize{{\textit{{image}}}}}}
\end{picture}}
%%%%%%%%%%%%%%%%%%%%%%%%%%%%%%%%%%%%%%%%%%%%%%%%%%
\put(120,33){\makebox(0,0){\footnotesize{{\textit{{data-type}}}}}}
\put(120,24){\makebox(0,0){\footnotesize{{\textit{{semi-join}}}}}}
\put(48,32){\makebox(0,0){\footnotesize{{\textit{{left}}}}}}
\put(48,46){\makebox(0,0){\footnotesize{$\oleft$}}}
\put(192,46){\makebox(0,0){\footnotesize{$\oright$}}}
\put(193,32){\makebox(0,0){\footnotesize{{\textit{{right}}}}}}
%\put(70,18){\makebox(0,0){\footnotesize{{\textit{{left}}}}}}
%\put(86,18){\makebox(0,0){\large{$\rhd$}}}
%\put(154,18){\makebox(0,0){\large{$\lhd$}}}
%\put(172,18){\makebox(0,0){\footnotesize{{\textit{{right}}}}}}
%%%%%%%%%%%%%%%%%%%%%%%%%%%%%%%%%%%%%%%%%%%%%%%%%%
\put(57,40){\vector(-1,0){20}}
\put(120,58.5){\line(0,-1){18.5}}
\put(120,40){\vector(-1,0){30}}
\put(120,40){\vector(1,0){30}}
\put(182,40){\vector(1,0){20}}
%\put(234,40){\vector(1,0){20}}
%%%%%%%%%%
\put(110,92){\vector(0,-1){12}}
\put(130,92){\vector(0,-1){12}}
%\thicklines
%\put(15,110){\line(1,0){210}}
%\put(55,10){\line(1,0){130}}
%\put(15,50){\line(0,1){60}}
%\put(225,50){\line(0,1){60}}
%\qbezier(15,50)(15,10)(55,10)
%\qbezier(185,10)(225,10)(225,50)
%%%%%%%%%%%%%%%%%%%%%%%%%%%%%%%%%%%%%%%%%%%%%%%%%%
\end{picture}
\end{tabular}}}}
\end{center}
\caption{\texttt{FOLE} data-type Semi-Join Flow Chart}
\label{fole:boole:semi:join:flo:chrt}
\end{figure}
%
%{\fbox{\textbf{Lets try to reduce this explanation as much as possible.}}}
%\newline
%The left and right data-type semi-join are two other data-type join-like operations.
%We make the same assumptions as in \S\,\ref{sub:sub:sec:boole:join}
%on the data-type join.
%
Let $\mathcal{S}$
% = {\langle{I,x,X}\rangle}$
be a fixed signature.
For any two $\mathcal{S}$-tables 
$\mathcal{T}_{1} \in \mathrmbf{Tbl}_{\mathcal{S}}(\mathcal{A}_{1})$ and 
$\mathcal{T}_{2} \in \mathrmbf{Tbl}_{\mathcal{S}}(\mathcal{A}_{2})$
that are linked through 
\comment{
%an $\mathcal{S}$-table span
%
%\newline\mbox{}\hfill\rule[-10pt]{0pt}{26pt}
%{\footnotesize{$
%%\begin{align}\label{tbl:opspan:redux}
%\mathcal{T}_{1}={\langle{\mathcal{A}_{1},K_{1},t_{1}}\rangle} 
%\xleftarrow{\langle{g_{1},k_{1}}\rangle} 
%%\overset{\textstyle
%\mathcal{T}
%%{\overbrace{\langle{I'',s'',K'',t''}\rangle}}
%\xrightarrow{\langle{g_{2},k_{2}}\rangle} {\langle{\mathcal{A}_{2},K_{2},t_{2}}\rangle}=\mathcal{T}_{2}
%\end{align}
%$}\normalsize}
%\hfill\mbox{}
%\newline
%with 
%%(\S\,\ref{sub:sub:sec:tbl})
%key span
%$K_{1} \xleftarrow{k_{1}} K \xrightarrow{k_{2}} K_{2}$
%and 
}
an $X$-sorted type domain opspan 
$\mathcal{A}_{1}
%={\langle{I_{1},s_{1}}\rangle} 
\xrightarrow{g_{1}} 
%\overset{
\mathcal{A}
%}{\mathcal{S}} 
\xleftarrow{g_{2}} 
%{\langle{I_{2},s_{2}}\rangle} = 
\mathcal{A}_{2}$
in $\mathrmbf{Cls}(X)$,
\comment{The limiting cone of this type domain opspan, 
has pullback type domain
$\mathcal{A}_{1}{\times_{\!\mathcal{A}\,}}\mathcal{A}_{2}$
%={\langle{I{+}_{I}I_{2},[s,s_{2}]}\rangle}$
and 
%surjective 
$X$-type domain morphisms
(span)
{\footnotesize{$
\mathcal{A}_{1} 
%= {\langle{I_{1},s_{1}}\rangle} 
\xleftarrow{\;\tilde{g}_{1}\,} 
%\overset{\textstyle
{\mathcal{A}_{1}{\times_{\mathcal{A}}}\mathcal{A}_{2}}
%}{\overbrace{{\langle{I_{1}{+}_{I}I_{2},[s_{1},s_{2}]}\rangle}}} 
\xrightarrow{\;\tilde{g}_{2}} 
%{\langle{I_{2},s_{2}}\rangle} = 
\mathcal{A}_{2}.
$}\normalsize}
Data-type join within the context $\mathrmbf{Tbl}(\mathcal{S})$
is expansion followed by join.
This is the two-step process 
\newline\mbox{}\hfill
\rule[-10pt]{0pt}{26pt}
$\mathcal{T}_{1}{\,\oplus_{\mathcal{S}}}\mathcal{T}_{2} 
\doteq  
\acute{\mathrmbfit{tbl}}_{\mathcal{S}}(\tilde{g}_{1})(\mathcal{T}_{1})
{\;\vee\;}
\acute{\mathrmbfit{tbl}}_{\mathcal{S}}(\tilde{g}_{2})(\mathcal{T}_{2})$.
\hfill\mbox{}\newline
}
the left data-type semi-join $\mathcal{T}_{1}{\;\oleft_{\mathcal{S}}\,}\mathcal{T}_{2}$
is 
the set of all 
%tuples in the data-type join
%${\mathcal{T}_{1}}{\,\bigcup_{\mathcal{T}}}{\mathcal{T}_{2}}$
%with data values in $Y_{1}$;
%or more simply, all 
tuples in $\mathcal{T}_{1}$ plus 
tuples in $\mathcal{T}_{2}$ with data values in $Y_{1}$.
%
%%%%%%%%%%%%%%%%%%%%%%%%%%%%%%%%%%%%%%%%%%%%%%%%%%%%%%%%%%%%%%%%%%%%%%
%%%%%%%%%%%%%%%%%%%%%%%%%%%%%%%%%%%%%%%%%%%%%%%%%%%%%%%%%%%%%%%%%%%%%%
\comment{This is the union of the tuples in $\mathcal{T}_{1}$
with the restriction of the expansion 
(through the pullback type domain) 
of the tuples in $\mathcal{T}_{2}$.}
%%%%%%%%%%%%%%%%%%%%%%%%%%%%%%%%%%%%%%%%%%%%%%%%%%%%%%%%%%%%%%%%%%%%%%
%%%%%%%%%%%%%%%%%%%%%%%%%%%%%%%%%%%%%%%%%%%%%%%%%%%%%%%%%%%%%%%%%%%%%%
%
Hence,
the left data-type semi-join is defined to be the restriction on the data-type join.
The right data-type semi-join is similar.
We use the following routes of flow from 
%Tbl.\;\ref{tbl:routes:flow}.
Fig.\;\ref{fig:routes:flow:colim}.
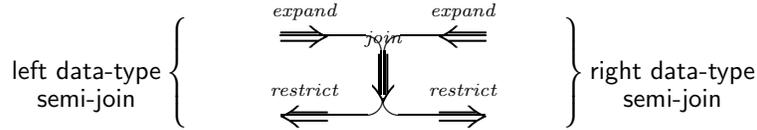
\begin{center}
{{\begin{tabular}{c}
\setlength{\unitlength}{0.6pt}
\begin{picture}(320,80)(0,-10)
%\put(157,63){\makebox(0,0){\normalsize{$\textbf{2.}$}}}
%
\put(325,25){\makebox(0,0){\normalsize{
$\left.\rule{0pt}{24pt}\right\}
\underset{\textstyle{\textsf{semi-join}}}{\textsf{right data-type}}$}}}
\put(-25,25){\makebox(0,0){\normalsize{
$\underset{\textstyle{\textsf{semi-join}}}{\textsf{left data-type}}
\left\{\rule{0pt}{24pt}\right.$}}}
\put(100,55){\makebox(0,0){\huge{
$\overset{\textit{\scriptsize{expand}}}{\Rightarrow}$}}}
\put(200,55){\makebox(0,0){\huge{
$\overset{\textit{\scriptsize{expand}}}{\Leftarrow}$}}}
\put(150,30){\makebox(0,0){\huge{
$\overset{\textit{\scriptsize{join}}}{\Downarrow}$}}}
\put(100,4.5){\makebox(0,0){\huge{
$\overset{\textit{\scriptsize{restrict}}}{\Leftarrow}$}}}
%\put(105,-14.5){\makebox(0,0){$\textit{\scriptsize{image}}$}}
\put(200,4.5){\makebox(0,0){\huge{
$\overset{\textit{\scriptsize{restrict}}}{\Rightarrow}$}}}
%\put(205,-14.5){\makebox(0,0){$\textit{\scriptsize{image}}$}}
%
\put(90,48){\line(1,0){54}}
\put(144,38){\oval(20,20)[tr]}
\put(166,38){\oval(20,20)[tl]}
\put(220,48){\line(-1,0){54}}
\put(145,-2){\line(-1,0){54}}
\put(165,-2){\line(1,0){54}}
\put(154,11){\line(0,1){27}}
\put(156,11){\line(0,1){27}}
\put(145,8){\oval(20,20)[br]}
\put(165,8){\oval(20,20)[bl]}
\end{picture}
\end{tabular}}}
\end{center}
The constraint, construction and input for data-type semi-join 
are identical to that for natural join. 
Only the output is different.
\begin{description}
\item[Constraint:] 
The constraint for data-type semi-join is the same as the constraint for data-type join
(Tbl.\,\ref{tbl:fole:boolean:join:input:output}):
an $X$-sorted type domain opspan 
$\mathcal{A}_{1}\xrightarrow{g_{1}}\mathcal{A}\xleftarrow{g_{2}}\mathcal{A}_{2}$
in $\mathrmbf{Cls}(X)$.
\newline
\item[Construction:] 
The construction for data-type semi-join is the same as the construction for data-type join 
(Tbl.\,\ref{tbl:fole:boolean:join:input:output}):
the span
{\footnotesize{$
\mathcal{A}_{1} 
%= {\langle{I_{1},s_{1}}\rangle} 
\xleftarrow{\;\tilde{g}_{1}\,} 
%\overset{\textstyle
{\mathcal{A}_{1}{\times_{\mathcal{A}}}\mathcal{A}_{2}}
%}{\overbrace{{\langle{I_{1}{+}_{I}I_{2},[s_{1},s_{2}]}\rangle}}} 
\xrightarrow{\;\tilde{g}_{2}} 
%{\langle{I_{2},s_{2}}\rangle} = 
\mathcal{A}_{2}
$}\normalsize}
of projection $X$-type domain morphisms
with pullback type domain
${\mathcal{A}_{1}{\times_{\mathcal{A}}}\mathcal{A}_{2}}$.
%and
%data value function opspan
%{\footnotesize{$Y_{1} 
%\xrightarrow{\tilde{g}_{1}\,} 
%{\langle{Y_{1}{+}_{Y}Y_{2},[g_{1},g_{2}]}\rangle}
%\xleftarrow{\;\tilde{g}_{2}}
%Y_{2}.$}\normalsize}
%
\newline
\item[Input:] 
The input for data-type semi-join is the same as the input for data-type join 
(Tbl.\,\ref{tbl:fole:boolean:join:input:output}):
two tables
$\mathcal{T}_{1} = {\langle{K_{1},t_{1}}\rangle} \in 
\mathrmbf{Tbl}_{\mathcal{S}}(\mathcal{A}_{1})$
and
$\mathcal{T}_{2} = {\langle{K_{2},t_{2}}\rangle} \in 
\mathrmbf{Tbl}_{\mathcal{S}}(\mathcal{A}_{2})$.
\newpage
\item[Output:] 
The output is data-type join followed by restriction.
\newline
\begin{itemize}
\item 
Data-type join results in the table
%$\mathcal{T}_{1}{\times_{\mathcal{T}}}\mathcal{T}_{2}$
$\mathcal{T}_{1}{\,\oplus_{\mathcal{S}}}\mathcal{T}_{2} = 
\acute{\mathrmbfit{tbl}}_{\mathcal{S}}(\tilde{g}_{1})(\mathcal{T}_{1})
{\;\vee\;}
\acute{\mathrmbfit{tbl}}_{\mathcal{S}}(\tilde{g}_{2})(\mathcal{T}_{2})$
with key set $K_{1}{+}K_{2}$
and tuple function
$K_{1}{+}K_{2}\xrightarrow{[\tilde{t}_{1},\tilde{t}_{2}]}
\mathrmbfit{tup}_{\mathcal{A}}(\mathcal{S})$.
\newline
\item 
Restriction 
{\footnotesize{$
\mathrmbf{Tbl}_{\mathcal{S}}(\mathcal{A}_{1})
{\;\xleftarrow
%[{\tilde{g}_{1}}^{\ast}]
{\;\grave{\mathrmbfit{tbl}}_{\mathcal{S}}(\tilde{g}_{1})\;}\;}
\mathrmbf{Tbl}_{\mathcal{S}}(\mathcal{A}_{1}{\times_{\mathcal{A}}}\mathcal{A}_{2})
$}\normalsize}
(\S\,\ref{sub:sub:sec:adj:flow:S})
along the tuple function
of the $X$-type domain morphism
{\footnotesize{$\mathcal{A}_{1} 
\xleftarrow{\;\tilde{g}_{1}\,} 
{\mathcal{A}_{1}{\times_{\mathcal{A}}}\mathcal{A}_{2}}
$}\normalsize}
maps the data-type join table
$\mathcal{T}_{1}{\,\oplus_{\mathcal{S}}}\mathcal{T}_{2}$
to the left semi-join table
${\mathcal{T}_{1}}{\,\oleft_{\mathcal{S}}}{\mathcal{T}_{2}} =
\grave{\mathrmbfit{tbl}}_{\mathcal{S}}(\tilde{g}_{1})
(\mathcal{T}_{1}{\,\oplus_{\mathcal{S}}}\mathcal{T}_{2}) = 
%{\tilde{g}_{1}}^{\ast}(\mathcal{T}_{1}{\,\cup_{\mathcal{T}}}\mathcal{T}_{2}) =
{\langle{\widehat{K},\hat{t}_{1}}\rangle}$ 
with key set $\widehat{K}_{1}$  and
tuple function
$\widehat{K}_{1}\xrightarrow{\hat{t}_{1}}\mathrmbfit{tup}_{\mathcal{S}}(\mathcal{A}_{1})$
defined by pullback. 
\newline
\end{itemize}
%
%The data-type semi-join flowchart input/output is displayed in 
%Tbl\;\ref{tbl:fole:boolean:join:input:output}.
Data-type semi-join is data-type join followed by restriction. 
For left semi-join, 
this is the two-step process
\[\mbox
{\footnotesize{{$
{\mathcal{T}_{1}}{\,\oleft_{\mathcal{S}}}{\mathcal{T}_{2}}
\doteq 
\grave{\mathrmbfit{tbl}}_{\mathcal{S}}(\tilde{g}_{1})
\bigl(\mathcal{T}_{1}{\,\oplus_{\mathcal{S}}}
\mathcal{T}_{2}\bigr)
$.}}\normalsize}
%
%%%%%%%%%%%%%%%%%%%%%%%%%%%%%%%%%%%%%%%%%%%%%%%%%%%%%%%%%%%%%%%%%%%%%%%%%%%%%%%%
%%%%%%%%%%%%%%%%%%%%%%%%%%%%%%%%%%%%%%%%%%%%%%%%%%%%%%%%%%%%%%%%%%%%%%%%%%%%%%%%
\footnote{The data-type semi-join of a disjoint sum gives one of the tables.}
%%%%%%%%%%%%%%%%%%%%%%%%%%%%%%%%%%%%%%%%%%%%%%%%%%%%%%%%%%%%%%%%%%%%%%%%%%%%%%%%
%%%%%%%%%%%%%%%%%%%%%%%%%%%%%%%%%%%%%%%%%%%%%%%%%%%%%%%%%%%%%%%%%%%%%%%%%%%%%%%%
%
\]
%\begin{itemize}
%\item 
This defines the table morphism
%\newline\mbox{}\hfill
{\footnotesize{{$
{\mathcal{T}_{1}}{\,\oleft_{\mathcal{S}}}{\mathcal{T}_{2}}
%= {\langle{\widehat{K},\hat{t}_{1}}\rangle}
\xrightarrow{{\langle{\tilde{g}_{1},\hat{k}_{1}}\rangle}} 
%{\langle{\mathcal{A}_{1}{{\times}_{\mathcal{S}}}\mathcal{A}_{2},
%K_{1}{+}_{K}K_{2},[\tilde{t}_{1},\tilde{t}_{2}]}\rangle} = 
\mathcal{T}_{1}{\,\oplus_{\mathcal{S}}}\mathcal{T}_{2}
$.}}\normalsize}
%\hfill\mbox{}\newline
in $\mathrmbf{Tbl}_{\mathcal{S}}(\mathcal{A}_{1})$.
%\item 
There is a sub-table relationship 
{\footnotesize{{$
\mathcal{T}_{1}
%%= {\langle{K_{1},t_{1}}\rangle}
\xrightarrow{\hat{m}_{1}} 
%{\langle{\widehat{K},\hat{t}_{1}}\rangle} =
%\widehat{\mathcal{T}}_{1} =
\mathcal{T}_{1}{\;\oleft_{\mathcal{S}}\,}\mathcal{T}_{2}
$}}\normalsize}
in the small fiber table context
$\mathrmbf{Tbl}_{\mathcal{S}}(\mathcal{A}_{1})$,
which is the (left-side) output of data-type semi-join. 
%\end{itemize}
%
The right data-type semi-join has a similar definition
(Tbl.\,\ref{tbl:fole:boolean:join:input:output}).
\end{description}
These factor 
%\[\mbox
{\footnotesize{
{$
\mathcal{T}_{1}
\xrightarrow{\hat{m}_{1}} 
\mathcal{T}_{1}{\;\oleft_{\mathcal{S}}\,}\mathcal{T}_{2}
\xrightarrow{{\langle{\tilde{g}_{1},\hat{k}_{1}}\rangle}} 
\mathcal{T}_{1}{\,\oplus_{\mathcal{S}}}\mathcal{T}_{2}
\xleftarrow{{\langle{\tilde{g}_{2},\hat{k}_{2}}\rangle}} 
\mathcal{T}_{1}{\;\oright_{\mathcal{S}}\,}\mathcal{T}_{2}
\xleftarrow{\hat{m}_{2}}
\mathcal{T}_{2}
$,}}\normalsize}
%\]
%
the opspan of $\mathcal{S}$-table morphisms
%\[\mbox
{\footnotesize{
{$
\mathcal{T}_{1}
\xrightarrow[\;\hat{m}_{1}{\circ\,}{\langle{\tilde{g}_{1},\hat{k}_{1}}\rangle}\;]
{\;{\langle{\tilde{g}_{1},\tilde{k}_{1}}\rangle}\;} 
\mathcal{T}_{1}{\,\oplus_{\mathcal{S}}}\mathcal{T}_{2}
\xleftarrow[\;\hat{m}_{2}{\circ\,}{\langle{\tilde{g}_{2},\hat{k}_{2}}\rangle}\;]
{\;{\langle{\tilde{g}_{2},\tilde{k}_{2}}\rangle}\;} 
\mathcal{T}_{2}
$,}}\normalsize}
%\]
%
%which is 
the output for data-type join. 
%
%%%%%%%%%%%%%%%%%%%%%%%%%%%%%%%%%%%%%%%%%%%%%%%%%%%%%%%%%%%%%%%%%%%%%%
%%%%%%%%%%%%%%%%%%%%%%%%%%%%%%%%%%%%%%%%%%%%%%%%%%%%%%%%%%%%%%%%%%%%%%
{\footnote{The data-type semi-join of a disjoint sum of non-empty tables 
gives either of the tables: the left data-type semi-join gives the left table,
and the right gives the right.}}
\newpage
\subsubsection{Data-type Anti-join.}
\label{sub:sub:sec:boole:anti:join}
%%%%%%%%%%%%%%%%%%%%%%%%%%%%%%%%%%%%%%%%%%%%%%%%%%%%%%%%%%%%%
%
\begin{figure}
\begin{center}
{{{\begin{tabular}{c}
\begin{picture}(130,80)(50,15)
\setlength{\unitlength}{0.97pt}
%%%%%%%%%%%%%%%%%%%%%%%%%%%%%%%%%%%%%%%%%%%%%%%%%%
\put(98,60){\begin{picture}(0,0)(0,3)
\setlength{\unitlength}{0.35pt}
\put(42,24){\makebox(0,0){\normalsize{$\oleft$}}}
\put(78,24){\makebox(0,0){\normalsize{$\oright$}}}
%\thicklines
%\put(33,76){\makebox(0,0){\normalsize{$\boldsymbol{\circ}$}}}
%\put(87,76){\makebox(0,0){\normalsize{$\boldsymbol{\circ}$}}}
%\put(40,3){\makebox(0,0){\normalsize{$\boldsymbol{\circ}$}}}
%\put(80,3){\makebox(0,0){\normalsize{$\boldsymbol{\circ}$}}}
\put(40,10){\line(1,0){40}}
\put(10,70){\line(1,0){100}}
\put(10,70){\line(0,-1){30}}
\put(110,70){\line(0,-1){30}}
\put(40,40){\oval(60,60)[bl]}
\put(80,40){\oval(60,60)[br]}
%\put(61,52){\makebox(0,0){\scriptsize{{\textit{{left}}}}}}
\put(61,59){\makebox(0,0){\scriptsize{{\textit{{data-type}}}}}}
\put(61,42){\makebox(0,0){\scriptsize{{\textit{{semi-join}}}}}}
\end{picture}}
%%%%%%%%%%%%%%%%%%%%%%%%%%%%%%%%%%%%%%%%%%%%%%%%%%
\put(44,21.3){\begin{picture}(0,0)(0,3)
\setlength{\unitlength}{0.35pt}
%\put(60,40){\makebox(0,0){\large{$-$}}}
%\thicklines
\put(50,33){\line(1,0){20}}
%\put(33,76){\makebox(0,0){\normalsize{$\boldsymbol{\circ}$}}}
%\put(87,76){\makebox(0,0){\normalsize{$\boldsymbol{\circ}$}}}
%\put(60,3){\makebox(0,0){\normalsize{$\boldsymbol{\circ}$}}}
\put(40,10){\line(1,0){40}}
\put(10,70){\line(1,0){100}}
\put(10,70){\line(0,-1){30}}
\put(110,70){\line(0,-1){30}}
\put(40,40){\oval(60,60)[bl]}
\put(80,40){\oval(60,60)[br]}
\put(61,50){\makebox(0,0){\scriptsize{{\textit{{diff}}}}}}
\end{picture}}
%%%%%%%%%%%%%%%%%%%%%%%%%%%%%%%%%%%%%%%%%%%%%%%%%%
\put(151,21.3){\begin{picture}(0,0)(0,3)
\setlength{\unitlength}{0.35pt}
%\put(60,40){\makebox(0,0){\large{$-$}}}
%\thicklines
\put(50,33){\line(1,0){20}}
%\put(33,76){\makebox(0,0){\normalsize{$\boldsymbol{\circ}$}}}
%\put(87,76){\makebox(0,0){\normalsize{$\boldsymbol{\circ}$}}}
%\put(60,3){\makebox(0,0){\normalsize{$\boldsymbol{\circ}$}}}
\put(40,10){\line(1,0){40}}
\put(10,70){\line(1,0){100}}
\put(10,70){\line(0,-1){30}}
\put(110,70){\line(0,-1){30}}
\put(40,40){\oval(60,60)[bl]}
\put(80,40){\oval(60,60)[br]}
\put(61,50){\makebox(0,0){\scriptsize{{\textit{{diff}}}}}}
\end{picture}}
%%%%%%%%%%%%%%%%%%%%%%%%%%%%%%%%%%%%%%%%%%%%%%%%%%
\put(120,35){\makebox(0,0){\footnotesize{{\textit{{data-type}}}}}}
\put(120,25){\makebox(0,0){\footnotesize{{\textit{{anti-join}}}}}}
\put(72,14){\makebox(0,0)[l]{\footnotesize{{\textit{{left}}}}}}
\put(58,14){\makebox(0,0){{$\oslash$}}}
\put(184,14){\makebox(0,0){{$\obackslash$}}}
\put(170,14){\makebox(0,0)[r]{\footnotesize{{\textit{{right}}}}}}
%\put(70,18){\makebox(0,0){\footnotesize{{\textit{{left}}}}}}
%\put(86,18){\makebox(0,0){\large{$\rhd$}}}
%\put(154,18){\makebox(0,0){\large{$\lhd$}}}
%\put(172,18){\makebox(0,0){\footnotesize{{\textit{{right}}}}}}
%%%%%%%%%%%%%%%%%%%%%%%%%%%%%%%%%%%%%%%%%%%%%%%%%%
%\put(49,40){\vector(-1,0){20}}
%\put(191,40){\vector(1,0){20}}
\put(112.7,60.3){\line(0,-1){8.5}}
\put(112.7,52){\line(-1,0){36}}
\put(77,46){\line(0,1){6}}
\put(77,43){\vector(0,-1){0}}
\put(127,60.3){\line(0,-1){8.5}}
\put(127,52){\line(1,0){36}}
\put(163,46){\line(0,1){6}}
\put(163,43){\vector(0,-1){0}}
\put(109.7,102.2){\vector(0,-1){20}}
\put(129.3,102.2){\vector(0,-1){20}}
\put(110,97){\line(-1,0){54}}
\put(56,97){\vector(0,-1){53.5}}
\put(129,97){\line(1,0){52}}
\put(181,97){\vector(0,-1){53.5}}
\put(66.6,22){\vector(0,-1){16}}
\put(173.6,22){\vector(0,-1){16}}
%\thicklines
%\put(15,110){\line(1,0){210}}
%\put(55,0){\line(1,0){130}}
%\put(15,40){\line(0,1){70}}
%\put(225,40){\line(0,1){70}}
%\qbezier(15,40)(15,0)(55,0)
%\qbezier(185,0)(225,0)(225,40)
%%%%%%%%%%%%%%%%%%%%%%%%%%%%%%%%%%%%%%%%%%%%%%%%%%
\end{picture}
\end{tabular}}}}
\end{center}
\caption{\texttt{FOLE} Data-type Anti-Join Flow Chart}
\label{fole:boole:anti:join:flo:chrt}
\end{figure}
The data-type anti-join operations 
are related to the data-type semi-join operations.
The left (right) data-type anti-join of two tables 
is the complement of the left (right) data-type semi-join. 
%It returns rows from the first table, 
%where no matches occur
%on the common attributes of the two tables.
%
Let $\mathcal{S}$ 
%= {\langle{I,x,X}\rangle}$ 
be a fixed signature.
%
%Let $\mathcal{A} = {\langle{X,Y,\models_{\mathcal{A}}}\rangle}$
%be a fixed type domain.
%
For any two $\mathcal{S}$-tables $\mathcal{T}_{1}$ and $\mathcal{T}_{2}$
that are linked through 
an $X$-sorted type domain opspan 
$\mathcal{A}_{1}
%={\langle{I_{1},s_{1}}\rangle} 
\xrightarrow{g_{1}} 
%\overset{
\mathcal{A}
%}{\mathcal{S}} 
\xleftarrow{g_{2}} 
%{\langle{I_{2},s_{2}}\rangle} = 
\mathcal{A}_{2}$,
%in $\mathrmbf{Cls}(X)$,
%
the left data-type anti-join 
$\mathcal{T}_{1}{\;\oslash_{\mathcal{S}}\,}\mathcal{T}_{2}$
is the set of all tuples in $\mathcal{T}_{1}$
for which there is no tuple in $\mathcal{T}_{2}$ 
%that is 
equal on their common data values. 
We use the following routes of flow.
% from Tbl.\;\ref{tbl:routes:flow}.
%
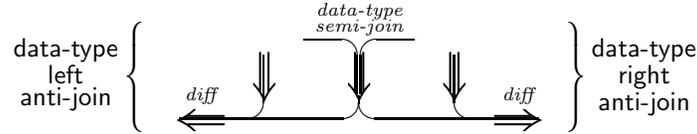
\begin{center}
{{\begin{tabular}{c}
\setlength{\unitlength}{0.6pt}
\begin{picture}(320,80)(0,-10)
\put(325,25){\makebox(0,0){\normalsize{$\left.\rule{0pt}{24pt}\right\}
\overset{\textstyle{\textsf{data-type}}}
{\underset{\textstyle{\textsf{anti-join}}}{\textsf{right}}}$}}}
\put(-20,25){\makebox(0,0){\normalsize{$
\overset{\textstyle{\textsf{data-type}}}
{\underset{\textstyle{\textsf{anti-join}}}{\textsf{left}}}
\left\{\rule{0pt}{24pt}\right.$}}}
\put(90,30){\makebox(0,0){\huge{
$\overset{\textit{\scriptsize{}}}{\Downarrow}$}}}
\put(150,65){\makebox(0,0){\huge{
$\overset{\textit{\scriptsize{data-type}}}{}$}}}
\put(150,55){\makebox(0,0){\huge{
$\overset{\textit{\scriptsize{semi-join}}}{}$}}}
\put(150,30){\makebox(0,0){\huge{
$\overset{\textit{\scriptsize{}}}{\Downarrow}$}}}
\put(210,30){\makebox(0,0){\huge{
$\overset{\textit{\scriptsize{}}}{\Downarrow}$}}}
\put(50,5){\makebox(0,0){\huge{
$\overset{\textit{\scriptsize{diff}}}{\Leftarrow}$}}}
\put(250,5){\makebox(0,0){\huge{
$\overset{\textit{\scriptsize{diff}}}{\Rightarrow}$}}}
\put(95,8){\line(0,1){30}}
\put(85,8){\oval(20,20)[br]}
\put(215,8){\line(0,1){30}}
\put(225,8){\oval(20,20)[bl]}
%\put(82,48){\vector(1,0){0}}
\put(120,48){\line(1,0){24}}
\put(144,38){\oval(20,20)[tr]}
\put(190,48){\line(-1,0){24}}
\put(166,38){\oval(20,20)[tl]}
\put(154,14){\line(0,1){24}}
%\put(155,12){\line(0,-1){8}}
\qbezier(155,12)(155,10)(155,8)
\put(156,14){\line(0,1){24}}
\put(145,8){\oval(20,20)[br]}
\put(144,-2){\line(-1,0){102}}
\put(165,8){\oval(20,20)[bl]}
\put(166,-2){\line(1,0){102}}
\end{picture}
\end{tabular}}}
\end{center}
Left data-type anti-join
within the context $\mathrmbf{Tbl}(\mathcal{S})$
is left data-type semi-join, followed by difference.
This is the two-step process
illustrated in 
Fig.\;\ref{fole:boole:anti:join:flo:chrt}.
The constraint, construction and input for data-type anti-join 
are identical to that for data-type join.
Only the output is different.
\begin{description}
\item[Constraint:] 
The constraint for data-type anti-join is the same as the constraint for data-type join
(Tbl.\,\ref{tbl:fole:boolean:join:input:output}):
an $X$-sorted type domain opspan 
$\mathcal{A}_{1}\xrightarrow{g_{1}}\mathcal{A}\xleftarrow{g_{2}}\mathcal{A}_{2}$
in $\mathrmbf{Cls}(X)$.
\newline
\item[Construction:] 
The construction for data-type anti-join is the same as 
the construction for data-type join 
(Tbl.\,\ref{tbl:fole:boolean:join:input:output}):
the span
{\footnotesize{$\mathcal{A}_{1} 
\xleftarrow{\;\tilde{g}_{1}\,} 
{\mathcal{A}_{1}{\times_{\mathcal{A}}}\mathcal{A}_{2}}
\xrightarrow{\;\tilde{g}_{2}} 
\mathcal{A}_{2}$}\normalsize}
of projection $X$-type domain morphisms
with pullback type domain
$\mathcal{A}_{1}{\times_{\mathcal{A}}}\mathcal{A}_{2}$.
\newline
\item[Input:] 
The input for data-type anti-join is the same as the input for data-type join 
(Tbl.\,\ref{tbl:fole:boolean:join:input:output}):
a pair of tables
$\mathcal{T}_{1}\in\mathrmbf{Tbl}_{\mathcal{S}}(\mathcal{A}_{1})$
and
$\mathcal{T}_{2}\in\mathrmbf{Tbl}_{\mathcal{S}}(\mathcal{A}_{2})$.
\newpage
\item[Output:] 
The output 
for data-type anti-join 
is 
%defined in two steps:
%Data-type 
semi-join followed by difference.
\newline
\begin{itemize}
\item 
Left data-type semi-join 
%is defined by the two-step process of 
%Data-type join followed by restriction.
%This 
results in the table
${\mathcal{T}_{1}}{\,\oleft_{\mathcal{S}}}{\mathcal{T}_{2}} =
\grave{\mathrmbfit{tbl}}_{\mathcal{S}}(\tilde{g}_{1})
(\mathcal{T}_{1}{\,\oplus_{\mathcal{S}}}\mathcal{T}_{2}) = 
%{\tilde{g}_{1}}^{\ast}(\mathcal{T}_{1}{\,\cup_{\mathcal{T}}}\mathcal{T}_{2}) =
{\langle{\widehat{K},\hat{t}_{1}}\rangle}$ 
with key set $\widehat{K}_{1}$  and
tuple function
$\widehat{K}_{1}\xrightarrow{\hat{t}_{1}}\mathrmbfit{tup}_{\mathcal{S}}(\mathcal{A}_{1})$.
%defined by pullback. 
%equal to the composite 
%$\mathrmbfit{tup}_{\mathcal{A}}(\mathcal{S}_{1})
%\xleftarrow{\mathrmbfit{tup}_{\mathcal{A}}(\iota_{1})}
%\mathrmbfit{tup}_{\mathcal{A}}({\mathcal{S}_{1}{+_{\mathcal{S}}}\mathcal{S}_{2}})
%\xleftarrow{(\hat{t}_{1},\hat{t}_{2})}\widehat{K}_{12}$.
\newline
\item 
Difference 
in the small table fiber context
$\mathrmbf{Tbl}_{\mathcal{S}}(\mathcal{A}_{1})$
gives the left data-type anti-join table
${\mathcal{T}_{1}}{\,\oslash_{\mathcal{S}}}{\mathcal{T}_{2}} 
= \mathcal{T}_{1}{\,-\,}({\mathcal{T}_{1}}{\,\oleft_{\mathcal{S}}}{\mathcal{T}_{2}})$.
\newline
\end{itemize}
%
%
%The data-type anti-join flowchart input/output is displayed in 
%Tbl\;\ref{tbl:fole:boolean:join:input:output}.
Data-type anti-join is data-type semi-join followed by difference. 
For left data-type anti-join, 
this is the two-step process
\newline\mbox{}\hfill\rule[-10pt]{0pt}{26pt}
${\mathcal{T}_{1}}{\,\oslash_{\mathcal{S}}}{\mathcal{T}_{2}} 
\doteq
\mathcal{T}_{1}{\,-\,}({\mathcal{T}_{1}}{\,\oleft_{\mathcal{S}}}{\mathcal{T}_{2}})$.
%
%%%%%%%%%%%%%%%%%%%%%%%%%%%%%%%%%%%%%%%%%%%%%%%%%%%%%%%%%%%%%%%%%%%%%%%%%%%%%%%%
%%%%%%%%%%%%%%%%%%%%%%%%%%%%%%%%%%%%%%%%%%%%%%%%%%%%%%%%%%%%%%%%%%%%%%%%%%%%%%%%
\footnote{The data-type anti-join of a disjoint sum is the empty table.}
%%%%%%%%%%%%%%%%%%%%%%%%%%%%%%%%%%%%%%%%%%%%%%%%%%%%%%%%%%%%%%%%%%%%%%%%%%%%%%%%
%%%%%%%%%%%%%%%%%%%%%%%%%%%%%%%%%%%%%%%%%%%%%%%%%%%%%%%%%%%%%%%%%%%%%%%%%%%%%%%%
%
\hfill\mbox{}\newline
%The left anti-join $\mathcal{T}_{1}{\,\unrhd_{\mathcal{S}}}\mathcal{T}_{2}$
%is the difference between the table $\mathcal{T}_{1}$ and the semi-join
%$\mathcal{T}_{1}{\,\rhd_{\mathcal{S}}}\mathcal{T}_{2}$
%\hfill\mbox{}
%
\end{description}
There is an inclusion morphism
%\newline\mbox{}\hfill
$\mathcal{T}_{1}
\xhookleftarrow{\bar{\omega}_{1}}
{\mathcal{T}_{1}}{\,\oslash_{\mathcal{S}}}{\mathcal{T}_{2}}$
in the small table fiber context
$\mathrmbf{Tbl}_{\mathcal{S}}(\mathcal{A}_{1})$,
which is the output for left data-type anti-join.
The right data-type anti-join has a similar definition
(Tbl\;\ref{tbl:fole:boolean:join:input:output}).

%%%%%%%%%%%%%%%%%%%%%%%%%%%%%%%%%%%%%%%%%%%%%%%%%%%%%%%%%%%%%%
%
\newpage
\subsection{Generic Join.}\label{sub:sub:sec:generic:join}
%\ast\ast\ast$} 
%(Pullback)}
%%%%%%%%%%%%%%%%%%%%%%%%%%%%%%%%%%%%%%%%%%%%%%%%%%%%%%%%%%%%%%
%
%
\begin{figure}
\begin{center}
{{{\begin{tabular}{c}
\begin{picture}(160,75)(37,27)
\setlength{\unitlength}{0.97pt}
%%%%%%%%%%%%%%%%%%%%%%%%%%%%%%%%%%%%%%%%%%%%%%%%%%
%\put(44,62){\begin{picture}(0,0)(0,0)
%\setlength{\unitlength}{0.46pt}
\put(54,65){\begin{picture}(0,0)(0,0)
\setlength{\unitlength}{0.35pt}
%\thicklines
%\put(106,40){\makebox(0,0){\normalsize{$\boldsymbol{\circ}$}}}
%\put(4.7,40){\makebox(0,0){\normalsize{$\boldsymbol{\circ}$}}}
\put(10,10){\line(1,0){60}}
\put(10,70){\line(1,0){60}}
\put(10,70){\line(0,-1){60}}
\put(70,40){\oval(60,60)[br]}
\put(70,40){\oval(60,60)[tr]}
\put(55,55){\makebox(0,0){\scriptsize{{\textit{{project}}}}}}
\put(56,38){\makebox(0,0){\Large{${\Rightarrow}$}}}
\put(55,25){\makebox(0,0){\scriptsize{{\textit{{expand}}}}}}
\end{picture}}
%%%%%%%%%%%%%%%%%%%%%%%%%%%%%%%%%%%%%%%%%%%%%%%%%%
\put(146.5,65){\begin{picture}(0,0)(0,0)
\setlength{\unitlength}{0.35pt}
%\thicklines
%\put(106,40){\makebox(0,0){\normalsize{$\boldsymbol{\circ}$}}}
%\put(4.7,40){\makebox(0,0){\normalsize{$\boldsymbol{\circ}$}}}
\put(40,10){\line(1,0){60}}
\put(40,70){\line(1,0){60}}
\put(100,70){\line(0,-1){60}}
\put(40,40){\oval(60,60)[bl]}
\put(40,40){\oval(60,60)[tl]}
\put(58,55){\makebox(0,0){\scriptsize{{\textit{{project}}}}}}
\put(56,38){\makebox(0,0){\Large{${\Leftarrow}$}}}
\put(58,25){\makebox(0,0){\scriptsize{{\textit{{expand}}}}}}
\end{picture}}
%%%%%%%%%%%%%%%%%%%%%%%%%%%%%%%%%%%%%%%%%%%%%%%%%%
\put(98,37){\begin{picture}(0,0)(0,3)
\setlength{\unitlength}{0.35pt}
\put(60,30){\makebox(0,0){\normalsize{$\vee$}}}
%\thicklines
\put(40,10){\line(1,0){40}}
\put(10,70){\line(1,0){100}}
\put(10,70){\line(0,-1){30}}
\put(110,70){\line(0,-1){30}}
\put(40,40){\oval(60,60)[bl]}
\put(80,40){\oval(60,60)[br]}
\put(60,55){\makebox(0,0){\scriptsize{{\textit{{join}}}}}}
\end{picture}}
%%%%%%%%%%%%%%%%%%%%%%%%%%%%%%%%%%%%%%%%%%%%%%%%%%
\put(120,100){\makebox(0,0){\footnotesize{{\textit{{generic join}}}}}}
\put(120,88){\makebox(0,0){\scriptsize{$\coprod$}}}
\put(121,72){\makebox(0,0){\ldots}}
%%%%%%%%%%%%%%%%%%%%%%%%%%%%%%%%%%%%%%%%%%%%%%%%%%
\put(38,80){\line(0,1){20}}
\put(38,80){\vector(1,0){20}}
\put(110,80){\line(-1,0){20}}
\put(110,80){\vector(0,-1){21}}
\put(120,38){\vector(0,-1){15}}
\put(130,80){\vector(0,-1){21}}
\put(130,80){\line(1,0){20}}
\put(203,80){\line(0,1){20}}
\put(203,80){\vector(-1,0){20}}
%\thicklines
%\put(15,110){\line(1,0){210}}
%\put(55,10){\line(1,0){130}}
%\put(15,50){\line(0,1){60}}
%\put(225,50){\line(0,1){60}}
%\qbezier(15,50)(15,10)(55,10)
%\qbezier(185,10)(225,10)(225,50)
%%%%%%%%%%%%%%%%%%%%%%%%%%%%%%%%%%%%%%%%%%%%%%%%%%
\end{picture}
\end{tabular}}}}
\end{center}
\caption{\texttt{FOLE} Generic Join Flow Chart}
\label{fig:fole:boolean:join:flo:chrt}
\end{figure}
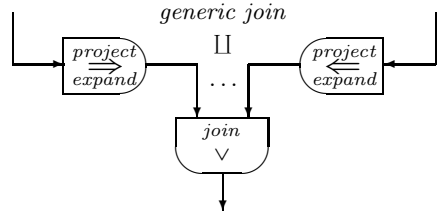
The generic join for tables 
is the relational counterpart of
the logical disjunction for predicates.
Where 
the \emph{join} operation (\S\,\ref{sub:sub:sec:boole}) 
is the analogue for logical disjunction 
at the small scope $\mathrmbf{Tbl}(\mathcal{D})$ of a signed domain table fiber,
and the \emph{co-quotient} operation (\S\,\ref{sub:sub:sec:co-quotient})
and
the \emph{data-type join} operation (\S\,\ref{sub:sub:sec:boole:join})
are special cases of the analogue 
%for logical conjunction 
at the intermediate scope 
$\mathrmbf{Tbl}(\mathcal{S})$ of a signature table fiber,
the \emph{generic join} operation
%(\S\,\ref{sub:sub:sec:generic:meet})
is defined at the large scope $\mathrmbf{Tbl}$ of all tables.
We identify all of these concepts as colimits at different scopes.
%We identify \texttt{FOLE} generic meets with all limits.
%

In this section, we focus on tables in the full context $\mathrmbf{Tbl}$ 
of all tables.
These colimits are resolvable into expansion-projections
followed by join.
The generic join operation is dual to the generic meet operation
(\S\,\ref{sub:sub:sec:generic:meet}).
%
%%%%%%%%%%%%%%%%%%%%%%%%%%%%%%%%%%%%%%%%%%%%%%%%%%%%%%%%%%%%%%%%%%%%%%
%%%%%%%%%%%%%%%%%%%%%%%%%%%%%%%%%%%%%%%%%%%%%%%%%%%%%%%%%%%%%%%%%%%%%%
\footnote{Generic meets and limits in the context of $\mathrmbf{Tbl}$ 
(\S\,\ref{sub:sub:sec:generic:meet})
can be constructed out of 
colimits in the context of signed domains $\mathrmbf{Dom}$,
the table restriction-inflation operation along signed domain morphisms, and 
meets (limits) in small table fibers.}
%%%%%%%%%%%%%%%%%%%%%%%%%%%%%%%%%%%%%%%%%%%%%%%%%%%%%%%%%%%%%%%%%%%%%%
%%%%%%%%%%%%%%%%%%%%%%%%%%%%%%%%%%%%%%%%%%%%%%%%%%%%%%%%%%%%%%%%%%%%%%
%
%\newpage
%We focus on tables in the context $\mathrmbf{Tbl}(\mathcal{A})$ 
%for fixed type domain $\mathcal{A}$.
%
%\newpage
%For the general case,
%we will use the terminology \emph{generic meet} when we use a sufficient
%
The \emph{generic join} operation only needs a sufficient collection of tables
(Def.\;\ref{def:suff:adequ:colim}).
To reiterate,
we identify \texttt{FOLE} generic joins with all colimits
in the context $\mathrmbf{Tbl}$.
\begin{description}
\item[Constraint:] 
Consider a diagram 
$\mathrmbfit{D} 
%= \mathrmbfit{T}^{\text{op}}{\,\circ\,}\mathrmbfit{dom} 
: \mathrmbf{I}^{\text{op}} \rightarrow \mathrmbf{Dom}$
consisting of a linked collection of signed domains 
$\bigl\{ 
\mathcal{D}_{i} 
\xrightarrow{{\langle{{h},{f},{g}}\rangle}} 
\mathcal{D}_{j} 
%\mid i \in \mathrmbf{I} 
\bigr\}$.
%$\bigl\{ 
%\mathcal{D}_{i} 
%= \mathrmbfit{D}(i)
%= \mathrmbfit{dom}(\mathrmbfit{T}(i))
%\mid i \in \mathrmbf{I} \bigr\}$.
This is the constraint for generic join 
(Tbl.\,\ref{tbl:fole:generic:join:input:output}).
\newline
\item[Construction:] 
Let $\widetilde{\mathcal{D}} = \prod \mathrmbfit{D}$ be the limit
% of 
%$\mathrmbfit{D} = \mathrmbfit{T}^{\text{op}}{\,\circ\,}\mathrmbfit{dom} :
%\mathrmbf{I}^{\text{op}} \rightarrow 
in $\mathrmbf{Dom}$
with projection signed domain morphisms 
$\bigl\{ 
\mathcal{D}_{i} 
\xleftarrow{{\langle{\tilde{h}_{i},\tilde{f}_{i},\tilde{g}_{i}}\rangle}} 
\widetilde{\mathcal{D}} 
\mid i \in \mathrmbf{I} \bigr\}$
that commute with the links in the constraint:
${\langle{\tilde{h}_{i},\tilde{f}_{i},\tilde{g}_{i}}\rangle}
{\;\circ\;}{\langle{{h},{f},{g}}\rangle}
= 
{\langle{\tilde{h}_{j},\tilde{f}_{j},\tilde{g}_{j}}\rangle}$.  
This is the construction for generic join 
(Tbl.\,\ref{tbl:fole:generic:join:input:output}).
\newline
\item[Input:] 
Let $I \xrightarrow{\mathrmbfit{T}} \mathrmbf{Tbl}$ 
be a \underline{sufficient} indexed collection of tables 
(Def.\;\ref{def:suff:adequ:colim})
{\footnotesize{$
\bigl\{ 
\mathcal{T}_{i} = \mathrmbfit{T}(i) \in \mathrmbf{Tbl}(\mathcal{D}_{i}) 
\mid i \in I \bigr\}$}}
for some indexing set
$I \subseteq obj(\mathrmbf{I})$.
%
%Let
%$I \xrightarrow{\mathrmbfit{D}} \mathrmbf{Dom}$
%be its signed domain projection
%consisting of an (unlinked) indexed collection of signed domains 
%$\bigl\{ 
%\mathcal{D}_{i} = \mathrmbfit{D}(i)
%\mid i \in I \bigr\}$.
%
This is the input for generic join 
(Tbl.\,\ref{tbl:fole:generic:join:input:output}).
\newline\newline
\item[Output:]\mbox{}
Generic join is projection/expansion ($i \in I$ times) followed by join.
\newline
\begin{itemize}
\item 
For each index $i \in I$,
projection/expansion
{\footnotesize{$\mathrmbf{Tbl}(\mathcal{D}_{i}) 
\xrightarrow{\acute{\mathrmbfit{tbl}}({\tilde{h}_{i},\tilde{f}_{i},\tilde{g}_{i}})} 
\mathrmbf{Tbl}(\widehat{\mathcal{D}})$}}
(\S\,\ref{sub:sub:sec:flow:sign:dom:mor})
along the tuple function
of the signed domain morphism 
$\mathcal{D}_{i} 
\xleftarrow{{\langle{\tilde{h}_{i},\tilde{f}_{i},\tilde{g}_{i}}\rangle}} 
\widetilde{\mathcal{D}} = 
\prod\mathrmbfit{D}$
maps the table 
$\mathcal{T}_{i} 
%= {\langle{K_{i},t_{i}}\rangle} 
\in \mathrmbf{Tbl}(\mathcal{D}_{i})$
%with its tuple function 
%$K_{i} \xrightarrow{t_{i}} \mathrmbfit{tup}(\mathcal{D}_{i})$
to the table
$\widetilde{\mathcal{T}}_{i}
%\widehat{\mathrmbfit{T}}(i)
%= \grave{\mathrmbfit{tbl}}({h,f,g})(\mathrmbfit{T}(i))
= {\langle{K_{i},\tilde{t}_{i}}\rangle} 
\in \mathrmbf{Tbl}({\widetilde{\mathcal{D}}})$
with its tuple function
$K_{i} \xrightarrow{\tilde{t}_{i}} 
\mathrmbfit{tup}(\widetilde{\mathcal{D}})$
defined by composition, 
$t_{i}{\,\cdot\,}{\mathrmbfit{tup}(\tilde{h}_{i},\tilde{f}_{i},\tilde{g}_{i})}
= \tilde{t}_{i}$. 
Here we have ``horizontally abridged'' and then ``vertically extended'' 
tuples in $\mathrmbfit{tup}(\mathcal{D}_{i}) \subseteq \mathrmbf{List}(Y_{i})$ 
by composition 
along the tuple function
%\[\mbox
{\footnotesize{$
\mathrmbfit{tup}(\mathcal{D}_{i})
\xrightarrow
%[(h{\cdot}{(\mbox{-})})\cdot({(\mbox{-})}{\cdot}g)]
{\mathrmbfit{tup}(\hat{h}_{i},\hat{f}_{i},\hat{g}_{i})}
\mathrmbfit{tup}(\widehat{\mathcal{D}})
$}\normalsize}
(see LHS Fig.\;\ref{fig:flow:sign:dom}).
This is linked 
%(Fig.\;\ref{fig:fole:boolean:meet})
to the table $\mathcal{T}_{i}$ 
by the table morphism 
%\newline
%\[\mbox
{\footnotesize{{$
\mathcal{T}_{i} = {\langle{K_{i},t_{i}}\rangle}
\xrightarrow{{\langle{{\langle{\tilde{h}_{i},\tilde{f}_{i},\tilde{g}_{i}}\rangle},1_{K_{i}}}\rangle}}
{\langle{K_{i},\tilde{t}_{i}}\rangle} 
= \widetilde{\mathcal{T}}_{i}$.}}\normalsize}
%\]
\newline
\item 
Union (\S\,\ref{sub:sub:sec:boole})
of the tables $\{ \widetilde{\mathcal{T}}_{i} \mid i \in I \}$
in the fiber context $\mathrmbf{Tbl}(\widetilde{\mathcal{D}})$
defines the generic join
$\coprod\mathrmbfit{T}
= \widetilde{\mathrmbfit{T}} 
= \bigvee \bigl\{\widetilde{\mathcal{T}}_{i} \mid i \in I \bigr\}
= \bigvee_{i \in I} 
\acute{\mathrmbfit{tbl}}({\tilde{h}_{i},\tilde{f}_{i},\tilde{g}_{i}})(\mathcal{T}_{i})
= {\langle{\widetilde{K},\tilde{t}}\rangle}$,
whose key set 
is the disjoint union 
$\widetilde{K} = {+}\{K_{i}\mid i \in I \}$ 
and 
whose tuple map is the mediating function
$\widetilde{K}\xrightarrow{{[\,\tilde{t}_{i}]}}
\mathrmbfit{tup}(\widetilde{\mathcal{D}})$
of the multi-opspan
$\bigl\{ K_{i}\xrightarrow{\tilde{t}_{i}}
\mathrmbfit{tup}(\widetilde{\mathcal{D}}) \mid i \in I \bigr\}$,
resulting in the discrete multi-opspan (cocone)
{\footnotesize{{$\Bigl\{ 
\widetilde{\mathcal{T}}_{i} 
%= {\langle{\widehat{K}_{i},\hat{t}_{i}}\rangle} 
\xrightarrow{\;\tilde{\iota}_{i}\;} 
%{\langle{\widehat{K},\hat{t}}\rangle} = 
\widetilde{\mathrmbfit{T}} 
\mid i \in I \Bigr\}$.}}\normalsize}
\newline
\end{itemize}
Projection-expansion composed with join 
defines the multi-opspan of table morphisms
\[\mbox{\footnotesize{{$\bigl\{ 
\mathcal{T}_{i}
\xrightarrow
[\;
{\langle{
{\langle{\tilde{h}_{i},\tilde{f}_{i},\tilde{g}_{i}}\rangle},1_{K_{i}}
}\rangle}
{\circ\,}\tilde{\iota}_{i}\;]
{\;{\langle{
{\langle{\tilde{h}_{i},\tilde{f}_{i},\tilde{g}_{i}}\rangle},\tilde{k}_{i}
}\rangle}\;} 
\coprod\mathrmbfit{T} = \widetilde{\mathrmbfit{T}}
\mid i \in I \bigr\}$,}
}\normalsize}\]
illustrated in 
Fig.\;\ref{fig:fole:generic:join},
which is the output for generic join 
(Tbl.\,\ref{tbl:fole:generic:join:input:output}).
%\begin{description}
%\item[projection/expansion:] 
%\end{description}
%
\end{description}
Generic join is projection/expansion ($i \in I$ times) followed by join.
This is the two-step process
%\newline\mbox{}\hfill\rule[-10pt]{0pt}{26pt}
%
\[\mbox{\footnotesize{{$
\coprod\mathrmbfit{T} 
\doteq
\bigvee\;
\Bigl\{  
\acute{\mathrmbfit{tbl}}({\tilde{h}_{i},\tilde{f}_{i},\tilde{g}_{i}})(\mathcal{T}_{i}) 
\mid i \in I \Bigr\}$.}}\normalsize}\]
%\hfill\mbox{}\newline
%
\begin{aside}
Theoretically
this would represent the colimit
%(see the application discussion for completeness in \S\,\ref{sub:sub:sec:lim:colim:tbl})
of a diagram
$\mathrmbf{I}\xrightarrow{\mathrmbfit{T}}\mathrmbf{Tbl}$ 
consisting of a linked collection of tables.
But practically,
we are only given the constraint (a diagram)
$\mathrmbf{I}^{\text{op}}\xrightarrow{\mathrmbfit{D}}\mathrmbf{Dom}$ 
consisting of a linked collection of signed domains 
$\bigl\{ 
\mathcal{D}_{i} = \mathrmbfit{D}(i)
\mid i \in \mathrmbf{I} \bigr\}$
and the input 
$I\xrightarrow{\mathrmbfit{T}}\mathrmbf{Tbl}$ 
consisting of a \underline{sufficient} indexed collection of tables 
(Def.\;\ref{def:suff:adequ:colim})
{\footnotesize{$
\bigl\{ 
\mathcal{T}_{i} =
\mathrmbfit{T}(i) \in \mathrmbf{Tbl}(\mathcal{D}_{i}) 
\mid i \in I \subseteq obj(\mathrmbf{I}) \bigr\}$}}.
\end{aside}
%\newpage

%
\begin{table}
\begin{center}
{{\fbox{\begin{tabular}{c}
\setlength{\extrarowheight}{2pt}
{\scriptsize{$\begin{array}[c]{c@{\hspace{12pt}}l}
\mathcal{D}_{i}
\xrightarrow{{\langle{h,f,g}\rangle}}
\mathcal{D}_{j}
&
\textit{constraint}
\\
\mathcal{D}_{i}
\xleftarrow{{\langle{\tilde{h}_{i},\tilde{f}_{i},\tilde{g}_{i}}\rangle}}
\prod \mathrmbfit{D}
%\widetilde{\mathcal{D}} 
&
\textit{construction}
\\\hline
%\mathcal{T}_{i}
%\xleftarrow{{\langle{{\langle{h,f,g}\rangle},k}\rangle}}
%\mathcal{T}_{j}
\bigl\{ 
\mathcal{T}_{i} 
%=\mathrmbfit{T}(i) 
\in \mathrmbf{Tbl}(\mathcal{D}_{i}) 
\mid i \in I \bigr\}
&
\textit{input}
\\
\mathcal{T}_{i}
\xrightarrow
%[\;{\langle{
%{\langle{\tilde{h}_{i},\tilde{f}_{i},\tilde{g}_{i}}\rangle},1_{K_{i}}
%}\rangle}
%{\circ\,}\tilde{\iota}_{i}\;]
{\;{\langle{
{\langle{\tilde{h}_{i},\tilde{f}_{i},\tilde{g}_{i}}\rangle},\tilde{k}_{i}
}\rangle}\;} 
\coprod \mathrmbfit{T}
&
\textit{output}
\end{array}$}}
\end{tabular}}}}
\end{center}
\caption{\texttt{FOLE} Generic Join I/O}
\label{tbl:fole:generic:join:input:output}
\end{table}
%

%\mbox{}\newline\rule{300pt}{1pt}\newline%

%
\begin{figure}
\begin{center}
{{\begin{tabular}{c}
%@{\hspace{75pt}}c}
%%%%%%%%%%%%%%%%%%%%%%%%%%%%%%%%%%%%%%%%%%%%%%%%%%
{{\begin{tabular}{c}
\setlength{\unitlength}{0.56pt}
\begin{picture}(320,160)(0,-5)
\put(0,80){\makebox(0,0){\footnotesize{$K_{i}$}}}
\put(100,80){\makebox(0,0){\footnotesize{$K_{i}$}}}
\put(220,80){\makebox(0,0){\footnotesize{$K_{j}$}}}
\put(324,80){\makebox(0,0){\footnotesize{$K_{j}$}}}
\put(160,148){\makebox(0,0){\footnotesize{$\widetilde{K}$}}}
\put(0,0){\makebox(0,0){\footnotesize{$
{\mathrmbfit{tup}(\mathcal{D}_{i})}$}}}
\put(320,0){\makebox(0,0){\footnotesize{$
{\mathrmbfit{tup}(\mathcal{D}_{j})}$}}}
\put(160,0){\makebox(0,0){\footnotesize{$
{\mathrmbfit{tup}(\widehat{\mathcal{D}})}$}}}
\put(55,90){\makebox(0,0){\scriptsize{$1_{K_{i}}$}}}
\put(265,90){\makebox(0,0){\scriptsize{$1_{K_{j}}$}}}
\put(80,-12){\makebox(0,0){\scriptsize{$
\mathrmbfit{tup}(\tilde{h}_{i},\tilde{f}_{i},\tilde{g}_{i})$}}}
\put(240,-12){\makebox(0,0){\scriptsize{$
\mathrmbfit{tup}(\tilde{h}_{j},\tilde{f}_{j},\tilde{g}_{j})$}}}
\put(-6,40){\makebox(0,0)[r]{\scriptsize{$t_{i}$}}}
\put(125,40){\makebox(0,0)[r]{\scriptsize{$\tilde{t}_{i}$}}}
\put(125,116){\makebox(0,0)[r]{\scriptsize{$\tilde{\iota}_{i}$}}}
\put(327,40){\makebox(0,0)[l]{\scriptsize{$t_{j}$}}}
\put(200,40){\makebox(0,0)[l]{\scriptsize{$\tilde{t}_{j}$}}}
\put(195,115){\makebox(0,0)[l]{\scriptsize{$\tilde{\iota}_{j}$}}}
\put(64,128){\makebox(0,0)[r]{\scriptsize{$\tilde{\iota}_{i}$}}}
\put(260,128){\makebox(0,0)[l]{\scriptsize{$\tilde{\iota}_{j}$}}}
\put(0,65){\vector(0,-1){50}}
\put(320,65){\vector(0,-1){50}}
\put(105,65){\vector(1,-1){50}}
\put(215,65){\vector(-1,-1){50}}
\put(105,90){\vector(1,1){45}}
\put(210,90){\vector(-1,1){45}}
\put(20,80){\vector(1,0){60}}
\put(300,80){\vector(-1,0){60}}
\put(35,0){\vector(1,0){85}}
\put(285,0){\vector(-1,0){85}}
\qbezier(20,96)(80,120)(140,144)\put(140,144){\vector(2,1){0}}
\qbezier(180,144)(240,120)(300,96)\put(180,144){\vector(-2,1){0}}
%%%%%%%%%%
\put(60,45){\makebox(0,0){\huge{
$\overset{\textit{\scriptsize{expand}}}{\Rightarrow}$}}}
\put(160,110){\makebox(0,0){\ldots}}
\put(160,80){\makebox(0,0){${\scriptsize{join}}$}}
\put(250,45){\makebox(0,0){\huge{
$\overset{\textit{\scriptsize{expand}}}{\Leftarrow}$}}}
%%%%%%%%%%
\end{picture}
\end{tabular}}}
%
%%%%%%%%%%%%%%%%%%%%%%%%%%%%%%%%%%%%%%%%%%%%%%%%%%
\\\\
%%%%%%%%%%%%%%%%%%%%%%%%%%%%%%%%%%%%%%%%%%%%%%%%%%
{{\begin{tabular}{c}
\setlength{\unitlength}{0.7pt}
\begin{picture}(200,70)(-100,60)
\put(-140,57){\makebox(0,0){\footnotesize{${\mathcal{T}_{i}}$}}}
\put(0,120){\makebox(0,0){\footnotesize{$\widetilde{\mathrmbfit{T}}$}}}
\put(140,57){\makebox(0,0){\footnotesize{${\mathcal{T}_{j}}$}}}
\put(-63,57){\makebox(0,0){\footnotesize{$\widetilde{\mathcal{T}}_{i}$}}}
\put(63,57){\makebox(0,0){\footnotesize{$\widetilde{\mathcal{T}}_{j}$}}}
\put(-80,100){\makebox(0,0)[r]{\tiny{$
{\langle{{\langle{\tilde{h}_{i},\tilde{f}_{i},\tilde{g}_{i}}\rangle},
\tilde{\iota}_{i}}\rangle}$}}}
\put(-105,45){\makebox(0,0){\tiny{$
{\langle{{\langle{\tilde{h}_{i},\tilde{f}_{i},\tilde{g}_{i}}\rangle},
1_{K_{i}}}\rangle}$}}}
\put(82,100){\makebox(0,0)[l]{\tiny{$
{\langle{{\langle{\tilde{h}_{j},\tilde{f}_{j},\tilde{g}_{j}}\rangle},
\tilde{\iota}_{j}}\rangle}$}}}
\put(105,45){\makebox(0,0){\tiny{$
{\langle{{\langle{\tilde{h}_{j},\tilde{f}_{j},\tilde{g}_{j}}\rangle},
1_{K_{j}}}\rangle}$}}}
\put(-33,93){\makebox(0,0)[r]{\scriptsize{$\tilde{\iota}_{i}$}}}
\put(35,93){\makebox(0,0)[l]{\scriptsize{$\tilde{\iota}_{j}$}}}
\qbezier(-130,65)(-85,96)(-16,120)\put(-16,120){\vector(2,1){0}}
\qbezier(130,65)(85,96)(16,120)\put(16,120){\vector(-2,1){0}}
\put(-50,70){\vector(1,1){40}}
\put(50,70){\vector(-1,1){40}}
\put(-130,57){\vector(1,0){50}}
\put(130,57){\vector(-1,0){50}}
\put(0,80){\makebox(0,0){\ldots}}
\end{picture}
\end{tabular}}}
%%%%%%%%%%%%%%%%%%%%%%%%%%%%%%%%%%%%%%%%%%%%%%%%%%
\\
%%%%%%%%%%%%%%%%%%%%%%%%%%%%%%%%%%%%%%%%%%%%%%%%%%
\end{tabular}}}
\end{center}
\caption{\texttt{FOLE} Generic Join}
\label{fig:fole:generic:join}
\end{figure}
%

%%%%%%%%%%%%%%%%%%%%%%%%%%%%%%%%%%%%%%%%%%%%%%%%%%%%%%%%%%%%%%
%%%%%%%%%%%%%%%%%%%%%%%%%%%%%%%%%%%%%%%%%%%%%%%%%%%%%%%%%%%%%%
%
\newpage
\section{Unorthodox Composite Operations}
\label{sub:sec:non-trad:ops}
%%%%%%%%%%%%%%%%%%%%%%%%%%%%%%%%%%%%%%%%%%%%%%%%%%%%%%%%%%%%
%%%%%%%%%%%%%%%%%%%%%%%%%%%%%%%%%%%%%%%%%%%%%%%%%%%%%%%%%%%%
%$\bigstar$ 
Operations are put into the unorthodox category for various reasons.
Selection and select-join are special cases of a more general concept:
selection is a special case of natural join, and
select-join is a special case of natural multi-join.
Filter join and data-type meet
do not follow 
the dual concepts of either \emph{limit} or \emph{colimit} of category theory.
Subtraction, division and outer-join are complex.
\comment{
The orthodox operations follow 
the dual concepts of eithor \emph{limit} or \emph{colimit} of category theory
(\S\,\ref{sub:sec:lim:colim:tbl}).
The limit concept initially requires inflation followed by meet, and dually
the colimit concept initially requires expansion followed by join.
An orthodox operation may follow these first two steps with reverse flow and basic operations.}
%Any deviation from this
%ese first two steps 
%is unorthodox.
%The unorthodox operations deviate in some way. 
%
Tbl.\,\ref{tbl:fole:comp:rel:ops:unortho}
lists the composite relational operations defined in this section.
\begin{table}
\begin{center}
{{{\begin{tabular}{c}
\setlength{\extrarowheight}{2pt}
{\scriptsize{$\begin{array}[c]
{|@{\hspace{5pt}}r@{\hspace{10pt}}l@{\hspace{5pt}\in\hspace{4pt}}l@{\hspace{5pt}}|}
\hline
\textbf{selection:}
&
\sigma_{\mathcal{R}'}(\mathcal{T}) = 
\grave{\mathrmbfit{tbl}}_{\mathcal{A}}(h)(\mathcal{R}'){\;\wedge\;}\mathcal{T}
&
\mathrmbf{Tbl}(\mathcal{A})
\\
\textbf{select-join:}
&
\mathcal{T}_{1}{\,\hat{\boxtimes}^{\mathcal{R}}_{\mathcal{A}}}\mathcal{T}_{2} = 
\sigma_{\mathcal{R}}\bigl(\mathcal{T}_{1}{\,\boxtimes_{\mathcal{A}}}\mathcal{T}_{2}\bigr)
&
\mathrmbf{Tbl}(\mathcal{A})
\\\hline\hline
\textbf{filtered join:}
&
\mathcal{T}_{1}{\,\varominus_{\mathcal{S}}\,}\mathcal{T}_{2} = 
\grave{\mathrmbfit{tbl}}_{\mathcal{S}}(g_{1})(\mathcal{T}_{1})
{\;\vee\;}
\grave{\mathrmbfit{tbl}}_{\mathcal{S}}(g_{2})(\mathcal{T}_{2})
&
\mathrmbf{Tbl}(\mathcal{S})
\\
\textbf{data-type meet:}
&
\mathcal{T}_{1}{\,\boxbar_{\mathcal{S}}}\mathcal{T}_{2} = 
\acute{\mathrmbfit{tbl}}_{\mathcal{S}}(\tilde{g}_{1})(\mathcal{T}_{1})
{\;\wedge\;}
\acute{\mathrmbfit{tbl}}_{\mathcal{S}}(\tilde{g}_{2})(\mathcal{T}_{2})
%{\mathcal{T}_{1}}{\,\curlyeqsucc_{\mathcal{A}}}{\mathcal{T}_{2}} =
%\grave{\mathrmbfit{tbl}}_{\mathcal{S}}(\tilde{g}_{1})
%\bigl(\mathcal{T}_{1}{\,\bigcap_{\mathcal{A}}}\mathcal{T}_{2}\bigr)
&
\mathrmbf{Tbl}(\mathcal{S})
%\\
%\textbf{anti-meet:}
%&
%{\mathcal{T}_{1}}{\,\succcurlyeq_{\mathcal{A}}}{\mathcal{T}_{2}} 
%\doteq
%\mathcal{T}_{1}{\,-\,}({\mathcal{T}_{1}}{\,\curlyeqsucc_{\mathcal{A}}}{\mathcal{T}_{2}})
%&
%\mathrmbf{Tbl}_{\mathcal{S}}(\mathcal{A}_{1})
\\\hline\hline
\textbf{subtraction:}
&
\mathcal{T}{\,\thicksim\,}\mathcal{T}_{2} =
\mathcal{T}{\,-\,}\acute{\mathrmbfit{tbl}}_{\mathcal{S}}(g)(\mathcal{T}_{2})
&
\mathrmbf{Tbl}(\mathcal{S})
\\\hline\hline
\textbf{division:}
&
\mathcal{T}{\,\div_{\!\mathcal{A}}\,}\mathcal{T}_{2}
= \acute{\mathrmbfit{tbl}}_{\mathcal{A}}(\iota_{1})(\mathcal{T})
{\,-\,}\acute{\mathrmbfit{tbl}}_{\mathcal{A}}(\iota_{1})
\Bigl(
\bigl(
\acute{\mathrmbfit{tbl}}_{\mathcal{A}}(\iota_{1})(\mathcal{T}){\times}\mathcal{T}_{2}
\bigr)
{\,-\,}\mathcal{T}
\Bigr)
&
\mathrmbf{Tbl}(\mathcal{A})
\\
\textbf{outer-join:}
&
\multicolumn{2}{l|}{
\mathcal{T}_{12} = 
\mathcal{T}_{1}{\,\boxtimes_{\mathcal{A}}}\mathcal{T}_{2} = 
\grave{\mathrmbfit{tbl}}_{\mathcal{A}}(\iota_{1})(\mathcal{T}_{1})
\wedge
\grave{\mathrmbfit{tbl}}_{\mathcal{A}}(\iota_{2})(\mathcal{T}_{2})
\text{ implies }}
\\
&
{\mathcal{T}_{1}}{\,{\rgroup\!\boxtimes}_{\mathcal{A}}}{\mathcal{T}_{2}} =
\acute{\mathrmbfit{tbl}}_{\mathcal{S}}(\tilde{0})(\mathcal{T}_{12})
{\,\vee\,}
\bigl(
\acute{\mathrmbfit{tbl}}_{\mathcal{S}}(\tilde{0})(
\mathcal{T}_{1}{-}\acute{\mathrmbfit{tbl}}_{\mathcal{A}}(\iota_{1})(\mathcal{T}_{12})
)
{\,\times\,}\mathcal{T}_{\bullet}
\bigr)
&
\mathrmbf{Tbl}
\\\hline
\end{array}$}}
\end{tabular}}}}
\end{center}
\caption{\texttt{FOLE} Unorthodox Composite Relational Operations}
\label{tbl:fole:comp:rel:ops:unortho}
\end{table}

% does not work
%%%%%%%%%%%%%%%%%%%%%%%%%%%%%%%%%%%%%%%%%%%%%%%%%%%%%%%%%%%%%
%%%%%%%%%%%%%%%%%%%%%%%%%%%%%%%%%%%%%%%%%%%%%%%%%%%%%%%%%%%%%

%
%%%%%%%%%%%%%%%%%%%%%%%%%%%%%%%%%%%%%%%%%%%%%%%%%%%%%%%%%%%%%%
\newpage
\subsection{Selection.}\label{sub:sub:sec:sel}
%$\ast\ast\ast$}
%%%%%%%%%%%%%%%%%%%%%%%%%%%%%%%%%%%%%%%%%%%%%%%%%%%%%%%%%%%%%%
%
%%%%%%%%%%%%%%%%%%%%%%%%%%%%%%%%%%%%%%%%%%%%%%%%%%%%%%%%%%%%%
%\newpage
%\paragraph{Selection.}
%%%%%%%%%%%%%%%%%%%%%%%%%%%%%%%%%%%%%%%%%%%%%%%%%%%%%%%%%%%%
%%%%%%%%%%%%%%%%%%%%%%%%%%%%%%%%%%%%%%%%%%%%%%%%%%%%%%%%%%%%
%%%%%%%%%%%%%%%%%%%%%%%%%%%%%%%%%%%%%%%%%%%%%%%%%%%%%%%%%%%%
\comment{Here we define selection 
in the fiber context $\mathrmbf{Tbl}(\mathcal{A})$.
We can define selection 
in the more general context $\mathrmbf{Tbl}(\mathcal{A})$
by prefacing with a restriction step.}
%%%%%%%%%%%%%%%%%%%%%%%%%%%%%%%%%%%%%%%%%%%%%%%%%%%%%%%%%%%%
%%%%%%%%%%%%%%%%%%%%%%%%%%%%%%%%%%%%%%%%%%%%%%%%%%%%%%%%%%%%
%%%%%%%%%%%%%%%%%%%%%%%%%%%%%%%%%%%%%%%%%%%%%%%%%%%%%%%%%%%%%%
%%%%%%%%%%%%%%%%%%%%%%%%%%%%%%%%%%%%%%%%%%%%%%%%%%%%%%%%%%%%%%
\comment{\label{binary:ops}
Examples\;\ref{bin:op:bnd} and \ref{bin:op:gph} are special cases.
The binary operations in the set 
$\{\,<,\leq ,=,\neq ,\geq ,\;>\}$
can be handled in a similar fashion to these examples.}
%%%%%%%%%%%%%%%%%%%%%%%%%%%%%%%%%%%%%%%%%%%%%%%%%%%%%%%%%%%%%%
%%%%%%%%%%%%%%%%%%%%%%%%%%%%%%%%%%%%%%%%%%%%%%%%%%%%%%%%%%%%%%
%
%
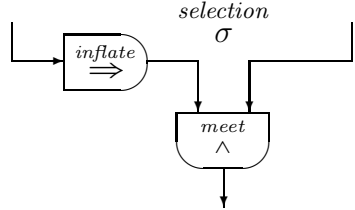
\begin{figure}
\begin{center}
{{{\begin{tabular}{c}
\begin{picture}(120,70)(40,28)
\setlength{\unitlength}{0.97pt}
%%%%%%%%%%%%%%%%%%%%%%%%%%%%%%%%%%%%%%%%%%%%%%%%%%
\put(54,65){\begin{picture}(0,0)(0,0)
\setlength{\unitlength}{0.35pt}
%\thicklines
%\put(106,40){\makebox(0,0){\normalsize{$\boldsymbol{\circ}$}}}
%\put(4.7,40){\makebox(0,0){\normalsize{$\boldsymbol{\circ}$}}}
\put(10,10){\line(1,0){60}}
\put(10,70){\line(1,0){60}}
\put(10,70){\line(0,-1){60}}
\put(70,40){\oval(60,60)[br]}
\put(70,40){\oval(60,60)[tr]}
\put(55,50){\makebox(0,0){\scriptsize{{\textit{{inflate}}}}}}
\put(56,30){\makebox(0,0){\Large{${\Rightarrow}$}}}
\end{picture}}
%%%%%%%%%%%%%%%%%%%%%%%%%%%%%%%%%%%%%%%%%%%%%%%%%%
\put(98,37.5){\begin{picture}(0,0)(0,3)
\setlength{\unitlength}{0.35pt}
\put(60,33){\makebox(0,0){\normalsize{$\wedge$}}}
%\thicklines
\put(40,10){\line(1,0){40}}
\put(10,70){\line(1,0){100}}
\put(10,70){\line(0,-1){30}}
\put(110,70){\line(0,-1){30}}
\put(40,40){\oval(60,60)[bl]}
\put(80,40){\oval(60,60)[br]}
\put(60,55){\makebox(0,0){\scriptsize{{\textit{{meet}}}}}}
\end{picture}}
%%%%%%%%%%%%%%%%%%%%%%%%%%%%%%%%%%%%%%%%%%%%%%%%%%
\put(120,100){\makebox(0,0){\footnotesize{{\textit{{selection}}}}}}
\put(120,90){\makebox(0,0){\large{$\sigma$}}}
%%%%%%%%%%%%%%%%%%%%%%%%%%%%%%%%%%%%%%%%%%%%%%%%%%
\put(37.5,80){\line(0,1){15}}
\put(37.5,80){\vector(1,0){20}}
\put(110,80){\line(-1,0){20}}\put(110,80){\vector(0,-1){20}}
\put(130,80){\line(1,0){40}}\put(130,80){\vector(0,-1){20}}
%\put(212,80){\vector(-1,0){20}}
\put(170,80){\line(0,1){15}}
\put(120,38){\vector(0,-1){15}}
%\thicklines
%\put(15,110){\line(1,0){190}}
%\put(55,10){\line(1,0){110}}
%\put(15,50){\line(0,1){60}}\put(205,50){\line(0,1){60}}
%\qbezier(15,50)(15,10)(55,10)
%\qbezier(165,10)(205,10)(205,50)
%%%%%%%%%%%%%%%%%%%%%%%%%%%%%%%%%%%%%%%%%%%%%%%%%%
\end{picture}
\end{tabular}}}}
\end{center}
\caption{\texttt{FOLE} Selection Flow Chart}
\label{fig:fole:select:aux:rel:flo:chrt}
\end{figure}
Let $\mathcal{A}$ be a fixed type domain.
% (conjunction, or intersection).
%We focus on tables in the context $\mathrmbf{Tbl}(\mathcal{A})$ 
%for fixed type domain $\mathcal{A}$.
General selection 
%
%%%%%%%%%%%%%%%%%%%%%%%%%%%%%%%%%%%%%%%%%%%%%%%%%%%%%%%%%%%%%%
%%%%%%%%%%%%%%%%%%%%%%%%%%%%%%%%%%%%%%%%%%%%%%%%%%%%%%%%%%%%%%
\footnote{\label{aux:rel}
Selection uses auxiliary relations.
See \S\,\ref{sub:sec:tbl:rel}
for some examples of auxiliary relations.}
%%%%%%%%%%%%%%%%%%%%%%%%%%%%%%%%%%%%%%%%%%%%%%%%%%%%%%%%%%%%%%
%%%%%%%%%%%%%%%%%%%%%%%%%%%%%%%%%%%%%%%%%%%%%%%%%%%%%%%%%%%%%%
%
is a binary operation 
$\sigma _{\mathcal{R}'}(\mathcal{T})$,
where $\mathcal{T}$ is a table
and $\mathcal{R}'$ is a relation
that may represent a propositional formula 
on sorts in the header of $\mathcal{T}$.
The relation $\mathcal{R}'$ might consists of atoms 
${n_{0} \theta n_{1}}$ or ${n \theta \hat{n}}$
as in the examples in \S\,\ref{sub:sec:tbl:rel},
plus 
the Boolean operators $\wedge$ (and), $\vee$ (or) and $\neg$ (negation). 
General selection selects all those tuples in $\mathcal{T}$ for which $\mathcal{R}'$ holds.
More precisely,
the selection $\sigma_{\mathcal{R}'}(\mathcal{T})$ 
denotes all tuples in $\mathcal{T}$ whose projection in
$\acute{\mathrmbfit{tbl}}_{\mathcal{A}}(h)(\mathcal{T})$
is also in $\mathcal{R}'$; 
equivalently,
a tuple is in the selection $\sigma_{\mathcal{R}'}(\mathcal{T})$ 
when it is 
in both $\mathcal{T}$ 
and the inflation $\grave{\mathrmbfit{tbl}}_{\mathcal{A}}(h)(\mathcal{R}')$.
%of $\mathcal{R}$.
%$\acute{\mathrmbfit{tbl}}_{\mathcal{A}}(h)(\mathcal{T}){\cap}\mathcal{R}$
%$\grave{\mathrmbfit{tbl}}_{\mathcal{A}}(h)(\mathcal{R}){\cap}\mathcal{T}$
%
%%%%%%%%%%%%%%%%%%%%%%%%%%%%%%%%%%%%%%%%%%%%%%%%%%%%%%%%%%%%%%%%%%%%%%
%%%%%%%%%%%%%%%%%%%%%%%%%%%%%%%%%%%%%%%%%%%%%%%%%%%%%%%%%%%%%%%%%%%%%%
\footnote{Selection is defined by using the fiber adjunction of tables
(Disp.\;\ref{def:fbr:adj:sign:mor} of \S\,\ref{sub:sub:sec:adj:flow:A})
of the $X$-sorted signature morphism 
$\mathcal{S}'\xrightarrow{\;h\;}\mathcal{S}$.
%tuple function
%{\footnotesize{$\mathrmbfit{tup}_{\mathcal{A}}(\mathcal{S}')
%\xleftarrow[h{\,\cdot\,}{(\mbox{-})}]{\mathrmbfit{tup}_{\mathcal{A}}(h)}
%\mathrmbfit{tup}_{\mathcal{A}}(\mathcal{S})
%$.}\normalsize}
%
The left adjoint fiber passage
$\mathrmbf{Tbl}_{\mathcal{A}}(\mathcal{S}')
\xleftarrow
%[{\scriptscriptstyle\sum}_{h}]
{\;\acute{\mathrmbfit{tbl}}_{\mathcal{A}}(h)\;}
\mathrmbf{Tbl}_{\mathcal{A}}(\mathcal{S})$
defines projection by composition.
The projection
$\mathcal{T}_{h} =
\acute{\mathrmbfit{tbl}}_{\mathcal{A}}(h)(\mathcal{T}) = 
{\langle{
%\mathcal{S}',\mathcal{A},
K,t{\,\cdot\,}\mathrmbfit{tup}_{\mathcal{A}}(h)}\rangle}$
consists of 
the columns of table $\mathcal{T}$ 
under sub-header $\mathcal{S}'$;
%%(the $h^{\text{th}}$ component of $\mathcal{T}$)
%%This contains only 
%The tuples in the projection
%$\mathcal{T}_{h}$
%are 
elements of the set 
$\mathrmbfit{tup}_{\mathcal{A}}(\mathcal{S}')$.
Some of these may be in $R$, and 
some in the complement $\mathrmbfit{tup}_{\mathcal{A}}(\mathcal{S}'){{-}}{R}$.
%}
%
The purpose of selection is to choose those tuples of $\mathcal{T}$
whose $\mathcal{S}'$-component is in $R$.}
%%%%%%%%%%%%%%%%%%%%%%%%%%%%%%%%%%%%%%%%%%%%%%%%%%%%%%%%%%%%
%%%%%%%%%%%%%%%%%%%%%%%%%%%%%%%%%%%%%%%%%%%%%%%%%%%%%%%%%%%%
%
Selection, within the context $\mathrmbf{Tbl}(\mathcal{A})$,
is inflation followed by meet.
%\newpage
%
\begin{description}
\item[Constraint:] 
Consider 
%The constraint for selection is 
an $X$-sorted signature morphism 
$\mathcal{S}'\xrightarrow{\;h\;}\mathcal{S}$
in $\mathrmbf{List}(X)$
consisting of an index function
$I'\xrightarrow{\;h\;}I$.
%satisfying the condition:
%$h{\;\cdot\;}s = s'$.
%
This forms a (trivial) $X$-sorted signature span 
$\mathcal{S}'\xleftarrow{{1}}\mathcal{S}'\xrightarrow{h}\mathcal{S}$,
which is a special case of 
the constraint for natural join (Tbl.\,\ref{tbl:fole:natural:join:input:output}).
This is the constraint for selection 
(Tbl.\;\ref{tbl:fole:selection:input:output}).
%:
%an $X$-sorted signature span 
%$\mathcal{S}_{1}\xleftarrow{h_{1}}\mathcal{S}\xrightarrow{h_{2}}\mathcal{S}_{2}$
%in $\mathrmbf{List}(X)$
%consisting of 
%a span of index functions
%$I_{1}\xleftarrow{h_{1}}I\xrightarrow{h_{2}}I_{2}$.
%
\newline
\item[Construction:] 
The 
%construction for selection is the 
pushout of this constraint
in $\mathrmbf{List}(X)$
is the (trivial) opspan
{\footnotesize{$\mathcal{S}'\xrightarrow{\,h\,} 
\mathcal{S}
\xleftarrow{\;1}\mathcal{S}$}\normalsize}
of injection $X$-signature morphisms
with pushout signature
$\mathcal{S}'{+_{\mathcal{S}'}}\mathcal{S}=\mathcal{S}$,
%(opspan)
%
which is a special case of the construction for natural join 
(Tbl.\,\ref{tbl:fole:natural:join:input:output}).
This is the construction for selection 
(Tbl.\;\ref{tbl:fole:selection:input:output}).
%,
%which is the opspan
%{\footnotesize{$\mathcal{S}_{1}\xrightarrow{\iota_{1}\,} 
%{\mathcal{S}_{1}{\!+_{\mathcal{S}}}\mathcal{S}_{2}}
%\xleftarrow{\;\iota_{2}}\mathcal{S}_{2}$}\normalsize}
%of injection $X$-signature morphisms
%with
%pushout signature
%$\mathcal{S}_{1}{+_{\mathcal{S}}}\mathcal{S}_{2}$.
%in $\mathrmbf{List}(X)$
%with injection $X$-signature morphisms
%
\newline
\item[Input:] 
Consider
a relation${}^{\ref{aux:rel}}$ 
$\mathcal{R}' = {\langle{R',i'}\rangle} \in \mathrmbf{Rel}_{\mathcal{A}}(\mathcal{S}')$
and a table
$\mathcal{T} = {\langle{K,t}\rangle} \in \mathrmbf{Tbl}_{\mathcal{A}}(\mathcal{S})$.
This is the input for selection
(Tbl.\;\ref{tbl:fole:selection:input:output}), 
a special case of the input for natural join 
(Tbl.\,\ref{tbl:fole:natural:join:input:output}),
since any relation is a table by reflection \S\,\ref{sub:sub:sec:reflect}.
\newline
\item[Output:] 
The output is in inflation followed by meet.
\newline
\begin{itemize}
\item 
%\begin{description}
%\item[inflation:] 
Inflation 
{\footnotesize{$
\mathrmbf{Tbl}_{\mathcal{A}}(\mathcal{S}')
{\;\xrightarrow
{\;\grave{\mathrmbfit{tbl}}_{\mathcal{A}}(h)\;}\;}
\mathrmbf{Tbl}_{\mathcal{A}}(\mathcal{S})
$}\normalsize}
(\S\,\ref{sub:sub:sec:adj:flow:A})
along the tuple function 
of the $X$-signature morphism
{\footnotesize{$\mathcal{S}'\xrightarrow{h\,}\mathcal{S}$}\normalsize}
maps the relation $\mathcal{R}'$
to the $\mathcal{A}$-table
$\widehat{\mathcal{R}}'
= \grave{\mathrmbfit{tbl}}_{\mathcal{A}}(h)(\mathcal{R}')
%\iota_{1}^{\ast}(\mathcal{T}_{1})
= {\langle{\widehat{R}',\hat{i}'}\rangle} 
\in \mathrmbf{Tbl}_{\mathcal{A}}(\mathcal{S})$
with its tuple function
$\widehat{R}' \xrightarrow{\hat{i}'} 
\mathrmbfit{tup}_{\mathcal{A}}(\mathcal{S})$
defined by pullback,
$\grave{k}{\,\cdot\,}i' 
= \hat{i}'{\,\cdot\,}\mathrmbfit{tup}_{\mathcal{A}}(h)$. 
This is linked to the relation $\mathcal{R}'$ 
by the $\mathcal{A}$-table morphism 
%\[\mbox
{\footnotesize{{$
\mathcal{R}' = {\langle{\mathcal{S}',R',i'}\rangle}
\xleftarrow{{\langle{h,\grave{k}}\rangle}} 
{\langle{\mathcal{S},\widehat{R}',\hat{i}'}\rangle} 
%= {\iota}_{1}^{\ast}(\mathcal{T}_{1})
= \widehat{\mathcal{R}}'
$.}}\normalsize}
%\]
%Similarly for $\mathcal{A}$-table
%$\mathcal{T}_{2} = {\langle{K_{2},t_{2}}\rangle} \in 
%\mathrmbf{Tbl}_{\mathcal{A}}(\mathcal{S}_{2})$.
Inflation of table $\mathcal{T}$ along the identity signature morphism is identity.
\newline
%\newline
%\item[meet:] 
%
\item 
Intersection (\S\,\ref{sub:sub:sec:boole})
of the two inflation tables $\widehat{\mathcal{R}}'$
%{i}_{1}^{\ast}(\mathcal{T}_{1})$
and $\mathcal{T}$
%{i}_{2}^{\ast}(\mathcal{T}_{2})$ 
in the context 
$\mathrmbf{Tbl}_{\mathcal{A}}(\mathcal{S})$
defines the selection table
$\sigma_{\mathcal{R}'}(\mathcal{T})
%= \mathcal{R}  
\doteq \widehat{\mathcal{R}}'{\;\wedge\;}\mathcal{T}$, 
whose key set $\widehat{R}$ is the pullback and 
whose tuple map is the mediating function 
$\widehat{R} 
%\subseteq \widehat{K}_{1}{\times}\widehat{K}_{2}
\xrightarrow{{(\hat{\iota}',t)}}
\mathrmbfit{tup}_{\mathcal{A}}(\mathcal{S})$
of the opspan
$\widehat{R}'\xrightarrow{\hat{i}'}
\mathrmbfit{tup}_{\mathcal{A}}(\mathcal{S})
\xleftarrow{t}\widehat{K}$,
resulting in the span 
%of 
%${\langle{I_{1}{+}_{I}I_{2},[s_{1},s_{2}],\mathcal{A}}\rangle}$-
%table morphisms
%\newline\mbox{}\hfill
%\rule[-10pt]{0pt}{26pt}\bowtie
%\[\mbox
{\footnotesize{
{$
\widehat{\mathcal{R}}'
%{i}_{1}^{\ast}(\mathcal{T}_{1})
%= {\langle{\widehat{R}',\hat{i}'}\rangle} 
\xleftarrow{\;\hat{\pi}_{1}\;} 
\sigma_{\mathcal{R}'}(\mathcal{T})
%\mathcal{R}
%\mathcal{R}'{\,\bowtie_{\mathcal{S}}}\mathcal{T}
\xrightarrow{\;\hat{\pi}_{2}\;} 
%{\langle{K,t}\rangle} = 
\mathcal{T}
%{i}_{2}^{\ast}(\mathcal{T}_{2})
$.}}\normalsize}
%\]
\newline
\end{itemize}
Inflation composed with meet 
defines the span of $\mathcal{A}$-table morphisms
%(Fig.\;\ref{fig:fole:select:aux:rel})
%
\[\mbox{\footnotesize{
{$\mathcal{R}'
\xleftarrow[\;\hat{\pi}_{1}{\circ\,}{\langle{h,\grave{k}}\rangle}\;]
{\;{\langle{h,\hat{k}}\rangle}\;} 
\sigma_{\mathcal{R}'}(\mathcal{T})
%\mathcal{R}
%\mathcal{R}'{\,\bowtie_{\mathcal{S}}}\mathcal{T}
\xrightarrow{\;\hat{\pi}_{2}\;} 
\mathcal{T}
$,}}\normalsize}\]
%
%illustrated in 
%Fig.\;\ref{fig:fole:nat:join}.
which is the output for selection.
%natural join (Tbl.\,\ref{tbl:fole:natural:join:input:output}).
%
%\end{description}
%
This is a special case 
$\sigma_{\mathcal{R}'}(\mathcal{T})  
%= \mathcal{R}
= {\mathcal{R}'}{\,\boxtimes_{\mathcal{A}}}{\mathcal{T}}$
of the output for natural join 
(Tbl.\,\ref{tbl:fole:natural:join:input:output}).
%\end{description}
%
\end{description}
The selection flowchart input/output is displayed in 
Tbl.\,\ref{tbl:fole:selection:input:output}.
Selection within the context $\mathrmbf{Tbl}(\mathcal{A})$
is inflation followed by meet.
% (conjunction, or intersection).
This is a two-step process
%(illustrated in Fig.\;\ref{fig:fole:select:aux:rel})
\newline\mbox{}\hfill
\rule[-10pt]{0pt}{26pt}
$\sigma_{\mathcal{R}'}(\mathcal{T}) 
\doteq 
\grave{\mathrmbfit{tbl}}_{\mathcal{A}}(h)(\mathcal{R}'){\;\wedge\;}\mathcal{T}$. 
\hfill\mbox{}\newline

% selection figure
%%%%%%%%%%%%%%%%%%%%%%%%%%%%%%%%%%%%%%%%%%%%%%%%%%%%%%%%%%%%%%%%%%%%%%%%%%%%%%%%
%%%%%%%%%%%%%%%%%%%%%%%%%%%%%%%%%%%%%%%%%%%%%%%%%%%%%%%%%%%%%%%%%%%%%%%%%%%%%%%%

%
\begin{table}
\begin{center}
{{\fbox{\begin{tabular}{c}
\setlength{\extrarowheight}{2pt}
{\scriptsize{$\begin{array}[c]{c@{\hspace{12pt}}l}
\mathcal{S}'\xleftarrow{{1}}\mathcal{S}'\xrightarrow{h}\mathcal{S}
&
\textit{constraint}
\\
\mathcal{S}'\xrightarrow{\,h\,} 
\mathcal{S}
\xleftarrow{\;1}\mathcal{S}
&
\textit{construction}
\\
\hline
\mathcal{R}' 
%= {\langle{R',i'}\rangle} 
\in \mathrmbf{Rel}_{\mathcal{A}}(\mathcal{S}')
\text{ and }
\mathcal{T} 
%= {\langle{K,t}\rangle} 
\in \mathrmbf{Tbl}_{\mathcal{A}}(\mathcal{S})
&
\textit{input}
\\
\mathcal{R}'
\xleftarrow
%[\;\hat{\iota}{\circ\,}{\langle{h,\grave{k}}\rangle}\;]
{\;{\langle{h,\hat{k}}\rangle}\;} 
\sigma_{\mathcal{R}'}(\mathcal{T})
\xrightarrow{\;\hat{\pi}_{2}\;} 
\mathcal{T}
&
\textit{output}
\end{array}$}}
\end{tabular}}}}
\end{center}
\caption{\texttt{FOLE} Selection I/O}
\label{tbl:fole:selection:input:output}
\end{table}

\begin{proposition}\label{select:preserve:join:meet}
Selection $\sigma$ preserves union $\vee$ and intersection $\wedge$.
\end{proposition}
\begin{proof}
Same proof as Prop.\;\ref{join:preserve:join:meet}.
\hfill\rule{5pt}{5pt}
\end{proof}
%

%
%%%%%%%%%%%%%%%%%%%%%%%%%%%%%%%%%%%%%%%%%%%%%%%%%%
%\newpage
%\paragraph{Wikipedia Explanation.}
%%%%%%%%%%%%%%%%%%%%%%%%%%%%%%%%%%%%%%%%%%%%%%%%%%
%

%
%%%%%%%%%%%%%%%%%%%%%%%%%%%%%%%%%%%%%%%%%%%%%%%%%%
%\newpage
%\paragraph{Higher Level: Selection via Formalism.}
%%%%%%%%%%%%%%%%%%%%%%%%%%%%%%%%%%%%%%%%%%%%%%%%%%
%We use the \textbf{yin} definition in \S\,\ref{sub:sub:sec:typ:dom:yin:yang}.

%%%%%%%%%%%%%%%%%%%%%%%%%%%%%%%%%%%%%%%%%%%%%%%%%%%%%%%%%%%%%%
\newpage
\subsection{Select-join.}\label{sub:sub:sec:sel:join}
%$\ast{-}\ast$}
%\subsubsection{$\theta$-join and Equi-join.}
%%%%%%%%%%%%%%%%%%%%%%%%%%%%%%%%%%%%%%%%%%%%%%%%%%%%%%%%%%%%%%
%
%%%%%%%%%%%%%%%%%%%%%%%%%%%%%%%%%%%%%%%%%%%%%%%%%%
%\newpage
%\paragraph{Select-join.}
%%%%%%%%%%%%%%%%%%%%%%%%%%%%%%%%%%%%%%%%%%%%%%%%%%

%
\begin{figure}
\begin{center}
{{{\begin{tabular}{c}
\begin{picture}(100,90)(80,28)
\setlength{\unitlength}{0.97pt}
%%%%%%%%%%%%%%%%%%%%%%%%%%%%%%%%%%%%%%%%%%%%%%%%%%%%%%%%%%%%
\put(148,94.5){\begin{picture}(0,0)(0,3)
\setlength{\unitlength}{0.35pt}
\put(60,23){\makebox(0,0){\normalsize{$\boxtimes$}}}
%\thicklines
%\put(33,76){\makebox(0,0){\normalsize{$\boldsymbol{\circ}$}}}
%\put(87,76){\makebox(0,0){\normalsize{$\boldsymbol{\circ}$}}}
%\put(60,3){\makebox(0,0){\normalsize{$\boldsymbol{\circ}$}}}
\put(40,10){\line(1,0){40}}
\put(10,70){\line(1,0){100}}
\put(10,70){\line(0,-1){30}}
\put(110,70){\line(0,-1){30}}
\put(40,40){\oval(60,60)[bl]}
\put(80,40){\oval(60,60)[br]}
\put(61,59){\makebox(0,0){\scriptsize{{\textit{{natural}}}}}}
\put(61,42){\makebox(0,0){\scriptsize{{\textit{{join}}}}}}
\end{picture}}
%%%%%%%%%%%%%%%%%%%%%%%%%%%%%%%%%%%%%%%%%%%%%%%%%%%%%%%%%%%%
\put(98,37.5){\begin{picture}(0,0)(0,3)
\setlength{\unitlength}{0.35pt}
\put(60,33){\makebox(0,0){\normalsize{$\sigma$}}}
%\thicklines
\put(40,10){\line(1,0){40}}
\put(10,70){\line(1,0){100}}
\put(10,70){\line(0,-1){30}}
\put(110,70){\line(0,-1){30}}
\put(40,40){\oval(60,60)[bl]}
\put(80,40){\oval(60,60)[br]}
\put(60,55){\makebox(0,0){\scriptsize{{\textit{{select}}}}}}
\end{picture}}
%%%%%%%%%%%%%%%%%%%%%%%%%%%%%%%%%%%%%%%%%%%%%%%%%%%%%%%%%%%%
\put(120,100){\makebox(0,0){\footnotesize{{\textit{{select}}}-{\textit{{join}}}}}}
%\put(120,90){\makebox(0,0){\large{$\sigma$}}}
\put(122,90){\makebox(0,0){\large{$\hat{\boxtimes}$}}}
%%%%%%%%%%%%%%%%%%%%%%%%%%%%%%%%%%%%%%%%%%%%%%%%%%%%%%%%%%%%
%\put(37.5,80){\line(0,1){15}}
%\put(37.5,80){\vector(1,0){20}}
\put(110,80){\line(-1,0){20}}\put(110,80){\vector(0,-1){20}}
\put(130,80){\vector(0,-1){20}}\put(130,80){\line(1,0){40}}
%\put(212,80){\vector(-1,0){20}}
\put(170,80){\line(0,1){15}}
\put(120,38){\vector(0,-1){15}}
\put(160,132){\vector(0,-1){15}}
\put(179.8,132){\vector(0,-1){15}}
%\thicklines
%\put(15,110){\line(1,0){190}}
%\put(55,10){\line(1,0){110}}
%\put(15,50){\line(0,1){60}}\put(205,50){\line(0,1){60}}
%\qbezier(15,50)(15,10)(55,10)
%\qbezier(165,10)(205,10)(205,50)
%%%%%%%%%%%%%%%%%%%%%%%%%%%%%%%%%%%%%%%%%%%%%%%%%%%%%%%%%%%%
%%%%%%%%%%%%%%%%%%%%%%%%%%%%%%%%%%%%%%%%%%%%%%%%%%%%%%%%%%%%
\end{picture}
\end{tabular}}}}
\end{center}
\caption{\texttt{FOLE} Select-Join Flow Chart}
\label{fig:fole:theta:join:flo:chrt}
\end{figure}
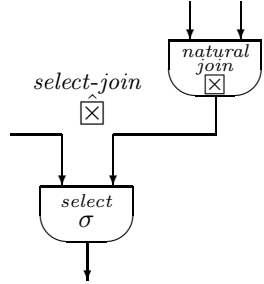
If we want to combine tuples from two tables, 
where the combination condition is not simply the equality of shared attributes, 
then it is convenient to have a more general form of the natural join operator, 
which is known as the select-join.
The select-join is a ternary operator that is written as 
$\mathcal{T}_{1}{\,\hat{\boxtimes}^{\mathcal{R}}_{\mathcal{A}}}\mathcal{T}_{2}$.
%In computing the natural join 
%in \S\,\ref{sub:sub:sec:nat:join}
%we formed 
%
The result of the select-join operation consists of 
all combinations of tuples in $\mathcal{T}_{1}$ and $\mathcal{T}_{2}$ 
that satisfy $\mathcal{R}$. 
%The result of the $\theta$-join is defined only if the attributes being $\theta$-joined, 
%namely $a$ and $b$, are distinct.
%headers of $S$ and $R$ are disjoint; 
%that is, 
%do not contain a common attribute.
%
The definition of this operation in terms of more fundamental operations 
%(see Fig.\;\ref{fole:theta:join})
is
$\mathcal{T}_{1}{\,\hat{\boxtimes}^{\mathcal{R}}_{\mathcal{A}}}\mathcal{T}_{2} = 
\sigma_{\mathcal{R}}\bigl(\mathcal{T}_{1}{\,\boxtimes_{\mathcal{A}}}\mathcal{T}_{2}\bigr)$.
%
%%%%%%%%%%%%%%%%%%%%%%%%%%%%%%%%%%%%%%%%%%%%%%%%%%%%%%%%%%%%
%%%%%%%%%%%%%%%%%%%%%%%%%%%%%%%%%%%%%%%%%%%%%%%%%%%%%%%%%%%%
\footnote{The $\theta$-join is a special case,
where the headers of the table $\mathcal{T}_{1}$ and $\mathcal{T}_{2}$ are disjoint,
so that we use the Cartesian product.
When the relation $\mathcal{R}$ represents the equality operator $(=)$ on two attributes,
this join is called an equi-join.}
%%%%%%%%%%%%%%%%%%%%%%%%%%%%%%%%%%%%%%%%%%%%%%%%%%%%%%%%%%%%
%%%%%%%%%%%%%%%%%%%%%%%%%%%%%%%%%%%%%%%%%%%%%%%%%%%%%%%%%%%%
%
The input/output for the natural join component is independent,
with the input/output for the selection component depending on it.
\begin{description}
\item[Constraint/Construction:] 
The constraint and construction is an interleaved process 
for the natural join and the selection aspects.
\begin{enumerate}
\item 
Consider an $X$-sorted signature span 
$\mathcal{S}_{1}\xleftarrow{h_{1}}\mathcal{S}\xrightarrow{h_{2}}\mathcal{S}_{2}$
in $\mathrmbf{List}(X)$
consisting of 
a span of index functions
$I_{1}\xleftarrow{h_{1}}I\xrightarrow{h_{2}}I_{2}$.
This is the constraint for the natural join aspect
(Tbl.\,\ref{tbl:fole:natural:join:input:output}).
This has pushout signature
$\mathcal{S}_{1}{+_{\mathcal{S}}}\mathcal{S}_{2}$
and injection $X$-signature morphisms
{\footnotesize{$
\mathcal{S}_{1} \xrightarrow{\iota_{1}\,} 
{\mathcal{S}_{1}{+_{\mathcal{S}}}\mathcal{S}_{2}}
\xleftarrow{\;\iota_{2}} \mathcal{S}_{2}
$}\normalsize}
with
index function opspan
{\footnotesize{$I_{1} 
\xrightarrow{\iota_{1}\,} 
{\langle{I_{1}{+}_{I}I_{2},[s_{1},s_{2}]}\rangle}
\xleftarrow{\;\iota_{2}}I_{2}.$.}\normalsize}
\item 
Consider a connecting $X$-sorted signature morphism 
$\mathcal{S}_{3}
%=\mathcal{S}'
\xrightarrow{\;\iota_{3}\;}
\mathcal{S}_{1}{+_{\mathcal{S}}}\mathcal{S}_{2}$
in $\mathrmbf{List}(X)$,
which is the RHS of a trivial $X$-sorted signature span.
This is the constraint for the selection aspect.
The construction for the selection aspect
is the pushout of this constraint
in $\mathrmbf{List}(X)$,
which is the (trivial) opspan
whose LHS
is the $X$-sorted signature morphism 
$\mathcal{S}_{3}
\xrightarrow{\;\iota_{3}\;}
\mathcal{S}_{1}{+_{\mathcal{S}}}\mathcal{S}_{2}$.
\end{enumerate}
%
%\item[Construction:] 
\mbox{}
\newpage
\item[Input/Output:] 
The input and output is an interleaved process 
for the natural join and the selection aspects.
\begin{enumerate}
\item 
Consider a pair of tables
$\mathcal{T}_{1} = {\langle{K_{1},t_{1}}\rangle} \in 
\mathrmbf{Tbl}_{\mathcal{A}}(\mathcal{S}_{1})$
and
$\mathcal{T}_{2} = {\langle{K_{2},t_{2}}\rangle} \in 
\mathrmbf{Tbl}_{\mathcal{A}}(\mathcal{S}_{2})$.
This is the input for natural join (Tbl.\,\ref{tbl:fole:natural:join:input:output}).
The natural join from \S\,\ref{sub:sub:sec:nat:join}
%$\mathcal{T}_{1}{\times}_{\mathcal{T}}\mathcal{T}_{2}$
%=
%{\langle{{\mathcal{S}{+_{\mathcal{S}''}}\mathcal{S}'},K{\times}_{K}K',
%t{\times}_{t''}t'}\rangle}$
is the top of the $\mathcal{A}$-table span 
%displayed in diagram~\ref{tbl:opspan},
%\mathcal{T}_{1}{\,\hat{\boxtimes}^{\mathcal{R}}_{\mathcal{T}}}\mathcal{T}_{2}
\newline\mbox{}\hfill
{\footnotesize{$
%\begin{align}\label{tbl:span:redux}
\mathcal{T}_{1}={\langle{\mathcal{S}_{1},K_{1},t_{1}}\rangle} 
\xleftarrow{\langle{\iota_{1},\hat{k}_{1}}\rangle} 
\mathcal{T}_{1}{\,\boxtimes_{\mathcal{A}}}\mathcal{T}_{2}
%\overset{\textstyle{\mathcal{T}{\bowtie}_{\mathcal{S}}\mathcal{T}_{2}}}
%{\overbrace{\langle{{\mathcal{S}{+_{\mathcal{S}_{2}_{2}}}\mathcal{S}_{2}},K{\times}_{K}K_{2},
%t{\times}_{t_{2}_{2}}t_{2}}\rangle}}
\xrightarrow{\langle{\iota_{2},\hat{k}_{2}}\rangle} 
{\langle{\mathcal{S}_{2},K_{2},t_{2}}\rangle}=\mathcal{T}_{2}
$}\normalsize}
\hfill\mbox{}\newline
with underlying $X$-signature opspan
\newline\mbox{}\hfill
{\footnotesize{$
\mathcal{S}_{1} 
%= {\langle{I_{1},s_{1}}\rangle} 
\xrightarrow{\iota_{1}\,} 
%\overset{\textstyle
{\mathcal{S}_{1}{+_{\mathcal{S}}}\mathcal{S}_{2}}
%}{\overbrace{{\langle{I_{1}{+}_{I}I_{2},[s_{1},s_{2}]}\rangle}}} 
\xleftarrow{\;\iota_{2}} 
%{\langle{I_{2},s_{2}}\rangle} = 
\mathcal{S}_{2}.
$}\normalsize}
\hfill\mbox{}\newline
%\newpage
This is the output for natural join (Tbl.\,\ref{tbl:fole:natural:join:input:output}).
\item 
Consider a relation
$\mathcal{R}={\langle{R,i}\rangle} \in 
\mathrmbf{Rel}_{\mathcal{A}}(\mathcal{S}_{3})$
and the table
${\mathcal{T}_{1}}{\,\boxtimes_{\mathcal{A}}}{\mathcal{T}_{2}}
= {\langle{\widehat{K}_{12},{(\hat{t}_{1},\hat{t}_{2})}}\rangle}
\in 
\mathrmbf{Tbl}_{\mathcal{A}}(\mathcal{S}_{1}{+_{\mathcal{S}}}\mathcal{S}_{2})$.
This is the input for selection (Tbl.\,\ref{tbl:fole:selection:input:output}).
Selection is inflation followed by meet.
The $\mathcal{A}$-relation
$\mathcal{R}={\langle{R,i}\rangle} \in 
\mathrmbf{Rel}_{\mathcal{A}}(\mathcal{S}_{3})$
is mapped by inflation to the $\mathcal{A}$-relation
$\widehat{\mathcal{R}} = 
\grave{\mathrmbfit{tbl}}_{\mathcal{A}}(\iota_{3})(\mathcal{R}) = 
%\iota_{3}^{\ast}(\mathcal{R}) =
{\langle{\widehat{R},\hat{i}}\rangle} 
\in \mathrmbf{Rel}_{\mathcal{A}}(\mathcal{S}_{1}{+_{\mathcal{S}}}\mathcal{S}_{2})$
with its tuple function
$\widehat{R} \xhookrightarrow{\hat{i}} \mathrmbfit{tup}_{\mathcal{A}}(\mathcal{S}_{1}{+_{\mathcal{S}}}\mathcal{S}_{2})$.
%defined by pullback,
%$k'{\,\cdot\,}i' 
%= \hat{\imath}{\,\cdot\,}\mathrmbfit{tup}_{\mathcal{A}}(h)$. 
%
This is linked to relation $\mathcal{R}$ by the $\mathcal{A}$-relation morphism 
%\[\mbox
{\footnotesize{{$
\mathcal{R} = {\langle{\mathcal{S}_{3},R,i}\rangle}
\xleftarrow{{\langle{\iota_{3},\grave{r}}\rangle}} 
{\langle{\mathcal{S}_{1}{+_{\mathcal{S}}}\mathcal{S}_{2},\hat{R},\hat{i}}\rangle} = \widehat{\mathcal{R}}
$.}}\normalsize}
%\]
%
Intersection, within the context $\mathrmbf{Tbl}_{\mathcal{A}}(\mathcal{S}_{1}{+_{\mathcal{S}}}\mathcal{S}_{2})$,
of $\mathcal{A}$-relation $\widehat{\mathcal{R}}$
with $\mathcal{A}$-table $\mathcal{T}_{1}{\,\boxtimes_{\mathcal{A}}}\mathcal{T}_{2}$
results in the table
\[\mbox{\footnotesize{{$
\sigma_{\mathcal{R}}(\mathcal{T}_{1}{\,\boxtimes_{\mathcal{A}}}\mathcal{T}_{2})
= 
\iota_{3}^{\ast}(\mathcal{R})
{\;\wedge\;}\bigl(\mathcal{T}_{1}{\,\boxtimes_{\mathcal{A}}}\mathcal{T}_{2}\bigr)
=
\iota_{3}^{\ast}(\mathcal{R})
{\;\wedge\;}\iota_{1}^{\ast}(\mathcal{T}_{1})
{\;\wedge\;}\iota_{2}^{\ast}(\mathcal{T}_{2})
$.}}\normalsize}\]
%\hfill\mbox{}
%\newline
%
This defines the span of $\mathcal{A}$-table morphisms
\[\mbox{\footnotesize{{$
\mathcal{R}
%\xleftarrow{\;{\langle{\iota_{3},\grave{r}}\rangle}\;} 
%\hat{\mathcal{T}}_{3}
%\xleftarrow{\;\hat{k}\;} 
\xleftarrow[\;\hat{\pi}_{3}{\,\circ\,}{\langle{\iota_{3},\grave{r}}\rangle}\;]
{\;{\langle{\iota_{3},\hat{k}_{3}}\rangle}\;} 
\sigma_{\mathcal{R}}(\mathcal{T}_{1}{\,\boxtimes_{\mathcal{A}}}\mathcal{T}_{2})
= \mathcal{T}_{1}{\,\hat{\boxtimes}^{\mathcal{R}}_{\mathcal{A}}}\mathcal{T}_{2}
\xrightarrow{\;\hat{\pi}\;} 
\mathcal{T}_{1}{\,\boxtimes_{\mathcal{A}}}\mathcal{T}_{2}
$,}}\normalsize}\]
which is the output for select-join.
\end{enumerate}
\end{description}
The select join flowchart input/output is displayed in 
Tbl.\,\ref{tbl:fole:select:join:input:output}.
Select-join 
%within the context $\mathrmbf{Tbl}(\mathcal{A})$
is natural join, followed by selection.
This is the two-step process
%(illustrated in Fig.\;\ref{fig:fole:select:aux:rel})
%
\newline\mbox{}\hfill
\rule[-10pt]{0pt}{26pt}
$\mathcal{T}_{1}{\,\hat{\boxtimes}^{\mathcal{R}}_{\mathcal{A}}}\mathcal{T}_{2} 
\doteq 
\sigma_{\mathcal{R}}\bigl(\mathcal{T}_{1}{\,\boxtimes_{\mathcal{A}}}\mathcal{T}_{2}\bigr)$. 
\hfill\mbox{}\newline
\begin{table}
\begin{center}
{{\fbox{\begin{tabular}{c}
\setlength{\extrarowheight}{2pt}
{\scriptsize{$\begin{array}[c]{c@{\hspace{12pt}}l}
\mathcal{S}_{1}\xleftarrow{h_{1}}\mathcal{S}\xrightarrow{h_{2}}\mathcal{S}_{2}
&
\textit{constraint nat-join}
\\
\mathcal{S}_{1} \xrightarrow{\iota_{1}\,} 
{\mathcal{S}_{1}{\!+_{\mathcal{S}}}\mathcal{S}_{2}}
\xleftarrow{\;\iota_{2}}\mathcal{S}_{2}
&
\textit{construction nat-join}
\\\cline{2-2}
\mathcal{S}_{3}\xleftarrow{\;1\;}
\mathcal{S}_{3}
\xrightarrow{\;\iota_{3}\;}\mathcal{S}_{1}{+_{\mathcal{S}}}\mathcal{S}_{2}
&
\textit{constraint select}
\\
\mathcal{S}_{3}\xrightarrow{\;\iota_{3}\;}
\mathcal{S}_{1}{+_{\mathcal{S}}}\mathcal{S}_{2}
\xleftarrow{\;1\;}\mathcal{S}_{1}{+_{\mathcal{S}}}\mathcal{S}_{2}
&
\textit{construction select}
\\
\hline
\mathcal{T}_{1}\in\mathrmbf{Tbl}_{\mathcal{A}}(\mathcal{S}_{1})
\text{ and }
\mathcal{T}_{1}\in\mathrmbf{Tbl}_{\mathcal{A}}(\mathcal{S}_{1})
&
\textit{input nat-join}
\\
\mathcal{T}
\xleftarrow{{\langle{\iota_{1},\hat{k}_{1}}\rangle}} 
\mathcal{T}_{1}{\,\boxtimes_{\mathcal{A}}}\mathcal{T}_{2}
\xrightarrow{{\langle{\iota_{2},\hat{k}_{2}}\rangle}} 
\mathcal{T}_{2}
&
\textit{output nat-join}
\\
\cline{2-2}
\mathcal{R}
%= {\langle{R,i}\rangle} 
\in \mathrmbf{Rel}_{\mathcal{A}}(\mathcal{S}_{3})
&
\textit{input select}
\\
\mathcal{R}
\xleftarrow
%[\;\hat{k}{\,\circ\,}{\langle{\iota_{3},\grave{r}}\rangle}\;]
{\;{\langle{\iota_{3},\hat{k}_{3}}\rangle}\;} 
%\sigma_{\mathcal{R}}(\mathcal{T}_{1}{\,\boxtimes_{\mathcal{A}}}\mathcal{T}_{2})
\mathcal{T}_{1}{\,\hat{\boxtimes}^{\mathcal{R}}_{\mathcal{A}}}\mathcal{T}_{2} 
\xrightarrow{\;\hat{\pi}\;} 
\mathcal{T}_{1}{\,\boxtimes_{\mathcal{A}}}\mathcal{T}_{2}
&
\textit{output select}
\end{array}$}}
\end{tabular}}}}
\end{center}
\caption{\texttt{FOLE} Select-Join I/O}
\label{tbl:fole:select:join:input:output}
\end{table}

\begin{note}
From another standpoint,
select-join is a form of multi-join,
where we start with the star-shaped diagram of $X$-signature morphisms
below-left,
and end with the
star-shaped diagram of table morphisms below-right.
\begin{center}
{{\begin{tabular}{c@{\hspace{40pt}}c}
%%%%%%%%%%%%%%%%%%%%%%%%%%%%%%%%%%%%%%%%%%%%%%%%%%%%%%%%%%%%
%%%%%%%%%%%%%%%%%%%%%%%%%%%%%%%%%%%%%%%%%%%%%%%%%%%%%%%%%%%%
{{\begin{tabular}{c}
\setlength{\unitlength}{0.8pt}
\begin{picture}(120,50)(0,-7)
\put(0,0){\makebox(0,0){\scriptsize{$\mathcal{S}_{1}$}}}
\put(120,0){\makebox(0,0){\scriptsize{$\mathcal{S}_{2}$}}}
\put(60,40){\makebox(0,0){\scriptsize{$\mathcal{S}_{3}$}}}
\put(62,0){\makebox(0,0){\scriptsize{${\mathcal{S}_{1}{+_{\mathcal{S}}}\mathcal{S}_{2}}$}}}
\put(25,-8){\makebox(0,0){\tiny{$\iota_{1}$}}}
\put(100,-8){\makebox(0,0){\tiny{$\iota_{2}$}}}
\put(54,20){\makebox(0,0)[r]{\tiny{$\iota_{3}$}}}
\put(10,0){\vector(1,0){24}}
\put(108,0){\vector(-1,0){24}}
\put(58,30){\vector(0,-1){20}}
%\qbezier(50,-30)(55,-25)(60,-20)
%
\end{picture}
\end{tabular}}}
%%%%%%%%%%%%%%%%%%%%%%%%%%%%%%%%%%%%%%%%%%%%%%%%%%%%%%%%%%%%
&
%%%%%%%%%%%%%%%%%%%%%%%%%%%%%%%%%%%%%%%%%%%%%%%%%%%%%%%%%%%%
{{\begin{tabular}{c}
\setlength{\unitlength}{0.8pt}
\begin{picture}(120,50)(0,-7)
\put(0,0){\makebox(0,0){\scriptsize{$\mathcal{T}_{1}$}}}
\put(120,0){\makebox(0,0){\scriptsize{$\mathcal{T}_{2}$}}}
\put(60,40){\makebox(0,0){\scriptsize{$\mathcal{R}$}}}
\put(62,0){\makebox(0,0){\scriptsize{$
\mathcal{T}_{1}{\,\hat{\boxtimes}^{\mathcal{R}}_{\mathcal{A}}}\mathcal{T}_{2}$}}}
\put(25,-8){\makebox(0,0){\tiny{${\langle{\iota_{1},\hat{k}_{1}}\rangle}$}}}
\put(100,-8){\makebox(0,0){\tiny{${\langle{\iota_{2},\hat{k}_{2}}\rangle}$}}}
\put(54,20){\makebox(0,0)[r]{\tiny{${\langle{\iota_{3},\hat{k}_{3}}\rangle}$}}}
\put(34,0){\vector(-1,0){24}}
\put(86,0){\vector(1,0){24}}
\put(58,10){\vector(0,1){20}}
%\qbezier(50,-30)(55,-25)(60,-20)
%
\end{picture}
\end{tabular}}}
%%%%%%%%%%%%%%%%%%%%%%%%%%%%%%%%%%%%%%%%%%%%%%%%%%%%%%%%%%%%
%%%%%%%%%%%%%%%%%%%%%%%%%%%%%%%%%%%%%%%%%%%%%%%%%%%%%%%%%%%%
\end{tabular}}}
\end{center}
Here,
select-join 
can also be defined by the two-step process
\[\mbox{\footnotesize{{$
\mathcal{T}_{1}{\,\hat{\boxtimes}^{\mathcal{R}}_{\mathcal{A}}}\mathcal{T}_{2} 
\doteq
\grave{\mathrmbfit{tbl}}_{\mathcal{A}}(\imath_{3})(\mathcal{T}_{1})
{\;\wedge\;}
\grave{\mathrmbfit{tbl}}_{\mathcal{A}}(\imath_{3})(\mathcal{T}_{2})
{\;\wedge\;}
\grave{\mathrmbfit{tbl}}_{\mathcal{A}}(\imath_{3})(\mathcal{R})
$}}\normalsize}\]
%}
of inflation (thrice) followed by multi-meet.
\end{note}
\begin{proposition}\label{select:join:preserve:join:meet}
Select-join $\hat{\boxtimes}$ preserves union $\vee$ and intersection $\wedge$.
\end{proposition}
\begin{proof}
Prop.\;\ref{join:preserve:join:meet}
and
Prop.\;\ref{select:preserve:join:meet}.
\hfill\rule{5pt}{5pt}
\end{proof}
%

%%%%%%%%%%%%%%%%%%%%%%%%%%%%%%%%%%%%%%%%%%%%%%%%%%%%%%%%%%%%%
\newpage
\subsection{Filtered Join.}
\label{sub:sub:sec:filtered:join}
%%%%%%%%%%%%%%%%%%%%%%%%%%%%%%%%%%%%%%%%%%%%%%%%%%%%%%%%%%%%%
%$\bigstar$ 

%
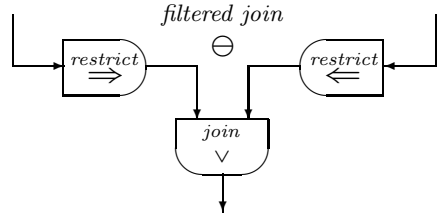
\begin{figure}
\begin{center}
{{{\begin{tabular}{c}
\begin{picture}(160,75)(37,27)
\setlength{\unitlength}{0.97pt}
%%%%%%%%%%%%%%%%%%%%%%%%%%%%%%%%%%%%%%%%%%%%%%%%%%
%\put(44,62){\begin{picture}(0,0)(0,0)
%\setlength{\unitlength}{0.46pt}
\put(54,65){\begin{picture}(0,0)(0,0)
\setlength{\unitlength}{0.35pt}
%\thicklines
%\put(106,40){\makebox(0,0){\normalsize{$\boldsymbol{\circ}$}}}
%\put(4.7,40){\makebox(0,0){\normalsize{$\boldsymbol{\circ}$}}}
\put(10,10){\line(1,0){60}}
\put(10,70){\line(1,0){60}}
\put(10,70){\line(0,-1){60}}
\put(70,40){\oval(60,60)[br]}
\put(70,40){\oval(60,60)[tr]}
\put(55,50){\makebox(0,0){\scriptsize{{\textit{{restrict}}}}}}
\put(56,30){\makebox(0,0){\Large{${\Rightarrow}$}}}
\end{picture}}
%%%%%%%%%%%%%%%%%%%%%%%%%%%%%%%%%%%%%%%%%%%%%%%%%%
\put(146.5,65){\begin{picture}(0,0)(0,0)
\setlength{\unitlength}{0.35pt}
%\thicklines
%\put(106,40){\makebox(0,0){\normalsize{$\boldsymbol{\circ}$}}}
%\put(4.7,40){\makebox(0,0){\normalsize{$\boldsymbol{\circ}$}}}
\put(40,10){\line(1,0){60}}
\put(40,70){\line(1,0){60}}
\put(100,70){\line(0,-1){60}}
\put(40,40){\oval(60,60)[bl]}
\put(40,40){\oval(60,60)[tl]}
\put(58,50){\makebox(0,0){\scriptsize{{\textit{{restrict}}}}}}
\put(56,30){\makebox(0,0){\Large{${\Leftarrow}$}}}
\end{picture}}
%%%%%%%%%%%%%%%%%%%%%%%%%%%%%%%%%%%%%%%%%%%%%%%%%%
\put(98,37){\begin{picture}(0,0)(0,3)
\setlength{\unitlength}{0.35pt}
\put(60,30){\makebox(0,0){\normalsize{$\vee$}}}
%\thicklines
\put(40,10){\line(1,0){40}}
\put(10,70){\line(1,0){100}}
\put(10,70){\line(0,-1){30}}
\put(110,70){\line(0,-1){30}}
\put(40,40){\oval(60,60)[bl]}
\put(80,40){\oval(60,60)[br]}
\put(60,55){\makebox(0,0){\scriptsize{{\textit{{join}}}}}}
\end{picture}}
%%%%%%%%%%%%%%%%%%%%%%%%%%%%%%%%%%%%%%%%%%%%%%%%%%
\put(120,100){\makebox(0,0){\footnotesize{{\textit{{filtered join}}}}}}
\put(120,88){\makebox(0,0){\large{$\varominus$}}}
%%%%%%%%%%%%%%%%%%%%%%%%%%%%%%%%%%%%%%%%%%%%%%%%%%
\put(38,80){\line(0,1){20}}
\put(38,80){\vector(1,0){20}}
\put(110,80){\line(-1,0){20}}
\put(110,80){\vector(0,-1){21}}
\put(120,38){\vector(0,-1){15}}
\put(130,80){\vector(0,-1){21}}
\put(130,80){\line(1,0){20}}
\put(203,80){\line(0,1){20}}
\put(203,80){\vector(-1,0){20}}
%\thicklines
%\put(15,110){\line(1,0){210}}
%\put(55,10){\line(1,0){130}}
%\put(15,50){\line(0,1){60}}
%\put(225,50){\line(0,1){60}}
%\qbezier(15,50)(15,10)(55,10)
%\qbezier(185,10)(225,10)(225,50)
%%%%%%%%%%%%%%%%%%%%%%%%%%%%%%%%%%%%%%%%%%%%%%%%%%
\end{picture}
\end{tabular}}}}
\end{center}
\caption{\texttt{FOLE} Filtered Join Flow Chart}
\label{fig:fole:filter:join:flo:chrt}
\end{figure}
Let $\mathcal{S}$ be a signature.
We are given 
two $\mathcal{S}$-tables 
$\mathcal{T}_{1} = {\langle{\mathcal{A}_{1},K_{1},t_{1}}\rangle}$ and
$\mathcal{T}_{2} = {\langle{\mathcal{A}_{2},K_{2},t_{2}}\rangle}$
connected by an 
%connecting 
$X$-sorted type domain opspan
$\mathcal{A}_{1}
% = {\langle{X,Y_{1},\models_{\mathcal{A}_{1}}}\rangle} 
\xrightarrow{\;g_{1}\,} 
%{\langle{X,Y,\models_{\mathcal{A}}}\rangle} = 
\mathcal{A}
\xleftarrow{\;g_{2}\,} 
\mathcal{A}_{2}$
via a third type domain $\mathcal{A}$,
%$\mathcal{S}$-table 
%$\mathcal{T} = {\langle{\mathcal{A},K,t}\rangle}$
and consisting of a span of data value functions
$Y_{1}\xleftarrow{\;g_{1}\;}Y\xrightarrow{\;g_{2}\;}Y_{2}$.
%satisfying the condition
%$g_{1}(y){\;\models_{\mathcal{A}_{1}}\;}x$
%\underline{iff}
%$y{\;\models_{\mathcal{A}}\;}x$
%for any sort $x{\,\in\,}X$ and 
%%source 
%data value $y{\,\in\,}Y$.
%
%%%%%%%%%%%%%%%%%%%%%%%%%%%%%%%%%%%%%%%%%%%%%%%%%%%%%%%%%%%%%%%%%%%%%%
%%%%%%%%%%%%%%%%%%%%%%%%%%%%%%%%%%%%%%%%%%%%%%%%%%%%%%%%%%%%%%%%%%%%%%
%\footnote{This means that $\mathcal{A} = g_{1}^{-1}(\mathcal{A}_{1})$,
%as discussed by the \textbf{Yin} definition in \S\,\ref{sub:sec:cls}.}
%%%%%%%%%%%%%%%%%%%%%%%%%%%%%%%%%%%%%%%%%%%%%%%%%%%%%%%%%%%%%%%%%%%%%%
%%%%%%%%%%%%%%%%%%%%%%%%%%%%%%%%%%%%%%%%%%%%%%%%%%%%%%%%%%%%%%%%%%%%%%
%
The set $Y$ represents authentic data values.
Filtered join is a binary operation 
$\mathcal{T}_{1}{\,\varominus_{\!\mathcal{S}}\,}\mathcal{T}$
that filters out the tuples with non-authentic data values,
and then joins the results.
When the data value functions are injections
$Y_{1}\xhookleftarrow{\;g_{1}\;}Y\xhookrightarrow{\;g_{2}\;}Y_{2}$,
%the set $Y$ may represent authentic data values.
%filtered join is a binary operation 
%$\mathcal{T}_{1}{\,\varominus_{\!\mathcal{T}}\,}\mathcal{T}$,
%where 
$\mathcal{T}_{1}$ and $\mathcal{T}_{2}$ 
will be tables whose tuples
\newline\mbox{}\hfill
$
\mathrmbfit{tup}_{\mathcal{S}}(\mathcal{A}_{1})
\subseteq 
{\wp}\mathrmbf{List}(Y_{1})
\supseteq 
{\wp}\mathrmbf{List}(Y)
\subseteq 
{\wp}\mathrmbf{List}(Y_{2})
\supseteq 
\mathrmbfit{tup}_{\mathcal{S}}(\mathcal{A}_{2})
$
\hfill\mbox{}\newline
have un-authenticated
data values in 
$Y_{1} \setminus Y$ and  
$Y_{2} \setminus Y$. 
Filtered join 
$\mathcal{T}_{1}{\,\varominus_{\!\mathcal{S}}\,}\mathcal{T}_{2}$
restricts to only those tuples in $ {\wp}\mathrmbf{List}(Y)$
with the authentic data values in $Y$.
%
%%%%%%%%%%%%%%%%%%%%%%%%%%%%%%%%%%%%%%%%%%%%%%%%%%%%%%%%%%%%%%%%%%%%%%%%%%%%%%%%
%%%%%%%%%%%%%%%%%%%%%%%%%%%%%%%%%%%%%%%%%%%%%%%%%%%%%%%%%%%%%%%%%%%%%%%%%%%%%%%%
\footnote{See the discussion of authenticity w.r.t. restriction 
in \S\,\ref{sub:sub:sec:adj:flow:S}.}
%%%%%%%%%%%%%%%%%%%%%%%%%%%%%%%%%%%%%%%%%%%%%%%%%%%%%%%%%%%%%%%%%%%%%%%%%%%%%%%%
%%%%%%%%%%%%%%%%%%%%%%%%%%%%%%%%%%%%%%%%%%%%%%%%%%%%%%%%%%%%%%%%%%%%%%%%%%%%%%%%
%
A tuple is in the filtered join 
$\mathcal{T}_{1}{\;\varominus_{\mathcal{S}}\;}\mathcal{T}_{2}$ 
when it is 
either in 
the restriction $\grave{\mathrmbfit{tbl}}_{\mathcal{S}}(g_{1})(\mathcal{T}_{1})$
or in 
the restriction $\grave{\mathrmbfit{tbl}}_{\mathcal{S}}(g_{2})(\mathcal{T}_{2})$.
Filtered join, within the context $\mathrmbf{Tbl}(\mathcal{S})$,
is restriction followed by join.
We use the following routes of flow.
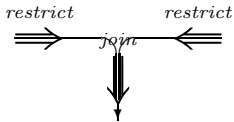
\begin{center}
{{\begin{tabular}{c}
\setlength{\unitlength}{0.6pt}
\begin{picture}(320,60)(0,5)
%\put(157,63){\makebox(0,0){\normalsize{$\textbf{1.}$}}}
%
%\put(325,25){\makebox(0,0){\normalsize{
%$\left.\rule{0pt}{24pt}\right\}
%\underset{\textstyle{\textsf{semi-join}}}{\textsf{right}}$}}}
%\put(-25,25){\makebox(0,0){\normalsize{
%$\underset{\textstyle{\textsf{semi-join}}}{\textsf{left}}
%\left\{\rule{0pt}{24pt}\right.$}}}
%
\put(100,55){\makebox(0,0){\huge{
$\overset{\textit{\scriptsize{restrict}}}{\Rightarrow}$}}}
\put(200,55){\makebox(0,0){\huge{
$\overset{\textit{\scriptsize{restrict}}}{\Leftarrow}$}}}
\put(150,30){\makebox(0,0){\huge{
$\overset{\textit{\scriptsize{join}}}{\Downarrow}$}}}
%\put(100,4.5){\makebox(0,0){\huge{
%$\overset{\textit{\scriptsize{project}}}{\Leftarrow}$}}}
%\put(105,-14.5){\makebox(0,0){$\textit{\scriptsize{image}}$}}
%\put(200,4.5){\makebox(0,0){\huge{
%$\overset{\textit{\scriptsize{project}}}{\Rightarrow}$}}}
%\put(205,-14.5){\makebox(0,0){$\textit{\scriptsize{image}}$}}
%
\put(90,48){\line(1,0){54}}
\put(144,38){\oval(20,20)[tr]}
\put(166,38){\oval(20,20)[tl]}
\put(220,48){\line(-1,0){54}}
%\put(145,-2){\line(-1,0){54}}
%\put(165,-2){\line(1,0){54}}
\put(154,9){\line(0,1){29}}
\put(156,9){\line(0,1){29}}
%\put(145,8){\oval(20,20)[br]}
%\put(165,8){\oval(20,20)[bl]}
%\put(155,6){\line(0,-1){8}}
\put(155,6){\vector(0,-1){10}}
\end{picture}
\end{tabular}}}
\end{center}
Similar to  natural join and data-type join,
we can define filtered join for any number of tables 
$\{ \mathcal{T}_{1}, \mathcal{T}_{2}, \mathcal{T}_{3}, \cdots , \mathcal{T}_{n} \}$
with a comparable constraint.
\begin{description}
\item[Constraint:] 
Consider an $X$-sorted type domain opspan 
$\mathcal{A}_{1}
%={\langle{I_{1},s_{1}}\rangle} 
\xrightarrow{g_{1}} 
%\overset{
\mathcal{A}
%}{\mathcal{S}} 
\xleftarrow{g_{2}} 
%{\langle{I_{2},s_{2}}\rangle} = 
\mathcal{A}_{2}$
in $\mathrmbf{Cls}(X)$
consisting of 
a span of data value functions
$Y_{1}\xleftarrow{g_{1}}Y\xrightarrow{g_{2}}Y_{2}$.
This is the constraint for filtered join 
(Tbl.\,\ref{tbl:fole:filter:join:input:output}).
\newline
\item[Construction:] 
The construction for filtered join is the same as the constraint for filtered join 
(Tbl.\,\ref{tbl:fole:filter:join:input:output}).
\newline
\item[Input:] 
Consider a pair of tables
$\mathcal{T}_{1} = {\langle{K_{1},t_{1}}\rangle} \in 
\mathrmbf{Tbl}_{\mathcal{S}}(\mathcal{A}_{1})$
and
$\mathcal{T}_{2} = {\langle{K_{2},t_{2}}\rangle} \in 
\mathrmbf{Tbl}_{\mathcal{S}}(\mathcal{A}_{2})$.
This is the input for filtered join 
(Tbl.\,\ref{tbl:fole:filter:join:input:output}).
\newpage
\item[Output:] 
The output is restriction (twice) followed by join.
\begin{itemize}
\item 
%\begin{itemize}
%\item 
Restriction 
(\S\,\ref{sub:sub:sec:adj:flow:S})
$\mathrmbf{Tbl}_{\mathcal{S}}(\mathcal{A}_{1})
\xrightarrow
%[{g_{1}}^{\ast}]
{\;\grave{\mathrmbfit{tbl}}_{\mathcal{S}}(g_{1})\;}
\mathrmbf{Tbl}_{\mathcal{S}}(\mathcal{A})
%\xleftarrow[{g_{2}}^{\ast}]
%{\;\grave{\mathrmbfit{tbl}}_{\mathcal{S}}(g_{2})\;}
%\mathrmbf{Tbl}_{\mathcal{S}}(\mathcal{A}_{2})
$
along the tuple function
\comment{\footnotesize{$
\mathrmbfit{tup}_{\mathcal{S}}(\mathcal{A}_{1})
\xleftarrow
%[{(\mbox{-})}{\,\cdot\,}g_{1}]
{\mathrmbfit{tup}_{\mathcal{S}}(g_{1})}
\mathrmbfit{tup}_{\mathcal{S}}(\mathcal{A})
$}\normalsize}
of the $X$-type domain morphism
$\mathcal{A}_{1}\xrightarrow{g_{1}}\mathcal{A}$,
maps the table
$\mathcal{T}_{1}$
% = {\langle{\mathcal{A}_{1},K_{1},t_{1}}\rangle}
%\in \mathrmbf{Tbl}_{\mathcal{S}}(\mathcal{A}_{1})$
to the table
$\widehat{\mathcal{T}}_{1} =
\grave{\mathrmbfit{tbl}}_{\mathcal{S}}(g_{1})(\mathcal{T}_{1})
%= {g_{1}}^{\ast}(\mathcal{T}_{1})
= {\langle{\widehat{K}_{1},\hat{t}_{1}}\rangle}
\in \mathrmbf{Tbl}_{\mathcal{S}}(\mathcal{A})$,
with its tuple function
$\widehat{K}_{1} \xrightarrow{\hat{t}_{1}} \mathrmbfit{tup}_{\mathcal{S}}(\mathcal{A})$
defined by pullback,
$k_{1}{\,\cdot\,}t_{1} 
= \hat{t}_{1}{\,\cdot\,}\mathrmbfit{tup}_{\mathcal{S}}(g_{1})$. 
This is linked to table $\mathcal{T}_{1}$ by the $\mathcal{S}$-table morphism 
%\[\mbox
{\footnotesize{{$
\mathcal{T}_{1} = {\langle{\mathcal{A}_{1},K_{1},t_{1}}\rangle}
\xleftarrow{{\langle{g_{1},k_{1}}\rangle}} 
{\langle{\mathcal{A},\widehat{K}_{1},\hat{t}_{1}}\rangle} 
= \widehat{\mathcal{T}}_{1}
$.}}\normalsize}
%\]
The same process can be defined for $\mathcal{S}$-table $\mathcal{T}_{2}$.
\newline
\item 
%
%\item 
%
Union (\S\,\ref{sub:sub:sec:boole})
of the two restriction tables 
$\widehat{\mathcal{T}}_{1}$
and 
$\widehat{\mathcal{T}}_{2}$
in the context 
$\mathrmbf{Tbl}_{\mathcal{S}}(\mathcal{A})$
defines the filtered join 
%$\mathcal{T}_{1}{\times_{\mathcal{T}}}\mathcal{T}_{2}$
${\mathcal{T}_{1}}{\,\varominus_{\mathcal{S}}}{\mathcal{T}_{2}}
= \widehat{\mathcal{T}}_{1} \vee \widehat{\mathcal{T}}_{2}
= {\langle{\widehat{K}_{1}{+}\widehat{K}_{2},[\hat{t}_{1},\hat{t}_{2}]}\rangle}$,
whose key set is the disjoint union $\widehat{K}_{1}{+}\widehat{K}_{2}$
and whose tuple map
$\widehat{K}_{1}{+}\widehat{K}_{2}
\xrightarrow{[\hat{t}_{1},\hat{t}_{2}]}
\mathrmbfit{tup}_{\mathcal{A}}(\mathcal{S})$
is the comediator of the opspan
$\widehat{K}_{1}\xrightarrow{\hat{t}_{1}} 
\mathrmbfit{tup}_{\mathcal{S}}
(\mathcal{A}_{1}{\times_{\mathcal{A}}}\mathcal{A}_{2})
\xleftarrow{\hat{t}_{2}}\widehat{K}_{2}$,
resulting in the span 
%of 
%${\langle{I_{1}{+}_{I}I_{2},[s_{1},s_{2}],\mathcal{A}}\rangle}$-
%table morphisms
%\newline\mbox{}\hfill
%\rule[-10pt]{0pt}{26pt}\bowtie
%\[\mbox
{\footnotesize{{$
%{\scriptstyle\sum}_{\tilde{g}_{1}}(\mathcal{T}_{1})
\widehat{\mathcal{T}}_{1}
%= {\langle{K_{1},\tilde{t}_{1}}\rangle} 
\xrightarrow{\;\check{\iota}_{1}\;} 
\mathcal{T}_{1}
{\,\varominus_{\mathcal{S}}}
\mathcal{T}_{2}
\xleftarrow{\;\check{\iota}_{2}\;} 
%{\langle{K_{2},\tilde{t}_{2}}\rangle} = 
\widehat{\mathcal{T}}_{2}
%{\scriptstyle\sum}_{\tilde{g}_{2}}(\mathcal{T}_{2})
$.}}\normalsize}
%\]
\newline
\end{itemize}
%
%\end{itemize}
%
Restriction composed with join 
defines the (\textsf{M}-shaped) multi-opspan of $\mathcal{S}$-table morphisms
%\newline\mbox{}\hfill
\[\mbox
{\footnotesize{$
\mathcal{T}_{1} 
%= {\langle{\mathcal{A}_{1},K_{1},t_{1}}\rangle}
\xleftarrow{{\langle{g_{1},k_{1}}\rangle}} 
\widehat{\mathcal{T}}_{1} 
%= {\langle{\widehat{K}_{1},\hat{t}_{1}}\rangle}
\xrightarrow{\;\check{\iota}_{1}\;} 
\mathcal{T}_{1}{\;\varominus_{\mathcal{S}}\;}\mathcal{T}_{2}
\xleftarrow{\;\check{\iota}_{2}\;}
%{\langle{\widehat{K}_{2},\hat{t}_{2}}\rangle} = 
\widehat{\mathcal{T}}_{2}
\xrightarrow{{\langle{g_{2},k_{2}}\rangle}} 
%= {\langle{\mathcal{A}_{1},K_{1},t_{1}}\rangle}
\mathcal{T}_{2}
$,}\normalsize}
\]
%\hfill\mbox{}\newline
which is the output for filtered join 
(Tbl.\,\ref{tbl:fole:filter:join:input:output}).
\end{description}
Filtered join 
%within the context $\mathrmbf{Tbl}(\mathcal{S})$
is restriction followed by join.
% (disjunction, or union).
This is a two-step process
%(illustrated in Fig.\;\ref{fig:fole:select:aux:rel})
%
%\comment{
\newline\mbox{}\hfill
\rule[-10pt]{0pt}{26pt}
%$\mathcal{T}_{1}{\;\varominus_{\mathcal{S}}\;}\mathcal{T} = 
%\grave{\mathrmbfit{tbl}}_{\mathcal{S}}(g_{1})(\mathcal{T}_{1}){\;\vee\;}\mathcal{T}$.
$\mathcal{T}_{1}{\;\varominus_{\mathcal{S}}\;}\mathcal{T}_{2} 
\doteq 
\grave{\mathrmbfit{tbl}}_{\mathcal{S}}(g_{1})(\mathcal{T}_{1})
{\;\vee\;}
\grave{\mathrmbfit{tbl}}_{\mathcal{S}}(g_{2})(\mathcal{T}_{2})$. 
%
%%%%%%%%%%%%%%%%%%%%%%%%%%%%%%%%%%%%%%%%%%%%%%%%%%%%%%%%%%%%%%%%%%%%%%
%%%%%%%%%%%%%%%%%%%%%%%%%%%%%%%%%%%%%%%%%%%%%%%%%%%%%%%%%%%%%%%%%%%%%%
\footnote{
The filtered join,
% in \S\,\ref{sub:sub:sec:filtered:join},
$\mathcal{T}_{1} 
\xleftarrow{{\langle{g_{1},k_{1}}\rangle}} 
\widehat{\mathcal{T}}_{1} 
\xrightarrow{\;\check{\iota}_{1}\;} 
\mathcal{T}_{1}{\;\varominus_{\mathcal{S}}\;}\mathcal{T}_{2}
\xleftarrow{\;\check{\iota}_{2}\;}
\widehat{\mathcal{T}}_{2}
\xrightarrow{{\langle{g_{2},k_{2}}\rangle}} 
\mathcal{T}_{2}$, 
is shielded from and 
has no direct connection to 
either table $\mathcal{T}_{1}$ or table $\mathcal{T}_{2}$.
This is comparable with
the data-type meet in \S\,\ref{sub:sub:sec:boolean:meet},
{{$\mathcal{T}_{1} 
\xrightarrow{{\langle{1,\tilde{g}_{1}}\rangle}} 
\widetilde{\mathcal{T}}_{1}
\xleftarrow{\hat{\pi}_{1}}
\mathcal{T}_{1}{\,\boxbar_{\mathcal{S}}}\mathcal{T}_{2}
%\widetilde{\mathcal{T}}_{1}{\,\wedge\,}\widetilde{\mathcal{T}}_{2}
\xrightarrow{\hat{\pi}_{2}}
\widetilde{\mathcal{T}}_{2}
\xleftarrow{{\langle{1,\tilde{g}_{2}}\rangle}} 
\mathcal{T}_{2}$,}}
which also is shielded from and 
has no direct connection to 
either table $\mathcal{T}_{1}$ or table $\mathcal{T}_{2}$.}
%%%%%%%%%%%%%%%%%%%%%%%%%%%%%%%%%%%%%%%%%%%%%%%%%%%%%%%%%%%%%%%%%%%%%%
%%%%%%%%%%%%%%%%%%%%%%%%%%%%%%%%%%%%%%%%%%%%%%%%%%%%%%%%%%%%%%%%%%%%%%
%
\hfill\mbox{}\newline
%}

%
\begin{table}
\begin{center}
{{\fbox{\begin{tabular}{c}
\setlength{\extrarowheight}{2pt}
{\scriptsize{$\begin{array}[c]{c@{\hspace{12pt}}l}
\mathcal{A}_{1}\xrightarrow{g_{1}}\mathcal{A}\xleftarrow{g_{2}}\mathcal{A}_{2}
&
\textit{constraint}
\\
\mathcal{A}_{1}\xrightarrow{g_{1}}\mathcal{A}\xleftarrow{g_{2}}\mathcal{A}_{2}
&
\textit{construction}
\\
\hline
\mathcal{T}_{1}\in\mathrmbf{Tbl}_{\mathcal{S}}(\mathcal{A}_{1})
\text{ and }
\mathcal{T}_{2}\in\mathrmbf{Tbl}_{\mathcal{S}}(\mathcal{A}_{2})
&
\textit{input}
\\
\mathcal{T}_{1} 
%= {\langle{\mathcal{A}_{1},K_{1},t_{1}}\rangle}
\xleftarrow{{\langle{g_{1},k_{1}}\rangle}} 
\widehat{\mathcal{T}}_{1} 
%= {\langle{\widehat{K}_{1},\hat{t}_{1}}\rangle}
\xrightarrow{\;\check{\iota}_{1}\;} 
\mathcal{T}_{1}{\;\varominus_{\mathcal{S}}\;}\mathcal{T}_{2}
\xleftarrow{\;\check{\iota}_{2}\;}
%{\langle{\widehat{K}_{2},\hat{t}_{2}}\rangle} = 
\widehat{\mathcal{T}}_{2}
\xrightarrow{{\langle{g_{2},k_{2}}\rangle}} 
%= {\langle{\mathcal{A}_{1},K_{1},t_{1}}\rangle}
\mathcal{T}_{2} 
&
\textit{output}
\end{array}$}}
\end{tabular}}}}
\end{center}
\caption{\texttt{FOLE} Filtered Join I/O}
\label{tbl:fole:filter:join:input:output}
\end{table}

% filter-join display
%%%%%%%%%%%%%%%%%%%%%%%%%%%%%%%%%%%%%%%%%%%%%%%%%%%%%%%%%%%%%%%%%%%%%%
%%%%%%%%%%%%%%%%%%%%%%%%%%%%%%%%%%%%%%%%%%%%%%%%%%%%%%%%%%%%%%%%%%%%%%

%%%%%%%%%%%%%%%%%%%%%%%%%%%%%%%%%%%%%%%%%%%%%%%%%%%%%%%%%%%%%%%%%%%%%%%
\newpage
\subsection{Data-type Meet.}
\label{sub:sub:sec:boolean:meet}
%%%%%%%%%%%%%%%%%%%%%%%%%%%%%%%%%%%%%%%%%%%%%%%%%%%%%%%%%%%%%%%%%%%%%%
%$\bigstar$ 

%
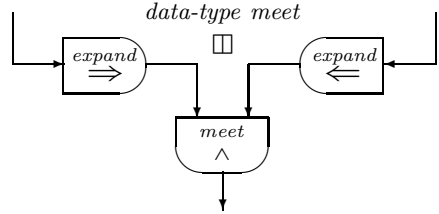
\begin{figure}
\begin{center}
{{{\begin{tabular}{c}
\begin{picture}(160,75)(37,27)
\setlength{\unitlength}{0.97pt}
%%%%%%%%%%%%%%%%%%%%%%%%%%%%%%%%%%%%%%%%%%%%%%%%%%
%\put(44,62){\begin{picture}(0,0)(0,0)
%\setlength{\unitlength}{0.46pt}
\put(54,65){\begin{picture}(0,0)(0,0)
\setlength{\unitlength}{0.35pt}
%\thicklines
%\put(106,40){\makebox(0,0){\normalsize{$\boldsymbol{\circ}$}}}
%\put(4.7,40){\makebox(0,0){\normalsize{$\boldsymbol{\circ}$}}}
\put(10,10){\line(1,0){60}}
\put(10,70){\line(1,0){60}}
\put(10,70){\line(0,-1){60}}
\put(70,40){\oval(60,60)[br]}
\put(70,40){\oval(60,60)[tr]}
\put(55,50){\makebox(0,0){\scriptsize{{\textit{{expand}}}}}}
\put(56,30){\makebox(0,0){\Large{${\Rightarrow}$}}}
\end{picture}}
%%%%%%%%%%%%%%%%%%%%%%%%%%%%%%%%%%%%%%%%%%%%%%%%%%
\put(146.5,65){\begin{picture}(0,0)(0,0)
\setlength{\unitlength}{0.35pt}
%\thicklines
%\put(106,40){\makebox(0,0){\normalsize{$\boldsymbol{\circ}$}}}
%\put(4.7,40){\makebox(0,0){\normalsize{$\boldsymbol{\circ}$}}}
\put(40,10){\line(1,0){60}}
\put(40,70){\line(1,0){60}}
\put(100,70){\line(0,-1){60}}
\put(40,40){\oval(60,60)[bl]}
\put(40,40){\oval(60,60)[tl]}
\put(58,50){\makebox(0,0){\scriptsize{{\textit{{expand}}}}}}
\put(56,30){\makebox(0,0){\Large{${\Leftarrow}$}}}
\end{picture}}
%%%%%%%%%%%%%%%%%%%%%%%%%%%%%%%%%%%%%%%%%%%%%%%%%%
\put(98,37){\begin{picture}(0,0)(0,3)
\setlength{\unitlength}{0.35pt}
\put(60,30){\makebox(0,0){\normalsize{$\wedge$}}}
%\thicklines
\put(40,10){\line(1,0){40}}
\put(10,70){\line(1,0){100}}
\put(10,70){\line(0,-1){30}}
\put(110,70){\line(0,-1){30}}
\put(40,40){\oval(60,60)[bl]}
\put(80,40){\oval(60,60)[br]}
\put(60,55){\makebox(0,0){\scriptsize{{\textit{{meet}}}}}}
\end{picture}}
%%%%%%%%%%%%%%%%%%%%%%%%%%%%%%%%%%%%%%%%%%%%%%%%%%
\put(120,100){\makebox(0,0){\footnotesize{{\textit{{data-type meet}}}}}}
\put(120,88){\makebox(0,0){\large{$\boxbar$}}}
%%%%%%%%%%%%%%%%%%%%%%%%%%%%%%%%%%%%%%%%%%%%%%%%%%
\put(38,80){\line(0,1){20}}
\put(38,80){\vector(1,0){20}}
\put(110,80){\line(-1,0){20}}
\put(110,80){\vector(0,-1){21}}
\put(120,38){\vector(0,-1){15}}
\put(130,80){\vector(0,-1){21}}
\put(130,80){\line(1,0){20}}
\put(203,80){\line(0,1){20}}
\put(203,80){\vector(-1,0){20}}
%\thicklines
%\put(15,110){\line(1,0){210}}
%\put(55,10){\line(1,0){130}}
%\put(15,50){\line(0,1){60}}
%\put(225,50){\line(0,1){60}}
%\qbezier(15,50)(15,10)(55,10)
%\qbezier(185,10)(225,10)(225,50)
%%%%%%%%%%%%%%%%%%%%%%%%%%%%%%%%%%%%%%%%%%%%%%%%%%
\end{picture}
\end{tabular}}}}
\end{center}
\caption{\texttt{FOLE} Data-type Meet Flow Chart}
\label{fig:fole:boole:meet:flo:chrt}
\end{figure}
The data-type meet for tables 
is the relational counterpart 
(like the natural join)
of
the logical conjunction for predicates.
Where the meet operation is the analogue for logical conjunction 
at the small scope $\mathrmbf{Tbl}(\mathcal{D})$ of a signed domain table fiber,
the data-type meet is defined at the intermediate scope of a signature table fiber
(in contrast to the natural join, 
which is define at the intermediate scope of a type domain table fiber).
%We identify \texttt{FOLE} data-type meets with limits in $\mathrmbf{Tbl}$.
%
%{\fbox{\textbf{The limit in $\mathrmbf{Tbl}(\mathcal{S})$ is not defined in the Table paper.}}}
%
We focus on tables in the context $\mathrmbf{Tbl}(\mathcal{S})$ 
for fixed signature (header) $\mathcal{S} = {\langle{I,x,X}\rangle}$.
We use the following routes of flow.
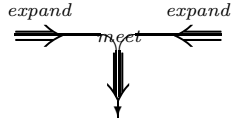
\begin{center}
{{\begin{tabular}{c}
\setlength{\unitlength}{0.6pt}
\begin{picture}(320,60)(0,5)
%\put(157,63){\makebox(0,0){\normalsize{$\textbf{1.}$}}}
%
%\put(325,25){\makebox(0,0){\normalsize{
%$\left.\rule{0pt}{24pt}\right\}
%\underset{\textstyle{\textsf{semi-join}}}{\textsf{right}}$}}}
%\put(-25,25){\makebox(0,0){\normalsize{
%$\underset{\textstyle{\textsf{semi-join}}}{\textsf{left}}
%\left\{\rule{0pt}{24pt}\right.$}}}
%
\put(100,55){\makebox(0,0){\huge{
$\overset{\textit{\scriptsize{expand}}}{\Rightarrow}$}}}
\put(200,55){\makebox(0,0){\huge{
$\overset{\textit{\scriptsize{expand}}}{\Leftarrow}$}}}
\put(150,30){\makebox(0,0){\huge{
$\overset{\textit{\scriptsize{meet}}}{\Downarrow}$}}}
%\put(100,4.5){\makebox(0,0){\huge{
%$\overset{\textit{\scriptsize{project}}}{\Leftarrow}$}}}
%\put(105,-14.5){\makebox(0,0){$\textit{\scriptsize{image}}$}}
%\put(200,4.5){\makebox(0,0){\huge{
%$\overset{\textit{\scriptsize{project}}}{\Rightarrow}$}}}
%\put(205,-14.5){\makebox(0,0){$\textit{\scriptsize{image}}$}}
%
\put(90,48){\line(1,0){54}}
\put(144,38){\oval(20,20)[tr]}
\put(166,38){\oval(20,20)[tl]}
\put(220,48){\line(-1,0){54}}
%\put(145,-2){\line(-1,0){54}}
%\put(165,-2){\line(1,0){54}}
\put(154,9){\line(0,1){29}}
\put(156,9){\line(0,1){29}}
%\put(145,8){\oval(20,20)[br]}
%\put(165,8){\oval(20,20)[bl]}
%\put(155,6){\line(0,-1){8}}
\put(155,6){\vector(0,-1){10}}
\end{picture}
\end{tabular}}}
\end{center}
\begin{description}
\item[Constraint:] 
The constraint for data-type meet is the same as the constraint for data-type join: 
an $X$-sorted type domain opspan 
$\mathcal{A}_{1}\xrightarrow{g_{1}}\mathcal{A}\xleftarrow{g_{2}}\mathcal{A}_{2}$
%in $\mathrmbf{Cls}(X)$
consisting of 
a span of data value functions
$Y_{1}\xleftarrow{g_{1}}Y\xrightarrow{g_{2}}Y_{2}$
(Tbl.\,\ref{tbl:fole:boolean:meet:input:output}).
\newline
\item[Construction:] 
The construction for data-type meet is the same as the construction for data-type join:
the span
{\footnotesize{$\mathcal{A}_{1} 
\xleftarrow{\;\tilde{g}_{1}\,} 
{\mathcal{A}_{1}{\times_{\mathcal{A}}}\mathcal{A}_{2}}
\xrightarrow{\;\tilde{g}_{2}} 
\mathcal{A}_{2}$}\normalsize}
(Tbl.\,\ref{tbl:fole:boolean:meet:input:output}).
\newline
\item[Input:] 
The input for data-type meet 
is the same as the input for data-type join
(Tbl.\,\ref{tbl:fole:boolean:meet:input:output}):
a pair of tables
$\mathcal{T}_{1} 
%= {\langle{K_{1},t_{1}}\rangle} 
\in \mathrmbf{Tbl}_{\mathcal{S}}(\mathcal{A}_{1})$
and
$\mathcal{T}_{2} 
%= {\langle{K_{2},t_{2}}\rangle} 
\in \mathrmbf{Tbl}_{\mathcal{S}}(\mathcal{A}_{2})$.
\newpage
\item[Output:] 
The output is expansion (twice) followed by meet.
\newline
%\mbox{}
%
\begin{itemize}
\item 
%\begin{description}
%\item[expansion:] 
Expansion 
$\mathrmbf{Tbl}_{\mathcal{S}}(\mathcal{A}_{1})
{\;\xrightarrow{\acute{\mathrmbfit{tbl}}_{\mathcal{S}}(\tilde{g}_{1})}\;}
\mathrmbf{Tbl}_{\mathcal{S}}(\mathcal{A}_{1}{\times_{\mathcal{A}}}\mathcal{A}_{2})$
along the tuple function of the $X$-type domain morphism 
{\footnotesize{$\mathcal{A}_{1} 
\xleftarrow{\;\tilde{g}_{1}\,} 
{\mathcal{A}_{1}{\times_{\mathcal{A}}}\mathcal{A}_{2}}$}\normalsize}
maps
the $\mathcal{S}$-table
$\mathcal{T}_{1}$
% = {\langle{K_{1},t_{1}}\rangle} \in 
%\mathrmbf{Tbl}_{\mathcal{S}}(\mathcal{A}_{1})$
to the $\mathcal{S}$-table
$\widetilde{\mathcal{T}}_{1}
= \acute{\mathrmbfit{tbl}}_{\mathcal{S}}(\tilde{g}_{1})(\mathcal{T}_{1})
%= {\scriptstyle\sum}_{\tilde{g}_{1}}(\mathcal{T}_{1})
= {\langle{K_{1},\tilde{t}_{1}}\rangle} 
\in \mathrmbf{Tbl}_{\mathcal{S}}(\mathcal{A}_{1}{\times_{\mathcal{A}}}\mathcal{A}_{2})$,
with its tuple function
$K_{1} \xrightarrow{\tilde{t}_{1}} 
\mathrmbfit{tup}_{\mathcal{S}}
(\mathcal{A}_{1}{\times_{\mathcal{A}}}\mathcal{A}_{2})$
defined by composition,
$\tilde{t}_{1} 
= t_{1}{\,\cdot\,}\mathrmbfit{tup}_{\mathcal{S}}(\tilde{g}_{1})$. 
This is linked to the table $\mathcal{T}_{1}$ 
by the $\mathcal{S}$-table morphism 
%\[\mbox
{\footnotesize{{$
\mathcal{T}_{1} = {\langle{\mathcal{A}_{1}.K_{1},t_{1}}\rangle}
\xrightarrow{{\langle{1,\tilde{g}_{1}}\rangle}} 
{\langle{
\mathcal{A}_{1}{\times_{\mathcal{A}}}\mathcal{A}_{2},K_{1},\tilde{t}_{1}}\rangle} 
= \widetilde{\mathcal{T}}_{1}
%{\scriptstyle\sum}_{\tilde{g}_{1}}(\mathcal{T}_{1})
$.}}\normalsize}
%\]
Similarly for $\mathcal{S}$-table
$\mathcal{T}_{2} = {\langle{K_{2},t_{2}}\rangle} \in 
\mathrmbf{Tbl}_{\mathcal{S}}(\mathcal{A}_{2})$.
\newline
%\item[meet:] 
%
\item 
Intersection (\S\,\ref{sub:sub:sec:boole})
of the two expansion tables 
$\widetilde{\mathcal{T}}_{1}$
%{\scriptstyle\sum}_{\tilde{g}_{1}}(\mathcal{T}_{1})$
and 
$\widetilde{\mathcal{T}}_{2}$
%${\scriptstyle\sum}_{\tilde{g}_{1}}(\mathcal{T}_{2})$ 
in the context 
$\mathrmbf{Tbl}_{\mathcal{S}}(\mathcal{A}_{1}{\times}_{\mathcal{A}}\mathcal{A}_{2})$
defines the data-type meet 
%$\mathcal{T}_{1}{\times_{\mathcal{T}}}\mathcal{T}_{2}$
${\mathcal{T}_{1}}{\,\boxbar_{\mathcal{S}}}{\mathcal{T}_{2}}
= \widetilde{\mathcal{T}}_{1} \wedge \widetilde{\mathcal{T}}_{2}
= {\langle{K_{1}{\times}K_{2},(\tilde{t}_{1},\tilde{t}_{2})}\rangle}$,
whose key set is the product $K_{1}{\times}K_{2}$
and whose tuple map
$K_{1}{\times}K_{2}\xrightarrow{(\tilde{t}_{1},\tilde{t}_{2})}
\mathrmbfit{tup}_{\mathcal{A}}(\mathcal{S})$
maps a pair of keys 
$(k_{1},k_{2}) \in K_{1}{\times}K_{2}$
to the common tuple
$\tilde{t}_{1}(k_{1})=\tilde{t}_{2}(k_{2}) 
\in 
\mathrmbfit{tup}_{\mathcal{S}}
(\mathcal{A}_{1}{\times_{\mathcal{A}}}\mathcal{A}_{2})$.
Intersection is the product in $\mathrmbf{Tbl}_{\mathcal{S}}(\mathcal{A})$
with span
%\newline\mbox{}\hfill
$\widetilde{\mathcal{T}}_{1}\xleftarrow{\hat{\pi}_{1}}
\widetilde{\mathcal{T}}_{1}{\,\wedge\,}\widetilde{\mathcal{T}}_{2}
\xrightarrow{\hat{\pi}_{2}}\widetilde{\mathcal{T}}_{2}$.
%\hfill\mbox{}\newline
%\end{description}
%
\newline
\end{itemize}
Expansion composed with meet 
defines the (\textsf{W}-shaped) multi-span of $\mathcal{S}$-table morphisms
%\newline\mbox{}\hfill
\[\mbox
{\footnotesize{$
\mathcal{T}_{1} 
\xrightarrow{{\langle{1,\tilde{g}_{1}}\rangle}} 
\widetilde{\mathcal{T}}_{1}
\xleftarrow{\hat{\pi}_{1}}
{\mathcal{T}_{1}}{\,\boxbar_{\mathcal{S}}}{\mathcal{T}_{2}}
%\widetilde{\mathcal{T}}_{1}{\,\wedge\,}\widetilde{\mathcal{T}}_{2}
\xrightarrow{\hat{\pi}_{2}}
\widetilde{\mathcal{T}}_{2}
\xleftarrow{{\langle{1,\tilde{g}_{2}}\rangle}} 
\mathcal{T}_{2}
$,}\normalsize}
\]
%\hfill\mbox{}\newline
which is the output for data-type meet 
(Tbl.\,\ref{tbl:fole:boolean:join:input:output}).
%
%
%%%%%%%%%%%%%%%%%%%%%%%%%%%%%%%%%%%%%%%%%%%%%%%%%%%%%%%%%%%%%%%%%%%%%%%%%%%%%%%%
%%%%%%%%%%%%%%%%%%%%%%%%%%%%%%%%%%%%%%%%%%%%%%%%%%%%%%%%%%%%%%%%%%%%%%%%%%%%%%%%
\footnote{The data-type meet of a disjoint sum is the empty table.}
%%%%%%%%%%%%%%%%%%%%%%%%%%%%%%%%%%%%%%%%%%%%%%%%%%%%%%%%%%%%%%%%%%%%%%%%%%%%%%%%
%%%%%%%%%%%%%%%%%%%%%%%%%%%%%%%%%%%%%%%%%%%%%%%%%%%%%%%%%%%%%%%%%%%%%%%%%%%%%%%%
%
\end{description}
Data-type meet 
%within the context $\mathrmbf{Tbl}(\mathcal{S})$
is expansion followed by meet.
This is the two-step process 
\newline\mbox{}\hfill
\rule[-10pt]{0pt}{26pt}
$\mathcal{T}_{1}{\,\boxbar_{\mathcal{S}}}\mathcal{T}_{2} 
\doteq  
\acute{\mathrmbfit{tbl}}_{\mathcal{S}}(\tilde{g}_{1})(\mathcal{T}_{1})
{\;\wedge\;}
\acute{\mathrmbfit{tbl}}_{\mathcal{S}}(\tilde{g}_{2})(\mathcal{T}_{2})$.
%
%%%%%%%%%%%%%%%%%%%%%%%%%%%%%%%%%%%%%%%%%%%%%%%%%%%%%%%%%%%%%%%%%%%%%%%%%%%%%%%%
%%%%%%%%%%%%%%%%%%%%%%%%%%%%%%%%%%%%%%%%%%%%%%%%%%%%%%%%%%%%%%%%%%%%%%%%%%%%%%%%
\footnote{The data-type meet,
%in \S\,\ref{sub:sub:sec:boolean:meet},
{{$\mathcal{T}_{1} 
\xrightarrow{{\langle{1,\tilde{g}_{1}}\rangle}} 
\widetilde{\mathcal{T}}_{1}
\xleftarrow{\hat{\pi}_{1}}
\mathcal{T}_{1}{\,\boxbar_{\mathcal{S}}}\mathcal{T}_{2}
%\widetilde{\mathcal{T}}_{1}{\,\wedge\,}\widetilde{\mathcal{T}}_{2}
\xrightarrow{\hat{\pi}_{2}}
\widetilde{\mathcal{T}}_{2}
\xleftarrow{{\langle{1,\tilde{g}_{2}}\rangle}} 
\mathcal{T}_{2}$,}}
is shielded from and 
has no direct connection to 
either table $\mathcal{T}_{1}$ or table $\mathcal{T}_{2}$.
This is comparable with
the filtered join in \S\,\ref{sub:sub:sec:filtered:join},
$\mathcal{T}_{1} 
\xleftarrow{{\langle{g_{1},k_{1}}\rangle}} 
\widehat{\mathcal{T}}_{1} 
\xrightarrow{\;\check{\iota}_{1}\;} 
\mathcal{T}_{1}{\,\varominus_{\mathcal{S}}\,}\mathcal{T}_{2}
\xleftarrow{\;\check{\iota}_{2}\;}
\widehat{\mathcal{T}}_{2}
\xrightarrow{{\langle{g_{2},k_{2}}\rangle}} 
\mathcal{T}_{2}$, 
which also is shielded from and 
has no direct connection to 
either table $\mathcal{T}_{1}$ or table $\mathcal{T}_{2}$.}
%%%%%%%%%%%%%%%%%%%%%%%%%%%%%%%%%%%%%%%%%%%%%%%%%%%%%%%%%%%%%%%%%%%%%%%%%%%%%%%%
%%%%%%%%%%%%%%%%%%%%%%%%%%%%%%%%%%%%%%%%%%%%%%%%%%%%%%%%%%%%%%%%%%%%%%%%%%%%%%%%
%
\hfill\mbox{}\newline
\begin{table}
\begin{center}
{{\fbox{\begin{tabular}{c}
\setlength{\extrarowheight}{2pt}
{\scriptsize{$\begin{array}[c]{c@{\hspace{12pt}}l}
\mathcal{A}_{1}\xrightarrow{g_{1}}\mathcal{A}\xleftarrow{g_{2}}\mathcal{A}_{2}
&
\textit{constraint}
\\
\mathcal{A}_{1} 
\xleftarrow{\;\tilde{g}_{1}\,} 
{\mathcal{A}_{1}{\times_{\mathcal{A}}}\mathcal{A}_{2}}
\xrightarrow{\;\tilde{g}_{2}} 
\mathcal{A}_{2}
&
\textit{construction}
\\
\hline
\mathcal{T}_{1}\in\mathrmbf{Tbl}_{\mathcal{S}}(\mathcal{A}_{1})
\text{ and }
\mathcal{T}_{2}\in\mathrmbf{Tbl}_{\mathcal{S}}(\mathcal{A}_{2})
&
\textit{input}
\\
\mathcal{T}_{1} 
\xrightarrow{{\langle{1,\tilde{g}_{1}}\rangle}} 
\widetilde{\mathcal{T}}_{1}
\xleftarrow{\hat{\pi}_{1}}
\mathcal{T}_{1}{\,\boxbar_{\mathcal{S}}}\mathcal{T}_{2}
\xrightarrow{\hat{\pi}_{2}}
\widetilde{\mathcal{T}}_{2}
\xleftarrow{{\langle{1,\tilde{g}_{2}}\rangle}} 
\mathcal{T}_{2}
&
\textit{output}
\end{array}$}}
\end{tabular}}}}
\end{center}
\caption{\texttt{FOLE} Data-type Meet I/O}
\label{tbl:fole:boolean:meet:input:output}
\end{table}
%
%{\fbox{\textbf{Combine with data-type join I/O??}}}

%\mbox{}\newline\newline
%{\fbox{$\blacktriangledown$\hspace{90pt}
%\textbf{Work zone: adjoint flow.}
%\hspace{90pt}$\blacktriangledown$}}
%\newline

%\mbox{}\newline\newline
%{\fbox{$\blacktriangle$\hspace{90pt}
%\textbf{Work zone: adjoint flow.}
%$\hspace{90pt}\blacktriangle$}}
%\newline

%%%%%%%%%%%%%%%%%%%%%%%%%%%%%%%%%%%%%%%%%%%%%%%%%%%%%%%%%%%%%%
%
\newpage
\subsection{Subtraction.}
\label{sub:sub:sec:subtrac}
%$\ast{-}{-}$}
%%%%%%%%%%%%%%%%%%%%%%%%%%%%%%%%%%%%%%%%%%%%%%%%%%%%%%%%%%%%%%

%
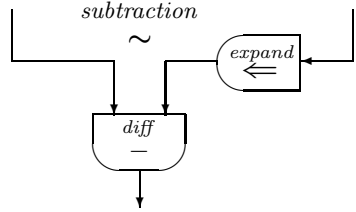
\begin{figure}
\begin{center}
{{{\begin{tabular}{c}
\begin{picture}(160,75)(37,27)
\setlength{\unitlength}{0.97pt}
%%%%%%%%%%%%%%%%%%%%%%%%%%%%%%%%%%%%%%%%%%%%%%%%%%
\put(146.5,65){\begin{picture}(0,0)(0,0)
\setlength{\unitlength}{0.35pt}
%\thicklines
%\put(106,40){\makebox(0,0){\normalsize{$\boldsymbol{\circ}$}}}
%\put(4.7,40){\makebox(0,0){\normalsize{$\boldsymbol{\circ}$}}}
\put(40,10){\line(1,0){60}}
\put(40,70){\line(1,0){60}}
\put(100,70){\line(0,-1){60}}
\put(40,40){\oval(60,60)[bl]}
\put(40,40){\oval(60,60)[tl]}
\put(58,50){\makebox(0,0){\scriptsize{{\textit{{expand}}}}}}
\put(56,30){\makebox(0,0){\Large{${\Leftarrow}$}}}
\end{picture}}
%%%%%%%%%%%%%%%%%%%%%%%%%%%%%%%%%%%%%%%%%%%%%%%%%%
\put(98,37){\begin{picture}(0,0)(0,3)
\setlength{\unitlength}{0.35pt}
\put(60,30){\makebox(0,0){\normalsize{$-$}}}
%\thicklines
\put(40,10){\line(1,0){40}}
\put(10,70){\line(1,0){100}}
\put(10,70){\line(0,-1){30}}
\put(110,70){\line(0,-1){30}}
\put(40,40){\oval(60,60)[bl]}
\put(80,40){\oval(60,60)[br]}
\put(60,55){\makebox(0,0){\scriptsize{{\textit{{diff}}}}}}
\end{picture}}
%%%%%%%%%%%%%%%%%%%%%%%%%%%%%%%%%%%%%%%%%%%%%%%%%%
\put(120,100){\makebox(0,0){\footnotesize{{\textit{{subtraction}}}}}}
\put(120,88){\makebox(0,0){\large{$\thicksim$}}}
%%%%%%%%%%%%%%%%%%%%%%%%%%%%%%%%%%%%%%%%%%%%%%%%%%
\put(70,80){\line(0,1){20}}
%\put(38,80){\vector(1,0){20}}
\put(110,80){\line(-1,0){40}}
\put(110,80){\vector(0,-1){21}}
\put(120,38){\vector(0,-1){15}}
\put(130,80){\vector(0,-1){21}}
\put(130,80){\line(1,0){20}}
\put(203,80){\line(0,1){20}}
\put(203,80){\vector(-1,0){20}}
%\thicklines
%\put(15,110){\line(1,0){210}}
%\put(55,10){\line(1,0){130}}
%\put(15,50){\line(0,1){60}}
%\put(225,50){\line(0,1){60}}
%\qbezier(15,50)(15,10)(55,10)
%\qbezier(185,10)(225,10)(225,50)
%%%%%%%%%%%%%%%%%%%%%%%%%%%%%%%%%%%%%%%%%%%%%%%%%%
\end{picture}
\end{tabular}}}}
\end{center}
\caption{\texttt{FOLE} Subtraction Flow Chart}
\label{fole:subtraction:flo:chrt}
\end{figure}
%
%For two sets $A$ and $A_{2}$,
%the difference is
%$A \setminus A_{2} 
%= A \setminus (A \cap A_{2})$.
%Extend this identity to tables:
%subtraction is defined as
%a special case of
%data-type anti-meet.
%{{\textbf{This ia a special case of data-type anti-join with identity morphisms on left side.}}}
%Hence,
%we need 
%type-domain semi-meet.
Here we define the operation of subtraction 
in terms of the basic operations. 
Let $\mathcal{S} = {\langle{I,x,X}\rangle}$
be a fixed signature.
Let
$\mathcal{T}$ and $\mathcal{T}_{2}$ be two $\mathcal{S}$-tables.
For subtraction,
we assume that the data values in the body of $\mathcal{T}_{2}$ 
are a subset of those of $\mathcal{T}$. 
Hence,
assume that 
$\mathcal{T}$ 
is a table in $\mathrmbf{Tbl}_{\mathcal{S}}(\mathcal{A})$,
$\mathcal{T}_{2}$ is a table in $\mathrmbf{Tbl}_{\mathcal{S}}(\mathcal{A}_{2})$,
and that these are connected with
an $X$-type domain morphism
{\footnotesize{$\mathcal{A}
\xrightarrow{\;g\;}\mathcal{A}_{2}$}\normalsize}
with data value function
{\footnotesize{$Y\xleftarrow{\;g}Y_{2}$.}\normalsize}
%
%is a table in $\mathrmbf{Tbl}_{\mathcal{S}}(\mathcal{A}_{1}{\times}\mathcal{A}_{2})$
%with product $X$-type domain 
%$\mathcal{A}_{1}{\times}\mathcal{A}_{2}$ 
%$\mathcal{A}_{1}{\times}\mathcal{A}_{2}
%={\langle{X,Y_{1}{+}Y_{2},\models_{\mathcal{A}_{1}{\times}\mathcal{A}_{1}}}\rangle}$ 
%in the 
%span
%{\footnotesize{$\mathcal{A}_{1}\xleftarrow{i_{1}\,}
%{\mathcal{A}_{1}{\times}\mathcal{A}_{2}}
%\xrightarrow{\;i_{2}}\mathcal{A}_{2}$}\normalsize}
%with data value functions
%{\footnotesize{$Y_{1}\xhookrightarrow{i_{1}\,} 
%Y_{1}{+}Y_{2}
%\xhookleftarrow{\;i_{2}}Y_{2}$.}\normalsize}
%
% Assume that
%$\mathcal{T}_{2}$ is a table in $\mathrmbf{Tbl}_{\mathcal{S}}(\mathcal{A}_{2})$.
%
The subtraction $\mathcal{T}{\,\thicksim\,}\mathcal{T}_{2}$ 
consists of 
the expansion of $\mathcal{T}_{2}$
along the $X$-type domain morphism
{\footnotesize{$\mathcal{A}
\xrightarrow{\;g\;}\mathcal{A}_{2}$}\normalsize}
%{\footnotesize{${\mathcal{A}_{1}{\times}\mathcal{A}_{2}}
%\xrightarrow{\;i_{2}}\mathcal{A}_{2}$}\normalsize}
%with data value function
%{\footnotesize{$Y_{1}{+}Y_{2}\xhookleftarrow{\;i_{2}}Y_{2}$.}\normalsize}
to the $X$-type domain $\mathcal{A}$, 
followed by subtraction from $\mathcal{T}$
in the signed domain
${\langle{\mathcal{S},\mathcal{A}}\rangle}$.
Subtraction is defined by the two-step process
illustrated in 
Fig.\;\ref{fole:subtraction:flo:chrt}.
\begin{description}
\item[Constraint/Construction:] 
The constraint and construction for subtraction are the same 
(Tbl.\,\ref{tbl:fole:subtraction:input:output}):
an $X$-type domain morphism
{\footnotesize{$\mathcal{A}\xrightarrow{\;g\;}\mathcal{A}_{2}$}\normalsize}
with data value function
{\footnotesize{$Y\xleftarrow{\;g}Y_{2}$.}\normalsize}
%is the same as the constraint for disjoint sum
%(Tbl.\,\ref{tbl:fole:disjoint:sum:input:output}): 
%a trivial
%$X$-sorted type domain opspan 
%$\mathcal{A}_{1}\xrightarrow{0_{1}}\mathcal{A}_{\top}\xleftarrow{0_{2}}\mathcal{A}_{2}$
%consisting of a span of data value functions
%in $\mathrmbf{Cls}(X)$
%consisting of the span of inclusion data value functions
%$Y_{1}\xhookleftarrow{0_{1}}\emptyset\xhookrightarrow{0_{2}}Y_{2}$.
%
%This is the constraint for subtraction .
%
\newline
\comment{
\item[Construction:] 
The construction for subtraction 
(Tbl.\,\ref{tbl:fole:subtraction:input:output})
is the same as 
the constraint for subtraction. 
%the construction for disjoint sum
%(Tbl.\,\ref{tbl:fole:disjoint:sum:input:output}): 
%the pullback of this constraint
%is the product $X$-type domain 
%%$\mathcal{A}_{1}{\times}\mathcal{A}_{2}$ 
%$\mathcal{A}_{1}{\times}\mathcal{A}_{2}
%={\langle{X,Y_{1}{+}Y_{2},\models_{\mathcal{A}_{1}{\times}\mathcal{A}_{1}}}\rangle}$ 
%in the span
%{\footnotesize{$\mathcal{A}_{1}\xleftarrow{i_{1}\,}
%{\mathcal{A}_{1}{\times}\mathcal{A}_{2}}
%\xrightarrow{\;i_{2}}\mathcal{A}_{2}$}\normalsize}
%with 
%the disjoint sum data value set $Y_{1}{+}Y_{2}$
%and
%projection
%$X$-type domain morphisms 
%%(span)
%{\footnotesize{$
%\mathcal{A}_{1}
%={\langle{X,Y_{1},\models_{\mathcal{A}_{1}}}\rangle}
%\xleftarrow{{\langle{\mathrmit{1}_{X},i_{1}}\rangle}}
%\xleftarrow{\,\tilde{0}_{1}}
%\mathcal{A}_{1}{\times}\mathcal{A}_{2}
%\xrightarrow{\tilde{0}_{2}\,}
%\xrightarrow{{\langle{\mathrmit{1}_{X},i_{2}}\rangle}}
%{\langle{X,Y_{2},\models_{\mathcal{A}_{2}}}\rangle}=
%\mathcal{A}_{2}$}}
%
%with inclusion data value functions 
%{\footnotesize{$Y_{1}
%\xhookrightarrow{\tilde{0}_{1}}
%Y_{1}{+}Y_{2}
%\xhookleftarrow{\tilde{0}_{2}}Y_{2}$.}}
%
%with inclusion data value functions
%{\footnotesize{$Y_{1}\xhookrightarrow{i_{1}\,} 
%Y_{1}{+}Y_{2}
%\xhookleftarrow{\;i_{2}}Y_{2}$}\normalsize}
%into the disjoint union set of data values.
%
%This is the construction for subtraction 
%(Tbl.\,\ref{tbl:fole:subtraction:input:output}).
\newline
}
\item[Input:] 
Consider a pair of tables
$\mathcal{T}= {\langle{K,t}\rangle} \in 
\mathrmbf{Tbl}_{\mathcal{S}}(\mathcal{A})$
and
$\mathcal{T}_{2} 
= {\langle{K_{2},t_{2}}\rangle} 
\in \mathrmbf{Tbl}_{\mathcal{S}}(\mathcal{A}_{2})$.
This is the input for subtraction 
(Tbl.\,\ref{tbl:fole:subtraction:input:output}).
\newline
\item[Output:] 
Subtraction is expansion (once) followed by difference.
\begin{itemize}
\item 
Expansion 
$\mathrmbf{Tbl}_{\mathcal{S}}(\mathcal{A})
{\;\xleftarrow{\acute{\mathrmbfit{tbl}}_{\mathcal{S}}(g)}\;}
\mathrmbf{Tbl}_{\mathcal{S}}(\mathcal{A}_{2})$
along the tuple function of the $X$-type domain morphism 
%{\footnotesize{$\mathcal{A}_{1} 
%\xleftarrow{\;\tilde{0}_{2}\,} 
%{\mathcal{A}_{1}{\times}\mathcal{A}_{2}}$}\normalsize}
%{\footnotesize{$
%\mathcal{A}_{1}{\times}\mathcal{A}_{2}
%\xrightarrow{\;\tilde{0}_{2}\,}
%\mathcal{A}_{2}$}}
{\footnotesize{$\mathcal{A}
\xrightarrow{\;g\;}\mathcal{A}_{2}$}\normalsize}
maps
the $\mathcal{S}$-table
$\mathcal{T}_{2} \in \mathrmbf{Tbl}_{\mathcal{S}}(\mathcal{A}_{2})$
% = {\langle{K_{1},t_{1}}\rangle} \in 
%\mathrmbf{Tbl}_{\mathcal{S}}(\mathcal{A}_{1})$
to the $\mathcal{S}$-table
$\widetilde{\mathcal{T}}_{2}
= \acute{\mathrmbfit{tbl}}_{\mathcal{S}}(g)(\mathcal{T}_{2})
%= {\scriptstyle\sum}_{\tilde{0}_{2}}(\mathcal{T}_{1})
= {\langle{K_{2},\tilde{t}_{2}}\rangle} 
\in \mathrmbf{Tbl}_{\mathcal{S}}(\mathcal{A})$
with its tuple function
$K_{2} \xrightarrow{\tilde{t}_{2}} 
\mathrmbfit{tup}_{\mathcal{S}}(\mathcal{A})$
defined by composition,
$\tilde{t}_{2} 
= t_{2}{\,\cdot\,}\mathrmbfit{tup}_{\mathcal{S}}(g)$. 
This is linked to the table $\mathcal{T}_{2}$ 
by the $\mathcal{S}$-table morphism 
%\[\mbox
{\footnotesize{{$
\mathcal{T}_{2} = {\langle{\mathcal{A}_{2}.K_{2},t_{2}}\rangle}
\xrightarrow{{\langle{1,g}\rangle}} 
{\langle{\mathcal{A},K_{2},\tilde{t}_{2}}\rangle} 
= \widetilde{\mathcal{T}}_{2}
%{\scriptstyle\sum}_{\tilde{0}_{2}}(\mathcal{T}_{1})
$.}}\normalsize}
%\]
%\begin{description}
%\item[expansion:] 
%Expand table $\mathcal{T}_{2}$
%getting the table
%$\acute{\mathrmbfit{tbl}}_{\mathcal{S}}(\tilde{0}_{2})(\mathcal{T}_{2})$
%in $\mathrmbf{Tbl}_{\mathcal{S}}(\mathcal{A}_{1}{\times}\mathcal{A}_{2})$. 
%
\newline
%\item[removal:] 
%
\item 
The difference 
$
\mathcal{T}{\,\thicksim\,}\mathcal{T}_{2} 
=
\mathcal{T}{\,-\,}\widetilde{\mathcal{T}}_{2} 
=
\mathcal{T}{\,-\,}\acute{\mathrmbfit{tbl}}_{\mathcal{S}}(g)(\mathcal{T}_{2})
= {\langle{\bar{K}_{2},\bar{t}_{2}}\rangle}$
is the desired table in 
$\mathrmbf{Tbl}_{\mathcal{S}}(\mathcal{A})$
%. 
%The difference operation defines the \texttt{FOLE} table
%$\mathcal{T}{-}\widetilde{\mathcal{T}}_{2}  
with key set
$\bar{K}_{2}$ 
%is the tuple inverse image of the difference tuple set
%$\bar{K} = t^{\text{-}1}({\wp}t(K){\,{-}\,}{\wp}t'(K')){\,\subseteq\,}K$
and tuple map
$\bar{t}_{2} : \bar{K}_{2} \xhookrightarrow{\bar{k}}
K_{2} \xrightarrow{\tilde{t}_{2}} 
\mathrmbfit{tup}_{\mathcal{S}}(\mathcal{A})$.
%K\xrightarrow{t}\mathrmbfit{tup}_{\mathcal{A}}(\mathcal{S})$
\newline
\end{itemize}
%
%\end{description}
%
Expansion composed with difference defines the inclusion $\mathcal{S}$-table morphism
$\mathcal{T}\xhookleftarrow{\bar{\omega}}\mathcal{T}{\,\thicksim\,}\mathcal{T}_{2}$,
which is the output for subtraction 
(Tbl.\,\ref{tbl:fole:subtraction:input:output}).
\end{description}
Subtraction
%within the context $\mathrmbf{Tbl}(\mathcal{S})$
is expansion followed by difference.
This is the two-step process 
\newline\mbox{}\hfill
\rule[-10pt]{0pt}{26pt}
$\mathcal{T}{\,\thicksim\,}\mathcal{T}_{2} 
\doteq  
\mathcal{T}{\,-\,}\acute{\mathrmbfit{tbl}}_{\mathcal{S}}(g)(\mathcal{T}_{2})$.
\hfill\mbox{}\newline
\begin{table}
\begin{center}
{{\fbox{\begin{tabular}{c}
\setlength{\extrarowheight}{2pt}
{\scriptsize{$\begin{array}[c]{c@{\hspace{12pt}}l}
\mathcal{A}\xrightarrow{\;g\;}\mathcal{A}_{2}
%\mathcal{A}_{1}  
%\text{ and }
%\mathcal{A}_{2} 
&
\textit{constraint}
\\
\mathcal{A}\xrightarrow{\;g\;}\mathcal{A}_{2}
%\mathcal{A}_{1} \xleftarrow{\tilde{0}_{1}} 
%{\mathcal{A}_{1}{\times}\mathcal{A}_{2}}
%\xrightarrow{\,\tilde{0}_{2}}\mathcal{A}_{2}
&
\textit{construction}
\\
\hline
\mathcal{T} \in \mathrmbf{Tbl}_{\mathcal{S}}(\mathcal{A})
\text{ and }
\mathcal{T}_{2} \in \mathrmbf{Tbl}_{\mathcal{S}}(\mathcal{A}_{2})
&
\textit{input}
\\
\mathcal{T}\xhookleftarrow{\bar{\omega}}\mathcal{T}{\,\thicksim\,}\mathcal{T}_{2}
&
\textit{output}
\end{array}$}}
\end{tabular}}}}
\end{center} 
\caption{\texttt{FOLE} Subtraction I/O}
\label{tbl:fole:subtraction:input:output}
\end{table}
\begin{proposition}\label{subtract:morph}
There is a table morphism
%%%%%%%%%%%%%%%%%%%%%%%%%%%%%%%%%%%%%%%%%%%%%%%%%%%%%%%%%%%%
%%%%%%%%%%%%%%%%%%%%%%%%%%%%%%%%%%%%%%%%%%%%%%%%%%%%%%%%%%%%
\footnote{Analogous to 
$B \subseteq (B \setminus A) \cup A = B \cup A$.}
%$7{\;\geq\;}\underset{2}{\underbrace{(7{\,\div\,}3)}}{\;\times\;}3$
%or modus ponens $q \leftarrow ((q \leftarrow p) \wedge p)$
%%%%%%%%%%%%%%%%%%%%%%%%%%%%%%%%%%%%%%%%%%%%%%%%%%%%%%%%%%%%
%%%%%%%%%%%%%%%%%%%%%%%%%%%%%%%%%%%%%%%%%%%%%%%%%%%%%%%%%%%%
%\hfill\mbox{}
%\newline
\[\mbox
{\footnotesize{{$
%\newline\mbox{}\hfill
\mathcal{T}\rightarrow
(\mathcal{T}{\,\thicksim\,}\mathcal{T}_{2})
{\;\oplus\;}\mathcal{T}_{2}
$.}}\normalsize}
\]
\end{proposition}
\begin{proof}
%
%Footnote\;\ref{footnote:refl}.
%%%%%%%%%%%%%%%%%%%%%%%%%%%%%%%%%%%%%%%%%%%%%%%%%%%%%%%%%%%%%%%%%%%%%%
%%%%%%%%%%%%%%%%%%%%%%%%%%%%%%%%%%%%%%%%%%%%%%%%%%%%%%%%%%%%%%%%%%%%%%
\footnote{Argument in terms of the underlying relations in the reflection
Prop.\;\ref{tbl:rel:refl}.}
%%%%%%%%%%%%%%%%%%%%%%%%%%%%%%%%%%%%%%%%%%%%%%%%%%%%%%%%%%%%%%%%%%%%%%
%%%%%%%%%%%%%%%%%%%%%%%%%%%%%%%%%%%%%%%%%%%%%%%%%%%%%%%%%%%%%%%%%%%%%%
%
%Suppose $t \in (\mathcal{T}{\,\thicksim\,}\mathcal{T}_{2}){\;\oplus\;}\mathcal{T}_{2}$,
%\begin{itemize}
%\item 
%Description of the tuples in 
%$\mathcal{T}{\,\thicksim\,}\mathcal{T}_{2}$:
%\newline
%$\mathcal{T}{\,\thicksim\,}\mathcal{T}_{2}$:
%is a table in $\mathrmbf{Tbl}_{\mathcal{S}}(\mathcal{A})$.
%with the product type domain
%$\mathcal{A}_{1}{\times}\mathcal{A}_{2}$.
% (not the pullback!);
%not necessarily in $\mathrmbf{Tbl}_{\mathcal{S}}(\mathcal{A}_{1})$.
$\mathcal{T}{\,\thicksim\,}\mathcal{T}_{2}$ has all the tuples in $\mathcal{T}$,
except for those in 
%the expansion
$\widetilde{\mathcal{T}}_{2} = \mathrmbfit{tup}_{\mathcal{S}}(g)(\mathcal{T}_{2})$.
%a table in $\mathrmbf{Tbl}_{\mathcal{S}}(\mathcal{A}_{2})$.
%\item 
The latter
%tuples in the expansion
%$\widetilde{\mathcal{T}}_{2}$
%\mathrmbfit{tup}_{\mathcal{S}}(g)(\mathcal{T}_{2})$
%that were taken away to define the subtraction
%$\mathcal{T}{\,\thicksim\,}\mathcal{T}_{2}$,
can be ``added'' back
with the 
%data-type 
join 
$(\mathcal{T}{\,\thicksim\,}\mathcal{T}_{2}){\;\oplus\;}\mathcal{T}_{2}
\cong
(\mathcal{T}{\,-\,}\widetilde{\mathcal{T}}_{2}) 
%(\widehat{\mathcal{T}}{\;-\;}\mathrmbfit{tup}_{\mathcal{S}}(g)(\mathcal{T}_{2}))
{\;\vee\;}
\widetilde{\mathcal{T}}_{2}
%\mathrmbfit{tup}_{\mathcal{S}}(g)(\mathcal{T}_{2})
$, 
%$\mathcal{T}{\,\thicksim\,}\mathcal{T}_{2}$,
thus defining a table morphism
%\newline\mbox{}\hfill
$\mathcal{T}\rightarrow
(\mathcal{T}{\,\thicksim\,}\mathcal{T}_{2})
{\;\oplus\;}\mathcal{T}_{2}$.
\comment{Here we make use of the ``contraction'' expansion
{\footnotesize{$\mathrmbf{Tbl}_{\mathcal{S}}(\mathcal{A}_{2})
{\;\xleftarrow
%[{\scriptscriptstyle\sum}_{\natural}]
{\acute{\mathrmbfit{tbl}}_{\mathcal{S}}(\triangledown)}
\;}
\mathrmbf{Tbl}_{\mathcal{S}}(\mathcal{A}_{2}{\times}\mathcal{A}_{2})$}}
for the $X$-type domain morphism
$\mathcal{A}_{2}
\xrightarrow{\;\triangledown\;}
\mathcal{A}_{2}{\times}\mathcal{A}_{2}$
with (surjective) data value function
{\footnotesize{$Y_{2}\xleftarrow{\;\triangledown\;}Y_{2}{+}Y_{2}$}\normalsize}
mentioned in 
\S\,\ref{sub:sub:sec:adj:flow:S}.
}
%%%%%%%%%%%%%%%%%%%%%%%%%%%%%%%%%%%%%%%%%%%%%%%%%%%%%%%%%%%%%%%%%%%%%%
%%%%%%%%%%%%%%%%%%%%%%%%%%%%%%%%%%%%%%%%%%%%%%%%%%%%%%%%%%%%%%%%%%%%%%
\comment{The sum
$(\mathcal{T}{\,\thicksim\,}\mathcal{T}_{2}){\;\oplus\;}\mathcal{T}_{2}$
is a table in 
$\mathrmbf{Tbl}_{\mathcal{S}}
(\mathcal{A}{\times}\mathcal{A}_{2})$.
%(\mathcal{A}_{1}{\times}\mathcal{A}_{2}{\times}\mathcal{A}_{2})$.
\comment{with product $X$-type domain 
$\mathcal{A}_{1}{\times}\mathcal{A}_{2}{\times}\mathcal{A}_{2}$ 
%in the span
%{\footnotesize{$\mathcal{A}_{1}\xleftarrow{\tilde{0}_{1}\,}
%{\mathcal{A}_{1}{\times}\mathcal{A}_{2}}
%\xrightarrow{\;\tilde{0}_{2}}\mathcal{A}_{2}$}\normalsize}
with data value functions
{\footnotesize{$Y_{1}\xhookrightarrow{\tilde{0}_{1}\,} 
Y_{1}{+}Y_{2}{+}Y_{2}
\xhookleftarrow{\;\tilde{0}_{2}}Y_{2}$.}\normalsize}}}
%%%%%%%%%%%%%%%%%%%%%%%%%%%%%%%%%%%%%%%%%%%%%%%%%%%%%%%%%%%%%%%%%%%%%%
%%%%%%%%%%%%%%%%%%%%%%%%%%%%%%%%%%%%%%%%%%%%%%%%%%%%%%%%%%%%%%%%%%%%%%
%
\hfill\rule{5pt}{5pt}
\end{proof}
%

%%%%%%%%%%%%%%%%%%%%%%%%%%%%%%%%%%%%%%%%%%%%%%%%%%%%%%%%%%%%%%
%
\newpage
\subsection{Division.}\label{sub:sub:sec:div}
%$\ast{-}{-}$}
%%%%%%%%%%%%%%%%%%%%%%%%%%%%%%%%%%%%%%%%%%%%%%%%%%%%%%%%%%%%%%

%
\begin{figure}
\begin{center}
{{{\begin{tabular}{c}
\begin{picture}(120,130)(32.5,30)
\setlength{\unitlength}{0.73pt}
%%%%%%%%%%%%%%%%%%%%%%%%%%%%%%%%%%%%%%%%%%%%%%%%%%
\put(95,191){\begin{picture}(0,0)(0,0)
\setlength{\unitlength}{0.35pt}
%\thicklines
%\put(106,40){\makebox(0,0){\normalsize{$\boldsymbol{\circ}$}}}
%\put(4.7,40){\makebox(0,0){\normalsize{$\boldsymbol{\circ}$}}}
\put(40,10){\line(1,0){60}}
\put(40,70){\line(1,0){60}}
\put(100,70){\line(0,-1){60}}
\put(40,40){\oval(60,60)[bl]}
\put(40,40){\oval(60,60)[tl]}
\put(58,50){\makebox(0,0){\scriptsize{{\textit{{project}}}}}}
\put(56,30){\makebox(0,0){\Large{${\Leftarrow}$}}}
%\put(61,22){\makebox(0,0){\tiny{{\textit{{image}}}}}}
\end{picture}}
%%%%%%%%%%%%%%%%%%%%%%%%%%%%%%%%%%%%%%%%%%%%%%%%%%
\put(151,149.7){\begin{picture}(0,0)(0,3)
\setlength{\unitlength}{0.35pt}
\put(60,29){\makebox(0,0){\normalsize{$\times$}}}
%\thicklines
%\put(33,75){\makebox(0,0){\normalsize{$\boldsymbol{\circ}$}}}
%\put(87,75){\makebox(0,0){\normalsize{$\boldsymbol{\circ}$}}}
%\put(60,5){\makebox(0,0){\normalsize{$\boldsymbol{\circ}$}}}
\put(40,10){\line(1,0){40}}
\put(9,70){\line(1,0){102}}
\put(10,70){\line(0,-1){30}}
\put(110,70){\line(0,-1){30}}
\put(40,40){\oval(60,60)[bl]}
\put(80,40){\oval(60,60)[br]}
%\put(61,58){\makebox(0,0){\scriptsize{{\textit{{Cart}}}}}}
\put(61,50){\makebox(0,0){\scriptsize{{\textit{{prod}}}}}}
\end{picture}}
%%%%%%%%%%%%%%%%%%%%%%%%%%%%%%%%%%%%%%%%%%%%%%%%%%
\put(151,100){\begin{picture}(0,0)(0,3)
\setlength{\unitlength}{0.35pt}
%\put(60,40){\makebox(0,0){\large{$-$}}}
%\thicklines
\put(50,33){\line(1,0){20}}
%\put(33,75){\makebox(0,0){\normalsize{$\boldsymbol{\circ}$}}}
%\put(87,75){\makebox(0,0){\normalsize{$\boldsymbol{\circ}$}}}
%\put(60,5){\makebox(0,0){\normalsize{$\boldsymbol{\circ}$}}}
\put(40,10){\line(1,0){40}}
\put(10,70){\line(1,0){100}}
\put(10,70){\line(0,-1){30}}
\put(110,70){\line(0,-1){30}}
\put(40,40){\oval(60,60)[bl]}
\put(80,40){\oval(60,60)[br]}
\put(61,50){\makebox(0,0){\scriptsize{{\textit{{diff}}}}}}
\end{picture}}
%%%%%%%%%%%%%%%%%%%%%%%%%%%%%%%%%%%%%%%%%%%%%%%%%%
\put(95,71){\begin{picture}(0,0)(0,0)
\setlength{\unitlength}{0.35pt}
%\thicklines
%\put(106,40){\makebox(0,0){\normalsize{$\boldsymbol{\circ}$}}}
%\put(4.7,40){\makebox(0,0){\normalsize{$\boldsymbol{\circ}$}}}
\put(40,10){\line(1,0){60}}
\put(40,70){\line(1,0){60}}
\put(100,70){\line(0,-1){60}}
\put(40,40){\oval(60,60)[bl]}
\put(40,40){\oval(60,60)[tl]}
\put(58,50){\makebox(0,0){\scriptsize{{\textit{{project}}}}}}
\put(56,30){\makebox(0,0){\Large{${\Leftarrow}$}}}
%\put(61,22){\makebox(0,0){\tiny{{\textit{{image}}}}}}
\end{picture}}
%%%%%%%%%%%%%%%%%%%%%%%%%%%%%%%%%%%%%%%%%%%%%%%%%%
\put(31,40){\begin{picture}(0,0)(0,3)
\setlength{\unitlength}{0.35pt}
%\put(60,40){\makebox(0,0){\large{$-$}}}
%\thicklines
\put(50,33){\line(1,0){20}}
%\put(33,75){\makebox(0,0){\normalsize{$\boldsymbol{\circ}$}}}
%\put(87,75){\makebox(0,0){\normalsize{$\boldsymbol{\circ}$}}}
%\put(60,5){\makebox(0,0){\normalsize{$\boldsymbol{\circ}$}}}
\put(40,10){\line(1,0){40}}
\put(10,70){\line(1,0){100}}
\put(10,70){\line(0,-1){30}}
\put(110,70){\line(0,-1){30}}
\put(40,40){\oval(60,60)[bl]}
\put(80,40){\oval(60,60)[br]}
\put(61,50){\makebox(0,0){\scriptsize{{\textit{{diff}}}}}}
\end{picture}}
%%%%%%%%%%%%%%%%%%%%%%%%%%%%%%%%%%%%%%%%%%%%%%%%%%
%%%%%%%%%%%%%%%%%%%%%%%%%%%%%%%%%%%%%%%%%%%%%%%%%%
\put(150,56){\makebox(0,0){\footnotesize{{\textit{{division}}}}}}
%\put(40,32){\makebox(0,0){\footnotesize{{\textit{{left}}}}}}
\put(150,40){\makebox(0,0){\large{$\div$}}}
%%%%%%%%%%%%%%%%%%%%%%%%%%%%%%%%%%%%%%%%%%%%%%%%%%
%\put(0,210){\line(1,0){240}}
%\put(0,180){\line(1,0){240}}
%\put(0,150){\line(1,0){240}}
%\put(0,120){\line(1,0){240}}
%\put(0,90){\line(1,0){240}}
%\put(0,60){\line(1,0){240}}
%\put(0,30){\line(1,0){240}}
%\put(60,0){\line(0,1){240}}
%\put(120,0){\line(0,1){240}}
%\put(180,0){\line(0,1){240}}
\put(193,230){\vector(0,-1){50}}
\put(154,210){\line(0,1){20}}\put(154,210){\vector(-1,0){11}}
\put(154,210){\line(1,0){37}}\put(195,210){\line(1,0){21}}
\put(216,210){\line(0,-1){67}}\put(216,143){\line(-1,0){23}}\put(193,143){\vector(0,-1){12}}
\put(46,210){\line(1,0){53}}\put(46,210){\vector(0,-1){139}}
\put(46,190){\line(1,0){121}}\put(167,190){\vector(0,-1){10}}
\put(166,143){\vector(0,-1){12}}
\put(165.8,143){\line(1,0){14.4}}
\put(180,143){\line(0,1){8.8}}
\put(180,90){\vector(-1,0){37}}\put(180,90){\line(0,1){12}}
\put(73,90){\line(1,0){26}}\put(73,90){\vector(0,-1){19}}
\put(60,42){\vector(0,-1){20}}
%\thicklines
%\put(15,240){\line(1,0){210}}
%\put(15,50){\line(0,1){190}}
%\put(225,50){\line(0,1){190}}
%\put(55,10){\line(1,0){130}}
%\qbezier(15,50)(15,10)(55,10)
%\qbezier(185,10)(225,10)(225,50)
%%%%%%%%%%%%%%%%%%%%%%%%%%%%%%%%%%%%%%%%%%%%%%%%%%
\end{picture}
\end{tabular}}}}
\end{center}
\caption{\texttt{FOLE} Division Flow Chart}
\label{fole:division:flo:chrt}
\end{figure}
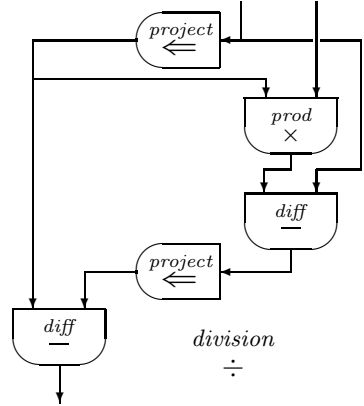
Here we define the operation of division 
in terms of the basic operations. 
Let $\mathcal{A} = {\langle{X,Y,\models_{\mathcal{A}}}\rangle}$
be a fixed type domain.
Let
$\mathcal{T}$ and $\mathcal{T}_{2}$ be two $\mathcal{A}$-tables.
For division,
we assume that the attribute names in the header of $\mathcal{T}_{2}$ 
are a subset of those of $\mathcal{T}$. 
The division $\mathcal{T}{\,\div_{\!\mathcal{A}}\,}\mathcal{T}_{2}$ 
consists of 
the contraction of tuples in $\mathcal{T}$ 
to the attribute names unique to $\mathcal{T}$
(the projection to $\mathcal{S}_{1}$),
%i.e., in the header of $\mathcal{T}$ but not in the header of $\mathcal{T}_{2}$, 
for which it holds that 
all their combinations with tuples in $\mathcal{T}_{2}$ 
(the Cartesian product)
are present in $\mathcal{T}$ 
(take the difference in $\mathcal{S}_{1}{+}\mathcal{S}_{2}$, 
then project to $\mathcal{S}_{1}$, 
finally take the difference in $\mathcal{S}_{1}$).
Division is defined by the multi-step process
%illustrated 
in Fig.\;\ref{fole:division:flo:chrt}.
\begin{description}
\item[Constraint:] 
The constraint for division (Tbl.\,\ref{tbl:fole:division:input:output})
is the same as the constraint for Cartesian product:
two $X$-signatures $\mathcal{S}_{1}$ and $\mathcal{S}_{2}$.
These are linked by the span of $X$-signatures 
{\footnotesize{$\mathcal{S}_{1}\xhookleftarrow{0_{I_{1}}} 
\mathcal{S}_{\bot}
\xhookrightarrow{0_{I_{2}}}\mathcal{S}_{2}$}\normalsize}
in $\mathrmbf{List}(X)$
consisting of 
a span of injection index functions
{\footnotesize{$I_{1}
\xhookleftarrow{0_{I_{1}}}
\emptyset
\xhookrightarrow{0_{I_{2}}}
I_{2}$.}}
\newline
\item[Construction:] 
The construction for division 
(Tbl.\,\ref{tbl:fole:division:input:output})
is the same as the construction for Cartesian product:
%
%The pushout of this constraint
%%colimiting cocone of this signature span
%in $\mathrmbf{List}(X)$
%is 
the opspan
{\footnotesize{$\mathcal{S}_{1} 
\xhookrightarrow{\iota_{1}\,} 
{\mathcal{S}_{1}{+\;}\mathcal{S}_{2}}
\xhookleftarrow{\;\iota_{2}} 
\mathcal{S}_{2}$}\normalsize}
of injection $X$-signature morphisms
with
coproduct $X$-signature
$\mathcal{S}_{1}{+}\mathcal{S}_{2}$
and injection
index function opspan
{\footnotesize{$I_{1} 
\xhookrightarrow{\iota_{1}\,} 
{\langle{I_{1}{+}I_{2},[s_{1},s_{2}]}\rangle}
\xhookleftarrow{\;\iota_{2}}I_{2}.$}\normalsize}
\newline
\item[Input:] Consider a pair of tables
$\mathcal{T} \in \mathrmbf{Tbl}_{\mathcal{A}}(\mathcal{S}_{1}{+}\mathcal{S}_{2})$
and
$\mathcal{T}_{2} \in \mathrmbf{Tbl}_{\mathcal{A}}(\mathcal{S}_{2})$.
This is the input for division 
(Tbl.\,\ref{tbl:fole:division:input:output}).
\newpage
\item[Output:] 
Division is projection, Cartesian product, difference, projection and difference.
\newline
\begin{itemize}
\item 
%\begin{description}
%\item[project:] 
Project table 
$\mathcal{T} \in \mathrmbf{Tbl}_{\mathcal{A}}(\mathcal{S}_{1}{+}\mathcal{S}_{2})$
back along the injection $X$-signature morphism
{\footnotesize{$\mathcal{S}_{1} 
\xhookrightarrow{\iota_{1}\,} 
{\mathcal{S}_{1}{+\;}\mathcal{S}_{2}}$}\normalsize}
to its unique attribute names $\mathcal{S}_{1}$,
getting the table
$\acute{\mathrmbfit{tbl}}_{\mathcal{A}}(\iota_{1})(\mathcal{T}) 
\in \mathrmbf{Tbl}_{\mathcal{A}}(\mathcal{S}_{1})$. 
\newline
%\newline\item[product:] 
%
\item 
Form the Cartesian product 
with the table 
$\mathcal{T}_{2} \in \mathrmbf{Tbl}_{\mathcal{A}}(\mathcal{S}_{2})$,
getting the table
$\widehat{\mathcal{T}}=
\acute{\mathrmbfit{tbl}}_{\mathcal{A}}(\iota_{1})(\mathcal{T}){\times}\mathcal{T}_{2} 
\in \mathrmbf{Tbl}_{\mathcal{A}}(\mathcal{S}_{1}{+}\mathcal{S}_{2})$. 
\newline
%\newline
%The Cartesian product 
%$\widehat{\mathcal{T}}=
%\acute{\mathrmbfit{tbl}}_{\mathcal{A}}(\iota_{1})(\mathcal{T})
%{\times}\mathcal{T}_{2}$
%is a table in $\mathrmbf{Tbl}_{\mathcal{A}}(\mathcal{S}_{1}{+}\mathcal{S}_{2})$. 
%\newline
%
\item 
%
%\item[removal:] 
The difference 
$\bar{\mathcal{T}} =
\widehat{\mathcal{T}}{\,-\,}\mathcal{T}
\in \mathrmbf{Tbl}_{\mathcal{A}}(\mathcal{S}_{1}{+}\mathcal{S}_{2})$
gives the unwanted tuples.
%,
%also a table in $\mathrmbf{Tbl}_{\mathcal{A}}(\mathcal{S}_{1}{+}\mathcal{S}_{2})$. 
\newline
%
%The projection
%$\tilde{\mathcal{T}}=
%\acute{\mathrmbfit{tbl}}_{\mathcal{A}}(\iota_{1})(\bar{\mathcal{T}})$
%is a table in $\mathrmbf{Tbl}_{\mathcal{A}}(\mathcal{S}_{1})$. 
%
\item 
The subtraction
$\mathcal{T}{\,\div_{\!\mathcal{A}}\,}\mathcal{T}_{2} = 
\acute{\mathrmbfit{tbl}}_{\mathcal{A}}(\iota_{1})(\mathcal{T})
{\,-\,}
\acute{\mathrmbfit{tbl}}_{\mathcal{A}}(\iota_{1})(\bar{\mathcal{T}})
\in \mathrmbf{Tbl}_{\mathcal{A}}(\mathcal{S}_{1})$
is the desired table. 
\newline
%\end{description}
%
\end{itemize}
Hence,
the division of $\mathcal{T}$ by $\mathcal{T}_{2}$ 
is defined by the multi-step process:
%has the following expression:
\newline\mbox{}\hfill
\newline\mbox{}\hfill
$\mathcal{T}{\,\div_{\!\mathcal{A}}\,}\mathcal{T}_{2}
\doteq 
\acute{\mathrmbfit{tbl}}_{\mathcal{A}}(\iota_{1})(\mathcal{T})
{\,-\,}
\acute{\mathrmbfit{tbl}}_{\mathcal{A}}(\iota_{1})
\Bigl(
\bigl(\acute{\mathrmbfit{tbl}}_{\mathcal{A}}(\iota_{1})(\mathcal{T})
{\times}\mathcal{T}_{2}
\bigr)
{\,-\,}\mathcal{T}
\Bigr)$.
\hfill\mbox{}\newline
%\in \mathrmbf{Tbl}_{\mathcal{A}}(I_{1},s_{1})$
\end{description}
\begin{table}
\begin{center}
{{\fbox{\begin{tabular}{c}
\setlength{\extrarowheight}{2pt}
{\scriptsize{$\begin{array}[c]{c@{\hspace{12pt}}l}
\mathcal{S}_{1}  
\text{ and }
\mathcal{S}_{2} 
&
\textit{constraint}
\\
\mathcal{S}_{1} \xhookrightarrow{\iota_{1}\,} 
{\mathcal{S}_{1}{+}\mathcal{S}_{2}}
\xhookleftarrow{\;\iota_{2}}\mathcal{S}_{2}
&
\textit{construction}
\\
\hline
\mathcal{T} \in \mathrmbf{Tbl}_{\mathcal{A}}(\mathcal{S}_{1}{+}\mathcal{S}_{2})
\text{ and }
\mathcal{T}_{2} \in \mathrmbf{Tbl}_{\mathcal{A}}(\mathcal{S}_{2})
&
\textit{input}
\\
\mathcal{T}{\,\div_{\!\mathcal{A}}\,}\mathcal{T}_{2}
&
\textit{output}
\end{array}$}}
\end{tabular}}}}
\end{center} 
\caption{\texttt{FOLE} Division I/O}
\label{tbl:fole:division:input:output}
\end{table}
\begin{proposition}\label{div:morph}
There is a table morphism
%%%%%%%%%%%%%%%%%%%%%%%%%%%%%%%%%%%%%%%%%%%%%%%%%%%%%%%%%%%%
%%%%%%%%%%%%%%%%%%%%%%%%%%%%%%%%%%%%%%%%%%%%%%%%%%%%%%%%%%%%
\footnote{Analogous to 
$7{\;\geq\;}\underset{2}{\underbrace{(7{\,\div\,}3)}}{\;\times\;}3$
or modus ponens $q \leftarrow ((q \leftarrow p) \wedge p)$.}
%%%%%%%%%%%%%%%%%%%%%%%%%%%%%%%%%%%%%%%%%%%%%%%%%%%%%%%%%%%%
%%%%%%%%%%%%%%%%%%%%%%%%%%%%%%%%%%%%%%%%%%%%%%%%%%%%%%%%%%%%
%%%%%%%%%%%%%%%%%%%%%%%%%%%%%%%%%%%%%%%%%%%%%%%%%%%%%%%%%%%%
%\newline\mbox{}\hfill
\[\mbox
{\footnotesize{{$
\mathcal{T}\xhookleftarrow{\;\;}\rule[-10pt]{0pt}{26pt}
(\mathcal{T}{\,\div_{\!\mathcal{A}}\,}\mathcal{T}_{2}){\;\times\;}\mathcal{T}_{2}
$.}}\normalsize}
\]
%\hfill\mbox{}
\end{proposition}
\begin{proof}
%
%Footnote\;\ref{footnote:refl}.
%%%%%%%%%%%%%%%%%%%%%%%%%%%%%%%%%%%%%%%%%%%%%%%%%%%%%%%%%%%%%%%%%%%%%%
%%%%%%%%%%%%%%%%%%%%%%%%%%%%%%%%%%%%%%%%%%%%%%%%%%%%%%%%%%%%%%%%%%%%%%
\footnote{Argument in terms of the underlying relations in the reflection
Prop.\;\ref{tbl:rel:refl}.}
%%%%%%%%%%%%%%%%%%%%%%%%%%%%%%%%%%%%%%%%%%%%%%%%%%%%%%%%%%%%%%%%%%%%%%
%%%%%%%%%%%%%%%%%%%%%%%%%%%%%%%%%%%%%%%%%%%%%%%%%%%%%%%%%%%%%%%%%%%%%%
%
Suppose 
$t \in (\mathcal{T}{\,\div_{\!\mathcal{A}}\,}\mathcal{T}_{2}){\;\times\;}\mathcal{T}_{2}$,
so that 
$t = (t_{1},t_{2})$ for 
$t_{1} \in (\mathcal{T}{\,\div_{\!\mathcal{A}}\,}\mathcal{T}_{2})$ 
and
$t_{2} \in \mathcal{T}_{2}$.
By definition,
$t_{1} \in \acute{\mathrmbfit{tbl}}_{\mathcal{A}}(\iota_{1})(\mathcal{T})$
and
$t_{1} \not\in
\acute{\mathrmbfit{tbl}}_{\mathcal{A}}(\iota_{1})
\Bigl(
\bigl(\acute{\mathrmbfit{tbl}}_{\mathcal{A}}(\iota_{1})(\mathcal{T})
{\times}\mathcal{T}_{2}
\bigr)
{\,-\,}\mathcal{T}
\Bigr)$.
%Hence,
%$t_{1} =
%\acute{\mathrmbfit{tbl}}_{\mathcal{A}}(\iota_{1})(t)$,
%where
%$t = (t_{1},t_{2}) \in \acute{\mathrmbfit{tbl}}_{\mathcal{A}}(\iota_{1})(\mathcal{T})
%{\times}\mathcal{T}_{2}$
%
This implies that
$t = (t_{1},t_{2}) \in  
\bigl(\acute{\mathrmbfit{tbl}}_{\mathcal{A}}(\iota_{1})(\mathcal{T})
{\times}\mathcal{T}_{2}
\bigr)$.
Suppose $t \not\in \mathcal{T}$.
Then
$t \in  
\bigl(\acute{\mathrmbfit{tbl}}_{\mathcal{A}}(\iota_{1})(\mathcal{T})
{\times}\mathcal{T}_{2}
\bigr)
{\,-\,}\mathcal{T}$.
%Suppose $t \not\in \mathcal{T}$.
Hence,
$t_{1} =
\acute{\mathrmbfit{tbl}}_{\mathcal{A}}(\iota_{1})(t)
\in
\acute{\mathrmbfit{tbl}}_{\mathcal{A}}(\iota_{1})
\Bigl(
\bigl(\acute{\mathrmbfit{tbl}}_{\mathcal{A}}(\iota_{1})(\mathcal{T})
{\times}\mathcal{T}_{2}
\bigr)
{\,-\,}\mathcal{T}
\Bigr)$,
a contradiction.
Hence,
$t \in \mathcal{T}$.
\hfill\rule{5pt}{5pt}
\end{proof}
%

%%%%%%%%%%%%%%%%%%%%%%%%%%%%%%%%%%%%%%%%%%%%%%%%%%%%%%%%%%%%%%%%%%%%%%%
%
\newpage
\subsection{Outer-join.}\label{sub:sub:sec:out:join}
%$\ast{-}{-}$}$\bigstar$ 
%%%%%%%%%%%%%%%%%%%%%%%%%%%%%%%%%%%%%%%%%%%%%%%%%%%%%%%%%%%%%%%%%%%%%%%
%

%
\begin{figure}
\begin{center}
{{\begin{tabular}{c}
\setlength{\unitlength}{0.6pt}
\begin{picture}(480,240)(-38,-44)
%%%%%%%%%%%%%%%%%%%%%%%%%%%%%%%%%%%%%%%%%%%%%%%%%%
\thicklines
%\put(0,180){\line(1,0){480}}
%\put(0,90){\line(1,0){480}}
%\put(0,0){\line(1,0){480}}
%\put(0,0){\line(0,1){180}}
%\put(120,0){\line(0,1){180}}
%\put(240,0){\line(0,1){180}}
%\put(360,0){\line(0,1){180}}
%\put(480,0){\line(0,1){180}}
\thinlines
%%%%%%%%%%%%%%%%%%%%%%%%%%%%%%%%%%%%%%%%%%%%%%%%%%
\put(0,0){\makebox(0,0){\scriptsize{${\bullet}$}}}
\put(240,36){\makebox(0,0){\scriptsize{${\bullet}$}}}
\put(120,120){\makebox(0,0){\scriptsize{$\textsl{left square}$}}}
\put(120,90){\makebox(0,0){\scriptsize{$\blacksquare$}}}
\put(65,130){\makebox(0,0){\footnotesize{$\mathcal{A}_{\bullet}$}}}
\put(175,60){\makebox(0,0){\footnotesize{$\mathcal{A}$}}}
%\dashline{3}(-80,100)(320,100)
\qbezier[100](-80,80)(20,105)(120,90)
\qbezier[100](120,90)(210,60)(380,80)
%%%%%%%%%%%%%%%%%%%%%%%%%%%%%%%%%%%%%%%%%%%%%%%%%%
%%%%%%%%%%%%%%%%%%%%%%%%%%%%%%%%%%%%%%%%%%%%%%%%%%
\put(420,180){\makebox(0,0)[l]{\scriptsize{$\mathcal{T}_{\bullet}:\bar{\mathcal{S}}_{2}$}}}
\put(420,0){\makebox(0,0)[l]{\scriptsize{$\mathcal{T}_{2}:\mathcal{S}_{2}$}}}
\put(-10,180){\makebox(0,0)[r]{\scriptsize{$
\acute{\acute{\mathcal{T}}}_{1} =
\acute{\mathrmbfit{tbl}}_{\mathcal{S}}(\tilde{0})(\acute{\mathcal{T}}_{1})
:\mathcal{S}_{1}$}}}
\put(250,153){\makebox(0,0)[l]{\scriptsize{$
\acute{\acute{\mathcal{T}}}_{1}{\,\times\,}\mathcal{T}_{\bullet}
:\mathcal{S}_{1}{+\,}\bar{\mathcal{S}}_{2}$}}}
\put(-10,64){\makebox(0,0)[r]{\scriptsize{
$\acute{\mathcal{T}}_{1}=\mathcal{T}_{1}{\,\boxslash_{\mathcal{A}}}\mathcal{T}_{2}
:\mathcal{S}_{1}$}}}
\put(-40,0){\makebox(0,0)[r]{\scriptsize{$\mathcal{T}_{1}:\mathcal{S}_{1}$}}}
\put(305,125){\makebox(0,0)[l]{\scriptsize{$
\mathcal{T}_{1}{\,\rgroup\!\boxtimes}_{\mathcal{A}}\mathcal{T}_{2} 
:\mathcal{S}_{1}{+\,}\bar{\mathcal{S}}_{2}\cong
\mathcal{S}_{1}{+_{\mathcal{S}}}\mathcal{S}_{2}$}}}
\put(250,97){\makebox(0,0)[l]{\scriptsize{$
\acute{\mathcal{T}}_{12}=
\acute{\mathrmbfit{tbl}}_{\mathcal{S}}(\tilde{0})(\mathcal{T}_{12})
:\mathcal{S}_{1}{+_{\mathcal{S}}}\mathcal{S}_{2}$}}}
\put(250,35){\makebox(0,0)[l]{\scriptsize{$
\mathcal{T}_{12} =
\mathcal{T}_{1}{\,\boxtimes_{\mathcal{A}}}\mathcal{T}_{2}
:\mathcal{S}_{1}{+_{\mathcal{S}}}\mathcal{S}_{2}$}}}
%%%%%%%%%%%%%%%%%%%%%%%%%%%%%%%%%%%%%%%%%%%%%%%%%%
%%%%%%%%%%%%%%%%%%%%%%%%%%%%%%%%%%%%%%%%%%%%%%%%%%
\put(-32,75){\begin{picture}(0,0)(0,3)
\setlength{\unitlength}{0.3pt}
\put(60,48){\makebox(0,0){\large{$\Uparrow$}}}
\put(10,10){\line(1,0){100}}
\put(40,70){\line(1,0){40}}
\put(10,40){\line(0,-1){30}}
\put(110,40){\line(0,-1){30}}
\put(40,40){\oval(60,60)[tl]}
%\put(40,40){\oval(60,60)[bl]}
\put(80,40){\oval(60,60)[tr]}
\put(60,23){\makebox(0,0){\scriptsize{{\textit{{expand}}}}}}
\put(60,70){\line(0,1){142}}
\end{picture}}
%%%%%%%%%%%%%%%%%%%%%%%%%%%%%%%%%%%%%%%%%%%%%%%%%%
%\put(208,75){\begin{picture}(0,0)(0,3)
\put(210,55){\begin{picture}(0,0)(0,3)
\setlength{\unitlength}{0.3pt}
\put(60,48){\makebox(0,0){\large{$\Uparrow$}}}
\put(10,10){\line(1,0){100}}
\put(40,70){\line(1,0){40}}
\put(10,40){\line(0,-1){30}}
\put(110,40){\line(0,-1){30}}
\put(40,40){\oval(60,60)[tl]}
%\put(40,40){\oval(60,60)[bl]}
\put(80,40){\oval(60,60)[tr]}
\put(60,23){\makebox(0,0){\scriptsize{{\textit{{expand}}}}}}
\put(60,70){\vector(0,1){44}}
%\put(60,136){\line(1,0){75}}
\end{picture}}
%%%%%%%%%%%%%%%%%%%%%%%%%%%%%%%%%%%%%%%%%%%%%%%%%%
%\put(72,158){\begin{picture}(0,0)(0,0)
\put(92,158){\begin{picture}(0,0)(0,0)
\setlength{\unitlength}{0.3pt}
%\thicklines
%\put(106,40){\makebox(0,0){\normalsize{$\boldsymbol{\circ}$}}}
%\put(5.8,40){\makebox(0,0){\normalsize{$\boldsymbol{\circ}$}}}
\put(10,10){\line(1,0){60}}
\put(10,70){\line(1,0){60}}
\put(10,70){\line(0,-1){60}}
\put(70,40){\oval(60,60)[br]}
\put(70,40){\oval(60,60)[tr]}
\put(55,50){\makebox(0,0){\scriptsize{{\textit{{inflate}}}}}}
\put(56,30){\makebox(0,0){\large{${\Rightarrow}$}}}
\put(-188,40){\vector(1,0){200}}
\put(100,40){\vector(1,0){142}}
\end{picture}}
%%%%%%%%%%%%%%%%%%%%%%%%%%%%%%%%%%%%%%%%%%%%%%%%%%
\put(92,14){\begin{picture}(0,0)(0,0)
\setlength{\unitlength}{0.3pt}
%\thicklines
%\put(106,40){\makebox(0,0){\normalsize{$\boldsymbol{\circ}$}}}
%\put(4.7,40){\makebox(0,0){\normalsize{$\boldsymbol{\circ}$}}}
\put(40,10){\line(1,0){60}}
\put(40,70){\line(1,0){60}}
\put(100,70){\line(0,-1){60}}
\put(40,40){\oval(60,60)[bl]}
\put(40,40){\oval(60,60)[tl]}
\put(58,50){\makebox(0,0){\scriptsize{{\textit{{project}}}}}}
\put(56,30){\makebox(0,0){\large{${\Leftarrow}$}}}
\put(10,42){\vector(-1,0){140}}
\put(295,42){\vector(-1,0){195}}
\end{picture}}
%%%%%%%%%%%%%%%%%%%%%%%%%%%%%%%%%%%%%%%%%%%%%%%%%%
\put(92,-21){\begin{picture}(0,0)(0,0)
\setlength{\unitlength}{0.3pt}
%\thicklines
%\put(106,40){\makebox(0,0){\normalsize{$\boldsymbol{\circ}$}}}
%\put(5.8,40){\makebox(0,0){\normalsize{$\boldsymbol{\circ}$}}}
\put(10,10){\line(1,0){60}}
\put(10,70){\line(1,0){60}}
\put(10,70){\line(0,-1){60}}
\put(70,40){\oval(60,60)[br]}
\put(70,40){\oval(60,60)[tr]}
\put(55,50){\makebox(0,0){\scriptsize{{\textit{{inflate}}}}}}
\put(56,30){\makebox(0,0){\large{${\Rightarrow}$}}}
\put(-250,40){\vector(1,0){260}}
\put(100,40){\vector(1,0){144}}
\end{picture}}
%%%%%%%%%%%%%%%%%%%%%%%%%%%%%%%%%%%%%%%%%%%%%%%%%%
\put(330,158){\begin{picture}(0,0)(0,0)
\setlength{\unitlength}{0.3pt}
%\thicklines
%\put(106,40){\makebox(0,0){\normalsize{$\boldsymbol{\circ}$}}}
%\put(4.7,40){\makebox(0,0){\normalsize{$\boldsymbol{\circ}$}}}
\put(40,10){\line(1,0){60}}
\put(40,70){\line(1,0){60}}
\put(100,70){\line(0,-1){60}}
\put(40,40){\oval(60,60)[bl]}
\put(40,40){\oval(60,60)[tl]}
\put(55,50){\makebox(0,0){\scriptsize{{\textit{{inflate}}}}}}
\put(56,30){\makebox(0,0){\large{${\Leftarrow}$}}}
\put(10,43){\vector(-1,0){136}}
\put(176,42){\vector(-1,0){76}}
\end{picture}}
%%%%%%%%%%%%%%%%%%%%%%%%%%%%%%%%%%%%%%%%%%%%%%%%%%
\put(330,-21){\begin{picture}(0,0)(0,0)
\setlength{\unitlength}{0.3pt}
%\thicklines
%\put(106,40){\makebox(0,0){\normalsize{$\boldsymbol{\circ}$}}}
%\put(4.7,40){\makebox(0,0){\normalsize{$\boldsymbol{\circ}$}}}
\put(40,10){\line(1,0){60}}
\put(40,70){\line(1,0){60}}
\put(100,70){\line(0,-1){60}}
\put(40,40){\oval(60,60)[bl]}
\put(40,40){\oval(60,60)[tl]}
\put(55,50){\makebox(0,0){\scriptsize{{\textit{{inflate}}}}}}
\put(56,30){\makebox(0,0){\large{${\Leftarrow}$}}}
\put(10,43){\vector(-1,0){136}}
\put(176,42){\vector(-1,0){76}}
\end{picture}}
%%%%%%%%%%%%%%%%%%%%%%%%%%%%%%%%%%%%%%%%%%%%%%%%%%
%\put(200,180){\begin{picture}(0,0)(0,0)
\put(240,180){\begin{picture}(0,0)(0,0)
\setlength{\unitlength}{0.55pt}
\put(0,0){\oval(58,34)[r]}
\put(0,0){\oval(58,34)[l]}
\put(0,8){\makebox(0,0){\footnotesize{$\textit{meet}$}}}
\put(0,-8){\makebox(0,0){\large{$\wedge$}}}
%\put(30,0){\vector(1,0){23}}
\put(0,-16){\vector(0,-1){28}}
\end{picture}}
%%%%%%%%%%%%%%%%%%%%%%%%%%%%%%%%%%%%%%%%%%%%%%%%%%
%\put(280,180){\begin{picture}(0,0)(0,0)
\put(240,125){\begin{picture}(0,0)(0,0)
\setlength{\unitlength}{0.55pt}
\put(0,0){\oval(58,34)[r]}
\put(0,0){\oval(58,34)[l]}
\put(0,8){\makebox(0,0){\footnotesize{$\textit{join}$}}}
\put(0,-8){\makebox(0,0){\large{$\vee$}}}
\put(29,0){\vector(1,0){36}}
%\put(0,-36){\vector(0,1){19}}
\end{picture}}
%%%%%%%%%%%%%%%%%%%%%%%%%%%%%%%%%%%%%%%%%%%%%%%%%%
\put(240,0){\begin{picture}(0,0)(0,0)
\setlength{\unitlength}{0.55pt}
\put(0,0){\oval(58,34)[r]}
\put(0,0){\oval(58,34)[l]}
\put(0,8){\makebox(0,0){\footnotesize{$\textit{meet}$}}}
\put(0,-8){\makebox(0,0){\large{$\wedge$}}}
\put(0,16){\vector(0,1){46}}
%\put(-10,-35){\vector(0,1){18}}
%\put(10,-35){\vector(0,1){18}}
\end{picture}}
%%%%%%%%%%%%%%%%%%%%%%%%%%%%%%%%%%%%%%%%%%%%%%%%%%
\put(0,35){\begin{picture}(0,0)(0,0)
\setlength{\unitlength}{0.55pt}
\put(0,0){\oval(58,34)[r]}
\put(0,0){\oval(58,34)[l]}
\put(0,8){\makebox(0,0){\footnotesize{$\textit{diff}$}}}
\put(0,-8){\makebox(0,0){\large{$-$}}}
\put(0,17){\vector(0,1){26}}
\put(0,-37){\vector(0,1){19}}
%\thicklines
%\put(0,10){\line(1,0){20}}
\end{picture}}
%%%%%%%%%%%%%%%%%%%%%%%%%%%%%%%%%%%%%%%%%%%%%%%%%%
\end{picture}
\end{tabular}}}
\end{center}
\caption{\texttt{FOLE} Left Outer Join Flow Chart}
\label{fole:outer:join:in:square}
\end{figure}
An outer join takes two operands as in the natural join.
It contains 
tuples formed by combining matching tuples in the two operands
(the natural join), plus (union)
tuples formed by extending an unmatched tuple in one of the operands
(an anti-join) 
by 
``null values for each of the attributes of the other operand''.

Let $\mathcal{A} = {\langle{X,Y,\models_{\mathcal{A}}}\rangle}$
be a fixed type domain.
%\begin{description}
Consider a pair of tables
$\mathcal{T}_{1} = {\langle{K_{1},t_{1}}\rangle} \in 
\mathrmbf{Tbl}_{\mathcal{A}}(\mathcal{S}_{1})$
and
$\mathcal{T}_{2} = {\langle{K_{2},t_{2}}\rangle} \in 
\mathrmbf{Tbl}_{\mathcal{A}}(\mathcal{S}_{2})$
%\item[Complement Signature:] 
%
%Let $\mathcal{T}_{1}$ and $\mathcal{T}_{2}$
%be two $\mathcal{A}$-tables that are 
linked through 
an underlying $X$-sorted signature span 
$\mathcal{S}_{1}\xhookleftarrow{h_{1}}\mathcal{S}\xhookrightarrow{h_{2}}\mathcal{S}_{2}$
with pushout opspan
{{\footnotesize{$
\mathcal{S}_{1}\xhookrightarrow{\iota_{1}\,} 
{\mathcal{S}_{1}{+_{\mathcal{S}}}\mathcal{S}_{2}}
\xhookleftarrow{\;\iota_{2}}\mathcal{S}_{2}
$.}\normalsize}}
%
%%%%%%%%%%%%%%%%%%%%%%%%%%%%%%%%%%%%%%%%%%%%%%%%%%%%%%%%%%%%%%%%%%%%%%
%%%%%%%%%%%%%%%%%%%%%%%%%%%%%%%%%%%%%%%%%%%%%%%%%%%%%%%%%%%%%%%%%%%%%%
\footnote{In this section we assume the index functions are injective
$I_{1}\xhookleftarrow{h_{1}}I\xhookrightarrow{h_{2}}I_{2}$.}
%Otherwise,
%we will need to add an inflation step before forming the Cartesian product.}
%%%%%%%%%%%%%%%%%%%%%%%%%%%%%%%%%%%%%%%%%%%%%%%%%%%%%%%%%%%%%%%%%%%%%%
%%%%%%%%%%%%%%%%%%%%%%%%%%%%%%%%%%%%%%%%%%%%%%%%%%%%%%%%%%%%%%%%%%%%%%
%
Define the $X$-signature
$\bar{\mathcal{S}}_{2}
={\langle{\bar{I}_{2},\bar{s}_{2}}\rangle}
={\langle{I_{2}{\,\setminus\,}I_{1},\bar{s}_{2}}\rangle}$
to be the restriction of 
$\mathcal{S}_{2}={\langle{I_{2},s_{2}}\rangle}$
to the complement index set
$\bar{I}_{2} = I_{2}{\,\setminus\,}I_{1}$.
Form the coproduct signature
in the opspan
%\newline\mbox{}\hfill
{{\footnotesize{$
\mathcal{S}_{1} 
%= {\langle{I_{1},s_{1}}\rangle} 
\xhookrightarrow{\iota_{1}\,} 
%\overset{\textstyle
{\mathcal{S}_{1}{+\,}\bar{\mathcal{S}}_{2}}
%}{\overbrace{{\langle{I_{1}{+}_{I}I_{2},[s_{1},s_{2}]}\rangle}}} 
\xhookleftarrow{\;\bar{\iota}_{2}} 
%{\langle{I_{2},s_{2}}\rangle} = 
\bar{\mathcal{S}}_{2}
$.}\normalsize}}
Then,
the pushout signature
%opspan 
above
is isomorphic 
to 
this coproduct signature:
$\mathcal{S}_{1}{+_{\mathcal{S}}}\mathcal{S}_{2}
\cong
\mathcal{S}_{1}{+\,}\bar{\mathcal{S}}_{2}$.

%in the opspan
%%\newline\mbox{}\hfill
%{{\footnotesize{$
%\mathcal{S}_{1} 
%%= {\langle{I_{1},s_{1}}\rangle} 
%\xhookrightarrow{\iota_{1}\,} 
%%\overset{\textstyle
%{\mathcal{S}_{1}{+\,}\bar{\mathcal{S}}_{2}}
%}{\overbrace{{\langle{I_{1}{+}_{I}I_{2},[s_{1},s_{2}]}\rangle}}} 
%\xhookleftarrow{\;\bar{\iota}_{2}} 
%{\langle{I_{2},s_{2}}\rangle} = 
%\bar{\mathcal{S}}_{2}
%$.}\normalsize}}

%\hfill\mbox{}\newline
%with the sum
%\newline

%
%\item[Null Type Domain:] 
To define the outer join, we need to expand the type domain 
$\mathcal{A} = {\langle{X,Y,\models_{\mathcal{A}}}\rangle}$ 
to the type domain 
$\mathcal{A}_{\bullet} = {\langle{X,Y{+}\{\cdot\},\models_{\mathcal{A}_{\bullet}}}\rangle}$ 
by adding a null value `$\cdot$':
(i)
$y{\;\models_{\mathcal{A}_{\bullet}}\;}x$
when
$y{\;\models_{\mathcal{A}}\;}x$
for all $x \in X$ and $y \in Y$; and
(ii)
${\cdot}{\;\models_{\mathcal{A}_{\bullet}}\;}x$
for all $x \in X$.
Hence,
in $\mathcal{A}_{\bullet}$
the (data type) 
extent of any sort $x \in X$
is the subset
$\mathrmbfit{ext}_{\mathcal{A}_{\bullet}}(x) = A_{x}{+}{\{\cdot\}}$.
%
%%%%%%%%%%%%%%%%%%%%
%\newpage
%%%%%%%%%%%%%%%%%%%%
The inclusion data value function 
$Y{+}\{\cdot\}\xhookleftarrow{\;\tilde{0}\;}Y$
defines an $X$-sorted type domain morphism
{{$\mathcal{A}_{\bullet}
%={\langle{X,Y{+}\{\cdot\},\models_{\mathcal{A}_{\bullet}}}\rangle}
\xhookrightarrow{\;\tilde{0}\;}
%{\langle{X,Y,\models_{\mathcal{A}}}\rangle}=
\mathcal{A}$.}}
%the identity type function $X\xrightarrow{\;1_{X}\;}X$ and 
%satisfying 
%the following condition
%\[\mbox
%{\footnotesize{
%$y=i(y){\;\models_{\mathcal{A}_{\bullet}}\;}x$
%\underline{iff}
%$y {\;\models_{\mathcal{A}}\;} x$
%}\normalsize}
%\]
%for any source sort $x{\,\in\,}X$ 
%and target data value $y{\,\in\,}Y$,
%or equivalently that
%$i^{-1}({\mathcal{A}_{\bullet}}_{x}) = A_{x}$ 
%for any source sort $x{\,\in\,}X$.
%
This defines a fiber adjunction of tables
%
%\newline\mbox{}\hfill
%\rule[8pt]{0pt}{10pt}
{\footnotesize{$
\mathrmbf{Tbl}_{\mathcal{S}}(\mathcal{A}_{\bullet})
{\;\xleftarrow
%[{\bigl\langle{{\scriptscriptstyle\sum}_{g}{\;\dashv\;}{g}^{\ast}\bigr\rangle}}]
{{\bigl\langle{\acute{\mathrmbfit{tbl}}_{\mathcal{S}}(\tilde{0})
{\;\dashv\;}
\grave{\mathrmbfit{tbl}}_{\mathcal{S}}(\tilde{0})}\bigr\rangle}}\;}
\mathrmbf{Tbl}_{\mathcal{S}}(\mathcal{A})$,}}
%\hfill\mbox{}\newline
%for an $X$-sorted type domain morphism
%$\mathcal{A}_{\bullet}
%%={\langle{X,Y{+}\{\cdot\},\models_{\mathcal{A}_{\bullet}}}\rangle}
%\xrightarrow{\;i\;}
%%{\langle{X,Y,\models_{\mathcal{A}}}\rangle}=
%\mathcal{A}$,
%
where the left adjoint 
%$\mathrmbf{Tbl}_{\mathcal{S}}(\mathcal{A}_{\bullet})
%{\;\xleftarrow{{\;\acute{\mathrmbfit{tbl}}_{\mathcal{S}}(\tilde{0})\;}}\;}
%\mathrmbf{Tbl}_{\mathcal{S}}(\mathcal{A})$
defines expansion.
%%%%%%%%%%%%%%%%%%%%%%%%%%%%%%%%%%%%%%%%%%%%%%%%%%%%%%%%%%%%%%%%%%%%%%
%%%%%%%%%%%%%%%%%%%%%%%%%%%%%%%%%%%%%%%%%%%%%%%%%%%%%%%%%%%%%%%%%%%%%%
%\footnote{
%We add one extra basic operator with no input that gives the null table as output.}
%%%%%%%%%%%%%%%%%%%%%%%%%%%%%%%%%%%%%%%%%%%%%%%%%%%%%%%%%%%%%%%%%%%%%%
%%%%%%%%%%%%%%%%%%%%%%%%%%%%%%%%%%%%%%%%%%%%%%%%%%%%%%%%%%%%%%%%%%%%%%
%\end{description}
%
%
%%%%%%%%%%%%%%%%%%%%%%%%%%%%%%%%%%%%%%%%%%%%%%%%%%%%%%%%%%%%
%%%%%%%%%%%%%%%%%%%%%%%%%%%%%%%%%%%%%%%%%%%%%%%%%%%%%%%%%%%%
%\footnote{
Consider a third table
$\mathcal{T}_{\bullet} = {\langle{1,{({\cdot},...,{\cdot})}}\rangle} \in 
\mathrmbf{Tbl}_{\mathcal{A}_{\bullet}}(\bar{\mathcal{S}}_{2})$
%The $\mathcal{A}_{\bullet}$-table
%$\mathcal{T}_{\bullet} 
%= {\langle{\bar{I}_{2},\bar{s}_{2},1,{({\cdot},...,{\cdot})}}\rangle}$
with underlying signature
$\bar{\mathcal{S}}_{2} = {\langle{\bar{I}_{2},\bar{s}_{2}}\rangle}$
that consists of a single tuple of null values.
%}
%%%%%%%%%%%%%%%%%%%%%%%%%%%%%%%%%%%%%%%%%%%%%%%%%%%%%%%%%%%%
%%%%%%%%%%%%%%%%%%%%%%%%%%%%%%%%%%%%%%%%%%%%%%%%%%%%%%%%%%%%
%with the one tuple filled with null values
%
%%%%%%%%%%%%%%%%%%%%
%%%%%%%%%%%%%%%%%%%%
%\newpage
%%%%%%%%%%%%%%%%%%%%
%%%%%%%%%%%%%%%%%%%%
%

\newpage

\begin{description}
\item[Constraint:] 
The constraint for outer-join 
has two parts:
(1)
an $X$-sorted signature span 
$\mathcal{S}_{1}\xleftarrow{h_{1}}\mathcal{S}\xrightarrow{h_{2}}\mathcal{S}_{2}$
in $\mathrmbf{List}(X)$,
the constraint for natural join
(Tbl.\,\ref{tbl:fole:natural:join:input:output});
and
(2)
an $X$-sorted type domain morphism
$\mathcal{A}_{\bullet}\xhookrightarrow{\;\tilde{0}\;}\mathcal{A}$
in $\mathrmbf{Cls}(X)$.
\newline
\item[Construction:] 
%
%\newpage
The flowchart
for left outer-join is illustrated in 
Fig.\;\ref{fole:outer:join:in:square}.
Since outer-join involves two type domains,
only the set of sorts $X$ is in common.
Taking together the discussions about signatures and type domains,
this defines an opspan of signed domain morphisms
\newline\mbox{}\hfill
{{$
{\langle{\mathcal{S}_{1},\mathcal{A}_{\bullet}}\rangle}
\xrightarrow{{\langle{\iota_{1},1_{X},\tilde{0}}\rangle}}
\underset{\textstyle{\langle{{\mathcal{S}_{1}{+\,}\bar{\mathcal{S}}_{2}},\mathcal{A}}\rangle}}
{\langle{{\mathcal{S}_{1}{+_{\mathcal{S}}}\mathcal{S}_{2}},\mathcal{A}}\rangle}
\xleftarrow
[{\langle{\bar{\iota}_{2},1_{X},\tilde{0}}\rangle}]
{{\langle{\iota_{2},1_{X},\tilde{0}}\rangle}}
{\langle{\mathcal{S}_{2},\mathcal{A}_{\bullet}}\rangle}
$.}}
\hfill\mbox{}\newline
%\newline
which factors as visualized in the following two squares.
%\begin{figure}
\begin{center}
{{\begin{tabular}{c}
%
%%%%%%%%%%%%%%%%%%%%%%%%%%%%%%%%%%%%%%%%%%%%%%%%%%%%%%%%%%%%%%%%%%%%%%%%%%%%%%%%
%%%%%%%%%%%%%%%%%%%%%%%%%%%%%%%%%%%%%%%%%%%%%%%%%%%%%%%%%%%%%%%%%%%%%%%%%%%%%%%%
\comment{{\begin{tabular}{c}
\setlength{\unitlength}{0.37pt}
{{$
{\langle{\mathcal{S}_{1},\mathcal{A}_{\bullet}}\rangle}
\xrightarrow{{\langle{\iota_{1},1_{X},\tilde{0}}\rangle}}
\underset{\textstyle{\langle{{\mathcal{S}_{1}{+\,}\bar{\mathcal{S}}_{2}},\mathcal{A}}\rangle}}
{\langle{{\mathcal{S}_{1}{+_{\mathcal{S}}}\mathcal{S}_{2}},\mathcal{A}}\rangle}
\xleftarrow
[{\langle{\bar{\iota}_{2},1_{X},\tilde{0}}\rangle}]
{{\langle{\iota_{2},1_{X},\tilde{0}}\rangle}}
{\langle{\mathcal{S}_{2},\mathcal{A}_{\bullet}}\rangle}
$.}}
\end{tabular}}}
%%%%%%%%%%%%%%%%%%%%%%%%%%%%%%%%%%%%%%%%%%%%%%%%%%%%%%%%%%%%%%%%%%%%%%%%%%%%%%%%
%%%%%%%%%%%%%%%%%%%%%%%%%%%%%%%%%%%%%%%%%%%%%%%%%%%%%%%%%%%%%%%%%%%%%%%%%%%%%%%%
%\\
{{\begin{tabular}{c}
\setlength{\unitlength}{0.37pt}
\begin{picture}(520,205)(-8,-5)
%
%\put(120,240){\makebox(0,0){\normalsize{${
%\overset{\textstyle{\mathrmbf{Tbl}(\mathcal{A}')}}{\overbrace{\hspace{120pt}}}}$}}}
%\put(120,-60){\makebox(0,0){\normalsize{${
%\underset{\textstyle{\mathrmbf{Tbl}(\mathcal{A})}}{\underbrace{\hspace{120pt}}}}$}}}
%\put(-50,90){\makebox(0,0)[r]{\footnotesize{$
%\mathrmbf{Tbl}(\mathcal{S}')\left\{\rule{0pt}{50pt}\right.$}}}
%\put(290,90){\makebox(0,0)[l]{\footnotesize{$
%\left.\rule{0pt}{50pt}\right\}\mathrmbf{Tbl}(\mathcal{S})$}}}
%
\put(0,180){\makebox(0,0){\footnotesize{$
{\langle{\mathcal{S}_{1},\mathcal{A}_{\bullet}}\rangle}$}}}
\put(0,0){\makebox(0,0){\footnotesize{$
{\langle{\mathcal{S}_{1},\mathcal{A}}\rangle}$}}}
\put(480,180){\makebox(0,0){\footnotesize{$
{\langle{\mathcal{S}_{2},\mathcal{A}_{\bullet}}\rangle}$}}}
\put(480,0){\makebox(0,0){\footnotesize{$
{\langle{\mathcal{S}_{2},\mathcal{A}}\rangle}$}}}
\put(240,180){\makebox(0,0){\footnotesize{$
{\langle{\mathcal{S}_{1}{+_{\mathcal{S}}}\mathcal{S}_{2},\mathcal{A}_{\bullet}}\rangle}$}}}
\put(240,0){\makebox(0,0){\footnotesize{$
{\langle{\mathcal{S}_{1}{+_{\mathcal{S}}}\mathcal{S}_{2},\mathcal{A}}\rangle}$}}}
\put(105,195){\makebox(0,0){\scriptsize{$
{\langle{\iota_{1},1_{\mathcal{A}_{\bullet}}}\rangle}$}}}
\put(105,165){\makebox(0,0){\scriptsize{$=\,\iota_{1}$}}}
%{\langle{\iota_{1},1_{X},1_{Y{+}1}}\rangle}$}}}
\put(105,15){\makebox(0,0){\scriptsize{$
{\langle{\iota_{1},1_{\mathcal{A}}}\rangle}$}}}
\put(105,-15){\makebox(0,0){\scriptsize{$=\,\iota_{1}$}}}
%\iota_{1}$}}}
\put(375,195){\makebox(0,0){\scriptsize{$
{\langle{\iota_{2},1_{\mathcal{A}_{\bullet}}}\rangle}$}}}
\put(375,165){\makebox(0,0){\scriptsize{$=\,\iota_{2}$}}}
%{\langle{\iota_{1},1_{X},1_{Y{+}1}}\rangle}$}}}
\put(375,15){\makebox(0,0){\scriptsize{$
{\langle{\iota_{2},1_{\mathcal{A}}}\rangle}$}}}
\put(375,-15){\makebox(0,0){\scriptsize{$=\,\iota_{2}$}}}
%\iota_{1}$}}}
\put(-10,95){\makebox(0,0)[r]{\scriptsize{$
{\langle{1_{\mathcal{S}_{1}},\tilde{0}}\rangle}$}}}
\put(-20,70){\makebox(0,0)[r]{\scriptsize{$=\,\tilde{0}$}}}
%{\langle{1_{I},1_{X},\tilde{0}}\rangle}$}}}
\put(195,110){\makebox(0,0)[l]{\scriptsize{$
{\langle{1_{\mathcal{S}_{1}{+_{\mathcal{S}}}\mathcal{S}_{2}},\tilde{0}}\rangle}$}}}
\put(255,85){\makebox(0,0)[l]{\scriptsize{$=\,\tilde{0}$}}}
\put(495,95){\makebox(0,0)[l]{\scriptsize{$
{\langle{1_{\mathcal{S}_{2}},\tilde{0}}\rangle}$}}}
\put(510,70){\makebox(0,0)[l]{\scriptsize{$=\,\tilde{0}$}}}
%{\langle{1_{I},1_{X},\tilde{0}}\rangle}$}}}
%\put(500,95){\makebox(0,0)[l]{\scriptsize{$\tilde{0}$}}}
\put(140,60){\makebox(0,0)[r]
{\scriptsize{${\langle{\iota_{1},1_{X},\tilde{0}}\rangle}$}}}
\put(340,60){\makebox(0,0)[l]
{\scriptsize{${\langle{\iota_{2},1_{X},\tilde{0}}\rangle}$}}}
\put(60,180){\vector(1,0){100}}
\put(60,0){\vector(1,0){100}}
\put(420,180){\vector(-1,0){100}}
\put(420,0){\vector(-1,0){100}}
\put(0,150){\vector(0,-1){120}}
\put(240,150){\vector(0,-1){120}}
\put(480,150){\vector(0,-1){120}}
\put(20,160){\vector(4,-3){180}}
\put(460,160){\vector(-4,-3){180}}
%
%\put(-45,85){\makebox(0,0)[r]{\huge{${\textit{\scriptsize{expand}}}$}}}
%\put(-30,90){\makebox(0,0){\huge{${\Uparrow}$}}}
%\put(215,90){\makebox(0,0){\huge{${\Uparrow}$}}}
%\put(215,50){\makebox(0,0){\huge{${\textit{\scriptsize{expand}}}$}}}
%\put(110,155){\makebox(0,0){\huge{${\Rightarrow}$}}}
%\put(110,135){\makebox(0,0){\huge{${\textit{\scriptsize{inflate}}}$}}}
%\put(370,155){\makebox(0,0){\huge{${\Leftarrow}$}}}
%\put(370,135){\makebox(0,0){\huge{${\textit{\scriptsize{inflate}}}$}}}
%\put(70,-25){\makebox(0,0){\huge{${\Leftarrow}$}}}
%\put(70,-45){\makebox(0,0){\huge{${\textit{\scriptsize{project}}}$}}}
%\put(150,-25){\makebox(0,0){\huge{${\Rightarrow}$}}}
%\put(150,-45){\makebox(0,0){\huge{${\textit{\scriptsize{inflate}}}$}}}
%\put(210,225){\makebox(0,0){\Large{${\wedge}$}}}
%\put(210,205){\makebox(0,0){\huge{${\textit{\scriptsize{meet}}}$}}}
%\put(260,225){\makebox(0,0){\Large{${\vee}$}}}
%\put(260,205){\makebox(0,0){\huge{${\textit{\scriptsize{join}}}$}}}
%\put(240,-25){\makebox(0,0){\Large{${\wedge}$}}}
%\put(240,-45){\makebox(0,0){\huge{${\textit{\scriptsize{meet}}}$}}}
%\put(370,-25){\makebox(0,0){\huge{${\Leftarrow}$}}}
%\put(370,-45){\makebox(0,0){\huge{${\textit{\scriptsize{inflate}}}$}}}
%
%\put(120,90){\makebox(0,0){\scriptsize{$\square$}}}
\put(120,90){\makebox(0,0){\scriptsize{$\blacksquare$}}}
%\put(120,90){\makebox(0,0){\scriptsize{$\square$}}}
\put(360,90){\makebox(0,0){\scriptsize{$\blacksquare$}}}
\end{picture}
\end{tabular}}}
%%%%%%%%%%%%%%%%%%%%%%%%%%%%%%%%%%%%%%%%%%%%%%%%%%%%%%%%%%%%%%%%%%%%%%%%%%%%%%%%
%%%%%%%%%%%%%%%%%%%%%%%%%%%%%%%%%%%%%%%%%%%%%%%%%%%%%%%%%%%%%%%%%%%%%%%%%%%%%%%%
\end{tabular}}}
\end{center}
%\caption{The Double Square}
%\label{fig:double:square}
%\end{figure}
Comparing the above figure with 
Fig.\;\ref{fig:square} in \S\;\ref{sub:sub:sec:adj:flo:square},
we see that flow in the square is appropriate.
\newline
\item[Input:] 
The input for outer-join 
is
%has two parts:
a pair of tables
$\mathcal{T}_{1} 
%= {\langle{K_{1},t_{1}}\rangle} 
\in \mathrmbf{Tbl}_{\mathcal{A}}(\mathcal{S}_{1})$
and
$\mathcal{T}_{2} 
%= {\langle{K_{2},t_{2}}\rangle} 
\in \mathrmbf{Tbl}_{\mathcal{A}}(\mathcal{S}_{2})$
as in the input for natural join 
(Tbl.\,\ref{tbl:fole:natural:join:input:output}).
%
%%%%%%%%%%%%%%%%%%%%%%%%%%%%%%%%%%%%%%%%%%%%%%%%%%%%%%%%%%%%%%%%%%%%%%%%%%%%%%%%
%%%%%%%%%%%%%%%%%%%%%%%%%%%%%%%%%%%%%%%%%%%%%%%%%%%%%%%%%%%%%%%%%%%%%%%%%%%%%%%%
\footnote{The table
$\mathcal{T}_{\bullet} 
%= {\langle{1,{({\cdot},...,{\cdot})}}\rangle} 
\in \mathrmbf{Tbl}_{\mathcal{A}_{\bullet}}(\bar{\mathcal{S}}_{2})$
%The $\mathcal{A}_{\bullet}$-table
%$\mathcal{T}_{\bullet} 
%= {\langle{\bar{I}_{2},\bar{s}_{2},1,{({\cdot},...,{\cdot})}}\rangle}$
%with underlying signature
%$\bar{\mathcal{S}}_{2} = {\langle{\bar{I}_{2},\bar{s}_{2}}\rangle}$
is an implicit table that helps build the output.}
%%%%%%%%%%%%%%%%%%%%%%%%%%%%%%%%%%%%%%%%%%%%%%%%%%%%%%%%%%%%%%%%%%%%%%%%%%%%%%%%
%%%%%%%%%%%%%%%%%%%%%%%%%%%%%%%%%%%%%%%%%%%%%%%%%%%%%%%%%%%%%%%%%%%%%%%%%%%%%%%%
%
\newline
\item[Output:] 
%
%Since each signed domain morphism has a fixed sort set $X$,
%%table processing suggests using adjoint flow in the square %(\S\;\ref{sub:sub:sec:adj:flo:square}).%
%
%%Hence,
%outer-join is best conceived in terms of adjoint flow in the square 
%(\S\,\ref{sub:sub:sec:adj:flo:square})
%and a better flowchart 
%for outer-join is illustrated in Fig.\;\ref{fole:outer:join:in:square}.
%
%Hence,
Left outer-join is defined by the following multi-step process.
\begin{itemize}
\item[\textbf{1}:] 
The \emph{natural join} 
(\S\,\ref{sub:sub:sec:nat:join}) 
{{$
\mathcal{T}_{12} =
\mathcal{T}_{1}{\,\boxtimes}_{\mathcal{A}}\mathcal{T}_{2} =
\grave{\mathrmbfit{tbl}}_{\mathcal{A}}(\iota_{1})(\mathcal{T}_{1})
{\;\wedge\;}
\grave{\mathrmbfit{tbl}}_{\mathcal{A}}(\iota_{2})(\mathcal{T}_{2})
$}\normalsize}
is the meet of inflations 
along 
$X$-signature opspan
{\footnotesize{$\mathcal{S}_{1}\xhookrightarrow{\iota_{1}\,} 
{\mathcal{S}_{1}{+_{\mathcal{S}}}\mathcal{S}_{2}}
\xhookleftarrow{\;\iota_{2}}\mathcal{S}_{2}$.}\normalsize}
%\newline
%
\item[\textbf{2}:] 
The \emph{left semi-join} 
(\S\,\ref{sub:sub:sec:semi:join})
${\mathcal{T}_{1}}{\,\boxleft_{\mathcal{A}}}{\mathcal{T}_{2}} 
= \acute{\mathrmbfit{tbl}}_{\mathcal{A}}(\iota_{1})
\bigl(\mathcal{T}_{1}{\,\boxtimes_{\mathcal{A}}}\mathcal{T}_{2}\bigr)
\in \mathrmbf{Tbl}_{\mathcal{S}_{1}}(\mathcal{A})$
is the projection 
(\S\,\ref{sub:sub:sec:adj:flow:A})
of the natural join 
along the $X$-signature morphism
{\footnotesize{$\mathcal{S}_{1}\xhookrightarrow{\iota_{1}\,} 
{\mathcal{S}_{1}{+_{\mathcal{S}}}\mathcal{S}_{2}}$.}\normalsize}
%\mbox{}\newline
\item[\textbf{3}:]
The \emph{left anti-join}
(\S\,\ref{sub:sub:sec:anti:join})
$\acute{\mathcal{T}}_{1} =
{\mathcal{T}_{1}}{\,\boxslash_{\mathcal{A}}}{\mathcal{T}_{2}} =
\mathcal{T}_{1}{\,-\,}({\mathcal{T}_{1}}{\,\boxleft_{\mathcal{A}}}{\mathcal{T}_{2}})
%\in \mathrmbf{Tbl}_{\mathcal{A}}(\mathcal{S}_{1})
\in \mathrmbf{Tbl}_{\mathcal{S}_{1}}(\mathcal{A})$
is the difference between the table $\mathcal{T}_{1}$ and the semi-join
$\mathcal{T}_{1}{\,\boxleft_{\mathcal{A}}}\mathcal{T}_{2}$.
%\mbox{}\newline
\item[\textbf{4}:]
The \emph{expansion} 
(\S\,\ref{sub:sub:sec:adj:flow:S})
$\acute{\acute{\mathcal{T}}}_{1} =
\acute{\mathrmbfit{tbl}}_{\mathcal{S}}(\tilde{0})({\mathcal{T}_{1}}{\,\boxslash_{\mathcal{A}}}{\mathcal{T}_{2}})
\in \mathrmbf{Tbl}_{\mathcal{S}_{1}}(\mathcal{A}_{\bullet})$
is the left adjoint flow
along the $X$-sorted type domain morphism
$\mathcal{A}_{\bullet}\xhookrightarrow{\;\tilde{0}\;}\mathcal{A}$.
%defines the table
%$\acute{\mathcal{T}}_{1}$.
%
\item[\textbf{5}:]
The \emph{Cartesian product} 
(\S\,\ref{sub:sub:sec:nat:join}) 
$\acute{\acute{\mathcal{T}}}_{1}{\,\times\,}\mathcal{T}_{\bullet} =
\grave{\mathrmbfit{tbl}}_{\mathcal{A}}(\iota_{1})(\acute{\acute{\mathcal{T}}}_{1})
{\;\wedge\;}
\grave{\mathrmbfit{tbl}}_{\mathcal{A}}(\iota_{2})(\mathcal{T}_{\bullet})$
%\in 
%\mathrmbf{Tbl}_{\mathcal{A}_{\bullet}}(\mathcal{S}_{1}{+\,}\bar{\mathcal{S}}_{2})
%\cong 
%\mathrmbf{Tbl}_{\mathcal{A}_{\bullet}}(\mathcal{S}_{1}{+_{\mathcal{S}}}\mathcal{S}_{2})
is the meet of inflations 
along the 
$X$-signature opspan
{\footnotesize{$\mathcal{S}_{1}\xhookrightarrow{\iota_{1}\,} 
{\mathcal{S}_{1}{+}\bar{\mathcal{S}}_{2}}
\xhookleftarrow{\;\bar{\iota}_{2}}\bar{\mathcal{S}}_{2}$.}\normalsize}
%
%Since $\mathcal{S}_{1}{+_{\mathcal{S}}}\mathcal{S}_{2}
%\cong
%\mathcal{S}_{1}{+\,}\bar{\mathcal{S}}_{2}$,
Here
$\acute{\acute{\mathcal{T}}}_{1}{\,\times\,}\mathcal{T}_{\bullet}
\in \mathrmbf{Tbl}_{\mathcal{A}_{\bullet}}(\mathcal{S}_{1}{+\,}\bar{\mathcal{S}}_{2})
\cong 
\mathrmbf{Tbl}_{\mathcal{A}_{\bullet}}(\mathcal{S}_{1}{+_{\mathcal{S}}}\mathcal{S}_{2})$.
%\newline
%
\item[\textbf{6}:]
The \emph{expansion} 
(\S\,\ref{sub:sub:sec:adj:flow:S})
$\acute{\mathcal{T}}_{12}=
\acute{\mathrmbfit{tbl}}_{\mathcal{S}}(\tilde{0})(
\mathcal{T}_{1}{\,\boxtimes_{\mathcal{A}}}\mathcal{T}_{2})
\in \mathrmbf{Tbl}_{\,\mathcal{S}_{1}{+_{\mathcal{S}}}\mathcal{S}_{2}}(\mathcal{A}_{\bullet})$
is the left adjoint flow
%of the natural join
along the $X$-sorted type domain morphism
$\mathcal{A}_{\bullet}\xhookrightarrow{\;\tilde{0}\;}\mathcal{A}$.
%\newline
%
\item[\textbf{7}:]
The left outer-join 
$\mathcal{T}_{1}{\,\rgroup\!\boxtimes}_{\mathcal{A}}\mathcal{T}_{2} =
\acute{\mathcal{T}}_{12}
\vee
(\acute{\acute{\mathcal{T}}}_{1}{\,\times\,}\mathcal{T}_{\bullet})$
%(\acute{\acute{\mathcal{T}}}_{1}{\,\boxtimes_{\mathcal{A}}}\mathcal{T}_{\bullet})$
is union of
%\begin{itemize}
%\item 
the expansion
$\acute{\mathcal{T}}_{12}$
%of the natural join $\mathcal{T}_{1}{\,\boxtimes}_{\mathcal{A}}\mathcal{T}_{2}$
%the set of all combinations of tuples in $\mathcal{T}_{1}$ and $\mathcal{T}_{2}$ 
%that are equal on their common attribute names in $\mathcal{S}$, 
%(the RHS of Disp.\;\ref{def:tbl:assoc:law})
and
%\item 
the Cartesian product 
$\acute{\acute{\mathcal{T}}}_{1}{\,\times\,}\mathcal{T}_{\bullet}$.
% =
%\grave{\mathrmbfit{tbl}}_{\mathcal{A}}(\iota_{1})(\acute{\mathcal{T}}_{1})
%{\;\wedge\;}
%\grave{\mathrmbfit{tbl}}_{\mathcal{A}}(\iota_{2})(\mathcal{T}_{\bullet})$.
%the expansion
%of 
%the left anti-join 
%$(\mathcal{T}_{1}{\,\boxslash_{\mathcal{A}}}\mathcal{T}_{2})$.
%
%%%%%%%%%%%%%%%%%%%%%%%%%%%%%%%%%%%%%%%%%%%%%%%%%%%%%%%%%%%%%%%%%%%%%%
%%%%%%%%%%%%%%%%%%%%%%%%%%%%%%%%%%%%%%%%%%%%%%%%%%%%%%%%%%%%%%%%%%%%%%
%\footnote{The set of all tuples in $\mathcal{T}_{1}$
%that have no matching tuples in $\mathcal{T}_{2}$
%on their common attribute names.}
%%%%%%%%%%%%%%%%%%%%%%%%%%%%%%%%%%%%%%%%%%%%%%%%%%%%%%%%%%%%%%%%%%%%%%
%%%%%%%%%%%%%%%%%%%%%%%%%%%%%%%%%%%%%%%%%%%%%%%%%%%%%%%%%%%%%%%%%%%%%%
%
%filled out with the null value `$\cdot$'.
%
The right outer-join has a similar definition. 
\end{itemize}
\end{description}
%

%\newpage

%%%%%%%%%%%%%%%%%%%%%%%%%%%%%%%%%%%%%%%%%%%%%%%%%%
%%%%%%%%%%%%%%%%%%%%%%%%%%%%%%%%%%%%%%%%%%%%%%%%%%
\comment{ % expression and changes
\\\\
If $\mathcal{T}_{12} =
\mathcal{T}_{1}{\,\boxtimes_{\mathcal{A}}}\mathcal{T}_{2} =
(\grave{\mathrmbfit{tbl}}_{\mathcal{A}}(\iota_{1})(\mathcal{T}_{1})
\wedge
\grave{\mathrmbfit{tbl}}_{\mathcal{A}}(\iota_{2})(\mathcal{T}_{2}))$,
\\
then
${\mathcal{T}_{1}}{\,{\rgroup\!\boxtimes}_{\mathcal{A}}}{\mathcal{T}_{2}} =
\acute{\mathrmbfit{tbl}}_{\mathcal{S}}(\iota)(\mathcal{T}_{12})
{\,\vee\,}
\bigl(
\acute{\mathrmbfit{tbl}}_{\mathcal{S}}(\iota)(
\mathcal{T}_{1}{-}(\acute{\mathrmbfit{tbl}}_{\mathcal{A}}(\iota_{1})(\mathcal{T}_{12})))
{\,\times\,}\mathcal{T}_{\bullet}
\bigr)$
\\\hline
Changes to generalize:
$\mathcal{A}_{\bullet}$ to $\mathcal{A}'$,
$\bar{\mathcal{S}}_{2}$ to $\mathcal{S}'$,
$\mathcal{A}_{\bullet}\xrightarrow{\;i\;}\mathcal{A}$
to
$\mathcal{A}'\xrightarrow{\;g\;}\mathcal{A}$
} % expression and changes
%%%%%%%%%%%%%%%%%%%%%%%%%%%%%%%%%%%%%%%%%%%%%%%%%%
%%%%%%%%%%%%%%%%%%%%%%%%%%%%%%%%%%%%%%%%%%%%%%%%%%

%
\begin{table}
\begin{center}
{{\fbox{\begin{tabular}{c}
\setlength{\extrarowheight}{2pt}
{\scriptsize{$\begin{array}[c]{c@{\hspace{12pt}}l}
%\mathcal{S}_{1}  
%\text{ and }
%\mathcal{S}_{2} 
\mathcal{S}_{1}\xleftarrow{h_{1}}\mathcal{S}\xrightarrow{h_{2}}\mathcal{S}_{2}
%in $\mathrmbf{List}(X)$,
%the constraint for natural join
%(Tbl.\,\ref{tbl:fole:natural:join:input:output});
\text{ and }
%(2)
%an $X$-sorted type domain morphism
\mathcal{A}_{\bullet}\xhookrightarrow{\;\tilde{0}\;}\mathcal{A}
&
\textit{constraint}
\\
{\langle{\mathcal{S}_{1},\mathcal{A}_{\bullet}}\rangle}
\xrightarrow{{\langle{\iota_{1},1_{X},\tilde{0}}\rangle}}
\underset{\textstyle{\langle{{\mathcal{S}_{1}{+\,}\bar{\mathcal{S}}_{2}},\mathcal{A}}\rangle}}
{\langle{{\mathcal{S}_{1}{+_{\mathcal{S}}}\mathcal{S}_{2}},\mathcal{A}}\rangle}
\xleftarrow
[{\langle{\bar{\iota}_{2},1_{X},\tilde{0}}\rangle}]
{{\langle{\iota_{2},1_{X},\tilde{0}}\rangle}}
{\langle{\mathcal{S}_{2},\mathcal{A}_{\bullet}}\rangle}
%\mathcal{S}_{1} \xhookrightarrow{\iota_{1}\,} 
%{\mathcal{S}_{1}{+}\mathcal{S}_{2}}
%\xhookleftarrow{\;\iota_{2}}\mathcal{S}_{2}
&
\textit{construction}
\\
\hline
\mathcal{T} \in \mathrmbf{Tbl}_{\mathcal{A}}(\mathcal{S}_{1})
\text{ and }
\mathcal{T}_{2} \in \mathrmbf{Tbl}_{\mathcal{A}}(\mathcal{S}_{2})
&
\textit{input}
\\
\mathcal{T}_{1}{\,\rgroup\!\boxtimes}_{\mathcal{A}}\mathcal{T}_{2}
&
\textit{output}
\end{array}$}}
\end{tabular}}}}
\end{center} 
\caption{\texttt{FOLE} Outer-Join I/O}
\label{tbl:fole:outer:join:input:output}
\end{table}

% eliminated 03-13-2021

% optimizing by temporary elimination of "`transforming the base"'
%%%%%%%%%%%%%%%%%%%%%%%%%%%%%%%%%%%%%%%%%%%%%%%%%%%%%%%%%%%%%%%%%%%%%%%%%%%%%%%%
%%%%%%%%%%%%%%%%%%%%%%%%%%%%%%%%%%%%%%%%%%%%%%%%%%%%%%%%%%%%%%%%%%%%%%%%%%%%%%%%

%%%%%%%%%%%%%%%%%%%%%%%%%%%%%%%%%%%%%%%%%%%%%%%%%%%%%%%%%%%%%%%%%%
%%%%%%%%%%%%%%%%%%%%%%%%%%%%%%%%%%%%%%%%%%%%%%%%%%%%%%%%%%%%%%%%%%
%%%%%%%%%%%%%%%%%%%%%%%%%%%%%%%%%%%%%%%%%%%%%%%%%%%%%%%%%%%%%%%%%%
%\newpage
%\section{Example: The Gene}\label{sec:gene}
%%%%%%%%%%%%%%%%%%%%%%%%%%%%%%%%%%%%%%%%%%%%%%%%%%%%%%%%%%%%%%%%%
%%%%%%%%%%%%%%%%%%%%%%%%%%%%%%%%%%%%%%%%%%%%%%%%%%%%%%%%%%%%%%%%%
%%%%%%%%%%%%%%%%%%%%%%%%%%%%%%%%%%%%%%%%%%%%%%%%%%%%%%%%%%%%%%%%%
%
%We illustrate \texttt{FOLE} relational operations by 
%providing a preliminary model for the \emph{gene}.

%%%%%%%%%%%%%%%%%%%%%%%%%%%%%%%%%%%%%%%%%%%%%%%%%%%%%%%%%%%%%%%%%%
%%%%%%%%%%%%%%%%%%%%%%%%%%%%%%%%%%%%%%%%%%%%%%%%%%%%%%%%%%%%%%%%%%
%%%%%%%%%%%%%%%%%%%%%%%%%%%%%%%%%%%%%%%%%%%%%%%%%%%%%%%%%%%%%%%%%%
%\newpage
\section{Conclusion and Future Work}\label{sec:conclu}
%%%%%%%%%%%%%%%%%%%%%%%%%%%%%%%%%%%%%%%%%%%%%%%%%%%%%%%%%%%%%%%%%
%%%%%%%%%%%%%%%%%%%%%%%%%%%%%%%%%%%%%%%%%%%%%%%%%%%%%%%%%%%%%%%%%
%%%%%%%%%%%%%%%%%%%%%%%%%%%%%%%%%%%%%%%%%%%%%%%%%%%%%%%%%%%%%%%%%

%\begin{description}
%\item[Preface:] 

%%%%%%%%%%%%%%%%%%%%%%%%%%%%%%%%%%%%%%%%%%%%%%%%%%%%%%%%%%%%%%%%%
\paragraph{Conclusion.}
%%%%%%%%%%%%%%%%%%%%%%%%%%%%%%%%%%%%%%%%%%%%%%%%%%%%%%%%%%%%%%%%%

%
This paper describes 
a well-founded semantics
for relational algebra 
in the first-order logical environment \texttt{FOLE},
thus providing a theoretical foundation for relational databases.
%
%In particular,
%this paper defines 
Here we have defined
a typed semantics for the flowcharts 
%used to define complex database queries
used in analyzing, designing, documenting and managing database query processing.
%
%%%%%%%%%%%%%%%%%%%%%%%%%%%%%%%%%%%%%%%%%%%%%%%%%%%%%%%%%%%%%%%%%
%\paragraph{Elaboration.}
%%%%%%%%%%%%%%%%%%%%%%%%%%%%%%%%%%%%%%%%%%%%%%%%%%%%%%%%%%%%%%%%%
%\item[Elaboration:] 
%
%\mbox{}\newline
%{\fbox{\hspace{340pt}}}
%\newline
%
In the \texttt{FOLE} approach to relational algebra,
each relational operator is a composite concept.
The structure of a relational operator is represented by a flowchart
made up of basic components.
%In the FOLE approach,
The basic components are divided into three categories:
reflection,
Boolean operators,
and 
adjoint flow.  
Each flowchart is typed:
the inputs are type
and 
%the type of multiple entries are 
linked,
each step in the flowchart is typed,
and the output is typed.
These types are of three kinds:
either a type domain $\mathcal{A}$,
a signature $\mathcal{S}$,
or a signed domain $\mathcal{D}$.
Implicit in the background of each non-generic flowchart is 
a particular tabular component
that is fixed:
either 
a type domain,
a signature, or
a signed domain. 
%
%%%%%%%%%%%%%%%%%%%%%%%%%%%%%%%%%%%%%%%%%%%%%%%%%%%%%%%%%%%%%%%%%%%%%%%%%%%%%%%%
%%%%%%%%%%%%%%%%%%%%%%%%%%%%%%%%%%%%%%%%%%%%%%%%%%%%%%%%%%%%%%%%%%%%%%%%%%%%%%%%
\footnote{Combining the type at each step 
with the fixed background 
defines the signed domain 
of each table at that step.}
%%%%%%%%%%%%%%%%%%%%%%%%%%%%%%%%%%%%%%%%%%%%%%%%%%%%%%%%%%%%%%%%%%%%%%%%%%%%%%%%
%%%%%%%%%%%%%%%%%%%%%%%%%%%%%%%%%%%%%%%%%%%%%%%%%%%%%%%%%%%%%%%%%%%%%%%%%%%%%%%%
%
%In the \texttt{FOLE} approach,
%each relational operator is regarded as a composite operation
%defined in terms of more basic components.
%
The paper demonstrates that the \texttt{FOLE} approach
for representing the relational model is very natural,
providing a clear approach to its implementation.
%
%\mbox{}\newline
%{\fbox{\hspace{340pt}}}
%\newline
%

%%%%%%%%%%%%%%%%%%%%%%%%%%%%%%%%%%%%%%%%%%%%%%%%%%%%%%%%%%%%%%%%%
\paragraph{Future Work.}
%%%%%%%%%%%%%%%%%%%%%%%%%%%%%%%%%%%%%%%%%%%%%%%%%%%%%%%%%%%%%%%%%
%\item[Equivalence:] 
%
Two forms for the first-order logical environment \texttt{FOLE}
have been developed:
the \emph{classification form} and the \emph{interpretation form}. 
%\begin{itemize}
%\item 
The \emph{classification form} 
of \texttt{FOLE}
is developed in the papers
``The {\ttfamily ERA} of {\ttfamily FOLE}: Foundation''
\cite{kent:fole:era:found}
and
``The {\ttfamily ERA} of {\ttfamily FOLE}: Superstructure''
\cite{kent:fole:era:supstruc}.
%\item 
The \emph{interpretation form} 
of \texttt{FOLE}
is developed in the papers
``The {\ttfamily FOLE} Table''
\cite{kent:fole:era:tbl}
and
the unpublished paper 
``The \texttt{FOLE} Database''
\cite{kent:fole:era:db}.
Both of the latter two papers expand on material found in the paper 
``Database Semantics''
\cite{kent:db:sem}.
%\end{itemize}
%
In the two papers
\cite{kent:fole:era:found}
and
\cite{kent:fole:era:supstruc}
that develop the classification form of \texttt{FOLE},
the classification concept of information flow 
\cite{barwise:seligman:97} 
is used at the two {\ttfamily ERA} levels of entities and attributes.
By using a slight generalization for the classification concept for entities,
the unpublished paper 
``The \texttt{FOLE} Equivalence''
\cite{kent:fole:equiv}
establishes
the equivalence between 
between
the classification and interpretation forms of \texttt{FOLE}.

%%%%%%%%%%%%%%%%%%%%%%%%%%%%%%%%%%%%%%%%%%%%%%%%%%%%%%%%%%%%%%%%%%%%%%
%%%%%%%%%%%%%%%%%%%%%%%%%%%%%%%%%%%%%%%%%%%%%%%%%%%%%%%%%%%%%%%%%%%%%%
%%%%%%%%%%%%%%%%%%%%%%%%%%%%%%%%%%%%%%%%%%%%%%%%%%%%%%%%%%%%%%%%%%%%%%
\newpage
\appendix
\section{Appendix}\label{sec:append}
%%%%%%%%%%%%%%%%%%%%%%%%%%%%%%%%%%%%%%%%%%%%%%%%%%%%%%%%%%%%%%%%%%%%%%
%%%%%%%%%%%%%%%%%%%%%%%%%%%%%%%%%%%%%%%%%%%%%%%%%%%%%%%%%%%%%%%%%%%%%%
%%%%%%%%%%%%%%%%%%%%%%%%%%%%%%%%%%%%%%%%%%%%%%%%%%%%%%%%%%%%%%%%%%%%%%

Here,
we review the various concepts 
that arise in the \texttt{FOLE} approach 
%\cite{kent:fole:era:tbl}
to relational algebra.
We first review tables and their components:
type domains, signatures and signed domains.
%We then discuss how these morphisms foster connections 
%between the various objects in 
%%the \texttt{FOLE} theory.
%the {\texttt{FOLE}} approach to multi-sorted first order logic.
%
We next review the completeness and co-completeness of the mathematical context of tables,
ending with the definitions of sufficiency and adequacy 
that are useful for simplifying the input for relational operations.

%%%%%%%%%%%%%%%%%%%%%%%%%%%%%%%%%%%%%%%%%%%%%%%%%%%%%%%%%%%%%%
%%%%%%%%%%%%%%%%%%%%%%%%%%%%%%%%%%%%%%%%%%%%%%%%%%%%%%%%%%%%%%%
%\newpage
\subsection{Tabular Components}
\label{sub:sec:tbl:comp}
%%%%%%%%%%%%%%%%%%%%%%%%%%%%%%%%%%%%%%%%%%%%%%%%%%%%%%%%%%%%%%
%%%%%%%%%%%%%%%%%%%%%%%%%%%%%%%%%%%%%%%%%%%%%%%%%%%%%%%%%%%%%%

%%%%%%%%%%%%%%%%%%%%%%%%%%%%%%%%%%%%%%%%%%%%%%%%%%%%%%%%%%%%%%
%%%%%%%%%%%%%%%%%%%%%%%%%%%%%%%%%%%%%%%%%%%%%%%%%%%%%%%%%%%%%%%
%\newpage
\paragraph{Signatures.}
%\subsubsection{Signatures}
%\label{sub:sec:sign}
%%%%%%%%%%%%%%%%%%%%%%%%%%%%%%%%%%%%%%%%%%%%%%%%%%%%%%%%%%%%%%
%%%%%%%%%%%%%%%%%%%%%%%%%%%%%%%%%%%%%%%%%%%%%%%%%%%%%%%%%%%%%%

A signature 
\cite{kent:fole:era:tbl}
is a list, 
which represents the header of a relational table;
it provides typing for the tuples permitted in the table. 
%
%%%%%%%%%%%%%%%%%%%%%%%%%%%%%%%%%%%%%%%%%%%%%%%%%%%%%%%%%%%%%%%%%%%%%%
%%%%%%%%%%%%%%%%%%%%%%%%%%%%%%%%%%%%%%%%%%%%%%%%%%%%%%%%%%%%%%%%%%%%%%
\footnote{As we remarked in \cite{kent:fole:era:tbl},
the use of lists for signatures (and tuples) 
follows Codd's recommendation to use attribute names to index the tuples of a relation instead of a numerical ordering.}
%%%%%%%%%%%%%%%%%%%%%%%%%%%%%%%%%%%%%%%%%%%%%%%%%%%%%%%%%%%%%%%%%%%%%%
%%%%%%%%%%%%%%%%%%%%%%%%%%%%%%%%%%%%%%%%%%%%%%%%%%%%%%%%%%%%%%%%%%%%%%
%
A signature
$\mathcal{S}={\langle{I,s,X}\rangle}$
consists of 
a set of sorts $X$,
an indexing set (arity) $I$, and 
a map $I\xrightarrow{\,s\,}X$ from indexes $I$ to sorts $X$.
A signature morphism (list morphism)
{\footnotesize{$
\mathcal{S}'\xrightarrow{{\langle{h,f}\rangle}}\mathcal{S}
$}\normalsize},
from signature (list)
$\mathcal{S}'={\langle{I',x',X'}\rangle}$
to signature (list)
$\mathcal{S}={\langle{I,x,X}\rangle}$,
consists of 
a sort function $X'\xrightarrow{f}X$ and an arity function $I'\xrightarrow{h}I$
satisfying the naturality condition $h{\,\cdot\,}s = s'{\,\cdot\,}f$.
Let $\mathrmbf{List}$ denote 
the mathematical context of signatures and signature morphisms. 
%\comment{
\begin{figure}
\begin{center}
{{\begin{tabular}{c}
\setlength{\unitlength}{0.48pt}
\begin{picture}(180,90)(0,0)
\put(2,80){\makebox(0,0){\footnotesize{$I'$}}}
\put(160,80){\makebox(0,0){\footnotesize{$I$}}}
\put(2,0){\makebox(0,0){\footnotesize{$X'$}}}
\put(160,0){\makebox(0,0){\footnotesize{$X$}}}
\put(8,40){\makebox(0,0)[l]{\scriptsize{$s'$}}}
\put(168,40){\makebox(0,0)[l]{\scriptsize{$s$}}}
\put(80,94){\makebox(0,0){\scriptsize{$h$}}}
\put(80,14){\makebox(0,0){\scriptsize{$f$}}}
\put(20,80){\vector(1,0){120}}
\put(20,0){\vector(1,0){120}}
\put(0,70){\vector(0,-1){60}}
\put(160,70){\vector(0,-1){60}}
\end{picture}
\end{tabular}}}
\end{center}
%}
\caption{Signature Morphism}
\label{fig:sign:mor}
\end{figure}
%
%%%%%%%%%%%%%%%%%%%%%%%%%%%%%%%%%%%%%%%%%%%%%%%%%%%%%%%%%%%%
%%%%%%%%%%%%%%%%%%%%%%%%%%%%%%%%%%%%%%%%%%%%%%%%%%%%%%%%%%%%
\comment{Hence, when $I' = \{i'\} = 1$,
a signature ${\langle{1,s',X'}\rangle}$
is essentially a sort $s_{i'}=x' \in X'$,
an arity function $I'\xrightarrow{h}I$
is essentially an index $h(i') = i \in I$,
and a signature morphism
%\[\mbox
{\footnotesize{$
%\mathcal{S}'=
{\langle{1,s',X'}\rangle}
\xrightarrow{{\langle{h,f}\rangle}}
{\langle{I,x,X}\rangle}=\mathcal{S}
$}\normalsize}
%\]
%
states that $f(x') = s(h(i')) = s(i)$.}
%%%%%%%%%%%%%%%%%%%%%%%%%%%%%%%%%%%%%%%%%%%%%%%%%%%%%%%%%%%%
%%%%%%%%%%%%%%%%%%%%%%%%%%%%%%%%%%%%%%%%%%%%%%%%%%%%%%%%%%%%
%
%
%Signature morphisms are illustrated in Fig.~\ref{fig:sign:mor}.
%%%%%%%%%%%%%%%%%%%%%%%%%%%%%%%%%%%%%%%%%%%%%%%%%%%%%%%%%%%%

%\comment{
%%%%%%%%%%%%%%%%%%%%%%%%%%%%%%%%%%%%%%%%%%%%%%%%%%
%\newpage\paragraph{Fibers.}
%%%%%%%%%%%%%%%%%%%%%%%%%%%%%%%%%%%%%%%%%%%%%%%%%%

Let $X$ be a fixed sort set.
An $X$-signature
$\mathcal{S} = \mathcal{S}$
is a signature with the sort set $X$.
%consists of an indexing set (arity) $I$ and a map %$I\xrightarrow{\,s\,}X$ from $I$ to the set of sorts $X$.
%
An $X$-signature morphism 
$\mathcal{S}' = {\langle{\mathcal{S}'}\rangle} \xrightarrow{h} \mathcal{S} = \mathcal{S}$
is a signature morphism
with an identity sort function
$X\xrightarrow{\;1_{X}\;}X$
%is an arity function $I'\xrightarrow{h}I$
%that preserves signatures by
%
and satisfying the naturality condition $h{\,\cdot\,}s = s'$.
Let $\mathrmbf{List}(X)$ denote 
the fiber context of $X$-signatures and $X$-signature morphisms. 
%
%
%In particular,
%for any index $i \in I$,
%the $i$-th projection 
%${\langle{1,s_{i}}\rangle} \xrightarrow{i} \mathcal{S} = %\mathcal{S}$
%is an arity function $1\xrightarrow{i}I$
%satisfying the naturality condition $i{\,\cdot\,}s = s_{i}$.
%
%%%%%%%%%%%%%%%%%%%%%%%%%%%%%%%%%%%%%%%%%%%%%%%%%%
%\newpage
%\paragraph{Fibered Context.}
%%%%%%%%%%%%%%%%%%%%%%%%%%%%%%%%%%%%%%%%%%%%%%%%%%%
%
For fixed arity function $I'\xrightarrow{h}I$,
the naturality condition 
$h{\,\cdot\,}s = s'{\,\cdot\,}f$
gives two alternate and adjoint fiber passages:
%\newline\mbox{}\hfill
%\rule[8pt]{0pt}{10pt}
$\mathrmbf{List}(X')
{\;\xrightarrow[{\langle{{\scriptscriptstyle\sum}_{f}{\;\dashv\;}f^{\ast}}\rangle}]{\mathrmbfit{list}(f)}\;}
\mathrmbf{List}(X)$.
%\hfill\mbox{}\newline
%
In terms of fibers,
a signature morphism
consists of
a sort function $X'\xrightarrow{f}X$ 
\underline{and} 
either a morphism 
$\mathcal{S}'\xrightarrow{\;\hat{h}\;}{f}^{\ast}(\mathcal{S})$
in the fiber context $\mathrmbf{List}(X')$ 
or adjointly a morphism 
${\scriptstyle\sum}_{f}(\mathcal{S}')\xrightarrow{\;h\;}\mathcal{S}$
in the fiber context $\mathrmbf{List}(X)$.
%
%%%%%%%%%%%%%%%%%%%%%%%%%%%%%%%%%%%%%%%%%%%%%%%%%%%%%%%%%%%%
%%%%%%%%%%%%%%%%%%%%%%%%%%%%%%%%%%%%%%%%%%%%%%%%%%%%%%%%%%%%
\footnote{${f}^{\ast}(\mathcal{S})$ 
is the pullback of signature $\mathcal{S}$
back along sort function $X'\xrightarrow{f}X$ 
and
${\scriptstyle\sum}_{f}(\mathcal{S}')$ 
is the composition of signature $\mathcal{S}'$
forward along sort function $X'\xrightarrow{f}X$.}
%%%%%%%%%%%%%%%%%%%%%%%%%%%%%%%%%%%%%%%%%%%%%%%%%%%%%%%%%%%%
%%%%%%%%%%%%%%%%%%%%%%%%%%%%%%%%%%%%%%%%%%%%%%%%%%%%%%%%%%%%
%
%%%%%%%%%%%%%%%%%%%%%%%%%%%%%%%%%%%%%%%%%%%%%%%%%%%%%%%%%%%%
%%%%%%%%%%%%%%%%%%%%%%%%%%%%%%%%%%%%%%%%%%%%%%%%%%%%%%%%%%%%
\footnote{For more on this see \S\;2.1 of
\cite{kent:fole:era:tbl}.}
\paragraph{Type Domains.}
%\subsubsection{Type Domains}
%\label{sub:sec:cls}
%%%%%%%%%%%%%%%%%%%%%%%%%%%%%%%%%%%%%%%%%%%%%%%%%%%%%%%%%%%%%%
%%%%%%%%%%%%%%%%%%%%%%%%%%%%%%%%%%%%%%%%%%%%%%%%%%%%%%%%%%%%%%

%%%%%%%%%%%%%%%%%%%%%%%%%%%%%%%%%%%%%%%%%%%%%%%%%%
%\newpage
%\paragraph{Type Domains.}
%%%%%%%%%%%%%%%%%%%%%%%%%%%%%%%%%%%%%%%%%%%%%%%%%

%A classification 
%$\mathcal{A} = {\langle{X,Y,\models_{\mathcal{A}}}\rangle}$
%consists of 
%a set of types $X$,
%a set of instances (tokens) $Y$,
%and a binary (classification) relation
%{\footnotesize{$\models_{\mathcal{A}}$}}
%between instances and types.
%
%The extent of any type $x \in X$
%is the subset
%$\mathrmbfit{ext}_{\mathcal{A}}(x) = A_{x} 
%= \{ y \in Y \mid y{\;\models_{\mathcal{A}}\;}x \}$.
%
%Hence,
%a classification is equivalent to be a type-indexed collection of %subsets of instances%
%$X \xrightarrow{\;\mathcal{A}\;} {\wp}Y: x \mapsto \mathrmbfit{ext}_{\mathcal{A}}(x) = A_{x}$.
%
\comment{
A simple example is a subset relationship:
given a set of instances $Y$,
a subset $Y'{\,\subseteq\,}Y$
determines a classification 
$Y_{Y'} = {\langle{1,Y,\models}\rangle}$
with a singleton typeset $1=\{\cdot\}$,
the given set of instances $Y$,
and the binary (classification) relation
{\footnotesize{$y\models{\cdot}$}}
iff $y{\,\in\,}Y'$.
}

In the \texttt{FOLE} theory of data-types \cite{kent:fole:era:found},
a classification 
$\mathcal{A} = {\langle{X,Y,\models_{\mathcal{A}}}\rangle}$
\cite{barwise:seligman:97} 
is known as a type domain.
%
%%%%%%%%%%%%%%%%%%%%%%%%%%%%%%%%%%%%%%%%%%%%%%%%%%
%%%%%%%%%%%%%%%%%%%%%%%%%%%%%%%%%%%%%%%%%%%%%%%%%%
\footnote{In \cite{ganter:wille:99} a classification is known as a formal context.}
%%%%%%%%%%%%%%%%%%%%%%%%%%%%%%%%%%%%%%%%%%%%%%%%%%
%%%%%%%%%%%%%%%%%%%%%%%%%%%%%%%%%%%%%%%%%%%%%%%%%%
%
A type domain
\cite{kent:fole:era:tbl}
%is an sort-indexed collection of data types 
%from which a table's tuples are chosen.
%It 
consists of 
a set of sorts (data types) $X$,
a set of data values (instances) $Y$,
and a binary (classification) relation
{\footnotesize{$\models_{\mathcal{A}}$}}
between data values and sorts.
The extent
of any sort (data type) $x \in X$
is the subset
$\mathrmbfit{ext}_{\mathcal{A}}(x) = A_{x} 
= \{ y \in Y \mid y{\;\models_{\mathcal{A}}\;}x \}$.
Hence,
a type domain is equivalent to be a sort-indexed collection of subsets of data values
$\mathcal{A} = \{ A_{x} \subseteq Y \mid x \in X \}$;
or more abstractly,
$X \xrightarrow{\;\mathcal{A}\;} {\wp}Y: 
x \mapsto \mathrmbfit{ext}_{\mathcal{A}}(x) = A_{x}$.
By being so explicit,
we have more exact control over the data.
\begin{flushleft}
{\fbox{\fbox{\footnotesize{\begin{minipage}{340pt}
Some examples of data-types 
useful in databases are as follows.
The real numbers might use 
sort symbol $\Re$ 
with extent $\{-\infty, \cdots, 0, \cdots, \infty\}$.
The alphabet might use
sort symbol $\aleph$ 
with extent $\{a,b,c, \cdots, x,y,z\}$.
Words, as a data-type, would be lists of alphabetic symbols
with
sort symbol $\aleph^{\ast}$ 
and extent $\{a,b,c, \cdots, x,y,z\}^{\ast}$
being all strings of alphabetic symbols.
The periodic table of elements might use
sort symbol $\textbf{E}$
with extent $\{\text{H, He, Li, He,} \cdots\text{, Hs,Mt}\}$.
Of course,
chemical elements can also be regarded as 
entities 
\cite{kent:fole:era:found}
%Kent, R.E.:
%``The {\ttfamily ERA} of {\ttfamily FOLE}: Foundation.''
in a database 
with various attributes
such as
name, 
symbol, 
atomic number, 
atomic mass, 
density, 
melting point, 
boiling point, etc.
\end{minipage}}}}}
\end{flushleft}
%
%%%%%%%%%%%%%%%%%%%%%%%%%%%%%%%%%%%%%%%%%%%%%%%%%%%%%%%%%%%%%%%%%%%%%%
%\paragraph{Lists.}
%%%%%%%%%%%%%%%%%%%%%%%%%%%%%%%%%%%%%%%%%%%%%%%%%%%%%%%%%%%%%%%%%%%%%%
%
For a given type domain $\mathcal{A}$, 
the list classification 
$\mathrmbf{List}(\mathcal{A}) = {\langle{\mathrmbf{List}(X),\mathrmbf{List}(Y),\models_{\mathrmbf{List}(\mathcal{A})}}\rangle}$
has $X$-signatures as types and
$Y$-tuples as instances,
with classification by common arity and universal $\mathcal{A}$-classification:
a $Y$-tuple ${\langle{J,t}\rangle}$ 
is classified by 
an $X$-signature $\mathcal{S}$ 
when
$J = I$ and
$t_{k} \models_{\mathcal{A}} s_{k}$
for all $k \in J = I$.
Hence,
a list type domain is equivalent to be a header-indexed collection of subsets of tuples
$\mathrmbf{List}(\mathcal{A}) 
= \{ \mathrmbf{List}(\mathcal{A})_{\mathcal{S}} \subseteq \mathrmbf{List}(Y) 
\mid \mathcal{S} \in \mathrmbf{List}(X) \}$;
or more abstractly,
$\mathrmbf{List}(X) \xrightarrow{\;\mathrmbf{List}(A)\;} {\wp}\mathrmbf{List}(Y): 
\mathcal{S} \mapsto \mathrmbfit{ext}_{\mathrmbf{List}(A)}(\mathcal{S}) 
= \mathrmbf{List}(\mathcal{A})_{\mathcal{S}}$.
%
%%%%%%%%%%%%%%%%%%%%%%%%%%%%%%%%%%%%%%%%%%%%%%%%%%%%%%%%%%%%
%%%%%%%%%%%%%%%%%%%%%%%%%%%%%%%%%%%%%%%%%%%%%%%%%%%%%%%%%%%%
\footnote{In particular,
when $I = 1$ is a singleton,
an $X$-signature ${\langle{1,s}\rangle}$ 
is the same as a sort  
$s({\cdot)} = x \in X$,
a $Y$-tuple ${\langle{1,t}\rangle}$ 
is the same as a data value  
$t({\cdot}) = y \in Y$,
and
$
%\mathrmbfit{tup}(1,s,\mathcal{A}) =\mathrmbfit{tup}_{\mathcal{A}}(1,s)=
\mathrmbfit{ext}_{\mathrmbf{List}(\mathcal{A})}(1,s) 
= \mathrmbf{List}(\mathcal{A})_{1,s}
=\mathcal{A}_{x}$.}
%%%%%%%%%%%%%%%%%%%%%%%%%%%%%%%%%%%%%%%%%%%%%%%%%%%%%%%%%%%%
%%%%%%%%%%%%%%%%%%%%%%%%%%%%%%%%%%%%%%%%%%%%%%%%%%%%%%%%%%%%
%%%%%%%%%%%%%%%%%%%%%%%%%%%%%%%%%%%%%%%%%%%%%%%%%%
%\newpage
%\paragraph{Type Domain Morphisms.}
%%%%%%%%%%%%%%%%%%%%%%%%%%%%%%%%%%%%%%%%%%%%%%%%%%%

%\comment{
\begin{figure}
\begin{center}
{{\begin{tabular}{c
@{\hspace{33pt}{\textit{or}}\hspace{20pt}}
c}
%%%%%%%%%%%%%%%%%%%%%%%%%%%%%%%%%%%%%%%%%%%%%%%%%%%%%%%%%%%%%%%%%%%%%%%%%%%%%%%%
%%%%%%%%%%%%%%%%%%%%%%%%%%%%%%%%%%%%%%%%%%%%%%%%%%%%%%%%%%%%%%%%%%%%%%%%%%%%%%%%
{{\begin{tabular}{c}
\setlength{\unitlength}{0.6pt}
\begin{picture}(120,90)(0,0)
\put(2,80){\makebox(0,0){\footnotesize{$X'$}}}
\put(120,80){\makebox(0,0){\footnotesize{$X$}}}
\put(2,0){\makebox(0,0){\footnotesize{$Y'$}}}
\put(120,0){\makebox(0,0){\footnotesize{$Y$}}}
\put(8,40){\makebox(0,0)[l]{\scriptsize{$\models_{\mathcal{A}'}$}}}
\put(130,40){\makebox(0,0)[l]{\scriptsize{$\models_{\mathcal{A}}$}}}
\put(60,94){\makebox(0,0){\scriptsize{$f$}}}
\put(62,-14){\makebox(0,0){\scriptsize{$g$}}}
\put(20,80){\vector(1,0){80}}
\put(100,0){\vector(-1,0){80}}
\put(0,65){\line(0,-1){50}}
\put(120,65){\line(0,-1){50}}
\end{picture}
\end{tabular}}}
%%%%%%%%%%%%%%%%%%%%%%%%%%%%%%%%%%%%%%%%%%%%%%%%%%%%%%%%%%%%%%%%%%%%%%%%%%%%%%%%
&
%%%%%%%%%%%%%%%%%%%%%%%%%%%%%%%%%%%%%%%%%%%%%%%%%%%%%%%%%%%%%%%%%%%%%%%%%%%%%%%%
{{\begin{tabular}{c}
\setlength{\unitlength}{0.6pt}
\begin{picture}(120,90)(0,0)
\put(0,80){\makebox(0,0){\footnotesize{$X'$}}}
\put(120,80){\makebox(0,0){\footnotesize{$X$}}}
\put(0,0){\makebox(0,0){\footnotesize{${\wp}Y'$}}}
\put(120,0){\makebox(0,0){\footnotesize{${\wp}Y$}}}
\put(8,40){\makebox(0,0)[l]{\scriptsize{$
\mathrmbfit{ext}_{\mathcal{A}'}$}}}
\put(128,40){\makebox(0,0)[l]{\scriptsize{$
\mathrmbfit{ext}_{\mathcal{A}}$}}}
\put(60,94){\makebox(0,0){\scriptsize{$f$}}}
\put(62,-14){\makebox(0,0){\scriptsize{$g^{-1}$}}}
\put(20,80){\vector(1,0){80}}
\put(20,0){\vector(1,0){80}}
\put(0,65){\vector(0,-1){50}}
\put(120,65){\vector(0,-1){50}}
\end{picture}
\end{tabular}}}
%%%%%%%%%%%%%%%%%%%%%%%%%%%%%%%%%%%%%%%%%%%%%%%%%%%%%%%%%%%%%%%%%%%%%%%%%%%%%%%%
%%%%%%%%%%%%%%%%%%%%%%%%%%%%%%%%%%%%%%%%%%%%%%%%%%%%%%%%%%%%%%%%%%%%%%%%%%%%%%%%
\end{tabular}}}
\end{center}
\caption{\texttt{FOLE} Type Domain Morphism}
\label{fig:fole:typ:dom:mor}
\end{figure}
%}

%
\comment{
An infomorphism
$\mathcal{A}'={\langle{X',Y',\models_{\mathcal{A}'}}\rangle}
\xrightleftharpoons{{\langle{f,g}\rangle}}
{\langle{X,Y,\models_{\mathcal{A}}}\rangle}=\mathcal{A}$
consists of
a type function $X'\xrightarrow{\;f\;}X$
and
an instance function $Y'\xleftarrow{\;g\;}Y$
that satisfy the condition
%\[\mbox
{\footnotesize{
$g(y){\;\models_{\mathcal{A}'}\;}x'$
\underline{iff}
$y{\;\models_{\mathcal{A}}\;}f(x')$
}\normalsize}
%\]
%
for source type $x'{\,\in\,}X'$ and target instance $y{\,\in\,}Y$.
%\newline
}
In the \texttt{FOLE} theory of data-types \cite{kent:fole:era:found},
an infomorphism  
$\mathcal{A}'\xrightleftharpoons{{\langle{f,g}\rangle}}\mathcal{A}$
\cite{barwise:seligman:97} 
is known as a type domain morphism,
and consists of
a sort function $X'\xrightarrow{\;f\;}X$ and
a data value function $Y'\xleftarrow{\;g\;}Y$
that satisfy any of the following equivalent conditions
%\[\mbox
%{\footnotesize{
%$g(y){\;\models_{\mathcal{A}'}\;}x'$
%\underline{iff}
%$y{\;\models_{\mathcal{A}}\;}f(x')$
%}\normalsize}
%\]
%
for any source sort $x'{\,\in\,}X'$ 
and target data value $y{\,\in\,}Y$:
\begin{equation}\label{typ:dom:mor:equiv:conds}
%\begin{center} 
%{{\setlength{\extrarowheight}{2pt}{\footnotesize{$
\begin{array}[c]{
@{\hspace{5pt}}
r@{\hspace{10pt}}r@{\hspace{7pt}}c@{\hspace{7pt}}l@{\hspace{5pt}}}
%\text{or}
&
g(y){\;\models_{\mathcal{A}'}\;}x'
&
\text{\underline{iff}}
&
y{\;\models_{\mathcal{A}}\;}f(x');
\\
%\text{or}
&
g(y){\;\in\;}\mathcal{A}'_{x'}=\mathrmbfit{ext}_{\mathcal{A}'}(x')
&
\text{\underline{iff}}
&
y{\;\in\;}\mathcal{A}_{f(x')}=
\mathrmbfit{ext}_{\mathcal{A}}(f(x'));
\\
%\text{or}
&
g^{-1}(\mathcal{A}'_{x'})
=g^{-1}(\mathrmbfit{ext}_{\mathcal{A}'}(x'))
&
=
&
\mathrmbfit{ext}_{\mathcal{A}}(f(x'))=\mathcal{A}_{f(x')};
\\
%\text{or}
&
X'\xrightarrow{\;\mathcal{A}'\;}{\wp}Y'\xrightarrow{\;g^{-1}\;}{\wp}Y 
&
=
&
X'\xrightarrow{\;f\;}X\xrightarrow{\;\mathcal{A}\;}{\wp}Y .
\end{array}
%$}}}}
%\end{center}
\end{equation}
The condition on a type domain morphism 
gives two alternate definitions.
%
%%%%%%%%%%%%%%%%%%%%%%%%%%%%%%%%%%%%%%%%%%%%%%%%%%%%%%%%%%%%%%%%%%%%%%
%%%%%%%%%%%%%%%%%%%%%%%%%%%%%%%%%%%%%%%%%%%%%%%%%%%%%%%%%%%%%%%%%%%%%%
\footnote{Yin-Yang is the interplay of two opposite 
but complementary forces that interact to form a dynamic whole.}
%%%%%%%%%%%%%%%%%%%%%%%%%%%%%%%%%%%%%%%%%%%%%%%%%%%%%%%%%%%%%%%%%%%%%%
%%%%%%%%%%%%%%%%%%%%%%%%%%%%%%%%%%%%%%%%%%%%%%%%%%%%%%%%%%%%%%%%%%%%%%
%
\begin{itemize}
\item[\textit{Yin:}]
(RHS Fig.\;\ref{fig:typ:dom:yin:yang})
In terms of fibers,
a type domain morphism
consists of
a data value function $Y'\xleftarrow{\;g\;}Y$, and
an infomorphism 
${g}^{-1}(\mathcal{A}')\xrightleftharpoons{{\langle{f,\mathrmit{1}_{Y}}\rangle}}\mathcal{A}$
from
the (Yin) classification 
${g}^{-1}(\mathcal{A}') = {\langle{X',Y,\models_{{g}^{-1}(\mathcal{A}')}}\rangle}$
defined by
$y{\;\models_{{g}^{-1}(\mathcal{A}')}\;}x'$
\underline{when}
$g(y){\;\models_{\mathcal{A}'}\;}x'$.
\newline
\item[\textit{Yang:}]
(LHS Fig.\;\ref{fig:typ:dom:yin:yang})
In terms of fibers,
a type domain morphism
consists of
a sort function $X'\xrightarrow{f}X$, and
an infomorphism 
$\mathcal{A}'\xrightleftharpoons{{\langle{\mathrmit{1}_{X},g}\rangle}}{f}^{-1}(\mathcal{A})$
to 
the (Yang) classification 
${f}^{-1}(\mathcal{A}) = {\langle{X',Y,\models_{{f}^{-1}(\mathcal{A})}}\rangle}$
defined by
$y{\;\models_{{f}^{-1}(\mathcal{A})}\;}x'$
\underline{when}
$y{\;\models_{\mathcal{A}}\;}f(x')$.
\end{itemize}
An infomorphism 
$\mathcal{A}'\xrightleftharpoons{{\langle{f,g}\rangle}}\mathcal{A}$
is identical to the composite infomorphism
\begin{center}
{\footnotesize{$\mathcal{A}'\xrightleftharpoons{{\langle{\mathrmit{1}_{X},g}\rangle}}{f}^{-1}(\mathcal{A})
={g}^{-1}(\mathcal{A}')\xrightleftharpoons{{\langle{f,\mathrmit{1}_{Y}}\rangle}}\mathcal{A}$\;.}}
\end{center}
When the sort function is an injection $X'\xhookrightarrow{f}X$,
the target set $X$ contains the source set $X'$;
hence,
we can think of the  set of sorts as becoming larger as we move to the right.
Similarly,
when the data value function is an injection $Y'\xhookleftarrow{g}Y$,
the target set $Y$ is contained in the source set $Y'$;
hence,
we can think of the set of data values as becoming smaller as we move to the right.

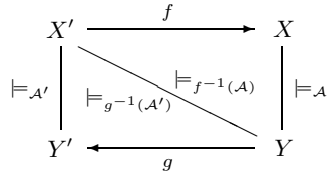
\begin{figure}
\begin{center}
{{\begin{tabular}{c}
\setlength{\unitlength}{0.45pt}
\begin{picture}(200,120)(0,-5)
\put(0,100){\makebox(0,0){\footnotesize{$X'$}}}
\put(188,100){\makebox(0,0){\footnotesize{$X$}}}
\put(0,0){\makebox(0,0){\footnotesize{$Y'$}}}
\put(188,0){\makebox(0,0){\footnotesize{$Y$}}}
\put(-8,50){\makebox(0,0)[r]{\scriptsize{$\models_{\mathcal{A}'}$}}}
\put(196,50){\makebox(0,0)[l]{\scriptsize{$\models_{\mathcal{A}}$}}}
\put(95,55){\makebox(0,0)[l]{\scriptsize{$\models_{f^{-1}(\mathcal{A})}$}}}
\put(95,40){\makebox(0,0)[r]{\scriptsize{$\models_{g^{-1}(\mathcal{A}')}$}}}
\put(90,114){\makebox(0,0){\scriptsize{$f$}}}
\put(90,-14){\makebox(0,0){\scriptsize{$g$}}}
\put(23,100){\vector(1,0){140}}
\put(163,0){\vector(-1,0){140}}
\put(0,85){\line(0,-1){70}}
\put(186,85){\line(0,-1){70}}
\put(15,85){\line(2,-1){150}}
\end{picture}
\end{tabular}}}
\end{center}
\caption{Type Domain Yin-Yang}
\label{fig:typ:dom:yin:yang}
\end{figure}
%
%}

%\newpage

\comment{\begin{flushleft}
{\footnotesize{\begin{minipage}{345pt}
A partial order 
$\mathcal{N}={\langle{N,\leq}\rangle}$
can be represent as a type domain (classification)
$\mathrmbfit{ord}(\mathcal{N})=
{\langle{N,N,\models_{\mathcal{N}}}\rangle}$,
where
$n_{0}{\,\models_{\mathcal{N}}\,}n$
when
$n_{0}{\,\leq\,}n$.
A subset $M \subseteq N$ will define two classifications
--- one backward and one forward.
%
%%%%%%%%%%%%%%%%%%%%%%%%%%%%%%%%%%%%%%%%%%%%%%%%%%%%%%%%%%%%%%%%%%%%%%
%%%%%%%%%%%%%%%%%%%%%%%%%%%%%%%%%%%%%%%%%%%%%%%%%%%%%%%%%%%%%%%%%%%%%%
\footnote{
A special case is a set regarded as the identity partial order 
${\langle{N,=}\rangle}$.}
%%%%%%%%%%%%%%%%%%%%%%%%%%%%%%%%%%%%%%%%%%%%%%%%%%%%%%%%%%%%%%%%%%%%%%
%%%%%%%%%%%%%%%%%%%%%%%%%%%%%%%%%%%%%%%%%%%%%%%%%%%%%%%%%%%%%%%%%%%%%%
%
\begin{itemize}
\item \textsf{Yin:}
%(RHS Fig.\;\ref{fig:typ:dom:yin:yang})
The inclusion data value function $M\xleftarrow{\;\iota\;}N$
defines
a classification ${\iota}^{-1}(\mathcal{N}) = {\langle{N,M,\models_{{\iota}^{-1}(\mathcal{N})}}\rangle}$
defined by
$m{\;\models_{{\iota}^{-1}(\mathcal{N})}\;}n$
\underline{when}
$\iota(m)=m{\;\models_{\mathcal{N}}\;}n$.
There is
an infomorphism ${\iota}^{-1}(\mathcal{N})\xrightleftharpoons{{\langle{1_{N},\iota}\rangle}}\mathcal{N}$.
\item ???
\end{itemize}
\end{minipage}}}
\end{flushleft}}

%%%%%%%%%%%%%%%%%%%%%%%%%%%%%%%%%%%%%%%%%%%%%%%%%%%%%%%%%%%%%%
%%%%%%%%%%%%%%%%%%%%%%%%%%%%%%%%%%%%%%%%%%%%%%%%%%%%%%%%%%%%%%%
%\newpage
\paragraph{Signed Domains.}
%\subsubsection{Signed Domains}
%\label{sub:sec:sign:dom}
%%%%%%%%%%%%%%%%%%%%%%%%%%%%%%%%%%%%%%%%%%%%%%%%%%%%%%%%%%%%%%
%%%%%%%%%%%%%%%%%%%%%%%%%%%%%%%%%%%%%%%%%%%%%%%%%%%%%%%%%%%%%%

%%%%%%%%%%%%%%%%%%%%%%%%%%%%%%%%%%%%%%%%%%%%%%%%%%
%\newpage
%\paragraph{Signed Domains.}
%%%%%%%%%%%%%%%%%%%%%%%%%%%%%%%%%%%%%%%%%%%%%%%%%

A signed domain 
\cite{kent:fole:era:tbl}
is a fundamental component used 
in the definition of database tables and 
in the database interpretation of \texttt{FOLE}.
%\S\ref{sub:sub:sec:data:model:interp}.
Signed domains are used to denote 
the valid tuples for a database header (signature).

A signed (headed/typed) domain
$\mathcal{D}={\langle{\mathcal{S},\mathcal{A}}\rangle}$
consists of
a type domain $\mathcal{A}={\langle{X,Y,\models_{\mathcal{A}}}\rangle}$ 
and a signature (database header) $\mathcal{S}={\langle{I,x,X}\rangle}$
with a common sort set $X$.
A signed domain morphism
%\[\mbox
{\footnotesize{$
\mathcal{D}'={\langle{\mathcal{S}',\mathcal{A}'}\rangle}
\xrightarrow{{\langle{h,f,g}\rangle}}
{\langle{\mathcal{S},\mathcal{A}}\rangle}=\mathcal{D}
$}\normalsize}
%\]
consists of 
a signature morphism 
$\mathcal{S}'={\langle{I',x',X'}\rangle}\xrightarrow{{\langle{h,f}\rangle}}{\langle{I,x,X}\rangle}=\mathcal{S}$ and 
a type domain morphism 
$\mathcal{A}'={\langle{X',Y',\models_{\mathcal{A}'}}\rangle}
\xrightleftharpoons{{\langle{f,g}\rangle}}
{\langle{X,Y,\models_{\mathcal{A}}}\rangle}=\mathcal{A}$
with a common sort function $X'\xrightarrow{f}X$.
Let $\mathrmbf{Dom}$ denote the mathematical context of signed domains. 

\begin{figure}
\begin{center}
{{\begin{tabular}{c
@{\hspace{25pt}{\textit{equivalently}}\hspace{25pt}}
c}
%%%%%%%%%%%%%%%%%%%%%%%%%%%%%%%%%%%%%%%%%%%%%%%%%%%%%%%%%%%%%%%%%%%%%%%%%%%%%%%%
%{{\begin{tabular}{c}
%\setlength{\unitlength}{0.45pt}
%\begin{picture}(200,190)(0,0)
{{\begin{tabular}{c}
\setlength{\unitlength}{0.45pt}
\begin{picture}(180,190)(0,0)
\put(0,180){\makebox(0,0){\footnotesize{$I'$}}}
\put(0,90){\makebox(0,0){\footnotesize{$X'$}}}
\put(0,0){\makebox(0,0){\footnotesize{$Y'$}}}
\put(180,180){\makebox(0,0){\footnotesize{$I$}}}
\put(180,90){\makebox(0,0){\footnotesize{$X$}}}
\put(180,0){\makebox(0,0){\footnotesize{$Y$}}}
\put(8,140){\makebox(0,0)[l]{\scriptsize{$s'$}}}
%\put(88,145){\makebox(0,0)[l]{\scriptsize{$s'{\cdot}f$}}}
\put(188,140){\makebox(0,0)[l]{\scriptsize{$s$}}}
\put(8,45){\makebox(0,0)[l]{\scriptsize{$
\models_{\mathcal{A}'}$}}}
%\put(85,55){\makebox(0,0)[l]{\scriptsize{$\models_{g^{-1}(\mathcal{A}')}$}}}
\put(188,45){\makebox(0,0)[l]{\scriptsize{$
\models_{\mathcal{A}}$}}}
\put(90,194){\makebox(0,0){\scriptsize{$h$}}}
\put(90,104){\makebox(0,0){\scriptsize{$f$}}}
\put(90,14){\makebox(0,0){\scriptsize{$g$}}}
\put(20,180){\vector(1,0){140}}
\put(20,90){\vector(1,0){140}}
\put(160,0){\vector(-1,0){140}}
\put(0,170){\vector(0,-1){60}}
\put(180,170){\vector(0,-1){60}}
\put(0,75){\line(0,-1){60}}
\put(180,75){\line(0,-1){60}}
%\put(20,170){\vector(2,-1){140}}
%\put(20,80){\line(2,-1){140}}
%
\end{picture}
\end{tabular}}}
%%%%%%%%%%%%%%%%%%%%%%%%%%%%%%%%%%%%%%%%%%%%%%%%%%%%%%%%%%%%%%%%%%%%%%%%%%%%%%%%
&
%%%%%%%%%%%%%%%%%%%%%%%%%%%%%%%%%%%%%%%%%%%%%%%%%%%%%%%%%%%%%%%%%%%%%%%%%%%%%%%%
{{\begin{tabular}{c}
\setlength{\unitlength}{0.45pt}
\begin{picture}(180,190)(0,0)
\put(2,180){\makebox(0,0){\footnotesize{$I'$}}}
\put(2,90){\makebox(0,0){\footnotesize{$X'$}}}
\put(2,0){\makebox(0,0){\footnotesize{$Y'$}}}
\put(180,180){\makebox(0,0){\footnotesize{$I$}}}
\put(180,90){\makebox(0,0){\footnotesize{$X$}}}
\put(180,0){\makebox(0,0){\footnotesize{$Y$}}}
\put(8,140){\makebox(0,0)[l]{\scriptsize{$s'$}}}
%\put(88,145){\makebox(0,0)[l]{\scriptsize{$s'{\cdot}f$}}}
\put(188,140){\makebox(0,0)[l]{\scriptsize{$s$}}}
\put(8,45){\makebox(0,0)[l]{\scriptsize{$
%\models_{\mathcal{A}'}
\mathrmbfit{ext}_{\mathcal{A}'}
$}}}
%\put(85,55){\makebox(0,0)[l]{\scriptsize{$\models_{g^{-1}(\mathcal{A}')}$}}}
\put(188,45){\makebox(0,0)[l]{\scriptsize{$
%\models_{\mathcal{A}}
\mathrmbfit{ext}_{\mathcal{A}}
$}}}
\put(90,194){\makebox(0,0){\scriptsize{$h$}}}
\put(90,104){\makebox(0,0){\scriptsize{$f$}}}
%\put(90,14){\makebox(0,0){\scriptsize{$g$}}}
\put(92,14){\makebox(0,0){\scriptsize{$g^{-1}$}}}
\put(20,180){\vector(1,0){140}}
\put(20,90){\vector(1,0){140}}
\put(20,0){\vector(1,0){140}}
\put(0,170){\vector(0,-1){60}}
\put(180,170){\vector(0,-1){60}}
%\put(0,75){\line(0,-1){60}}
\put(0,75){\vector(0,-1){60}}
%\put(180,75){\line(0,-1){60}}
\put(180,75){\vector(0,-1){60}}
%\put(20,170){\vector(2,-1){140}}
%\put(20,80){\line(2,-1){140}}
%
\end{picture}
\end{tabular}}}
%%%%%%%%%%%%%%%%%%%%%%%%%%%%%%%%%%%%%%%%%%%%%%%%%%%%%%%%%%%%%%%%%%%%%%%%%%%%%%%%
%%%%%%%%%%%%%%%%%%%%%%%%%%%%%%%%%%%%%%%%%%%%%%%%%%%%%%%%%%%%%%%%%%%%%%%%%%%%%%%%
\end{tabular}}}
\end{center}
\caption{\texttt{FOLE} Signed Domain Morphism}
\label{fig:sign:dom:mor}
\end{figure}
%

%%%%%%%%%%%%%%%%%%%%%%%%%%%%%%%%%%%%%%%%%%%%%%%%%%%%%%%%%%%%
%%%%%%%%%%%%%%%%%%%%%%%%%%%%%%%%%%%%%%%%%%%%%%%%%%%%%%%%%%%%
%\newpage
\subsection{Tables, Tabular Flow and Relations.}
\label{sub:sec:tbl:rel}
%%%%%%%%%%%%%%%%%%%%%%%%%%%%%%%%%%%%%%%%%%%%%%%%%%%%%%%%%%%%
%%%%%%%%%%%%%%%%%%%%%%%%%%%%%%%%%%%%%%%%%%%%%%%%%%%%%%%%%%%%

The \texttt{FOLE} relational operations are defined on tables.
The mathematical contexts of relations and tables 
are used for satisfaction and interpretation
\cite{kent:fole:era:found},
relations for traditional interpretation and 
tables for database interpretation.
In {\texttt{FOLE}}, 
both relational algebra and the tuple relational calculus 
are based upon the table concept.

\begin{definition}\label{def:sign:dom:tup}
There is a \underline{tuple passage}
\cite{kent:fole:era:tbl}
$\mathrmbfit{tup} : \mathrmbf{Dom}^{\mathrm{op}} \rightarrow \mathrmbf{Set}$.
\end{definition}
%
%\begin{sloppypar}\noindent
The tuple passage
$\mathrmbf{Dom}^{\mathrm{op}}\xrightarrow{\mathrmbfit{tup}}\mathrmbf{Set}$
maps a signed domain
${\langle{\mathcal{S},\mathcal{A}}\rangle}$
to its set of tuples
$\mathrmbfit{tup}(\mathcal{S},\mathcal{A})$,
%%%%%%%%%%%%%%%%%%%%%%%%%%%%%%%%%%%%%%%%%%%%%%%%%%%%%%%%%%%%
%%%%%%%%%%%%%%%%%%%%%%%%%%%%%%%%%%%%%%%%%%%%%%%%%%%%%%%%%%%%
\footnote{\label{sign:dom:tup}
This important concept
can intuitively be regarded as the set of legal tuples under the database header
$\mathcal{S}={\langle{I,x,X}\rangle}$.
It is define to be the extent in the list type domain $\mathrmbf{List}(\mathcal{A})$:
$\mathrmbfit{tup}(\mathcal{S},\mathcal{A})
=\mathrmbfit{ext}_{\mathrmbf{List}(\mathcal{A})}(\mathcal{S}) 
= \{ {\langle{J,t}\rangle} \in \mathrmbf{List}(Y) 
\mid {\langle{J,t}\rangle} \models_{\mathrmbf{List}(\mathcal{A})} \mathcal{S} \}$.
Various notations are used for this concept depending upon circumstance:
$\mathrmbfit{tup}(\mathcal{S},\mathcal{A})
= \mathrmbfit{tup}(\mathcal{S},\mathcal{A});
= \mathrmbfit{tup}_{\mathcal{A}}(\mathcal{S})
= \mathrmbfit{tup}_{\mathcal{A}}(\mathcal{S});
= \mathrmbfit{tup}_{\mathcal{S}}(\mathcal{A})
= \mathrmbfit{tup}_{\mathcal{S}}(Y,\models_{\mathcal{A}})
$.}
%%%%%%%%%%%%%%%%%%%%%%%%%%%%%%%%%%%%%%%%%%%%%%%%%%%%%%%%%%%%
%%%%%%%%%%%%%%%%%%%%%%%%%%%%%%%%%%%%%%%%%%%%%%%%%%%%%%%%%%%%
and maps a signed domain morphism
${\langle{I_{2},s_{2},\mathcal{A}_{2}}\rangle}\xrightarrow{{\langle{h,f,g}\rangle}}{\langle{I_{1},s_{1},\mathcal{A}_{1}}\rangle}$
to its tuple function
$\mathrmbfit{tup}(I_{2},s_{2},\mathcal{A}_{2})
\xleftarrow[(h{\cdot}{(\mbox{-})})\cdot({(\mbox{-})}{\cdot}g)]{\mathrmbfit{tup}(h,f,g)}
\mathrmbfit{tup}(I_{1},s_{1},\mathcal{A}_{1})$;
\newline
or visually,
$({\cdots\,}g(t_{h(i_{2})}){\,\cdots}{\,\mid\,}i_{2}{\,\in\,}I_{2})
\mapsfrom
({\cdots\,}t_{i_{1}}{\,\cdots}{\,\mid\,}i_{1}{\,\in\,}I_{1})$.
%\end{sloppypar}
%\begin{proposition}\label{prop:tup:pass:cts}
%The tuple passage 
%$\mathrmbf{List}(X)^{\mathrm{op}}\xrightarrow{\mathrmbfit{tup}_{\mathcal{A}}}\mathrmbf{Set}$
%is continuous.
%\end{proposition}
%%
%\begin{proof}
%See Prop.\;6 in \S\;4.1 of the paper \cite{kent:fole:era:tbl}.
%\mbox{}\hfill\rule{5pt}{5pt}
%\end{proof}
%

%%%%%%%%%%%%%%%%%%%%%%%%%%%%%%%%%%%%%%%%%%%%%%%%%%%%%%%%%%%%%%%%%%
%\newpage
\paragraph{Tables.}
%\subsubsection{Tables.}
%\label{sub:sub:sec:tbl}
%%%%%%%%%%%%%%%%%%%%%%%%%%%%%%%%%%%%%%%%%%%%%%%%%%%%%%%%%%%%%%%%%%

A \texttt{FOLE} table
$\mathcal{T} = {\langle{\mathcal{D},K,t}\rangle}$
%,
%visualized in Fig.~\ref{fig:fole:tbl},
consists of 
a signed domain $\mathcal{D}={\langle{\mathcal{S},\mathcal{A}}\rangle}$,
a set $K$ of (primary) keys and
a tuple function $K\xrightarrow{t}\mathrmbfit{tup}(\mathcal{D})=\mathrmbfit{tup}_{\mathcal{A}}(\mathcal{S})$
mapping keys to $\mathcal{D}$-tuples.
% too much from the table paper?
%%%%%%%%%%%%%%%%%%%%%%%%%%%%%%%%%%%%%%%%%%%%%%%%%%%%%%%%%%%%%%%%%%%%%%%%%%%%%%%%
%%%%%%%%%%%%%%%%%%%%%%%%%%%%%%%%%%%%%%%%%%%%%%%%%%%%%%%%%%%%%%%%%%%%%%%%%%%%%%%%
%
A \texttt{FOLE} table morphism 
%(morphism of database relations)
%\[\mbox
{\footnotesize{$
\mathcal{T}' = {\langle{\mathcal{D}',K',t'}\rangle}
\xleftarrow{{\langle{{\langle{h,f,g}\rangle},k}\rangle}} 
{\langle{\mathcal{D},K,t}\rangle} = \mathcal{T}
$}\normalsize}
%\hfill\mbox{}\newline\]
%
consists of 
a signed domain morphism
{\footnotesize{$
\mathcal{D}'={\langle{\mathcal{S}',\mathcal{A}'}\rangle}
\xrightarrow{{\langle{h,f,g}\rangle}}
{\langle{\mathcal{S},\mathcal{A}}\rangle}=\mathcal{D}
$}\normalsize}
and
a key function $K'\xleftarrow{k}K$,
which satisfy the naturality condition
$k{\;\cdot\;}t' = t{\;\cdot\;}\mathrmbfit{tup}(h,f,g)$.
Let $\mathrmbf{Tbl}$ denote the mathematical context of tables. 
\comment{
\begin{figure}
\begin{center}
%\begin{tabular}{c@{\hspace{90pt}}c}
%%%%%%%%%%%%%%%%%%%%%%%%%%%%%%%%%%%%%%%%%%%%%%%%%%
{{\begin{tabular}{c}
\setlength{\unitlength}{0.6pt}
\begin{picture}(120,90)(0,-5)
%\put(160,80){\framebox(0,0){}}
\put(0,80){\makebox(0,0){\footnotesize{$K'$}}}
\put(120,80){\makebox(0,0){\footnotesize{$K$}}}
\put(-27,0){\makebox(0,0){\footnotesize{$\mathrmbfit{tup}_{\mathcal{A}'}(\mathcal{S}')$}}}
\put(148,0){\makebox(0,0){\footnotesize{$\mathrmbfit{tup}_{\mathcal{A}}(\mathcal{S})$}}}
%\put(60,-40){\makebox(0,0){\footnotesize{$\mathrmbfit{tup}_{\mathcal{A}}(\mathcal{S}'{\cdot}f)$}}}
\put(63,95){\makebox(0,0){\scriptsize{$k$}}}
\put(63,-12){\makebox(0,0){\scriptsize{$\mathrmbfit{tup}(h,f,g)$}}}
\put(-6,40){\makebox(0,0)[r]{\scriptsize{$t'$}}}
\put(128,40){\makebox(0,0)[l]{\scriptsize{$t$}}}
%\put(105,-22){\makebox(0,0)[l]{\scriptsize{$\mathrmbfit{tup}_{\mathcal{A}}(h)$}}} 
%\put(15,-22){\makebox(0,0)[r]{\scriptsize{$\tau_{{\langle{f,g}\rangle}}(\mathcal{S}')$}}}
\put(0,65){\vector(0,-1){50}}
\put(120,65){\vector(0,-1){50}}
\put(100,80){\vector(-1,0){80}}
\put(100,0){\vector(-1,0){80}}
%\put(45,-30){\vector(-3,2){30}}
%\put(105,-10){\vector(-3,-2){30}}
\end{picture}
\end{tabular}}}
\end{center}
\caption{\texttt{FOLE} Table Morphism}
\label{fig:tbl:mor}
\end{figure}
}

We next show that the context of tables is a fibered context over signed domains; 
we first define the table fiber for fixed signed domain; 
%We next 
%move between table fibers along signed domain morphisms. 
%Finally, we 
%show that 
%tables in 
we then
%the large table context can be 
define 
adjoint flow along signed domain morphisms.

%%%%%%%%%%%%%%%%%%%%%%%%%%%%%%%%%%%%%%%%%%%%%%%%%%
%\newpage
\paragraph{Small Fibers.}
%%%%%%%%%%%%%%%%%%%%%%%%%%%%%%%%%%%%%%%%%%%%%%%%%%
%
%%%%%%%%%%%%%%%%%%%%%%%%%%%%%%%%%%%%%%%%%%%%%%%%%%%%%%%%%%%%
%%%%%%%%%%%%%%%%%%%%%%%%%%%%%%%%%%%%%%%%%%%%%%%%%%%%%%%%%%%%
\footnote{See the small-size fiber context of signed domain indexing in
\S\;3.2 of the paper 
``The {\ttfamily FOLE} Table'' 
\cite{kent:fole:era:tbl}.}
%%%%%%%%%%%%%%%%%%%%%%%%%%%%%%%%%%%%%%%%%%%%%%%%%%%%%%%%%%%%
%%%%%%%%%%%%%%%%%%%%%%%%%%%%%%%%%%%%%%%%%%%%%%%%%%%%%%%%%%%%
%
Let $\mathcal{D}
={\langle{\mathcal{S},\mathcal{A}}\rangle}
={\langle{\mathcal{S},\mathcal{A}}\rangle}$ 
be a fixed signed domain.
%\begin{itemize}
%\item 
The fiber mathematical context of 
of $\mathcal{D}$-tables 
is denoted by
$\mathrmbf{Tbl}(\mathcal{D}) 
= \mathrmbf{Tbl}_{\mathcal{A}}(\mathcal{S})
= \mathrmbf{Tbl}_{\mathcal{A}}(\mathcal{S})$.
A $\mathcal{D}$-table $\mathcal{T} = {\langle{K,t}\rangle}$ 
%called an ${\langle{\mathcal{S},\mathcal{A}}\rangle}$-table,
consists of a set $K$ of (primary) keys
and a tuple function
$K\xrightarrow{\,t\;}\mathrmbfit{tup}_{\mathcal{A}}(\mathcal{S})$
mapping each key to its descriptor $\mathcal{A}$-tuple of type (signature) $\mathcal{S}$.
A $\mathcal{D}$-table morphism 
$\mathcal{T}' = {\langle{K',t'}\rangle}\xleftarrow{k}{\langle{K,t}\rangle} = \mathcal{T}$
%An ${\langle{\mathcal{S},\mathcal{A}}\rangle}$-table morphism 
is a key function $K'\xleftarrow{\,k\,}K$
that preserves descriptors by
satisfying the naturality condition $k{\;\cdot\;}t' = t$.
As we show in \S\,\ref{sub:sub:sec:reflect} on reflection,
this means that
we have the relation morphism condition
${\wp{t'}}(K'){\;\supseteq\;}{\wp{t}}(K)$;
that is,
all the tuples in $\mathcal{T}$ are tuples in $\mathcal{T}'$.
Use the notation
$\mathcal{T}' \geq \mathcal{T}$
for this assertion.

%%%%%%%%%%%%%%%%%%%%%%%%%%%%%%%%%%%%%%%%%%%%%%%%%%%%%%%%%%%%%%%%%%
%%%%%%%%%%%%%%%%%%%%%%%%%%%%%%%%%%%%%%%%%%%%%%%%%%%%%%%%%%%%%%%%%%
%\newpage
%\subsection{Tabular Flow.}
%\label{sub:sec:tbl:flow}
%%%%%%%%%%%%%%%%%%%%%%%%%%%%%%%%%%%%%%%%%%%%%%%%%%%%%%%%%%%%%%%%%%
%%%%%%%%%%%%%%%%%%%%%%%%%%%%%%%%%%%%%%%%%%%%%%%%%%%%%%%%%%%%%%%%%%

%%%%%%%%%%%%%%%%%%%%%%%%%%%%%%%%%%%%%%%%%%%%%%%%%%
%\newpage
%\paragraph{Tabular Flow Adjunction.}
%%%%%%%%%%%%%%%%%%%%%%%%%%%%%%%%%%%%%%%%%%%%%%%%%%

%
%%%%%%%%%%%%%%%%%%%%%%%%%%%%%%%%%%%%%%%%%%%%%%%%%%%%%%%%%%%%%
%\newpage
\paragraph{Tabular Flow Adjunction.}
%\subsubsection{Tabular Flow Adjunction.}
%\label{sub:sub:sec:tbl:flow:adj}
%%%%%%%%%%%%%%%%%%%%%%%%%%%%%%%%%%%%%%%%%%%%%%%%%%%%%%%%%%%%%

%\newpage
%\mbox{}\newline\newline
%{\fbox{$\blacktriangledown$\hspace{30pt}
%\textbf{Work zone:signed domain indexed adjunction}
%\hspace{30pt}$\blacktriangledown$}}
%\newline\newline
As defined above,
a table morphism
%a $\mathrmbf{Tbl}$-morphism 
%\[\mbox{\footnotesize{
$\mathcal{T}' = {\langle{\mathcal{D}',K',t'}\rangle} 
\xleftarrow{\,{\langle{{\langle{h,f,g}\rangle},k}\rangle}\;} 
{\langle{\mathcal{D},K,t}\rangle} = \mathcal{T}$
%}\normalsize}\]
%called a table morphism,
consists of
a signed domain morphism
%\newline
$\mathcal{D}'\xrightarrow{{\langle{h,f,g}\rangle}}\mathcal{D}$
and 
%an ${\langle{\mathcal{S}',\mathcal{A}'}\rangle}$-table morphism 
%${\langle{K',t'}\rangle}\xleftarrow{\,k\;}\overleftarrow{\mathrmbfit{tbl}}_{{\langle{h,f,g}\rangle}}(K,t)$,
%consisting of
a key function
$K'\xleftarrow{\,k\;}K$
satisfying the naturality condition 
$k{\;\cdot\;}t' = t{\;\cdot\;}\mathrmbfit{tup}(h,f,g)$.
This condition gives two alternate and adjoint definitions.
%\end{itemize}
%
%\vspace{50pt}
\begin{figure}
\begin{center}
{{\begin{tabular}{@{\hspace{10pt}}r@{\hspace{30pt}}c@{\hspace{10pt}}l}
{{\begin{tabular}{c}
{\setlength{\unitlength}{0.45pt}\begin{picture}(240,180)(0,-55)
\put(0,80){\makebox(0,0){\footnotesize{$K'$}}}
\put(120,80){\makebox(0,0){\footnotesize{$K$}}}
\put(240,80){\makebox(0,0){\footnotesize{$K$}}}
\put(60,-25){\makebox(0,0){\footnotesize{$\mathrmbfit{tup}(\mathcal{D}')$}}}
\put(240,-25){\makebox(0,0){\footnotesize{$\mathrmbfit{tup}(\mathcal{D})$}}}
\put(30,-56){\makebox(0,0){\footnotesize{$
\underset{\textstyle{\mathcal{T}'}}{\underbrace{\rule{30pt}{0pt}}}$}}}
\put(115,-58){\makebox(0,0){\footnotesize{$
\underset{\textstyle{\acute{\mathrmbfit{tbl}}(h,f,g)
(\mathcal{T})}}{\underbrace{\rule{30pt}{0pt}}}$}}}
\put(240,-56){\makebox(0,0){\footnotesize{$\underset{\textstyle{\mathcal{T}}}{\underbrace{\rule{40pt}{0pt}}}$}}}
\put(5,26){\makebox(0,0)[l]{\scriptsize{$t'$}}}
\put(102,26){\makebox(0,0)[l]{\scriptsize{$t{\,\cdot\,}\mathrmbfit{tup}(h,f,g)$}}}
\put(248,26){\makebox(0,0)[l]{\scriptsize{$t$}}}
\put(160,-10){\makebox(0,0){\scriptsize{$\mathrmbfit{tup}(h,f,g)$}}}
\put(120,125){\makebox(0,0){\scriptsize{$k$}}}
\put(60,92){\makebox(0,0){\scriptsize{$k$}}}
\put(100,80){\vector(-1,0){80}}
\put(4,59){\vector(2,-3){46}}
\put(108,59){\vector(-2,-3){46}}
\put(190,80){\makebox(0,0){\normalsize{$=$}}}
\put(240,60){\vector(0,-1){68}}
\put(190,-25){\vector(-1,0){80}}
%\put(165,-35){\vector(-1,0){35}}
%
\put(10,100){\oval(20,20)[tl]}
\put(10,110){\line(1,0){220}}
\put(230,100){\oval(20,20)[tr]}
\put(0,94){\vector(0,-1){0}}
\put(60,-103){\makebox(0,0){\footnotesize{$
\underset{\textstyle{\mathrmbf{Tbl}(\mathcal{D}')}}{\underbrace{\rule{60pt}{0pt}}}$}}}
\end{picture}}
\end{tabular}}}
%%%%%%%%%%%%%%%%%%%%%%%%%%%%%%%%%%%%%%%%%%%%%%%%%%%%%%%%%%%%%%%%%%%%%%%%%%%%%%%%
${\;\;\;\;\;\;\;\;\;\;\cong\;\;\;\;\;\;\;\;\;\;}$
%%%%%%%%%%%%%%%%%%%%%%%%%%%%%%%%%%%%%%%%%%%%%%%%%%%%%%%%%%%%%%%%%%%%%%%%%%%%%%%%
{{\begin{tabular}{c}
{\setlength{\unitlength}{0.45pt}\begin{picture}(240,180)(-120,-55)
\put(-120,80){\makebox(0,0){\footnotesize{$K'$}}}
\put(0,80){\makebox(0,0){\footnotesize{$\widehat{K}$}}}
\put(120,80){\makebox(0,0){\footnotesize{$K$}}}
\put(60,-25){\makebox(0,0){\footnotesize{$\mathrmbfit{tup}(\mathcal{D})$}}}
\put(-120,-25){\makebox(0,0){\footnotesize{$\mathrmbfit{tup}(\mathcal{D}')$}}}
\put(-120,-56){\makebox(0,0){\footnotesize{$
\underset{\textstyle{\mathcal{T}'}}{\underbrace{\rule{40pt}{0pt}}}$}}}
\put(20,-60){\makebox(0,0){\footnotesize{$
\underset{\textstyle{\grave{\mathrmbfit{tbl}}(h,f,g)
(\mathcal{T}')}}{\underbrace{\rule{30pt}{0pt}}}$}}}
\put(100,-56){\makebox(0,0){\footnotesize{$\underset{\textstyle{\mathcal{T}}}{\underbrace{\rule{30pt}{0pt}}}$}}}
\put(-60,92){\makebox(0,0){\scriptsize{$e$}}}
\put(15,26){\makebox(0,0)[r]{\scriptsize{$\hat{t}$}}}
\put(105,26){\makebox(0,0)[l]{\scriptsize{$t$}}}
\put(-125,26){\makebox(0,0)[r]{\scriptsize{$t'$}}}
\put(-25,-10){\makebox(0,0){\scriptsize{$\mathrmbfit{tup}(h,f,g)$}}}
\put(60,92){\makebox(0,0){\scriptsize{$k'$}}}
\put(0,125){\makebox(0,0){\scriptsize{$k$}}}
\put(100,80){\vector(-1,0){80}}
\put(-25,80){\vector(-1,0){70}}
\put(4,59){\vector(2,-3){46}}
\put(116,59){\vector(-2,-3){46}}
\put(-120,60){\vector(0,-1){68}}
\put(5,-25){\vector(-1,0){80}}
%\put(-15,-30){\vector(-1,0){35}}
%
\put(-110,100){\oval(20,20)[tl]}
\put(-110,110){\line(1,0){220}}
\put(110,100){\oval(20,20)[tr]}
\put(-120,94){\vector(0,-1){0}}
\qbezier(-100,15)(-92,15)(-84,15)
\qbezier(-84,-1)(-84,7)(-84,15)
%\qbezier(-24,45)(-16,45)(-8,45)
%\qbezier(-24,45)(-24,53)(-24,62)
%
\put(60,-103){\makebox(0,0){\footnotesize{$
\underset{\textstyle{\mathrmbf{Tbl}(\mathcal{D})}}{\underbrace{\rule{60pt}{0pt}}}$}}}
\end{picture}}
\end{tabular}}}
%%%%%%%%%%%%%%%%%%%%%%%%%%%%%%%%%%%%%%%%%%%%%%%%%%%%%%%%%%%%%%%%%%%%%%%%%%%%%%%%
\\&&
\\&&
\end{tabular}}}
\end{center}
\caption{Table Morphism: Signed Domain}
\label{fig:tbl:mor:large}
\end{figure}
%
%{\fbox{\begin{tabular}{p{300pt}}
In terms of fibers,
a table morphism
consists of
a signed domain morphism
$\mathcal{D}'\xrightarrow{{\langle{h,f,g}\rangle}}\mathcal{D}$
\underline{and} 
either a morphism 
$\mathcal{T}'\xleftarrow{\;k\;}\acute{\mathrmbfit{tbl}}(h,f,g)(\mathcal{T})$
in the fiber context $\mathrmbf{Tbl}(\mathcal{D}')$ 
or a morphism 
$\grave{\mathrmbfit{tbl}}(h,f,g)(\mathcal{T}')\xleftarrow{\,k'\,}\mathcal{T}$
in the fiber context $\mathrmbf{Tbl}(\mathcal{D})$.
\comment{\begin{equation}
{{\begin{picture}(120,20)(0,-7)
\put(60,0){\makebox(0,0){\footnotesize{$
\underset{\textstyle{\text{in}\;\mathrmbf{Tbl}(\mathcal{D}')}}
{\overset{{\textstyle{\mathcal{T}'{\;\geq\;}\acute{\mathrmbfit{tbl}}(h,f,g)(\mathcal{T})}}}
{\overbrace{\mathcal{T}'\xleftarrow{\;k\;}\acute{\mathrmbfit{tbl}}(h,f,g)(\mathcal{T})}}}
{\;\;\;\;\;\;\;\;\rightleftarrows\;\;\;\;\;\;\;\;}
\underset{\textstyle{\text{in}\;\mathrmbf{Tbl}(\mathcal{D})}}
{\overset{{\textstyle{\grave{\mathrmbfit{tbl}}(h,f,g)(\mathcal{T}'){\;\geq\;}\mathcal{T}}}}
{\overbrace{\grave{\mathrmbfit{tbl}}(h,f,g)(\mathcal{T}')\xleftarrow{\,k'\,}\mathcal{T}}}}
$}}}
\end{picture}}}
\end{equation}}
The $\mathcal{D}'$-table morphism 
$\mathcal{T}'
\xleftarrow{\;k\;}\acute{\mathrmbfit{tbl}}(h,f,g)(\mathcal{T})$
is the composition (RHS of Fig.~\ref{fig:tbl:mor:large}) 
of the fiber morphism
$\acute{\mathrmbfit{tbl}}(h,f,g)\Bigl(
\grave{\mathrmbfit{tbl}}(h,f,g)(\mathcal{T}')\xleftarrow{k'}\mathcal{T}
\Bigr)$
with the $\mathcal{T}'^{\,\text{th}}$ counit component 
$\mathcal{T}'\xleftarrow{e}
\acute{\mathrmbfit{tbl}}(h,f,g)\Bigl(\grave{\mathrmbfit{tbl}}(h,f,g)(\mathcal{T}')\Bigr)$
for the fiber adjunction
\begin{equation}\label{def:sign:dom:mor:adj}
{{\begin{picture}(120,20)(0,-7)
\put(60,0){\makebox(0,0){\footnotesize{$
\mathrmbf{Tbl}(\mathcal{D}')
{\;\xleftarrow
[{\bigl\langle{{\scriptstyle\sum}_{{\langle{h,f,g}\rangle}}
{\;\dashv\;}{{\langle{h,f,g}\rangle}}^{\ast}\bigr\rangle}}]
{{\bigl\langle{\acute{\mathrmbfit{tbl}}(h,f,g)
{\;\dashv\;}
\grave{\mathrmbfit{tbl}}(h,f,g)}\bigr\rangle}}\;}
\mathrmbf{Tbl}(\mathcal{D})$.
}}}
\end{picture}}}
\end{equation}
This fiber adjunction is a component of
the signed domain indexed adjunction of tables
$\mathrmbf{Dom}^{\mathrm{op}}\!\xrightarrow{\;\mathrmbfit{tbl}\;}\mathrmbf{Adj}$.
For more on this see
%See the section on signed domain indexing in 
the paper
\cite{kent:fole:era:tbl}.

%%%%%%%%%%%%%%%%%%%%%%%%%%%%%%%%%%%%%%%%%%%%%%%%%%%%%%%%%%%%%%
%%%%%%%%%%%%%%%%%%%%%%%%%%%%%%%%%%%%%%%%%%%%%%%%%%%%%%%%%%%%%%%
%\newpage
%\subsection{Tabular Flow.}
%\label{sub:sec:tbl:flo}
%%%%%%%%%%%%%%%%%%%%%%%%%%%%%%%%%%%%%%%%%%%%%%%%%%%%%%%%%%%%%%
%%%%%%%%%%%%%%%%%%%%%%%%%%%%%%%%%%%%%%%%%%%%%%%%%%%%%%%%%%%%%%

%%%%%%%%%%%%%%%%%%%%%%%%%%%%%%%%%%%%%%%%%%%%%%%%%%%%%%%%%%%%%%%%%%
%\newpage
\paragraph{Tuple Passage Factorization.}
%\subsubsection{Tuple Passage Factorization.}
%\label{sub:sub:sec:tup:pass}
%%%%%%%%%%%%%%%%%%%%%%%%%%%%%%%%%%%%%%%%%%%%%%%%%%%%%%%%%%%%%%%%%%

Since a signed domain morphism 
%\[\mbox
{\footnotesize{$
%\mathcal{D}'=
{\langle{\mathcal{S}',\mathcal{A}'}\rangle}
\xrightarrow{{\langle{h,f,g}\rangle}}
{\langle{\mathcal{S},\mathcal{A}}\rangle}
%=\mathcal{D}
$}\normalsize}
%\]
factors into three parts
\[\mbox{\footnotesize{$
{\langle{\mathcal{S}',\mathcal{A}'}\rangle}
\xrightarrow{{\langle{1,1,g}\rangle}}
{\langle{\mathcal{S}',g^{-1}(\mathcal{A})}\rangle}
\xrightarrow{{\langle{1,f,1}\rangle}}
{\langle{{\scriptstyle\sum}_{f}(\mathcal{S}'),\mathcal{A}}\rangle}
\xrightarrow{{\langle{h,1,1}\rangle}}
{\langle{\mathcal{S},\mathcal{A}}\rangle}
$,}\normalsize}\]
%Hence,
its tuple function of Def.\;\ref{def:sign:dom:tup} factors into five parts:
\begin{center}
{{\begin{tabular}{c}
\setlength{\unitlength}{0.82pt}
\begin{picture}(340,100)(0,10)
\put(0,60){\makebox(0,0){\footnotesize{$
\underset{\textstyle{\mathrmbfit{tup}_{\mathcal{S}'}(\mathcal{A}')}}
{\mathrmbfit{tup}_{\mathcal{A}'}(\mathcal{S}')}
$}}}
\put(115,60){\makebox(0,0){\footnotesize{$
\mathrmbfit{tup}_{\mathcal{S}'}({g}^{-1}(\mathcal{A}'))$}}}
\put(225,60){\makebox(0,0){\footnotesize{$
\mathrmbfit{tup}_{\mathcal{A}}({\scriptstyle\sum}_{f}(\mathcal{S}'))$}}}
\put(340,60){\makebox(0,0){\footnotesize{$
\underset{\textstyle{\mathrmbfit{tup}_{\mathcal{S}}(\mathcal{A})}}
{\mathrmbfit{tup}_{\mathcal{A}}(\mathcal{S})}
$}}}
\put(60,67){\makebox(0,0){\scriptsize{$
\mathrmbfit{tup}_{\mathcal{S}'}(g)$}}}
\put(60,53){\makebox(0,0){\scriptsize{${(\mbox{-})}{\,\cdot\,}g$}}}
\put(110,107){\makebox(0,0){\scriptsize{$
\grave{\tau}_{{\langle{f,g}\rangle}}(\mathcal{S}')$}}}
\put(115,93){\makebox(0,0){\scriptsize{${(\mbox{-})}{\,\cdot\,}g$}}}
\put(170,67){\makebox(0,0){\scriptsize{$
\mathrmbfit{tup}(1,f,1)$}}}
\put(170,53){\makebox(0,0){\scriptsize{$1$}}}
\put(225,27){\makebox(0,0){\scriptsize{$
\tau_{{\langle{h,f}\rangle}}(\mathcal{A})$}}}
\put(225,13){\makebox(0,0){\scriptsize{$h{\,\cdot\,}{(\mbox{-})}$}}}
\put(295,67){\makebox(0,0){\scriptsize{$
\mathrmbfit{tup}_{\mathcal{A}}(h)$}}}
\put(295,53){\makebox(0,0){\scriptsize{$h{\,\cdot\,}{(\mbox{-})}$}}}
\put(70,60){\vector(-1,0){35}}
\put(180,60){\vector(-1,0){23}}
\put(308,60){\vector(-1,0){35}}
\put(0,70){\vector(0,-1){0}}
\put(30,70){\oval(60,60)[tl]}
\put(195,100){\line(-1,0){165}}
\put(195,70){\oval(60,60)[tr]}
\put(115,50){\vector(0,1){0}}
\put(145,50){\oval(60,60)[bl]}
\put(310,20){\line(-1,0){165}}
\put(310,50){\oval(60,60)[br]}
%
%\put(90,10){\makebox(0,0){\Large{$
%\overset{\textit{\scriptsize{hidden}}}
%{\textit{\scriptsize{factor}}}
%$}}}
%\dottedline{2}(110,15)(165,45)
%\dottedline{2}(110,5)(200,-60)
%
\put(110,78){\makebox(0,0){\scriptsize{${\textit{\scriptsize{Yin}}}$}}}
\put(220,42){\makebox(0,0){\scriptsize{${\textit{\scriptsize{Yang}}}$}}}
\end{picture}
\end{tabular}}}
\end{center}
%
%visualized in Fig.\;\ref{tup:func:fact}, 
both linearly and 2-dimensionally
%
%%%%%%%%%%%%%%%%%%%%%%%%%%%%%%%%%%%%%%%%%%%%%%%%%%%%%%%%%%%%
%%%%%%%%%%%%%%%%%%%%%%%%%%%%%%%%%%%%%%%%%%%%%%%%%%%%%%%%%%%%
\footnote{For 
(1) the definitions of the tuple bridges 
$\acute{\tau}_{{\langle{f,g}\rangle}}$,
$\grave{\tau}_{{\langle{f,g}\rangle}}$, and
$\tau_{{\langle{h,f}\rangle}}$
and 
(2) the signature/type domain tuple function factorization,
see the paper \cite{kent:fole:era:tbl}).}

% tuple function factorization removed for compactness

% temporary comment out 3/16/2020
%\newpage
%\mbox{}\newline
%
\begin{proposition}\label{prop:tup:func:ident}
The tuple function
{\footnotesize{$
\mathrmbfit{tup}_{\mathcal{S}'}({g}^{-1}(\mathcal{A}'))
\xleftarrow[1]{\;\mathrmbfit{tup}(1,f,1)\;}
\mathrmbfit{tup}_{\mathcal{A}}({\scriptstyle\sum}_{f}(\mathcal{S}')$}}
is a {hidden factor}.
Composition with it
gives the identities in Fig.\;\ref{tup:func:idents}.
%
%%%%%%%%%%%%%%%%%%%%%%%%%%%%%%%%%%%%%%%%%%%%%%%%%%%%%%%%%%%%
%%%%%%%%%%%%%%%%%%%%%%%%%%%%%%%%%%%%%%%%%%%%%%%%%%%%%%%%%%%%
\footnote{For tuple set notation,
see footnote\;\ref{sign:dom:tup}.}
%%%%%%%%%%%%%%%%%%%%%%%%%%%%%%%%%%%%%%%%%%%%%%%%%%%%%%%%%%%%
%%%%%%%%%%%%%%%%%%%%%%%%%%%%%%%%%%%%%%%%%%%%%%%%%%%%%%%%%%%%
%
\end{proposition}
\begin{proof}
$\mathrmbfit{tup}_{f^{-1}(\mathcal{A})}(\mathcal{S}') 
= \mathrmbfit{tup}_{g^{-1}(\mathcal{A}')}(\mathcal{S}')$
by the Yin-Yang in \S\,\ref{sub:sec:tbl:comp}.
To see that
$\mathrmbfit{tup}_{f^{-1}(\mathcal{A})}(\mathcal{S}') 
= \mathrmbfit{tup}_{\mathcal{A}}({\scriptstyle\sum}_{f}(\mathcal{S}')
= \mathrmbfit{tup}_{\mathcal{A}}(\mathcal{S}'{\,\cdot\,}f)$,
let
$(I',t) \in \mathrmbf{List}(Y)$
be any $Y$-tuple.
Then
%\newline
$(I',t) \in \mathrmbfit{tup}_{f^{-1}(\mathcal{A})}(\mathcal{S}')$
iff
%$i' \in I'$,
%\newline
$\forall_{i' \in I'}$
$t(i')\models_{f^{-1}(\mathcal{A})}s'(i')$
iff
%\newline
$\forall_{i' \in I'}$
$t(i')\models_{\mathcal{A}}f(s'(i'))$
iff
%\newline
$(I',t) \in 
\mathrmbfit{tup}_{\mathcal{A}}(\mathcal{S}'{\,\cdot\,}f)$.
Hence,
%the source and target are identical
$\mathrmbfit{tup}_{\mathcal{S}'}({g}^{-1}(\mathcal{A}'))
=\mathrmbfit{tup}_{\mathcal{A}}({\scriptstyle\sum}_{f}(\mathcal{S}')$.
The tuple function is the identity
{\footnotesize{$
\mathrmbfit{tup}_{\mathcal{S}'}({g}^{-1}(\mathcal{A}'))
\xleftarrow[(1{\cdot}{(\mbox{-})})\cdot({(\mbox{-})}{\cdot}1)]
{\;\mathrmbfit{tup}(1,f,1)\;}
\mathrmbfit{tup}_{\mathcal{A}}({\scriptstyle\sum}_{f}(\mathcal{S}')$,}}
since defined as pre-post composition with the identity function.
\mbox{}\hfill\rule{5pt}{5pt}
\end{proof}
In \S\,\ref{sub:sec:base:ops}
the basic operations of
\textit{projection/inflation} and \textit{restriction/expansion}
are explained in terms of 
the encircled tuple functions in Fig.\;\ref{tup:func:idents}.
%In the sections below,
%to define table flow along a signed domain morphism,
%pullback is restriction followed by inflation,
%whereas 
%composition is projection followed by expansion.
%
\begin{figure}
\begin{center}
{{\setlength{\extrarowheight}{4pt}{\footnotesize{$
\begin{array}[c]{@{\hspace{5pt}}
r@{\hspace{10pt}=\hspace{10pt}}l@{\hspace{5pt}}}
\text{\fbox{$\;
\underset{\mathrmbfit{tup}_{\mathcal{A}'}(\mathcal{S}')}
{\mathrmbfit{tup}_{\mathcal{S}'}(\mathcal{A}')}
\xleftarrow[{(\mbox{-})}{\,\cdot\,}g]
{\;\mathrmbfit{tup}_{\mathcal{S}'}(g)\;}
\mathrmbfit{tup}_{\mathcal{S}'}({g}^{-1}(\mathcal{A}'))
$}}
& 
\mathrmbfit{tup}_{\mathcal{A}'}(\mathcal{S}')
\xleftarrow[{(\mbox{-})}{\,\cdot\,}g]
{\grave{\tau}_{{\langle{f,g}\rangle}}(\mathcal{S}')}
{\mathrmbfit{tup}_{\mathcal{A}}({\scriptstyle\sum}_{f}(\mathcal{S}')}) 
\\
\mathrmbfit{tup}_{\mathcal{S}'}(f^{-1}(\mathcal{A}))
\xleftarrow[h{\,\cdot\,}{(\mbox{-})}]
{\tau_{{\langle{h,f}\rangle}}(\mathcal{A})}
\mathrmbfit{tup}_{\mathcal{S}}(\mathcal{A})
&
\text{\fbox{$
\mathrmbfit{tup}_{\mathcal{A}}({\scriptstyle\sum}_{f}(\mathcal{S}'))
\xleftarrow[h{\,\cdot\,}{(\mbox{-})}]
{\mathrmbfit{tup}_{\mathcal{A}}(h)}
\mathrmbfit{tup}_{\mathcal{A}}(\mathcal{S})
\;$}}
\\
\multicolumn{1}{c}{Yin}
&
\multicolumn{1}{c}{Yang}
\end{array}$}}}}
\end{center}
\caption{Tuple Function Identities}
\label{tup:func:idents}
\end{figure}
The adjoint flow factorization (Disp.\;\ref{fbr:adj:sign:dom:mor:factor})
is visualize in Fig.\ref{fig:flow:factors:graph}.
%\newpage
%
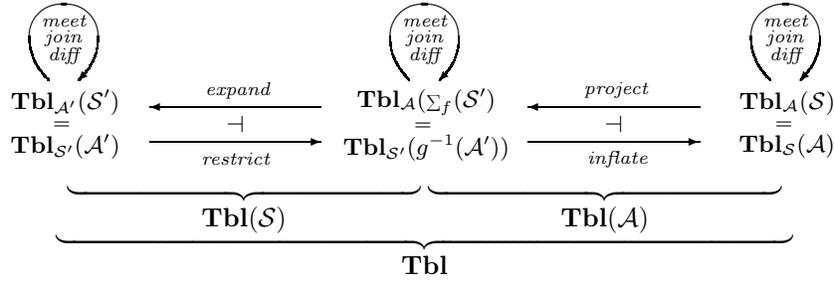
\begin{figure}
\begin{center}
{{\begin{tabular}{c}
\setlength{\unitlength}{0.65pt}
\begin{picture}(480,140)(-30,-20)
\put(1,60){\makebox(0,0){\footnotesize{$
\overset{\textstyle{\mathrmbf{Tbl}_{\mathcal{A}'}(\mathcal{S}')}}
{\underset{\textstyle{\mathrmbf{Tbl}_{\mathcal{S}'}(\mathcal{A}')}}
{=\rule[4pt]{0pt}{1pt}}}$}}}
\put(211,60){\makebox(0,0){\footnotesize{$
\overset{\textstyle{
\mathrmbf{Tbl}_{\mathcal{A}}({\scriptstyle\sum}_{f}(\mathcal{S}')}}
{\underset{\textstyle{
\mathrmbf{Tbl}_{\mathcal{S}'}({g}^{-1}(\mathcal{A}'))}}
{=\rule[4pt]{0pt}{1pt}}}$}}}
\put(421,60){\makebox(0,0){\footnotesize{$
\overset{\textstyle{\mathrmbf{Tbl}_{\mathcal{A}}(\mathcal{S})}}
{\underset{\textstyle{\mathrmbf{Tbl}_{\mathcal{S}}(\mathcal{A})}}
{=\rule[4pt]{0pt}{1pt}}}$}}}
%\put(0,60){\makebox(0,0){\footnotesize{$\bullet$}}}
%\put(210,60){\makebox(0,0){\footnotesize{$\bullet$}}}
%\put(420,60){\makebox(0,0){\footnotesize{$\bullet$}}}
%
\put(150,70){\vector(-1,0){100}}
\put(50,50){\vector(1,0){100}}
\put(370,70){\vector(-1,0){100}}
\put(270,50){\vector(1,0){100}}
\put(0,110){\oval(40,40)[t]}
\qbezier(-20,110)(-20,97)(-9,86)
\qbezier(20,110)(20,97)(9,86)
\put(8,85){\vector(-1,-1){0}}
\put(210,110){\oval(40,40)[t]}
\qbezier(190,110)(190,97)(201,86)
\qbezier(230,110)(230,97)(219,86)
\put(218,85){\vector(-1,-1){0}}
\put(420,110){\oval(40,40)[t]}
\qbezier(400,110)(400,97)(411,86)
\qbezier(440,110)(440,97)(429,86)
\put(428,85){\vector(-1,-1){0}}
\put(0,120){\makebox(0,0){\scriptsize{$\textit{meet}$}}}
\put(0,110){\makebox(0,0){\scriptsize{$\textit{join}$}}}
\put(0,100){\makebox(0,0){\scriptsize{$\textit{diff}$}}}
\put(210,120){\makebox(0,0){\scriptsize{$\textit{meet}$}}}
\put(210,110){\makebox(0,0){\scriptsize{$\textit{join}$}}}
\put(210,100){\makebox(0,0){\scriptsize{$\textit{diff}$}}}
\put(420,120){\makebox(0,0){\scriptsize{$\textit{meet}$}}}
\put(420,110){\makebox(0,0){\scriptsize{$\textit{join}$}}}
\put(420,100){\makebox(0,0){\scriptsize{$\textit{diff}$}}}
\put(120,80){\makebox(0,0)[r]{\scriptsize{$\textit{expand}$}}}
\put(120,40){\makebox(0,0)[r]{\scriptsize{$\textit{restrict}$}}}
\put(340,80){\makebox(0,0)[r]{\scriptsize{$\textit{project}$}}}
\put(340,40){\makebox(0,0)[r]{\scriptsize{$\textit{inflate}$}}}
\put(100,60){\makebox(0,0){\footnotesize{${\;\dashv\;}$}}}
\put(320,60){\makebox(0,0){\footnotesize{${\;\dashv\;}$}}}
\put(105,10){\makebox(0,0){\normalsize{${
\underset{\textstyle{\mathrmbf{Tbl}(\mathcal{S})}}{\underbrace{\hspace{135pt}}}}$}}}
\put(315,10){\makebox(0,0){\normalsize{${
\underset{\textstyle{\mathrmbf{Tbl}(\mathcal{A})}}{\underbrace{\hspace{135pt}}}}$}}}
\put(210,-16){\makebox(0,0){\normalsize{${
\underset{\textstyle{\mathrmbf{Tbl}}}{\underbrace{\hspace{280pt}}}}$}}}
\end{picture}
\end{tabular}}}
\end{center}
\caption{Adjoint Flow Factors}
\label{fig:flow:factors:graph}
\end{figure}
%

%%%%%%%%%%%%%%%%%%%%%%%%%%%%%%%%%%%%%%%%%%%%%%%%%%%%%%%%%%%%%%%%%%
%%%%%%%%%%%%%%%%%%%%%%%%%%%%%%%%%%%%%%%%%%%%%%%%%%%%%%%%%%%%%%%%%%
\newpage
\paragraph{Relations.}
%\subsubsection{Relations.}
%\label{sub:sub:sec:rel}
%%%%%%%%%%%%%%%%%%%%%%%%%%%%%%%%%%%%%%%%%%%%%%%%%%%%%%%%%%%%%%%%%%
%%%%%%%%%%%%%%%%%%%%%%%%%%%%%%%%%%%%%%%%%%%%%%%%%%%%%%%%%%%%%%%%%%

A \texttt{FOLE} relation
$\mathcal{R} = {\langle{\mathcal{D},R,i}\rangle}$
is a \texttt{FOLE} table, 
whose tuple function 
$R\xhookrightarrow{i}\mathrmbfit{tup}(\mathcal{D})=\mathrmbfit{tup}_{\mathcal{A}}(\mathcal{S})$
%mapping keys to $\mathcal{D}$-tuples.
is an inclusion.
Hence,
the set of keys 
is a subset
$R \subseteq \mathrmbfit{tup}_{\mathcal{A}}(\mathcal{S})$
of $\mathcal{D}$-tuples,
and the tuple function can be omitted 
in the designation of the relation
$\mathcal{R} = {\langle{\mathcal{D},R}\rangle}$.
%and whose tuple function is inclusion
%$R\xhookrightarrow{i}\mathrmbfit{tup}_{\mathcal{A}}(\mathcal{S})$.
%
A \texttt{FOLE} relation morphism 
%\[\mbox
{\footnotesize{$
\mathcal{R}' = {\langle{\mathcal{D}',R'}\rangle}
%\xleftarrow{{\langle{h,f,g}\rangle}} 
\xleftarrow{{\langle{{\langle{h,f,g}\rangle},r}\rangle}} 
{\langle{\mathcal{D},R}\rangle} = \mathcal{R}
$}\normalsize}
%\hfill\mbox{}\newline\]
%
consists of 
a signed domain morphism
{\footnotesize{$
\mathcal{D}'={\langle{\mathcal{S}',\mathcal{A}'}\rangle}
\xrightarrow{{\langle{h,f,g}\rangle}}
{\langle{\mathcal{S},\mathcal{A}}\rangle}=\mathcal{D}
$}\normalsize}
and
an inclusion key function $R'\xhookleftarrow{r}R$,
which satisfy the naturality condition
$r{\;\cdot\;}i' = i{\;\cdot\;}\mathrmbfit{tup}(h,f,g)$.
Let $\mathrmbf{Rel}$ denote the mathematical context of relations. 
%

%%%%%%%%%%%%%%%%%%%%%%%%%%%%%%%%%%%%%%%%%%%%%%%%%%
%\newpage
\paragraph{Small Fibers.}
%%%%%%%%%%%%%%%%%%%%%%%%%%%%%%%%%%%%%%%%%%%%%%%%%%

Let $\mathcal{D}
={\langle{\mathcal{S},\mathcal{A}}\rangle}
={\langle{\mathcal{S},\mathcal{A}}\rangle}$ 
be a fixed signed domain.
The fiber mathematical context of 
of $\mathcal{D}$-relations 
is denoted by
$\mathrmbf{Rel}(\mathcal{D}) 
= \mathrmbf{Rel}_{\mathcal{D}}(\mathcal{S})
= \mathrmbf{Rel}_{\mathcal{D}}(\mathcal{S})$.
A $\mathcal{D}$-relation $\mathcal{R} = {\langle{R,i}\rangle}$ 
%called an ${\langle{\mathcal{S},\mathcal{A}}\rangle}$-table,
consists of a subset $R \subseteq \mathrmbfit{tup}_{\mathcal{A}}(\mathcal{S})$ 
of tuples
with inclusion tuple function
$R\xhookrightarrow{\,i\;}\mathrmbfit{tup}_{\mathcal{A}}(\mathcal{S})$.
A $\mathcal{D}$-relation morphism 
$\mathcal{R}' = {\langle{R',i'}\rangle}\xhookleftarrow{r}{\langle{R,i}\rangle} = \mathcal{R}$
is an inclusion tuple function $R'\xhookleftarrow{\,r\,}R$
satisfying the condition
$\mathrmbfit{tup}_{\mathcal{A}}(\mathcal{S}) \supseteq R' \supseteq R$.
%\end{itemize}
In short,
the small fiber context of $\mathcal{D}$-relations 
is 
%exactly the same as 
the partial order of 
subsets of $\mathcal{D}$-tuples
$\mathrmbf{Rel}(\mathcal{D}) =
{\langle{{\wp}\mathrmbfit{tup}_{\mathcal{A}}(\mathcal{S}),\subseteq}\rangle}$.

%%%%%%%%%%%%%%%%%%%%%%%%%%%%%%%%%%%%%%%%%%%%%%%%%%%%%%%%%%%%%%%%%%
%5\newpage
%\subsubsection{Auxiliary Relations.}
%\label{sub:sec:aux:rel}
%%%%%%%%%%%%%%%%%%%%%%%%%%%%%%%%%%%%%%%%%%%%%%%%%%%%%%%%%%%%%%%%%%
%%%%%%%%%%%%%%%%%%%%%%%%%%%%%%%%%%%%%%%%%%%%%%%%%%
%\newpage
\paragraph{Auxiliary Relations.}
%%%%%%%%%%%%%%%%%%%%%%%%%%%%%%%%%%%%%%%%%%%%%%%%%%

Here are some examples of auxiliary relations.

%%%%%%%%%%%%%%%%%%%%%%%%%%%%%%%%%%%%%%%%%%%%%%%%%%%%%%%%%%%%%%
%\newpage
\begin{example}{{\footnotesize{$\sigma_{n \theta \hat{n}}(\mathcal{T}):$}}}\label{bin:op:bnd}
%%%%%%%%%%%%%%%%%%%%%%%%%%%%%%%%%%%%%%%%%%%%%%%%%%%%%%%%%%%%%%
%\begin{itemize}
%\item
Let $\hat{x} \in X$ be a sort
whose data-type
$\mathcal{A}_{\hat{x}} 
= \mathrmbfit{tup}_{\mathcal{A}}(1,\hat{x}) = P$ 
%contains the elements of some partial order 
%${\langle{P,\leq}\rangle}$.
contains the elements of some binary relation 
${\langle{P,\theta}\rangle}$
with
= $\theta \subseteq P{\times}P$.
%
%%%%%%%%%%%%%%%%%%%%%%%%%%%%%%%%%%%%%%%%%%%%%%%%%%%%%%%%%%%%
%%%%%%%%%%%%%%%%%%%%%%%%%%%%%%%%%%%%%%%%%%%%%%%%%%%%%%%%%%%%
\footnote{\label{egs}
Examples include a partial order ${\langle{P,\leq}\rangle}$,
where $\theta$ is one of the binary operations in the set 
$\{\;<,\leq ,=,\neq ,\geq ,\;>\}$.}
%%%%%%%%%%%%%%%%%%%%%%%%%%%%%%%%%%%%%%%%%%%%%%%%%%%%%%%%%%%%
%%%%%%%%%%%%%%%%%%%%%%%%%%%%%%%%%%%%%%%%%%%%%%%%%%%%%%%%%%%%
%\item
Then, 
$X$-signature (header) ${\langle{1,\hat{x}}\rangle}$
with sort function 
$1
%=\{\cdot\}
\xrightarrow{\,\hat{x}\,}X$
%mapping $0$ and $1$ to corresponding sorts
%$\hat{x}_{0},\;\hat{x}_{1} \in X$
has the extent
$\mathrmbfit{tup}_{\mathcal{A}}(1,\hat{x})
= \mathcal{A}_{\hat{x}}
= P$.
%\item
Let 
$N 
%= {\downarrow{\!\hat{n}}} 
= n \theta \hat{n}
%= \{ n {\,\theta\,} \hat{n} \}$
= \{ (n,\hat{n}) \in P{\times}P \mid n \theta \hat{n} \}$
%be the common set of data values.
be the elements of $n \in P$ related to $\hat{n}$
for some distinguished element $\hat{n} \in P$.
%\item 
Since
$N 
%= {\downarrow{\!\hat{n}}} 
\subseteq 
\mathcal{A}_{\hat{x}} 
=
\mathrmbfit{tup}_{\mathcal{A}}(1,\hat{x})$,
%
%Here we are using the $X$-signature 
%${\langle{1,\hat{s}}\rangle}$
%with sort function $1\xrightarrow{\,\hat{s}\,}X$.
%Hence,
$\mathcal{N} = 
{\langle{1,\hat{x},N}\rangle}$
is an $\mathcal{A}$-relation.
%in the fiber context $\mathrmbf{Rel}(\mathcal{A})$.
%\item 
Let $\mathcal{T} = {\langle{\mathcal{S},K,t}\rangle}$ 
be an $\mathcal{A}$-table.
%with $X$-signature 
%$\mathcal{S}$
%having sort map $I\xrightarrow{\,s\,}X$.
%\item 
Assume that $\mathcal{T}$ has an attribute 
${\langle{\hat{\imath},\hat{x}}\rangle}$ 
with index $\hat{\imath}{\,\in\,}I$ and sort $\hat{x}{\,\in\,}X$.
This defines an index function
$1\xrightarrow{\;\hat{\imath}\;}I$
satisfying the condition:
$\hat{\imath}{\;\cdot\;}s = \hat{x}$.
Hence, 
$\mathcal{N}$ is connected to $\mathcal{T}$
via the $X$-signature morphism
{{${\langle{1,\hat{x}}\rangle}
\xrightarrow{\;\hat{\imath}\;}\mathcal{S}$.}}
%\hfill\mbox{}\newline
%\end{itemize}
%
Selection 
%within the context $\mathrmbf{Tbl}(\mathcal{A})$
%is inflation followed by meet (conjunction, or intersection).
%This 
is the two-step process
%(illustrated in Fig.\;\ref{fig:fole:select:aux:rel})
%\newline\mbox{}\hfill\rule[-10pt]{0pt}{26pt}
$\sigma_{n \theta \hat{n}}(\mathcal{T}) 
\doteq  
\grave{\mathrmbfit{tbl}}_{\mathcal{A}}(\hat{\imath})(\mathcal{N}){\;\wedge\;}\mathcal{T}$. 
%\hfill\mbox{}\newline
%
%%%%%%%%%%%%%%%%%%%%%%%%%%%%%%%%%%%%%%%%%%%%%%%%%%%%%%%%%%%%%%
\end{example}
%%%%%%%%%%%%%%%%%%%%%%%%%%%%%%%%%%%%%%%%%%%%%%%%%%%%%%%%%%%%%%

%%%%%%%%%%%%%%%%%%%%%%%%%%%%%%%%%%%%%%%%%%%%%%%%%%%%%%%%%%%%%%
%\newpage
\begin{example}{{\footnotesize{$\sigma_{n_{0} \theta n_{1}}(\mathcal{T}):$}}}\label{bin:op:gph}
%%%%%%%%%%%%%%%%%%%%%%%%%%%%%%%%%%%%%%%%%%%%%%%%%%%%%%%%%%%%%%
%%%%%%%%%%%%%%%%%%%%%%%%%%%%%%%%%%%%%%%%%%%%%%%%%%%%%%%%%%%%
%%%%%%%%%%%%%%%%%%%%%%%%%%%%%%%%%%%%%%%%%%%%%%%%%%%%%%%%%%%%
\footnote{Generalized selection, 
written as $\sigma_{\varphi}(\mathcal{T})$,
where $\varphi$ is a propositional formula 
that consists of atoms as allowed in the normal selection and the logical operators
$\wedge$ (and), $\vee$ (or) and $\neg$ (negation). 
This selection chooses all those tuples in $\mathcal{T}$ 
for which $\varphi$ holds.}
%%%%%%%%%%%%%%%%%%%%%%%%%%%%%%%%%%%%%%%%%%%%%%%%%%%%%%%%%%%%
%%%%%%%%%%%%%%%%%%%%%%%%%%%%%%%%%%%%%%%%%%%%%%%%%%%%%%%%%%%%
%\begin{itemize}
%\item 
Let $\hat{x}_{0},\hat{x}_{1} \in X$ be two sorts
whose data-types
$\mathcal{A}_{\hat{x}_{0}} 
= \mathrmbfit{tup}_{\mathcal{A}}(1,\hat{x}_{0}) 
= P 
= \mathrmbfit{tup}_{\mathcal{A}}(1,\hat{x}_{1}) 
= \mathcal{A}_{\hat{x}_{1}}$  
contain the elements of some binary relation 
${\langle{P,\theta}\rangle}$
with
= $\theta \subseteq P{\times}P$.
$^{\ref{egs}}$
%\item 
%
Then, 
$X$-signature (header) ${\langle{2,\hat{x}}\rangle}$
with sort function $2=\{0,1\}\xrightarrow{\,\hat{x}\,}X$
%mapping $0$ and $1$ to corresponding sorts
%$\hat{x}_{0},\;\hat{x}_{1} \in X$
has the extent
$\mathrmbfit{tup}_{\mathcal{A}}(2,\hat{x})
= \mathcal{A}_{\hat{x}_{0}}{\!\times}\mathcal{A}_{\hat{x}_{1}}
= P{\times}P$.
%\item 
Let
$N 
%= {P{\theta}P} 
= n_{0} \theta n_{1}
= \{ (n_{0},n_{1}) \in P{\times}P \mid n_{0} \theta n_{1} \}$
be the graph of the binary relation.
%\item 
Since
$N \subseteq 
\mathcal{A}_{\hat{x}_{0}}{\!\times}\mathcal{A}_{\hat{x}_{1}}
= \mathrmbfit{tup}_{\mathcal{A}}(2,\hat{x})$,
$\mathcal{N} = 
{\langle{2,\hat{x},N}\rangle}$
is an $\mathcal{A}$-relation.
%\begin{itemize}
%\item 
Let $\mathcal{T} = {\langle{\mathcal{S},K,t}\rangle}$ 
be an $\mathcal{A}$-table.
%with $X$-signature 
%$\mathcal{S}$
%having sort map $I\xrightarrow{\,s\,}X$.
%\item 
Assume that $\mathcal{T}$ has 
a sub-signature (sub-header) 
consisting of two indexes 
$\hat{\imath}_{0},\;\hat{\imath}_{1} \in I$
with corresponding sorts
$\hat{x}_{0},\;\hat{x}_{1} \in X$.
This defines an index function
$2\xrightarrow{\;\hat{\imath}\;}I$
satisfying the condition:
$\hat{\imath}{\;\cdot\;}s = \hat{x}$.
%\end{itemize}
%
Hence, 
$\mathcal{N}$ is connected to $\mathcal{T}$
via the $X$-signature morphism
{{${\langle{2,\hat{x}}\rangle}
\xrightarrow{\;\hat{\imath}\;}\mathcal{S}$.}}
Selection 
is the two-step process
$\sigma_{n_{0} \theta n_{1}}(\mathcal{T}) 
\doteq  
\grave{\mathrmbfit{tbl}}_{\mathcal{A}}(\hat{\imath})(\mathcal{N}){\;\wedge\;}\mathcal{T}$. 
%\hfill\mbox{}\newline
%
%%%%%%%%%%%%%%%%%%%%%%%%%%%%%%%%%%%%%%%%%%%%%%%%%%%%%%%%%%%%%%
\end{example}
%%%%%%%%%%%%%%%%%%%%%%%%%%%%%%%%%%%%%%%%%%%%%%%%%%%%%%%%%%%%%%

%
%%%%%%%%%%%%%%%%%%%%%%%%%%%%%%%%%%%%%%%%%%%%%%%%%%%%%%%%%%%%
%\newpage
\begin{example}
%%%%%%%%%%%%%%%%%%%%%%%%%%%%%%%%%%%%%%%%%%%%%%%%%%%%%%%%%%%%
%
Suppose one is working with a database table 
$\mathcal{T} = {\langle{\mathcal{S},\mathcal{A},K,t}\rangle}$
concerned with countries of the world,
where population is an attribute $(\hat{\imath},\hat{x})$ 
of index $\hat{\imath}{\,\in}I$ 
and sort $s(\hat{\imath})=\hat{x}{\,\in}X$ 
representing the data-type 
$\mathcal{A}_{\hat{x}} = \{0,1,\cdots\}$
of natural numbers.
One might want to select only those tuples from $\mathcal{T}$ 
representing countries
whose population is above a certain threshold 
$\hat{n} \in \{0,1,\cdots\}$.
Let 
$N 
= {\uparrow{\!\hat{n}}} 
= \{ n \geq \hat{n} \mid n \in \{0,1,\cdots\} \}$.
%
%%%%%%%%%%%%%%%%%%%%%%%%%%%%%%%%%%%%%%%%%%%%%%%%%%
\end{example}
%%%%%%%%%%%%%%%%%%%%%%%%%%%%%%%%%%%%%%%%%%%%%%%%%%
%

%%%%%%%%%%%%%%%%%%%%%%%%%%%%%%%%%%%%%%%%%%%%%%%%%%
%\newpage
\begin{example}
%%%%%%%%%%%%%%%%%%%%%%%%%%%%%%%%%%%%%%%%%%%%%%%%%%
%
Another example using the natural numbers data-type 
is the DMV (Department of Motor Vehicles),
which will need to work with a motorists database table 
$\mathcal{M} = 
{\langle{{\langle{\mathcal{S},\mathcal{A}}\rangle},K,t}\rangle}$. 
%with signature $\mathcal{S} = {\langle{I,x,X}\rangle}$
%and typed domain (attribute classification)
%$\mathcal{A} = {\langle{X,Y,\models_{\mathcal{A}}}\rangle}$
%set $K$ of keys, and
%tuple function $K\xrightarrow{t}\mathrmbfit{tup}_{\mathcal{A}}(\mathcal{S})$.
%
This table will have attributes such as 
\textsf{name}, \textsf{age}, \textsf{address}, \textsf{phone}, \textsf{medical}, \textsf{soc-sec-no},  etc.
Each \textsf{motorist} in the database 
will be represented by a key value $m \in K$.
The tuple for this key
$t(m) = (I,y) \in \mathrmbfit{tup}_{\mathcal{A}}(\mathcal{S}) \subseteq \mathrmbf{List}(Y)$
will have values for the various attributes such as \textsf{name}, \textsf{address} and \textsf{age}.
Here, \textsf{age} 
%in the motorists table $\mathcal{M}$
is an attribute $(\hat{\imath},\hat{x})$ 
of index $\hat{\imath}{\,\in}I$ 
and sort $s(\hat{\imath})=\hat{x}{\,\in}X$ 
representing the data-type 
$\mathcal{A}_{\hat{x}} = \{0,1,\cdots\}$
of natural numbers.
To define a sub-table for all minors 
--- motorist below the age of 21
---
we need to select only those tuples from $\mathcal{M}$ 
representing people
whose \textsf{age} attribute is below the threshold 
$\hat{n}=21 \in \{0,1,\cdots\}$.
Let 
$N 
= {\downarrow{\!21}} 
= \{ n \leq 21 \mid n \in \{0,1,\cdots\} \}$.
%%%%%%%%%%%%%%%%%%%%%%%%%%%%%%%%%%%%%%%%%%%%%%%%%%
\end{example}
%%%%%%%%%%%%%%%%%%%%%%%%%%%%%%%%%%%%%%%%%%%%%%%%%%
%

%%%%%%%%%%%%%%%%%%%%%%%%%%%%%%%%%%%%%%%%%%%%%%%%%%%%%%%%%%%%%%%%%%
%%%%%%%%%%%%%%%%%%%%%%%%%%%%%%%%%%%%%%%%%%%%%%%%%%%%%%%%%%%%%%%%%%
%%%%%%%%%%%%%%%%%%%%%%%%%%%%%%%%%%%%%%%%%%%%%%%%%%%%%%%%%%%%%%%%%%
%\comment{% not needed
%}% not needed
%%%%%%%%%%%%%%%%%%%%%%%%%%%%%%%%%%%%%%%%%%%%%%%%%%%%%%%%%%%%%%%%%%
%%%%%%%%%%%%%%%%%%%%%%%%%%%%%%%%%%%%%%%%%%%%%%%%%%%%%%%%%%%%%%%%%%
%%%%%%%%%%%%%%%%%%%%%%%%%%%%%%%%%%%%%%%%%%%%%%%%%%%%%%%%%%%%%%%%%%

%%%%%%%%%%%%%%%%%%%%%%%%%%%%%%%%%%%%%%%%%%%%%%%%%%%%%%%%%%%%%%%%%%
%%%%%%%%%%%%%%%%%%%%%%%%%%%%%%%%%%%%%%%%%%%%%%%%%%%%%%%%%%%%%%%%%%
%
\newpage
\subsection{Completeness and Co-completeness.}
\label{sub:sec:lim:colim:tbl}
%%%%%%%%%%%%%%%%%%%%%%%%%%%%%%%%%%%%%%%%%%%%%%%%%%%%%%%%%%%%%%%%%%
%%%%%%%%%%%%%%%%%%%%%%%%%%%%%%%%%%%%%%%%%%%%%%%%%%%%%%%%%%%%%%%%%%

%
The following proposition 
using comma contexts
appears as 
Prop.\;7  in
``The {\ttfamily FOLE} Table''
\cite{kent:fole:era:tbl}.
\begin{proposition}\label{prop:com:cxt:lim:colim}
The contexts 
of signed domains $\mathrmbf{Dom}$ 
and 
%the small fiber context 
%of 
tables 
$\mathrmbf{Tbl}(\mathcal{D})=\mathrmbf{Tbl}_{\mathcal{A}}(\mathcal{S})$
%for fixed signed domain $\mathcal{D}={\langle{\mathcal{S},\mathcal{A}}\rangle}$
are 
%comma contexts
associated with the following passage opspans:
\begin{center}
{{\footnotesize\setlength{\extrarowheight}{4pt}
$\begin{array}{r@{\;=\;}l@{\hspace{10pt}:\hspace{16pt}}l}
\multicolumn{2}{c}{\textit{comma context}} & \textit{passage opspan}
\\
\mathrmbf{Dom}
& \bigl(\mathrmbf{Set}{\,\downarrow\,}\mathrmbfit{sort}\bigr)
& \mathrmbf{Set}\xrightarrow{\;\mathrmbfit{1}\;}\mathrmbf{Set}\xleftarrow{\;\mathrmbfit{sort}\;}\mathrmbf{Cls}
\\
\mathrmbf{Tbl}(\mathcal{D})
& \bigl(\mathrmbf{Set}{\,\downarrow\,}\mathrmbfit{tup}(\mathcal{D})\bigr)
& \mathrmbf{Set}\xrightarrow{\;\mathrmbfit{1}\;}
\mathrmbf{Set}
\xleftarrow[\mathrmbfit{tup}_{\mathcal{A}}(\mathcal{S})]
{\;\mathrmbfit{tup}(\mathcal{D})\;}\mathrmbf{1}
\end{array}$.}}
\end{center}
respectively.
Hence, they are complete and co-complete and their projections
%
%%%%%%%%%%%%%%%%%%%%%%%%%%%%%%%%%%%%%%%%%%%%%%%%%%%%%%%%%%%%%%%%%%%%%%%%%%%%%%%%
%%%%%%%%%%%%%%%%%%%%%%%%%%%%%%%%%%%%%%%%%%%%%%%%%%%%%%%%%%%%%%%%%%%%%%%%%%%%%%%%
\footnote{The completeness and co-completeness of
the small fiber context of tables 
$\mathrmbf{Tbl}(\mathcal{D})=\mathrmbf{Tbl}_{\mathcal{A}}(\mathcal{S})$
is represented by the data-type operations of \S\,\ref{sub:sub:sec:boole}.}
%%%%%%%%%%%%%%%%%%%%%%%%%%%%%%%%%%%%%%%%%%%%%%%%%%%%%%%%%%%%%%%%%%%%%%%%%%%%%%%%
%%%%%%%%%%%%%%%%%%%%%%%%%%%%%%%%%%%%%%%%%%%%%%%%%%%%%%%%%%%%%%%%%%%%%%%%%%%%%%%%
%
\begin{center}
{\begin{tabular}{@{\hspace{20pt}}c}
{{\footnotesize\setlength{\extrarowheight}{4pt}$\begin{array}[t]
{@{\hspace{5pt}}r@{\hspace{5pt}}c@{\hspace{5pt}}l@{\hspace{5pt}}}
\mathrmbf{Set}\xleftarrow{\;\mathrmbfit{arity}\;}
\!\!\!\!\!
&
{\mathrmbf{Dom}}
&
\!\!\!\!
\xrightarrow{\;\mathrmbfit{data}\;}\mathrmbf{Cls}
\\
\mathrmbf{Set}\xleftarrow{\;\mathrmbfit{key}(\mathcal{D})\;}
\!\!\!
&
{\mathrmbf{Tbl}(\mathcal{D})}
&
\!
\xrightarrow{\;\;\;}\mathrmbf{1}
\end{array}$}}
\end{tabular}}
\end{center}
are continuous and co-continuous.
\end{proposition}
\begin{proof}
The contexts $\mathrmbf{Set}$, $\mathrmbf{Cls}$ and $\mathrmbf{1}$ 
are complete and co-complete; 
the passages 
$\mathrmbf{Set}\xrightarrow{\mathrmbfit{1}}\mathrmbf{Set}$,
$\mathrmbf{Cls}\xrightarrow{\mathrmbfit{sort}}\mathrmbf{Set}$, and
$\mathrmbf{1}\xrightarrow{\;\mathrmbfit{tup}_{\mathcal{A}}(\mathcal{S})}\mathrmbf{Set}$
are continuous; and
the passage $\mathrmbf{Set}\xrightarrow{\mathrmbfit{1}}\mathrmbf{Set}$ 
is continuous and co-continuous.
\mbox{}\hfill\rule{5pt}{5pt}
\end{proof}
%\newpage
%
The following proposition using the Grothendieck construction
appears as Prop.\;9 in ``The {\ttfamily FOLE} Table''
\cite{kent:fole:era:tbl}.
This provides a framework for relational operations.
\begin{proposition}\label{prop:lim:tbl}
The fibered context of tables $\mathrmbf{Tbl}$ is complete and co-complete 
and the projection 
$\mathrmbf{Tbl}\;\xrightarrow{\;\mathrmbfit{dom}^{\mathrm{op}}}\mathrmbf{Dom}^{\mathrm{op}}$
is continuous and co-continuous.
%
%%%%%%%%%%%%%%%%%%%%%%%%%%%%%%%%%%%%%%%%%%%%%%%%%%%%%%%%%%%%%%%%%%%%%%
%%%%%%%%%%%%%%%%%%%%%%%%%%%%%%%%%%%%%%%%%%%%%%%%%%%%%%%%%%%%%%%%%%%%%%
\comment{The fibered context of tables
$\mathrmbf{Tbl}\xrightarrow{\;\mathrmbfit{dom}\;}\mathrmbf{Dom}^{\mathrm{op}}$
is the Grothendieck construction of 
the signed domain indexed adjunction
$\mathrmbf{Dom}^{\mathrm{op}}\!\xrightarrow{\;\mathrmbfit{tbl}\;}\mathrmbf{Adj}$.
(Thm.\;3 in \cite{kent:fole:era:tbl}).}
%%%%%%%%%%%%%%%%%%%%%%%%%%%%%%%%%%%%%%%%%%%%%%%%%%%%%%%%%%%%%%%%%%%%%%
%%%%%%%%%%%%%%%%%%%%%%%%%%%%%%%%%%%%%%%%%%%%%%%%%%%%%%%%%%%%%%%%%%%%%%
%
\end{proposition}
\begin{proof}\mbox{}
The fibered context of tables
$\mathrmbf{Tbl}\xrightarrow{\;\mathrmbfit{dom}\;}\mathrmbf{Dom}^{\mathrm{op}}$
is the Grothendieck construction of 
the signed domain indexed adjunction
$\mathrmbf{Dom}^{\mathrm{op}}\!\xrightarrow{\;\mathrmbfit{tbl}\;}\mathrmbf{Adj}$
(Thm.\;3 in \cite{kent:fole:era:tbl}).
%(Thm.\;3  in
%%``The {\ttfamily FOLE} Table''
%\cite{kent:fole:era:tbl}) 
%
\end{proof}
%

%%%%%%%%%%%%%%%%%%%%%%%%%%%%%%%%%%%%%%%%%%%%%%%%%%%%%%%%%%%%%%%%%%
%
\newpage
\subsubsection{Completeness.}
\label{sub:sub:sec:lim:tbl}
%%%%%%%%%%%%%%%%%%%%%%%%%%%%%%%%%%%%%%%%%%%%%%%%%%%%%%%%%%%%%%%%%%

%\begin{description}
%\item[Completeness:] 
%For limits in $\mathrmbf{Tbl}$
%[\textbf{signed domain colimit}]
For the completeness of $\mathrmbf{Tbl}$
and the continuity of
$\mathrmbf{Tbl}\;\xrightarrow{\;\mathrmbfit{dom}^{\mathrm{op}}}\mathrmbf{Dom}^{\mathrm{op}}$,
use the right adjoint 
%contravariant 
pseudo-passage
$\mathrmbf{Dom}\;\xrightarrow{\;\grave{\mathrmbfit{tbl}}\;}\mathrmbf{Cxt}$ 
(See Fact.\;3 of \cite{kent:fole:era:tbl}):
%\begin{enumerate}
%\item 
the indexing context $\mathrmbf{Dom}^{\mathrm{op}}$ is complete;
%\item 
the fiber context $\mathrmbf{Tbl}(\mathcal{D})$ 
of \S\,\ref{sub:sub:sec:boole}
is complete (meets exist)
for each signed domain $\mathcal{D}$; and,
%\item 
the fiber passage 
$\mathrmbf{Tbl}(\mathcal{D}')
\xrightarrow{\grave{\mathrmbfit{tbl}}({h,f,g})}
\mathrmbf{Tbl}(\mathcal{D})$ 
of \S\,\ref{sub:sub:sec:flow:sign:dom:mor}
(\textit{restrict-inflate} operation)
is continuous 
for each signed domain morphism 
$\mathcal{D}'\xrightarrow{{\langle{h,f,g}\rangle}}\mathcal{D}$. 
%\end{enumerate}
\mbox{}
%
%\end{description}
%
%%%%%%%%%%%%%%%%%%%%%%%%%%%%%%%%%%%%%%%%%%%%%%%%%%%%%%%%%%%%%%%%%%
%\newpage
%\paragraph{\textbf{Application.}}
%%%%%%%%%%%%%%%%%%%%%%%%%%%%%%%%%%%%%%%%%%%%%%%%%%%%%%%%%%%%%%%%%

%
Given a diagram of tables
$\mathrmbf{I} \xrightarrow{\mathrmbfit{T}} \mathrmbf{Tbl}$ 
consisting of an indexed (and linked) collection of tables 
{\footnotesize{$
\underset{\textbf{theoretical input}}
{\underbrace{
\bigl\{ \mathcal{T}_{i} = \mathrmbfit{T}(i) \in \mathrmbf{Tbl} 
\mid i \in \mathrmbf{I} \bigr\}}}$}}
with links
$\mathcal{T}_{i'} 
\xleftarrow[\mathrmbfit{T}(e)]{{\langle{{\langle{h,f,g}\rangle},k}\rangle}} 
\mathcal{T}_{i}$
for each $i' \xrightarrow{e} i$ in $\mathrmbf{I}$,
Prop.\;\ref{prop:lim:tbl}
states that the diagram of tables
$\mathrmbfit{T}$
%$\mathrmbf{I} \xrightarrow{\mathrmbfit{T}} \mathrmbf{Tbl}$ 
has a limit table $\prod\mathrmbfit{T}$  
with projection bridge
$\Delta(\prod\mathrmbfit{T}) \xRightarrow{\,\pi\;} \mathrmbfit{T}$ 
consisting of an indexed collection of table morphisms 
{\footnotesize{$
\underset{\textbf{theoretical output}}
{\underbrace{\textstyle{
\bigl\{ \prod\mathrmbfit{T} \xrightarrow
[{\langle{{\langle{\hat{h}_{i},\hat{f}_{i},\hat{g}_{i}}\rangle},\hat{k}_{i}}\rangle}]
{\pi_{i}} \mathcal{T}_{i} 
\mid i \in \mathrmbf{I} \bigr\}
}}}$}}
satisfying naturality.
% 
%\begin{description}
%
%%%%%%%%%%%%%%%%%%%%%%%%%%%%%%%%%%%%%%%%%%%%%%%%%%%%%%%%%%%%%%%%%%
%\newpage\paragraph{\textbf{Application (limits).}}
%%%%%%%%%%%%%%%%%%%%%%%%%%%%%%%%%%%%%%%%%%%%%%%%%%%%%%%%%%%%%%%%%
%\item[\underline{Limits:}] 
%
%
\begin{description}
\item[\underline{Input:}] 
We use a sufficient subset of tables 
(Def.\;\ref{def:suff:adequ:lim})
\newline\mbox{}\hfill
{\footnotesize{$\underset{\textbf{practical input}}
{\underbrace{
\bigl\{ 
\mathcal{T}_{i} = \mathrmbfit{T}(i) \in \mathrmbf{Tbl} 
\mid i \in I \subseteq obj(\mathrmbf{I}) \bigr\}}}$}} 
\hfill\mbox{}\newline
%
%consisting of an indexed, but unlinked, collection of tables 
reachable \underline{to} other tables
in the collection 
{\footnotesize{$\bigl\{ 
\mathcal{T}_{i} = \mathrmbfit{T}(i) \in \mathrmbf{Tbl} 
\mid i \in \mathrmbf{I} \bigr\}$.}}
\item[\underline{Constraint:}] 
The projection 
$\mathrmbf{Tbl}^{\mathrm{op}}\xrightarrow{\;\mathrmbfit{dom}\;}\mathrmbf{Dom}$
maps the diagram of tables $\mathrmbfit{T}$ to a diagram of signed domains
$
\mathrmbf{I}^{\mathrm{op}}\xrightarrow{\mathrmbfit{D}} \mathrmbf{Dom}
$  
consisting of an indexed collection 
%of signed domains 
$\mathcal{D}_{i} = \mathrmbfit{D}(i) \in \mathrmbf{Dom} 
\mid i \in \mathrmbf{I} \bigr\}$
with links
{\footnotesize{$
\underset{\textbf{constraint}}
{\underbrace{ \bigl\{ 
\mathcal{D}_{i'} 
\xrightarrow
%[\mathrmbfit{D}(e)]
{{\langle{h,f,g}\rangle}} 
\mathcal{D}_{i}
\mid i \in \mathrmbf{I} \bigr\} }}$.}}
%for each $(i' \xrightarrow{e} i) \in \mathrmbf{I}$,
%
%\end{description}
%
%\begin{itemize}
%
\item[\underline{Construction:}] 
Since
$\mathrmbf{Tbl}
\xrightarrow{\;\mathrmbfit{dom}^{\mathrm{op}}}
\mathrmbf{Dom}^{\mathrm{op}}$
is continuous
(Prop.\;\ref{prop:lim:tbl}), 
it maps
the limit table 
%$\prod\mathrmbfit{T}$  
with projection bridge
%$\Delta(\prod\mathrmbfit{T}) \xRightarrow{\,\pi\;} \mathrmbfit{T}$  
to
the colimit signed domain $\coprod\mathrmbfit{D}=\widehat{\mathcal{D}}$  
with injection bridge
%$\Delta(\coprod\mathrmbfit{D}) \xRightarrow{\,\iota\;} \mathrmbfit{D}$  
%to the injection bridge
$
\Delta(\widehat{\mathcal{D}}) \xLeftarrow{\;\iota\,} \mathrmbfit{D}
$ 
consisting of an indexed collection of signed domain morphisms 
{\footnotesize{$
\underset{\textbf{construction}}
{\underbrace{
\bigl\{ 
\widehat{\mathcal{D}} 
\xleftarrow[{\langle{\hat{h}_{i},\hat{f}_{i},\hat{g}_{i}}\rangle}]{\iota_{i}} 
\mathcal{D}_{i} 
\mid i \in \mathrmbf{I} \bigr\} }}
$}}
satisfying naturality.  
\item[\underline{Output:}] 
%
%\end{description}
%
%\begin{itemize}
%
%\newline\item 
Using the completeness aspect of Prop.\;\ref{prop:lim:tbl},
the limit table is 
the meet (intersection)
in the fiber context $\mathrmbf{Tbl}(\widehat{\mathcal{D}})$
of the restriction-inflation of the component tables
$\mathcal{T}_{i}$
along the injections
%signed domain morphisms 
{\footnotesize{$
\{ \mathcal{D}_{i} 
\xrightarrow
{{\langle{\hat{h}_{i},\hat{f}_{i},\hat{g}_{i}}\rangle}}
%{\iota_{i}} 
\widehat{\mathcal{D}} 
\mid i \in \mathrmbf{I} \}$.}}
%for each index $i \in \mathrmbf{I}$

%$\prod\mathrmbfit{T}$
%essentially 
%\newline\mbox{}\hfill
%\hfill\mbox{}\newline
%\begin{itemize}
%\item

$\bullet$ For each index $i \in I$,
restriction/inflation
{\footnotesize{$\mathrmbf{Tbl}(\mathcal{D}_{i}) 
\xrightarrow{\grave{\mathrmbfit{tbl}}({\hat{h}_{i},\hat{f}_{i},\hat{g}_{i}})} 
\mathrmbf{Tbl}(\widehat{\mathcal{D}})$}}
(\S\,\ref{sub:sub:sec:flow:sign:dom:mor})
%along the tuple function
%of the signed domain morphism 
%$\mathcal{D}_{i} 
%\xrightarrow{{\langle{\hat{h}_{i},\hat{f}_{i},\hat{g}_{i}}\rangle}} 
%\widehat{\mathcal{D}} = 
%\coprod \mathrmbfit{D}$
maps the table 
$\mathcal{T}_{i} 
%= {\langle{K_{i},t_{i}}\rangle} 
\in \mathrmbf{Tbl}(\mathcal{D}_{i})$
%with its tuple function 
%$K_{i} \xrightarrow{t_{i}} \mathrmbfit{tup}(\mathcal{D}_{i})$
by pullback
to the table
$\widehat{\mathcal{T}}_{i}
%\widehat{\mathrmbfit{T}}(i)
%= \grave{\mathrmbfit{tbl}}({h,f,g})(\mathrmbfit{T}(i))
= {\langle{\widehat{K}_{i},\hat{t}_{i}}\rangle} 
\in \mathrmbf{Tbl}({\widehat{\mathcal{D}}})$.
%with its tuple function
%$\widehat{K}_{i} \xrightarrow{\hat{t}_{i}} 
%\mathrmbfit{tup}(\widehat{\mathcal{D}})$
%defined by pullback, 
%$\grave{k}_{i}{\,\cdot\,}t_{i} 
%= \hat{t}_{i}{\,\cdot\,}{\mathrmbfit{tup}(\hat{h}_{i},\hat{f}_{i},\hat{g}_{i})}$. 
%
%Here we have ``vertically restricted'' and then ``horizontally inflated'' 
%tuples in $\mathrmbfit{tup}(\mathcal{D}_{i})$ 
%by pullback 
%along the tuple function
%%\[\mbox
%{\footnotesize{$
%\mathrmbfit{tup}(\mathcal{D}_{i})
%\xleftarrow
%%[(h{\cdot}{(\mbox{-})})\cdot({(\mbox{-})}{\cdot}g)]
%{\mathrmbfit{tup}(\hat{h}_{i},\hat{f}_{i},\hat{g}_{i})}
%\mathrmbfit{tup}(\widehat{\mathcal{D}})
%$}\normalsize}
%(see RHS Fig.\;\ref{fig:flow:sign:dom}).
%
This 
defines
%is linked 
%(Fig.\;\ref{fig:fole:boolean:meet})
%to the table $\mathcal{T}_{i}$ 
%by the 
a table morphism 
%\newline
%\[\mbox
{\footnotesize{{$
\mathcal{T}_{i} = {\langle{K_{i},t_{i}}\rangle}
\xleftarrow{{\langle{{\langle{\hat{h}_{i},\hat{f}_{i},\hat{g}_{i}}\rangle},\grave{k}_{i}}\rangle}}
{\langle{\widehat{K}_{i},\hat{t}_{i}}\rangle} 
= \widehat{\mathcal{T}}_{i}$.}}\normalsize}
%\]
%\newline

%\item
$\bullet$
Intersection (\S\,\ref{sub:sub:sec:boole})
of the tables $\{ \widehat{\mathcal{T}}_{i} \mid i \in I \}$
in the fiber context $\mathrmbf{Tbl}(\widehat{\mathcal{D}})$
defines the 
%natural join 
%limit 
generic meet
$\prod\mathrmbfit{T}
= \widehat{\mathrmbfit{T}} 
= \bigwedge \bigl\{\widehat{\mathcal{T}}_{i} \mid i \in I \bigr\}
= {\langle{\widehat{K},\hat{t}}\rangle}$,
%whose key set $\widehat{K}$ is the pullback and 
%whose tuple map is the mediating function
%$\widehat{K}\xrightarrow{{(\hat{t}_{i})}}
%\mathrmbfit{tup}(\widehat{\mathcal{D}})$
%of the multi-opspan
%$\bigl\{ \widehat{K}_{i}\xrightarrow{\hat{t}_{i}}
%\mathrmbfit{tup}(\widehat{\mathcal{D}}) \mid i \in I \bigr\}$,
%
resulting in the 
discrete multi-span (cone)
{\footnotesize{{$\Bigl\{ 
\widehat{\mathcal{T}}_{i} 
%= {\langle{\widehat{K}_{i},\hat{t}_{i}}\rangle} 
\xleftarrow{\;\hat{\pi}_{i}\;} 
%{\langle{\widehat{K},\hat{t}}\rangle} = 
\widehat{\mathrmbfit{T}} 
\mid i \in I \Bigr\}$.}}\normalsize}
Restriction-inflation composed with meet 
defines the span of table morphisms
$\underset{\textbf{practical output}}
{\underbrace{
\textstyle{
\bigl\{ 
\mathcal{T}_{i}
\xleftarrow
[\;\hat{\pi}_{i}{\circ\,}{\langle{{\langle{\hat{h}_{i},\hat{f}_{i},\hat{g}_{i}}\rangle},\grave{k}_{i}}\rangle}\;]
{\;{\langle{{\langle{\hat{h}_{i},\hat{f}_{i},\hat{g}_{i}}\rangle},\hat{k}_{i}}\rangle}\;} 
\prod\mathrmbfit{T} 
%= \widehat{\mathrmbfit{T}}
%\mid i \in I 
\bigr\}
}
}}$. 
%illustrated in 
%Fig.\;\ref{fig:fole:boolean:meet},
%which is the output for generic meet 
%(Tbl.\,\ref{tbl:fole:generic:meet:input:output}).
%\end{itemize}
%
\end{description}
%
%\newpage
\begin{aside}
To construct the table
$\prod\mathrmbfit{T}$
it is not necessary to use all of the tables
and table morphisms in the indexed collection above. 
Let 
$\mathcal{T}_{i'} 
\xleftarrow[\mathrmbfit{T}(e)]{{\langle{{\langle{h,f,g}\rangle},k}\rangle}} 
\mathcal{T}_{i}$
%$\mathrmbfit{T}(e) = {{\langle{{\langle{h,f,g}\rangle},k}\rangle}}
%:\mathcal{T}_{i'}\xeftarrow\mathcal{T}_{i}$
be part of the diagram $\mathrmbfit{T}$
satisfying the naturality condition 
$\pi_{i'} = \pi_{i} \circ \mathrmbfit{T}(e)$
for some morphism $i' \xleftarrow{e} i$ in $\mathrmbf{I}$.
%
%\item 
Let 
{\footnotesize{
$\mathcal{D}_{i'} 
\xrightarrow[\mathrmbfit{D}(e)]
{{\langle{h,f,g}\rangle}}
\mathcal{D}_{i}$
}\normalsize}
be its signed domain morphism 
satisfying the naturality condition
$\iota_{i'} = \mathrmbfit{D}(e) \circ \iota_{i}$.
\begin{center}
{{\begin{tabular}{@{\hspace{50pt}}c@{\hspace{10pt}}c@{\hspace{10pt}}c}
%%%%%%%%%%%%%%%%%%%%%%%%%%%%%%%%%%%%%%%%%%%%%%%%%%
{{\begin{tabular}{c}
\setlength{\unitlength}{0.56pt}
\begin{picture}(120,100)(0,-15)
\put(60,80){\makebox(0,0){\footnotesize{$\prod\mathrmbfit{T}$}}}
\put(0,0){\makebox(0,0){\footnotesize{$\mathcal{T}_{i'}$}}}
\put(120,0){\makebox(0,0){\footnotesize{$\mathcal{T}_{i}$}}}
\put(25,43){\makebox(0,0)[r]{\scriptsize{$\pi_{i'}$}}}
\put(97,43){\makebox(0,0)[l]{\scriptsize{$\pi_{i}$}}}
\put(66,10){\makebox(0,0){\scriptsize{$\mathrmbfit{T}(e)$}}}
\put(66,-10){\makebox(0,0){\tiny{$
{\langle{{\langle{h,f,g}\rangle},k}\rangle}$}}}
\put(45,64){\vector(-3,-4){40}}
\put(75,64){\vector(3,-4){40}}
\put(105,0){\vector(-1,0){90}}
\put(-15,25){\makebox(0,0)[r]{\footnotesize{$
\underset{\textstyle{\textit{needed}}}
{\overset{\textstyle{\textit{component}}}{\textit{not}}}
\left\{\rule{0pt}{20pt}\right.$}}}
%
%\put(60,-25){\makebox(0,0){\tiny{$
%\pi_{i}=
%{\langle{{\langle{\hat{h}_{i},\hat{f}_{i},\hat{g}_{i}}\rangle},\hat{k}_{i}}\rangle}
%$}}}
\end{picture}
\end{tabular}}}
%%%%%%%%%%%%%%%%%%%%%%%%%%%%%%%%%%%%%%%%%%%%%%%%%%
{{\begin{tabular}{c}
\setlength{\unitlength}{0.56pt}
\begin{picture}(48,100)(0,0)
\put(30,60){\makebox(0,0){\footnotesize{$
\overset{\textstyle{\textit{naturality}}}
{\textstyle{\textit{diagrams}}}
$}}}
\end{picture}
\end{tabular}}}
%%%%%%%%%%%%%%%%%%%%%%%%%%%%%%%%%%%%%%%%%%%%%%%%%%
&
%%%%%%%%%%%%%%%%%%%%%%%%%%%%%%%%%%%%%%%%%%%%%%%%%%
{{\begin{tabular}{c}
\setlength{\unitlength}{0.56pt}
\begin{picture}(120,100)(0,-15)
\put(60,80){\makebox(0,0){\footnotesize{$\coprod\mathrmbfit{D}$}}}
\put(0,0){\makebox(0,0){\footnotesize{$\mathcal{D}_{i'}$}}}
\put(120,0){\makebox(0,0){\footnotesize{$\mathcal{D}_{i}$}}}
\put(25,43){\makebox(0,0)[r]{\scriptsize{$\iota_{i'}$}}}
\put(97,43){\makebox(0,0)[l]{\scriptsize{$\iota_{i}$}}}
\put(66,10){\makebox(0,0){\scriptsize{$\mathrmbfit{D}(e)$}}}
\put(66,-10){\makebox(0,0){\tiny{${\langle{h,f,g}\rangle}$}}}
\put(6,12){\vector(3,4){39}}
\put(114,12){\vector(-3,4){39}}
\put(15,0){\vector(1,0){90}}
%
%\put(60,-25){\makebox(0,0){\tiny{$
%\iota_{i}=
%{\langle{\hat{h}_{i},\hat{f}_{i},\hat{g}_{i}}\rangle}
%$}}}
\end{picture}
\end{tabular}}}
%%%%%%%%%%%%%%%%%%%%%%%%%%%%%%%%%%%%%%%%%%%%%%%%%%
%%%%%%%%%%%%%%%%%%%%%%%%%%%%%%%%%%%%%%%%%%%%%%%%%%
\end{tabular}}}
\end{center}
%
%\begin{itemize}
%\item 
Using adjoint flow
(Disp.\;\ref{def:sign:dom:mor:adj}),
the table morphism
$\mathcal{T}_{i'} 
\xleftarrow
%[\mathrmbfit{T}(e)]
{{\langle{{\langle{h,f,g}\rangle},k}\rangle}} 
\mathcal{T}_{i}$
implies existence of a morphism 
$\grave{\mathrmbfit{tbl}}(h,f,g)(\mathcal{T}_{i'})
\xleftarrow{\,k'\,}\mathcal{T}_{i}$
in the fiber context $\mathrmbf{Tbl}(\mathcal{D})$.
%\item 
%the fiber context $\mathrmbf{Tbl}(\mathcal{D})$ is a preorder,
For all practical purposes,
by reflection (\S\,\ref{sub:sub:sec:reflect})
we 
essentially 
have the inclusion
$\grave{\mathrmbfit{tbl}}({h,f,g})(\mathcal{T}_{i'}) 
\geq \mathcal{T}_{i}$.
%\item 
Applying restriction-inflation to the naturality condition
$\iota_{i'} = \mathrmbfit{D}(e) \circ \iota_{i}$,
get
$\grave{\mathrmbfit{tbl}}(\iota_{i'})
= \grave{\mathrmbfit{tbl}}({h,f,g}) \cdot \grave{\mathrmbfit{tbl}}(\iota_{i})$.
%\item
Hence,
$\widehat{\mathcal{T}}_{i'} =
\grave{\mathrmbfit{tbl}}(\iota_{i'})(\mathcal{T}_{i'}) =
\grave{\mathrmbfit{tbl}}(\iota_{i})
(\grave{\mathrmbfit{tbl}}({h,f,g})(\mathcal{T}_{i'})) \geq 
\grave{\mathrmbfit{tbl}}(\iota_{i})(\mathcal{T}_{i}) = 
\widehat{\mathcal{T}}_{i}$.
%\end{itemize}
Thus,
it is clear that we do not need to use the table $\mathcal{T}_{i'}$
to define the meet
$\prod\mathrmbfit{T} = 
\bigwedge \bigl\{\widehat{\mathcal{T}}_{i} \mid i \in I \bigr\}$,
but we do need at least one table, such as $\mathcal{T}_{i}$,
from which $\mathcal{T}_{i'}$ can be connected.
%\end{itemize}
%
We only need a discrete collection of tables 
\newline\mbox{}\hfill
$\underset{\textbf{practical input}}
{\underbrace{
\bigl\{ 
\mathcal{T}_{i} = \mathrmbfit{T}(i) \in \mathrmbf{Tbl} 
\mid i \in I \subseteq obj(\mathrmbf{I}) \bigr\}}}$ 
\hfill\mbox{}\newline
%
%consisting of an indexed, but unlinked, collection of tables 
reachable to other tables
in the collection 
{\footnotesize{$\bigl\{ 
\mathcal{T}_{i} = \mathrmbfit{T}(i) \in \mathrmbf{Tbl} 
\mid i \in \mathrmbf{I} \bigr\}$.}}
\end{aside}
%

%\newpage

%
\begin{definition}\label{def:suff:adequ:lim}
We will call such an collection of tables a \underline{sufficient} collection.
The minimal such collection can be called an \underline{adequate} collection.
\end{definition}
\begin{flushleft}
{\fbox{\fbox{\footnotesize{\begin{minipage}{345pt}
{\underline{\textsf{Special Cases:}}}
Three limits are of special interest.
\begin{itemize}
\item 
\textbf{Equalizer:}
the constraint diagram 
%{{${\bullet}{\,\leftleftarrows\,}{\bullet}$.}}
is of shape $1 \leftleftarrows 0$.
The subsets $\{0,1\}$ and $\{1\}$ are sufficient,
with the subset $\{1\}$ being adequate. 
The subsets $\{0\}$ and $\emptyset$ are not sufficient.
\emph{Quotient} 
(\S\,\ref{sub:sub:sec:quotient})
models equalizer, 
using the minimal table index set $\{1\}$
with just one table and no table morphisms.
\item 
\textbf{Pullback:}
the constraint diagram is of shape $1 \leftarrow 0 \rightarrow 2$.
The subsets $\{0,1,2\}$ and $\{1,2\}$ are sufficient,
with the subset $\{1,2\}$ being adequate.  
The subsets $\{0\}$, $\{1\}$, $\{2\}$ and $\emptyset$ are not sufficient.
\emph{Natural join} 
(\S\,\ref{sub:sub:sec:nat:join})
models pullback,
%the input table diagram is of shape the 
%(for pullback)
%with its underlying constraint diagram of shape 
%{{${\bullet}{\leftarrow}{\bullet}{\rightarrow}{\bullet}$.}}
%For constructing this limit,
%\item \textbf{Pullback:}
using the minimal table index set $\{1,2\}$
with just two tables and no table morphisms.
\item 
\textbf{Limit:}
the constraint diagram is of shape $\mathrmbf{I}$
with indexes $i \in obj(\mathrmbf{I})$
linked by morphisms $i \xrightarrow{e} i'$. 
%$\bigl\{ i \xrightarrow{e} j \mid e \in \mathrmbf{I} \bigr\}$.
\emph{Generic meet} 
(\S\,\ref{sub:sub:sec:generic:meet})
models limit,
using a sufficient (Def.\;\ref{def:suff:adequ:lim})
table index set 
{\footnotesize{$\bigl\{i \in I \subseteq obj(\mathrmbf{I}) \bigr\}$}}
with no table morphisms.
The linked collection of tables
$\mathrmbf{I}\xrightarrow{\mathrmbfit{T}}\mathrmbf{Tbl}$ 
is replaced by
the collection of tables
$I\xrightarrow{\mathrmbfit{T}}\mathrmbf{Tbl}$ 
with
{\footnotesize{$
\bigl\{ 
\mathcal{T}_{i} =
\mathrmbfit{T}(i) \in \mathrmbf{Tbl}(\mathcal{D}_{i}) 
\mid i \in I \subseteq obj(\mathrmbf{I}) \bigr\}$}}.
\end{itemize}
\end{minipage}}}}}
\end{flushleft}

%special cases
%%%%%%%%%%%%%%%%%%%%%%%%%%%%%%%%%%%%%%%%%%%%%%%%%%%%%%%%%%%%%%
%%%%%%%%%%%%%%%%%%%%%%%%%%%%%%%%%%%%%%%%%%%%%%%%%%%%%%%%%%%%%%

%%%%%%%%%%%%%%%%%%%%%%%%%%%%%%%%%%%%%%%%%%%%%%%%%%%%%%%%%%%%%%%%%%
%
\newpage
\subsubsection{Co-completeness.}
\label{sub:sub:sec:colim:tbl}
%%%%%%%%%%%%%%%%%%%%%%%%%%%%%%%%%%%%%%%%%%%%%%%%%%%%%%%%%%%%%%%%%%

%\item[Co-completeness:] 
For the co-completeness of $\mathrmbf{Tbl}$
and the co-continuity of
$\mathrmbf{Tbl}\;\xrightarrow{\;\mathrmbfit{dom}^{\mathrm{op}}}\mathrmbf{Dom}^{\mathrm{op}}$,
use the left adjoint 
%covariant 
pseudo-passage
$\mathrmbf{Dom}^{\mathrm{op}}\!\xrightarrow{\;\acute{\mathrmbfit{tbl}}\;}\mathrmbf{Cxt}$ 
(See Fact.\;4 of \cite{kent:fole:era:tbl}):
%use duality:
%\begin{enumerate}
%\item 
the indexing context 
$\mathrmbf{Dom}^{\mathrm{op}}$ is co-complete;
%\item 
the fiber context $\mathrmbf{Tbl}(\mathcal{D})$ 
of \S\,\ref{sub:sub:sec:boole}
is co-complete (joins exist)
for each signed domain $\mathcal{D}$; and
%(See Chap.\;4 of \cite{kent:fole:era:tbl}). 
%\item 
the fiber passage 
$\mathrmbf{Tbl}(\mathcal{D}')
\xleftarrow{\acute{\mathrmbfit{tbl}}({h,f,g})}
\mathrmbf{Tbl}(\mathcal{D})$ 
of \S\,\ref{sub:sub:sec:flow:sign:dom:mor}
(\textit{project-expand} operation)
is co-continuous 
for each signed domain morphism 
$\mathcal{D}'\xrightarrow{{\langle{h,f,g}\rangle}}\mathcal{D}$.
\mbox{}\hfill\rule{5pt}{5pt}
%\end{enumerate}

Given a diagram of tables
$\mathrmbf{I} \xrightarrow{\mathrmbfit{T}} \mathrmbf{Tbl}$ 
consisting of an indexed (and linked) collection of tables 
{\footnotesize{$
\underset{\textbf{theoretical input}}{\underbrace{
\bigl\{ 
\mathcal{T}_{i} = \mathrmbfit{T}(i) \in \mathrmbf{Tbl} 
\mid i \in \mathrmbf{I} \bigr\}
}}
$}}
with links
$\mathcal{T}_{i'} 
\xleftarrow[\mathrmbfit{T}(e)]{{\langle{{\langle{h,f,g}\rangle},k}\rangle}} 
\mathcal{T}_{i}$
for each $i' \xrightarrow{e} i$ in $\mathrmbf{I}$,
Prop.\;\ref{prop:lim:tbl}
states that the diagram of tables
$\mathrmbfit{T}$
%$\mathrmbf{I} \xrightarrow{\mathrmbfit{T}} \mathrmbf{Tbl}$ 
has a colimit table $\coprod\mathrmbfit{T}$  
with injection bridge
$\Delta(\coprod\mathrmbfit{T}) \xLeftarrow{\;\iota\,} \mathrmbfit{T}$ 
consisting of an indexed collection of table morphisms 
{\footnotesize{$
\underset{\textbf{theoretical output}}
{\underbrace{\textstyle{
\bigl\{ \coprod\mathrmbfit{T} \xleftarrow
[{\langle{{\langle{\tilde{h}_{i},\tilde{f}_{i},\tilde{g}_{i}}\rangle},\tilde{k}_{i}}\rangle}]
{\iota_{i}} \mathcal{T}_{i} 
\mid i \in \mathrmbf{I} \bigr\}
}}}$}}
satisfying naturality.
% 
%\begin{description}
%
%%%%%%%%%%%%%%%%%%%%%%%%%%%%%%%%%%%%%%%%%%%%%%%%%%%%%%%%%%%%%%%%%%
%\newpage\paragraph{\textbf{Application (limits).}}
%%%%%%%%%%%%%%%%%%%%%%%%%%%%%%%%%%%%%%%%%%%%%%%%%%%%%%%%%%%%%%%%%
%\item[\underline{Limits:}] 
%\newpage
%
\begin{description}
\item[\underline{Input:}] 
We use a sufficient subset of tables 
(Def.\;\ref{def:suff:adequ:colim})
\newline\mbox{}\hfill
{\footnotesize{$\underset{\textbf{practical input}}
{\underbrace{
\bigl\{ 
\mathcal{T}_{i} = \mathrmbfit{T}(i) \in \mathrmbf{Tbl} 
\mid i \in I \subseteq obj(\mathrmbf{I}) \bigr\}}}$}} 
\hfill\mbox{}\newline
%
%consisting of an indexed, but unlinked, collection of tables 
reachable \underline{from} other tables
in the collection 
{\footnotesize{$\bigl\{ 
\mathcal{T}_{i} = \mathrmbfit{T}(i) \in \mathrmbf{Tbl} 
\mid i \in \mathrmbf{I} \bigr\}$.}}
\item[\underline{Constraint:}] 
The projection 
$\mathrmbf{Tbl}^{\mathrm{op}}\xrightarrow{\;\mathrmbfit{dom}\;}\mathrmbf{Dom}$
maps the diagram of tables $\mathrmbfit{T}$ to a diagram of signed domains
$
\mathrmbf{I}^{\mathrm{op}}\xrightarrow{\mathrmbfit{D}} \mathrmbf{Dom}
$  
consisting of an indexed collection 
%of signed domains 
$\mathcal{D}_{i} = \mathrmbfit{D}(i) \in \mathrmbf{Dom} 
\mid i \in \mathrmbf{I} \bigr\}$
with links
{\footnotesize{$
\underset{\textbf{constraint}}
{\underbrace{ \bigl\{ 
\mathcal{D}_{i'} 
\xrightarrow
%[\mathrmbfit{D}(e)]
{{\langle{h,f,g}\rangle}} 
\mathcal{D}_{i}
\mid i \in \mathrmbf{I} \bigr\} }}$.}}
%for each $(i' \xrightarrow{e} i) \in \mathrmbf{I}$,
%
%\end{description}
%
%\begin{itemize}
%
\item[\underline{Construction:}] 
Since
$\mathrmbf{Tbl}
\xrightarrow{\;\mathrmbfit{dom}^{\mathrm{op}}}
\mathrmbf{Dom}^{\mathrm{op}}$
is co-continuous
(Prop.\;\ref{prop:lim:tbl}), 
it maps
the colimit table 
%$\prod\mathrmbfit{T}$  
with injection bridge
%$\Delta(\prod\mathrmbfit{T}) \xRightarrow{\,\pi\;} \mathrmbfit{T}$  
to
the limit signed domain $\prod\mathrmbfit{D}=\widetilde{\mathcal{D}}$  
with projection bridge
%$\Delta(\coprod\mathrmbfit{D}) \xRightarrow{\,\iota\;} \mathrmbfit{D}$  
%to the injection bridge
$\Delta(\widehat{\mathcal{D}}) \xRightarrow{\,\pi\;} \mathrmbfit{D}$ 
consisting of an indexed collection of signed domain morphisms 
{\footnotesize{$
\underset{\textbf{construction}}
{\underbrace{
\bigl\{ 
\widetilde{\mathcal{D}} 
\xrightarrow[{\langle{\tilde{h}_{i},\tilde{f}_{i},\tilde{g}_{i}}\rangle}]{\pi_{i}} 
\mathcal{D}_{i} 
\mid i \in \mathrmbf{I} \bigr\} }}
$}}
satisfying naturality.  
%
%
%\end{description}
%
\item[\underline{Output:}] 
%
%\begin{itemize}
%
%\newline\item 
Using the co-completeness aspect of Prop.\;\ref{prop:lim:tbl},
the colimit table is 
the join (union)
in the fiber context $\mathrmbf{Tbl}(\widetilde{\mathcal{D}})$
of the projection-expansion of the component tables
$\mathcal{T}_{i}$
along the projections
%signed domain morphisms 
{\footnotesize{$
\{ \mathcal{D}_{i} 
\xleftarrow
{{\langle{\tilde{h}_{i},\tilde{f}_{i},\tilde{g}_{i}}\rangle}}
%{\iota_{i}} 
\widetilde{\mathcal{D}} 
\mid i \in \mathrmbf{I} \}$.}}
%for each index $i \in \mathrmbf{I}$

%$\prod\mathrmbfit{T}$
%essentially 
%\newline\mbox{}\hfill
%\hfill\mbox{}\newline
%\begin{itemize}
%\item

$\bullet$ For each index $i \in I$,
projection-expansion
{\footnotesize{$\mathrmbf{Tbl}(\mathcal{D}_{i}) 
\xrightarrow{\acute{\mathrmbfit{tbl}}({\tilde{h}_{i},\tilde{f}_{i},\tilde{g}_{i}})} 
\mathrmbf{Tbl}(\widetilde{\mathcal{D}})$}}
(\S\,\ref{sub:sub:sec:flow:sign:dom:mor})
maps the table 
$\mathcal{T}_{i} \in \mathrmbf{Tbl}(\mathcal{D}_{i})$
by composition
to the table
$\widetilde{\mathcal{T}}_{i}
%\widehat{\mathrmbfit{T}}(i)
%= \grave{\mathrmbfit{tbl}}({h,f,g})(\mathrmbfit{T}(i))
= {\langle{K_{i},\tilde{t}_{i}}\rangle} 
\in \mathrmbf{Tbl}({\widetilde{\mathcal{D}}})$.
This defines a table morphism 
%\newline
%\[\mbox
{\footnotesize{{$
\mathcal{T}_{i} = {\langle{K_{i},t_{i}}\rangle}
\xrightarrow{{\langle{{\langle{\tilde{h}_{i},\tilde{f}_{i},\tilde{g}_{i}}\rangle},1_{K_{i}}}\rangle}}
{\langle{K_{i},\tilde{t}_{i}}\rangle} 
= \widetilde{\mathcal{T}}_{i}$.}}\normalsize}
%\]
%\newline
%\mbox{}\rule{200pt}{1pt}
%\newline

%\item
$\bullet$
Union (\S\,\ref{sub:sub:sec:boole})
of the tables $\{ \widetilde{\mathcal{T}}_{i} \mid i \in I \}$
in the fiber context $\mathrmbf{Tbl}(\widetilde{\mathcal{D}})$
defines the 
%natural join 
%limit 
generic join
$\coprod\mathrmbfit{T}
= \widetilde{\mathrmbfit{T}} 
= \bigvee \bigl\{\widetilde{\mathcal{T}}_{i} \mid i \in I \bigr\}
= {\langle{\widetilde{K},\tilde{t}}\rangle}$,
resulting in the 
discrete multi-opspan (cocone)
{\footnotesize{{$\Bigl\{ 
\widetilde{\mathcal{T}}_{i} 
%= {\langle{\widehat{K}_{i},\hat{t}_{i}}\rangle} 
\xrightarrow{\;\tilde{\iota}_{i}\;} 
%{\langle{\widehat{K},\hat{t}}\rangle} = 
\widetilde{\mathrmbfit{T}} 
\mid i \in I \Bigr\}$.}}\normalsize}
Projection-expansion composed with join 
defines the opspan of table morphisms
$\underset{\textbf{practical output}}
{\underbrace{
\textstyle{
\bigl\{ 
\mathcal{T}_{i}
\xrightarrow
[\;{\langle{{\langle{\tilde{h}_{i},\tilde{f}_{i},\tilde{g}_{i}}\rangle},\acute{k}_{i}}\rangle}{\circ\,}\tilde{\iota}_{i}\;]
{\;{\langle{{\langle{\tilde{h}_{i},\tilde{f}_{i},\tilde{g}_{i}}\rangle},\tilde{k}_{i}}\rangle}\;} 
\prod\mathrmbfit{T} 
%= \widehat{\mathrmbfit{T}}
%\mid i \in I 
\bigr\}
}
}}$. 
%illustrated in 
%Fig.\;\ref{fig:fole:data-type:meet},
%which is the output for generic meet 
%(Tbl.\,\ref{tbl:fole:generic:meet:input:output}).
%\end{itemize}
%
%
\end{description}
%

%\newpage
\begin{aside}
To construct the table
$\coprod\mathrmbfit{T}$
it is not necessary to use all of the tables
and table morphisms in the indexed collection above. 
Let 
$\mathcal{T}_{i'} 
\xleftarrow[\mathrmbfit{T}(e)]{{\langle{{\langle{h,f,g}\rangle},k}\rangle}} 
\mathcal{T}_{i}$
%$\mathrmbfit{T}(e) = {{\langle{{\langle{h,f,g}\rangle},k}\rangle}}
%:\mathcal{T}_{i'}\xeftarrow\mathcal{T}_{i}$
be part of the diagram $\mathrmbfit{T}$
satisfying the naturality condition 
$\mathrmbfit{T}(e) \circ \iota_{i'} = \iota_{i}$
for some morphism $i' \xleftarrow{e} i$ in $\mathrmbf{I}$.
%
%\item 
Let 
{\footnotesize{
$\mathcal{D}_{i'} 
\xleftarrow[\mathrmbfit{D}(e)]
{{\langle{h,f,g}\rangle}}
\mathcal{D}_{i}$
}\normalsize}
be its signed domain morphism 
satisfying the naturality condition
$\pi_{i'} \circ \mathrmbfit{D}(e) = \pi_{i}$.
\begin{center}
{{\begin{tabular}{@{\hspace{-30pt}}c@{\hspace{10pt}}c@{\hspace{10pt}}c}
%%%%%%%%%%%%%%%%%%%%%%%%%%%%%%%%%%%%%%%%%%%%%%%%%%
{{\begin{tabular}{c}
\setlength{\unitlength}{0.56pt}
\begin{picture}(120,100)(0,-15)
\put(60,80){\makebox(0,0){\footnotesize{$\prod\mathrmbfit{D}$}}}
\put(0,0){\makebox(0,0){\footnotesize{$\mathcal{D}_{i'}$}}}
\put(120,0){\makebox(0,0){\footnotesize{$\mathcal{D}_{i}$}}}
\put(25,43){\makebox(0,0)[r]{\scriptsize{$\pi_{i'}$}}}
\put(97,43){\makebox(0,0)[l]{\scriptsize{$\pi_{i}$}}}
\put(66,10){\makebox(0,0){\scriptsize{$\mathrmbfit{D}(e)$}}}
\put(66,-10){\makebox(0,0){\tiny{${\langle{h,f,g}\rangle}$}}}
\put(45,64){\vector(-3,-4){40}}
\put(75,64){\vector(3,-4){40}}
\put(15,0){\vector(1,0){90}}
%
%\put(60,-25){\makebox(0,0){\tiny{$
%\pi_{i}=
%{\langle{{\langle{\hat{h}_{i},\hat{f}_{i},\hat{g}_{i}}\rangle},\hat{k}_{i}}\rangle}
%$}}}
\end{picture}
\end{tabular}}}
%%%%%%%%%%%%%%%%%%%%%%%%%%%%%%%%%%%%%%%%%%%%%%%%%%
{{\begin{tabular}{c}
\setlength{\unitlength}{0.56pt}
\begin{picture}(48,100)(0,0)
\put(30,60){\makebox(0,0){\footnotesize{$
\overset{\textstyle{\textit{naturality}}}
{\textstyle{\textit{diagrams}}}
$}}}
\end{picture}
\end{tabular}}}
%%%%%%%%%%%%%%%%%%%%%%%%%%%%%%%%%%%%%%%%%%%%%%%%%%
&
%%%%%%%%%%%%%%%%%%%%%%%%%%%%%%%%%%%%%%%%%%%%%%%%%%
{{\begin{tabular}{c}
\setlength{\unitlength}{0.56pt}
\begin{picture}(120,100)(0,-15)
\put(60,80){\makebox(0,0){\footnotesize{$\coprod\mathrmbfit{T}$}}}
\put(0,0){\makebox(0,0){\footnotesize{$\mathcal{T}_{i'}$}}}
\put(120,0){\makebox(0,0){\footnotesize{$\mathcal{T}_{i}$}}}
\put(25,43){\makebox(0,0)[r]{\scriptsize{$\iota_{i'}$}}}
\put(97,43){\makebox(0,0)[l]{\scriptsize{$\iota_{i}$}}}
\put(66,10){\makebox(0,0){\scriptsize{$\mathrmbfit{T}(e)$}}}
\put(66,-10){\makebox(0,0){\tiny{$
{\langle{{\langle{h,f,g}\rangle},k}\rangle}$}}}
\put(6,12){\vector(3,4){39}}
\put(114,12){\vector(-3,4){39}}
\put(105,0){\vector(-1,0){90}}
%
%\put(-15,25){\makebox(0,0)[r]{\footnotesize{$
%\underset{\textstyle{\textit{needed}}}
%{\overset{\textstyle{\textit{component}}}{\textit{not}}}
%\left\{\rule{0pt}{20pt}\right.$}}}
\put(135,25){\makebox(0,0)[l]
{\footnotesize{$
\left.
\rule{0pt}{20pt}
\right\}
\underset{\textstyle{\textit{needed}}}
{\overset{\textstyle{\textit{component}}}{\textit{not}}}
$}}}
%
%\put(60,-25){\makebox(0,0){\tiny{$
%\iota_{i}=
%{\langle{\hat{h}_{i},\hat{f}_{i},\hat{g}_{i}}\rangle}
%$}}}
\end{picture}
\end{tabular}}}
%%%%%%%%%%%%%%%%%%%%%%%%%%%%%%%%%%%%%%%%%%%%%%%%%%
%%%%%%%%%%%%%%%%%%%%%%%%%%%%%%%%%%%%%%%%%%%%%%%%%%
\end{tabular}}}
\end{center}
%
%\begin{itemize}
%\item 
Using adjoint flow
(Disp.\;\ref{def:sign:dom:mor:adj}),
the table morphism
$\mathcal{T}_{i'} 
\xleftarrow
%[\mathrmbfit{T}(e)]
{{\langle{{\langle{h,f,g}\rangle},k}\rangle}} 
\mathcal{T}_{i}$
implies existence of a morphism 
$\mathcal{T}_{i'}
\xleftarrow{\,k\,}
\acute{\mathrmbfit{tbl}}(h,f,g)(\mathcal{T}_{i})
$
in the fiber context $\mathrmbf{Tbl}(\mathcal{D}')$.
%\item 
%the fiber context $\mathrmbf{Tbl}(\mathcal{D})$ is a preorder,
For all practical purposes,
by reflection (\S\,\ref{sub:sub:sec:reflect})
we 
essentially 
have the inclusion
$\mathcal{T}_{i'}
\supseteq 
\acute{\mathrmbfit{tbl}}({h,f,g})(\mathcal{T}_{i}) $.
%\item 
Applying projection-expansion to the naturality condition
$\pi_{i'} \circ \mathrmbfit{D}(e) = \pi_{i}$,
get
$\acute{\mathrmbfit{tbl}}(\pi_{i}) = 
\acute{\mathrmbfit{tbl}}({h,f,g}) \cdot \acute{\mathrmbfit{tbl}}(\pi_{i'})$.
%\item
Hence,
$\widetilde{\mathcal{T}}_{i'} = 
\acute{\mathrmbfit{tbl}}(\pi_{i'})(\mathcal{T}_{i'}) \geq 
\acute{\mathrmbfit{tbl}}(\pi_{i'})
(\acute{\mathrmbfit{tbl}}({h,f,g})(\mathcal{T}_{i})) =
\acute{\mathrmbfit{tbl}}(\pi_{i})(\mathcal{T}_{i}) =
\widetilde{\mathcal{T}}_{i}$.
%\end{itemize}
Thus,
it is clear that we do not need to use the table $\mathcal{T}_{i}$
to define the join
$\coprod\mathrmbfit{T} = 
\bigvee \bigl\{\widetilde{\mathcal{T}}_{i} \mid i \in I \bigr\}$,
but we do need at least one table, such as $\mathcal{T}_{i'}$,
to which $\mathcal{T}_{i}$ can be connected.
%\end{itemize}
%\mbox{}\newline\rule{300pt}{1pt}\newline
%
We only need a discrete collection of tables 
\newline\mbox{}\hfill
$\underset{\textbf{practical input}}
{\underbrace{
\bigl\{ 
\mathcal{T}_{i} = \mathrmbfit{T}(i) \in \mathrmbf{Tbl} 
\mid i \in I \subseteq obj(\mathrmbf{I}) \bigr\}}}$ 
\hfill\mbox{}\newline
%
%consisting of an indexed, but unlinked, collection of tables 
reachable from other tables
in the collection 
{\footnotesize{$\bigl\{ 
\mathcal{T}_{i} = \mathrmbfit{T}(i) \in \mathrmbf{Tbl} 
\mid i \in \mathrmbf{I} \bigr\}$.}}
\end{aside}
\begin{definition}\label{def:suff:adequ:colim}
We will call such a collection of tables a \underline{sufficient} collection.
The minimal such collection can be called an \underline{adequate} collection.
\end{definition}
\begin{flushleft}
{\fbox{\fbox{\footnotesize{\begin{minipage}{345pt}
{\underline{\textsf{Special Cases:}}}
Three colimits are of special interest.
\begin{itemize}
\item 
\textbf{Coequalizer:}
the constraint diagram 
%{{${\bullet}{\,\leftleftarrows\,}{\bullet}$.}}
is of shape $1 \rightrightarrows 0$.
The subsets $\{0,1\}$ and $\{1\}$ are sufficient,
with the subset $\{1\}$ being adequate. 
The subsets $\{0\}$ and $\emptyset$ are not sufficient.
\emph{Co-quotient} 
(\S\,\ref{sub:sub:sec:co-quotient})
models coequalizer, 
using the minimal table index set $\{1\}$
with just one table and no table morphisms.
\item 
\textbf{Pushout:}
the constraint diagram is of shape $1 \rightarrow 0 \leftarrow 2$.
The subsets $\{0,1,2\}$ and $\{1,2\}$ are sufficient,
with the subset $\{1,2\}$ being adequate.  
The subsets $\{0\}$, $\{1\}$, $\{2\}$ and $\emptyset$ are not sufficient.
\emph{data-type join} 
(\S\,\ref{sub:sub:sec:boole:join})
models pushout,
%the input table diagram is of shape the 
%(for pullback)
%with its underlying constraint diagram of shape 
%{{${\bullet}{\leftarrow}{\bullet}{\rightarrow}{\bullet}$.}}
%For constructing this limit,
%\item \textbf{Pullback:}
using the minimal table index set $\{1,2\}$
with just two tables and no table morphisms.
\item 
\textbf{Colimit:}
the constraint diagram is of shape $\mathrmbf{I}$
with indexes $i \in obj(\mathrmbf{I})$
linked by morphisms $i' \xleftarrow{\,e} i$. 
%$\bigl\{ i \xrightarrow{e} j \mid e \in \mathrmbf{I} \bigr\}$.
\emph{Generic join} 
(\S\,\ref{sub:sub:sec:generic:join})
models colimit,
using a sufficient (Def.\;\ref{def:suff:adequ:colim})
table index set 
{\footnotesize{$\bigl\{i \in I \subseteq obj(\mathrmbf{I}) \bigr\}$}}
with no table morphisms.
The linked collection of tables
$\mathrmbf{I}\xrightarrow{\mathrmbfit{T}}\mathrmbf{Tbl}$ 
is replaced by
the collection of tables
$I\xrightarrow{\mathrmbfit{T}}\mathrmbf{Tbl}$ 
with
{\footnotesize{$
\bigl\{ 
\mathcal{T}_{i} =
\mathrmbfit{T}(i) \in \mathrmbf{Tbl}(\mathcal{D}_{i}) 
\mid i \in I \subseteq obj(\mathrmbf{I}) \bigr\}$}}.
\end{itemize}
\end{minipage}}}}}
\end{flushleft}
%

%\input{overview}
%\input{tup-calc}
%\input{tup-calc2}
%\input{tup-calc3}
%\input{tup-calc4}

%\newpage


\begin{thebibliography}{4}

\bibitem{barwise:seligman:97} 
Barwise, J., and Seligman, J.:
{\itshape Information Flow: The Logic of Distributed Systems}.
Cambridge University Press, Cambridge (1997).

\bibitem{chen:76} 
Chen, P.:
``The Entity-Relationship Model --- Toward a Unified View of Data.'' 
(1976).
{\itshape ACM Trans. on Database Sys.}, 1 (1): 9--36. 
%doi:10.1145/320434.320440.

\bibitem{codd:70} 
Codd, E.F.: 
``A Relational Model of Data for Large Shared Data Banks.''
(1970).
{\itshape Comm. of the ACM.} 13 (6): 377--387. 
%doi:10.1145/362384.362685.

\bibitem{codd:90} 
Codd, E.F.:
``The Relational Model for Database Management (Version 2 ed.).''
(1990).
{\itshape Addison Wesley}, Boston: 371--388.
ISBN 0-201-14192-2.

\bibitem{ganter:wille:99} 
Ganter, B., and Wille, R.:
{\itshape Formal Concept Analysis: Mathematical Foundations}.
Springer, New York (1999).

%\bibitem{goguen:cm91} 
%Goguen, J.:
%{\itshape A categorical manifesto}.
%Mathematical Structures in Computer Science 1, 49--67 (1991).

%\bibitem{goguen:iii2004} 
%Goguen, J.:
%{\itshape Information Integration in Institutions} 
%Paper for Jon Barwise memorial volume (2004).
%Available online: 
%\url{http://cseweb.ucsd.edu/~goguen/pps/ifi04.pdf}.

%\bibitem{goguen:burstall:92} 
%Goguen, J., and Burstall, R.:
%``Institutions: Abstract Model Theory for Specification and Programming''. 
%J. Assoc. Comp. Mach. vol. 39, pp. 95--146 (1992).

%\bibitem{google:spanner} 
%Google authors:
%J. Corbett, J. Dean, M. Epstein, C. Frost, J. Furman, S. Ghemawat, A. Gubarev, C. Heiser, P. Hochschild, W. Hsieh, S. Kanthak, E. Kogan, H. Li, A. Lloyd,, S. Melnik, D. Mwaura, D. Nagle, S. Quinlan, R. Rao, L. Rolig, Y. Saito, M. Szymaniak, C. Taylor, R. Wang, and D. Woodford:
%``Spanner: Google's Globally-Distributed Database''.
%Available online: 
%\url{http://research.google.com/archive/spanner.html}.
%\url{https://www.usenix.org/system/files/conference/osdi12/osdi12-final-16.pdf}.

%\bibitem{gruber:eds2009} 
%Gruber, T.
%``Ontology''.
%In: Ling Liu and M. Tamer \"Ozsu (eds.) 
%{\itshape The Encyclopedia of Database Systems}, 
%Springer-Verlag (2009).
%Available online: 
%\url{tomgruber.org/writing/ontology-definition-2007.htm}.

%\bibitem{johnson:rosebrugh:07}
%Johnson, M., Rosebrugh, R.: 
%``Fibrations and Universal View Updatability''. 
%Th. Comp. Sci. 388: 109--129 (2007).
%\bibitem{johnson:rosebrugh:wood:era} 
%Johnson, M., Rosebrugh, R., and Wood, R.
%``Entity Relationship Attribute Designs and Sketches''.
%Theory and Application of Categories 10, 3, 94--112 (2002). 

%\bibitem{kalfoglou:schorlemmer:ont:map} 
%Kalfoglou, Y., and Schorlemmer M.
%``Ontology Mapping: The State of the Art''.
%In: 
%Kalfoglou, Y., Schorlemmer M., Sheth, A., Staab, S., and Uschold, M. (eds.)
%{\itshape Semantic Interoperability and Integration},
%Dagstuhl Seminar Proceedings 04391 (2005).
%issn:1862-4405.
%Available online: 
%\url{http://drops.dagstuhl.de/opus/volltexte/2005/40}.

%\bibitem{kent:tao}
%Kent, R.E.:
%``The Institutional Approach''.
%In: Poli, R., Healy, M., Kameas, A. (eds.)
%(invited chapter)
%{\itshape Theory and Applications of Ontology: Computer Applications}.
%Springer, Heidelberg, 
%533--563 
%(2010).
%\bibitem{johnstone:77}
%Johnstone, P.T.:
%{\itshape Topos Theory}.
%Academic Press, London (1977).

%\bibitem{kent:iccs2009}
%Kent, R.E.
%``System Consequence''.
%In: 
%Rudolph, S., Dau, F., and Kuznetsov, S.O. (eds.)
%{\itshape Conceptual Structures: Leveraging Semantic Technologies},
%LNCS vol. 5662, pp. 201--218.
%Springer, Heidelberg (2009). 

\bibitem{kent:db:sem}
Kent, R.E.:
``Database Semantics.''
(2011).
Available online: 
\url{http://arxiv.org/abs/1209.3054}.

\bibitem{kent:iccs2013}
Kent, R.E.:
``The First-order Logical Environment.''
In: 
Pfeiffer, H.D., Ignatov, D.I., Poelmans, J., and Nagarjuna G. (eds.)
{\itshape Conceptual Structures in Research and Education},
LNCS (7735): 210--230.
Springer, Heidelberg (2013). 
Available online: 
\url{arxiv.org/abs/1305.5240}.

\bibitem{kent:fole:era:found}
Kent, R.E.:
``The {\ttfamily ERA} of {\ttfamily FOLE}: Foundation.''
(2015).
%Submitted for publication.
Available online: 
\url{arxiv.org/abs/1512.07430}.

\bibitem{kent:fole:era:supstruc}
Kent, R.E.:
``The {\ttfamily ERA} of {\ttfamily FOLE}: Superstructure.''
(2016).
Available online: 
\url{arxiv.org/abs/1602.04268}.

%\bibitem{kent:fole:era:interp}
%Kent, R.E.:
%``The {\ttfamily ERA} of {\ttfamily FOLE}: Interpretation''.

\bibitem{kent:fole:era:tbl}
Kent, R.E.:
``The {\ttfamily FOLE} Table.''
%Submitted for publication.
Available online: 
\url{arxiv.org/abs/1810.12100}.
%\href{https://arxiv.org/abs/1810.12100}{The {\ttfamily FOLE} Table}

\bibitem{kent:fole:era:db}
Kent, R.E.:
``The {\ttfamily FOLE} Database''.
Unpublished paper.

\bibitem{kent:fole:equiv}
Kent, R.E.:
``The {\ttfamily FOLE} Equivalence''.
Unpublished paper.

\bibitem{maclane:71}
Mac Lane, S.,
{\itshape Categories for the Working Mathematician}.
Springer-Verlag (1971).

%\bibitem{sowa:kr}
%Sowa, J.F.:
%{\itshape Knowledge Representation: Logical, Philosophical, and Computational Foundations}.
%Brookes/Coles (2000).
%\bibitem{spivak:sd}
%Spivak,D.I.: 
%``Simplicial Databases'' (2009. 
%\newline
%Available online: 
%\url{http://arxiv.org/abs/0904.2012}). 
%\bibitem{spivak:fdm}
%Spivak,D.I.: 
%``Functorial Data Migration'' (2010). 
%\newline
%Available online: 
%\url{http://arxiv.org/abs/1009.1166}.
%\bibitem{spivak:kent:olog}
%Spivak, D.I., and Kent, R.E.:
%``Ologs: a categorical framework for knowledge representation''.
%PLoS ONE 7(1): e24274. doi:10.1371/journal.pone.0024274.
%(2012). 
%Available online: 
%\url{http://arxiv.org/abs/1102.1889}.
%\bibitem{tarlecki:burstall:goguen:91}
%Tarlecki, A., Burstall, R., Goguen, J.: 
%``Some Fundamental Algebraic Tools for the Semantics of Computation, Part 3: 
%Indexed Categories''. 
%Th. Comp. Sci. vol. 91, pp. 239--264.
%Elsevier (1991). 

\bibitem{iff} 
{\itshape The Information Flow Framework (IFF)}. 
The Standard Upper Ontology (SUO) working group IEEE P1600.1. 
%\newline
Available online: 
\newline
\url{http://web.archive.org/web/20121008145548/http://suo.ieee.org/IFF/}.
%\url{http://web.archive.org/web/20071018073538/http://suo.ieee.org/IFF/}
%\newline
%Previously available online: 
%\url{http://suo.ieee.org/IFF/}.
%\bibitem{nlab:grothendieck} 
%nLab Grothendieck fibration. 
%Online:
%\url{ncatlab.org/nlab/show/Grothendieck+fibration}.

\end{thebibliography}
\end{document}